\pdfoutput=1
\RequirePackage[current]{latexrelease}
\documentclass[%
fontsize=11pt, %
BCOR=15mm, %
DIV=13, %
toc=bib,
index=toc,
listof=flat,
listof=leveldown,
listof=nochaptergap,
chapterprefix=true,
headsepline=true, %
headings=optiontoheadandtoc,
]{scrbook}[2020/04/19 v3.30 KOMA-Script package (KOMA-Script-independent]
\usepackage[utf8]{inputenc}[2018/08/11 v1.3c Input encoding file]%
\usepackage[T1]{fontenc}[2020/02/11 v2.0o Standard LaTeX package]%
\usepackage{microtype}[microtype 2019/11/18 v2.7d Micro-typographical refinements (RS)]%
\usepackage{newtxtext}[2020/03/02 v1.625]%
\usepackage[french,main=british]{babel}[2020/05/13 3.44 The Babel package]%
\usepackage[useregional]{datetime2}[2020/03/02 v1.5.6 (NLCT) date and time formats]
\usepackage{csquotes}[2019-12-06 v5.2j csquotes generic definitions (JAW)]%
\usepackage[inline]{enumitem}[2019/06/20 v3.9 Customized lists]%
\usepackage[symbol]{footmisc}[2011/06/06 v5.5b a miscellany of footnote facilities]%
\usepackage{multicol}[2019/12/09 v1.8y multicolumn formatting (FMi)]%
\usepackage{tabularx}[2020/01/15 v2.11c `tabularx' package (DPC)]%
\usepackage{etoolbox}[2019/09/21 v2.5h e-TeX tools for LaTeX (JAW)]%
\usepackage{xifthen}[2015/11/05 v1.4.0 Extended ifthen features]
\usepackage{xspace}[2014/10/28 v1.13 Space after command names (DPC,MH)]%
\usepackage{scrwfile}[2013/08/05 v0.1f-alpha KOMA-Script package (write and clone files)]%
\title{Asymptotic analysis of the form-factors of the quantum spin chains}%
\subtitle{Étude asymptotique des facteurs de formes des chaînes de spins}%
\author{Giridhar V. KULKARNI}%
\date{20/11/2020}
\DTMsavedate{defencedate}{2020-11-20}
\DTMsavedate{releasedate}{2020-12-04}
\newcommand{\engtitle}{Asymptotic analysis of the form-factors of the quantum spin chains}%
\newcommand{\frtitle}{Étude asymptotique des facteurs de forme des chaînes de spin quantiques}%
\newcommand{\frkeywords}{%
ansatz de Bethe algebrique, méthode de la dispersion inverse, facteurs de forme, chaîne de spin XXX, limite thermodynamique, systèmes intégrables%
}
\newcommand{\engkeywords}{%
algebraic Bethe ansatz, quantum inverse scattering method, form-factors, XXX spin chain, thermodynamic limit, integrable systems%
}
\KOMAoption{chapterprefix}{true}
\setkomafont{chapterprefix}{\large\bfseries\itshape}
\renewcommand*{\chapterformat}{%
    \chapappifchapterprefix\space{\huge \thechapter\autodot}%
  \IfUsePrefixLine{%
    \par\nobreak\vspace{-\parskip}\vspace{-.6\baselineskip}%
    \rule{\textwidth}{1pt}%
    \vspace{-.6\baselineskip}%
  }{\enskip}%
}

\usepackage{amsthm}[2020/05/29 v2.20.6]
\usepackage[upint]{newtxmath}[2020/09/19 v1.630]
\usepackage[only=bigtriangleup]{stmaryrd}[1994/03/03 St Mary's Road symbol package] %
\SetSymbolFont{stmry}{bold}{U}{stmry}{m}{n}
\usepackage{mymacros}[2020/12/02 custom math macros v2.0 by GK] %
\usepackage{mathtools}[2020/03/24 v1.24 mathematical typesetting tools] %
\usepackage{array}[2019/08/31 v2.4l Tabular extension package (FMi)]   %
\newcolumntype{L}{>{$}l<{$}} %
\newcolumntype{C}{>{$}c<{$}} %
\usepackage{multirow}[2019/05/31 v2.5 Span multiple rows of a table]
\theoremstyle{plain}
\newtheorem{thm}{Theorem}[chapter]
\newtheorem{prop}[thm]{Proposition}

\newtheorem{lem}[thm]{Lemma}
\newtheorem*{coro}{Corollary}
\theoremstyle{definition}
\newtheorem{defn}{Definition}
\newtheorem{notn}[defn]{Notation}
\theoremstyle{remark}
\newtheorem*{rem}{Remark}
\newtheorem*{example}{Example}

\usepackage{graphicx}[2019/11/30 v1.2a Enhanced LaTeX Graphics (DPC,SPQR)]
\usepackage{xcolor}[2016/05/11 v2.12 LaTeX color extensions (UK)]
\definecolor{violetUBFC}{rgb}{.71,.12,.56}
\usepackage{tcolorbox}[2020/09/25 version 4.40 text color boxes]
\usepackage[framemethod=tikz]{mdframed}[2013/07/01 1.9b: mdframed]
\newtcolorbox{resumebox}[1]{%
colback=white,colframe=violetUBFC,
colbacktitle=violetUBFC,
fonttitle=\bfseries,
title=#1
}
\newcommand{\halfboxitem}[1]{%
	\tikz[baseline=(item.base)]{%
		\node[inner sep=2pt, minimum size=1em] (item) {#1};%
		\draw[thick] (item.south west) -| (item.north east);%
		\draw[thin] (item.south west) |- (item.north east);%
	}%
}%
\setlist{parsep=0pt}
\setlist[enumerate, 1]{label=\textbf{\arabic*.}}
\setlist[enumerate, 2]{label={(\arabic{enumi}\alph*).}}
\usepackage{subcaption}[2020/08/23 v1.3g Sub-captions (AR)]
\AtEndEnvironment{figure}{%
\vspace{\stretch{1}}
\rule{\linewidth}{.5pt}
}
\setcapindent{1em}
\addtokomafont{captionlabel}{\bfseries}
\usepackage{tikz}[2020/09/28 v3.1.6 (3.1.6)]
\usetikzlibrary{calc, arrows, fadings, matrix, scopes, positioning, chains, shapes, decorations.pathmorphing, chains}
\DeclareRobustCommand{\tikzcross}{\tikz{\node[draw, cross out, inner sep=2pt] {};}\xspace}
\DeclareRobustCommand{\tikzcircfill}{\tikz{\node[fill, circle, inner sep =1.5pt] {};}\xspace}
\DeclareRobustCommand{\tikzzigzag}{\tikz{\draw[decoration = {zigzag, segment length = 2pt, amplitude = .5pt}, decorate, baseline=(current bounding box.center)] (0,0)--(1,0);}\xspace}
\usepackage{imakeidx}[2016/10/15 v1.3e Package for typesetting indices in a synchro]
\makeindex[%
	title=Index of notations,
	columns=1,
	intoc,
	options=-s myidx.ist
	]

\usepackage{appendix}[2020/02/08 v1.2c extra appendix facilities]

\usepackage[%
			backend=biber,%
			style=ext-alphabetic-verb,%
			articlein=false,%
			maxalphanames=4,%
			minalphanames=3,%
			maxnames=7,%
			minnames=5,
			giveninits=true,%
			date=year,%
			useprefix=false,%
			]{biblatex}
\DeclareBibliographyCategory{ignore}
\addtocategory{ignore}{Hei07}
\addtocategory{ignore}{Sch35}
\usepackage[%
	colorlinks=true,
	linktoc=all,
	bookmarks=true,
	bookmarksnumbered=true,
	bookmarksopen=true,
	bookmarksopenlevel=0,
	pdftitle={Asymptotic analysis of the form-factors of quantum spin chains{}},
	pdfauthor={Giridhar V. Kulkarni{}},
	pdfsubject={mathematical physics},
	pdfkeywords={%
	algebraic Bethe ansatz, quantum inverse scattering method, form-factors, XXX spin chain, thermodynamic limit, integrable systems%
	}
	]{hyperref}[2020-05-15 v7.00e Hypertext links for LaTeX]
\hypersetup{unicode=true}
\usepackage{cleveref}[2018/03/27 v0.21.4 Intelligent cross-referencing]
\crefname{table}{table}{tables}
\crefname{figure}{fig.}{figs.}
\crefname{appendix}{appendix}{appendices}
\crefname{lem}{lemma}{lemmas}
\let\oldlem\lem
\renewcommand{\lem}{%
\crefalias{thm}{lem}
\oldlem%
}%
\let\oldapp\appendices
\renewcommand{\appendices}{%
\crefalias{chapter}{appendix}
\oldapp%
}%
\crefname{thm}{theorem}{theorems}
\crefname{prop}{proposition}{propositions}
\let\oldprop\prop
\renewcommand{\prop}{%
\crefalias{thm}{prop}
\oldprop%
}%
\let\oldnotn\notn
\renewcommand{\notn}{%
\crefalias{defn}{notn}
\oldnotn%
}%
\crefname{defn}{definition}{definitions}
\crefname{notn}{notation}{notations}
\crefname{conj}{conjecture}{conjectures}
\TOCclone[Short Table of Contents]{toc}{stoc}%
\addtocontents{stoc}{\protect\value{tocdepth}=0}
\let\oldfootnotemark\footnotemark
\renewcommand{\footnotemark}{\oldfootnotemark\xspace}
\deffootnote{1em}{\parindent}{%
	\textsuperscript{\thefootnotemark\space}%
}
\setfootnoterule[1pt]{.5\textwidth}
\begin{document}
\selectlanguage{french}
\begin{titlepage}
	\setlength{\unitlength}{1cm}
	\begin{picture}(4,3)
	\put(2.85,1.5){\makebox(0,0)[cc]%
	{\includegraphics*[height=3cm]{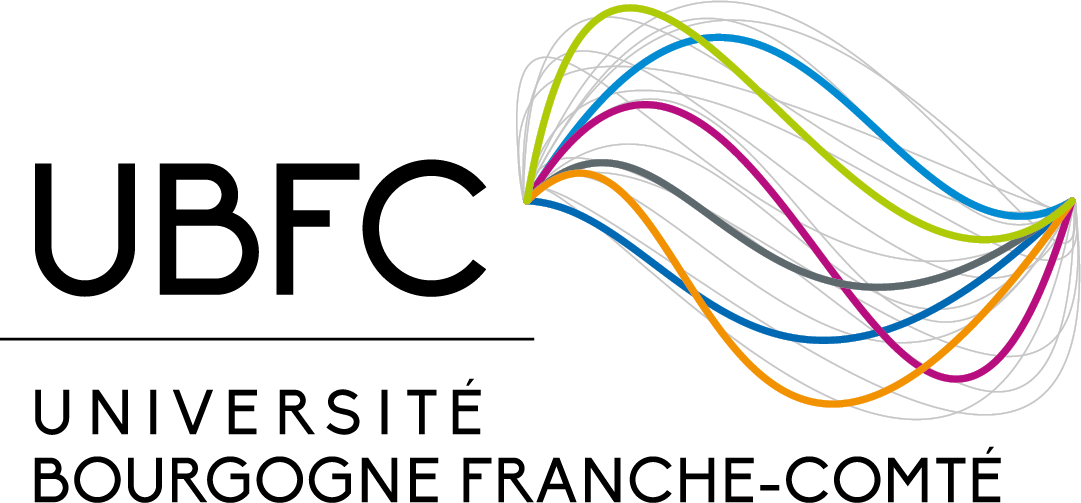}}
	}
	\put(11.85,1.5){\makebox(0,0)[cc]%
	{\includegraphics*[height=3cm]{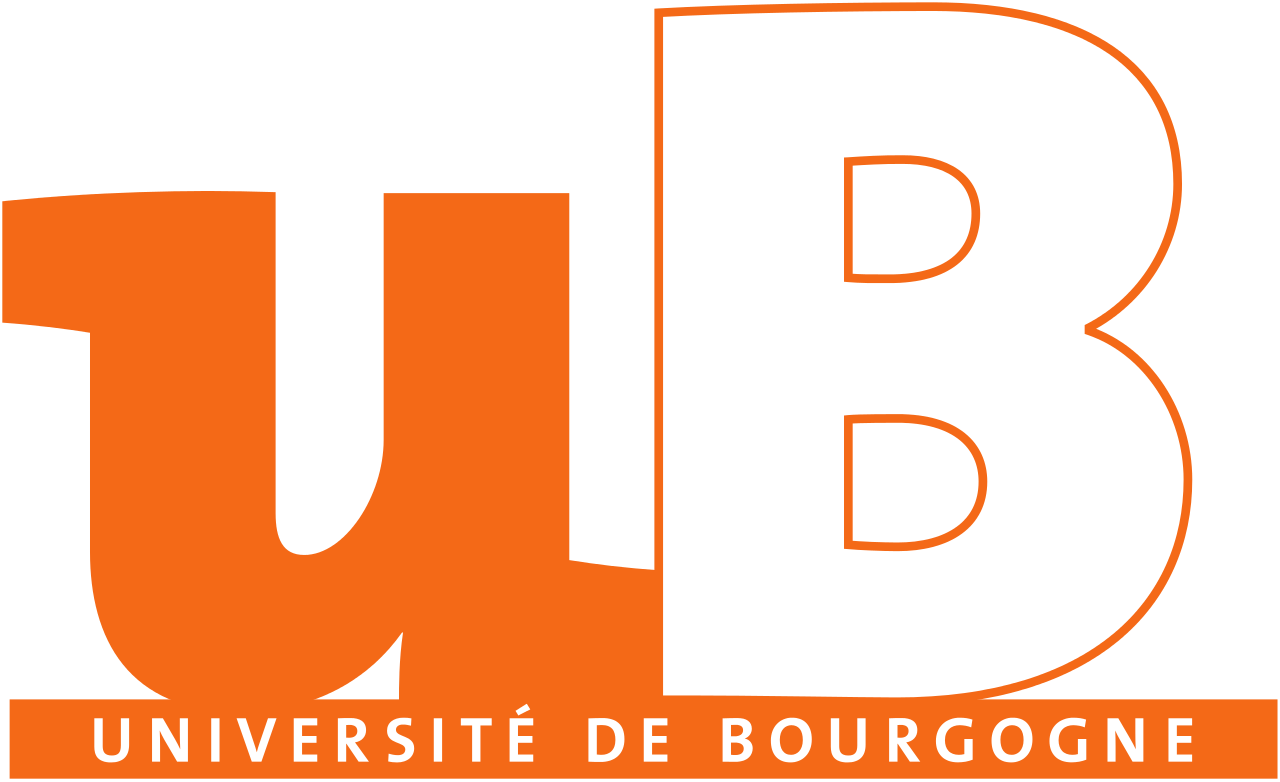}}
	}
	\end{picture}
	\vspace{\stretch{2}}
	\begin{center}
	{%
	\usekomafont{title}
	\large
	\MakeUppercase{Thèse de doctorat de l'établissement
	\mbox{Université Bourgogne Franche-Comté}}
	\\[1ex]
	\MakeUppercase{préparée à l'université de Bourgogne}
	}
	\\[\stretch{1}]
	\mdseries
	École doctorale n\textdegree 553\\
	Carnot-Pasteur
	\\[\stretch{1}]
	Doctorat de mathématiques
	\\[\stretch{1}]
	Par
	\\[1ex]
	{\usekomafont{author} \large Giridhar Vidyadhar \thinspace\MakeUppercase{Kulkarni}}
	\\[\stretch{1.5}]
	{\usekomafont{title} \Large \mbox{Étude asymptotique} des \mbox{facteurs de forme} des \mbox{chaînes de spin quantiques}}
	\\[1ex]
	\rule[1ex]{.75\textwidth}{.6pt}\par%
{%
\begin{otherlanguage}{british}
	{\usekomafont{disposition} \LARGE \mbox{Asymptotic analysis} of \mbox{the form-factors} of \mbox{quantum spin chains}}
\end{otherlanguage}
}
\rule{.75\textwidth}{.4pt}\par%
\vspace{2pt}%
\vspace{-\baselineskip}
\rule{.75\textwidth}{1pt}\par%
\end{center}
\vspace{\stretch{1}}
Thèse présentée et soutenue publiquement à Dijon le \DTMusedate{defencedate} devant le jury composé de:
\\[2ex]
\begin{tabularx}{\textwidth}{>{\hsize=.35\hsize\linewidth=\hsize\raggedright\arraybackslash}X>{\hsize=0.45\hsize\linewidth=\hsize\itshape\arraybackslash}X@{\enspace}>{\hsize=.2\hsize\linewidth=\hsize\raggedright\arraybackslash}X}
Olalla \medspace\MakeUppercase{Castro Alvaredo}
&
City University of London, Royaume-uni
&
Rapporteuse
\\
Robert \medspace\MakeUppercase{Weston}
&
Heriot-Watt University, Royaume-uni
&
Rapporteur
\\
Karol K. \medspace\MakeUppercase{Kozlowski}
&
CNRS, ENS de Lyon, France
&
Examinateur
\\
Dominique \medspace\MakeUppercase{Sugny}
&
Université de Bourgogne, France
&
Président
\\
Sébastien \medspace\MakeUppercase{Leurent}
&
Université de Bourgogne, France
&
Examinateur
\\
Nikolaï \medspace\MakeUppercase{Kitanine}
&
Université de Bourgogne, France
&
Directeur de thèse
\end{tabularx}
\end{titlepage}
\clearpage
\thispagestyle{empty}
\vspace*{\stretch{1}}
\begin{otherlanguage}{british}
\begin{center}
\usekomafont{title}
\Large
Note from the author
\end{center}
This Thesis is a synthesis of academic work done in collaboration with N. Kitanine. While a part of the original work presented in this Thesis has already been published elsewhere, the remaining part will appear in other joint publications.
\vspace{2\baselineskip}
\begin{flushright}
    G. V. Kulkarni
    \\
    \DTMusedate{releasedate} 
\end{flushright}
\end{otherlanguage}
\vspace{\stretch{2}}
\hypersetup{pageanchor=true}
\normalfont
\normalsize
\frontmatter
\thispagestyle{plain}
\vfill
\begin{flushleft}
\usekomafont{chapter}
\noindent
{\Huge R}{\huge ésumé}
\\
\rule[.85em]{.5\textwidth}{2pt}
\end{flushleft}
Les systèmes intégrables quantiques restaient longtemps un domaine où des méthodes mathématiques modernes permettaient d’accéder aux résultats intéressants pour l‘étude de systèmes physiques.
Le calcul exacte, numérique et asymptotique de fonction de corrélation reste un de sujets les plus importants de la théorie de modèles intégrables quantiques. 
Dans ce cadre l’approche basée sur le calcul des facteurs de forme s’est révélée la plus efficace.
Dans ce thèse, une méthode alternative fondée sur l'ansatz de Bethe algébrique est développée pour calculer des facteurs de formes dans la limite thermodynamique.
Elle est appliqué et décrit dans le contexte de chaîne de spin isotrope XXX, qui est un des cas plus intéressant des modèles critiques où la zone de Fermi est non-compacte.
Dans le cas particulière des facteurs de formes à deux-spinons, on obtient un résultat exact en forme close qui est comparable à celui-ci obtenu initialement dans le formalisme de l'algèbre des opérateurs de $q$-vertex.
Cette méthode est aussi généralisée au calcul des facteurs de formes dans les secteurs de spinons plus hauts, donnant une représentation en déterminants réduits, dont une structure de haut-niveau à l'échelle des facteurs de formes est révélée.
\vspace*{\stretch{4}}
{%
\vfill
\begin{otherlanguage*}{british}
\begin{flushleft}
\usekomafont{chapter}
{\Huge A}{\huge bstract}
\\
\rule[.85em]{.5\textwidth}{2pt}
\end{flushleft}
Since a long-time, the quantum integrable systems have remained an area where modern mathematical methods have given an access to interesting results in the study of physical systems.
The exact computations, both numerical and asymptotic, of the correlation function is one of the most important subject of the theory of the quantum integrable models.
In this context an approach based on the calculation of form factors has been proved to be a more effective one.
In this thesis, we develop a new method based on the algebraic Bethe ansatz for the computation of the form-factors in thermodynamic limit.
It is both applied to and described in the context of isotropic XXX Heisenberg chain, which is one of the examples of an interesting case of critical models where the Fermi-zone is non-compact.
In a particular case of two-spinon form-factors, we obtain an exact result in a closed-form which matches the previous result obtained from an approach based on $q$-vertex operator algebra.
This method is then generalised to form-factors in higher spinon sectors where we find a reduced determinant representation for the form-factors, in which a higher-level structure for the form-factors is revealed.
 \end{otherlanguage*}
}
\vfill
\cleardoubleplainpage
\clearpage{}%
\chapter*{Remerciements}
\selectlanguage{french}
\begingroup
\setlength{\parskip}{.5em}
\setlength{\parindent}{2ex}
Tout d’abord, je tiens à remercier Nikolaï \textsc{Kitanine}, mon directeur de thèse pour m'avoir proposé ce sujet de thèse, et aussi pour tous les conseils, explications et suggestions très pertinents qu'il m'a offert, avec une énorme patience et intérêt dans mes travaux.
Je veux lui aussi remercier pour ses commentaires et les corrections dans ce manuscrit au cours de sa rédaction.
\par
Je remercie
Olalla \textsc{Castro Alvaredo} %
et
Robert \textsc{Weston} %
pour leurs rapports sur ce manuscrit. C'est grace aux leurs commentaires très importants et aussi pertinents que je puisse vous presenter cette version améliorée. 
Je remercie
Karol \textsc{Kozlowski} %
et
Dominique \textsc{Sugny} %
pour avoir faire partie du comité de suivi de cette thèse ainsi que pour la jury de soutenance.
Je remercie également Sébastien \textsc{Leurent} pour l'avoir faire partie du jury. 
\par
Cette thèse a été financé par l'École Doctorale (ED) de Carnot-Pasteur pour les trois premières années depuis 2016 et par l'Université de Bourgogne (UB) en contrat d'ATER pour la dernière année.
Je remercie
Hans \textsc{Jauslin}, %
le directeur de l'ED;
Johannes \textsc{Nagel}, %
le directeur du département de mathématiques;
Abderrahim \textsc{Jourani} %
et
Lucy \textsc{Moser-Jauslin}, %
le directeur et l'ancienne directrice adjointe
du laboratoire de l'Institut de Mathématiques de Bourgogne (IMB) à Dijon, où cette thèse a été préparée pendant les quatre dernières années.
Je remercie également tous les membres de laboratoire. En particulière, je veux remercier
Daniele, %
Guido, %
Samuel H., %
Michele, %
Hervé, %
Ricardo %
avec qui j'ai eu l'occasion d'enseigner les maths en tant qu'un ATER.
Je remercie en aussi l'équipe math physique :
Kolya (Nikolaï), %
Sébastien, %
José-Luis, %
Giuseppe, %
Christian, %
Peter, %
Guido, %
Simona,
Taro; %
j'ai énormément profité par les séminaires et les groupes de travail organisés par l'équipe.
Je remercie les doctorants, y compris les anciens :
Michaël, %
Ajinkya, %
Olga, %
Mario, %
Rémi, %
Mireille, %
Cindy, %
les deux Nicolas %
et les deux Quentins, %
Sion, %
Lamis, %
Omar, %
et,
Helal. %
J'ai aussi beaucoup profité des séminaires des doctorants organisés par l'association des doctorants de mathématiques à Dijon D.M.D, duquel j'ai eu opportunité d'organiser, avec Mario et Nicolas en 2017-18.
\par
J'ai eu l'occasion de participer, grâce au financement de l'ED et celui de ces organisateurs, à un école d'été en 2018 aux Houches et un école d'hivers en 2017 au GGI à Florence. Je remercie les organisateurs, dont Kolya, J.-S. \textsc{Caux}, A. \textsc{Klumper} à l'école de physique des Houches; P. \textsc{Calabrese} et Prof. \textsc{Mussardo} à \textit{Galileo Galilei Institute}, pour ma'avoir donné une occasion d'y participer; et ainsi que les orateurs \textsc{Essler}, J.M. \textsc{Maillet}, \textsc{Göhmann}, \textsc{Slavnov}, \textsc{Doyon}, \textsc{Prosen}, \textsc{Sirker} pour leurs cours éclairants sur les thématiques de recherche.
\par
Je remercie les membres du SAFT de l'IMB : 
Magali, %
Caroline, %
Anissa, %
Nadia, %
Francis, %
Sébastien, %
Noémie %
pour leurs aides dans les démarches administratives et informatiques. Je remercie
Emeline \textsc{Iltis} %
au secrétariat de l'ED pour son aide dans les nombreuses démarches administratives et financières au seins de l'ED;
Mylène \textsc{Mongin} %
au secrétariat du département de maths, ainsi que
Nawal \textsc{Mounir}
et
Fabien \textsc{George} %
du pôle international de l'UB.
\par
Finalement, je remercie ma famille en Inde : ma sœur Prajakta, mes parents et mes grandes parents qui m'ont toujours soutenu dans la vie.
Je remercie Thibault, mon compagnon, tout en lui excusant pour ce qu'il a dû subir pendant une période longue de rédaction de ce manuscrit. Je remercie également sa famille, qui m'a toujours accueilli si chaleureusement dans ses milieux, voir comme ma deuxième famille en France. Je remercie les voisins en Or que j'ai eu la chance d'avoir : Bruno, Franck, Jérémy et Cécile, pendant les trois derniers ans de mon séjour à Dijon.
\endgroup
\clearpage{}%
\clearpage{}%
\chapter*{Acknowledgements}
 \begin{otherlanguage*}{british}
\begingroup
\setlength{\parskip}{.5em}
\setlength{\parindent}{2ex}
First of all, I would like thank my thesis advisor Nikolai Kitanine for giving me an opportunity to work on this project as well as for all his extremely useful advises, suggestions and explanations that he always gave me with an enormous amount patience and keen interest interest in my work.
I also want to thank him for the numerous corrections and suggestions on this manuscript during its preparation.
\par
I want to thank
Olalla {Castro Alvaredo} %
and
Robert {Weston} %
for their reports on my thesis. It is thanks to their comments and suggestions that I can present you this improved version today.
I would like to thank
Karol {Kozlowski} %
and
Dominique {Sugny} %
for having been part of my thesis monitoring committee as well as the defnece jury for this thesis. 
I want to also thank Sébastien {Leurent} for being a part of this jury. 
\par
This thesis was sponsored by Carnot-Pasteur doctoral school (ED) for its initial three years since 2016 and then by University of Burgundy (UB) with an ATER contract for the last year.
I thank the director of ED 
Hans {Jauslin}, %
the director of mathematics department
Johannes {Nagel}, %
as well as
the director and ex co-director of the IMB Dijon:
Abderrahim {Jourani} %
and
Lucy {Moser-Jauslin}, %
where this thesis was prepared and written.
First I would like thank all the members of IMB. Particularly, many thanks to those with whom I had an opportunity to teach maths as a part of ATER program:
Daniele, %
Guido, %
Samuel H., %
Michele, %
Hervé, %
Ricardo. %
I want to thank the entire math physics group:
Kolya (Nikolai), %
Sébastien, %
José-Luis, %
Giuseppe, %
Christian, %
Peter, %
Guido, %
Simona,
Taro; %
I have benifited quite a lot thanks to the regularly organised seminars of the group.
Next, I want thank all PhD students, including those who left the IMB before me:
Michaël, %
Ajinkya, %
Olga, %
Mario, %
Rémi, %
Mireille, %
Cindy, %
both Nicolas %
and both Quentins, %
Sion, %
Lamis, %
Omar, %
and,
Helal. %
I also benifitted from the doctoral students seminars organised by DMD, of which I also an opportunity to organise with Mario and Nicolas from 2017 to 2018.
\par
Thanks to the funding by IMB and ED, I also had an opportunity to attend workshops and conferences, among which I attended a summer school at Les Houches in 2018 and a winter school in Florence in 2017.
I thank the organisers, among which Kolya, J.-S., Klumper for Les Houches and Calabrese, Prof. Mussardo at GGI Florence.
Many thanks to the lecturers in these two venues: {Essler}, J.M. {Maillet}, {Göhmann}, {Slavnov}, {Doyon}, {Prosen}, {Sirker} for their enlightening lectures focused on current topics of research in the field.
\par
I thank administrative staff of IMB: 
Magali, %
Caroline, %
Anissa, %
Nadia, %
Francis, %
Sébastien, %
Noémie.
I also thank Emeline {Iltis} %
at the secretariat of ED for her help in many administrative and financial affairs related with doctoral school.
I thank Mylène {Mongin} %
at the secretariat of maths department, as well as Nawal {Mounir}
and
Fabien {George} %
from the international office of UB.
\par
Finally I want to thank my extended family back in India: my sister Prajakta, my parents and grand parents who always supported me unconditionally.
I thank my partner Thibault, while also begging a pardon for all the trouble he had to put up with during a long process of writing this thesis.
I want to also thank his family for always welcoming me and making me feel at home.
Finally, I want to thank all my neighbours in Dijon: Bruno, Cécile, Jérémy and Franck, who went out of their way to make my stay in Dijon more comfortable.
\endgroup
 \end{otherlanguage*}
\clearpage{}%

\listofstoc
\clearpage{}%
\addchap[tocentry={Récapitulatif de la thèse (in French)},head={Récapitulatif de la thèse}]{Récapitulatif de la thèse}
La méthode des facteurs de forme est un outil très puissant \cite[voir][]{Smi92} pour étudier la dynamique des systèmes intégrables quantiques.
La théorie des systèmes intégrables donne une possibilité pour réaliser les calculs exacts des facteurs de forme. Les moyens utilisés dans ce contexte entrent dans le cadre des mathématiques modernes.
Un exemple fondamental de ces systèmes est celui d'une chaîne de spin XXX qui est déterminée par l’Hamiltonien ci-dessous.
\begin{align}
	H_{XXX} = \sum_{a=1}^{M}
	\left\lbrace
	\sigma^{1}_{n}\sigma^{1}_{n+1}+\sigma^{2}_{n}\sigma^{2}_{n+1}+\sigma^{3}_{n}\sigma^{3}_{n+1}-1
	\right\rbrace
	.
	\label{XXX_H}
\end{align}
C'est un opérateur sur l'espace de Hilbert (ou l'espace quantique) qui se décompose en produit tensoriel $V_{q}={\mathbb{C}^{2}}^{\otimes M}$.
Soit la longeur de chaîne $M$ un entier pair.
De plus, nous allons y imposer la condition de bords périodique $\sigma^\alpha_{M+1}=\sigma^\alpha_M,~\forall \alpha\in\set{1,2,3}$.
\par
La résolution du problème spectral de ce système a été étudiée par \textcite{Bet31}. Il a montré que les paramètres spectraux $\bm\la\in
\Cset$ décrivant les vecteurs propres satisfont un système d'équations transcendantales. Une variation algébrique de cette approche a été développée par \textcite{FadST79}, surnommée \og Ansatz de Bethe Algébrique\fg{} (ABA). 
Pour brièvement parler de cette approche, considérons un vecteur (ou son dual) d'une forme
\begin{subequations}
\begin{align}
	\ket{\psi(\bm\la)}&=\prod_{a=1}^{N}\opB(\la_{a})\pvac,
	&
	\bra{\psi(\bm\la)}&=\pvac*\prod_{a=1}^{N}\opC(\la_{a})
	\label{ABA_vec_fr}
	;
\end{align}
où $\pvac$ est un vecteur de référence, les opérateurs $\opB$ et $\opC$ sont des opérateurs de l'espace quantique $V_{q}$ et, $N$ est la cardinalité de l'ensemble $\bm\la$.
\end{subequations}
On peut montrer qu'un tel vecteur (ou son dual) est un vecteur propre non-trivial du Hamiltonien $H$ si les paramètres spectraux $\bm\la$ correspondent aux solutions admissibles d'un système d'équations
\begin{subequations}
\begin{flalign}
	(\forall a\leq N)
	&&
	\aux(\la_{a})+1&=0
	.
	&&
	\label{BAE_fr}
\end{flalign}
La fonction auxiliaire $\aux(\la)$ qu'on trouve dans l'\cref{BAE_fr} est donnée par
\begin{align}
	\aux(\la)=\left(\frac{\la-\frac{i}{2}}{\la+\frac{i}{2}}\right)^M\prod_{a=1}^{N}\frac{\la-\la_{a}+i}{\la-\la_{a}-i}.
\end{align}
\end{subequations}
Ces équations ainsi obtenues par l'approche de \cite{FadST79} sont équivalentes à celles de \textcite{Bet31}. Elles sont appelées les équations de Bethe et leurs solutions sont appelées les racines de Bethe. De la même manière, un vecteur de la forme \eqref{ABA_vec_fr} dont les paramètres spectraux satisfont \cref{BAE_fr} est appelé un vecteur de Bethe \textit{on-shell}, par opposition au vecteur de Bethe \textit{off-shell}.
\par
Ici, nous nous intéresserons au calcul des facteurs de forme qui se présentent notamment dans le développement des fonctions de corrélations dynamiques en deux points, tel que l'exemple ci-dessous. Notons que nous utiliserons ici la représentation de Heisenberg où les opérateurs évoluent avec le temps alors que les états sont constants. 
\begin{align}
	\braket{\sigma^{3}_{1}(0)\sigma^{3}_{n+1}(t)}=\sum_{\text{exc}}e^{-i(E_{\text{exc}}-E_{\text{vide}})t}e^{-i(p_{\text{exc}}-p_{\text{vide}})n}\left|\FF_{\text{exc}}\right|^2.
\end{align}
La somme des facteurs de forme de ce développement est prise sur tous les états excités (vecteurs propres) de l'Hamiltonien.
Le produit des facteurs de forme $|\FF_{\text{exc}}|^2$ s'écrit alors
\begin{align}
	\left|\FF_{\text{exc}}\right|^2=
	\frac{%
	\braket{%
	\psi_{\text{vide}}
	|\sigma^{3}_{n}|
	\psi_{\text{exc}}
	}
	\braket{%
	\psi_{\text{exc}}
	|\sigma^{3}_{n}|
	\psi_\text{vide}%
	}
	}
	{%
	\braket{\psi_{\text{vide}}|\psi_{\text{vide}}}
	\braket{\psi_{\text{exc}}|\psi_{\text{exc}}}
	}
	.
\end{align}
Avec les états excités, apparaissent des particules de spinons. Dans le cadre de l'ansatz de Bethe algébrique, ces spinons sont générés en ajoutant des trous dans la distribution des racines pour l'état vide, ce que l'on décrit par l'ensemble $\bm\hle$ des paramètres de la cardinalité $n_h$. On trouve que le nombre des spinons $n_h$ est toujours un pair et par conséquences elles apparaissent toujours en couple.
Les états liés de spinons sont décrits dans le cadre de l'ansatz de Bethe algébrique avec des racines de Bethe complexes non-réelles.
Nous allons utiliser ici la description de \textcite{DesL82} pour ces dernières, selon laquelle elles sont classifiées soit dans une \emph{paire étroite} (\textit{close-pair}), soit dans une \emph{paire étendue} (\textit{wide-pair}). 
Les paires étroites forment une des deux dispositions suivantes dans la limite thermodynamique $N,M\to\infty$, $N\sim \frac{1}{2}M$ : une \emph{corde} de longueur deux (2-\textit{string}) ou un \emph{quartet} qui consiste en quatre racines de Bethe.
Les trois types de racines complexes: 2-\textit{string}, quartet et \textit{wide-pair} sont décrites par un ensemble des racines du haut niveau $\bm\cid$ de la cardinalité $\ho{n}$. Elles satisfont un système d'équations de Bethe non-homogènes, que l'on appelle équation de Bethe du haut niveau, s'écrivant : 
\begin{subequations}
\begin{flalign}
	(a\leq \ho{n})
	&&
	\aux*(\cid_a)+1&=0
	.
	&&
	\label{hl_bae_fr}
\end{flalign}
La fonction $\aux*$ dans l'\cref{hl_bae_fr} étant donnée par :
\begin{align}
	\aux*(\nu)&=
	\prod_{a=1}^{n_h}
	\frac{%
	\nu-\hle_a-\frac{i}{2}
	}{%
	\nu-\hle_a+\frac{i}{2}
	}
	\prod_{a=1}^{\ho{n}}
	\frac{%
	\nu-\cid_a+i
	}{%
	\nu-\cid_a-i
	}
	.
\end{align}
\end{subequations}
\par
À ce stade, on se retrouve face au problème de diffusion inverse quantique, c.à.d. le problème de déterminer l'action d'un opérateur de spin $\sigma^{3}_{n}$ 
sur un état \emph{on-shell} défini par \cref{ABA_vec_fr} et \cref{BAE_fr}. 
Ce problème a été résolu par \textcite{KitMT99}, grâce auquel on obtient des quotients de produits scalaires, où au moins un vecteur est 
\emph{on-shell}. Pour un tel genre du produit scalaire; \textcite{Sla89,Gau83,Kor82} ont trouvé des représentations déterminants parmi lesquelles on a ces deux cas différents :
\begin{itemize}[label=---, wide=0pt, noitemsep]
\item la représentation en déterminant des produit scalaire en déterminant de Slavnov dont au moins un des deux états de produit scalaire est \emph{on-shell}. On note la matrice de Slavnov par lettre $\Mcal$.
\item la représentation de la norme au carré d'un état de Bethe \textit{on-shell} en déterminant de Gaudin. On note la matrice de Gaudin par lettre $\Ncal$.
\end{itemize}
Ces deux représentations, avec la résolution du problème de diffusion inverse, nous permettent \cite{KitMT99} d'écrire une représentation des facteurs de forme en déterminants pour les chaînes de spin d'une longueur finie.
Nous avons ici une représentation des facteurs de forme finis en déterminants, qui s'écrit comme un quotient des déterminants de Slavnov et Gaudin.
À partir de là, la méthode qu'on propose aussi se poursuite aux études asymptotiques des facteurs de forme dans la limite thermodynamique, où la  longueur de chaîne $M\to\infty$.
Cette méthode est divisée en trois étapes :
\begin{enumerate}[wide=0pt, %
label=\protect\halfboxitem{\textbf{étape \arabic*}},ref={étape \arabic*}]
\item \textbf{Extraction de la matrice de Gaudin:}
\label{gau_ex_etape_fr}
dans cette première étape, nous extrayons\footnotemark\space la matrice de Gaudin en prenant l'action de son inverse sur la matrice de Slavnov ou une version \textcite{FodW12a} qui généralise la matrice de Slavnov.
\footnotetext{%
\label{foot:extn_fr}
Ce que l'on entend ici par l'\og extraction d'une matrice\fg{} est une action de sa matrice inverse, ou une matrice équivalente à celle dernière au déterminant près.}
\begin{align}
	\Fmat&=
	\Ncal^{-1}\Mcal,
	&&
	\det\Fcal=
	\frac{\det\Mcal}{\det\Ncal}
	.
\end{align}
Après cette étape, nous obtenons une représentation en termes de déterminants des matrices de Cauchy modifiées.
On évoque souvent la propriété de condensation des racines de Bethe dans les calculs thermodynamiques, qui permet d'écrire les sommes en tant qu'intégrales. La mesure de celles-ci est la fonction de densité des racines qui satisfait par conséquence une équation intégrale.
Dans certains cas qu'on précisera dans les calculs, une extension de cette propriété pour les fonctions méromorphes sera utilisé, c'est ce qu'on appellera la propriété de condensation généralisée.
\\
Un résultat surprenant que nous avons obtenu, montre l'émergence de la matrice de Gaudin du \emph{haut niveau} et son extraction:
\begin{align}
	\ho{\Scal}=
	\Ho{\Ncal}^{-1}
	\Ho{\Tcal}
	.
\end{align}
Cette structure du haut niveau se trouve dans un bloc des colonnes de la matrice de Cauchy modifiée provenant des racines complexes.
L’émergence ici de cette structure du \emph{niveau supérieur} pour les facteurs de forme est comparable avec celle de l'équation de Bethe du haut niveau \eqref{hl_bae_fr} obtenue dans \cite{DesL82} pour le spectre.
\item \textbf{Extraction de la matrice de Cauchy(-Vandermonde):}
\label{cau_ex_etape_fr}
dans cette étape, nous extrayons\footref{foot:extn_fr} la plus large matrice de Cauchy contenue dans la représentation obtenue dans l'étape \hyperref[gau_ex_etape_fr]{précédente}.
Malheureusement la matrice de Cauchy plus large est souvent rectangulaire.
Dans ce cas, la matrice mixte de Cauchy-Vandermonde est un recours pour avancer l'extraction puisqu'elle généralise l'identité du déterminant de Cauchy aux cas rectangulaires.
\begin{align}
	\det\Cmat<\ptn\gamma>[\bm\alpha\Vert\bm\beta]
	=
	\frac{%
	\prod_{j>k}^{m+n}\sinh\pi(\alpha_j-\alpha_k)
	\prod_{j<k}^{m}\sinh\pi(\beta_j-\beta_k)
	}{%
	\prod_{j=1}^{m+n}
	\prod_{k=1}^{m}
	\sinh\pi(\alpha_j-\beta_k)
	}
	.
\end{align}
Ceci nous montre aussi l’intérêt d'extraire une matrice au sens plus général. Nous allons plus loin et extrayons une matrice duale de Cauchy-Vandermonde grâce à sa dualité. Cette dualité nous permet de remplacer le bloc Vandermonde par un bloc équivalent, composé de polynômes supersymétriques élémentaires.
\item \textbf{Calcul des déterminants de Cauchy infinis dans la limite thermodynamique :}
\label{pref_tdl_etape_fr}
dans cette dernière étape, nous calculons les déterminants de Cauchy, avec les préfacteurs dans la limite thermodynamique. Comme $M,N\to\infty$ dans cette limite, la matrice de Cauchy (ou plus généralement la matrice de Cauchy-Vandermonde) devient une matrice infinie.
Pour calculer cette limite, nous exprimons d’abord les déterminants de Cauchy et les préfacteurs comme un produit infini des fonctions auxiliaires.
On peut calculer la limite thermodynamique de ces fonctions auxiliaires avec la méthode de condensation.
Après la substitution de cette limite pour les fonctions auxiliaires, on obtient un nouveau produit infini. En comparant celui-ci avec la forme Weierstrass de fonction de Barnes-G, on obtient le résultat final.
\end{enumerate}
Dans le cas particulier des facteurs de forme pour les états excités à deux spinons tous les racines de Bethe sont réelles et les racines non-réelles s'absentent.
Par conséquence, on trouve que la modification dans la matrice de Cauchy que l'on obtient après l'\ref{gau_ex_etape_fr} est minimale, de même que le calcul de l'extraction de Cauchy dans l'\ref{cau_ex_etape_fr} donne un simple résultat. Ainsi on a obtenu dans \cite{KitK19} la forme exact close des facteurs de formes thermodynamiques à deux spinons :
\begin{align}
	\left|\FF^{z}(\hle_1,\hle_2)\right|^2
	&=
	\frac{2}{M^2 G^4\left(\frac{1}{2}\right)}
	\prod_{\sigma=\pm}
	\frac{%
	G(\frac{\hle_{2}-\hle_{1}}{2i\sigma})
	G(1+\frac{\hle_{2}-\hle_{1}}{2i\sigma})
	}{%
	G(\frac{1}{2}+\frac{\hle_{2}-\hle_{1}}{2i\sigma})
	G(\frac{3}{2}+\frac{\hle_{2}-\hle_{1}}{2i\sigma})
	}
	.
	\label{2sp_ff_result_fr}
\end{align}
Ce résultat a été comparé avec celui-ci obtenu dans \cite{BouCK96,BouKM98}, ces dernières sont obtenus en utilisant une méthode fondé sur l'algèbre des opérateurs de $q$-vertex \cite{JimM95}.
Car l'approche de l'ABA a une domaine d'applicabilité plus large, on voit immédiate l’aspect intéressant derrière notre méthode.
Ce comparaison entre les résultat obtenus dans les deux cadre différent est également important pour justifier les hypothèses utilisées dans nos calculs, notamment la propriété de condensation généralisée.
\par
Nous avons calculé aussi les facteurs de forme dans le cas plus général des états excités liés qui contiennent nécessairement les racine non-réelles. On obtient dans ce cas une représentation des facteurs de formes en déterminant \emph{réduit}, que s'écrivant :
\begin{multline}
	\left|\FF^{z}(\set{\hle_a}_{a=1}^{n_h})\right|^2=
	(-1)^{\frac{n_h+2}{2}}
	M^{-n_h}
	2^{\frac{n_h(n_h-2)+2}{2}}	
	\pi^{\frac{n_h(n_h-3)+2}{2}}	
	\frac{\prod_{a=1}^{\ho{n}}\prod_{b=1}^{n_h}(\cid_a-\hle_b-\frac{i}{2})}{\prod_{a,b=1}^{\ho{n}}(\cid_a-\cid_b-i)}
	\\
	\times
	\frac{1}{G^{2n_h}(\frac{1}{2})}
	\prod_{\underset{a\neq b}{a,b=1}}^{n_h}
	\frac{%
	G(\frac{\hle_a-\hle_b}{2i})
	G(1+\frac{\hle_a-\hle_b}{2i})
	}{%
	G(\frac{1}{2}+\frac{\hle_a-\hle_b}{2i})
	G(\frac{3}{2}+\frac{\hle_a-\hle_b}{2i})
	}
	~
	\frac{%
	\det_{\ho{n}}\resmat*[g]
	\det_{n_h}\resmat*[e]
	}{\det\vmat[\bm\hle]}
	.
	\label{red_det_rep_generic_fr}
\end{multline}
Ce que l'on entend ici dans la phrase \og représentation des facteurs de forme en déterminant réduit \fg{} c'est que les matrices $\resmat*[g/e]$ sont finies, ce qui est vrai car on ne s’intéresse qu'aux calculs des déterminants pour les états excités plus proches (\textit{low-lying}) de l’état vide $n_h,\ho{n}<<N$. Malheureusement ce résultat de l'\cref{red_det_rep_generic_fr} n'est pas écrit sous une forme close, car les composants des matrices résiduelles $\resmat*[g/e]$ restent toujours sous une forme d'intégrales contenants des fonctions auxiliaires.
Pourtant il est bien important car il montre qu'on peut obtenir, dans le cadre de ABA, une représentation des facteurs de formes des états liés en déterminants d'une taille finis.
Cependant, on peut toujours simplifier cette expression. Dans le cas des facteurs de forme à quatre spinons, nous montrons que l'on peut exprimer le résultat sous la forme suivante, en éliminant tous les déterminants résiduels.
\begin{align}
 	\left|\FF^z(\hle_1,\hle_2,\hle_3,\hle_4)\right|^2
 	=
	-
	\frac{32\pi^3}{M^{4} G^8(\frac{1}{2})}
	\sum_{a\neq b}
	\frac{%
	G(\frac{\hle_a-\hle_b}{2i})
	G(1+\frac{\hle_a-\hle_b}{2i})
	}{%
	G(\frac{1}{2}+\frac{\hle_a-\hle_b}{2i})
	G(\frac{3}{2}+\frac{\hle_a-\hle_b}{2i})
	}
	\frac{%
	\Jcal_g
	\Jcal_e
	}{%
	\sum_{a=1}^{n_h} \ho{\rden}(\clp-\hle_a)
	}
	.
	\label{4sp_ff_result_fr}
\end{align}
Finalement, remarquons que l'avantage principal derrière notre méthode est qu'elle fondé sur ABA. Ça lui donne une applicabilité dans le contexte du modèle plus général d'une chaîne de spin anisotrope dite XXZ.
On peut aussi estimer qu'on puisse généraliser notre méthode aux autres modèles intégrables similaires à condition que la structure que donne la matrice de Cauchy (ou une structure équivalente) reste préservée.
\begin{flushright}
Le reste de ce manuscrit est rédigé en anglais.
\end{flushright}
\clearpage{}%
\selectlanguage{british}
\renewcommand*{\contentsname}{Table of Contents}
\pdfbookmark{Table of contents}{toc}
\tableofcontents
\listoffigures
\listoftables
\mainmatter
\clearpage{}%
\setchapterpreamble[or][0.4\textwidth]{\dictum[Werner Heisenberg]{``what we observe is not nature in itself, but nature exposed to our method of questioning.''}\hfill\footnotesize\citetitle{Hei07}}
\addchap[Introduction]{Introduction}
A quantum spin chain is a prototype of an interacting many-body quantum system whose origins can be found in the initial attempts of W. Heisenberg to demonstrate the quantum origin of the magnetism \cite{Hei28}.
It came soon after the failure of the one dimensional Ising model \cite{Isi25} to demonstrate the phase transition to disordered phase that Heisenberg proposed his model, albeit it is important to note that \textcite{Ons44} did manage to show the phase transition in two-dimensional Ising model.
\par
The Heisenberg model is a lattice of spin-1/2 particles where each of them interact with its nearest neighbours.
In its most general form, one can write the totally anisotropic quantum Hamiltonian with the coupling to an external field
\begin{align*}
	H_{XYZ,h}
	=
	\sum_{<j,k>}
	\sum_{a=1}^{3}
	J_a \sigma_j^a\sigma_k^a
	+
	h\sum_{j}\sigma^3_j
	.
\end{align*}
In this expression $\sigma^a_j$ denote the local spin operators which act as a Pauli matrices in the subspace for the $j$\textsuperscript{th} lattice site of the total Hilbert space.
Here we consider only the one-dimensional Heisenberg chains of length $M$ with periodicity condition $\sigma^a_{M+1}=\sigma^a_{M}$.
We also distinguish the model with isotropic coupling $J_1=J_2=J_3=J$ which is called the \emph{XXX model}, the model with longitudinal anisotropy $J_1=J_2=J$ and $J_3=\Delta J$ called \emph{XXZ model}.
\par
In 1931, Hans Bethe \cite{Bet31} realised that the isotropic Heisenberg model can be solved analytically to obtain exact wavefunctions and their eigenvalues.
This method came to be known as the \emph{coordinate Bethe ansatz}.
It tells us that the Bethe wavefunction is parametrised by a set of complex \emph{spectral parameters} which satisfy a system of coupled transcendental equations, known as the \emph{Bethe equations}.
Its roots are correspondingly called the \emph{Bethe roots}.
This method was extended to the other one dimensional quantum models opening a new paradigm of the integrable one-dimensional quantum models, notable examples of which include the one dimensional Bose gas (or the non-linear Schrodinger equation NLS model) as well as the anisotropic version of the Heisenberg spin-1/2 model denoted XXZ or the XYZ for the completely anisotropic case. 
\Textcite{LieL63} resolved the NLS model for the one-dimensional Bose gas whereas XXZ model was resolved by \textcite{Orb58} for values of the anisotropy parameter $\Delta$.
\\
Following Bethe's seminal work, the study of the ground state and its excitations became the primary focus of the investigation.
It was \textcite{Hul38} who first came up with a conjecture which determines the anti-ferromagnetic ground state of the Heisenberg model. He also made the assumption that ground state Bethe roots are distributed densely on the real line in the thermodynamic limit and gave the integral equation satisfied by the density function. This allowed him to compute the energy of the ground state in the thermodynamic limit.
It was thus known that the ground state of the XXZ model for $\Delta>-1$ is a disordered anti-ferromagnetic ground state which has a very non-trivial description in the Bethe ansatz.
These results were extended to the anisotropic anti-ferromagnetic XXZ model in \cite{Orb58} giving the integral equations for the density of the ground state roots for all values of anisotropy parameter in $\Delta>1$.
Some of the assumptions inherent in these computations, including the condition used to determine the ground state, were rigorously proved by Yang and Yang in \cite{YanY66,YanY66a}.
However, the fact that ground state roots condenses in the thermodynamic limit with the given density function was only recently proved by \textcite{Koz18}.
The excited states `near' the ground state are called \emph{low-lying excitations}.
These low-lying excitations for $\Delta>-1$ of the XXZ model were first studied by des \textcite{CloP62} walking along the footsteps of Hulthén.
However there it was misconstrued that the low-lying excitations are made of the spin-1 particles and thus the dispersion relation obtained was wrong.
This error was corrected by \textcite{FadT81} (see also \cite{FadT84}) as they showed that the low-lying excitations are made up of spin-1/2 particles called spinons which always comes in pairs.
In addition to it, the XXZ chain also contain the complex Bethe roots which represent physically the spinon bound states.
Since the original work of the Bethe himself, it has been widely believed that the complex roots can be arranged in the specific formations called \emph{strings} in the thermodynamic limit where $M\to\infty$.
This \emph{string hypothesis} although it is frequently used in the computation involving the complex roots, remains a contentious issue.
The analysis of \Textcite{DesL82,BabVV83} provides an alternate approach which do not make any \textit{a priori} assumptions of the string hypothesis, before passing to the thermodynamic limit.
In this picture all the complex roots are classified into the two categories called \emph{close-pairs} and \emph{wide-pairs} where the former are sub-divided into two types of special formations: \emph{2-strings} and \emph{quartets}.
The nomenclature `2-string' is borrowed from the traditional \emph{string picture} and it refers to the strings of length two. In the alternate picture proposed of the \emph{Destri-Lowenstein}, we find that strings of length higher than two do not appear in the low-lying spectrum. Combinatorially, their absence is compensated by the new type of formations called quartets and wide-pairs.
\par
The quantum integrable models in one dimension demonstrate uncanny similarities with the exactly solvable models of the two dimensional statistical physics. 
During his investigations into the six-vertex model \textcite{Lie67} found that wavefunction for its transfer matrix are same as the wavefunction obtained through the Bethe ansatz for the isotropic (XXX) Heisenberg model.
Subsequently \textcite{McCW68} made a similar observation for the anisotropic (XXZ) Heisenberg model and showed that the Hamiltonian commute with the transfer matrix whereas \textcite{Sut70} found a similar link between the totally anisotropic (XYZ) model and the transfer matrix for the eight-vertex model.
This correspondence between the quantum integrable chains and two dimensional exactly solvable lattice model is more profound. It was further explored by Baxter in \cite{Bax71,Bax72,Bax89} where he also outlines the algebraic nature of this link.
The same period also saw the radical changes in our understanding of the classical integrable systems.
The works of \textcite{Lax68,GarGKM67} on the Korteweg and de Vries (KdV) equation led to the realisation that non-linear problem of the KdV equation can be rephrased in terms of evolution problem of a linear operator, which surprisingly turns out to be the Schrödinger operator in the case of KdV equation.
This method of using the Lax operator was further developed by Zakharov, Faddeev and Shabat in \cite{ZakS72,ZakF72,ZakS74,ZakS79} and it came to be known as classical inverse scattering method.
These simultaneous developments in the classical integrable systems as well as exactly solvable lattice models played a quintessential role towards the development of the \emph{quantum inverse scattering method} by \textcite{FadST79}.
Consolidating on these previous developments, they firmly established the algebraic origin of the quantum integrability hence giving us the algebraic version of the Bethe ansatz.
There it was shown that the Lax operators in the quantum sense are given by a representation of the Baxter's $\Rm$ matrix satisfying so-called \emph{Yang-Baxter equation}.
The transfer matrix is defined as the trace of the product of Lax operators and it generates a family of commutating operators, among which one can find the Hamiltonian of the quantum integrable model.
This also meant that one can define in principle new integrable quantum models by looking for the solutions of the Yang-Baxter equation and the representations of the $\Rm$ matrices.
This line of reasoning sparked an interest in study of quantum groups \cite[and the refs. therein]{Jim90Book} (\cite{KulSk82,Dri85,Dri88,Jim85,FadRT90}) which developed into entire new domain of mathematics.
\par
In the algebraic formulation of the Bethe ansatz, we generate a Fock space starting from a (pseudo-)vacuum vector. This method gives a convenient description where both on-shell and off-shell vectors are obtained by action of lowering operator in this algebra. This simpler formulation opens the doors to the computations of the physically more meaningful quantities: the correlation functions, which we were able to compute only for a handful of special cases of the spin chains such as its free-fermion point \cite{LieSM61} and the Ising model \cite{McCTW77} and conformal field theories in the larger picture.
The correlation function can be either equal-time correlations for a system at equilibrium or the dynamic correlations for a system near or away from the equilibrium. Fourier transform of the two-point dynamic correlation function is called the \emph{dynamic structure factor}, it can be studied experimentally in the neutron scattering experiments \cite{MouEKCS13}.
The matrix elements of the local operators are called the \emph{form-factors} provides a powerful tool \cite{Smi92} to study the quantum integrable models, these are the central object of this thesis. 
These are related with the correlation functions and dynamic structure factor through the so-called \emph{form-factor expansion}.
\par
There are two main approaches for the computations of the correlation functions and form-factors of the quantum spin chains, both of which ultimately rely on the algebraic structure behind the integrability.
The first approach is based on the $q$-vertex operator algebra formalism put-forth by \textcite{JimM95} which relies on the affine ${U_q(\hat{\mathfrak{sl}_2})}$ symmetry of the infinite quantum spin chains which they exhibit directly in the thermodynamic limit.
\Textcite{JimMMN92} were able to obtain the multiple integral representation of the equal-time correlation functions for the massive anti-ferromagnetic XXZ model with anisotropy $\Delta>1$ using the approach based on the $q$-vertex operator algebra.
The form-factors can also be written in the multiple integral representation form \cite{JimM95} for the massive XXZ model. 
It is important to remark that this method works directly in the thermodynamic limit and for the massive regime of the XXZ spin chains only where $|q|\neq 1$.
However it is important to note that there is a way to get around this problem, which allowed \textcite{JimM96} to compute the correlation functions in the massless $-1<\Delta\leq 1$ regime of the XXZ chain.
One can also compute the form-factors in the isotropic limit $q\to 1, \Delta\to 1^+$ from the multiple integral representation \cite{JimM95} to obtain the form-factors of the isotropic (XXX) Heisenberg model. 
This isotropic limit for the two-spinon form-factors was computed by \textcite{BouCK96,BouKM98} whereas for the four-spinon form-factor it was computed by \textcite{AbaBS97,CauH06}.
\Textcite{CauKSW12} showed that one can obtain the form-factors of the massless anisotropic XXZ model $|\Delta|<1$ with a similar limit for the form-factors starting from the elliptic XYZ model.
\par
The second approach that we shall use here is based directly on the algebraic Bethe ansatz. Using the determinant formulae for the scalar products obtained by \textcite{Gau83,Kor82,Sla89} as well as through the resolution of the quantum inverse scattering problem by \textcite{KitMT99}, we can obtain the determinant representation for the form-factors of the finite length XXZ chain for all values of the $\Delta>-1$. This result was also obtained in \cite{KitMT99} and the same group went to compute in \cite{KitMT00} the multiple integral representation for the correlation functions in the thermodynamic limit which was found in agreement with the prior result \cite{JimMMN92,JimM96} from the $q$-vertex operator algebra method. 
This was extended in \cite{KitMST02a} to compute the correlation function for the XXZ model in the presence of an external magnetic field $h\neq 0$.
\par
The extraction of the long-distance asymptotic behaviour of the correlations from its multiple integral representation for the correlation function is a truly daunting task.
Remarkably, it is still possible to do so as shown by group of \textcite{KitKMST07}; \cite{KitKMST09,KitKMT14}.
It is also argued that the form-factor expansion comes out as very powerful method long-distance asymptotic behaviour of the two-point (dynamic) correlation functions \cite{KitKMST09,KitKMST11a,KitKMST12}.
There it was found that the large-distance asymptotic behaviour of the two-point function for the massless non-zero field XXZ chain is dominated by the form-factor of the umklapp excitation generated by adding particle hole pair at the Fermi-boundary.
\\
However the computations of form-factors in the thermodynamic limit from an algebraic Bethe ansatz based method had remained an unexplored territory until the very end of the first decade of the millennium, barring an important exception of the spontaneous magnetisation computed by \textcite{IzeKMT99}.
The renewed interest in the recent years has led to the computation of thermodynamic form-factors \cite{KitKMST09b,KitKMST11} for the massless XXZ chain in the presence of an external field $h\neq 0$ as well as by \textcite{DugGKS15} in the massive regime $\Delta>1$ of the XXZ model.
However, it is important to remark that these results are always represented in terms of the Fredholm determinants for which it is not yet known how it could be converted to the multiple integral form.
As a result, the results for the form-factors from the $q$-vertex operator algebra and algebraic Bethe ansatz based methods were never successfully compared.
It is also important to note that there is no Fermi-boundary for the zero-field case $h=0$ of the massless XXZ model as well as the XXX model since the Fermi distribution of its roots is non-compact.
As a result there is no question of particle-hole excitation at the boundary and the previous results of the non-zero external field cannot be simply extended to the zero-field case.
\\
One of the main objective behind the work presented in this thesis was to get these problems and build up a method to compute the thermodynamic form-factors which can allow us to compare the results with those obtained from the $q$-vertex operator algebra.
This was successfully demonstrated in \cite{KitK19} for the isotropic Heisenberg (XXX) model where we reproduced the result of \cite{BouCK96} from the $q$-vertex operator approach.
Another equally important objective behind this new method proposed here is to understand the role of the complex roots in the form-factors.
For the XXZ chain in the presence of an external field, the form-factor for the excitations involving complex roots was done in \cite{Koz17}.
\textcite{BooJMST07a} (BJMST) found hidden fermionic operators in the space quasi-local operators for the XXZ model, based on which they propose a new approach \cite{JimMS11} towards the computation of the form-factors of the bound-states.
An important conclusion drawn from this formalism is a prediction that all the form-factors can be written as smaller, finite determinants.
We do not use the BJMST approach here, although our results based on the method generalised from \cite{KitK19} does satisfy this criteria.
\addsec{Outline of the Thesis}
This thesis is organised into three parts.
The \hyperref[gen_descrptn_aba]{first} part is primarily of introductory nature.
In the \hyperref[comp_ff_XXX]{second} part we discuss our method of computation of the form-factors in the context of the isotropic Heisenberg (XXX) model.
The \hyperref[conclusion]{third} part summarises the result and lays down the conclusions.
Here we also briefly discuss scope for the generalisation of this method to the XXZ model and beyond.
\par
The first part is organised into two chapters.
In the \hyperref[chap:qism]{first chapter} we introduce the quantum inverse scattering method and determinant formulae for the scalar products and the finite form-factors.
The \hyperref[chap:spectre]{second chapter} is entirely devoted to the spectrum in the thermodynamic limit. The \hyperref[sec:gs_gen]{first section} deals with the ground state of the XXZ model for different values of $\Delta$. Here we introduce a conjecture for a generalised version of the condensation property which we use to write the sum over the roots of the ground state (or real roots of a low-lying excitation) involving a meromorphic function as integrals in the thermodynamic limit. In the \hyperref[sec:spectre_XXX]{second section} of this chapter we also introduce the Destri-Lowenstein picture which describes the bound state excitations in the low-lying spectrum which necessarily contain complex Bethe roots.
Although most of the discussion revolves around the XXX chain, we shall also discuss very briefly \cref{sec:xxz_spectre} the generalisation of the Destri-Lowenstein picture to the XXZ model by Babelon, de Vega and Viallet.
At the end of \cref{chap:spectre} we give an auxiliary result for the asymptotic form of the $\phifn$ function representing the ratio of Baxter polynomials, which later plays an important role in our computations.
\par
The second part contains all the technical details of our method employed on the XXX model.
Our method can be thought of as a three step process:
\begin{enumerate}[wide=0pt, label={\textbf{Step \arabic*: }},ref={step-\arabic*}]
\item In the first step we perform the so-called \emph{Gaudin extraction} that allows us to write a ratio of two determinants as a single determinant. This can be achieved by taking an action of the inverse of a Gaudin matrix, or an equivalent matrix\footnote{here equivalence means the equality of their determinants} on the Slavnov matrix.
This procedure gives us a determinant representation involving an infinite Cauchy matrix.
\item 
	In the second step, we \emph{extract} the infinite Cauchy matrix.
	This procedure leaves behind a small residual matrix of a finite size. In a particular case of two-spinon form-factors, it turns out that we can compute the determinant of the reduced matrix exactly, leading to a final result in closed-form. However, we cannot extend the same method to compute the residual matrices for the form-factors of bound states, as they must involve the complex Bethe roots which complicates the computations of auxiliary integrals. As a result we have not yet obtained a result for generic form-factors that can be expressed in closed-form.
\item 
In the third and final step we compute the infinite determinant together with the prefactor in the thermodynamic limit to obtain our final results.
\end{enumerate}
This part is organised into three chapters from \crefrange{chap:2sp_ff}{chap:gen_FF}.
In \cref{chap:2sp_ff} we first introduce our method and compute the thermodynamic limit of the two-spinon form-factor.
This chapter is entirely based on our published result 
\begin{quote}
\fullcite{KitK19}
\end{quote}
Here we use the determinant representations due to Slavnov and Gaudin for the scalar product and also a version of the Slavnov's determinant for excitations that are obtained as $\mathfrak{su}_2$ descendants of the leading Bethe vectors. 
With the procedure of Gaudin extraction and Cauchy extraction described earlier we obtain an exact result for the two spinon form-factor.
In particular, here we find that the residual determinant that is left behind after the Cauchy extraction (step two) is a Vandermonde matrix of size two. Hence its determinant can be easily computed which gives us a closed form expression for the two-spinon form-factor after computing the thermodynamic limit of the infinite Cauchy determinants and prefactors. This  result is also compared with the result from the $q$-vertex operator formalism.
\\
Over the next two \cref{chap:cau_det_rep_gen,chap:gen_FF}, we extend this method to the computation of generic form-factors for the bound states.
\Cref{chap:cau_det_rep_gen} is devoted to the Gaudin extraction in the generic case, where we discuss how the emergence of complex Bethe roots influences this process.
As a result of this procedure we obtain a modified Cauchy determinant representation where a small number of modifications are brought due to the presence of complex roots and not all of these extra terms associated with the Cauchy part coming from the real roots.
Nonetheless, the low-lying criteria for the excitation mean that the complex roots form a small fraction of the excited state Bethe roots and hence the determinant representation is still dominated by an infinite Cauchy matrix in the thermodynamic limit.
\Cref{chap:gen_FF} is devoted to the extraction of these infinite Cauchy matrices and computation of their determinants in the  thermodynamic limit.
Here we encounter a subtle impediment since we find that the Cauchy matrices that we wish to extract are rectangular.
This problem is resolved through the extraction of a Cauchy-Vandermonde matrix which is composed by mixing rectangular Cauchy and Vandermonde matrices. Its determinant formula is a simple generalisation of the Cauchy determinant.
With this extraction we obtain a determinant representation in terms of the reduced matrices of small and finite size in the thermodynamic limit.
This result is further examined for the four-spinon case where we find that the residual determinants can be computed to as a summation over some general terms. However, this does not yet give us a closed-form representation since these general terms are expressed in the form of integrals of a auxiliary $\Phifn$ functions which are hyperolic equivalent of the ratio of Baxter polynomials $\phi$.
\par
In the conclusion, we summarise the results obtained for the XXX model. Here we again compare the entire result for the two-spinon form-factors and the prefactors in the result for four-spinon form-factor with those obtained from the $q$-vertex operator algebra framework.
At the end, we will briefly discuss the possible extensions of this method to the XXZ model and to more diverse scenarios.
\par
There are three appendices to supplement all the computations. \Cref{chap:spl_fns} gives all the definitions and useful properties of the special functions used here.
\Cref{chap:den_int_aux} contains all the auxiliary computations that involve the density functions.
We study a general version of Lieb integral equations to define the density that encompasses all the variations of this integral equation that we need in our computations.
\Cref{chap:mat_det_extn} gives useful results for the determinants and the extractions of matrices. In this appendix we also discuss the mixed Cauchy-Vandermonde matrix, its determinant and inversion.
\clearpage{}%
\clearpage{}%
\addchap{Notations}
\label{ind_free_notn}	
A more comprehensive list of all the notations can be found in the index at the back of this Thesis.
\begin{table}[h]
\begin{tabularx}{\textwidth}{l|@{\enspace}X}
	$\Nset$	& set of natural numbers
	\\
	\hline
	$\Nset_\ast=\Nset\setminus\set{0}$	& set of natural numbers exlcuding zero
	\\
	\hline
	$\Zset$	&	set of integers
	\\
	\hline
	$\Qset$	&	set of rationals
	\\
	\hline
	$\Rset$	&	set of real numbers
	\\
	\hline
	$\Rset+i\alpha$, $\alpha\in\Rset$		&	a line parallel to the real line in complex plane
	\\
	\hline
	$\Cset$		&	set of complex numbers
	\\
	\hline
	$\delta_{j,k}$	& Kronecker's delta function
	\\
	\hline
	$\delta(x)$	& Dirac's delta function (as a distribution)
	\\
	\hline
	$I_S$	&	characterisitc function of a set $S$
	\\
	\hline
	$H(x)=I_{x>0}$	&	Heaviside step function	
\end{tabularx}
\end{table}
\addsec*{Index-free notation}
In this thesis we use a non-conventional%
\footnote{Although this is not a conventional notation, similar notations have been used by N. Slavnov, O. Foda (to name a few) in their works, see for example \cite{FodW12} or (math.ph/1911.12811).}
\emph{index-free} notation for sets of rapidities or spectral parameters, and sums and products involving such sets, that is denoted without writing the dummy indices explicitly. 
This is summarised in the following table, which is followed by some important clarifications presented in the remainder of this section.
\begin{table}[h]
	\begin{tabularx}{\textwidth}{l|@{\enspace}X}
	\renewcommand{\arraystretch}{2}
	$\bm z=\set{z_1,z_2,\ldots,z_n}$		&		set of complex parameters 
	\\
	\hline
	\\[-1.5ex]
	$n_{\bm z}=\#\bm z$		& cardinality of a set
	\\[0.5em]
	\hline
	\\[-1.5ex]
	$\bm{z_{\hat{a}}}=\bm z\setminus\set{z_a}$	&	set with an omission 
	\\[0.5em]
	\hline
	\\[-1.5ex]
	$\displaystyle\bmprod f(\bm z)=\prod_{\forall z_a\in\bm z}f(z_a)$		&	product over all elements in the set $\bm\la$

	\\[0.5em]
	\hline
	\\[-1.5ex]
	$A[\bm z]$ or $B[\bm z \Vert\bm w]$	&	 matrix parametrised by set(s) of parameters
	\\[0.5em]
	\hline
	\\[-1.5ex]
	$\bmalt f(\bm z)$ or $\bmalt f(\bm z\Vert\bm w)$	&	 (super-)alternant product
	\\[0.5em]
	\hline
	\\[-1.5ex]
	$\ptn*\la = \set{\la_1,\la_2,\ldots|\la_{i}\in \Nset, \la_{i+1}<\la_{i}}$
														& partition of integers
\end{tabularx}
\end{table}
\subsection*{Index-free set}
\label{ind_free_set}
\index{Index-free@\textbf{Index-free}!set@$\bm\la$ (e.g.) : \rule{3em}{.5 pt} set}%
	We will use the \textbf{bold} typeface for mathematical symbols that denote a set of variables.
	\begin{example}
		The notation $	\bm\la$ denotes the set
		\begin{align*}
			\bm\la=\set{\la_{j}}_{j=1}^{n}
		\end{align*}
		The cardinality $n(\bm\la)$ has to be explicitly given in this notation.
		In this example, we have $n(\bm\la)=n$.
		We will also use the notation $n_{\bm\la}$ for the cardinality.
		An omission of an index will be denoted as
		\begin{align*}
			\bm\la_{\hat{a}}=
			\bm\la\setminus \set{\la_{a}}
			.
		\end{align*}
	\end{example}
	Set operations in this notations are defined as follows:
	\begin{enumerate}[leftmargin=\parindent]
		\item The addition with a scalar (complex) can be used to define a shifted set.
		\begin{example}
			$\bm\la+\eta$ denotes the set $\set{\la_{j}+\eta}_{j}$.
		\end{example}
		\item Multiplication by a scalar corresponds to the dilation or rescaling of the set
		\begin{example}
			$\alpha\bm\la$ denotes the set $\set{\alpha\la_{j}}_{j}$.
		\end{example}
		\item Addition (or subtraction) of two sets can be defined as
		\begin{example}
			$\bm\la-\bm\mu$ can be used to denote $\set{\la_{j}-\mu_{k}}_{j,k}$.
		\end{example}
		We also note that as far as the addition of sets is concerned, we will drop the condition for its elements to be distinct and allow for repetitions.
		That is to say that we interpret, $\bm\la-\bm\mu$ as a collection rather than a set in the set theoretic parlance.
	\end{enumerate}
\subsection*{Index-free products (or summations)}
\label{ind_free_prod}
\index{Index-free@\textbf{Index-free}!product@$\bmprod$ : \rule{3 em}{.5 pt} product|textbf}%
\index{Index-free@\textbf{Index-free}!summation@$\bmsum$ : \rule{3 em}{.5 pt} summation|textbf}%
We use the product and sum operators in \textbf{bold} typeface $\bmprod$ and $\bmsum$ to denote the product and sum running over a set (or collection) of variables.
	\begin{example}
		For any function $f:\Cset\to\Cset$, we can define
		\begin{align*}
			\bmprod f(\bm\la) &= \prod_{j=1}^{n_{\bm\la}} f(\la_{j})
			&&\text{and}
			&
			\bmsum f(\bm\la) &= \sum_{j=1}^{n_{\bm\la}} f(\la_{j})
			.
		\end{align*}
		This can be combined with the operations defined on the sets above.
		\begin{align*}
			\bmprod f(\bm\la-\bm\mu) &= \prod_{j=1}^{n_{\bm\la}}\prod_{k=1}^{n_{\bm\mu}} f(\la_{j}-\mu_{k})
		\end{align*}
	\end{example}
	We may face a situation where the function $f$ has poles (or zeroes) that we want to avoid from the product (or sum), in such case, we will attach an attribute $'$ to the product (or sum).
	\begin{example}
	\begin{align*}
		\bmprod^\prime (\bm{\la}-\bm{\la})
		&=
		\prod_{\underset{j\neq k}{j,k=1}}^{n_{\bm\la}} (\la_{j}-\la_{k})
		.
	\end{align*}
	\end{example}
	When there are more than one sets of variables present in summand (product term), the dummy set can be often deduced by comparison.
Nonetheless for the sake of clarity, we will be sometimes add explicit indications for dummy variables in the subscript as it is shown in the following example:
	\begin{example}
	\begin{align*}
		g(\bm\mu)
		=
		\bmsum_{\bm \la} f(\bm\la,\bm\mu)
		=
		\sum_{j=1}^{n_{\bm\la}} f(\la_{j},\bm\mu)
		\intertext{which denotes a partial sum in contrast to the double sum}
		h=
		\bmsum_{\bm \la, \bm \mu} f(\bm \mu,\bm \la)
		=
		\sum_{j=1}^{n_{\bm\la}}\sum_{k=1}^{n_{\bm_{\bm\mu}}}
		f(\bm\mu,\bm\la)
		.
	\end{align*}
	\end{example}
In the scenario where a set appears in the definition of a function, a single vertical bar $\vert$ will be used to seperate the implicit or hidden variables. For example, let us take the precedent partial sum, we can have
\begin{align*}
	g(\la_a|\bm\la) = \bmsum f(\la_a,\bm\la).
\end{align*}
Any further operation with this function will now only involve the parameters before the vertical bar $|$, as the parameters after the bar are hidden variables that occur implicitly in its definition. For example
\begin{align*}
	e = \bmsum g(\bm\la|\bm\la) q(\bm\la)
	=
	\sum_{j=1}^{n_{\bm\la}} g(\la_a|\bm\la) q(\la_a)
\end{align*}
Whenever there is the slightest indication that the ambiguity due to the use of this notation cannot be resolved, we will fall back to the classical notation. All final results presented as conlusions will also be expressed in classical notations.
\subsection*{Parametrised matrices}
\label{ind_free_mats}
\index{Index-free@\textbf{Index-free}!parametrised matrix@$[\cdot\Vert\cdot]$ or $[\cdot]$ : parametrised matrix|textbf}%
We will use the following notation for a parametrised matrix $A$:
\begin{align*}
	A[\bm x\Vert\bm y]
	=
	[a(x_j,y_k)]_{j,k}
\end{align*}
A complex valued function $a:\Cset^2\to\Cset$ needs to be explicitly given on case-by-case basis for each individual usage of this notation.
\\
When parametrised by a single set, it will become necessary to tell the reader explicitly about the form of the matrix in order to avoid confusion while reading. 
This can avoided through the use phrases such as:
\begin{itemize}
\item B is a square/ rectangular matrix given by
\begin{align*}
	B[\bm x]=\begin{bmatrix} b_1(x_j) & \ldots & b_n(x_j) \end{bmatrix}
\end{align*}
\item C is a column (row) vector given by
\begin{align*}
	C[\bm x]=[c(x_j)]_{j}
\end{align*}
\item D is a diagonal matrix given by
\begin{align*}
	D[\bm x]=
	\diag(d(x_1),d(x_2),\ldots,d(x_n))
	=
	\begin{bmatrix}
		d(x_1)	&	& 0	\\
		&	\ddots	& \\
		0		&	&	d(x_n)
	\end{bmatrix}
\end{align*}
\end{itemize}
Matrix products in this notation can be denoted with the contractions of dummy variables. For example,
\begin{align*}
	P[\bm u \Vert \bm z ]\cdot Q[\bm z\Vert \bm v] &= PQ[\bm u\Vert\bm v]
	\\[0.1em]
	P[\bm u\Vert\bm z]\cdot B[\bm z] &= PB[\bm u] && \text{both $B$ and $PB$ are square/ rectangular}
	\\[0.1em]
	P[\bm u\Vert\bm z]\cdot C[\bm z] &= PC[\bm u]	&&	\text{$C$ and $PC$ are column \& row vectors resp.}
	\\[0.1em]
	P[\bm u\Vert\bm z]\cdot D[\bm z] &= PD[\bm u \Vert \bm z] && D \text{ is diagonal} 
\end{align*}
Note that in the case of multiplication by diagonal matrix, there is no dummy variable that is summed over. Instead what we get is a modification of the original matrix that can be seen as a \emph{diagonal dressing}.
\\
Finally, let us note that parametrised notation for matrices can also be combined with the previous notation for sums and products. For example in the following
\begin{align*}
	V[\bm y] = \bmsum_{\bm x} W[\bm x\Vert \bm y]
\end{align*}
we sum over each elements in the rows of the matrix $W$ to obtain a column matrix $V$. 
\subsection*{Alternant product}
\label{ind_free_alt_prod}
\index{Index-free@\textbf{Index-free}!product alternant@$\bmalt f(\cdot)$ : alternant product|textbf}%
	We will use the notation $\bmalt$ to denote a very special type of product.
	Given a set of variables $\bm{x}$ and a function $f: \Cset\to\Cset$, we define the alternant $\bmalt f:\Cset^{n_{\bm x}}\to \Cset$ as a product 
	\begin{align*}
		\bmalt f(\bm{x})&=
		\prod_{j<k}^{n_{\bm{x}}}f(x_{j}-x_{k})
		.
	\end{align*}
	For some particular choices of the function $f$, this product can be interpreted as the Vandermonde determinant.
	Furthermore, we shall use the notation $\bmalt^2$ to denote
	\begin{align*}
		\bmalt^2f(\bm x)
		=
		\bmalt f(\bm x)
		\bmalt f(-\bm x)
		=
		\prod_{j\neq k}^{n_{\bm x}} f(x_j-x_k)
		.
	\end{align*}
	\subsection*{Superalternant product}
\index{Index-free@\textbf{Index-free}!product superalternant@$\bmalt f(\cdot\Vert\cdot)$ : superalternant product|textbf}%
	\label{ind_free_supalt_prod}
	Similarly, let us define a \emph{supersymmetric} variant of the alternant product. Given $f:\Cset\to\Cset$, we define the superalternant $\bmalt f:\Cset^{n_{\bm x}}\times\Cset^{n_{\bm y}}\to\Cset$ as the product
	\begin{align*}
		\bmalt f(\bm{x}\Vert\bm{y})
		=
		\frac{\bmalt f(\bm{x})\bmalt f(\bm{-y})}{\prod f(\bm{x}-\bm{y})}
		=
		\frac{%
		\prod_{j>k}^{n_{\bm x}} f(x_j-x_k)
		\prod_{j<k}^{n_{\bm y}} f(y_j-y_k)
		}{%
		\prod_{j=1}^{n_{\bm x}}
		\prod_{k=1}^{n_{\bm y}}
		f(x_j-y_k)
		}
		.
	\end{align*}
	For the particular choice of function $f$, it can be interpreted as the Cacuhy-Vandermonde determinant.
	We shall sometimes use the notation $\bmalt^2$ to denote the following product of superalternants
	\begin{align*}
		\bmalt^2 f(\bm x\Vert\bm y)
		=
		\bmalt f(\bm x\Vert \bm y)
		\bmalt f(-\bm x\Vert -\bm y)
		=
		\frac{%
		\prod_{j\neq k}^{n_{\bm x}} f(x_j-x_k)
		\prod_{j\neq k}^{n_{\bm y}} f(y_j-y_k)
		}{%
		\prod_{j=1}^{n_{\bm x}}
		\prod_{k=1}^{n_{\bm y}}
		f(x_j-y_k)
		f(y_k-x_j)
		}
		.
	\end{align*}
\label{ind_free_notn_end}
\addsec*{Partition of integers}
\label{ptn_notn_page_begin}
We denote a partition of non-negative integers as $\ptn*{\la}=\set{\la_1,\la_2,\ldots,}$ with descending property $\la_1\geq \la_2 \cdots$.
The number of non-zero integers in any partition is always finite.
\begin{itemize}[leftmargin=*, wide]
\item 
The length of a partition (stripped of the trailing zeroes) is denoted by $\ell(\ptn*\la)$, while its weight is denoted by $w(\ptn*\la)=\sum_{a}\ptn*\la_a$.
Whenever it would be felt necessary, the length of partition is indicated explicitly as $\ptn\la(n)$ which means $\ell(\ptn\la(n))=n$.
\item
Here we will also allow partitions in ascending order, the reason behind it is to avoid unnecessary sign corrections in the determinants.
Our notation is adopted for this reordering with a left-ward pointing arrow $\ptn*{\la}$ indicating the descending order and a right-ward arrow $\ptn{\la}$ indicating an ascending order.
It also helps us distinguish a partition of integers from spectral parameter, both of which are usually denoted by Greek letters.
\begin{rem}
It should be noted that reordering in $\ptn\la$ is only symbolic. 
For that matter, an addition of partitions is always carried in the descending order, no matter in which order it is written.
For example, let $\ptn\la=\set{0,1,3,4,7}$ and $\ptn\mu=\set{1,2,5}$, the sum is given by $\ptn\la+\ptn\mu=\set{0,1,4,6,12}$.
\end{rem}
\end{itemize}
The partition of consecutive integers of length $n$ will be denoted as $\ptn\delta$.
\begin{align*}
	\ptn*\delta(n)=
	\set{n-1,n-2,\ldots,0}
	\label{ptn_consec}
	.
\end{align*}
Similarly, we also define the partition into even or odd integers $\ptn\gamma$ as
\begin{equation*}
\begin{aligned}
	\ptn*\gamma&=
	\set{n-1,n-3,\ldots,0}
	,
	&
	&(\text{for } n \text{ odd})
	.
	\\
	\ptn*\gamma&=
	\set{n-1,n-3,\ldots,1}
	,
	&
	&(\text{for } n \text{ even})
	.
\end{aligned}
\end{equation*}
Its length is $\ell(\ptn\gamma)=\lfloor \frac{n}{2}\rfloor$ where $\lfloor\cdot\rfloor$ denotes the integer part.
\label{ptn_notn_page_end}
\clearpage{}%
\setpartpreamble[c][.9\textwidth]{%
\vspace{3em}
\small
This part is divided in two chapters.
In the beginning of the \hyperref[chap:qism]{first chapter} we recall the basic framework of the algebraic Bethe ansatz or also called the quantum inverse scattering method for quantum spin chains developed by \textcite{FadST79}.
In the latter half of this chapter from \cref{sec:qism_det_rep} onwards, we will discuss how this framework enables the computation of correlation function and form-factors for this model.
\par
The second chapter is dedicated to the treatment of the spectrum in the thermodynamic limit.
The nature of the anti-ferromagnetic ground state, the excitations which are energetically close to this ground state and the condensation property for these states is discussed here.
We also devote the end of this chapter to the complex solutions of the Bethe equations in the picture presented by \textcite{DesL82}.	
}
\part[Quantum inverse scattering method]{{\huge An Introduction \\ to the} \\[.5em] {Quantum Inverse Scattering Method}}
\label{gen_descrptn_aba}
\clearpage{}%
\chapter{Quantum inverse scattering method}
\label{chap:qism}
In this chapter we will lay down the framework of the algebraic Bethe ansatz that we use in our computations. All the discussion in this chapter is made in the context of the anisotropic Heisenberg or the XXZ model which describes a one-dimensional periodic chain of even length $M=2N$ composed of interacting spin-$\frac 12$ particles.
The spins interact with their nearest neighbours and this interaction is governed by the quantum Hamiltonian $H_\Delta$ which is given by,%
\begin{align}%
	H_{\Delta}=J\sum_{m=1}^{M}%
	\sigma^{1}_{m}\sigma^{1}_{m+1}%
	+%
	\sigma^{2}_{m}\sigma^{2}_{m+1}%
	+\Delta%
	(\sigma^{3}_{m}\sigma^{3}_{m+1}-\Id)%
	.%
	\label{xxz_ham}%
\end{align}%
The coupling parameter $J\in\Rset$ represents the interaction strength, whereas the \emph{anisotropy} parameter ${\Delta}\in\Rset$ governs the longitudinal anisotropy of the XXZ model.
At the isotropic point $\Delta=1$, we find the Hamiltonian $H_1$ of the isotropic XXX Heisenberg model. Whereas, at the free fermion point where the anisotropy parameter vanishes $\Delta=0$, we obtain the $XX$ model, which is also called the free-fermion model since its Hamiltonian $H_0$ can be mapped to that of the free-fermion gas using the Jordan-Wigner transformation.
\par
To each individual site on the lattice of a quantum spin chain we associate a $\Cset^2$ vector space. 
Hence the total Hilbert space for the XXZ model can be decomposed as a tensor product ${\qsp}=\otimes^{M}{\Cset^{2}}$.
In the terminology of the algebraic Bethe ansatz, the Hilbert space of a model is often called \emph{quantum space}, which we will denote as $\qsp$.
The operators $\sigma^{\alpha}_{n}$ in \cref{xxz_ham} for the Hamiltonian $H_\Delta$ denote the local spin operators, which act non-trivially only on the $n$\textsuperscript{th} site on the lattice as
\begin{align}%
	\sigma^{\alpha}_{n}=\Id_{n-1}\otimes \sigma^{\alpha} \otimes \Id_{M-n}.%
	\label{loc_sp_op_def}
\end{align}%
Here, $\sigma^\alpha$ are the Pauli matrices:
\begin{align}
	\sigma^{1}&=
	\begin{pmatrix}
		0	&	1 \\ 1 & 0
	\end{pmatrix}
	,
	&
	\sigma^{2}&=
	\begin{pmatrix}
		0	&	-i \\ i & 0
	\end{pmatrix}
	,
	&
	\sigma^{3}&=
	\begin{pmatrix}
		1	&	0 \\ 0 & -1
	\end{pmatrix}
	;
\end{align}
and the $\Id$ denotes the identity matrix.
Thus, the spin operators $\sigma^\alpha_n$ satisfy the local $\mathfrak{su}_{2}$ algebra: %
\begin{align}%
	[\sigma^{\alpha}_{n},\sigma^{\beta}_{m}]=2i\delta_{m,n}\,\epsilon^{\alpha\beta\gamma}\sigma^{\gamma}_{n}.%
	\label{loc_su2_alg}
\end{align}%
The local raising and lowering spin operators $\sigma^\pm_m$ are defined as:
\begin{align}
	\sigma^\pm_m
	=
	\sigma^1_m\pm i\sigma^2_m
	.
	\label{pauli_low-raising}
\end{align}
\par
The sign of coupling constant $J$ in \cref{xxz_ham} for the Hamiltonian can be fixed so that it is always positive $J>0$ since we have the similarity transformation which maps%
\begin{align}%
	-H_{\Delta}&=U^{-1}H_{-\Delta}U,%
	&%
	U=\prod_{m=1}^{\frac M2}\sigma_{2m}^{3}%
	\label{xxz_similarity_transf}
	.
\end{align}%
The total spin operator $S$ on $\qsp$ is defined as the sum of local spin operators on each lattice site, it has the components:
\begin{align}%
	S^{\alpha}&=\frac{1}{2}\sum_{m=1}^{M}\sigma^{\alpha}_{m}.
	\label{tot_spin_op}
\end{align}%
It can be readily checked that the components of the total spin operators follow the $\mathfrak{su}_2$ algebra and the XXZ Hamiltonian commutes with the third component $S^{3}$ of the total spin operator 
\begin{align}%
	[S^{3}, H_{\Delta}]=0.%
	\label{ham_xxz_comm_S3}%
\end{align}%
For the XXX model at the isotropic point $\Delta=1$, the Hamiltonian $H_1$ possesses extended $\mathfrak{su}_{2}$ symmetry since all the total spin operators commute with the Hamiltonian %
\begin{align}%
	[S^{\alpha}, H_{1}]&=0%
	&%
	\forall & \alpha\in \set{1,2,3}.%
	\label{xxx_ham_iso}%
\end{align}%
The anisotropy parameter $\Delta$ also determines the nature of the spectrum which we will study in more detail in \cref{chap:spectre}.
We are always interested in the values of the anisotropy parameter $\Delta>-1$ for which the lowest energy state or ground state has anti-ferromagnetic nature.
We also further divide this into the following regimes based on the nature of excitations and symmetry:
\begin{itemize}[noitemsep, labelsep=1em]
	\item \emph{massless} XXZ for the values $-1<\Delta<1$
	\item {XXX} for the isotropic point $\Delta=1$
	\item \emph{massive} XXZ for the values $\Delta>1$.
\end{itemize}
At this point it is important to remark that the {XXX model} is the primary focus of all our computations carried out in \cref{comp_ff_XXX} of this thesis.
Nonetheless, we begin here in this introductory \cref{gen_descrptn_aba} with the {XXZ model} for its virtue of being more general. %
This would facilitate our later discussion focused around broader applicability of the method in conclusions (\cref{concl}). %
\par
The XXZ model governed by $H_\Delta$ \eqref{xxz_ham} is integrable for all values of the anisotropy parameter $\Delta$.
The exact solutions for the spectral problem were originally given by \textcite{Bet31} for the XXX chain and these were extended to the XXZ model for an arbitrary value of $\Delta$ by \textcite{Orb58}.%
The method that they used to obtain these solutions is now known as the coordinate Bethe ansatz. %
In this thesis, we will not use this method but its algebraic reformulation developed by \textcite{FadST79}. This latter method which is discussed here goes by the name of the algebraic Bethe ansatz (ABA) or more appropriately the quantum inverse scattering method. %
\section{Algebraic formulation of Bethe ansatz}
\label{sec:aba}
The Hilbert space of a quantum spin chain of length $M$ can be expressed as a tensor product of local quantum spaces $\qsp = \otimes_{m=1}^{M} V_{m}$.
We can define the permutation operator which acts on the tensor product of vector spaces $V\otimes V$ with an action given by the exchange property:
\begin{align}
	\Pm (v_{1}\otimes v_{2})=v_{2}\otimes v_{1}.
	\label{pmat_exchange}
\end{align}
Let us now define the permutation operators $\Pm_{jk}$ as natural extension of $\Pm$ onto a $M$-fold tensor product $\otimes_{a} V_a$ of local vector spaces such that it acts trivially everywhere except on $V_{j}$ and $V_{k}$.
We can check that the permutation operators $\Pm_{j,k}$ satisfy the relations:
\begin{align}
	\Pm^{2}_{j,k}&=\Id,
	&&\text{and}
	&
	\Pm_{jk}\Pm_{kl}&=\Pm_{kl}\Pm_{jl}=\Pm_{jl}\Pm_{kl}.
	\label{pmat_props}
\end{align}
For the XXZ model which is a fundamental spin-$\frac{1}{2}$ representation of the $\mathfrak{su}_2$ spin chains,
the individual vector spaces $V_{m}$ are all isomorphic to $\Cset^2$.
Here the permutation operator is a $4\times 4$ matrix which admits the decomposition
\begin{align}
	\Pm = \frac{1}{2} \sum_{\alpha=1}^{3}\sigma^{\alpha}\otimes\sigma^{\alpha}
	.
\end{align}
Let us now define the $\Rm$-matrix which was originally found by Baxter \cite{Bax89} for the six-vertex model.
Incidentally, it also plays a central role in the algebraic Bethe ansatz for the XXZ model and related integrable models.
\begin{defn}[Six-vertex \emph{$\Rm$-matrix}]
\index{aba@\textbf{Algebraic Bethe ansatz (ABA)}!R matrix@$\Rm$-matrix|textbf}%
\label{defn:R-mat}
We introduce an operator-valued function $\Rm$ of a spectral parameter $\la\in\Cset$ which acts on the tensor product space $\Cset^{2}\otimes\Cset^{2}$. 
In the elementary basis $e_a\otimes e_b$, it can be expressed in matrix form as
\begin{align}
	\Rm(\la)=
	\begin{pmatrix}
	1		&		&		&		\\
			&f(\la)	&g(\la)	&		\\
			&g(\la)	&f(\la)	&		\\
			&		&		&1
	\end{pmatrix}
	.
	\label{Rmat_6v}
\end{align}
We also require that it satisfy the \emph{Yang-Baxter equation}:
\begin{align}
	\Rm_{12}(\la)\,\Rm_{13}(\la+\mu)\,\Rm_{23}(\mu)
	=
	\Rm_{23}(\mu)\,\Rm_{13}(\la+\mu)\,\Rm_{12}(\la)
	.
	\label{YBE}
\end{align}
\end{defn}
We are not going to discuss all the possible solutions of the Yang-Baxter equation here. This was the premise behind a range of investigations \cite{Jim90Book} which led to the development of quantum groups.
Here, we will directly begin our discussion with the following solution, that leads us back to the XXZ model through a relation which is known as the \emph{trace identity}.
\begin{lem}
With the following weight functions $f(\la)$ and $g(\la)$ for the $\Rm$-matrix the Yang-Baxter equation \eqref{YBE} is satisfied.
\begin{align}
 	f(\la)&=\frac{\varphi(\la)}{\varphi(\la+i\gamma)},
 	&
 	g(\la)&=\frac{\varphi(i\gamma)}{\varphi(\la+i\gamma)}
	.
 	\label{rmat_wts_paramn_gen}
\end{align}
The function $\varphi$ can correspond to any one of the following possible choices:
\index{aba@\textbf{Algebraic Bethe ansatz (ABA)}!weight matrix@$\varphi$ : weight function|textbf}%
\begin{align}
	\varphi(\la)=
	\begin{dcases}
		\la 	& \text{rational parametrisation},
		\\
		\sin(\la) 	& \text{trigonometric parametrisation},
		\\
		\sinh(\la) 	& \text{hyperbolic parametrisation}.
	\end{dcases}
	\label{rmat_paramn_types}
\end{align}
\end{lem}
\begin{rem}
These three parametrisations are related to the different regimes of the XXZ model according to the value of $\Delta$ that we discussed earlier.
Within the different regimes, the parameter $\gamma$ is related to the anisotropy parameter $\Delta$, this relationship will be established when we come to the trace identities.
\end{rem}
We can take note that the our \Rm-matrix \eqref{Rmat_6v} with the parametrisations \eqref{rmat_wts_paramn_gen} can be expressed as the sum:
\begin{multline}
	\Rm(\la)=
	\frac{\varphi(\la+i\gamma)+\varphi(\la)}{2\varphi(\la+i\gamma)}\,\Id
	+ \frac{\varphi(i\gamma)}{2\varphi(\la+i\gamma)} (\sigma^{1}\otimes\sigma^{1}+\sigma^{2}\otimes\sigma^{2})
	\\
	+ \frac{\varphi(\la+i\gamma)-\varphi(\la)}{2\varphi(\la+i\gamma)}\,(\sigma^{3}\otimes\sigma^{3})
	.
\end{multline}
Let us also note that \Rm-matrix possesses the following properties:
\begin{subequations}
\label{Rmat_6v_props}
\begin{flalign}
	&\textbf{Initial condition:}
	&
	&\Rm(0)= \Pm,
	&&&&&&
	\label{Rmat_6v_init_condn}
	\\
	&\textbf{Unitarity:}
	&
	&\Rm_{ab}(\la)\Rm_{ba}(-\la)=\Id,
	\label{Rmat_6v_unitarity}
	\\
	&\textbf{Crossing symmetry:}
	&
	&\sigma^{2}_{a}\Rm_{ab}^{t_{a}}(\la)\sigma^{2}_{a}=\Rm_{ba}(-\la).
	\label{Rmat_6v_crossing}
\end{flalign}
\end{subequations}
In \cref{Rmat_6v_crossing}, $t_a$ denotes the paritial transposition in the local vector space $V_a$
\par
Let us now write down an augmented Hilbert space $V_{a}\otimes V_{q}$ which is obtained by enlarging the original quantum space $\qsp$ with the \emph{auxiliary space} $V_{a}\eqsim\Cset^2$.
On this augmented Hilbert space we define the following operator.
\begin{defn}[\emph{Monodromy matrix}]
\label{defn:monodromy}
\index{aba@\textbf{Algebraic Bethe ansatz (ABA)}!monodromy matrix@$\Mon$ : monodromy matrix|textbf}%
\index{aba@\textbf{Algebraic Bethe ansatz (ABA)}!abcd block operators@$\opA$, $\opB$, $\opC$ or $\opD$ operators|see{$\Mon$}}%
We define the \emph{monodromy matrix} as an operator on the augmented quantum space, which is given by the following product:
\begin{align}
	\Mon_{a}(\la)=\Rm_{aM}\left(\la-\tfrac{i\gamma}{2}\right)\Rm_{a,M-1}\left(\la-\tfrac{i\gamma}{2}\right)\cdots\Rm_{a1}\left(\la-\tfrac{i\gamma}{2}\right)
	\label{mon_mat}
	.
\end{align}
It is often represented as follows:
\begin{align}
	\Mon_{a}(\la)=
	\begin{pmatrix}
		\opA(\la)	&	\opB(\la)
		\\
		\opC(\la)	&	\opD(\la)
	\end{pmatrix}_{a}
	\label{monodromy_blocks}
\end{align}
where the blocks $\opA$, $\opB$, $\opC$, $\opD$ are operator valued functions taking their values as operators on the quantum space $\qsp$.
\end{defn}
\begin{lem}
The monodromy matrix also satisfies the Yang-Baxter equation, since
\begin{align}
	\Rm_{ab}(\la-\mu) \Mon_{a}(\la) \Mon_{b}(\mu)
	=
	\Mon_{b}(\mu) \Mon_{a}(\la) \Rm_{ab}(\la-\mu).
	\label{monodromy_yb}
\end{align}
\end{lem}
\begin{coro}[{Fundamental commutation relations, \emph{FCR}}]
\Cref{monodromy_yb} imposes the following {commutation relations} on the block operators of the monodromy matrix:
\begingroup
\allowdisplaybreaks
\begin{subequations}
\begin{gather}
	[\opA(\la),\opA(\mu)]=
	[\opB(\la),\opB(\mu)]=
	[\opC(\la),\opC(\mu)]=
	[\opD(\la),\opD(\mu)]=0
	.
	\label{FCR_self}
	\\[0.5em]
\begin{aligned}
	\opA(\la)\opB(\mu)&=f(\la-\mu)\opB(\mu)\opA(\la)+g(\la-\mu)\opA(\la)\opB(\mu)
	,
	\\
	\opD(\la)\opB(\mu)&=f(\mu-\la)\opB(\mu)\opD(\la)+g(\mu-\la)\opD(\la)\opB(\mu)
	;
	\\[0.25em]
	\opB(\la)\opA(\mu)&=f(\la-\mu)\opA(\mu)\opB(\la)+g(\la-\mu)\opB(\la)\opA(\mu)
	,
	\\
	\opB(\la)\opD(\mu)&=f(\mu-\la)\opD(\mu)\opB(\la)+g(\mu-\la)\opB(\la)\opB(\mu)
	.
	\label{FCR_AB_DB}
\end{aligned}
	\\[0.5em]
\begin{aligned}
	\opA(\la)\opC(\mu)&=f(\mu-\la)\opC(\mu)\opA(\la)+g(\mu-\la)\opA(\la)\opC(\mu)
	,
	\\
	\opD(\la)\opC(\mu)&=f(\la-\mu)\opC(\mu)\opD(\la)+g(\la-\mu)\opD(\la)\opC(\mu)
	;
	\\[0.25em]
	\opC(\la)\opA(\mu)&=f(\mu-\la)\opA(\mu)\opC(\la)+g(\mu-\la)\opC(\la)\opC(\mu)
	,
	\\
	\opC(\la)\opD(\mu)&=f(\la-\mu)\opD(\mu)\opC(\la)+g(\la-\mu)\opC(\la)\opD(\mu)
	.
	\label{FCR_AC_DC}
\end{aligned}
	\\[0.5em]
\begin{aligned}
	f(\la-\mu)[\opA(\la),\opD(\mu)]&=g(\la-\mu)(\opC(\mu)\opB(\la)-\opC(\la)\opB(\mu))
	,
	\\
	f(\la-\mu)[\opC(\la),\opB(\mu)]&=g(\la-\mu)(\opA(\mu)\opD(\la)-\opA(\la)\opD(\mu))
	.
	\label{FCR_BC}
\end{aligned}
\end{gather}
\label{fcr_aba}
\end{subequations}
\endgroup
\end{coro}
\begin{defn}[\emph{Transfer matrix}]
\label{def:trans_mat}
\index{aba@\textbf{Algebraic Bethe ansatz (ABA)}!transfer matrix@$\tf$ : transfer matrix|textbf}%
The \emph{transfer matrix} $\tf$ is obtain from the monodromy matrix $\Mon$ by taking the partial trace over the auxiliary space
\begin{align}
	\tf(\la)
	=\tr_{a}\Mon_{a}(\la)
	=\opA(\la)+\opD(\la)
	.
	\label{def_tfmat}
\end{align}
\end{defn}	
From the Yang-Baxter relation for the monodromy matrices \eqref{monodromy_yb} we can see that the transfer matrix generates a one-parameter family of commuting operators on the quantum space $\qsp$.
\begin{align}
	[\tf(\mu),\tf(\la)]&=0; \qquad	\forall \la,\mu\in\Cset.
	\label{tfmat_commutator}
\end{align}
This tells us that the transfer matrix is the generating function of an infinite set of mutually commuting operators.
These are obtained by expanding the transfer matrix as formal series at any particular value. However in general, such an expansion does not give us the local operators (i.e. operators with finite support in the $M\to\infty$ limit).
However, for a particular case where the expansion is taken around $i\frac{\gamma}{2}$, all the operators obtained are quasi-local operators which can be expressed as sums of local operators.
Among them we also find the Hamiltonian of the XXZ chain as was defined in \cref{xxz_ham}.
Therefore the transfer matrix and all the quasi-local operators generated by the transfer matrix also commute with the Hamiltonian of the XXZ model $H_\Delta$ (including the isotropic XXX model for $\Delta=1$) giving us an infinite set of conserved charges.
The equation that relates the conserved charges to the derivatives of the transfer matrix are called the \emph{trace identities}.
The central sub-algebra of the commuting conserved charges which include the Hamiltonian is called the \emph{Bethe subalgebra}.
\minisec{Trace identities}
From the property \eqref{Rmat_6v_init_condn} of the $\Rm$-matrix we can deduce that the following evaluation of the monodromy matrix can be expressed in terms of the product of permutation matrices 
\begin{align}
	\Mon_{a}\left(\tfrac{i\gamma}{2}\right)=\Pm_{aM}\Pm_{a,M-1}\cdots\Pm_{a1}.
\end{align}
The property \eqref{pmat_props} of the permutation matrix allows us to rewrite this to obtain the cyclic \emph{shift operator} which is given by the following expression
\begin{align}
	\tf\left(\tfrac{i\gamma}{2}\right)&=
	\Pcal_{12}\Pcal_{23}\Pcal_{34}\cdots\Pcal_{M-1,M}\Pcal_{M,1}
	\label{tr_id_shft_op}
	\shortintertext{and its logarithm gives us a trace identity for the total momentum operator}
	P&=-iJ\,\log\tf\left(\frac{i\gamma}{2}\right).
	\label{tr_id_mom}
\end{align}
A similar computation with the first derivative shows that we can extract the shift operator from $\tf^\prime(\frac{i\gamma}{2})$. The summation that remains after this procedure is nothing but the Hamiltonian $H_\Delta$ of XXZ chain as we have defined in \cref{xxz_ham}.
This leads us to the following trace identity that gives a relation between the Hamiltonian $H_\Delta$ and the transfer matrix. It is given by,
\begin{align}
	H_{\Delta}=
	{2J}{\varphi(i\gamma)}\,
	\log T^\prime\left(\tfrac{i\gamma}{2}\right)
	.
	\label{tr_id_ham}
\end{align}
To obtain this we have identified the parameter {$\Delta$} with the parameter $\gamma$ according to the following relation:
\begin{align}
	\Delta=-i\varphi^\prime(i\gamma)
	.
	\label{tr_id_aniso}
\end{align}
For the rational parametrisation, we immediately obtain the {XXX} Hamiltonian with $\Delta=1$. Let us remark that rational parametrisation \eqref{rmat_paramn_types} that we use in the case of the XXX model, the spectral parameter can be always rescaled to set $\gamma$ to any particular value of our choice. Here we shall fix it to $\gamma=1$.
\\ 
For the hyperbolic parametrisation, we obtain $\Delta=\cos\gamma$. The parameter $\gamma$ takes values in the interval $]0,\pi[$ and this corresponds to the \emph{anisotropic} \emph{disordered} regime $-1<\Delta<1$ of the XXZ chain.
\\
For the trigonometric parametrisation, we obtain $\Delta=\cosh\gamma$. The parameter $\gamma$ takes values in $\Rset^{+}_{*}$ and it corresponds to the massive regime $\Delta>1$ of the XXZ chain. 
\minisec{Isotropic symmetry of the XXX model}
We have seen in \cref{ham_xxz_comm_S3} that the third component of the total spin operator $S^{3}$ also commutes with the Hamiltonian of the XXZ model.
Following the approach of \cite{FadT84}, here we can show that this $U(1)$ symmetry can be given a broader sense through the following relation that connects the commutator of the $S^{3}$ with the monodromy matrix in the quantum space with the commutator in the auxiliary space.
\begin{align}
	[\Id_{a}\otimes S^{3}, \Mon_{a}(\la)]&=-\frac{1}{2}[\sigma^{3}_{a}\otimes\Id_{V_{q}},\Mon_{a}(\la)].
	\label{comm_S3_qsp_to_aux}
\end{align}
In particular this relation contains the commutators:
\begin{subequations}
\begin{align}
	[S^{3},\tf(\la)]&=0
	\label{S3_comm_tf}
	\shortintertext{and}
	[S^{3},\opB(\la)]&=-\opB(\la)
	.
	\label{S3_comm_opB}
\end{align}
\end{subequations}
The extended $\mathfrak{su}_2$ symmetry \eqref{xxx_ham_iso} of the XXX model means that the relation like in \cref{comm_S3_qsp_to_aux} holds for all of the $\mathfrak{su}_{2}$ generators,
\begin{align}
 	[\Id_{a}\otimes S^{\alpha}, \Mon_{a}(\la)]=-\frac{1}{2}[\sigma^{\alpha}_{a}\otimes \Id_{\qsp}, \Mon_{a}(\la)]
 	\label{XXX_comm_mon}
\end{align}
which gives us many more commutators which are equivalent to \cref{S3_comm_tf,S3_comm_opB}. Following are few important examples of these commutators:
\begin{subequations}
\begin{align}
	[S^{\alpha}, \tf(\la)]&=0, \quad (\forall \alpha)
 	\label{XXX_comm_tf}
	\\
 	[S^{+}, \opB(\la)]&=\opA(\la)-\opD(\la).
 	\label{XXX_comm_S+}
 	\\
 	[S^{-}, \opC(\la)]&=\opD(\la)-\opA(\la).
 	\label{XXX_comm_S-}
\end{align}
\label{XXX_comm_block_ops}
\end{subequations}
Here $S^\pm$ denotes the total raising and lowering operators which are defined through \cref{tot_spin_op,pauli_low-raising}.
Furthermore, we can also see that these the lowering and raising operators $S^\pm$ in the case of XXX model can be written as limit of $\opB$ and $\opC$ operators with infinite spectral parameters in appropriate normalisation
\begin{subequations}
\begin{align}
	S^{-}&=-i\lim_{\la\to\infty}\la\opB(\la),
	\label{low_ops_xxx}
	\\
	S^{+}&=-i\lim_{\la\to\infty}\la\opC(\la).
	\label{raise_ops_xxx}
\end{align}
	\label{low_raise_ops_xxx}
\end{subequations}
This symmetry of the XXX model plays an important role in our computations and it will be invoked whenever it would be necessary to do so. 
\subsection{Bethe equations}
The trace identity \eqref{tr_id_ham} and the commutator \eqref{S3_comm_tf}, both tell us that the spectral problem for the Hamiltonian \eqref{xxz_ham} can be resolved by finding the eigenvectors of the transfer matrix and $S^{3}$ instead of the Hamiltonian $H_\Delta$.
This is one of the founding arguments of the algebraic Bethe ansatz that allows us to reproduce the Bethe equations.
\begin{defn}[\emph{Reference} vector]
We define the reference vector $\pvac$ as
\begin{align}
	\pvac&=
	\begin{pmatrix} 1 \\ 0 \end{pmatrix}_{1}
	\otimes
	\cdots
	\otimes
	\begin{pmatrix} 1 \\ 0 \end{pmatrix}_{M}
	.
	\label{ferro_vac}
\end{align}
Sometimes it will be also referred as \emph{ferromagnetic} vacuum vector since it represents the fully magnetised, ferromagnetic ground state in the regime $\Delta<-1$ of the XXZ model.
\end{defn}
We can check that the reference vector $\pvac$ is an eigenvector of the transfer matrix.
In fact, it is an eigenvector of the diagonal blocks of operators $\opA$ and $\opD$ separately whereas the block operator $\opC$ annihilates this vector since,
\begin{align}
	\opA(\la)\pvac &= \pvac,
	&
	\opD(\la)\pvac &= r(\la)\pvac,
	&
	\opC(\la)\pvac &= 0
	.
\end{align}
Similarly, we can write the vector which is dual to \eqref{ferro_vac} and show that
\begin{align}
	\pvac*\opA(\la) &= \pvac*,
	&
	\pvac*\opD(\la) &= r(\la)\pvac*,
	&
	\pvac*\opC(\la) &= 0.
\end{align}
The eigenvalue function $r(\la)$ in these expressions is given by the following expression:
\begin{align}
	r(\la)&=\left(\frac{\varphi\left(\la-\frac{i\gamma}{2}\right)}{\varphi\left(\la+\frac{i\gamma}{2}\right)}\right)^M
	.
	\label{ev_opD_pvac}
\end{align}
\index{aba@\textbf{Algebraic Bethe ansatz (ABA)}!momentum expo@$r$: exponentiation of the bare momentum|textbf}%
The function $r(\la)$ is always a model dependant function, for the spin chains it depends implicitly on the anisotropy parameter $\Delta$ through the identity \eqref{tr_id_aniso}.
Meanwhile, we can also see that the operators $\opB$ and $\opC$ can act as raising and lowering operators respectively, in the Fock space built upon the reference vector $\pvac$ (or vice-versa in the dual space). 
This permits us to define the following type of vectors $\ket{\psi(\bm\la)}$ in the Fock space generated by these lowering-raising operators and which are built upon the reference vector $\pvac$.
\begin{defn}[\emph{Bethe vector}]
\sloppy Given a set of $N$ distinct (complex) spectral parameters $\bm\la =\allowbreak \set{\la_{1},\ldots,\la_{N}}$, we look for the eigenvectors (or their duals) of the transfer matrix which have the form
\begin{subequations}
\begin{align}
	\ket{\psi(\bm\la)}=\bmprod \opB(\bm \la)\pvac&= \opB(\la_{1})\cdots\opB(\la_{N})\pvac,
	\label{bv_arbit}
	\shortintertext{or in the case of dual,}
	\bra{\psi(\bm\la)}=\pvac*\bmprod \opC(\bm \la)&= \pvac*\opC(\la_{1})\cdots\opC(\la_{N}).
	\label{bv_arbit_dual}
\end{align}
	\label{bv_arbit_all}
\end{subequations}
Note that here we are invoking for the first time in this chapter the \emph{index-free} notation for the set $\bm\la$ and the product $\bmprod$ over this set. 
This was also defined on \cpagerefrange{ind_free_notn}{ind_free_notn_end}.
\end{defn}
As the quantum space can be written as direct sum with respect to the action of the third component of the total spin operator $S^{3}$ as
\begin{align}
	\qsp&=\mathop{\bigoplus}_{\ell=1}^{M} \qsp^{(\ell)},
	&
	\qsp^{(\ell)}=\Set{\Ket{\psi}\in\qsp | S^{3}\ket{\psi}=\left(\frac{M}{2}-\ell\right)\ket{\psi}}.
	\label{qsp_decomp_ev_s3}
\end{align}
We can see from \cref{S3_comm_opB}, that any Bethe vector in the subspace $\qsp^{(\ell)}$ is determined by the set $\bm\la$ of spectral parameters with the cardinality $n_{\bm\la}=N_\ell$, where it is given by,
\index{exc@\textbf{Excitations}!DL@\textbf{- Destri-Lowenstein (DL)}!num roots@$N_s$: often denotes the cardinality of on-shell Bethe roots}%
\begin{align}
	N_\ell = \frac{M}{2}-\ell.
\end{align}
\par
Let us now use the commutation relations in \cref{fcr_aba} to write the action of the diagonal block operators $\opA$ and $\opD$ on an arbitrary Bethe vector.
\begin{lem}
\label{lem:bae_gen}
The action of the block operators $\opA(\mu)$ and $\opD(\mu)$ on the Bethe vector $\ket{\psi(\bm\la)}$ is described by,
\begin{subequations}
\begin{align}
	\opA(\mu)\Ket{\psi(\bm\la)}&=
	\Lambda_{A}(\mu|\bm \la) \ket{\psi(\bm\la)}
	-\sum_{a=1}^{N}
	\Lambda_{A,a}(\mu|\bm\la)
	\Ket{\psi(\bm\la_{\hat{a}}\bm\cup\set{\mu})}
	\label{action_opA_bv_decomp}
	\\
	\opD(\mu)\Ket{\psi(\bm\la)}&=
	\Lambda_{D}(\mu|\bm \la) \ket{\psi(\bm\la)}
	-\sum_{a=1}^{N}
	\Lambda_{D,a}(\mu|\bm\la)
	\ket{\psi(\bm\la_{\hat{a}}\bm\cup\set{\mu})}
	\label{action_opD_bv_decomp}
\end{align}
	\label{action_opAD_bv_decomp}
\end{subequations}
where,
\begin{align}
	\Lambda_{A}(\mu|\bm\la)&=
	\bmprod\frac{\varphi(\mu-\bm\la-i\gamma)}{\varphi(\mu-\bm\la)},
	&
	\Lambda_{D}(\mu|\bm\la)&=
	r(\mu)\bmprod\frac{\varphi(\mu-\bm\la+i\gamma)}{\varphi(\mu-\bm\la)}
	\label{action_opAD_bv_direct}
\end{align}
and,
\begin{subequations}
\begin{align}
	\Lambda_{A,a}(\mu|\bm\la)&=
	\frac{\varphi(i\gamma)}{\varphi(\la_{a}-\mu)}
	\bmprod\frac{\varphi(\la_{a}-\bm\la_{\hat{a}}-i\gamma)}{\varphi(\la_{a}-\bm\la_{\hat{a}})},
	\label{action_opA_bv_cross}
	\\
	\Lambda_{D,a}(\mu|\bm\la)&=
	r(\la_{a})
	\frac{\varphi(i\gamma)}{\varphi(\mu-\la_{a})}
	\bmprod\frac{\varphi(\la_{a}-\bm\la_{\hat{a}}+i\gamma)}{\varphi(\la_{a}-\bm\la_{\hat{a}})}
	.
	\label{action_opD_bv_cross}
\end{align}
	\label{action_opAD_bv_cross}
\end{subequations}
\label{lem:diag_action_bv}
Note that the index $\hat{a}$ in \cref{action_opAD_bv_decomp,action_opD_bv_cross} denotes the removal of one parameter $\bm{\la_{\hat{a}}}=\bm\la\setminus\set{\la_a}$ as defined in the summary of the index-free notations on \cpageref{ind_free_notn}.
\end{lem}
\begin{proof}
Let us first note that since the operators $\opB(\la_{a})$ commute with each other due to the commutator \eqref{FCR_self}, it does not matter the order in which the operator $\opA(\mu)$ or $\opD(\mu)$ is pushed through the products of operators $\opB$, i.e. we have the freedom to reorder the indices of this product.
With this remark in mind, let us push an operator $\opA(\mu)$ or $\opD(\mu)$ through the first operator $\opB(\la_{a})$ in the product for certain index $a\leq N$, by using the commutation relation found in \cref{FCR_AB_DB}.
After this process we get:
\begin{subequations}
\begin{multline}
	\opA(\mu)\prod\opB(\bm{\la})
	=
	\frac{1}{f(\la_{a}-\mu)}\opB(\la_{a})\opA(\mu)\bmprod\opB(\bm\la_{\hat{a}})\pvac
	\\
	-
	\frac{g(\la_{a}-\mu)}{f(\la_{a}-\mu)}\opB(\mu)\opA(\la_{a})\bmprod\opB(\bm\la_{\hat{a}})\pvac
	\label{action_opA_bv_first}
\end{multline}
{and similarly,}
\begin{multline}
	\opD(\mu)\prod\opB(\bm{\la})
	=
	\frac{1}{f(\mu-\la_{a})}\opB(\la_{a})\opD(\mu)\bmprod\opB(\bm\la_{\hat{a}})\pvac
	\\
	-
	\frac{g(\mu-\la_{a})}{f(\mu-\la_{a})}\opB(\mu)\opD(\la_{a})\bmprod\opB(\bm\la_{\hat{a}})\pvac
	.
	\label{action_opD_bv_first}
\end{multline}
	\label{action_opAD_bv_first}
\end{subequations}
Notice that there are two types of terms produced after each commutation, the first type of terms exchanges the operators without exchanging the spectral parameters whereas the second type of terms also exchanges the spectral parameters. But since $\opB(\la_{a})$ commute among themselves, the action of these operators must lie in the eigenspace spanned by $\ket{\psi(\bm\la)}$ and $\ket{\psi(\mu,\bm\la_{\hat{a}})}$ and hence it has the decomposition given by \cref{action_opA_bv_decomp,action_opD_bv_decomp}, which signify the same thing while they are written in a compact notation.
\par
Since the vector $\ket{\psi(\bm\la)}$ can only be obtained by taking cross terms without any exchange of spectral parameters as we commute the operator $\opA(\mu)$ or $\opD(\mu)$ with the product of $\opB(\la_{a})$, its coefficient $\Lambda_{A}$ or $\Lambda_{D}$ is given by the product of direct term coefficients with the eigenvalues of $\opA(\mu)$ or $\opD(\mu)$ for the reference vector. On the other hand, the remaining vector with coefficients $\Lambda_{A,a}$ or $\Lambda_{D,a}$ can be obtained by taking a cross-term exchanging the parameters at the very first step as we have seen in \cref{action_opAD_bv_first}. After this, we only need to take the direct terms without any exchange for all the remaining spectral parameters. This gives us the expressions that were shown in \cref{action_opAD_bv_direct,action_opAD_bv_cross} for the coefficients $\Lambda_{A,a}$ and $\Lambda_{D,a}$, where we used the expression \eqref{rmat_wts_paramn_gen} for the weight functions $f$ and $g$ of the $\Rm$-matrix.
\end{proof}
\par
\begin{thm}[{\emph{Bethe equations}} \cite{FadT84}]
A Bethe vector as defined in \cref{bv_arbit} can be an eigenvector of the transfer matrix only if its spectral parameters satisfy the following set of equations:
\begin{flalign}
	(a\leq N=n_{\bm\la}),
	&&
	r(\la_{a})
	\bmprod \frac{\varphi(\la_{a}-\bm\la+i\gamma)}{\varphi(\la_{a}-\bm\la-i\gamma)}
	&=
	-1
	.
	&&
	\label{bae_gen}
\end{flalign}
The eigenvalue of the transfer matrix for such a vector is given by the sum $\Lambda_{A}+\Lambda_{D}$ which gives rise to the following expression:
\index{aba@\textbf{Algebraic Bethe ansatz (ABA)}!transfer matrix ev@$\evtf$: eigenvalue of transfer matrix|textbf}%
\begin{align}
	\evtf(\mu)=
	\bmprod\frac{\varphi(\mu-\bm\la-i\gamma)}{\varphi(\mu-\bm\la)}+
	r(\mu)
	\bmprod\frac{\varphi(\mu-\bm\la+i\gamma)}{\varphi(\mu-\bm\la)}
	.
	\label{evtf_gen}
\end{align}
\end{thm}
\begin{proof}
	It follows from the \cref{lem:diag_action_bv} that for $\ket{\psi(\bm\la)}$ to be an eigenvector of the transfer matrix $T(\la)=\opA(\la)+\opD(\la)$, we must have $\Lambda_{A,a}+\Lambda_{D,a}=0$ $(\forall a)$. From \cref{action_opAD_bv_direct,action_opAD_bv_cross}, we obtain the necessary condition: 
	\begin{align}
	\frac{\varphi(i\gamma)}{\varphi(\la_{a}-\mu)}
	\bmprod\frac{\varphi(\la_{a}-\bm\la_{\hat{a}}-i\gamma)}{\varphi(\la_{a}-\bm\la_{\hat{a}})}
	+
	r(\la_{a})
	\frac{\varphi(i\gamma)}{\varphi(\mu-\la_{a})}
	\bmprod\frac{\varphi(\la_{a}-\bm\la_{\hat{a}}+i\gamma)}{\varphi(\la_{a}-\bm\la_{\hat{a}})}
	=0
	.
	\label{tfmat_ev_bv_gen}
	\end{align}
	This can be simplified to write the equation in the form \eqref{bae_gen}, here we use the fact that $\varphi$ \eqref{rmat_paramn_types} is always an odd function and assume that there are no singular terms\footnote{see \cref{sub:admissible_bae} on admissibility} in the above expression. The eigenvalue can then be readily seen as a sum of the coefficients $\Lambda_{A}+\Lambda_{D}$.
\end{proof}
\begin{rem}	
A similar result can be also derived for the dual off-shell Bethe vector \eqref{bv_arbit_dual}.
The Bethe equations \eqref{bae_gen} and the expression \eqref{evtf_gen} eigenvalue of the transfer matrix are identical in the dual case.
\end{rem}
\begin{notn}[Auxiliary function]
\label{def:exp_cfn_fn}
\index{aux@\textbf{Auxiliary functions}!exp cfn@$\aux$: exponential counting function|textbf}%
Let us define an auxiliary function $\aux$ as follows:
\begin{align}
	\aux(\mu|\bm\la)&=
	r(\mu)
	\bmprod_{\bm\la}
	\frac{%
	\varphi(\mu-\bm\la+i)
	}{%
	\varphi(\mu-\bm\la-i)
	}
	\label{exp_cfn_gen_def}
\end{align}
In terms of this function the Bethe equations \eqref{bae_gen} can be recast in a compact form:
\begin{flalign}
	(\forall a\leq n_{\bm\la}),
	&&
	1+\aux(\la_{a}|\bm\la)&=0
	.
	&&
	\label{bae_gen_aux_form}
\end{flalign}
\end{notn}
\minisec{Energy and momentum of the magnons}
As the Hamiltonian is included through the trace identity \eqref{tr_id_ham} in the Bethe sub-algebra generated by transfer matrix \eqref{def_tfmat}, the eigenstates found by solving the Bethe equations with admissible solutions are the eigenstates of the Hamiltonian, as desired. The energy and momentum eigenvalue of such an eigenstate can be computed from \cref{tfmat_ev_bv_gen}. These are given the following expressions:
\begin{subequations}
\begin{align}
	H_{\Delta}\ket{\psi(\bm\la)}&=
	J\bmsum\varepsilon_{0}(\bm\la|\Delta),
	&
	\varepsilon_{0}(\la|\Delta)&=
	\begin{dcases}
	\frac{-2}{\la^2+\frac{1}{4}},	& \Delta=1;
	\\
	\frac{-2\sin^2\gamma}{\sinh(\la+\frac{i\gamma}{2})\sinh(\la-\frac{i\gamma}{2})}
	,
	&|\Delta|<1;
	\\
	\frac{-2\sinh^2\gamma}{\sin(\la+\frac{i\gamma}{2})\sin(\la-\frac{i\gamma}{2})}
	,
	&\Delta>1.
	\end{dcases}
	\label{bare_energy_gen}
	\shortintertext{And}
	P\ket{\psi(\bm\la)}&=
	J\bmsum p_{0}(\bm\la|\Delta),
	&
	p_{0}(\la|\Delta)&=
	-i\log\left(\frac{\varphi(\la+\frac{i\gamma}{2})}{\varphi(\la-\frac{i\gamma}{2})}\right)
	+\pi
	\nonumber
	\\[.5em]
	&&\text{or,} \quad
	p_{0}(\la|\Delta)&=
	\begin{dcases}
	-2\arctan(2\la)+\pi,
	&\Delta=1;
	\\
	-2\arctan\left(\frac{\tanh\la}{\tan(\frac{\gamma}{2})}\right)+\pi,
	&|\Delta|<1;
	\\
	-2\arctan\left(\frac{\tan\la}{\tanh(\frac{\gamma}{2})}\right)+\pi,
	&\Delta>1.
	\end{dcases}
	\label{bare_mom_gen}
	\end{align}
	\label{ev_energy_mom_gen}
\end{subequations}
Combining these two expression we obtain the dispersion relation for the magnons, which can be presented as
\begin{align}
	\epsilon_{0}(\la|\Delta)=\cos p_{0}(\la)-\Delta
	.
	\label{disp_rel_xxz}
\end{align}
Note that this dispersion relation is not a non-negative function in the range of anisotropies $\Delta>-1$ that we are interested in.
This indicates that the magnon excitations are not true excitations of the XXZ model for $\Delta>-1$ but rather, they represent the pseudo-excitations over the reference vector $\pvac$ which are created by the action of $\opB(\la_{a})$.
Consequently, the reference vector $\pvac$ is not the true ground state of our model either, as we shall see in \cref{chap:spectre}, the latter is of composed of a large number of pseudo-particles forming a sea of magnons.
\minisec{Logarithmic form of the Bethe equations}
\label{sub:log_bae_gen}
Let us take the logarithm of \cref{bae_gen} with appropriately chosen branch cuts.
This gives us the logarithmic form of the Bethe equations:
\begin{subequations}
\begin{flalign}
	(\forall a\leq N=n_{\bm\la}),
	&&
	M\Theta_{1}(\la_{a}|\Delta)-\sum_{k=1}^{N}\Theta_{2}(\la_{a}-\la_{k}|\Delta)
	&=
	2\pi Q_{a}
	.
	&&
	\label{log_bae}
\end{flalign}
On the right hand side we obtain the quantum numbers that form a set $\bm Q\subset\frac{1}{2}\Zset$ of either all integers or all half-integers, depending upon the value of $N$ and $M$.
The function $\Theta_\kappa$ in \cref{log_bae} is defined piecewise in different regimes according to the value of the anisotropy parameter $\Delta$ as follows:
\begin{align}
	\Theta_{\kappa}(\la|\Delta)&=
	\begin{dcases}
		2\arctan\left(\frac{2\la}{\kappa}\right),	&	\Delta=1;
		\\
		2\arctan\left(\frac{\tanh\la}{\tan\left(\tfrac{\kappa\gamma}{2}\right)}\right),
		&	|\Delta|<1;
		\\
		2\arctan\left(\frac{\tan\la}{\tanh\left(\tfrac{\kappa\gamma}{2}\right)}\right),
		& \Delta>1.
	\end{dcases}
	\label{log_bae_funs_gen}
\end{align}
	\label{log_bae_all}
\end{subequations}
\index{misc@\textbf{Miscellaneous functions}!THETA fn@$\Theta_\kappa$: terms in the logarithmic Bethe \cref{log_bae_all,log_bae_xxx}|textbf}%
The branch cut of the logarithm in each of these expressions for $\Theta_\kappa$ is chosen such that we have a continuous branch on the real line.
One important consequence of this choice is demonstrated through \cref{fig_branch_cuts}, which clearly shows that there are two distinct types of the complex roots based on the structure of branch cuts chosen for these roots.
As we shall see in \cref{chap:spectre}, this difference has important implications for the nature of complex roots.
\begin{figure}[tb]
\centering
\includegraphics[width=\textwidth, height=0.2\paperheight]{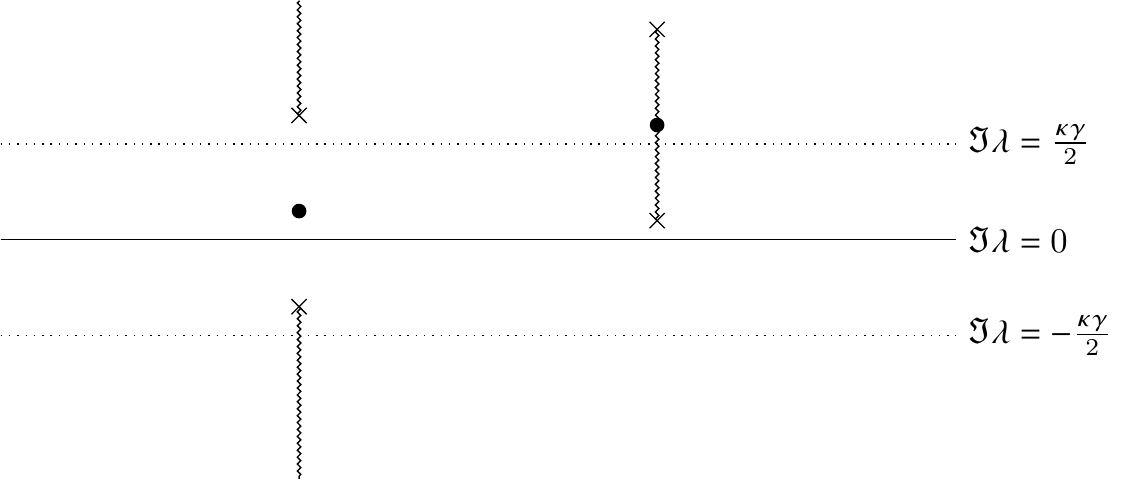}
\caption[Choice of branch cuts of the function $\Theta_\kappa$.]{The choice of branch cuts of the function $\Theta_\kappa(\la)$. The points `\tikzcircfill' denote the parameter $\la$ whereas the cross `\tikzcross'
denote a branch point of the function $\Theta_\kappa(\la)$. The branch cut that is chosen is denoted with a snaked line `\raisebox{.4ex}{\tikzzigzag}'.}
\label{fig_branch_cuts}
\end{figure}
We can also see that this distinction is solely driven by the relative placement of the branch points, vis-à-vis the real line.
\par
The term with $\Theta_{1}(\la|\Delta)$ in \cref{log_bae} is related to the bare momentum \eqref{bare_mom_gen} whereas the terms $\Theta_{2}(\la_{a}-\la_{k}|\Delta)$ denote the scattering terms.
This observation prompts us to define the counting function as shown below.
\begin{defn}
\index{aba@\textbf{Algebraic Bethe ansatz (ABA)}!counting function@$\cfn$: counting function|textbf}%
\label{defn_cfn_gen}
Given a set of Bethe roots $\bm\la$, we define the \emph{counting function} $\cfn(\nu|\bm\la)$ as follows:
\begin{align}
	\cfn(\nu|\bm\la,\Delta)=\frac{1}{2\pi}\Theta_{1}(\nu|\Delta)-\frac{1}{2\pi M}\bmsum\Theta_{2}(\nu-\bm\la|\Delta)
	\label{cfn_gen}
	.
\end{align}
\end{defn}
Let us pause here to remark that the auxiliary function defined in \cref{exp_cfn_gen_def} earlier can also be viewed as the \emph{exponential counting function}, since we have the following relation among these two:
\begin{align}
	\aux(\nu|\bm\la)=
	e^{2\pi i M\cfn(\nu|\bm\la)}
	.
	\label{exp_cfn_gen}
\end{align}
The logarithmic Bethe equations \eqref{log_bae} can be expressed in a compact form in terms of the counting function $\cfn$ as follows:
\begin{flalign}
	(\forall a)&& \cfn(\la_{a}|\bm\la)&=\frac{1}{M} Q_{a}
	.
	&&
\end{flalign}
The monotonicity of the counting function means that all the real Bethe roots are in one-to-one correspondence with quantum numbers.
The question of the complex roots is more delicate, it will be addressed partly in \cref{sub:string_pic} of this chapter and revisited again in \cref{chap:qism}.
\subsection{Admissibility of solutions}
\label{sub:admissible_bae}
The Baxter polynomial%
\footnote{in the anisotropic case $\Delta\neq 1$, we mean here the trigonometric or hyperbolic polynomials, i.e. the functions which can be written as polynomials in $q$-variables which are exponentials of the spectral parameters}
$q(\cdot|\bm\la)$ \cite{Bax89} for a given set of spectral parameters $\bm\la$ is defined as
\index{aba@\textbf{Algebraic Bethe ansatz (ABA)}!Baxter polynomial@$\baxq$: Baxter polynomial|textbf}%
\begin{align}
	\baxq(\mu|\bm\la)&=
	\bmprod\varphi(\mu-\bm\la)
	.
	\label{baxq}
\end{align}
Let us also remark that the exponential counting function can be expressed in terms of the Baxter polynomials $q$ as
\begin{align}
	\aux(\mu|\bm\la)&=r(\mu)\frac{q(\mu+i\gamma|\bm\la)}{q(\mu-i\gamma|\bm\la)}.
	\label{bae_aux_gen}
\end{align}
We will often drop the implicit dependence on the set of Bethe roots $\bm\la$ in the Baxter polynomials \eqref{baxq} as well as in the exponential counting function \eqref{bae_aux_gen} whenever it is made clear from the context. 
\par
We can see that the Bethe equations $1+\aux(\la_a)=0$ \eqref{bae_gen} arise as the necessary condition for the following polynomial factorisation problem known as the $T$-$Q$ relation to be solved.
\begin{align}
	\baxq(\mu)t(\mu)=
	\varphi^{M}\left(\mu-\tfrac{i\gamma}{2}\right)\bmprod\varphi(\mu-\bm\la+i\gamma)
	+
	\varphi^{M}\left(\mu+\tfrac{i\gamma}{2}\right)\bmprod\varphi(\mu-\bm\la-i\gamma)
	.
	\label{bax_tq_rel}
\end{align}
The factorised polynomial $t$ in this case turns out to be nothing but the eigenvalue of the transfer matrix \eqref{evtf_gen} which is rescaled to a polynomial form:
\begin{align}
	t(\la)=\varphi^{M}\left(\la+\tfrac{i\gamma}{2}\right)\,\evtf(\la).
	\label{baxq_t_fn}
\end{align}
\par
Although the Bethe equation \eqref{bae_gen} gives the sufficient condition under which an arbitrary Bethe vector \eqref{bv_arbit} can become an eigenvector of the transfer matrix, it is not necessary that the solutions to \cref{bae_gen} lead to a non-trivial eigenvector of the transfer matrix.
One of the possible counter-examples can be seen from the violation of distinctness criteria.
Let us recall that we have already made use of this property in \cref{lem:bae_gen}, however we can see that the system of Bethe equations \eqref{bae_gen} in itself does not forbid such a solution.
\Textcite{MukTV09} demonstrated that the $T$-$Q$ relation from \cref{bax_tq_rel} can be used to partly address this particular problem.
\\
Another problem arises when we have an exact string pair with two roots whose difference is $\la_{a}-\la_{b}=\pm i\gamma$.
Since these coincide with the poles or zeroes of the exponential counting function, we must discard such singular solutions.
However, this gets tricky for the exact string pair $\set{\frac{i\gamma}{2},\frac{-i\gamma}{2}}$ at the origin and it could be allowed in \cref{bae_gen} if the remaining roots satisfy it, even for the model at finite chain length $M$.
This was studied by \textcite{NepW13}.
It turns out that the corresponding vector is a null vector with vanishing norm since one can show that following product is a null operator,
\begin{align}
	\opB\left(\tfrac{i\gamma}{2}\right)\opB\left(-\tfrac{i\gamma}{2}\right)=0.
\end{align}
However, this does represent the complete picture as we can still achieve a normalised vector with an exact string $\Set{\frac{i\gamma}{2},-\frac{i\gamma}{2}}$ by taking the limit such that it is non-singular. But we find that such a limit is not well defined unless we have some additional information that dictates the way in which such a limit is taken.
The limit from the eigenvector of the twisted transfer matrix is an example of such a case where it can be properly defined.
We will not further devolve into this issues, the reader may refer to \cite{AvdV86}.
Here in this thesis we will always impose the following admissibility criteria on the solutions $\bm\la$ of \cref{bae_gen}:
\begin{enumerate}[noitemsep]
	\item All the roots are pairwise distinct;
	\item They do not contain the zeroes or the poles of the exponential counting function $\aux$ \eqref{exp_cfn_gen_def};
	\item The exact string solution, even if it could be otherwise permitted when written in polynomial form, is to be discarded.	
\end{enumerate}
These conditions can be summarised as
\begin{align}
	\la_{a}-\la_{b} \notin \set{0, i\gamma, -i\gamma}; \quad \forall a,b \leq N.
\end{align}
The solutions satisfying these conditions are called \emph{admissible Bethe roots} or simply, the Bethe roots. Correspondingly, a Bethe vector given by an admissible set of Bethe roots will be called \emph{on-shell} Bethe vector and \emph{off-shell} otherwise.
\minisec{Complex roots}
\label{sub:string_pic}
The Bethe equations \eqref{bae_gen} can also admit non real complex roots.
In the thermodynamic limit $M\to\infty$ (see \cref{chap:spectre}), a simplified description can be found for the complex roots if we make certain assumptions.
The most common description is the one given by the \emph{string hypothesis} which tells us that complex Bethe roots in the thermodynamic limit form 
string complexes of the following type:
\begin{align}
	\mix{\la}[a]^{(\ell)}_{j}&=
	{z_{a}^{(\ell)}+i\left(\tfrac{\gamma}{2}(\ell+1-2j)+\mix{\stdv}_{a}^{j}\right)}
	,
	&
	j&= 1, 2, \ldots, \ell
	.
	\label{str_cmplx}
\end{align}
The parameter $z^{(\ell)}_{a}\in\Rset$ is the centre of the string of length $\ell$ and the parameters $\stdv_{a}$ characterise the deviation from the ideal string.
In this picture, the real roots can be seen as strings of length $1$ with exactly zero string deviation.
The deviation parameters for complex strings are assumed to be exponentially small for large chain lengths $M$ which we need to balance the $r(\la_a)$ term in the Bethe equation \eqref{bae_gen}.
With the string hypothesis, the logarithmic form of the Bethe equation is:
\begin{subequations}
\begin{align}
	M\Theta_{\ell}\left(z_{a}^{(\ell)}\right)
	-
	\sum_{k\geq1}\bmsum
	\Theta_{\ell, k}(z^{(\ell)}_{a}-\bm{w^{k}})
	=
	2\pi Q^{(\ell)}_{a}
	\label{log_bae_str}
\end{align}
where,
\index{misc@\textbf{Miscellaneous functions}!THETA fn@$\Theta_\kappa$: terms in the logarithmic Bethe \cref{log_bae_all,log_bae_xxx}}%
\begin{align}
	\Theta_{\ell,k}(\la)
	&=
	\sideset{}{'}\sum_{|\ell-k|\leq r\leq \ell+k}
	\Theta_{r}\left(\la\right)
	.
	\label{log_bae_str_funs}
\end{align}
	\label{log_bae_str_all}
\end{subequations}
The primed sum $\sum'$ omits the singularity term for $r=0$ (if it is present). In contrast to the \cref{log_bae}, here we get different flavours of quantum numbers $\bm{Q^{(\ell)}}$ depending on the length $\ell$ of the string.
\par
We will revisit the string hypothesis in \cref{sec:spectre_XXX} of the next chapter where it is discussed in the context of the XXX model.
Let us note that although it is widely used, the string hypothesis remains contentious as it makes a very strong assumption.
The violations of string hypothesis has been investigated in numerous works \cite{HagC07,Vla84,EssKS92} and alternate description that do not rely \textit{a priori} on the string hypothesis are also presented in \cite{DesL82,BabVV83}.
In-fact we use Destri-Lowenstein picture \cite{DesL82} in our computations, it will be presented here in \cref{sub:DL_picture} of the next chapter.
\subsection{Decomposition of the XXX spectrum into multiplets}
\label{sub:XXX_mult}
We know that the XXX model enjoys a larger $\mathfrak{su}_2$ symmetry which manifests itself in the form of identities shown in \cref{XXX_comm_mon,XXX_comm_block_ops}.
Using these identities we can show that an on-shell Bethe vector $\ket{\psi(\bm\la)}$ (or its dual) is annihilated by the natural action of the global raising operator $S^+$ (or $S^-$ in the case of the dual) in the XXX model.
\begin{align}
	S^{+}\ket{\psi(\bm\la)}&=0,
	\label{bv_highest_wt_xxx}
	&
	\bra{\psi(\bm\la)}S^{-}&=0.
\end{align}
The eigenvalues of the third component of the total spin operator $S^3$ and Casimir operator $C$ of the global $\mathfrak{su}_2$ algebra for an on-shell Bethe vector $\ket{\psi(\bm\la)}$ are given by the following expressions.
\begin{subequations}
\begin{align}
	S^{3}\ket{\psi(\bm\la)}&
	=\left(\frac{M}{2}-n_{\bm\la}\right)\ket{\psi(\bm\la)}
	=s\ket{\psi(\bm\la)},
	\label{bv_ev_S3}
	\\
	C\ket{\psi(\bm\la)}
	&=\left(\frac{M}{2}-n_{\bm\la}\right)\left(\frac{M}{2}-n_{\bm\la}+1\right)\ket{\psi(\bm\la)}
	=s(s+1)\ket{\psi(\bm\la)}
	.
	\label{bv_ev_cas}
\end{align}
	\label{bv_ev_S3_cas}
\end{subequations}
The eigenvalue $s$ in these expression is the \emph{total spin} of the vector $\ket{\psi(\bm\la)}$.
\par
As a consequence of the symmetry, the eigenspaces of the XXX chain for transfer matrix form the irreducible representations of $\mathfrak{su}_{2}$ which we call the \emph{multiplets} and the on-shell Bethe vector have the highest weight in a multiplet.
An interesting consequence of this is that the number of Bethe roots can never exceed $\frac{M}{2}$.
We introduce the following notations that takes into account these factors.
\begin{notn}		
For a given value of the total spin $s$, let us define integer $N_s$ according to the following expression
\index{exc@\textbf{Excitations}!DL@\textbf{- Destri-Lowenstein (DL)}!num roots@$N_s$: often denotes cardinality of on-shell Bethe roots|textbf}%
\index{exc@\textbf{Excitations}!DL@\textbf{- Destri-Lowenstein (DL)}!nspin@$s$: total spin|textbf}%
\begin{align}
	N_{s}&=\frac{M}{2}-s
	\qquad
	\left(
	0\leq s\leq \frac{M}{2}
	\right)
	.
	\label{num_roots_xxx}
\end{align}
\end{notn}
\begin{notn}
\index{exc@\textbf{Excitations}!DL@\textbf{- Destri-Lowenstein (DL)}!xxx_vec@$\ket{\psi^\ell_s}$: a vector of a XXX multiplet|textbf}%
\label{notn:descendant_xxx_notn}
Let $\bm\la$ be a set of Bethe roots of the cardinality $n_{\bm\la}=N_s$.
We use the notation $\ket{\psi_{s}^{\ell}}$ to denote the set of vectors:
\begin{align}
	\ket{\psi_s^\ell}
	&=
	(S^{-})^{\ell}\ket{\psi(\bm\la)}
	&&
	\text{where,}
	\quad
	0\leq \ell \leq 2s.
	.
	\label{descendant_xxx_notn}
\end{align}
\end{notn}
The dimension of the multiplet generated by $\ket{\psi_{d}^{\ell}}$ is $d=2s+1$. This will be called a $s$-multiplet.
For the first few values of $s$, we shall give the names \emph{singlet} $(s=0)$, \emph{triplet} $(s=1)$ and \emph{quintplet} $(s=2)$ and so on.
Due to the identity given in  \cref{low_raise_ops_xxx}, the action of lowering operator in \cref{descendant_xxx_notn} is synonymous with
\begin{subequations}
\begin{align}
	\ket{\psi_s^\ell(\bm\la)} &= (-i)^{\ell}
	\lim_{\bm\mu\to\infty}
	\left(
	\bmprod \bm\mu\,
	\opB(\bm\mu)
	\right)
	\ket{\psi(\bm\la)}
	,
	\label{descendants_xxx}
	\\
	\bra{\psi^\ell_s(\bm\la)} &= (-i)^{\ell}
	\lim_{\bm\mu\to\infty}
	\bra{\psi(\bm\la)}
	\left(
	\bmprod \bm\mu\,
	\opC(\bm\mu)
	\right)
	.
	\label{descendants_xxx_dual}
\end{align}
	\label{descendants_xxx_both}
\end{subequations}
The multiplets are degenerate eigenspaces of the transfer matrix since the latter commutes with the lowering operator \eqref{XXX_comm_S-}, therefore all the vectors have the same energy and momentum eigenvalues.
Meanwhile the eigenvalues of the Casimir operator and $S^{3}$ for the vectors $\ket{\psi^\ell_s}$ of a multiplet are as follows:
\begin{subequations}
\begin{align}
	S^{3}\ket{\psi_s^\ell(\bm\la)}
	&=
	(s-\ell)\ket{\psi_s^\ell(\bm\la)},
	&
	\bra{\psi_s^\ell(\bm\la)}S^{3}
	&=
	(s-\ell)\bra{\psi_s^\ell(\bm\la)};
	\label{ev_S3_descendant}
	\\
	C\ket{\psi_s^\ell(\bm\la)}
	&=
	s(s+1)\ket{\psi_s^\ell(\bm\la)},
	&
	\bra{\psi_s^\ell(\bm\la)}C
	&=
	s(s+1)\bra{\psi_s^\ell(\bm\la)}.
	\label{ev_cas_descendant}
\end{align}
	\label{ev_S3_cas_descendant}
\end{subequations}
Decomposition of the XXX spectrum into degenerate multiplets has many important consequences. In \cref{sec:det_rep_ff} we will see how this symmetry can be utilised to reduce the complexity in our computations of the form-factors.
We will also study the nature of the eigenvectors of the XXX model in the next \cref{chap:spectre}.
There we characterise the \emph{ground state} of the XXX model as well as its excitations.
For the rest of the discussion in this chapter, we only need to know that the ground state of the XXX model in the absence of any external field is a uniquely determined by a singlet $s=0$ Bethe vector,  
denoted by $\ket{\psi_g}$.
\section{Quantum inverse scattering and determinant representation}
\label{sec:qism_det_rep}
This path from the spin chain Hamiltonian to transfer matrix in the algebraic Bethe ansatz framework is comparable to the direct scattering part in the literature of classical integrable systems.
However, as we know from its classical analogue, this is only half the story, and the other half is the inverse map.
The quantum inverse scattering map that we will discuss now, is a key that gives access to the computations of physically relevant quantities for integrable models, notably the correlation functions and form-factors.
\par
A physical state of a quantum system is best described by a density matrix.
At thermal equilibrium at inverse temperature $\beta$, the density matrix is given by the exponential of the Hamiltonian $\varrho=e^{-\beta H}$.
The expectation value of an operator $\mathcal{O}$ is given by the trace:
\begin{align}
	\braket{\mathcal{O}}=\frac{\tr(\varrho\mathcal{O})}{\tr\varrho}.
\end{align}
All the computations in this thesis are carried in the zero temperature limit, where the expectation value of an operator is dominated by the ground state, hence
\begin{align}
	\braket{\mathcal{O}} \overset{T\to0}{\longrightarrow} \frac{\braket{\psi_g|\mathcal{O}|\psi_g}}{\braket{\psi_g|\psi_g}}
	.
\end{align}
In the case of spin chains, the local operator $\mathcal{O}$ can be expressed as a product of $n$ spin operators, the expectation value of this operator is the $n$-point correlation function.
In this thesis we will be mostly interested with the two-point longitudinal dynamic correlations and longitudinal form-factors both of which we will discuss again in \cref{sec:det_rep_ff}.
There are two main points that need to be discussed before we could move forward.
These are the two problems whose solutions pave the way for out computations, which are summarised in the following:
\begin{enumerate}[%
wide=\parindent 
]
	\item Computing the action of local operators on the on-shell Bethe vectors and to express the result as an off-shell Bethe vector.
	This is precisely what we mean by the \emph{quantum inverse scattering problem}.
	\item Finding exact representations for the scalar products of the on-shell and off-shell Bethe vector as well as the norms of the on-shell Bethe vector.
\end{enumerate}
We will first see the solution of the second problem in \cref{sub:det_rep_scal_prod} below.
The solution of the quantum inverse scattering problem for the XXZ model and its application will be discussed in \cref{sub:qisp}.
\subsection{Determinant representation for the scalar product}
\label{sub:det_rep_scal_prod}
Let us consider the scalar product of the two Bethe vectors:
\begin{align}
	\braket{\psi(\bm\mu)|\psi(\bm\la)}
	&=
	\Braket{\vac|\bmprod\opC(\bm\mu)\bmprod\opB(\bm\la)|\vac}
	.
\end{align}
The action of operators $\opC$ on the product of operators $\opB$ can be easily accessed from \cref{FCR_BC}, which leads to the following lemma.
\begin{lem}
	The left action (or, rightward action) of the operator $\opC(\mu)$ on the Bethe vector $\ket{\psi(\bm\la)}$ (let $N=n_{\bm\la}$) is given by,
\begin{subequations}
\begin{align}
	\opC(\mu)\Ket{\psi(\bm\la)}
	=
	\sum_{\underset{k\neq j}{j,k=1}}^{N+1}
	r(\check\la_j)
	\frac{\bmprod\varphi(\check\la_{j}-\bm\la+i\gamma)}{\bmprod\varphi(\check\la_j-\bm{\check\la_{\hat{j}}})}
	\frac{\bmprod\varphi(\check\la_{k}-\bm{\la_{\hat{j}}}-i\gamma)}{\bmprod\varphi(\check\la_k-\bm{\check\la_{\hat{j},\hat{k}}})}
	\Ket{\psi(\bm{\check\la}_{\hat{j},\hat{k}})}
	\label{action_C_bv}
\end{align}
and similarly the right action (or, the leftward action) of the operator $\opB(\mu)$ on the dual Bethe vector $\bra{\psi(\bm\la)}$ is given by,
\begin{align}
	\bra{\psi(\bm\la)}\opB(\mu)
	=
	\sum_{\underset{k\neq j}{j,k=1}}^{N+1}
	r(\check\la_j)
	\frac{\bmprod\varphi(\check\la_{j}-\bm\la+i\gamma)}{\bmprod\varphi(\check\la_j-\bm{\check\la_{\hat{j}}})}
	\frac{\bmprod\varphi(\check\la_{k}-\bm{\la_{\hat{j}}}-i\gamma)}{\bmprod\varphi(\check\la_k-\bm{\check\la_{\hat{j},\hat{k}}})}
	\Bra{\psi(\bm{\check\la}_{\hat{j},\hat{k}})}
	\label{action_B_bv_dual}
	,
\end{align}
	\label{action_BC_bv}
\end{subequations}
where the set $\bm{\check\la}$ denotes the union $\bm{\check\la}=\bm\la\cup\set{\mu}$ with the additional parameter added at $\check\la_{N+1}=\mu$.
\label{lem_action_BC_bv}
\end{lem}
\begin{proof}
Let us start with the right action $\opC(\mu)\ket{\la}$. From the commutator in \cref{FCR_BC}, we can see that
\begin{multline}
	\left[\opC(\mu),\bmprod\opB(\bm\la)\right]
	=
	\\
	\sum_{k=1}^{N}
	\left\lbrace
	\frac{\varphi(i\gamma)}{\varphi(\mu-\la_k)}
	\left(\prod_{j=1}^{k-1}\opB(\la_j)\right)
	\left(\opA(\mu)\opD(\la_k)-\opA(\la_k)\opD(\mu)\right)
	\left(\prod_{j=k+1}^{N}\opB(\la_j)\right)
	\right\rbrace
	\label{comm_C_bv}
\end{multline}
The action of $\opA$ and $\opD$ on the product of $\opB$ was given in the set of \cref{action_opAD_bv_direct,action_opAD_bv_cross} of \cref{lem:diag_action_bv}.
From there we know that the action of diagonal block operators $\opA$ and $\opD$ take the form of either a direct action or it can be a cross action involving an exchange of a parameter.
From this we can already expect that we have an expansion of the type:
\begin{align}
	\opC(\mu)\ket{\psi(\bm\la)}=
	\sum_{k=1}^{N}\chi_{k}\Ket{\psi(\bm{\la_{\hat{k}}})}
	+
	\sum_{j<k}^{N}\chi_{j,k}\Ket{\psi(\bm{\la_{\hat{j},\hat{k}}}\cup\set{\mu})}.
	\label{action_C_bethe_expansion}
\end{align}
To compute the coefficients $\chi_{{k}}$ and $\chi_{{j},{k}}$ we will use a method which is analogous to the one used in \cref{lem:diag_action_bv}.
Since the operators $\opB$ commute among themselves [see \cref{FCR_self}], we can reorder the product $\bmprod\opB(\bm\la)$ to bring the operator $\opB(\la_k)$ to the front.
The coefficient $\chi_k$ can be then easily obtained from the direct terms for actions of the operators $\opA$ and $\opD$
\begin{subequations}
\begin{multline}
	\chi_k
	=
	\frac{\varphi(i\gamma)}{\varphi(\mu-\la_k)}
	\left\lbrace
	r(\la_k)
	\frac{\bmprod\varphi(\la_k-\bm{\la_{\hat{k}}}+i\gamma)}{\bmprod\varphi(\la_k-\bm{\la_{\hat{k}}})}
	\frac{\bmprod\varphi(\mu-\bm{\la_{\hat{k}}}-i\gamma)}{\bmprod\varphi(\mu-\bm{\la_{\hat{k}}})}
	\right.
	\\
	\left.
	-
	r(\mu)
	\frac{\bmprod\varphi(\mu-\bm{\la_{\hat{k}}}+i\gamma)}{\bmprod\varphi(\mu-\bm{\la_{\hat{k}}})}
	\frac{\bmprod\varphi(\la_k-\bm{\la_{\hat{k}}}-i\gamma)}{\bmprod\varphi(\la_k-\bm{\la_{\hat{k}}})}
	\right\rbrace
	.
	\label{action_C_bethe_coeff_single}
\end{multline}
The coefficient $\chi_{j,k}$ can be computed using the symmetry argument that comes from \cref{FCR_self}.
We take only the cross-terms for the operator $\opA(\mu)$ which exchange spectral parameters $\mu$ and $\la_j$. 
The remaining exchanges are already managed through the symmetry.
Therefore we see that,
\begin{multline}
	\chi_{j,k}=
	r(\la_k)
	\frac{\varphi(i\gamma)}{\varphi(\mu-\la_k)}
	\frac{\varphi(i\gamma)}{\varphi(\la_j-\mu)}
	\frac{\bmprod\varphi(\la_k-\bm{\la_{\hat{k}}}+i\gamma)}{\bmprod\varphi(\la_{k}-\bm{\la_{\hat{k}}})}
	\frac{\bmprod\varphi(\la_j-\bm{\la_{\hat{j},\hat{k}}}-i\gamma)}{\bmprod\varphi(\la_{j}-\bm{\la_{\hat{j},\hat{k}}})}
	\\
	+
	r(\la_j)
	\frac{\varphi(i\gamma)}{\varphi(\mu-\la_j)}
	\frac{\varphi(i\gamma)}{\varphi(\la_k-\mu)}
	\frac{\bmprod\varphi(\la_j-\bm{\la_{\hat{j}}}+i\gamma)}{\bmprod\varphi(\la_{j}-\bm{\la_{\hat{j}}})}
	\frac{\bmprod\varphi(\la_k-\bm{\la_{\hat{j},\hat{k}}}-i\gamma)}{\bmprod\varphi(\la_{k}-\bm{\la_{\hat{j},\hat{k}}})}
	\label{action_C_bethe_coeff_double}
\end{multline}
\label{action_C_bethe_coeff_both}
\end{subequations}
Now we can see that two types of sums in the expansion \eqref{action_C_bethe_expansion} over the coefficients given by \cref{action_C_bethe_coeff_single,action_C_bethe_coeff_double} can be consolidated into one larger double sum over the set of parameters $\bm{\check\la}$ if we set $\check\la_{N+1}=\mu$. The terms with coefficients $\chi_{k}$ are the special cases of the double sum in \cref{action_C_bv} when we set either $\check\la_j=\mu$ or $\check\la_k=\mu$.
Whereas the two types of terms in $\chi_{j,k}$ are symmetric.
These can be simply used to extend the double sum by removing its ordering condition from $j<k$ to simply $j\neq k$. 
This proves the result in \cref{action_C_bv} whereas the result in \cref{action_B_bv_dual} for the dual can be proved in a similar manner.
\end{proof}
However, we find that the expression \eqref{action_BC_bv} describing the action of an off-diagonal block operator $\opB$ or $\opC$ often leads to a complicated double summation, hence limiting its utility to some particular scenarios.
A rather trivial example of such a scenario can be found in the corollary presented below.
We note that this can also be obtained by comparing the eigenvalues of third component of the spin for the two vectors.
\begin{coro}
The scalar product $\braket{\psi(\bm\la)|\psi(\bm\mu)}$ vanishes whenever $n_{\bm\la}\neq n_{\bm\mu}$.
\begin{align}
	\braket{\psi(\bm\la)|\psi(\bm\mu)}\big|_{n_{\bm\la}\neq n_{\bm\mu}}=0
	\label{corr_sum_ov_ptn_card}
\end{align}
\end{coro}
\begin{proof}
This is clear from the commutator in \cref{comm_C_bv}.
From \cref{action_C_bethe_expansion}, we can recursively express a scalar product as linear sum of scalar product of the Bethe vectors with progressively smaller cardinalities on each side
\begin{align}
	\braket{\psi(\bm\mu)|\psi(\bm\la)}&=\bmsum\alpha_{ab}\braket{\psi(\bm\nu)|\psi(\bm\eta_{ab})}
	,
	&
	n_{\bm\nu}&=n_{\bm\mu}-1,
	\quad
	n_{\bm\eta_{ab}}=n_{\bm\la}-1.
\end{align}
Whenever we have $n_{\bm\mu}\neq n_{\bm\la}$, we would get either $C(\mu)\pvac=0$ or, $\pvac*\opB(\eta)=0$ as a result of the imbalance in the cardinalities.
\end{proof}
\par
When the cardinalities of the vectors match, we need to take the action for all the raising operators.
The \cref{lem_action_BC_bv} provides a combinatorial representation of the scalar products of the Bethe vectors which expresses the scalar products as a sum over partitions \cite{IzeK85,BogIK93}.
It turns out that this combinatorial formula is not very computationally efficient, as the number of terms in the summations grows exponentially with $N$. This is particularly problematic when we are interested in the computations in the thermodynamic limit.
However, a more efficient form for the scalar products can be obtained in the form of determinant whenever at-least one of the Bethe vector is on-shell \cite{Sla89}.
We introduce this result in the following theorem.
\begin{thm}[Slavnov determinant representation \cite*{Sla89}]
\label{thm:sla} 
Let the set $\bm\la$ of Bethe roots \eqref{bae_gen} represent an on-shell Bethe vector and the set $\bm \mu$ of spectral parameter represent an off-shell Bethe vector, with the condition that we have the same cardinalities for two sets: $n_{\bm\mu}=n_{\bm\la}=N$.
Their scalar product can be represented in the form of a determinant
\begin{subequations}
\begin{align}
	\braket{\psi(\bm\la)|\psi(\bm\mu)}&=
	\frac{%
	\bmprod\varphi(\bm\mu-\bm\la-i\gamma)
	}{%
	\bmalt\varphi(\bm\la)
	\bmalt\varphi(-\bm\mu)
	}
	\det\Mcal[\bm\mu\Vert\bm\la]
	.
	\label{sla_det_rep_gen}
\end{align}
\par
The matrix $\Mcal[\bm\mu\Vert\bm\la]$ in \cref{sla_det_rep_gen} is called the \emph{Slavnov matrix} whose elements are described by,
\index{ff@\textbf{Form-factors}!mat Sla@$\Mcal[\cdot\Vert\cdot]$: Slavnov matrix|textbf}%
\begin{align}
	\Mcal_{j,k}[\bm\mu\Vert\bm\la]&=
	\aux(\mu_{j}|\bm\la)
	t(\mu_j-\la_k)
	-
	t(\la_k-\mu_j)
	.
	\label{sla_mat_gen}
\end{align}
The function $t(\nu)$ used here is a rational form in $\varphi$ which is defined as
\index{misc@\textbf{Miscellaneous functions}!t fn@$t$: rational form in $\varphi$ in the Slavnov matrix \cref{sla_mat_gen}}%
\begin{align}
	t(\nu)=\frac{\varphi(i\gamma)}{\varphi(\nu)\varphi(\nu+i\gamma)}
	\label{sla_det_rep_def_t}
\end{align}
\label{sla_det_rep_gen_all}
\end{subequations}
Note that this function is completely different from the polynomial $t$ defined earlier in \cref{baxq_t_fn} the context of Baxter's T-Q equation.
The function $\aux(\mu|\bm\la)$ is the exponential counting function which was defined in \cref{exp_cfn_gen_def}.
\end{thm}
\begin{rem}
Note that we are using the index free notations in \cref{sla_det_rep_gen_all} for the product $\bmprod$ and an alternant product $\bmalt$ which are defined on \cpageref{ind_free_prod,ind_free_alt_prod}.
We have also invokes a parametrised notation from \cpageref{ind_free_mats} for Slavnov matrix $\Mcal[\bm\mu\Vert\bm\la]$.
\end{rem}
We do not prove the \cref{thm:sla} here. There are multiple proofs that are known for this result. 
The original proof by Slavnov can be found in \cite{Sla89}.
An alternative proof based on the $\Fcal$-basis formalism was obtained in \cite{KitMT99}.
In addition to these, a simpler proof using the interpolation theorems for symmetric polynomials can also be written. 
\par
We now discuss two important corollaries of Slavnov's theorem.
The first one ultimately leads to the orthogonality of the on-shell Bethe vectors. Although the orthogonality of the Bethe roots is a trivial result that has been known before the Slavnov's theorem, it is the method used to prove this corollary that is more important for us.
\begin{coro}[Slavnov vector]
The scalar product $\braket{\psi(\bm\la)|\psi(\bm\mu)}$ of two distinct on-shell Bethe vectors vanishes.
The vector $\Vcal$ as defined in \cref{vect_sla} below is in the kernel of the Slavnov matrix $\Ycal\in\ker{\Mcal}$.
\begin{align}
	\Ycal_j=
	\frac{\bmprod\varphi(\la_j-\bm\mu)}{\bmprod'\varphi(\la_j-\bm\la)}
	\label{vect_sla}
	.
\end{align}
\end{coro}
\begin{proof}
This result will be proven here for the rational parametrisation (XXX) only.
We can write components of the product $\Mcal\Ycal$ as
\begin{align}
	(\Mcal\Ycal)_j
	=
	\sum_{k=1}^{N}
	\aux(\mu_a|\bm\la)
	t(\mu_j-\la_k)
	\Ycal_k
	-
	\sum_{k=1}^{N}
	t(\la_k-\mu_j)
	\Ycal_k
	.
\end{align}
Let us look at these two summation terms individually. Since we are in rational parametrisation $\varphi(\la)=\la$.
We can easily show the following result by comparing the residues of the meromorphic functions on the both side of the following expression:
\begin{align}
	\sum_{k=1}^{N}
	\frac{i}{\mu_j-\la_k+ i\sigma}
	\frac{\bmprod(\la_k-\bm{\mu_{\hat{j}}})}{\bmprod(\la_k-\bm{\la_{\hat{k}}})}
	&=
	\sigma
	\frac{\bmprod(\mu_j-\bm\mu+i\sigma)}{\bmprod(\mu_j-\bm\la+i\sigma)}
	,
	&
	(\sigma=\pm 1)
	.
\end{align}
This lead us to the following result:
\begin{align}
	(\Mcal\Ycal)_j
	=
	\aux(\mu_j|\bm\la)
	\frac{\bmprod(\mu_j-\bm\mu+i)}{\bmprod(\mu_j-\bm\la+i)}
	+
	\frac{\bmprod(\mu_j-\bm\mu-i)}{\bmprod(\mu_j-\bm\la-i)}
	.
	\label{action_sla_vect}
\end{align}
Now it only remains to observe that if $\bm\mu$ are also Bethe roots then we can extract $r(\mu_a)$ from their Bethe equations to write
\begin{align}
	\aux(\mu_j|\bm\la)&=
	-
	\frac{\bmprod(\mu_j-\bm\mu-i)}{\bmprod(\mu_j-\bm\mu+i)}
	\frac{\bmprod(\mu_j-\bm\la+i)}{\bmprod(\mu_j-\bm\la-i)}
	.
\end{align}
Substitution of this into \cref{action_sla_vect} tell us that $\Ycal\in\ker(\Mcal_{\bm\la})$. 
Since $\Ycal$ is a non-zero vector for this choice of distinct set of roots, we can conclude that the matrix $\Mcal_{\bm\la}$ given in \cref{thm:sla} is singular when both vectors are on-shell and distinct.
In other words, this shows that on-shell Bethe vectors are orthogonal.
The proof is similar for the trigonometric and hyperbolic parametrisations as we still obtain meromorphic functions $t(\nu)$ and for $\Ycal_a$. The only difference is that there are additional zeroes due to periodicity which needs to be accounted for.
\end{proof}
The second corollary to the \cref{thm:sla} is the Gaudin's determinant representation for the norms.
This result was also first conjectured by \textcite{Gau83} and proved by \textcite{Kor82}, well before the Slavnov determinant representation \eqref{sla_det_rep_gen_all} for scalar products came to light.
But it can also be seen as the corollary of the Slavnov's theorem \ref{thm:sla} by allowing the two Bethe vectors in the representation \eqref{sla_det_rep_gen_all} to coincide.
\begin{coro}[Gaudin determinant representation]
Given an on-shell Bethe vector which is described by Bethe roots $\bm\la$ ($n_{\bm\la}=N$), its squared norm can be represented by a determinant of the Gaudin matrix $\Ncal[\bm\la\Vert\bm\la]$ as
\begin{subequations}
\begin{align}
	\braket{\psi(\bm\la)|\psi(\bm\la)}&=
	(-1)^{N}
	\frac{\bmprod\varphi(\bm\la-\bm\la-i\gamma)}{\bmprod'\varphi(\bm\la-\bm\la)}
	\det\Ncal[\bm\la\Vert\bm\la]
	.
	\label{gau_det_rep_gen}
\end{align}
The Gaudin matrix $\Ncal$ has elements:
\index{ff@\textbf{Form-factors}!mat Gau@$\Ncal[\cdot\Vert\cdot]$: Gaudin matrix|textbf}%
\begin{align}
	\Ncal_{j,k}[\bm\la\Vert\bm\la]&=
	\aux'(\la_{j})\delta_{j,k}
	-
	2\pi i K(\la_j-\la_k)
	.
	\label{gau_mat_gen}
\end{align}
Here the function $K(\nu)$ is defined as the following for any of the three parametrisations in \cref{rmat_paramn_types}.
\begin{align}
	K(\nu)
	= \frac{t(\nu)+t(-\nu)}{2\pi i}
	= \frac{1}{\pi}
	\frac{-i\varphi(2i\gamma)}{\varphi(\nu-i\gamma)\varphi(\nu+i\gamma)}
	.
	\label{gau_mat_K_def}
\end{align}
\label{gau_det_rep_gen_all}
\end{subequations}
Note that the primed product $\bmprod'$ in the denominator here in \cref{gau_det_rep_gen} is a double product that excludes all the zeroes, this notation was introduced on \cpageref{ind_free_notn,ind_free_notn_end}.
\begin{align}
	\bmprod'\varphi(\bm\la-\bm\la)=\prod_{j\neq k}^{N}\varphi(\la_j-\la_k)
\end{align}
\end{coro}
\begin{proof}
We can rewrite \cref{sla_mat_gen} as
\begin{align}
	\Mcal_{jk}[\bm\mu\Vert\bm\la]=
	(1+\aux(\mu_j|\bm\la))
	t(\mu_j-\la_k)
	-
	2\pi i K(\mu_j-\la_k)
\end{align}
In the limit where $\mu_j\to\la_j$, the first term gives rise to the derivative $\aux'(\la_j|\bm\la)$ on the diagonal terms, hence obtaining \cref{gau_mat_gen}.
Similarly computing the prefactors of the determinant representation \eqref{sla_det_rep_gen} in this limit gives us the determinant representation \eqref{gau_det_rep_gen}.
\end{proof}
\minisec{Scalar product formula for the descendants of a Bethe vector in the XXX model}
A variant of this formula that is applicable for the isotropic XXX model is discussed in the following \cref{lem:foda_wh}, originally obtained by \textcite{FodW12a}.
The multiplet structure of the XXX model \eqref{sub:XXX_mult} tells us that adding a root at infinity corresponds to the action of the spin lowering operator $S^-$ \eqref{low_ops_xxx}. In this way, the scalar products of the descendants of the Bethe vector in a XXX multiplet can be accessed.
However, it should be noted that the limit where roots are sent to infinity must be taken with proper normalisation.
\begin{lem}[Foda-Wheeler version of the Slavnov formula \cite{FodW12a}]
\label{lem:foda_wh}
Let $\bm\la$ denote the set of Bethe roots of an on-shell Bethe vector of the XXX model and $\bm\mu$ a set of arbitrary complex parameters.
Let $\ell=n_{\bm\mu}-n_{\bm\la}$.
Then scalar product $\braket{\psi_s^\ell(\bm{\la})|\psi(\bm\mu)}$ of the descendant of the Bethe vector with an off-shell Bethe vector $\ket{\psi(\bm\mu)}$ can be represented in the form: 
\begin{align}
	\braket{\psi_{s}^{\ell}(\bm{\la})|\psi(\bm\mu)}	
	=
	\ell!	
	(-1)^{N\ell+\frac{\ell^2}{2}}
	\frac{%
	\bmprod(\bm\mu-\bm\la-i)%
	}{%
	\bmalt(\bm\la)\bmalt(-\bm\mu)%
	}
	\det\Mcal^{(\ell)}[\bm\mu\Vert\bm\la]
	.
	\label{det_rep_fw}
\end{align}
The matrix $\Mcal^{(\ell)}[\bm\mu|\bm\la]$ is composed of the blocks of columns
\index{ff@\textbf{Form-factors}!mat Sla FW@$\Mcal^{(\ell)}[\cdot\Vert\cdot]$: Foda-Wheeler version of Slavnov matrix|textbf}%
\begin{align}
	\Mcal^{(\ell)}[\bm\mu\Vert\bm\la]
	=
	\left(
	\Mcal[\bm\mu\Vert\bm\la]
	~\bigg|~
	\Ucal[\bm\mu]
	\right)
	.
\end{align}
The first block $\Mcal[\bm\mu\Vert\bm\la]$ is a rectangular Slavnov matrix \eqref{sla_mat_gen} of $N=n_{\bm\la}$ columns. 
Since we are dealing with the XXX model, it is written in the rational parametrisation
\begin{subequations}
\begin{align}
	\Mcal_{j,k}[\bm\mu|\bm\la]
 	=
 	\aux(\mu_j|\bm\la)
 	t(\mu_j-\la_k)
 	-
 	t(\la_k-\mu_j)
 	.
 	\label{sla_block_fw_version}
\end{align}
Whereas the matrix $\Ucal[\bm\mu]=\left(\Ucal_1[\bm\mu]\big|\cdots\big|\Ucal_{\ell}[\bm\mu]\right)$ forms a block of $\ell$ columns whose components are described by the following expression
\index{ff@\textbf{Form-factors}!mat Sla FW@$\Mcal^{(\ell)}[\cdot\Vert\cdot]$: Foda-Wheeler version of Slavnov matrix!block@$\Ucal[\cdot]$: Foda-Wheeler block in \rule{3em}{.5pt}|textbf}%
\begin{align}
	\Ucal_{j,a}[\bm\mu]
	=
	\aux(\mu_j|\bm\la)
	(\mu_j+i)^{a-1}
	-
	\mu_{j}^{a-1}
	.
	\label{van_block_fw_version}
\end{align}
	\label{fw_version_sla_mat_blocks}
\end{subequations}
We remark that the matrix $\Ucal[\bm\mu]$ can be seen as a sum of two Vandermonde matrices.
\end{lem}
\begin{proof}
This result was originally obtained by \textcite{FodW12a}, we reproduce this proof in \cref{sec:FW_append_pf} at the end of this chapter.
\end{proof}
\subsection{Quantum inverse scattering problem and its solution}
\label{sub:qisp}
The problem of expressing the local operators like $\sigma^\alpha_n$ in terms of the operators of the Yang-Baxter algebra is known as \emph{quantum inverse scattering problem}, or QISP.
The solution of this problem for the XXZ chains was originally obtained by \textcite{KitMT99}.
In the following lemma we summarise their result.
\begin{lem}[Solution of QISP for XXZ chains \cite*{KitMT99}]
Let $e^{\alpha\beta}$ denote the elementary matrices in $\mathop{end}(\Cset^2)$ forming its basis.
These can be represented in terms of the block operators of the monodromy matrix \eqref{monodromy_blocks} as the product
\begin{align}
	e^{\alpha\beta}_{n}&=
	\left(\tf(\tfrac{i\gamma}{2})\right)^{n-1}
	\Mon^{\alpha\beta}(\tfrac{i\gamma}{2})
	\left(\tf^{-1}(\tfrac{i\gamma}{2})\right)^{n}
	.
	\label{qism_elem_ops}
\end{align}
Here $\tf(\nu)=\opA(\nu)+\opD(\nu)$ denotes the transfer matrix \eqref{def_tfmat}.
\label{lem:qisp}
\end{lem}
We do not prove the \cref{lem:qisp} here. It was first proved in \cite{KitMT99} using the $\Fcal$-basis method and again in \cite{MaiT00} using the shift operator.
In particular, the result \eqref{qism_elem_ops} leads to the following identities for the local spin operators.
\begin{subequations}
\begin{align}
	\sigma^{z}_{n}&=
	\left(\tf(\tfrac{i\gamma}{2})\right)^{n-1}
	\left\lbrace\opA(\tfrac{i\gamma}{2})-\opD(\tfrac{i\gamma}{2})\right\rbrace
	\left(\tf^{-1}(\tfrac{i\gamma}{2})\right)^{n}
	\label{qism_rel_z}
	\\
	\sigma^{+}_{n}&=
	\left(\tf(\tfrac{i\gamma}{2})\right)^{n-1}
	\opC(\tfrac{i\gamma}{2})
	\left(\tf^{-1}(\tfrac{i\gamma}{2})\right)^{n}
	\label{qism_rel_+}
	\\
	\sigma^{-}_{n}&=
	\left(\tf(\tfrac{i\gamma}{2})\right)^{n-1}
	\opB(\tfrac{i\gamma}{2})
	\left(\tf^{-1}(\tfrac{i\gamma}{2})\right)^{n}
	\label{qism_rel_-}
\end{align}
	\label{qism_rel}
\end{subequations}
With the solution of the quantum inverse scattering problem, we are better placed to introduce the determinant representation for the form-factors, which is the starting point of all the computations carried out in this thesis.
But before we do that, let us briefly consider the case of two-point correlations, primary to see how these objects and the form-factors are related.
We are interested in the longitudinal two-point functions $\braket{\sigma^\alpha_n\sigma^\beta_m}$ at zero temperature.
Upon substituting the inverse relation \eqref{qism_rel_z}, we get the following expression for it
\begin{subequations}
\begin{multline}
	\frac{\braket{\psi_g|\sigma^3_n\sigma^3_m|\psi_g}}{\braket{\psi_g|\psi_g}}
	=
	\left(\evtf_{g}(\tfrac{i\gamma}{2})\right)^{n-m-1}
	\\
	\times
	\frac{\Braket{%
		\psi_g|%
		\left(\opA(\tfrac{i\gamma}{2})-\opD(\tfrac{i\gamma}{2})\right)%
		\left(\tf(\tfrac{i\gamma}{2})\right)^{m-n-1}
		\left(\opA(\tfrac{i\gamma}{2})-\opD(\tfrac{i\gamma}{2})\right)%
		|\psi_g}%
		}{%
		\braket{\psi_g|\psi_g}%
		}%
		.
		\label{qism_2pt_corr_long}
\end{multline}
The $\evtf$ function denotes an eigenvalue of the transfer matrix \eqref{evtf_gen} for the ground state vector that we have factored out into the prefactor.
\index{aba@\textbf{Algebraic Bethe ansatz (ABA)}!transfer matrix ev@$\evtf$: eigenvalue of transfer matrix}%
Similarly for the transverse two-point correlation function $\braket{\sigma^+_n\sigma^-_m}$, we get
\begin{align}
	\frac{\braket{\psi_g|\sigma^+_n\sigma^-_m|\psi_g}}{\braket{\psi_g|\psi_g}}
	=
	\left(\evtf_{g}(\tfrac{i\gamma}{2})\right)^{n-m-1}
	\frac{\Braket{%
		\psi_g|%
		\opC(\tfrac{i\gamma}{2})%
		\left(\tf(\tfrac{i\gamma}{2})\right)^{m-n-1}
		\opB(\tfrac{i\gamma}{2})%
		|\psi_g}%
		}{%
		\braket{\psi_g|\psi_g}%
		}%
		.
	\label{qism_2pt_corr_trans}
	\end{align}
\label{qism_2pt_corr}
\end{subequations}
There are several approaches available to us.
In the first one, an action of the entire block of operators is taken on one of the sides, let us say the left-action for example, using \cref{action_BC_bv,action_opAD_bv_decomp}.
This gives us a multiple sum over the scalar products of the ground state vector with an off-shell vector resulting from this action.
We will not write down this expression here since we do not use it here. From this approach, a multiple integral representation for the n-point correlation function was found in \cite{KitMT00}, which is consistent with the results of \cite{JimMMN92} obtained from the $q$-vertex operator algebra \cite{JimM95}.
\par
Another approach, that is more important to us in this thesis, makes use of the \emph{form-factor} expansion:
\begin{align}
	\frac{\braket{\psi_g|\sigma^{\alpha}_n\sigma^\beta_m|\psi_g}}{\braket{\psi_g|\psi_g}}	
	&=
	\sum_{\ket{\psi_e}\in\qsp}
	\frac{%
	\braket{\psi_g|\sigma^\alpha_n|\psi_e}%
	\braket{\psi_e|\sigma^\beta_m|\psi_g}%
	}{%
	\braket{\psi_g|\psi_g}%
	\braket{\psi_e|\psi_e}%
	}
	.
	\label{qism_ff_expn}
\end{align}
The vectors $\ket{\psi_e}$ denote the eigenvectors of the Hamiltonian, called the excited states.
The completeness of the set of eigenvectors obtained from the algebraic Bethe ansatz is assumed before we could obtain \cref{qism_ff_expn}, so that a resolution of the identity as a sum over projectors can be introduced in two-point function $\braket{\sigma^\alpha_n\sigma^\beta_m}$ to get \cref{qism_ff_expn}.
Let us substitute the identities \eqref{qism_rel} for the solution of the inverse problem in \cref{qism_ff_expn}. We see that for the longitudinal and transverse two-point functions, the scalar products in the numerators of the form-factors in this expansion \eqref{qism_ff_expn} can be rewritten as
\begin{subequations}
\begin{align}
	\braket{\psi_g|\sigma^3_n|\psi_e}
	\braket{\psi_g|\sigma^3_m|\psi_e}
	&=
	\left(\evtf_g(\tfrac{i\gamma}{2})\right)^{n-m-1}%
	\left(\evtf_e(\tfrac{i\gamma}{2})\right)^{m-n-1}%
	\left|
	\braket{\psi_g|\opA(\tfrac{i\gamma}{2})-\opD(\tfrac{i\gamma}{2})|\psi_e}
	\right|^2
	\label{qism_ff_long}
	,
	\shortintertext{and}
	\braket{\psi_g|\sigma^+_n|\psi_e}
	\braket{\psi_g|\sigma^-_m|\psi_e}
	&=
	\left(\evtf_g(\tfrac{i\gamma}{2})\right)^{n-m-1}%
	\left(\evtf_e(\tfrac{i\gamma}{2})\right)^{m-n-1}%
	\braket{\psi_g|\opB(\tfrac{i\gamma}{2})|\psi_e}%
	\braket{\psi_g|\opC(\tfrac{i\gamma}{2})|\psi_e}%
	.
	\label{qism_ff_trans}
\end{align}
	\label{qism_ff}
\end{subequations}
Although we are entirely committed to the form-factor approach in this thesis, it is important to note that the two approaches work complementary to each other for the two-point function.
The form-factor approach turns out to be more efficient in analysing the large distance behaviour of the correlation function $\braket{\sigma^\alpha_n\sigma^\beta_m}$, since the number of multiple-integral terms from the direct approach \eqref{qism_2pt_corr} would grow with the lattice distance $|m-n|$.
The form-factor approach also finds a greater utility for the dynamical correlations where the time dependence is introduced in the two-point function.
\begin{align}
	\Braket{\sigma^{\alpha}_{n}(t_{1})\sigma^{\beta}_{m}(t_{2})}
	&=
	\frac{%
	\braket{%
	\psi_{g}%
	|%
	\sigma^{\alpha}_{n}(t_{1})%
	\sigma^{\beta}_{m}(t_{2})%
	|%
	\psi_{g}%
	}%
	}{%
	\braket{\psi_{g}|\psi_{g}}
	}
\end{align}
During the form-factor expansion we shall use the Hamilton's equation for the time evolution of the operators, which factorises out the time variables $t$ into the exponential phase factors
\begin{align}
	\braket{\sigma^{\alpha}_{n}(t_1)\sigma^{\bar{\alpha}}_{m}(t_2)}
	&=
	\sum_{\ket{\psi_e}\in\qsp}
	e^{-i(t_2-t_1)(E_{e}-E_{g})}
	e^{-i(m-n)(p_{e}-p_{g})}
	\left|\FF^\alpha\right|^2
	.
	\label{ff_expn_loc_ops}
\end{align}
The phase factors are the trivial part of this equation \eqref{ff_expn_loc_ops}, which only depend on the differences of time and lattice positions, which is a direct consequence of time-invariance and translational symmetry of the model.
In the Fourier transform of the dynamic two-point function, called the dynamic structure factor: 
\begin{align}
	S^{\alpha\bar{\alpha}}(k,\omega)
	=
	\int_{-\infty}^{\infty}
	dt\, e^{-i\omega t}
	\sum_{m=1}^{M}
	e^{-ikm}
	\braket{\sigma^\alpha_{m+1}(t)\sigma^{\bar\alpha}_{1}}
\end{align}
all the trivial phase terms can be summed over to write the form-factor expansion.
\begin{align}
	S^{\alpha\bar{\alpha}}
	=
	2\pi  M
	\sum_{\text{exc}}
	\delta(\omega-\varepsilon_{\text{exc}})
	\delta_{k,p_{\text{exc}}}
	\big|\FF^{\alpha}\big|^2
	.
	\label{ff_expn_dsf}
\end{align}
The terms $\varepsilon_{\text{exc}}$ and $p_\text{exc}$ denote the energy and momentum of the excited states measured over the ground state. 
The amplitude term $\big|\FF^\alpha\big|^2$ in \cref{ff_expn_loc_ops,ff_expn_dsf} are the product of form-factors $\FF^\alpha$ and $\FF^{\bar{\alpha}}$.
These are defined as the ratio of scalar products
\index{ff@\textbf{Form-factors}!FF@$\FF^z$: longitudinal form-factor}%
\begin{align}
	\left|\FF^\alpha\right|^2
	=
	\frac{%
	\braket{\psi_{g}|\sigma_m^{\alpha}|\psi_{\text{exc}}}
	\braket{\psi_{\text{exc}}|\sigma_m^{\bar{\alpha}}|\psi_{g}}
	}{%
	\braket{\psi_{g}|\psi_{g}}
	\braket{\psi_{\text{exc}}|\psi_{\text{e}}}
	}
	.
	\label{ff_def_gen}
\end{align}
\section{Determinant representation for the form-factors}
\label{sec:det_rep_ff}
We shall now present the determinant representations for the longitudinal and transverse form-factors for a finite chain, starting from the expressions \eqref{qism_ff}.
For the longitudinal form-factors, we can take the left-action of the operators $\opA$ and $\opD$ and we get the following result.
\begin{lem}
The determinant representation of the scalar product in the longitudinal form-factor \eqref{qism_ff_long} is given by
\begin{subequations}
\begin{align}
	\braket{\psi(\bm\la)|\opA(\tfrac{i\gamma}{2})-\opD(\tfrac{i\gamma}{2})|\psi(\bm\mu)}
	&=
	\tau(\tfrac{i\gamma}{2}|\bm\mu)
	\frac{\bmprod\varphi(\bm\mu-\bm\la-i\gamma)}{\bmalt\varphi(\bm\la)\bmalt\varphi(-\bm\mu)}
	\det\big[\Mcal[\bm\mu\Vert\bm\la]-2\Pcal[\bm\mu\Vert\bm\la]\big]
	\label{qism_det_rep_ff_long}
\end{align}
where $\Pcal(\bm\la|\bm\mu)$ is a rank-1 matrix given by elements
\begin{align}
	\Pcal_{j,k}[\bm\mu\Vert\bm\la]=
	\frac{\bmprod\varphi(\bm\la+\frac{i\gamma}{2})}{\bmprod\varphi(\bm\mu+\frac{i\gamma}{2})}
	\frac{\bmprod\varphi(\mu_j-\bm\la-i\gamma)}{\bmprod\varphi(\mu_j-\bm\mu-i\gamma)}
	~
	t(\la_k-\tfrac{i\gamma}{2})
	.
	\label{qism_rank-1_mat_ff_long}
\end{align}
	\label{qism_det_rep_ff_long_all}
\end{subequations}
\end{lem}
\begin{proof}
Here we first utilise the relation
\begin{align}
	\opA(\tfrac{i\gamma}{2})
	-
	\opD(\tfrac{i\gamma}{2})
	=
	2\opA(\tfrac{i\gamma}{2})
	-
	\tf(\tfrac{i\gamma}{2})
	.
\end{align}
The action of the operator $\opA$ on a Bethe vector was computed in \cref{action_opA_bv_cross,action_opAD_bv_direct,action_opA_bv_decomp}.
The direct term for this action is identical to action of $\tf(\tfrac{i\gamma}{2})$.
Using this we get
\begin{multline}
	\braket{\psi(\bm\la)|\opA(\tfrac{i\gamma}{2})-\opD(\tfrac{i\gamma}{2})|\psi(\bm\mu)}
	=	
	\evtf(\tfrac{i\gamma}{2}|\bm\mu)
	\braket{\psi(\bm\la)|\psi(\bm\mu)}
	\\
	+
	2	\sum_{a=1}^{N}
	\frac{\bmprod\varphi(\mu_{a}-\bm{\mu}-i\gamma)}{\bmprod'\varphi(\mu_{a}-\bm{\check\mu})}
	\braket{\psi(\bm\la)|\psi(\bm{\check\mu_{\hat{a}}})}
\end{multline}
where $\bm{\check\mu}=\bm\mu\bm\cup\set{\tfrac{i\gamma}{2}}$. Using the determinant representation from \cref{thm:sla} we obtain the sum over determinants
\begin{multline}
	\braket{\psi(\bm\la)|\opA(\tfrac{i\gamma}{2})-\opD(\tfrac{i\gamma}{2})|\psi(\bm\mu)}
	=
	\evtf(\tfrac{i}{2}|\bm\mu)
	\frac{\bmprod\varphi(\bm\mu-\bm\la-i\gamma)}{\bmalt\varphi(\bm\la)\bmalt\varphi(-\bm\mu)}
	\left\lbrace
	\det\Mcal[\bm\mu\Vert\bm\la]
	\phantom{%
	\frac{\bmprod\varphi(\tfrac{i\gamma}{2}+\bm\la)}{\bmprod\varphi(\tfrac{i\gamma}{2}+\bm\mu)}
	}
	\right.
	\\
	\left.
	+
	2
	\sum_{a=1}^{N}
	(-1)^{N-a}
	\frac{\bmprod\varphi(\tfrac{i\gamma}{2}+\bm\la)}{\bmprod\varphi(\tfrac{i\gamma}{2}+\bm\mu)}
	\frac{\bmprod\varphi(\mu_a-\bm\mu-i\gamma)}{\bmprod\varphi(\mu_a-\bm\la-i\gamma)}
	\det\Mcal\left[\bm{\check\mu_{\hat{a}}}\big\Vert\bm\la\right]
	\right\rbrace
	\label{qism_sum_det_long_ff}
	.
\end{multline}
The last column rank-1 modifications in $\Mcal\left[\bm{\check\mu_{\hat{a}}}\Vert\bm\la\right]$ are described by following the expression since we have $\aux(\tfrac{i\gamma}{2})=0$:
\begin{align}
	\Mcal_{N,j}\left[\bm{\check\mu_{\hat{a}}}\big\Vert\bm\la\right]
	=
	-t(\la_j-\tfrac{i\gamma}{2}).
\end{align}
The sum over determinants in \cref{qism_sum_det_long_ff} can be written down as common determinant representation \eqref{qism_det_rep_ff_long_all} using \cref{lem:rank-1_det} shown in \cref{chap:mat_det_extn}.
\end{proof}
In the case of transverse form-factors, we have two distinct channels of computations since the left and right action of the off-diagonal operators $\opB$ or $\opC$ have different forms [see the \cref{lem_action_BC_bv}], unlike the diagonal blocks $\opA$ or $\opD$ where this choice does not make a significant difference [see \cref{action_opAD_bv_decomp}].
In the channel where the off-diagonal acts as raising operator of the Bethe algebra (i.e., a left-action of $\opC$ or a right action of $\opB$), we get a complicated double-sum \eqref{action_BC_bv} as we have seen in the \cref{lem_action_BC_bv}.
This expression leads to the following representation for the transverse form-factor:
\begin{lem}
The scalar product in the transverse form-factor for $\sigma^+_n$ \eqref{qism_ff_trans} when evaluated with the right-action can be expressed as
\begin{subequations}
\begin{multline}
	\braket{\psi(\bm\la)|\opC(\tfrac{i}{2})|\psi(\bm\mu)}
	=
	\frac{%
	\bmprod\varphi(\bm{\check\mu}-\bm\la-i\gamma)%
	}{%
	\bmalt\varphi(\bm\la) \bmalt\varphi(\bm{-\check\mu})
	}
	\\
	\times
	\sum_{j=1}^{N}
	\sum_{\underset{k\neq j}{k=1}}^{N+1}
	\frac{(-1)^{j+k+I_{j>k}}}{\varphi(\mu_j-\check\mu_k+i\gamma)}
	\frac{\bmprod\varphi(\mu_j-\bm\mu-i\gamma)}{\bmprod\varphi(\mu_j-\bm\la-i\gamma)}
	\frac{\bmprod\varphi(\check\mu_k-\bm\mu-i\gamma)}{\bmprod\varphi(\check\mu_k-\bm\la-i\gamma)}
	\det\Mcal\left[\bm{\check{\mu}_{\hat{j},\hat{k}}}\big\Vert\bm\la\right]
	\label{ff_trans_double_sum_opC}
\end{multline}
Similarly, the scalar product in the transverse form-factor for $\sigma^-_n$ \eqref{qism_ff_trans} when evaluated with the left-action can be expressed as
\begin{multline}
	\braket{\psi(\bm\mu)|\opB(\tfrac{i}{2})|\psi(\bm\la)}
	=
	\frac{%
	\bmprod\varphi(\bm{\check\mu}-\bm\la-i\gamma)%
	}{%
	\bmalt\varphi(\bm\la) \bmalt\varphi(\bm{-\check\mu})
	}
	\\
	\times
	\sum_{j=1}^{N}
	\sum_{\underset{k\neq j}{k=1}}^{N+1}
	\frac{(-1)^{j+k+I_{j>k}}}{\varphi(\mu_j-\check\mu_k+i\gamma)}
	\frac{\bmprod\varphi(\mu_j-\bm\mu-i\gamma)}{\bmprod\varphi(\mu_j-\bm\la-i\gamma)}
	\frac{\bmprod\varphi(\check\mu_k-\bm\mu-i\gamma)}{\bmprod\varphi(\check\mu_k-\bm\la-i\gamma)}
	\det\Mcal\left[\bm{\check{\mu}_{\hat{j},\hat{k}}}\big|\bm\la\right]
	\label{ff_trans_double_sum_opB}
\end{multline}
\label{ff_trans_double_sum}
\end{subequations}
Here the function $I_{j>k}$ is the characteristic function for $\set{(j,k)|j>k}$\footnote{this can also be seen as the Heaviside step function} and the notation $\bm{\check\mu}$ is used to denote the union $\bm{\check\mu}=\bm\mu\bm\cup\set{\tfrac{i}{2}}$.
\label{lem:ff_trans_double_sum}
\end{lem}
\begin{proof}
This follows from the \cref{lem_action_BC_bv}. Using \cref{action_C_bv} and the determinant representation \eqref{sla_det_rep_gen_all} we can write
\begin{multline}
	\braket{\psi(\bm\la)|\opC(\tfrac{i}{2})|\psi(\bm\mu)}
	=
	\sum_{j=1}^{N}
	\sum_{\underset{k\neq j}{k=1}}^{N+1}
	\left\lbrace
	r(\mu_j|\bm\mu)
	\frac{\bmprod\varphi(\mu_j-\bm\mu+i\gamma)}{\bmprod\varphi(\mu_j-\bm{\check{\mu}_{\hat{j}}})}
	\frac{\bmprod\varphi(\check\mu_k-\bm{\mu_{\hat{j}}}-i\gamma)}{\bmprod\varphi(\check\mu_k-\bm{\check{\mu}_{\hat{j},\hat{k}}})}
	\phantom{%
	\frac{%
	\bmprod\varphi(\bm{\check\mu_{\hat{j},\hat{k}}}-\bm\la-i\gamma)%
	}{%
	\bmalt\varphi(\bm\la)
	\bmalt\varphi(\bm{-\check\mu_{\hat{j},\hat{k}}})
	}}
	\right.
	\\
	\left.
	\times
	\frac{%
	\bmprod\varphi(\bm{\check\mu_{\hat{j},\hat{k}}}-\bm\la-i\gamma)%
	}{%
	\bmalt\varphi(\bm\la)
	\bmalt\varphi(\bm{-\check\mu_{\hat{j},\hat{k}}})
	}
	\det\Mcal\left[\bm{\check\mu_{\hat{j},\hat{k}}}\big\Vert\bm\la\right]
	\right\rbrace
	.
	\label{lem_trans_ff_double_sum_expn}
\end{multline}
Since the parameter $\mu_j$ is a Bethe root, it satisfies \cref{bae_gen}. Hence we can see that
\begin{align}
	r(\mu_j|\bm\mu)\bmprod\varphi(\mu_j-\bm\mu+i\gamma)
	=
	-\bmprod\varphi(\mu_j-\bm\mu-i\gamma)
	.
	\label{lem_trans_ff_double_sum_bae}
\end{align}
In the denominator, we can combine the terms to write a larger alternant (Vandermonde determinant)
\begin{align}
	\bmalt\varphi(\bm{-\check{\mu}_{\hat{j},\hat{k}}})
	\bmprod\varphi(\mu_j-\bm{\check{\mu}_{\hat{j}}})
	\bmprod\varphi(\check{\mu}_j-\bm{\check{\mu}_{\hat{j},\hat{k}}})
	=
	(-1)^{j+k+I_{j>k}}	
	\bmalt\varphi(\bm{-\check\mu})
	\label{lem_trans_ff_double_sum_alt_decomp}
	.
\end{align}
Let us now substitute expressions derived in the above \cref{lem_trans_ff_double_sum_bae,lem_trans_ff_double_sum_alt_decomp} into \cref{lem_trans_ff_double_sum_expn}.
Here we also take the common prefactors out of the summation and we obtain as a result \eqref{ff_trans_double_sum_opC}.
Similarly one can also obtain the result \eqref{ff_trans_double_sum_opB} using \cref{action_B_bv_dual} for the scalar product in the $\sigma^-_n$ form-factor.
\end{proof}
In the second channel, where the off-diagonal operator $\opB$ or $\opC$ act as a lowering operator in the Bethe algebra, the computation is very straightforward.
This can be easily seen from the definitions given in \cref{bv_arbit,bv_arbit_dual}.
The determinant representation obtained in this way is summarised in the following \cref{lem:det_rep_trans_direct}.
\begin{lem}
\label{lem:det_rep_trans_direct}
The scalar product in the transverse form-factor for $\sigma^+_n$ when evaluated with the left-action of the operator $\opC(\tfrac{i}{2})$ leads to the determinant representation
\begin{subequations}
\begin{align}
	\braket{\psi(\bm\la)|\opC(\tfrac{i}{2})|\psi(\bm\mu)}
	&=
	\frac{%
	\bmprod\varphi(\bm{\check\la}-\bm\mu-i\gamma)%
	}{%
	\bmalt\varphi(\bm\mu)
	\bmalt\varphi(\bm{-\check\la})
	}
	\det\Mcal\left[\bm{\check\la}\big\Vert\bm\mu\right]
	\label{ff_trans_direct_opC}
\end{align}
and similarly the right-action of the operator $\opB(\tfrac{i}{2})$ in the scalar product in the transverse form-factors for the $\sigma^-$ leads to the following determinant representation:
\begin{align}
	\braket{\psi(\bm\mu)|\opB(\tfrac{i}{2})|\psi(\bm\la)}
	&=
	\frac{%
	\bmprod\varphi(\bm{\check\la}-\bm\mu-i\gamma)%
	}{%
	\bmalt\varphi(\bm\mu)
	\bmalt\varphi(\bm{-\check\la})
	}
	\det\Mcal\left[\bm{\check\la}|\bm\mu\right]
	\label{ff_trans_direct_opB}
\end{align}
\label{ff_trans_direct}
\end{subequations}
\end{lem}
The next logical step is the computation of the form-factors in the thermodynamic limit, starting from the determinant representations in \cref{qism_det_rep_ff_long_all,ff_trans_double_sum,ff_trans_direct}.
It is a challenging problem.
This project of thermodynamic form-factors from ABA has been realised only for a handful of scenarios in the XXZ chains, which includes: the spontaneous magnetisation for massive XXZ chain $\Delta>1$ in \cite{IzeKMT99}, form-factors of the massless XXZ chain $-1<\Delta\leq 1$ where the Hamiltonian is modified by coupling to an external field \cite{KitKMST11a}, and the result which we will present later in this thesis for two-spinon form-factors for the XXX chain \cite{KitK19}.
We should also remark that the form-factors can also be computed using the $q$-vertex operator algebra approach. The results obtained in the algebraic Bethe ansatz framework are found in good agreement with the result from the $q$-vertex algebra framework wherever there is an overlap.
In this thesis we are going to present the method used to obtain the result of \cite{KitK19} and examine the broader utility of this method.
In the latter case, the symmetry of the XXX model brings additional freedom when it comes to the form-factors, which is highlighted in the following paragraphs.
\minisec{Determinant representation for form-factors of the XXX model}
In this thesis, we are primarily interested in the longitudinal form-factors for the {XXX} model computed in the thermodynamic limit.
\begin{align}
	|\FF^{z}|^2=
	\frac{%
	\braket{\psi_{g}|\sigma^{3}_{m}|\psi_{e}}%
	}{%
	\braket{\psi_{g}|\psi_{g}}
	}
	\frac{%
	\braket{\psi_{e}|\sigma^{3}_{m}|\psi_{g}}%
	}{%
	\braket{\psi_{g}|\psi_{e}}
	}
	\label{long_ff_scal_prod}
\end{align}
A more detailed understanding of excitations entering in this formula is required for the computations in the thermodynamic limit, which we will do in \cref{chap:spectre}.
For the demonstration of the following result, that is applicable to finite chain before we take the thermodynamic limit, it would be sufficient to know that the ground state $\ket{\psi_g}$ of the XXX model is a singlet.
\begin{lem}[Selection criteria]
\label{lem:seln_criteria}
Only the first descendant of the triplet excitations $\ket{\psi_{1}^{1}}$ can have non-trivial contribution to the longitudinal form-factors \eqref{long_ff_scal_prod} of the XXX model.
\end{lem}
\begin{proof}
Let us first remark that the scalar product $\braket{\psi_g|\sigma^3_n|\psi_s^\ell}$ would vanish unless $\ell=s$. It can be seen by comparing the eigenvalues of $S^3$ or through \cref{corr_sum_ov_ptn_card} from the corollary to the \cref{lem_action_BC_bv}.
Now we shall show that for the scalar product $\braket{\psi_g|\sigma^3_n|\psi_s^s}$ vanishes for both singlet $s=0$ and higher multiplets $s>1$.
\par
For the case of singlet excitations ($s=0$), we can replace the operator $\sigma^{3}_{m}$ by the commutator $[S^{+},\sigma^{-}_{m}]$ and see that we have
\begin{align}
	\Braket{\psi_{g}|\sigma^{3}_{m}|\mex{\psi}[0]_{0}}
	&=
	\Braket{\psi_{g}|[S^{+},\sigma^{-}_{m}]|\mex{\psi}[0]_{0}}=0
	 .
	\label{sel_rule_singlet}
\end{align}
This is because of the fact that all singlets, which includes the ground state $\ket{\psi_{g}}$, are annihilated by the action of both lowering and raising operators.
Similarly, we can also show that for quintuplet or higher multiplets ($s\geq 2$), the longitudinal form-factor is also zero due to 
\begin{align}
	\Braket{\psi_{g}|\sigma^{3}_{m}|\psi^{s}_s}&=
	\Braket{\psi_{g}|\sigma^{-}_{m}|\psi^{s-1}_s}=
	\Braket{\psi_{g}|S^{-}\sigma^{-}_{m}|\psi^{s-2}_s}=0
	,
	& s\geq 2
	.
	\label{sel_rule_higher_mult}
\end{align}
This shows that only the scalar products $\Braket{\psi_g|\sigma^3_{n}|\psi_{1}^{1}(\bm\hle)}$ can have non-zero value.
\end{proof}
\begin{coro}
The procedure in the proof of the above \cref{lem:seln_criteria} can also be used to prove the following identities for the scalar product of the triplet.
It permits us to switch between the scalar products for the longitudinal and transverse local spin operators
\begin{subequations}
\begin{align}
	\Braket{\psi_{g}|\sigma^{3}_{m}|\psi_1^{1}}
	&=
	-2\braket{\psi_{g}|\sigma^{-}_{m}|\psi_1^{0}}
	\shortintertext{or,}
	\Braket{\psi_{g}|\sigma^{3}_{m}|\psi_1^{1}}
	&=
	\Braket{\psi_{g}|\sigma^{+}_{m}|\psi_1^{2}}
	.
	\intertext{And the squared norm for the first descendant is revealed to be twice that of the corresponding leading Bethe vector}
	\Braket{\psi_1^{1}|\psi_1^{1}}
	&=
	2 \Braket{\psi_1^{0}|\psi_1^{0}}
	.
\end{align}
\end{subequations}
These identities can allow us to rewrite the longitudinal form-factor \eqref{long_ff_scal_prod} in the transverse mode either as
\begin{subequations}
\begin{align}
	|\FF^{z}|^2
	&=
	2
	\frac{\Braket{\psi_{g}|\sigma^{-}_{m}|\psi^{0}_1}}{\Braket{\psi_{g}|\psi_{g}}}
	\frac{\Braket{\psi^{0}_1|\sigma^{+}_{m}|\psi_{g}}}{\Braket{\psi^{0}_1|\psi^{0}_1}}
	\label{ff_transmap_double_sum}
\shortintertext{or,}
	|\FF^{z}|^2
	&=
	-
	\frac{\Braket{\psi_{g}|\sigma^{-}_{m}|\psi^{0}_1}}{\Braket{\psi_{g}|\psi_{g}}}
	\frac{\Braket{\psi^{2}_1|\sigma^{-}_{m}|\psi_{g}}}{\Braket{\psi^{0}_1|\psi^{0}_1}}
	.
	\label{ff_transmap}
\end{align}
\label{ff_transmap_all}
\end{subequations}
\end{coro}
The identities \eqref{ff_transmap_all} mean that we can exploit the $\mathfrak{su}_2$ symmetry of the XXX model to find a simpler route for the computation of the form-factors.
To do this, let us first compare the representations for longitudinal and transverse form-factors obtained in \cref{qism_det_rep_ff_long_all,ff_trans_double_sum,ff_trans_direct}.
Let us also recall that the complexity of the representations in the transverse case is different for the two channels which we discussed earlier.
While we have a very simple formula \eqref{ff_trans_direct} for the direct channel with $\opB$ or $\opC$ acting as creation operators; in contrast we obtain a complicated double sum for the channel \eqref{ff_trans_double_sum} where the operator $\opB$ or $\opC$ acts as an annihilation operator.
\\
Although this means that a use of direct transverse channel \eqref{ff_trans_direct} is always preferable, we should also note that we do not have absolute freedom when choosing the direction of this action.
For a reason that will become evident in \cref{comp_ff_XXX}, the action of a local operator must be taken rightward in our method, no matter which one of these three representation from \cref{long_ff_scal_prod,ff_transmap_double_sum,ff_transmap} is chosen.
\\
With this piece of information provided, we can evidently see that the representation \eqref{ff_transmap} makes for an appropriate choice.
We also note that its use is facilitated by a version of the Slavnov determinant formula for the descendants due to \cite{FodW12a}, which we have presented in the \cref{lem:foda_wh}.
\begin{subappendices}
	\section[Proof of the Foda-Wheeler version of Slavnov's theorem]{Proof of \texorpdfstring{\cref{lem:foda_wh} :}{} Foda-Wheeler version of Slavnov's theorem}
\label{sec:FW_append_pf}
\begin{proof}
Here we will use the notation $\bm{\zeta}^{(\ell)}$ to denote a set with $\ell $ extra roots added $\bm{\zeta}^{(\ell)}=\bm\la\cup\bm\eta$, with the cardinality $n_{\bm\eta}=\ell$.
At any intermediate stage, we will write the set $\bm{\zeta^{(k)}}=\bm\la\cup\bm{\eta^{(k)}}$ where $\bm{\eta^{(k)}}=\set{\eta_{k+1},\ldots,\eta_{\ell}}$.
Initially, we have the determinant representation
\begin{align}
	\braket{\psi_{s}^{\ell}(\bm\la)|\psi(\bm\mu)}
	=
	(-i)^{\ell}
	\lim_{\bm\eta\to\infty}
	\left(\bm\smallprod \bm\eta\right)
	\frac{%
	\bmprod(\bm\mu-\bm{\zeta^{(\ell)}}-i)
	}{%
	\bmalt(-\bm\mu)
	\bmalt(\bm{\zeta^{(\ell)}})
	}
	\det\Mcal[\bm\mu\Vert\bm{\zeta^{(\ell)}}]
	.
\end{align}
Let us begin with the first limit where $\eta_1\to\infty$.
The prefactors in this limit are transformed as
\begin{align}
	\lim_{\eta_1\to\infty}
	\frac{1}{\eta_1}
	\frac{%
	\bmprod(\bm\mu-\bm{\zeta^{(\ell)}}-i)
	}{%
	\bmalt(-\bm\mu)
	\bmalt(\bm{\zeta^{(\ell)}})
	}
	=
	(-1)^{n_{\bm\mu}}
	\frac{%
	\bmprod(\bm\mu-\bm{\zeta^{(\ell-1)}}-i)
	}{%
	\bmalt(-\bm\mu)
	\bmalt(\bm{\zeta^{(\ell-1)}})
	}
	.
\end{align}
We can also see that all the terms with exponential counting reduce in this limit simply as
\begin{align}
	\lim_{\eta_1\to\infty}
	\aux(\mu_j|\bm{\zeta^{(\ell)}})
	=
	\aux(\mu_j|\bm{\zeta^{(\ell-1)}})
	.
\end{align}
We will multiply the column of index $n_{\bm\la}+1$ with a factor of $\eta_1^2$  before taking the limit. It leads to the Foda-Wheeler column $\Ucal_1$
\begin{align}
	\Ucal_1(\mu_j)
	\equiv
	\lim_{\eta_1\to\infty}\Mcal_{j,n_{\bm\la}+1}\left[\bm\mu\big\Vert\bm{\zeta^{(\ell)}}\right]
	=
	\aux(\mu_j|\bm{\zeta^{(\ell-1)}})-1
	.
\end{align}
For the successive limit $\eta_{a}\to\infty$ ($a\geq 2$), we first develop the fractional terms in the column of index $n_{\bm\la}+a$ as
\footnote{We use the notation $z^\pm=z\pm\frac{i}{2}$ in this expression.}
\begin{align}
	t(\pm(\eta_a-\mu_j))
	\sim_{\eta_{a}\to\infty}
	\sum_{r=0}^{\infty}
	f_r(\mu^\pm_j)
	\eta_{a}^{-r-2}
	.
	\label{fod_wh_pol_expn}
\end{align}
Here $f_r(z)\in\Cset_{r}[z]$ is a polynomial of degree $r$ with a known coefficient for the term of highest degree.
\begin{align}
	f_r(z)=(r+1)z^r+\sum_{a<r}\chi_a z^a
	.
\end{align}
Let us project in this way all the polynomials of degree $0\leq r< a$ so that the expression \eqref{fod_wh_pol_expn} can be expanded as
\begin{align}
	t(\pm(\eta_a-\mu_j))
	\sim
	(\mu_{j}^\pm)^{a-1}
	\eta_a^{-a-1}
	+
	g_{a-2|a+1}(\mu_j^\pm|\eta_a^{-1})
	+
	\sum_{r\geq a}
	f_r(\mu_j^\pm)\eta_a^{-r-2}
	\label{foda_wh_rat_expn_proj}
\end{align}
where $g_{r,s}(u,v)$ denotes a polynomial of degrees $r$ and $s$ in the variables $u$ and $v$ respectively.
\par
We can see that the limit $\eta_a\to\infty$ must be taken with the normalisation $\eta_a^{a}$ added to it in the denominator, so that
\begin{align}
	\lim_{\eta_a\to\infty}
	\frac{1}{\eta_a^{a}}
	\frac{%
	\bmprod(\bm\mu-\bm{\zeta^{(\ell-a+1)}})
	}{%
	\bmalt(-\bm\mu)
	\bmalt(\bm{\zeta^{(\ell-a+1)}})
	}
	=
	(-1)^{n_{\bm\mu}}
	\frac{%
	\bmprod(\bm\mu-i\bm{\zeta^{(\ell-a)}})
	}{%
	\bmalt(-\bm\mu)
	\bmalt(\bm{\zeta^{(\ell-a)}})
	}
	.
\end{align}
The limit for the counting functions can be simply evaluated as
\begin{align}
	\lim_{\eta_a\to\infty}
	\aux(\mu_j|\bm{\zeta^{(\ell-a+1)}})
	=
	\aux(\mu_j|\bm{\zeta^{(\ell-a)}})
	.
\end{align}
Before taking the limit for the column of index $n_{\bm\la}+a$ in $\Mcal^{a-1}[\bm\mu|\bm{\zeta^{(\ell-a+1)}}]$, we will first use the expansion \eqref{foda_wh_rat_expn_proj}.
During this substitution we also cancel the polynomial $g_{a-1,a+1}(\mu_j|\eta_a^{-a})$ using the recombination with the previously obtained Foda-Wheeler columns $\Ucal_a$, up to the degree $a-1$.
\begin{align}
	\what{\Mcal}^{(a-1)}_{j,n_{\bm\la}+a}
	=
	\frac{1}{a}
	\left\lbrace
	\Mcal^{(a-1)}_{j,n_{\bm\la}+a}
	+
	\sum_{r=1}^{a-1}
	\chi_r(\eta_a^{-1})
	\Ucal_a(\mu_j)
	\right\rbrace
	.
\end{align}
Now taking the limit $\eta_a\to\infty$ on this new column, with the normalisation $\eta_a^{a+1}$, gives us
\begin{align}
	\Ucal_{ja}[\bm\mu]
	\equiv
	\lim_{\eta_a\to\infty}
	\what{\Mcal}^{(a-1)}_{j,a}
	.
\end{align}
Repeating this process for all $1<a\leq \ell$ would give us the limit
\begin{align}
	\lim_{\bm\eta\to\infty}
	(\bm\eta)
	\frac{%
	\bmprod(\bm\mu-\bm{\zeta^{(\ell)}})
	}{%
	\bmalt(-\bm\mu)
	\bmalt(\bm{\zeta^{(\ell)}})
	}
	\det\Mcal[\bm\mu|\bm{\zeta^{(\ell)}}]
	=
	(-1)^{n_{\bm\mu}\ell}
	\ell! 
	\frac{%
	\bmprod(\bm\mu-\bm\la)
	}{%
	\bmalt(-\bm\mu)
	\bmalt(\bm\la)
	}
	.
\end{align}	
This leads us to the result \eqref{det_rep_fw} of the \cref{lem:foda_wh}.
\end{proof}
\end{subappendices}
\clearpage{}%
\clearpage{}%
\chapter{Spectrum in the thermodynamic limit}
\label{chap:spectre}
In this chapter, we will study the nature of the spectrum generated by the Bethe eigenvectors and its thermodynamic limit, where we let the length of the spin-chain $M$ as well as the number of magnons $N$ (or the pseudo-particles) approach infinity ($M\to\infty$, $N\to\infty$) while keeping the density of the magnons fixed ($d=\frac{N}{M}$).
\par
This chapter is divided into three sections.
In the \hyperref[sec:gs_gen]{first section}, we discuss the nature of the ground state for the XXZ model $\Delta>-1$.
In particular, the criteria for determining the ground state, first laid down in \cite{Hul38} and later proved by C.N. Yang and C.P. Yang in \cite{YanY66,YanY66a}, is discussed here.
We shall also discuss the condensation property for the ground state roots, which was originally conjectured by Yang and Yang, and proved by \textcite{Koz18}.
Subsequently, we propose a more general version of the condensation property which applies to a class of meromorphic functions.
It is demonstrated for the cases where the Fermi-zone is compact.
We then propose that it can be used for the non-compact case with certain assumptions.
\par
In the \hyperref[sec:spectre_XXX]{second section} we discuss the excitations in the vicinity of the ground state of the isotropic XXX model.
Special attention is paid towards the behaviour of complex roots in the thermodynamic limit.
In this context, we introduce the \emph{Destri-Lowenstein \emph{(DL)} picture} \cite{DesL82}, which divides complex roots into classes of close-pairs and wide-pairs.
A comparison between the DL picture and the string picture is made, based on which we strengthen the argument for why the DL picture is better suited to our computations than the string picture.
\par
In the \hyperref[sec:xxz_spectre]{third section}, we briefly discuss the generalisation of the DL picture to the excitation of the anisotropic XXZ model $\Delta>-1$ by \textcite{BabVV83}.
\section{The true ground state of the XXZ model}
\label{sec:gs_gen}
First of all, let us recall that the dispersion relation \eqref{disp_rel_xxz} for magnon excitations, which are created by the action of $\opB$, is not a non-negative function.
Therefore the ferromagnetic ground state $\pvac$ is not the real ground state of the anti-ferromagnetic chain. In fact, 
it is the highest energy state.
More importantly, we can see that the true ground state (or the antiferromagnetic ground state) must be a highly disordered state, composed by a sea of magnons, which lies in the subspace with $N=\frac{M}{2}$.
\par
Let us recall from \cref{chap:qism} the logarithmic form of the Bethe equations \eqref{log_bae_all}. Let us also recall that in terms of the counting function $\cfn$ \eqref{cfn_gen}, the system of logarithmic Bethe equations can be expressed as
\begin{subequations}
\label{log_bae_spectre}
\begin{flalign}
	(\forall a\leq N),
	&&
	\cfn(\la_a)
	&=
	2\pi Q_{a},
	&&
	\label{log_bae_cfn_form_spectre}
	\\
	(N=n_{\bm\la}),
	&&
	\cfn(\nu)
	&=
	\Theta_{1}(\nu)-\frac{1}{M}\sum_{k=1}^{N}\Theta_{2}(\nu-\la_{k})
	.
	&&
	\label{cfn_spectre}
\end{flalign}
\end{subequations}
From the monotonicity of the $\cfn$ function, we see that the quantum numbers $\bm{Q}$ are half-integers which are in one-to-one correspondence with the Bethe roots on the real line.
Now to characterise the ground state, \textcite{Hul38} conjectured criteria which identify the ground state from its quantum numbers $\bm Q$. His conjecture was later proved by \textcite{YanY66}. It states that the ground state has the following characteristics:
\begin{enumerate}
	\item A set of real, pairwise distinct Bethe roots $\bm\la\subset\Rset$ of cardinality $n_{\bm\la}=\frac{M}{2}$.
	\item And the set of corresponding quantum numbers related through \cref{log_bae_spectre} are all consecutive
	\begin{align}
		Q_{a}&=
		a-\frac{M+2}{4},
		\qquad a=1,2,\ldots, \frac{M}{2}.
 	\end{align}
\end{enumerate}
The ground state is non-degenerate in the regime $-1<\Delta\leq 1$. As we move to the regime $\Delta>1$, due to the periodicity of the $\varphi$ function, we get a two-fold quasi-degeneracy in the thermodynamic limit as $M\to\infty$.
The quasi-degenerate ground state for $\Delta>1$ is obtained \cite{Orb58} \cite[see also][]{Tak99a} from the choice of quantum numbers that are shifted $\bm{Q^\prime}=\bm{Q}\bm+1$.
\begin{defn}[Density]
We define the density function $\rden(\nu|\bm\la)$ as the derivative of the counting function \eqref{cfn_gen}
\begin{align}
	\rden(\nu|\bm\la)=
	\frac{d}{d\nu}\cfn(\nu|\bm\la)
	.
	\label{def_rden}
\end{align}
We will often drop the second argument whenever it is unambiguously clear from the context.
For example, here we can write $\rden_g(\nu)=\rden(\nu|\bm\la)$ for the ground state.
\index{exc@\textbf{Excitations}!condn@\textbf{- condensation}!den_gs_real@$\rden_g$: density function for the ground state.|textbf}%
Similarly, while speaking about an excited state, we can denote $\rden_e(\nu)$.
\index{exc@\textbf{Excitations}!condn@\textbf{- condensation}!den_es_real@$\rden_e$: total density function for the excited state (incl. holes).|textbf}%
\end{defn}
\subsection{Condensation of roots in the thermodynamic limit}
\label{sub:condn}
\Textcite{Hul38} also postulated that in the thermodynamic limit, the ground state roots condense on the real line with a density given by the function $\rden_g(\la)$.
The support of the counting function is called the Fermi-zone $\mathfrak{F}\subset\Rset$ or the Fermi distribution.
\index{exc@\textbf{Excitations}!condn@\textbf{- condensation}!Fermi zone / distribution @$\mathfrak{F}$: Fermi-zone.|textbf}%
The condensation property allows us to replace a sum over ground state Bethe roots by an integral with density function $\rden_g$ as measure. In particular, starting from the logarithmic Bethe equation \eqref{log_bae_spectre} this procedure gives us the integral equation for the density function $\rden_g$ itself.
\begin{subequations}
\begin{align}
	\rden_g(\la)
	+
	\int_{\mathfrak{F}}K(\la-\tau)\rden_g(\tau)d\tau
	&=
	\frac{1}{2\pi}
	p_{0}^\prime(\la)
	.
	\label{lieb_eq_gen}
\end{align}
We call it the Lieb equation since it was obtained first by \textcite{LieL63} in the setting of an integrable model known as one-dimensional Bose gas.
Due to parity symmetry of the counting function $\cfn$, we can see that the Fermi-zone is symmetric at the origin, which can be either compact $\mathfrak{F}=[-\Lambda,\Lambda]$ or non-compact $\mathfrak{F}=\Rset$.
The compactness is determined by the following relation for the Fermi-boundary parameter $\Lambda$:
\begin{align}
	\int_{-\Lambda}^{\Lambda}\rden_g(\la)d\la = \frac{1}{2}.
	\label{fermi_condn_gen}
\end{align}
which is coupled to the integral equation \eqref{lieb_eq_gen}.
\end{subequations}
This heuristic argument allowed Hulthén to compute the ground state energy to leading order.
This approach was then adopted by others.
It was brought to mathematical rigour by \textcite{YanY66a}
where they proved, most importantly, that the coupled system of integral \cref{fermi_condn_gen,lieb_eq_gen} admits a unique solution.
\par
The \emph{Lieb kernel} $K(\la)$ in the integral equation \eqref{lieb_eq_gen} is given by the derivative of the scattering phase $\Theta^\prime(\la,i\gamma)$.
\index{exc@\textbf{Excitations}!condn@\textbf{- condensation}!Lieb kernel@$K$: Lieb kernel|seealso{$\Ncal$}}%
\index{exc@\textbf{Excitations}!condn@\textbf{- condensation}!Lieb kernel@$K$: Lieb kernel|textbf}%
According to the anisotropy $\Delta$, the function $K$ assumes the form:
\begin{align}
	K(\la)
	&=
	\begin{dcases}
	\frac{1}{\pi(\la^2+1)},	&	\Delta=1
	\\
	\frac{\sin(2\gamma)}{2\pi\sinh\left(\la+i\gamma\right)\sinh\left(\la-i\gamma\right)}
	,
	&
	|\Delta|<1
	\\
	\frac{\sinh(2\gamma)}{2\pi\sin\left(\la+i\gamma\right)\sin\left(\la-i\gamma\right)}
	,
	&
	\Delta>1
	\end{dcases}
	.
	\label{lieb_kernel}
\end{align}
For the values of the anisotropy parameter in the massless regime $-1<\Delta\leq 1$, we find that the Fermi zone is non-compact $\mathfrak{F}=\Rset$ and the integral equation \eqref{lieb_eq_gen} can be solved using the Fourier transform.
Whereas for the values of the anisotropy in the massive regime $\Delta>1$, the Fermi zone is compact with the Fermi-boundary $\Lambda=\frac{\pi}{2}$. In fact, due to the periodicity of the parametrisation $\varphi$ in the massive regime, we can see that the Fermi-zone can be mapped onto the unit circle. 
The solution to the Lieb equation \eqref{lieb_eq_gen} can be written as a Fourier series, or in terms of elliptic functions.
With all the three results aggregated, we can write 
\begin{align}
	\rden_g(\la)=
	\begin{dcases}
	\frac{1}{2\cosh\pi\la},	&	\Delta=1
	\\
	\frac{1}{2\cosh\frac{\pi\la}{\gamma}},	&	|\Delta|<1
	\\
	\frac{1}{2\pi}\sum_{n\in\Zset}\frac{e^{2in\la}}{\cosh(n\gamma)}
	,
	&
	\Delta>1.
	\end{dcases}
	\label{lieb_den_sol}
\end{align}
By integrating the density function $\rden_g$, we obtain the leading order term for the counting function $\cfn$ and hence the same for the exponential counting function $\aux$.
For the XXX model $\Delta=1$ we get
\begin{subequations}
\begin{align}
	\cfn_g(\la)&=
	\frac{1}{2\pi}\arctan\sinh\pi\la
	.
	\label{cfn_xxx_est}
	\shortintertext{Hence,}
	\aux_{g}(\la)&=
	\left(
	\frac%
	{\sinh\frac{\pi(\la-\frac{i}{2})}{2}}
	{\sinh\frac{\pi(\la+\frac{i}{2})}{2}}
	\right)^\frac{M}{2}
	.
	\label{aux_xxx_est}
\end{align}
\end{subequations}
Similarly, for the massless XXZ chain $|\Delta|<1$, we get 
\begin{subequations}
\begin{align}
	\cfn_g(\la)&=
	\frac{\gamma}{2\pi}\arctan\sinh\frac{\pi\la}{\gamma}
	.
	\label{cfn_xxz_massless_est}
	\shortintertext{Hence, we can write}
	\aux_{g}(\la)&=
	\left(
	\frac%
	{\sinh\frac{\pi(\la-\frac{i\gamma}{2})}{2\gamma}}
	{\sinh\frac{\pi(\la+\frac{i\gamma}{2})}{2\gamma}}
	\right)^\frac{\gamma M}{2}
	.
	\label{aux_xxz_massless_est}		
\end{align}
\end{subequations}
For the massive XXZ model $\Delta>1$, although a closed form expressions for the counting function can be obtained in terms elliptic functions,
we will not present it explicitly here.
\par
However, Yang and Yang did not prove the condensation property that allows us to rewrite the sum over Bethe roots of the ground state as an integral with the Lieb density.
Recently, this major gap was filled by \textcite{Koz18} as he rigorously proved the condensation property for all values of the anisotropy $\Delta>-1$.
\begin{thm}[Condensation property (Yang-Yang), (\textcite{Koz18})]
\label{thm:cndn_YY_Koz}
For a sufficiently regular function $f$, we can write the sum over ground state roots as the integral over density in the thermodynamic limit
\begin{align}
	\lim_{M\to\infty}
	\frac{1}{M}
	\bm\sum f(\bm\la)
	=
	\int_{\mathfrak{F}}f(\tau)\rden(\tau)d\tau
	\label{condn_prop_gs_gen}
\end{align}
\end{thm}
It is also important to note that in \cite{Koz18}, the condensation property was proved also for a class of excitations of the particle-hole type, where a root is removed (hole added) from the Fermi-zone and an additional root (particle) is placed outside it, on the real line.
It is also expected that this property applies to the excitations with the complex roots which are energetically close to the ground state, that are more relevant here in the scope of this thesis.
These excitations and their nature will be discussed in more detail in the next \cref{sec:spectre_XXX}.
There we will see that the holes play the role of true excitations of the ground state, called spinons. Thus the low-lying condition translates to having a finite number of holes or spinons in the thermodynamic limit.
\par
We will use in this thesis a generalisation of the condensation property \ref{thm:cndn_YY_Koz} that extends it to meromorphic functions, which may have poles on the real line.
It is stated in the following proposition, that we shall also demonstrate, for the compact Fermi-zone case which occurs in the massive XXZ model.
\begin{prop}[Generalised condensation property for massive chains]
\label{gen_condn_prop}
Let $\bm\la$ be the set of Bethe roots for the ground state of a massive XXZ chain $\Delta>1$.
Let $f:\Cset\to\Cset$ be meromorphic function periodic on the real line with principle domain $\mathfrak{F}$ and with poles $\bm w$ on the real line such that its intersection with the set of ground state Bethe roots is empty $\bm w\cap\bm \la=\emptyset$.
Then the thermodynamic limit of the sum over ground state Bethe roots is given by the integral
\begin{align}
	\lim_{M\to\infty}
	\frac{1}{M}
	\bmsum f(\bm\la)
	&=
	\int_{\mathfrak{F}+i0}f(\tau)\rden_g(\tau)d\tau
	+
	2\pi i
	\bmsum
	\frac{\rden_{g}(\bm w)}{1+\aux_{g}(\bm w)}\res f(\bm w)
	\label{gen_cndn_mass_prop}
	.
\end{align}
\end{prop}
\begin{proof}
Let us define the function $g$ as
\begin{align}
	g(\tau)&=
	\frac{\rden_g(\tau)}{1+\aux_{g}(\tau)}
	f(\tau)
	.
\end{align}
We can easily see that the function $g$ is also periodic with $\mathfrak{F}$ as the fundamental domain.
Since $1+\aux_g(\la_a)=0$ for all the Bethe roots $\bm\la$, we have poles of the function $g$ in the set $\bm\la\bm\sqcup\bm w$. 
From \cref{exp_cfn_gen,def_rden} we can see that its residue at each of these poles in $\bm\la$ is given by
\begin{align}
	\res_{\tau=\la_{a}}g(\tau)
	&=
	-
	\frac{1}{2\pi iM}
	f(\la_{a})
	.
	\label{gen_cndn_mass_res_bethe}
\end{align}
On the other hand, the poles in the set $\bm w$ that were inherited from the function $f$ remain simple because it is disjoint with the set of Bethe roots $\bm w\bm\cap\bm\la=\emptyset$ and the residue at each of the poles in $\bm w$ is hence given by
\begin{align}
	\res_{\tau=w_a}g(\tau)
	&=
	\frac{\rden_g(w_a)}{1+\aux_g(w_a)}
	\res_{\nu=w_a} f(\nu)
	\label{gen_cndn_mass_res_extra}
\end{align}
Therefore from \cref{gen_cndn_mass_res_bethe} the sum over $f(\la_{a})$ can be written as
\begin{align}
	\frac{1}{M}\bmsum f(\bm\la)
	=
	-2\pi i\bmsum \res g(\bm\la)&=
	-
	\oint_{\partial[(\mathfrak{F}+i(-\alpha,\alpha))\setminus{\bm w}]}
	g(\tau)d\tau
	.
	\label{gen_cndn_mass_contour}
\end{align}
Here the contour of integration is the boundary of the region shown in \cref{fig:fz_th_punct} which shows a thickened Fermi-zone $\mathfrak{F}+i(-\alpha,\alpha)$ punctured to remove unwanted poles in the set $\bm w$.
\begin{figure}[tb]
    \includegraphics[width=\textwidth]{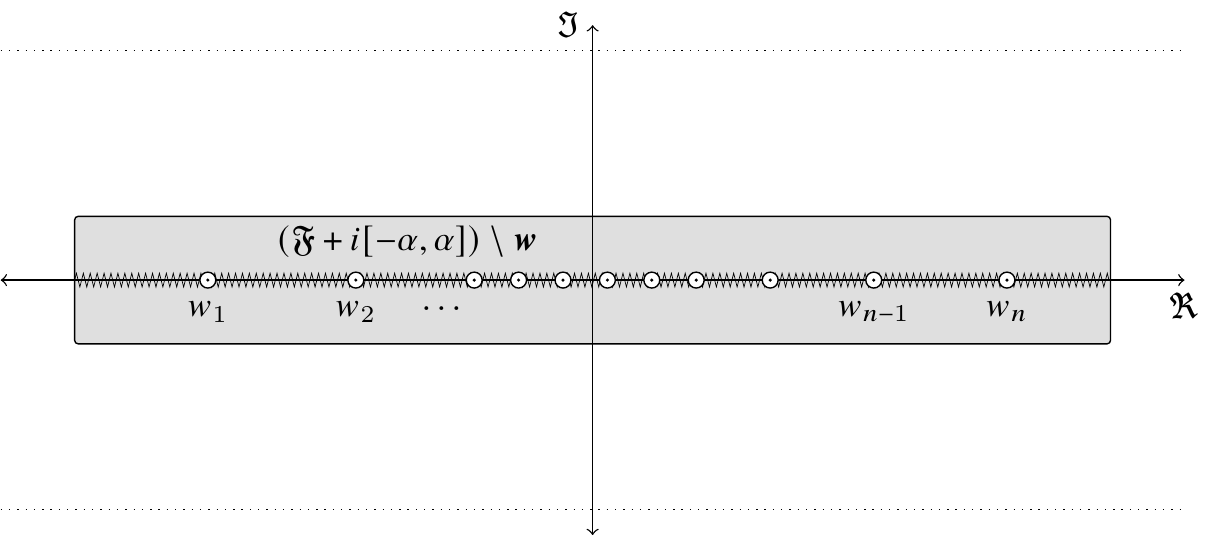}
	\label{fig:fz_th_punct}
	\caption[Thickened Fermi-zone for generalised condensation.]{Thickened Fermi-zone $\mathfrak{F}$ of the width $2\alpha$, punctured at $w_a\in\bm w$ for $a\leq n$ ($n=n_{\bm w}$)}
\end{figure}
The integration on the inner contour simply gives the sum over residues \eqref{gen_cndn_mass_res_extra} for the excluded poles.
The sum of integrals for the edges $\Lambda+i(-\alpha,\alpha)$ and $-\Lambda+i(-\alpha,\alpha)$ of the out contour is zero due to periodicity of the integrand
\begin{align}
	\left(
	\int_{\Lambda-i\alpha}^{\Lambda+i\alpha}
	-
	\int_{-\Lambda-i\alpha}^{-\Lambda+i\alpha}
	\right)
	g(\tau)d\tau
	=0
	.
	\label{gen_cndn_mass_edges}
\end{align}
This leaves us with the part of the contour which we can rewrite as
\begin{multline}
	\left(
	\int_{\mathfrak{F}+i\alpha}
	-
	\int_{\mathfrak{F}-i\alpha}
	\right)
	g(\tau)d\tau
	=
	\int_{\mathfrak{F}+i\alpha}f(\tau)\rden_g(\tau)d\tau
	\\
	-
	\left\lbrace
	\int_{\mathfrak{F}+i\alpha}f(\tau)\rden_g(\tau)
	\frac{\aux_g(\tau)}{1+\aux_g(\tau)}d\tau
	+
	\int_{\mathfrak{F}-i\alpha}f(\tau)\rden_g(\tau)
	\frac{1}{1+\aux_g(\tau)}d\tau
	\right\rbrace
	\label{gen_cndn_mass_para}
\end{multline}
From substituting \cref{gen_cndn_mass_para,gen_cndn_mass_edges,gen_cndn_mass_res_extra} into \cref{gen_cndn_mass_contour} we obtain
\begin{multline}
	\frac{1}{M}
	\bmsum f(\bm\la) 
	=
	\int_{\mathfrak{F}+i\alpha}f(\tau)\rden_g(\tau)d\tau
	+
	2\pi i
	\bmsum
	\frac{\rden_g(\bm w)}{1+\aux_g(\bm w)}
	\res f(\bm w)
	\\
	-
	\left\lbrace
	\int_{\mathfrak{F}+i\alpha}f(\tau)\rden_g(\tau)
	\frac{e^{2\pi iM\cfn(\tau)}}{1+e^{2\pi iM\cfn(\tau)}}d\tau
	+
	\int_{\mathfrak{F}-i\alpha}f(\tau)\rden_g(\tau)
	\frac{e^{-2\pi iM\cfn(\tau)}}{1+e^{-2\pi iM\cfn(\tau)}}d\tau
	\right\rbrace
	\label{gen_cndn_mass_int_corr}
\end{multline}
In the thermodynamic limit, this leads to the result of \cref{gen_cndn_mass_prop} where the error term can be estimated as
\begin{align}
	2\Re
	\int_{\mathfrak{F}+i\alpha}f(\tau)\rden_g(\tau)
	\frac{e^{2\pi iM\cfn(\tau)}}{1+e^{2\pi iM\cfn(\tau)}}d\tau
	=
	O(M^{-\infty})
\end{align}
as we substitute the thermodynamic limit of the counting function $\cfn(\tau\pm i\alpha)$.
Here we also remark that this finite-size correction term is identical to the one obtained by \textcite{VegW85}.
\end{proof}
\subsection{Generalised condensation in the non-compact case}
\label{sub:gen_condn_non-comp_F-zn}
For the massless chains $-1<\Delta\leq 1$ where the Fermi-zone is non-compact $\mathfrak{F}=\Rset$, a similar approach can be used to obtain \cref{gen_cndn_mass_int_corr}.
Instead of the periodicity, here we would need to have the function $f$ that vanishes sufficiently rapidly at infinity, which makes the contribution of the edges of the larger contour vanish as in \cref{gen_cndn_mass_edges}.
\par
However, a problem arises in the estimation of the error terms due to the non-compactness of the Fermi-zone $\mathfrak{F}$.
This is due to the fact that exponential counting function near the tails of the Fermi-distribution do not behave exponentially, rather it is $o(1)$ as the spectral parameter $\la\to \infty$ [see \cref{aux_xxz_massless_est,aux_xxx_est}].
Therefore, we encountered a problem of non-vanishing tails in the integrals
\begin{align}
	2\Re
	\int_{\Rset+i\alpha}f(\tau)\rden_g(\tau)
	\frac{e^{2\pi iM\cfn(\tau)}}{1+e^{2\pi iM\cfn(\tau)}}d\tau
	.
	\label{finite_size_corr_non_compact}
\end{align}
\par
This problem puts us in a precarious position, preventing a complete proof of the generalised condensation property for massless XXZ chain $-1<\Delta\leq 1$.
At this juncture, it is very important to realise that this problem is not unique to the meromorphic nature of the function $f$. It also appears in the case of condensation property for regular functions in non-compact Fermi-zone \cite{Koz18}.
In the regular case, it is important to see that a significant fraction of the Bethe roots are populated in the \emph{bulk} of the Fermi-zone in order to prove the condensation property.
By comparison, we see that we ought to consider in the case of meromorphic functions, two different scenarios.
\\
In the first case, all the poles are taken well within the \emph{bulk} of the Fermi-distribution, where we can write
\begin{align}
	\aux_g(\nu)= O(e^{\sigma \kappa M}),
	\qquad \sigma \Im\nu > 0.
	\label{expcfn_expn_behaviour}
\end{align}
Therefore the function $f$ is regular in the neighbourhood of the problematic part in \cref{finite_size_corr_non_compact}, i.e. the tails of the Fermi-distribution. Therefore it is reasonable to say, by comparison, that in such cases the condensation property holds and the corrections from tails are sub-leading. 
\\
The second scenario that is truly problematic is where at-least one pole of the function $f$ lies near the tail region where \cref{expcfn_expn_behaviour} does not hold. In this part, the corrections from the tails can become non-negligible.
However, in practice the function $f$ always has poles determined by the roots of ground state or a low-lying excited state\footnotemark. For such an eigenvector, most of the roots lie inside the bulk and do not fall in the problematic tail region. Thus it can be argued that that this problematic scenario in question, is an exception rather than a rule. We believe that in this way, the problematic cases arising from the pole outside the bulk can be tamed.
\footnotetext{see for example the functions $t(\la-\mu)$ in the Slavnov determinant representation \eqref{sla_det_rep_gen_all}. \textit{Apropos}, we shall see that these terms are at the origin of the poles that enter our computations in \cref{comp_ff_XXX}.}
But it is hard to show this rigorously, one of the reason is that there is no clear boundary in the Fermi-zone that separates the tails from the bulk, although the bulk is heavily populated in comparison to the tails.
\par
We will assume that the condensation property can be extended to such a meromorphic function $f$ and that the finite size corrections from \cref{finite_size_corr_non_compact} remain sub-dominant, at-least in the \emph{bulk part}.
We also assume that the order of sub-leading correction is $o(\frac{1}{M})$, in order to write
\begin{align}
	\frac{1}{M}
	\bmsum f(\bm\la) 
	=
	\int_{\mathfrak{F}+i\alpha}f(\tau)\rden_g(\tau)d\tau
	+
	2\pi i
	\bmsum
	\frac{\rden_g(\bm w)}{1+\aux_g(\bm w)}
	\res f(\bm w)
	+
	o\left(\frac{1}{M}\right)
	.
	\label{gen_cndn_bulk_asmp}
\end{align}
We will further assume that the contribution of sub-leading terms in the determinant representations, as well as the contribution due the presence of poles outside the bulk, are negligible in the final result for the form-factors in the thermodynamic limit.
In other words, we assume that the contribution of the leading order terms from the bulk part of the Fermi-distribution dominates in the computation of form-factors in the thermodynamic limit.
With this \emph{bulk assumption}, we have managed to obtain the correct result for the thermodynamic limit of the two-spinon form-factor in \cite{KitK19} that is compatible with the previous result \cite{BouCK96}. It will be reproduced in this thesis in \cref{chap:2sp_ff}.
\par
Although we introduced the condensation properties in the context of the ground state, it also applies to the low-lying excitations, obtained by adding a finite number of holes and complex roots.
As an effect of their introduction, we have correction terms of order $O(\frac{1}{M})$ in the total excited state density function $\rden_e$, that appear on top of the leading order terms which is dominated by the ground state density function.
Note that these correction terms are called so only due to their $M^{-1}$ coefficients and the behaviour of these terms is well understood. It is well known that they all satisfy an integral equation, similar to the Lieb \cref{lieb_eq_gen}.
Having a better understanding of these density terms for holes and complex roots tell us about the nature of complex Bethe roots themselves, as we shall see in \cref{sub:DL_picture} where we introduce the DL picture.
\section{Excitations of the XXX ground state}
\label{sec:spectre_XXX}
The Bethe equations \eqref{bae_xxx} for the XXX model are written in the rational parametrisation:
\begin{flalign}
	(\forall j\leq N),
	&&
	\left(
	\frac{%
	\la_j-\frac{i}{2}
	}{%
	\la_j+\frac{i}{2}
	}
	\right)^M
	\bmprod
	\frac{%
	\la_j-\bm\la+i
	}{%
	\la_j-\bm\la-i
	}
	&=
	-1
	.
	&&
	\label{bae_xxx}
\end{flalign}
Its logarithmic form can be written similar to \cref{log_bae_spectre} as follows:
\index{aba@\textbf{Algebraic Bethe ansatz (ABA)}!counting function@$\cfn$: counting function|textbf}%
\begin{flalign}
	(\forall j\leq N),
	&&
	\cfn(\la_j)
	&=
	M\Theta_2(\la_j)
	-
	\bmsum
	\Theta_1(\la_j-\bm\la)
	=
	2\pi i 
	Q_{j}
	.
	&&
	\label{log_bae_xxx}
\end{flalign}
The functions $\Theta_\kappa$ in the rational parametrisation \eqref{log_bae_funs_gen} can be expressed with the $\arctan$ function as follows:
\begin{align}
	\Theta_\kappa(\la)
	=
	2\arctan\left(\frac{2\la}{\kappa}\right)
	.
	\label{log_bae_funs_xxx}
\end{align}
\index{misc@\textbf{Miscellaneous functions}!THETA fn@$\Theta_\kappa$: terms in the logarithmic Bethe \cref{log_bae_all,log_bae_xxx}}%
One of the most characteristic features of the XXX model is its $\mathfrak{su}_2$ symmetry, which distinguishes it from the anisotropic XXZ model model.
An important consequence of its symmetry is the decomposition of the XXX spectrum into multiplets, as we saw it in \cref{sub:XXX_mult}.
We will now take it a step further as we study the spectrum of XXX model in the thermodynamic limit.
Let us first begin with the simpler types of excitations, where there are only real Bethe roots involved.
The treatment of complex roots will be taken up later in \cref{sub:DL_picture,sub:hl_bae}.
\subsection{Spin waves (excitations generated by holes)}
\label{sub:spinon}
Let us start with an eigenstate described by a set $\bm\mu\subset\Rset$ of $n_{\bm\mu}=N_s$ real Bethe roots. 
\index{exc@\textbf{Excitations}!DL@\textbf{- Destri-Lowenstein (DL)}!nr@$n_r$: number of real Bethe roots|textbf}%
\index{exc@\textbf{Excitations}!DL@\textbf{- Destri-Lowenstein (DL)}!real Bethe root@$\bm\rl$: set of real Bethe roots (excl. holes)|textbf}%
We will denote the set of real Bethe roots as $\bm\rl$ and its cardinality by $n_r=n_{\bm{\rl}}$, more so in the future, as in the current context it is rather redundant: $\bm\mu=\bm\rl$ and $n_r=N_s$.
The set of quantum numbers $\bm Q$, which are related to this set through the logarithmic Bethe \cref{log_bae_xxx}, form a subset of (half-)integers $\bm Q\subset [-Q_\text{max};Q_\text{max}]$ which is placed symmetrically around zero.
The value of the maximal quantum number can be computed from the logarithmic Bethe equation \eqref{log_bae_xxx}, as it corresponds to the penultimate quantum number for which the inverse counting function is finite \cite{FadT84}.
This allows us to write,
\begin{align}
  Q_\text{max}=
	M\cfn(\infty)-1
	=
	\frac{M}{4}+\frac{s-1}{2}
  .	
  \label{max_qno_xxx}
\end{align}
Therefore the cardinality $\hat{n}_r$ of the set $\bm{\hat{Q}}=[-Q_\text{max},Q_\text{max}]$ is given by the formula:
\index{exc@\textbf{Excitations}!DL@\textbf{- Destri-Lowenstein (DL)}!no@$\hat n_r$: {occupancy number}|textbf}%
\index{exc@\textbf{Excitations}!DL@\textbf{- Destri-Lowenstein (DL)}!set hle@$\bm{\hle}$: set of holes or spinon spect. param.|textbf}%
\begin{align}
	\hat{n}_r=
	2 Q_\text{max}+1
	=
	\frac{M}{2}+s
	.
	\label{occupancy_rl_only}
\end{align}
It is called the \emph{occupancy number} $\hat{n}_r$. The set $\bm{\hat{Q}}$ represents the range of all permissible quantum numbers.
The difference of cardinalities denoted $n_h=\hat{n}_r-n_r$ gives us number of the quantum numbers that remain unoccupied.
We find that the solution of the logarithmic Bethe \cref{log_bae_xxx} for these missing quantum numbers can also satisfy the Bethe equation \eqref{bae_xxx} for given set $\bm\mu$, although they do not belong to this set of Bethe roots. Such extra real roots are referred to as \emph{holes}, which form a set denoted by $\bm\hle$ ($n_{\bm\hle}=n_h$).
Naturally, in the current context, we have the relation:
\begin{align}
	n_h=
	\hat{n}_r-n_r
	=2s
	\label{num_hle_rle_only_xxx}
\end{align}
\index{exc@\textbf{Excitations}!DL@\textbf{- Destri-Lowenstein (DL)}!nh@$n_h$: num. of holes/ spinons}%
\index{exc@\textbf{Excitations}!DL@\textbf{- Destri-Lowenstein (DL)}!nspin@$s$: total spin}%
\begin{notn}
We denote the union of all real roots of the logarithmic Bethe equations as
\index{exc@\textbf{Excitations}!DL@\textbf{- Destri-Lowenstein (DL)}!real all Bethe root@$\bm\rh$: set of all real Bethe roots (incl. holes)|textbf}%
\begin{align}
	\bm\rh=\bm\rl\bm\cup\bm\hle
	.
	\label{def_rh_bethe_root}
\end{align}
Quite clearly, its cardinality is the same as the occupancy number $n_{\bm\rh}=\hat{n}_r$.
\end{notn}
From the relation \eqref{num_hle_rle_only_xxx} we can again see that the ground state $\ket{\psi_g}$ is a singlet $s=0$ that is fully occupied $n_r=\hat{n}_r=N_0$. 
In other words, it does not contain holes and as a singlet, it is annihilated by the action of both global lowering and raising operators
\begin{subequations}
\begin{align}
	S^{+}\ket{\psi_{g}}=S^{-}\ket{\psi_{g}}=0
	\shortintertext{similarly for the dual,}
	\bra{\psi_g}S^{+}=\bra{\psi_g}S^{-}=0.
\end{align}
\end{subequations}
In all the higher multiplets $s>0$, we always have an excitation given by a Bethe vector, that contains $n_h=2s$ holes in its Fermi distribution.
We call the quasi-particle obtained by addition of holes a spin-wave or \emph{spinon}.
The relation \eqref{num_hle_rle_only_xxx} indicates that spinons obey fractional statistics and always occur in pairs.
\\
Due to the condensation of real roots in the thermodynamic limit, the hole parameters $\bm\hle$ are more suitable to characterise the excitations in comparison to the Bethe roots $\bm\mu$, where the latter are better described through their density function.
Therefore, we will denote the leading Bethe eigenvector of a $(2s+1)$-plet in the spinon picture as $\ket{\psi_{s}^{0}(\bm\hle)}$, in contrast to the magnon picture where the Bethe roots $\bm\mu$ were explicitly used in the vector $\ket{\psi_s^0(\bm\mu)}$, like for example in \cref{sub:XXX_mult}.
\\
Let us recall that the Bethe vectors have highest weight in the $\mathfrak{su}_2$ multiplets
\begin{align}
	S^{+}\Ket{\psi_{s}^{0}(\bm\hle)}=0
	.
\end{align}
Similar to \cref{descendants_xxx} in \cref{sub:XXX_mult}, the descendants of the Bethe vector in a multiplet are denoted:
\begin{align}
	\Ket{\psi_{s}^{\ell}(\bm\hle)}=
	(S^-)^\ell\ket{\psi_s^0(\bm\hle)}
	.	
\end{align}
The ground state in this notation is $\Ket{\psi_{0}^{0}}$ since it does not contain holes.
\minisec{Density functions and their integral equations}
One of the ways to define the density of the roots in the excited state is as follows:
\begin{defn}
\label{def:rden_exc_def}
The \emph{total} density function for the excited state $\rden_e(\nu)$ is defined as the derivative of its counting function.
\begin{align}
	\rden_e(\nu)
	=
	\frac{d\cfn_e(\nu)}{d\nu}
	.
	\label{rden_exc_def}
\end{align}
It is called \emph{total} density because it includes the holes.
\end{defn}
\index{exc@\textbf{Excitations}!condn@\textbf{- condensation}!den_es_real@$\rden_e$: total density function for the excited state (incl. holes).|textbf}%
The condensation property also applies to a low-lying excitation, where this condition means that the number of holes is finite $n_h<\infty$ in the thermodynamic limit $N,M\to\infty$.
The integral equation that we obtain from \cref{log_bae_spectre} for total root density function $\rden_e$ [see \cref{def:rden_exc_def} above] is written down in the following equation:
\begin{align}
	\rden_e(\nu)+\int_{\Rset} K(\nu-\tau)\rden_e(\tau)d\tau
	=
	\frac{1}{2\pi}p'_0(\nu)
	+ 
	\frac{1}{M}
	\bmsum
	K(\nu-\bm\hle)
	.
	\label{inteq_exc_rl_only}
\end{align}
Comparing it with the Lieb equation \eqref{lieb_eq_gen} for the ground state, we see that the total root density function for this type of excitation admit a decomposition:
\begin{align}
	\rden_e(\nu)=
	\rden_g(\nu)
	+
	\frac{1}{M}
	\bmsum
	\rden_h(\nu-\bm\hle)
	.
	\label{rden_exc_decomp_rl_only}
\end{align}
The $\rden_g$ is the ground state density satisfying the Lieb integral \cref{lieb_eq_gen}, thus it admits the Fourier transform and the solution:
\begin{subequations}
\begin{align}
	\what{\rden_g}(t)
	&=
	\frac{e^{-\frac{|t|}{2}}}{1+e^{-|t|}},
	\label{rden_gs_ft_xxx}
	\\
	\rden_g(\nu)
	&=
	\frac{1}{2\cosh\pi\nu}
	\label{rden_gs_sol_xxx}
\end{align}
\end{subequations}
The density term $\rden_h$ due to the presence of holes enters as a correction term of order $O(\frac{1}{M})$ in this expansion \eqref{rden_exc_decomp_rl_only}.
We can see that it satisfies the integral equation:
\begin{align}
	\rden_h(\nu)
	+
	\int_{\Rset}
	K(\nu-\tau)
	\rden_h(\tau)
	d\tau
	=
	K(\nu)
	.
	\label{inteq_rden_hle}
\end{align}
Let us remark that the above integral \cref{inteq_rden_hle} for $\rden_h$ stands for the resolvent of the Lieb \cref{lieb_eq_gen}, in this respect, the density function $\rden_h$ is the resolvent of the Lieb \cref{lieb_eq_gen}.
\\
The computations in this thesis will often require us to consider several variations of the Lieb integral equation.
To make this process more efficient, we introduce a shifted and rescaled version of the Lieb equation as follows:
\index{exc@\textbf{Excitations}!condn@\textbf{- condensation}!den_generic@$\rden_\kappa(\cdot,\alpha)$: generic density function with shift $\alpha$|textbf}%
\begin{align}
	\rden_{\kappa}(\la,\alpha)+
	\int_{\Rset}K(\la-\tau)\rden_{\alpha}(\tau,\alpha)
	d\tau
	&=
	K_{\kappa}(\la-\alpha)
	\label{lieb_inteq_shft_scld_def}
\end{align}
where the function $K_\kappa(\la-\alpha)$ is the Lieb kernel $K$ that is shifted by a complex parameter $\alpha\in\mathbb{C}$ and rescaled by positive real parameter $\kappa>0$, as shown in the following expression:
\index{misc@\textbf{Miscellaneous functions}!Lieb K gen@$K_\kappa$: generic Lieb kernel|textbf}%
\begin{align}
	K_{\kappa}(\nu)&=
	\kappa K(\kappa(\nu))
	.
	\label{lieb_kernel_shft_scl_def}
\end{align}
Since we can write the bare momentum as
\begin{align}
	\frac{1}{2\pi} p'_0(\la)=K_2(\la)
\end{align}
the ground state root density function $\rden_g$ is nothing but $\rden_2(\la)=\rden_g(\la)$. Similarly, we can also see that $\rden_h(\nu)=\rden_1(\nu)$.
However, the integral \cref{lieb_inteq_shft_scld_def} for the generic term $\rden_\kappa(\la,\alpha)$ also allow us to study where the argument or the line of integration is shifted in the imaginary direction.
The different scenarios arising from this generalisation are studied in \cref{chap:den_int_aux}.
We borrow the results from \cref{lieb_hle_den_shftd_inside,ft_shft_scl_rhden} obtained there, to write the Fourier transform of the hole density term
\begin{subequations}
\begin{gather}
	\what{\rden_h}(t)=
	\frac{e^{-|t|}}{1+e^{-|t|}}
	\label{rden_hle_ft}
	\shortintertext{and its solution in closed-form, which can be expressed in terms of the digamma function [see \cref{chap:spl_fns}] as follows:}
	\rden_{h}(\nu)=
	\frac{1}{4\pi}
	\sum_{\sigma=\pm 1}
	\left\lbrace
	\dgamma\left(\frac{1}{2}+\frac{\nu}{2i\sigma}\right)
	-
	\dgamma\left(1+\frac{\nu}{2i\sigma}\right)
	\right\rbrace
	.
	\label{lieb_resolvant}
\end{gather}
\end{subequations}
\begin{defn}
\label{defn:rden*_exc_def}
Let us define the density function of the real roots without the holes through the following expression:
\index{exc@\textbf{Excitations}!condn@\textbf{- condensation}!den_es_real@$\rden*_e$: density function for the real roots of the excited state state (excl. holes).|textbf}%
\begin{align}
	\rden*_e(\nu)=
	\frac{d\cfn_e(\nu)}{d\nu}
	-
	\frac{1}{M}
	\bmsum \delta(\nu-\bm\hle)
	.
	\label{rden*_exc_def}
\end{align}
\end{defn}
The integral equation for the density function $\rden*_e$ \eqref{rden*_exc_def} can be obtained from the condensation property.
It is as follows:
\begin{align}
	\rden*_e(\nu)
	+
	\int_{\Rset}
	K(\nu-\tau)
	\rden*_e(\tau)
	d\tau
	=
	\frac{1}{2\pi}
	p'_0(\nu)
	-
	\frac{1}{M}
	\bmsum
	\delta(\nu-\bm\hle)
	.
	\label{inteq_exc_rl_only_*}
\end{align}
Therefore, it admits the decomposition:
\begin{align}
	\rden*_{e}(\mu)=
	\rden_{g}(\mu)
	+
	\frac{1}{M}\bmsum
	\rden*_{h}(\mu-\bm\hle)
	\label{rden_exc_decomp_rl_only_*}
	.
\end{align}
Let us observe that the leading order term in both \cref{rden_exc_decomp_rl_only,rden_exc_decomp_rl_only_*} is common, and it is equal to the density function $\rden_g$ for the ground state.
Hence, the terms due to the holes $\rden_h$ and $\rden*_h$ in the expansions of the densities $\rden_e$ and $\rden*_e$ [see \cref{rden_exc_decomp_rl_only,rden_exc_decomp_rl_only_*}] are related to each other by the following relation:
\begin{align}
	\rden_h(\nu)
	-
	\rden*_h(\nu)
	=
	\delta(\nu)
	\label{rel_bw_hle_dens}
\end{align}
This leads us to the Fourier transform of $\rden*_h$, as seen in the following expression:
\begin{align}
	\what{\rden*_h}(t)
	=
	\frac{-1}{1+e^{-|t|}}
	.
	\label{rden*_ft}
\end{align}
\minisec{Energy-momentum eigenvalues in the thermodynamic limit}
In the thermodynamic limit the eigenvalues \eqref{ev_energy_mom_gen} can be computed as integrals with densities, using the condensation property.
The density function $\rden*_e$ and its decomposition \eqref{rden_exc_decomp_rl_only_*} allows us to see that the leading term will be dominated by the eigenvalues of the ground state.
\begin{subequations}
\begin{align}
	H\ket{\psi_{s}^{\ell}(\bm\hle)}
	&=
	J(E_g+\bmsum \varepsilon_e(\bm\hle))
	,
	\shortintertext{and}
	P\ket{\psi_{s}^{\ell}(\bm\hle)}
	&=
	J(P_g+\bmsum p_e(\bm\hle))
	.
\end{align}
\end{subequations}
The difference, i.e. the energy when measured over the ground state, is a physically important quantity.
We can see from the above expression that it can be represented as sum over terms for individual spinons, where the functions for $\varepsilon_e$ and $p_e$ characterising these spinon energy and momentum are given by the convolutions:
\begin{subequations}
\begin{align}
	\varepsilon_{e}(\hle)
	&=
	\int_{\Rset}
	\varepsilon_{0}(\tau)\rden*_{h}(\tau-\hle)
	d\tau,
	\\
	p_{e}(\hle)
	&=
	\int_{\Rset}
	p_{0}(\tau)\rden*_{h}(\tau-\hle)
	d\tau.
\end{align}
\end{subequations}
This leads us to the following expressions for the energy and momentum of the spinons.
\begin{subequations}
\begin{align}
	\varepsilon_{e}(\hle)&=
	\frac{\pi}{2\cosh\pi\hle}
	\label{spinon_ene_rl_only}
	,
	\\
	p_{e}(\hle)&=
	\arctan\sinh\pi\hle
	-\frac{\pi}{2}
	\quad (\text{mod }\pi)
	.
	\label{spinon_mom_rl_only}
\end{align}
\label{spinon_ene_mom_rl_only}
\end{subequations}
By comparing the energy and momentum values of the spinons in \cref{spinon_ene_mom_rl_only} above, we get the des \textcite{CloP62} dispersion relation for the \emph{spinons}:
\begin{align}
	\varepsilon_{e}(p)
	=
	-\frac{\pi}{2}\sin p
	\label{disp_rel_spinon}
	.
\end{align}
Similarly, by integrating the $\rden*_h$ function \eqref{rden*_ft} we obtain 
\begin{align}
	\int_{\Rset}\rden*(\tau)d\tau=\frac{-1}{2}
\end{align}
which gives us correctly the value of total spin of the excited state Bethe vector
\begin{align}
	S^{3}\ket{\psi_{s}^{0}(\bm\hle)}
	=
	s\ket{\psi_{s}^{0}(\bm\hle)}
	,
\end{align}
since we have by \cref{num_hle_rle_only_xxx} number spinons or holes $n_h$ is given by $n_h=2s$, in the current context of only real excitations.
\\
The dispersion relation \eqref{disp_rel_spinon} also tells us that the necessary condition for an excitation to be a low-lying state is that the number of holes $n_h$ remains finite in the thermodynamic limit.
It also tells us that the spinon excitations have a vanishing mass gap, therefore we say that the XXX model is massless.
\subsection*{Complex Bethe roots in the thermodynamic limit}
\label{sub:cmplx_bethe_xxx}
The nature of complex Bethe roots of \cref{bae_xxx} is a very delicate issue.
In the finite case, we can only claim with certainty that complex roots of the XXX chain always appear in conjugated pairs.
In the thermodynamic limit however, as \textcite{Bet31} himself realised, we can expect that complex roots form a discernible pattern called \emph{string complex}.
\par
Since we can see that in the Bethe \cref{bae_xxx} for the XXX model reproduced below
\begin{align}
	r(\la_j)
	\bmprod
	\frac{\la_j-\bm\la+i}{\la_j-\bm\la-i}
	&=-1,
	&
	r(\la)&=
	\left(
	\frac{\la-\frac{i}{2}}{\la+\frac{i}{2}}
	\right)^M
	\label{bae_xxx_r}
	.
\end{align}
The function $r(\la)$ becomes singular for non-real values of the spectral parameter as 
\begin{align}
	r(\la)&=O(e^{-\sigma \kappa M}), & \text{for} \quad \sigma\Im\la&>0, ~\sigma=\pm 1.
\end{align}
Thus we can expect that this singularity is counter-balanced by the singularity in the remaining phase terms in the \cref{bae_xxx}.
This can happen if the complex Bethe roots in question approach the singularity $\la_{a}-\la_{b}\pm i = i\delta_{ab}$ up-to an exponentially small correction $\stdv_{ab}=O(M^{-\infty})$ of the same as order of divergence as $r(\la_a)$. This parameter is called the string deviation.
\par
This \emph{string hypothesis} about the behaviour of complex Bethe roots allow us to classify them in string complexes of various length $\ell>1$ called \emph{$\ell$-string} in the thermodynamic limit as follows:
\begin{align}
	\mix{\la}[a]^{(\ell)}_{j}&=
	{z_{a}^{(\ell)}+i\left(\tfrac{1}{2}(\ell+1-2j)+\mix{\stdv}_{a}^{j}\right)}
	,
	&
	j&= 1, 2, \ldots, \ell
	.
	\label{str_cmplx_xxz}
\end{align}
The parameter $z_{a}\in\Rset$ is called the string centre and the upper index tells us the length of string complex generated by this centre.
The real Bethe roots can be seen as $1$-string (string complex of length $1$) with exactly zero string deviation in this picture.
\par
In the stronger version of the string hypothesis where all the deviations are exponentially small, we can rewrite the logarithmic Bethe equation in terms of the aggregated phase terms for string-complexes:
\begin{align}
	\bm\prod\frac{\nu-\bm{\la^{\ell}_{a}}+i}{\nu-\bm{\la^{\ell}_{a}}-i}
	=
	\frac{\nu-z_{a}+(\ell+1)\tfrac{i}{2}}{\nu-z_{a}-(\ell+1)\tfrac{i}{2}}
	(1+O(M^{-\infty}))
	\label{str_phase_term_xxx}
\end{align}
where only the centres of the strings, which are all real parameters enter the final expression.
In the same manner as earlier, we can compute the bare energy and momentum for the complex strings $\varepsilon^{\ell\text{-str}}_0$ and $p_0^{\ell\text{-str}}$.
Using \cref{str_phase_term_xxx}, the logarithmic Bethe equation \cref{log_bae_str_all_xxx} reproduced here in the context of the XXX model by,
\begin{subequations}
\begin{align}
	M\Theta_{\ell}\left(z_{a}^{(\ell)}\right)
	-
	\sum_{k\geq1}\bmsum
	\Theta_{\ell, k}(z^{(\ell)}_{a}-\bm{w^{k}})
	=
	2\pi Q^{(\ell)}_{a}
	\label{log_bae_str_xxx}
\end{align}
where,
\begin{align}
	\Theta_{\ell,k}(\la)
	&=
	\sideset{}{'}\sum_{|\ell-k|\leq r\leq \ell+k}
	\Theta_{r}\left(\la\right)
	.
	\label{log_bae_str_funs_xxx}
\end{align}
	\label{log_bae_str_all_xxx}
\end{subequations}
The primed sum $\sum'$ omits the singularity term for $r=0$ (if it is present). In contrast to the \cref{log_bae}, here we get a set of quantum numbers $\bm{Q^{(\ell)}}$ associated to each subset of Bethe roots forming $\ell$-string complexes.
As we did in the case of 1-strings in \cref{max_qno_xxx,occupancy_rl_only}, we can compute the maximum quantum number for $\ell$-string and the occupancy number.
However, the main difference is that we have to subtract the length of the chain $\ell$ to get the maximal quantum number:
\begin{align}
	Q^{(\ell)}_{\text{max}}=
	\cfn_{\ell\text{-str}}(\infty)-\ell.
\end{align}
This gives us the formula:
\begin{align}
	Q^{(\ell)}_\text{max}=
	\frac{M}{2}
	-
	\sum_{k=1}^{\infty}
	J(\ell,k)n_{k\text{-str}}
	-\frac{1}{2}
	\label{max_qno_str}
\end{align}
and hence the occupancy number $\hat{n}_{\ell\text{-str}}$:
\begin{subequations}
\begin{align}
	\hat{n}_{\ell\text{-str}}
	=
	M
	-2
	\sum_{k=1}^{\infty}
	J(\ell,k)n_{k\text{-str}}
	\label{occupancy_str}
	.
\end{align}
where $J:\Nset^2\to\frac{1}{2}\Nset$ is given by,
\begin{align}
	J(\ell,k)
	=
	\begin{cases}
		\min(\ell,k), &	\ell\neq k;
		\\
		\ell-\frac{1}{2}, & \ell = k.
	\end{cases}
	\label{qno_str_comb_fn}
\end{align}
\end{subequations}
It is possible to compute combinatorially, using \crefrange{max_qno_str}{qno_str_comb_fn}, the total number of eigenvectors obtained from the Bethe ansatz which correctly gives the dimension of the quantum space \cite{FadT84,Kir85}.
However, this does not prove the completeness problem at all since it uses the string hypothesis for the counting. In order to write \cref{str_cmplx_xxz} we require that string deviations $\bm\delta$ vanish uniformly, which is problematic.
The violations of the string hypothesis are demonstrated in \cite{EssKS92,IslP93,HagC07,Vla84} where strings with large deviation are shown to exist or some extra solutions are found which do not fit the scheme of the string hypothesis.
The correct proof of the completeness problem for the XXX chain is provided in \cite{MukTV09}, which changes the perspective and reformulates the problem around the Baxter polynomial. For the XXZ chain the completeness problem for the spectrum remains unsolved.
\\
In addition to the completeness problem, the string picture is also known to be problematic in the thermodynamic analysis, particularly for the XXX chain.
It has been pointed out \cite{Woy82} that the low-lying excitations of the XXZ $0<\Delta\leq 1$ ground state do not contain higher strings of length longer than two.
\par
\Textcite{DesL82} gave an alternate description for the complex roots, in the context of the chiral Gross-Neveu model\footnotemark.
This description does not \textit{a priori} assume the string hypothesis as in \cref{str_cmplx}.
\footnotetext{{The Bethe equations for the chiral Gross-Neveu model are identical to Bethe equations in rational parametrisation obtained here for the XXX chain}}
It splits the complex roots into 2-strings, quartets and wide-pairs based on the imaginary part which determines the choice of branch cuts in the counting function.
This classification presents a more accurate description of the spectrum in the thermodynamic limit.
While speaking strictly in combinatorial terms, the absence of higher strings in this formulation is compensated by the emergence of quartet and wide-pair formations.
\\
The new picture due to \cite{DesL82} is more general in the sense that it can also give rise to higher strings, albeit only in some extreme scenarios.
For these and other reasons, it is this description of \emph{Destri-Lowenstein} (DL) that we will use throughout the computations carried out in \cref{comp_ff_XXX} of this thesis.
\\
The description provided in \cite{DesL82} was extended to anisotropic XXZ model $\Delta>-1$ by \textcite{BabVV83}.
Here we first give detailed description of the DL picture of the excitations in the XXX model in \cref{sub:DL_picture,sub:hl_bae}. The generalisation to the XXZ picture due to \cite{BabVV83} will be briefly discussed in \cref{sec:xxz_spectre}.
\subsection{Destri-Lowenstein picture of excitations}
\label{sub:DL_picture}
Based on the difference in the choice of branch cuts of the function $\Theta_2(\nu-\mu)$ as seen in \cref{fig_branch_cuts}, we classify the complex Bethe roots of the excited state $\bm\mu$ into two categories, called \emph{close-pairs} and \emph{wide-pairs}.
\begin{subequations}
\begin{flalign}
	&\textbf{Close-pair :}	&
	\bm\mu^{\txtcp}&\subset
	\Set{\alpha+i\beta | \alpha\in\Rset,\, 0<|\Im\beta|<1}
	,
	&&&&
	\label{clp_region}
	\\
	&\textbf{Wide-pair :}	&
	\bm\mu^{\txtwp}&\subset
	\Set{\alpha+i\beta |\alpha\in\Rset,\, |\Im\beta|>1}.
	&&&&
	\label{wdp_region}
\end{flalign}
\end{subequations}
Together with the set of real Bethe roots, this gives the complete partition into three categories of cardinalities
\begin{align}
	\bm\mu
	&=
	\bm\rl
	\bm\cup
	\bm\mu^\txtcp
	\bm\cup
	\bm\mu^\txtwp
	,
	&
	\text{with cardinality,}
	\quad
	N_s=n_r+2n_\txtcp+2n_\txtwp
	.
	\label{DL_partn}
\end{align}
Let us recall from the discussion in the preceding \cref{sub:spinon}, that we also have a set of holes $\bm\hle$.
The cardinality of the set of hole parameters $\bm\hle$ will be determined \textit{a posteriori} in this picture.
\\
Let us now define an auxiliary function that will facilitate our computations.
\index{aux@\textbf{Auxiliary functions}!phi@$\phifn$: ratio of Baxter polynomials (or similar)|textbf}%
\begin{defn}
\label{defn:phifn_rat}
We define the $\phifn$ function as the ratio of Baxter polynomials \eqref{baxq} of the excited state and the ground state
\begin{align}
	\phifn(\nu|\bm\mu,\bm\la)
	=
	\frac{q_e(\nu)}{q_g(\nu)}
	=
	\frac{\bmprod(\nu-\bm\mu)}{\bmprod(\nu-\bm\la)}
	.
	\label{phifn_rat_bax_def_spectre}
\end{align}
The sets $\bm\mu$ and $\bm\la$ which corresponds here to the sets of Bethe roots of the excited state and the ground state respectively are deliberately kept in the argument very explicitly to ease the redefinition of the $\phifn$ function producing several different versions that we will use in this thesis.
In the absence of explicit second arguments such as $\phifn(\nu)$, the interpretation will be taken strictly according to the definition \cref{phifn_rat_bax_def_spectre}.
\end{defn}
\index{aux@\textbf{Auxiliary functions}!rat eval@$\revtf$: ratio of eigenvalues of the transfer matrix|textbf}%
\begin{notn}[Ratio of eigenvalues]
Given a set of excited state and ground state, the function $\revtf(\nu)$ defines the ratio of the eigenvalues of the transfer matrix for these two states as
\begin{align}
	\revtf(\nu)=
	\frac{\evtf_e(\nu)}{\evtf_g(\nu)}
	.
	\label{def_revtf}
\end{align}
The ratio of eigenvalues is naturally present in the prefactors of the form-factors as we can see from \cref{qism_ff}.
\label{notn_def_revtf}
\end{notn}
\par
Let us recall that complex roots appear in conjugated pairs, hence their cardinalities are always even and we can further talk about the partitions based on the sign of their imaginary values
\begin{align}
	\bm\mu^{\txtcp}&=\bm\mu^{\txtcp+}\bm\cup\bm\mu^{\txtcp-}
	,
	&
	\bm\mu^{\txtwp}&=\bm\mu^{\txtwp+}\bm\cup\bm\mu^{\txtwp-}
	.
	\label{DL_clp_wdp_pm_ptn}
\end{align}
Due to the auto-conjugacy of the set of Bethe roots, we can see that
\begin{align}
	\bm\mu^{\txtcp-}&=\wbar{\bm\mu^{\txtcp+}}
	,
	&
	\bm\mu^{\txtwp-}&=\wbar{\bm\mu^{\txtwp+}}
	.
	\label{DL_auto_conj}
\end{align}
\begin{notn}
We adopt the notation of shifted parameters $z^\pm$ with respect to its `anchor' $z$ as follows:%
\begin{align}
	z^{\sigma}&=z+\frac{i\sigma}{2}
	,
		   &
	\sigma\Im z > -\frac{1}{2}
	,
	\qquad
	(\sigma=\pm)
	.
	\label{DL_ntn_shft}
\end{align}
The condition $\sigma \Im z>-\frac{1}{2}$ is imposed to ensure that parameters $z^+$ are always in the positive half of the complex plane and $z^-$ in the negative half.
\label{ntn:DL_shft}
\end{notn}
With this notation, we shall now denote a close-pair root in the positive half of the complex plane by the symbol $\clp<+>$ and that on the negative half by $\clp*<->$.
Similarly, the wide pairs by $\wdp<+>$ and $\wdp*<->$ on the positive and negative half of the complex plane respectively.
Here, $\clp*$ and $\wdp*$ denote the complex conjugations of anchors $\clp$ and $\wdp$.
Thus auto-conjugacy requirement of the Bethe roots demand that
\begin{itemize}
	\item if $\clp<+>$ is a close-pair Bethe roots, then $\clp*<->$ should be as well and vice-versa;
	\item if $\wdp<+>$ is a close-pair Bethe roots, then $\wdp*<->$ should be as well and vice-versa.
\end{itemize}
Hence, it is sufficient to parametrise the Bethe roots by specifying all the positive parts of the close-pairs and wide-pairs:
\begin{align}
	\bm\mu^\txtcp&=\bmclp<+>\bm\cup\bmclp*<->
	,
	&
	\bm\mu^\txtwp&=\bmwdp<+>\bm\cup\bmwdp*<->
	.
\end{align}
From the derivative of the counting function, we obtain the integral equation for the total root density function $\rden_e$:
\begin{multline}
	\rden_e(\nu)
	+
	\int_{\Rset}
	K(\nu-\tau)
	\rden_e(\tau)
	d\tau
	=
	\frac{1}{2\pi}
	p'_0(\nu)
	+
	\frac{1}{M}
	\bmsum
	K(\nu-\bm\hle)
	\\
	-
	\frac{1}{M}
	\bmsum
	\left\lbrace
	K(\nu-\bmclp<+>)
	+
	K(\nu-\bmclp*<->)
	\right\rbrace
	-
	\frac{1}{M}
	\bmsum
	\left\lbrace
	K(\nu-\bmwdp<+>)
	+
	K(\nu-\bmwdp*<->)
	\right\rbrace
	.
	\label{inteq_DL_no_str}
\end{multline}
This allows us to write the total density as
\begin{multline}
	\rden_{e}(\nu)
	=
	\rden_{g}(\nu)
	+
	\frac{1}{M}
	\bmsum \rden_{1}(\nu,\bm\hle)
	-
	\frac{1}{M}
	\left\lbrace
	\bmsum \rden_{1}(\nu,\bmclp<+>+i\bm\stdv)
	+
	\bmsum \rden_{1}(\nu,\bmclp<->-i\bm\stdv)
	\right\rbrace
	\\
	-
	\frac{1}{M}
	\left\lbrace
	\bmsum \rden_{1}(\nu,\bmwdp<+>)
	+
	\bmsum \rden_{1}(\nu,\bmwdp*<->)
	\right\rbrace
	+o\left(\frac{1}{M}\right)
	\label{den_expn_no-str_DL}
\end{multline}
where the function $\rden_1(\la,\mu)$ represents the shifted density function, satisfying the integral equation \eqref{lieb_inteq_shft_scld_def} for $\kappa=1$.
In \cref{chap:den_int_aux} we study this generalised form of the integral equation \eqref{lieb_inteq_shft_scld_def}.
In \cref{sec:den_terms_DL_comp_append} at the end of this chapter we recall these results for $\kappa=1$.
\par
In \cref{sec:den_complex} we extend the domain of the density function $\rden_g(\la)=\rden_2(\la)$ to complex values $\la\in\Cset$.
There we found \eqref{lieb_den_clp} that $\rden_g$ can be analytically continued to the region $|\Im\la|<1$ which tells us that the expression \eqref{aux_xxx_est} for the exponential counting function can also be extended to the close-pair strip, as shown in the following:
\begin{align}
	\aux_g(\la)
	&=
	\left(%
	\frac{%
	\sinh\frac{\pi(\la-\frac i2)}{2}
	}{%
	\sinh\frac{\pi(\la+\frac i2)}{2}
	}
	\right)^\frac{M}{2}
	,
	&
	|\Im\la|&<1
	.
	\label{aux_g_clp_tdl}
\end{align}
Due to its periodicity we can write the following
\begin{subequations}
\begin{align}
	\aux_g(\la^+)\aux_g(\la^-)&=1, \qquad (|\Im \la|<\frac{1}{2})
	\label{aux_g_clp_prod_id}
\end{align}
{whereas for the wide-pair we found that it vanishes $\rden_2(\la)=0$ \eqref{lieb_den_wdp}, which leads to}
\begin{align}
	\aux_g(\la)&=1, \qquad (|\Im\la|>1).
	\label{aux_g_wdp_id}
\end{align}
\label{aux_g_cmplx_ids}
\end{subequations}
We reformulate in this thesis the method of \cite{DesL82} to determine the nature of the complex roots in thermodynamic limit  in terms of auxiliary functions.
To do this, we factorise the exponential counting function of the ground state from that of the excited state as follows:
\begin{align}
	\frac{\aux_{e}(\la)}{\aux_{g}(\la)}
	&=
	\frac{\phi(\la+i|\bm\mu,\bm\la)}{\phi(\la-i|\bm\mu,\bm\la)}
	.
	\label{rat_exp_cfn_DL}
\end{align}
The reasoning behind this factorisation is clear, it follows from the observation that the leading order term in the density of roots for the excited state is the density function for the ground state $\rden_g$ as evident from \cref{den_expn_no-str_DL}. We note that a similar reasoning can be found in \cite{DesL82} and also in \cite{BabVV83} for the XXZ model. 
However, here we choose to work in terms of the auxiliary $\phifn$ function [see \cref{defn:phifn_rat}], which will be a recurrent feature throughout our computations.
We compute its thermodynamic limit in \cref{sec:tdl_phifn_append} at the end of this chapter.
The result \eqref{phifn_tdl_nostr} obtained there allows us to write down the thermodynamic limit for \cref{rat_exp_cfn_DL} for the different scenarios, which is summarised in the following paragraphs.
\minisec{For close-pairs}
Substituting the result obtained in \cref{phifn_tdl_nostr} for the thermodynamic limit of the $\phifn$ and \cref{aux_g_clp_tdl} for the ground state exponential counting function $\aux_g$ into \cref{rat_exp_cfn_DL} gives us the following thermodynamic limit of the $\aux_e(\la)$ in close-pair strip $|\Im\la|<1$ :
\begin{multline}
	\aux_{e}(\la)=
	\left(%
	\frac{%
	\sinh\frac{\pi(\la-\frac i2)}{2}
	}{%
	\sinh\frac{\pi(\la+\frac i2)}{2}
	}
	\right)^\frac{M}{2}
	\bmprod\frac{%
	(\la-\bmwdp<->)
	(\la-\bmwdp*<->)
	}{%
	(\la-\bmwdp<+>)
	(\la-\bmwdp*<+>)
	}%
	\\
	\times
	\bmprod\frac{%
	\Gamma\left(\frac{1}{2}+\frac{\la-\bmclp<+>}{2i}\right)
	\Gamma\left(\frac{1}{2}+\frac{\la-\bmclp*<->}{2i}\right)
	\Gamma\left(1-\frac{\la-\bmclp<+>}{2i}\right)
	\Gamma\left(1-\frac{\la-\bmclp*<->}{2i}\right)
	}{%
	\Gamma\left(1+\frac{\la-\bmclp<+>}{2i}\right)
	\Gamma\left(1+\frac{\la-\bmclp*<->}{2i}\right)
	\Gamma\left(\frac{1}{2}-\frac{\la-\bmclp<+>}{2i}\right)
	\Gamma\left(\frac{1}{2}-\frac{\la-\bmclp*<->}{2i}\right)
	}%
	\\
	\times
	\bmprod\frac{%
	\Gamma\left(1+\frac{\nu-\bm\hle}{2i}\right)
	\Gamma\left(\frac{1}{2}-\frac{\nu-\bm\hle}{2i}\right)
	}{%
	\Gamma\left(\frac{1}{2}+\frac{\nu-\bm\hle}{2i}\right)
	\Gamma\left(1-\frac{\nu-\bm\hle}{2i}\right)
	}%
	.
	\label{ex_aux_estmn_nostr}
\end{multline}
Since the term due to $\aux_g(\la)$ in \cref{ex_aux_estmn_nostr} can be estimated from \cref{aux_g_clp_tdl} as
\begin{align}
	\aux_{g}(\nu)
	=
	\left(%
	\frac{%
	\sinh\frac{\pi(\nu-\frac i2)}{2}
	}{%
	\sinh\frac{\pi(\nu+\frac i2)}{2}
	}
	\right)^\frac{M}{2}
	,
\end{align}
it turns out to be singular in the thermodynamic limit for complex values for the parameter $\nu$ in $0<|\Im\nu|<1$. Hence, we can see that this behaviour should be compensated by the remaining terms in \cref{ex_aux_estmn_nostr} in order to have $1+\aux_e(\clp<+>)=0$ or $1+\aux(\clp*<->)=0$ for the close-pair roots.
\par
We now argue that this can be achieved by having a pole (or a zero) in these remaining terms.
We can also reasonably believe that such a pole must be from a close-pair part of this expression as the remaining terms do not have any pole (or zero) in $|\Im\la|<1$.
This forces us to write
\begin{align}
	\clp_{a}-\clp*_{b}=i\stdv_{ab}
	\label{str_condn_clp}
	,
\end{align}
where the string deviation parameter is exponentially small $\stdv_{ab}=O(M^{-\infty})$, as it is responsible for countering the singular term.
Note that unlike the string hypothesis in \cref{str_cmplx} , we did not make this assumption \textit{a priori} in order to arrive at \cref{str_condn_clp}.
We first computed the relevant quantities ($\phifn$ in our case) in the thermodynamic limit, which naturally led us to \cref{str_condn_clp}.
This is the key difference of this approach in contrast to the string hypothesis.
It is also important to point out that the length of strings obtained in this way is limited to two, since this result is obtained in the context of close-pairs.
In addition to 2-strings, we also find that there can exist a new type of formation called a \emph{quartet} consisting of four complex roots.
This is discussed in the following paragraph.
\par
To see all the possible formations resulting from \cref{str_condn_clp}, we let the two sets of anchors for the close pairs coincide, halving the number of anchors for close-pairs
\begin{align}
	\bmclp = \bmclp*
	.
	\label{clp_anchor_redundancy}
\end{align}
Let us note that while writing this, we have silently, i.e. without changing the notation, separated the deviation terms from the set of anchors and incorporated them into the close-pair roots as deviated shifts:
\begin{align}
	\clp<+>+ i\stdv=\clp+\frac{i}{2}+\frac{i}{2}\stdv
	.
\end{align}
For the close-pair in the negative half of the complex plane, the notation $\clp*$ for the anchor becomes redundant due to \cref{clp_anchor_redundancy} hence we can write
\begin{align}
	\clp<->- i\stdv=\clp-\frac{i}{2}-\frac{i}{2}\stdv
	.
\end{align}
The parameter $\clp$ is thus called the \emph{centre} of the close-pair $\clp<\pm>\pm i\stdv$.
However, there are two ways the anchors can be identified in pairs according to \cref{clp_anchor_redundancy}.
This gives rise to the following two types of formations:
\begin{subequations}
\begin{enumerate}[wide=0pt]
	\item[\textbf{2-string:}] if we set $\clp*_{a}=\clp_{a}$, then this parameter must be real and it form a 2-string:
\index{exc@\textbf{Excitations}!DL@\textbf{- Destri-Lowenstein (DL)}!anchor clp@$\bmclp$: set of centres of the close-pairs|textbf}%
\index{exc@\textbf{Excitations}!DL@\textbf{- Destri-Lowenstein (DL)}!anchor clp plus@\hspace{1em}$\bmclp<+>$: set of positive close-pair roots|textbf}%
\index{exc@\textbf{Excitations}!DL@\textbf{- Destri-Lowenstein (DL)}!anchor clp minus@\hspace{1em}$\bmclp<->$: set of negative close-pair roots|textbf}%
	\begin{align}
		\set{\clp<+>_a+i\stdv_a,\,\clp<->_a-i\stdv_a | \clp\in\Rset}.	
		\label{clp_2s}
	\end{align}
	\item[\textbf{quartet:}] if we set $\clp_{a}=\clp*_{b}$ for distinct indices $a$ and $b$, then these parameters can be outside the real line in the strip $|\Im\la|<\frac{1}{2}$. The auto-conjugacy demands that we have $\clp_{b}=\bar\clp_{a}$. This leads to a complex of four close-pair roots:
	\begin{align}
		\Set{%
		\clp<+>_a+i\stdv_{ab},\,%
		\clp<+>_b+i\stdv_{ab},\,%
		\clp<->_a-i\stdv_{ab},\,%
		\clp<->_b-i\stdv_{ab}\,%
		|
		0<\Im(\clp_a)<\frac{1}{2},
		\clp_b=\wbar{\clp_a}
		}
		\label{clp_quartet}
		.
	\end{align}
\end{enumerate}
\label{clp_DL}
\end{subequations}
The set of \emph{centres} of the close-pairs forms a self conjugate set \eqref{clp_anchor_redundancy} which lies in the strip $|\Im\nu|<\frac{1}{2}$.
Its cardinality is equal to the number $n_{\bmclp}=n_c$ and each centre determines two roots $\clp<\pm>$ forming either 2-strings or quartet.
\index{exc@\textbf{Excitations}!DL@\textbf{- Destri-Lowenstein (DL)}!nc@$n_c$: number of close-pair centres|textbf}%
\\
Upon substitution of weaker string hypothesis \eqref{str_condn_clp} into the thermodynamic limit \eqref{phifn_tdl_nostr} for the $\phifn$ function, it takes a simplified form \eqref{phi_tdl}.
Consequently, the expression \eqref{ex_aux_estmn_nostr} also simplifies to the following:
\begin{multline}
	\aux_{e}(\nu)=
	\left(%
	\frac{%
	\sinh\frac{\pi(\la-\frac i2)}{2}
	}{%
	\sinh\frac{\pi(\la+\frac i2)}{2}
	}
	\right)^M
	\bmprod
	\frac{\la-\bmclp+\frac{i}{2}}{\la-\bmclp-\frac{i}{2}}
	\bmprod
	\frac{\la-\bmwdp+\frac{i}{2}}{\la-\bmwdp-\frac{i}{2}}
	\bmprod
	\frac{\la-\bmwdp*+\frac{i}{2}}{\la-\bmwdp*-\frac{i}{2}}
	\\
	\times
	\bmprod\frac{%
	\Gamma\left(1+\frac{\la-\bm\hle}{2i}\right)
	\Gamma\left(\frac{1}{2}-\frac{\la-\bm\hle}{2i}\right)
	}{%
	\Gamma\left(\frac{1}{2}+\frac{\la-\bm\hle}{2i}\right)
	\Gamma\left(1-\frac{\la-\bm\hle}{2i}\right)
	}%
	.
	\label{ex_aux_estmn_str}
\end{multline}
We now move to discuss the case of wide-pairs consisting of complex roots in the region \eqref{wdp_region}.
\minisec{For wide-pairs}
A wide-pairs is still parametrised in terms of its two anchors $\wdp$ and $\wdp*$.
\index{exc@\textbf{Excitations}!DL@\textbf{- Destri-Lowenstein (DL)}!anchor wdp@$\bmwdp$: set of positive anchors of the wide-pairs|textbf}%
\index{exc@\textbf{Excitations}!DL@\textbf{- Destri-Lowenstein (DL)}!anchor wdp-p@\hspace{1em}$\bmwdp<+>$: set of positive wide-pair roots|textbf}%
\index{exc@\textbf{Excitations}!DL@\textbf{- Destri-Lowenstein (DL)}!anchor wdp-m@\hspace{1em}$\bmwdp*<->$: set of negative wide-pair roots|textbf}%
\begin{flalign}
	&\textbf{wide-pair :}
	&
	&\Set{\wdp+\frac{i}{2},\,\wdp*-\frac{i}{2} | \Im{\wdp}>\frac{1}{2}}
	.
	&&
	\label{wdp_DL}
\end{flalign}
But as remarked earlier, when \cref{ntn:DL_shft} was introduced, we can claim with certainty in this case that the anchor $\wdp$ has a positive imaginary part, likewise the anchor $\wdp*$ has a negative imaginary part.
This statement is equivalent to the difference in the nature of branch cuts in the counting function for wide-pairs.
An interesting consequence follows from this remark.
To see it, let us compute now the thermodynamic limit of the expression \eqref{rat_exp_cfn_DL}, using the result for the $\phifn$ function from \cref{chap:den_int_aux}.
At this juncture, it is important to remark that there are two main differences in its computation here as compared to the close-pair case \cref{ex_aux_estmn_nostr}:
\begin{enumerate}[wide]
\item this time we take the thermodynamic limit of the $\phifn$ function given in \cref{phi_tdl} that already uses the weaker string hypothesis \eqref{str_condn_clp} for the close-pair terms appearing in it.
\item here the exponential counting function for the ground state is constant for wide-pairs.
\begin{align}
	\aux_g(\wdp<+>)=\aux_g(\wdp*<->)=1
	\label{wdp_aux_gs_const}
\end{align}
This follows from the fact that the density of ground state roots vanishes in the region where wide-pairs are found. It is seen from \cref{lieb_den_wdp} that we derive in \cref{chap:den_int_aux}.
\end{enumerate}
\par
Substituting the result of \cref{phi_tdl,wdp_aux_gs_const} in \cref{rat_exp_cfn_DL} leads us to
\begin{align}
	\aux_{e}(\la)=
	\begin{dcases}
	\bmprod\frac{\la-\bmclp+\frac{i}{2}}{\la-\bmclp-\frac{3i}{2}}
	\bmprod\frac{(\la-\bmwdp+\frac{i}{2})(\la-\bmwdp*+\frac{i}{2})}{(\la-\bmwdp-\frac{3i}{2})(\la-\bmwdp*-\frac{3i}{2})}
	\bmprod\frac{\la-\bm\hle-i}{\la-\bm\hle}
	,
	&
	\Im\la>1
	;
	\\
	\bmprod\frac{\la-\bmclp+\frac{3i}{2}}{\la-\bmclp-\frac{i}{2}}
	\bmprod\frac{(\la-\bmwdp+\frac{3i}{2})(\la-\bmwdp*+\frac{3i}{2})}{(\la-\bmwdp-\frac{i}{2})(\la-\bmwdp*-\frac{i}{2})}
	\bmprod\frac{\la-\bm\hle+i}{\la-\bm\hle}
	,
	&
	\Im\la<-1
	.
	\end{dcases}
	\label{exc_aux_wdp_region}
\end{align}
Note that there is no singular term here and that means wide-pairs do not form any particular formation.
The imaginary part of the wide-pair is a free parameter as long as it is confined to $\Im\wdp<+>>1$ for the wide-pairs in the positive half of the complex plane or to $\Im\wdp*<-><-1$ for wide-pairs in the negative half.
\begin{rem}
It also means that there is no halving of the number of parameters in the wide-pair case. The set of anchors $\wdp \bm\cup \wdp*$ continues to be the same.
The number $n_w$ denotes the number of positive wide-pair anchors.
\index{exc@\textbf{Excitations}!DL@\textbf{- Destri-Lowenstein (DL)}!nw@$n_w$: number of wide-pair anchors|textbf}%
More importantly, let us note that there are no roots such as $\wdp<->$ or $\wdp*<+>$. In fact, writing it so would be in conflict with the \cref{ntn:DL_shft} and hence it ought to avoided at all costs. Whenever such a parameter may arise in our computations, for whatever reasons, we shall write it explicitly as $\wdp-\frac{i}{2}$ or as, $\wdp*+\frac{i}{2}$, or alternatively, as $\wdp<+>-i$ or as, $\wdp*<->+i$.
\end{rem}
A schematic representation is given in \cref{fig:DL_pic} to present a combined picture due to \cite{DesL82} that is portrayed here.
\begin{figure}[tb]
\begin{subfigure}[t]{.58\textwidth}
\centering
\includegraphics[width=\columnwidth, height=0.5\columnwidth]{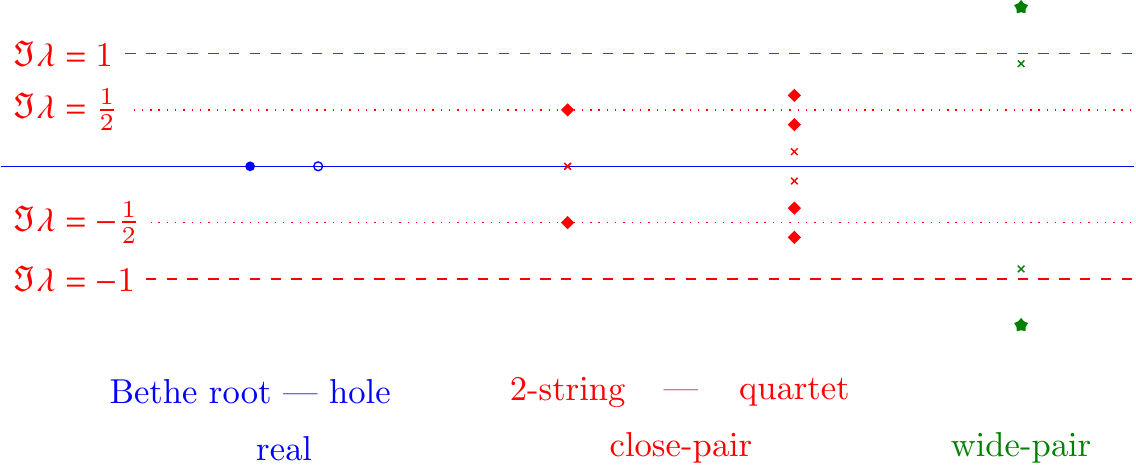}
\caption{A schematic representation of the different types of roots in the complex plane. Here we can also see that close-pairs are condensed into either 2-strings or quartets. The centres of the close-pairs and anchors of the wide-pair are denoted with a cross \tikzcross.}
\end{subfigure}
\quad
\begin{subfigure}[t]{.38\textwidth}
\centering
\includegraphics[width=\columnwidth, height=.75\columnwidth]{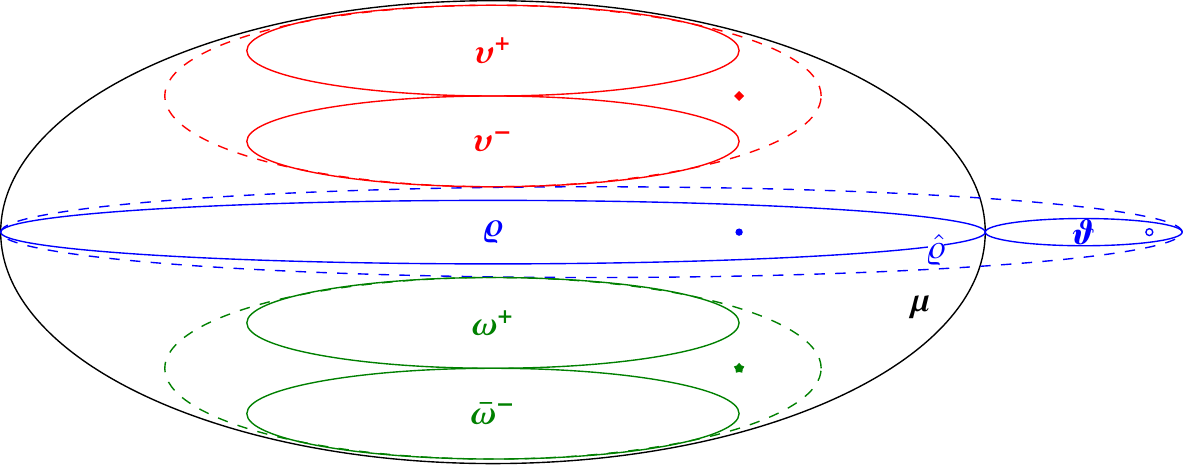}
\caption{Venn diagram for the classification. The set $\bm\hle$ contains holes. The union $\bm\rh=\bm\rl\cup\bm\hle$ represents the complete set of real solutions to the logarithmic Bethe equations.}
\end{subfigure}
\caption[Complex roots in Destri-Lowenstein picture.]{Classifications of roots into the sets of real $\bm\rl$, close-pair $\bmclp<+>$ and $\bmclp<->$ and, wide-pairs $\bmwdp<+>$ and $\bmwdp*<->$ in the Destri-Lowenstein picture.}
\label{fig:DL_pic}
\end{figure}
Let us introduce the notation combining the close-pair and wide-pair parameters.
\minisec{Higher-level roots and common density term}
\begin{notn}[higher-level Bethe roots]
\label{ntn:ho_roots}
Let us define the set of \emph{higher-level} roots composed of all the centres of close-pairs and anchors of the wide-pairs
\index{exc@\textbf{Excitations}!DL@\textbf{- Destri-Lowenstein (DL)}!hl set@$\cid$: set of all higher-level roots|textbf}%
\begin{subequations}
\label{ho_roots_def_all}
\begin{gather}
	\bm\cid=
	\bmclp
	\bm\cup
	\bmwdp
	\bm\cup
	\bmwdp*
	\label{ho_roots_def}
\end{gather}
Its cardinality $\ho{n}=n_{\bm\cid}$ can be expressed in terms of the number of close-pairs and wide-pairs as
\index{exc@\textbf{Excitations}!DL@\textbf{- Destri-Lowenstein (DL)}!nhl@$\ho{n}$: number of higher-level Bethe roots|textbf}%
\begin{gather}
	\ho{n}
	=n_\txtcp+2n_\txtwp.
	\label{ho_roots_card}
\end{gather}
\end{subequations}
\end{notn}
With the weaker string hypothesis of \cref{str_condn_clp}, the integral equation \eqref{inteq_DL_no_str} also becomes simplified.
It is studied in \cref{chap:den_int_aux}.
Let us highlight an important point that can be drawn from there.
We find after solving the respective integral equations that the density terms for the close-pair and wide-pair contribution have the same functional form, although different Fourier transform due to their differing pole structure.
This means that we can write a common density term $\ho{\rden}$.
It is given by the following expression:
\index{exc@\textbf{Excitations}!condn@\textbf{- condensation}!den_hl@$\ho{\rden}$: common density function for close-pairs and wide-pairs|textbf}%
\begin{align}
	\ho{\rden}(\la)
	=
	\frac{1}{2\pi}\frac{1}{\la^2+\frac{1}{4}}
	\label{hl_rden_sol}
\end{align}
Coincidentally, we find that the function $\ho{\rden}$ is the same as the bare momentum function $p'_0$ \eqref{bare_mom_gen} of the XXX model.
However, its Fourier transform is sensitive to the imaginary part of argument and admits different representations in the regions for the close-pairs and wide-pairs.
\begin{subequations}
\begin{align}
	\what{\ho{\rden}}(t,\clp)
	&=
	e^{-|\frac{t}{2}|}e^{-i\clp t}
	,
	\label{ft_ho_rden_clp}
	\\
	\what{\ho{\rden}}(t,\wdp)
	&=
	I_{t<0}(1-e^{-t})e^{-i\wdp t}
	,
	\label{ft_ho_rden_wdp+}
	\\
	\what{\ho{\rden}}(t,\wdp*)
	&=
	I_{t>0}(1-e^{t})e^{-i\wdp* t}
	\label{ft_ho_rden_wdp-}
	.
\end{align}
\end{subequations}
Together, \cref{ntn:ho_roots,hl_rden_sol} allows us to rewrite the expansion \eqref{den_expn_no-str_DL} in a simplified form:
\begin{align}
	\rden_e(\nu)=
	\rden_g(\nu)
	+
	\frac{1}{M}
	\bmsum
	\rden_h(\nu-\bm\hle)
	-
	\frac{1}{M}
	\bmsum
	\ho{\rden}(\nu-\bm\cid)
	.
	\label{den_decomp_DL}
\end{align}
Similarly, we can also write the density function of the real roots without holes, using \cref{rden*_exc_def}. It has the decomposition:
\begin{align}
	\rden*_e(\nu)=
	\rden_g(\nu)
	+
	\frac{1}{M}
	\bmsum
	\rden*_h(\nu-\bm\hle)
	-
	\frac{1}{M}
	\bmsum
	\ho{\rden}(\nu,\bm\cid)
	.
	\label{den_decomp_DL_*}
\end{align}
Let us also recall that in the above expressions, the $\rden_g$ denotes density of the ground state roots while $\rden_h$ and $\rden*_h$ denotes the density terms for the holes, all of which are exactly the same as in \cref{sub:spinon}.
\minisec{Energy and momentum eigenvalues}
Using the condensation property, we can again compute the energy and momentum eigenvalues for the excited state in the DL picture.
Let us denote the leading Bethe vector of a generic low-lying excitation as $\Ket{\psi_{s}^{(\ell)}(\bm\hle|\bm\cid)}$.
The energy and momentum eigenvalues for this vector and its descendant
are determined by the integrals with the density \eqref{den_decomp_DL_*}, which are given in the following expressions.
\begin{subequations}
\begin{align}
	E_e-E_g
	&=
	\bmsum \varepsilon_{2s}(\bmclp)	
	+
	\bmsum \varepsilon_{0}(\bmwdp<+>)
	+
	\bmsum \varepsilon_{0}(\bmwdp*<->)
	+
	\int_{\Rset}
	\varepsilon_0(\tau)
	(\rden*_e(\tau)-\rden_g(\tau))
	d\tau
	;
	\\
	p_e-p_g
	&=
	\bmsum p_{2s}(\bmclp)	
	+
	\bmsum p_{0}(\bmwdp<+>)
	+
	\bmsum p_{0}(\bmwdp*<->)
	+
	\int_{\Rset}
	p_0(\tau)
	(\rden*_e(\tau)-\rden_g(\tau))
	d\tau
	.
\end{align}
\end{subequations}
For the close-pair we assume the vanishing deviation parameters and write the combined terms:
\begin{subequations}
\begin{align}
	\varepsilon_{2s}(\nu)&=
	\varepsilon_{0}(\nu-\tfrac{i}{2})
	+
	\varepsilon_{0}(\nu+\tfrac{i}{2})
	=
	-4\pi K(\nu)
	=
	\frac{-4}{\nu^2+1}
	,	
	\\
	p_{2s}(\nu)&=
	p_0(\nu-\tfrac{i}{2})
	+
	p_0(\nu+\tfrac{i}{2})
	=
	\pi-
	\Theta_{2}(\nu)
	=
	\pi-2\arctan(\nu)
	.
\end{align}
\label{ene_mom_DL_expn}
\end{subequations}
The integrals with density terms of the functions $\varepsilon_0$ and $p_0$ given in \cref{bare_mom_gen,bare_energy_gen} can be computed with the Fourier transform.
We have already seen the result for the convolution with density terms for hole in \cref{spinon_ene_mom_rl_only}.
The convolution integral with density term for complex roots can be obtained using the Fourier transforms in \crefrange{ft_ho_rden_clp}{ft_ho_rden_wdp-}.
This is computed in \cref{chap:den_int_aux}.
After this computation we find that
\begin{subequations}
\begin{align}
	\int_{\Rset}\varepsilon(\tau)
	\ho{\rden}(\tau,\clp)
	d\tau
	&=
	-\varepsilon_{2s}(\tau,\clp)
	,
	&
	\int_{\Rset}p_0(\tau)
	\ho{\rden}(\tau,\clp)
	d\tau
	&=
	-p_{2s}(\clp)
	;
	\label{ene_mom_den_int_clp}
	\\
	\int_{\Rset}\varepsilon(\tau)
	\ho{\rden}(\tau,\wdp<+>)
	d\tau
	&=
	-\varepsilon_{0}(\tau,\wdp<+>)
	,
	&
	\int_{\Rset}p_0(\tau)
	\ho{\rden}(\tau,\wdp<+>)
	d\tau
	&=
	-p_{0}(\wdp<+>)
	;
	\label{ene_mom_den_int_wdp+}
	\\
	\int_{\Rset}\varepsilon(\tau)
	\ho{\rden}(\tau,\wdp*<->)
	d\tau
	&=
	-\varepsilon_{0}(\tau,\wdp*<->)
	,
	&
	\int_{\Rset}p_0(\tau)
	\ho{\rden}(\tau,\wdp*<->)
	d\tau
	&=
	-p_{0}(\wdp*<->)
	.
	\label{ene_mom_den_int_wdp-}
\end{align}
\label{ene_mom_den_int_DL}
\end{subequations}
As a consequence of \cref{ene_mom_den_int_DL}, all the terms due to the complex roots in \cref{ene_mom_DL_expn} are cancelled out in the thermodynamic limit.
This means that the energy momentum eigenvalues are independent from the complex roots $\cid$ and only depend on the hole parameters.
\begin{subequations}
\begin{align}
	\varepsilon_{e}(\hle)&=
	\frac{\pi}{2\cosh\pi\hle}
	\label{spinon_ene}
	,
	\\
	p_{e}(\hle)&=
	\arctan\sinh\pi\hle
	-\frac{\pi}{2}
	\quad (\text{mod }\pi)
	.
	\label{spinon_mom}
\end{align}
\label{spinon_ene_mom}
\end{subequations}
In the form of \cref{spinon_ene_mom}, note that we have obtained exactly the same expressions as we did in \cref{sub:spinon}.
We can also show that the thermodynamic limit ratio of eigenvalues of the transfer matrix [see \cref{notn_def_revtf}] is invariant in this eigenspace.
This is computed in \cref{sec:tdl_revtf_append_sec} at the end of this chapter and the final expression can be found in \cref{revtf_tdl}, which is reproduced below: 
\begin{align}
	\revtf(\la)
	=
	\bmprod
	\tanh\frac{\pi(\la-\bm\hle)}{2}
	.
\end{align}
In the above expression, we can see that the function $\revtf(\la)$ depends only on the hole parameters $\bm\hle$, as we have expected.
\\
In the upcoming \cref{sub:hl_bae}, we will obtain the set of equations which completely determines the set $\bm\cid$ in terms of $\bm\hle$.
We will see that these equations resemble the Bethe \cref{bae_xxx}, it will also permits us to compute the dimension of the degenerate eigenspace with fixed $\bm\hle$.
\par
Although we have seen that complex roots do not affect the eigenvalue of the transfer matrix (and the conserved charges generated by it), they do play a key role in the computation of the total spin $s$.
To see this, let us compute the number of real roots by integrating the density function $\rden*_e$.
It gives us the following relation between the number of real roots on hand and the numbers of holes and complex roots on the other.
\begin{subequations}
\begin{align}
	n_r&=
	M\what{\rden*_{e}}(0)=
	\frac{M}{2}-\frac{n_h}{2}-n_\txtcp
	.
\end{align}
In this expression $n_c$ is the number of close-pairs, while the number of wide-pairs does not enter this expression. Comparing it with the following expression:
\begin{align}
	n_{r}&=
	\frac{M}{2}-s-2n_c-2n_w
\end{align}
\end{subequations}
allows us to express the number of holes in terms of the total spin $s$ and $\ho{n}$ as
\index{exc@\textbf{Excitations}!DL@\textbf{- Destri-Lowenstein (DL)}!nh@$n_h$: num. of holes/ spinons|textbf}%
\index{exc@\textbf{Excitations}!DL@\textbf{- Destri-Lowenstein (DL)}!nspin@$s$: total spin}%
\begin{align}
	n_h=
	2s+2\ho{n}
	\label{hle_num_ho_num_rel}
	.
\end{align}
There are different ways of looking at this formula. With fixed number of holes, we see that adding a higher-level root we move to a lower multiplet $s'=s-1$. In addition to the multiplets, we will often speak of spinon sectors with constant number of spinons.
With fixed total spin $s$ (hence in a given multiplet), we can see that each close-pair require two holes to form while each wide-pair requires four holes to form.
We will now establish the relations that determine the set $\bm\cid$ of higher level roots in terms of the holes $\bm\hle$.
\subsection{Higher-level Bethe equations}
\label{sub:hl_bae}
Let us first observe that both results that were obtained in \cref{ex_aux_estmn_str,exc_aux_wdp_region} can be rewritten with \cref{ntn:ho_roots} as
\begin{align}
	\aux_{e}(\nu)
	&=
	\begin{dcases}
	\bmprod\frac{\nu-\bm\hle-i}{\nu-\bm\hle}
	\bmprod\frac{\nu-\bm\cid+\frac{3i}{2}}{\nu-\bm\cid-\frac{i}{2}}
	;
	& \Im\nu>1
	,
	\\
	\left(%
	\frac{%
	\sinh\frac{\pi(\nu-\frac i2)}{2}
	}{%
	\sinh\frac{\pi(\nu+\frac i2)}{2}
	}
	\right)^M
	\bmprod\frac{%
	\Gamma\left(1+\frac{\nu-\bm\hle}{2i}\right)
	\Gamma\left(\frac{1}{2}-\frac{\nu-\bm\hle}{2i}\right)
	}{%
	\Gamma\left(\frac{1}{2}+\frac{\nu-\bm\hle}{2i}\right)
	\Gamma\left(1-\frac{\nu-\bm\hle}{2i}\right)
	}%
	\bmprod
	\frac{\nu-\bm\cid+\frac{i}{2}}{\nu-\bm\cid-\frac{i}{2}}
	,
	& |\Im\nu|<1
	;
	\\
	\bmprod\frac{\nu-\bm\hle+i}{\nu-\bm\hle}
	\bmprod\frac{\nu-\bm\cid+\frac{i}{2}}{\nu-\bm\cid-\frac{3i}{2}}
	,
	& \Im\nu<-1
	.
	\end{dcases}
	\label{exc_aux_DL}
\end{align}
We will now show that higher-level roots $\bm\cid$ satisfy a set of inhomogeneous Bethe equations, which are known as the higher-level Bethe equations.
The inhomogeneity parameters entering these equations are nothing but the hole parameters $\bm\hle$.
\begin{lem}
Let $\bm\cid$ denote the set of higher-level roots \eqref{ho_roots_def} composed of the parameters from $\bmclp$, $\bmwdp$ and $\bmwdp*$ which describe the complex roots $\bmclp<+>$, $\bmclp<->$,\footnote{deviation for the close-pair are assumed to be exponentially small} $\bmwdp<+>$ and $\bmwdp*<->$ as described in \cref{clp_DL,wdp_DL}.
Then $\bm\cid$ satisfies in the thermodynamic limit the following set of equations:
\begin{flalign}
	(\forall a\leq \ho{n})
	&&
	\bmprod\frac%
	{\cid_{a}-\bm\hle-\frac{i}{2}}%
	{\cid_{a}-\bm\hle+\frac{i}{2}}%
	\bmprod\frac%
	{\cid_{a}-\bm\cid+i}%
	{\cid_{a}-\bm\cid-i}%
	&=-1
	.
	&&
	\label{hl_bae}
\end{flalign}
\label{lem_hl_bae}
\end{lem}
\begin{proof}
Let us begin with the case of a wide-pair $\cid_a=\wdp_{a'}$. Since $\wdp<+>_{a'}$ is a Bethe root, we find that from \cref{exc_aux_DL} for $1+\aux_e(\wdp<+>_{a'})$ leads to the expression:
\begin{align}
	\bmprod\frac%
	{\wdp_{a'}-\bm\hle-\frac{i}{2}}%
	{\wdp_{a'}-\bm\hle+\frac{i}{2}}%
	\bmprod\frac%
	{\wdp_{a'}-\bm\cid+i}%
	{\wdp_{a'}-\bm\cid-i}%
	=-1
	.
	\label{hlbae_wdp+}
\end{align}
Similarly for $\cid_{a}=\wdp*<->_{a'}$, since $\wdp*<->_{a'}$ is a Bethe root we get:
\begin{align}
	\bmprod\frac%
	{\wdp*_{a'}-\bm\hle-\frac{i}{2}}%
	{\wdp*_{a'}-\bm\hle+\frac{i}{2}}%
	\bmprod\frac%
	{\wdp*_{a'}-\bm\cid+i}%
	{\wdp*_{a'}-\bm\cid-i}%
	=-1	
	.
	\label{hlbae_wdp-}
\end{align}
Finally for a close-pair $\cid_a=\clp_{a'}$, we can see that both $\clp<+>_{a'}+i\stdv_{a'}$ and $\clp<->_{a'}-i\stdv{a'}$ are Bethe roots.
This means that we have $\aux_{e}(\clp<\pm>_{a'}\pm i\stdv_{a'})=-1$, however the expression for the exponential counting function obtained in \eqref{exc_aux_DL} contains singular terms for this.
We have seen that this singularity is indeed balanced by a pole of the Gamma function in the parameter $\delta_{a'}$.
Here we them multiply together in the limit $\stdv_{a'}\to 0$ to get:
\begin{align}
	\lim_{\stdv_{a'}\to 0}
	\aux_e(\clp<+>_{a'}+i\stdv_{a'})\aux_e(\clp<->_{a'}-i\stdv_{a'})
	=
	-
	\bmprod\frac%
	{\clp_{a'}-\bm\hle-\frac{i}{2}}%
	{\clp_{a'}-\bm\hle+\frac{i}{2}}%
	\bmprod\frac%
	{\clp_{a'}-\bm\cid+i}%
	{\clp_{a'}-\bm\cid-i}%
	.
	\label{exc_aux_clp_prod}
\end{align}
Since we have $\aux_{e}(\clp<\pm>_{a'}\pm i\stdv_{a'})=-1$, this tells us that for the close-pair also we can write,
\begin{align}
	\bmprod\frac%
	{\clp_{a'}-\bm\hle-\frac{i}{2}}%
	{\clp_{a'}-\bm\hle+\frac{i}{2}}%
	\bmprod\frac%
	{\clp_{a'}-\bm\cid+i}%
	{\clp_{a'}-\bm\cid-i}%
	=-1
	.
	\label{hlbae_clp}
\end{align}
Through the \cref{ntn:ho_roots} that we defined earlier, \cref{hlbae_clp,hlbae_wdp-,hlbae_wdp+} can be collectively written as \cref{hl_bae}.
\end{proof}
\begin{defn}
Let us define the higher-level auxiliary function, or the higher-level version of the exponential counting function $\aux*$ as
\begin{subequations}
\begin{align}
	\aux*(\nu|\bm\cid,\bm\hle)&=
	\bmprod\frac{\nu-\bm\hle-\frac{i}{2}}{\nu-\bm\hle+\frac{i}{2}}
	\bmprod\frac{\nu-\bm\cid+i}{\nu-\bm\cid-i}
	.
	\label{aux_hl_xxx}
\end{align}
\index{aux@\textbf{Auxiliary functions}!exp cfn hl@\hspace{1em}$\aux*$: higher-level equivalent of \rule{3em}{1pt}|textbf}%
More commonly, we will denote it as simply $\aux*(\nu)$.
In terms of it, the higher-level Bethe \cref{hl_bae} can be recast as
\begin{flalign}
	(\forall a\leq \ho{n})	&&
	1+\aux*(\cid_a|\bm\cid,\bm\hle)&=0
	.
	&&
\end{flalign}
\end{subequations}
\end{defn}
Finding explicit solutions of the Bethe equations is a difficult task, the inhomogeneous nature of the higher-level Bethe equations makes it even more so.
But we can still do some prediction by analysing the \cref{hl_bae} alone.
First we can see that the higher-level roots are also either real or they occur in conjugated pairs.
For the low-lying excitations, we also know that the system \eqref{hl_bae} is always finite.
We can compute the number of possible solutions to it based on the following argument.
\par
Let us fix the number of holes $n_h$ as any positive even integer.
The number of solutions $Z(n_h,\ho{n})$ for the higher-level Bethe equations \eqref{hl_bae} for any given $\ho{n}\leq \frac{1}{2}n_h$ can be computed as follows:
We first compute the number of solutions for a parameter $\cid_{a}$ with all other parameters $\bm\cid_{\hat{a}}$ fixed.
Since \cref{hl_bae} is symmetric, the choice of the parameter $\cid_{a}$ does not matter.
Assuming that there are no singularities in \cref{hl_bae}, we can rewrite it in a polynomial form:
\begin{align}
	A(\cid_a|\bm\hle,\bm\cid_{\hat{a}})=0
	.
\end{align}
While doing so we also replace the cross-terms in \eqref{hl_bae} with the other fixed roots using the higher-level Bethe equation for these roots, with the help of the following expression:
\begin{align}
	\frac{\cid_a-\cid_b+i}{\cid_a-\cid_b-i}
	=
	\bmprod
	\frac{\cid_b-\bm\hle-\frac{i}{2}}{\cid_b-\bm\hle+\frac{i}{2}}
	\bmprod
	\frac{\cid_b-\bm{\cid_{\hat{a},\hat{b}}}+i}{\cid_b-\bm{\cid_{\hat{a},\hat{b}}}-i}
	.
\end{align}
With this substitution the cross terms are replaced by constant terms and the polynomial thus obtained after simplification has degree determined by the number of holes $n_h$ and the number of solutions for the variable $\cid_a$ in polynomial $A$ is given by $\binom{n_h}{\ho{n}}$.
However, since we have fixed the remaining roots while doing so, the over-counted solutions must be removed, giving us the formula:
\begin{align}
	P(n_h,\ho{n})
	=
	\binom{n_h}{\ho{n}}
	-
	\binom{n_h}{\ho{n}-1}
	\qquad
	1\leq \ho{n} \leq \frac{n_h}{2}
	.
	\label{num_roots_hl_bae}
\end{align}
The number of solutions for $\ho{n}=0$ is evidently $P(n_h,0)=1$.
For a particular choice of $n_h$ and $\ho{n}$ the total spin is determined by \cref{hle_num_ho_num_rel}.
For total spin $s$ we have additional degeneracy in the $2s+1$ multiplet.
Due to the convexity of this relation, both the total spin $s$ and number of higher-level roots $\ho{n}$ lie in the range $0\leq s, \ho{n} \leq \frac{1}{2}n_{h}$.
We have already seen that energy eigenvalue for the excitation is function of the hole parameters alone \eqref{spinon_ene_mom}.
The presence of complex root only affects the total spin through the relation \eqref{hle_num_ho_num_rel}.
For a given set of holes $\bm\hle$, we get a degenerate eigenspace whose dimension can be obtained by summing up the expression \eqref{num_roots_hl_bae} with the $2s+1$ multiplicity accounting for the descendants of $s=\frac{n_h}{2}-\ho{n}$ multiplet
\begin{align}
	Z(n_h)=
	\sum_{\ho{n}=1}^{\frac{n_h}{2}}
	\left(n_h-2\ho{n}+1\right)
	P(n_h,\ho{n})
	.
\end{align}
This sum can be recast as follows, which gives the dimension
\begin{align}
	Z(n_h)=
	\binom{n_h}{\frac{n_h}{2}}+2\sum_{\ho{n}=0}^{\frac{n_h}{2}-1}\binom{n_h}{\ho{n}}
	=2^{n_h}
	.
	\label{dim_spinon_eigenspace}
\end{align}
This can also be seen from the fact that the higher level Bethe equations are nothing but inhomogeneous version of the Bethe equations for the site of length $n_h$.
\subsubsection{Examples}
We will now study some of the examples of the low-lying excitations.
The state with zero spinons is trivial, it corresponds to the ground state which we thoroughly discussed in the beginning of this chapter.
Let us start with the two-spinon sector.
\minisec{Two-spinon sector}
There are two excitations. The triplet ($s=1$) do not contain any complex root.
We can write down from \cref{phi_tdl}, the thermodynamic limit of the auxiliary $\phifn$ function for this excitation:
\begin{align}
	\phifn(\nu|\bm\mu,\bm\la)
	=
	(2i\sigma)^{-1}
	\frac{%
	\Gamma\left(\frac{\nu-\hle_1}{2i\sigma}\right)%
	\Gamma\left(\frac{\nu-\hle_2}{2i\sigma}\right)%
	}
	{%
	\Gamma\left(\frac{1}{2}+\frac{\nu-\hle_1}{2i\sigma}\right)%
	\Gamma\left(\frac{1}{2}+\frac{\nu-\hle_2}{2i\sigma}\right)%
	}
	,
	\qquad
	\sigma\Im\nu>0
	.
	\label{phifn_tdl_2sp_trip}
\end{align}
In two-spinon singlet $n_h=2$, $s=0$, we have one 2-string since $\ho{n}=1$ whose centre is denoted by $\clp$.
The higher-level Bethe equation \eqref{hl_bae} in this case can be reduced to a linear equation in $\clp$ and thus it can be readily solved to obtain
\begin{align}
	\clp=
	\frac{\hle_1+\hle_2}{2}
	.
	\label{hl_bae_sol_2sp_sing}
\end{align}
The thermodynamic limit of the auxiliary $\phifn$ function is given by,
\begin{align}
	\phifn(\nu|\bm\mu,\bm\la)
	=
	(2i\sigma)^{-1}
	\left(\nu-\clp-\frac{i\sigma}{2}\right)
	\frac{%
	\Gamma\left(\frac{\nu-\hle_1}{2i\sigma}\right)%
	\Gamma\left(\frac{\nu-\hle_2}{2i\sigma}\right)%
	}
	{%
	\Gamma\left(\frac{1}{2}+\frac{\nu-\hle_1}{2i\sigma}\right)%
	\Gamma\left(\frac{1}{2}+\frac{\nu-\hle_2}{2i\sigma}\right)%
	}
	,
	\qquad
	\sigma\Im\nu>0
	.	
	\label{phifn_tdl_2sp_sing}
\end{align}
Another auxiliary result of importance is the thermodynamic limit of the function $\revtf$ for ratio of eigenvalues of the transfer matrix defined in \cref{def_revtf}.
It is natural to expect from our observation in \cref{spinon_ene_mom} that it only depends on the choice of holes $\bm\hle$ in the thermodynamic limit. This is exactly what we find in the explicit computation \eqref{revtf_tdl} for this function in \cref{chap:den_int_aux}.
Here for the two-spinon sector (for both singlet and triplet), the thermodynamic limit of this ratio is given by the expression:
\begin{align}
	\revtf(\nu)
	=
	\tanh\frac{\pi(\nu-\hle_1)}{2}
	\tanh\frac{\pi(\nu-\hle_2)}{2}
	.
	\label{revtf_tdl_2sp}
\end{align}
All the possible configurations and their multiplicities are tabulated below up-to six-spinon sector in the \cref{tab:exc_DL}.
Now we will take here only two examples from the four and six spinon sectors, both of which are triplets.
This bias towards the triplet is a very deliberate choice and the reason behind it will become clear in the next chapter.
\minisec{Four-spinon and six-spinon triplet}
For a four-spinon triplet ($n_h=4$, $s=1$) we get $\ho{n}=1$ and hence it also consists of one 2-string with centre $\clp$.
However, note that this is also the first occurrence of complex root in a triplet.
The higher-level Bethe equation can be reduced to a cubic polynomial in $\clp$ admitting three real roots for the value of center $\clp$.
\begin{align}
	4\clp^3
	- 3\clp^2
	\sum_a \hle_a	
	+\clp
	\left(2\sum_{a\neq b}\hle_a\hle_b-1\right)
	-
	\left(\sum_{a\neq b\neq c}\hle_a\hle_b\hle_c-\frac{1}{4}\sum_{a}\hle_a\right)
	=
	0.
	\label{hl_bae_4sp_poly}
\end{align}
\par
In the six spinon triplet ($n_h=6$, $s=1$) we get two higher-level roots since $\ho{n}=2$.
The higher-level Bethe equations are two coupled equations of degree $6$ for $\clp_1, \clp_2$.
The number of solutions computed from the expression in \cref{num_roots_hl_bae} is $9$.
The nature of these roots is not easy to determine.
Here we can get any of the three possible configurations.
If the two solutions are real then it forms two 2-strings \eqref{clp_2s}, if they occur in conjugated pair, then it forms either a quartet \eqref{clp_quartet} or wide-pair \eqref{wdp_DL} depending on the value of its imaginary part.
In some extreme case it can also lead to the formation of a 3-string if the two roots are located such that their difference is close to $\clp_1-\clp_2\simeq 1$.
We exclude such extreme scenario in all our computations here since \cref{hl_bae} becomes singular in this case. However, this remark is behind the following observation where we compare the DL picture with the string picture.
\begin{table}[tb]
\centering
\arraycolsep=1.5ex
\def\arraystretch{1.25}
\begin{tabular}{|C|C|C||C|C||C|C|}
\hline
	n_h
	&
	s
	&
	\ho{n}
	&
	(n_r,n_c,n_w)
	&
	\hat{n}_r
	&
	P(n_h,\ho{n})
	&
	Z(n_h)
	\\[.25pt]
	\hline
	\hline
	0
	&
	0
	&
	0
	&
	(N_0,0,0)
	&
	N_0
	&
	1
	&
	1
	\\
	\hline
	\hline
	\multirow[c]{2}{*}{2}
	&
	1
	&
	0
	&
	(N_1,0,0)
	&
	N_{-1}
	&
	1
	&
	\multirow[c]{2}{*}{4}
	\\
	\cline{2-6}
	&
	0
	&
	1
	&
	(N_2,1,0)
	&
	N_0
	&
	1
	&
	\\
	\hline
	\hline
	\multirow[c]{4}{*}{4}
	&
	2
	&
	0
	&
	(N_2,0,0)
	&
	N_{-2}
	&
	1
	&
	\multirow[c]{4}{*}{16}
	\\
	\cline{2-6}
	&
	1
	&
	1
	&
	(N_3,1,0)
	&
	N_{-1}
	&
	3
	&
	\\
	\cline{2-6}
	&
	\multirow[c]{2}{*}{0}
	&
	\multirow[c]{2}{*}{2}
	&
	(N_{4},2,0)
	&
	N_{0}
	&
	\multirow[c]{2}{*}{2}
	&
	\\
	\cline{4-5}
	&
	&
	&
	(N_2,0,1)
	&
	N_{-2}
	&
	&
	\\
	\hline
	\hline
	\multirow[c]{6}{*}{6}
	&
	3
	&
	0
	&
	(N_3,0,0)
	&
	N_{-3}
	&
	1
	&
	\multirow[c]{6}{*}{64}
	\\
	\cline{2-6}
	&
	2
	&
	1
	&
	(N_{4},1,0)
	&
	N_{-2}
	&
	5
	&
	\\
	\cline{2-6}
	&
	\multirow[c]{2}{*}{1}
	&
	\multirow[c]{2}{*}{2}
	&
	(N_5,2,0)
	&
	N_{-1}
	&
	\multirow[c]{2}{*}{9}
	&
	\\
	\cline{4-5}
	&
	&
	&
	(N_3,0,1)
	&
	N_{-3}
	&
	&
	\\
	\cline{2-6}
	&
	\multirow[c]{2}{*}{0}
	&
	\multirow[c]{2}{*}{3}
	&
	(N_6,3,0)
	&
	N_0
	&
	\multirow[c]{2}{*}{5}
	&
	\\
	\cline{4-5}
	&
	&
	&
	(N_4,1,1)
	&
	N_{-2}
	&
	&
	\\
	\hline
\end{tabular}
\caption[Table of excitations in DL picture]{Examples of the excitation in the DL picture}
\label{tab:exc_DL}
\vspace{.5em}
\rule{\linewidth}{1pt}
\index{exc@\textbf{Excitations}!DL@\textbf{- Destri-Lowenstein (DL)}!nhl@$\ho{n}$: number of higher-level Bethe roots}%
\index{exc@\textbf{Excitations}!DL@\textbf{- Destri-Lowenstein (DL)}!nr@$n_r$: number of real Bethe roots}%
\index{exc@\textbf{Excitations}!DL@\textbf{- Destri-Lowenstein (DL)}!nspin@$s$: total spin}%
\index{exc@\textbf{Excitations}!DL@\textbf{- Destri-Lowenstein (DL)}!no@$\hat n_r$: {occupancy number}}%
\index{exc@\textbf{Excitations}!DL@\textbf{- Destri-Lowenstein (DL)}!nh@$n_h$: num. of holes/ spinons}%
\index{exc@\textbf{Excitations}!DL@\textbf{- Destri-Lowenstein (DL)}!nh@$n_h$: num. of holes/ spinons}%
\index{exc@\textbf{Excitations}!DL@\textbf{- Destri-Lowenstein (DL)}!nw@$n_w$: number of wide-pair anchors}%
\end{table}
\subsubsection{Compatibility with the string picture}
We saw that the DL picture also consists of strings of lengths no longer than 2 and the higher strings are replaced by quartets and wide-pairs.
However, we can also see that a string of length $3$ can also occur in some extreme cases where we have a singular term $\cid_a-\cid_b=i$ in the higher-level Bethe equation.
This formation can occur from either close-pair or the wide-pair side.
When two close-pairs centres comes close to the difference $\cid_a-\cid_b=i$, this can be seen as a quartet becoming a three string.
In this case one of the roots from the quartet turns into a real root.
From the wide-pair side as two wide-pair anchors come close to the difference $\wdp_a-\wdp*_a=i$, we also get a 3-string.
In this case the real root is turned into one of the three roots of the 3-string on the real line.
This process is consistent with the computation of the occupancy numbers for the real roots, although we will not demonstrate it here.
We also remark that higher string complexes can also form in the same process where quartets and wide-pairs come together.
Such extreme cases are very rare occurrences in the low-lying sector and it cannot affect the thermodynamic computation at near equilibrium and zero-temperature that is presented here.
However, these extreme cases do occur more frequently when we move away from equilibrium or at higher temperature where the low-lying condition is violated and for such studies, the string hypothesis provides more accurate and convenient description for the nature of complex roots.
\par
Let us also remark that the eigenvalues \eqref{spinon_ene_mom} can also be computed in the string picture, starting from \cref{log_bae_str}.
Similarly we can also compute the occupancy numbers for the string using \cref{occupancy_str}.
Here we also find that the two descriptions are compatible. In the string picture also the eigenvalues are given by the expression \eqref{spinon_ene_mom} which is the function of the hole parameters alone.
The computation of the dimensions \eqref{dim_spinon_eigenspace} of the eigenspaces generated from this degeneracy are also found in agreement.
\par
Henceforth we will only use the DL picture with 2-strings \eqref{clp_2s}, quartets \eqref{clp_quartet} and wide-pairs \eqref{wdp_DL}.
The string deviations will be assumed to be small unless when these parameters are important for regularising the expressions.
\section{\mbox{Excitations of the ground state} for the XXZ model \mbox{for \texorpdfstring{$\Delta>-1$}{\textbackslash Delta>-1}}}
\label{sec:xxz_spectre}
For the XXZ model with $\Delta>-1$, the excitations of the ground state are also described by the spinons and their bound states which are described by the holes and the complex Bethe roots in the algebraic Bethe ansatz formalism.
The Destri-Lowenstein picture for the description of complex roots was extended to the XXZ chain by \textcite{BabVV83} for both the massive and massless regime.
We can distinguish between the close-pairs and the wide-pairs and we once again find that the close-pairs are organised in 2-string or quartet formations in the thermodynamic limit.
In the massive regime $\Delta>1$ where we use the hyperbolic parametrisation \eqref{rmat_paramn_types} of the $\Rm$-matrix, the formations prescribed in this picture are similar to the XXX case.
\begin{subequations}
\label{DL_pic_XXZ_massive}
\begin{flalign}
&	\textbf{2-string: }
&& \Set{\clp<+>+i\gamma\stdv,\clp<->-i\gamma\stdv|\clp \in\Rset}
,
&&
\\
&	\textbf{quartet: }
&& \Set{\clp<+>+i\gamma\stdv,\clp*<+>+i\gamma\stdv,\clp<->-i\gamma\stdv,\clp*<->-i\gamma\stdv|0<\Im{\clp}<\tfrac{\gamma}{2}}
,
&&
\\
& \textbf{wide-pair: }
&& \Set{\wdp<+>,\wdp*<->|\Im{\wdp}>\tfrac{\gamma}{2}}
.
&&
\end{flalign}
\end{subequations}
However the shifted $z^\pm$ parameters must be redefined in our notation as follows:
\begin{align}
	z^\pm=
	z\pm\frac{i\gamma}{2}
	.
\end{align}
to include the parameter $\gamma$ which is related with the anisotropy parameter $\Delta=\cosh\gamma$ \eqref{tr_id_aniso}.
The set of centres and anchors $\bm\cid=\bmclp\cup\bmwdp\cup\bmwdp*$ satisfy the set of higher-level Bethe equations. \Textcite{VirW84} found the correct form\footnotemark\space of the higher-level Bethe equations which resembles the \cref{hl_bae} for the XXX model written in trigonometric parametrisation.
\footnotetext{correcting an error in higher-level equation originally found in \cite{BabVV83} for the massive XXZ model}
\begin{align}
	\bmprod
	\frac{\sin(\cid_a-\bm\hle-\frac{i\gamma}{2})}{\sin(\cid_a-\bm\hle+\frac{i\gamma}{2})}
	\bmprod
	\frac{\sin(\cid_a-\bm\cid-i\gamma)}{\sin(\cid_a-\bm\cid+i\gamma)}
	=-1.
	\label{hl_bae_xxz_massive}
\end{align}
It was also shown in \cite{VirW84} that there is a two-fold increase in the number of excitations with given set of hole parameters because we have a two-fold degeneracy of the ground state.
This fact can also be attributed to the fact that the Fermi-zone of the massive XXZ model can be mapped onto the unit circle since the function $\varphi(\la)=\sin\la$ is periodic on the real line.
We can also compute the total density function for the excited state which includes density term for the holes and complex roots (i.e. spinons and their bound states).
Integrating this expression one finds a relation similar to the one obtained for the XXX model \eqref{hle_num_ho_num_rel} that gives the total spin of the excited state $s$ in terms of the number of holes $n_h$ and the number of higher-level roots:
\begin{align}
	s=\frac{1}{2}{n_h}-\ho{n}.
	\label{hle_num_ho_num_spin_rel_xxz_massive}
\end{align}
\par
Let us recall that in the massless region we parametrise the anisotropy parameter as $\Delta=\cos\gamma$.
Here the region needs to be divided into two parts centred around the free fermion point $\Delta=0$ which in terms of the parameter corresponds to $0<\gamma<\frac{\pi}{2}$ and $\frac{\pi}{2}<\gamma<\pi$.
In the first case $0<\gamma<\frac{\pi}{2}$, we can again divide the principle domain of the $\varphi(\la)=\sinh\la$ into the close-pair and wide-pair regions as follows:
\begin{subequations}
\begin{flalign}
&\textbf{close-pair region:}
&&	0<|\Im\la|<\gamma
,
&&
\\
&\textbf{wide-pair region:}
&& \gamma<|\Im\la|<\frac{\pi}{2}
.
\end{flalign}
\end{subequations}
which is similar to the case of XXX model and massive ($\Delta>1$) XXZ model.
We again find the formations of the 2-string, close-pairs and wide-pairs identical to the \cref{DL_pic_XXZ_massive}.
We can also check that in its isotropic limit $\gamma\to 0$ taken with a proper rescaling of the spectral parameters
\begin{align}
	\la\mapsto \alpha=\frac{\la}{\gamma}
	;
	&&
	\la,\gamma \to 0
	.
	\label{xxz_massless_scaling_isotropic}
\end{align}
It leads to the Destri-Lowenstein picture for the XXX model as described through \cref{clp_DL,wdp_DL}.
For the region beyond the free-Fermi point $\frac{\pi}{2}<\gamma<\pi$, the close-pair and wide-pair regions are defined as
\begin{subequations}
\begin{flalign}
&\textbf{close-pair region:}
&&	0<|\Im\la|<\pi-\gamma
,
&&
\\
&\textbf{wide-pair region:}
&& \pi-\gamma<|\Im\la|<\frac{\pi}{2}
.
\end{flalign}
\end{subequations}
Unlike the $\Delta>0$ case, here we do not have 2-string, quartet and wide-pair formations but more string like formations in the thermodynamic limit \cite{BabVV83}.
Nonetheless we can still parametrise the complex roots in terms of the higher-level roots $\bm\cid$.
For both cases, the set $\bm\cid$ solves the higher-level Bethe equations which take a form similar to the \cref{hl_bae,hl_bae_xxz_massive} with hyperbolic parametrisation.
\begin{align}
	\bmprod
	\frac{\sinh(\cid_a-\bm\hle-\frac{i\gamma}{2})}{\sinh(\cid_a-\bm\hle+\frac{i\gamma}{2})}
	\bmprod
	\frac{\sinh(\cid_a-\bm\cid-i\gamma)}{\sinh(\cid_a-\bm\cid+i\gamma)}
	=-1.
	\label{hl_bae_xxz_massless}
\end{align}
\begin{rem}
However the massless case $|\Delta|<1$ differs from both the isotropic and massive cases $\Delta\geq 1$ in two important aspects.
\begin{enumerate}[wide]
\item 
\label{item:xxz_massless_rem_commensur}
From the process of integrating the total density function for these excitation, one arrives to the conclusion \cite{Woy82,BabVV83} that a case must be made separately for the values of $\gamma$ which are commensurate and non-commensurate with respect to $\pi$.
For the non-commensurate values $\gamma\notin\pi\Qset$, we find that only spin zero $s=0$ eigenstates are allowed and hence we always have the relation $n_h=2\ho{n}$ while in the case of commensurate values of $\gamma\in\pi\Qset$, some non-trivial but select values of the total spin are allowed.
\item 
\label{item:xxz_massless_rem_neg_par}
Due to the periodicity of the $\varphi(\la)=\sinh\la$ function along the imaginary axis, we also have the vacancies for the Bethe roots on the line $\Rset+i\pi$ creating the so-called \emph{negative parity} roots.
It is important to remark that the excitations with negative parity roots are equivalent to the descendants \eqref{descendants_xxx} in the XXX model since we can see that under isotropic limit $\gamma\to 0$ with the rescaling \eqref{xxz_massless_scaling_isotropic} we find that the line $\Rset+i\pi$ of negative parity is send to infinity.
Hence the adding negative parameter roots is equivalent in the isotropic limit to the action of lowering operator \eqref{low_raise_ops_xxx} which generates descendants.
\end{enumerate}
\end{rem}
\begin{subappendices}
\section{Density terms in the Destri-Lowenstein picture}
\label{sec:den_terms_DL_comp_append}
\subsection{Before the formation of 2-strings and quartets}
From the integral equation \eqref{inteq_DL_no_str} for the total root density (including holes) we saw that we can write
\begin{multline}
	\rden_{e}(\nu)
	=
	\rden_{g}(\nu)
	+
	\frac{1}{M}
	\bmsum \rden_{1}(\nu,\bm\hle)
	-
	\frac{1}{M}
	\left\lbrace
	\bmsum \rden_{1}(\nu,\bmclp<+>+i\bm\stdv)
	+
	\bmsum \rden_{1}(\nu,\bmclp<->-i\bm\stdv)
	\right\rbrace
	\\
	-
	\frac{1}{M}
	\left\lbrace
	\bmsum \rden_{1}(\nu,\bmwdp<+>)
	+
	\bmsum \rden_{1}(\nu,\bmwdp*<->)
	\right\rbrace
	+o\left(\frac{1}{M}\right)
	.
	\label{ex_tot_den_expn_nostr}
\end{multline}
We compute in \cref{sec:shftd_den_terms_append} of \cref{chap:den_int_aux} the function $\rden_1(\la,\mu)$ for different values of $\Im\mu$.
We find in \cref{lieb_hle_den_shftd_inside} that for $|\Im\mu|<1$ it can be expressed in terms of the $\dgamma$ function.
This includes the density terms for the holes
\begin{align}
	\rden_{1}(\la,\hle_a)&=
	\frac{1}{4\pi}
	\sum_{\sigma=\pm 1}
	\left\lbrace
	\dgamma\left(\frac{1}{2}+\frac{\la-\hle_a}{2i\sigma}\right)
	-
	\dgamma\left(1+\frac{\la-\hle_a}{2i\sigma}\right)
	\right\rbrace
	\label{hle_den_term_dgamma_append_sec}
\end{align}
as well as the density terms for the close-pairs
\begin{subequations}
\begin{multline}
	\rden_1(\la,\clp<+>_a)=
	\frac{1}{4\pi}
	\left\lbrace
	\dgamma\left(\frac{1}{4}+\frac{\la-\clp_a}{2i}\right)
	+
	\dgamma\left(\frac{3}{4}-\frac{\la-\clp_a}{2i}\right)
	\right.
	\\
	\left.
	-
	\dgamma\left(\frac{3}{4}+\frac{\la-\clp_a}{2i}\right)
	-
	\dgamma\left(\frac{5}{4}-\frac{\la-\clp_a}{2i}\right)
	\right\rbrace
	,
	\label{clp+_den_term_dgamma_append_sec}
\end{multline}
and,
\begin{multline}
	\rden_1(\la,\clp*<->_a)=
	\frac{1}{4\pi}
	\left\lbrace
	\dgamma\left(\frac{5}{4}+\frac{\la-\clp*_a}{2i}\right)
	+
	\dgamma\left(\frac{3}{4}-\frac{\la-\clp*_a}{2i}\right)
	\right.
	\\
	\left.
	-
	\dgamma\left(\frac{3}{4}+\frac{\la-\clp*_a}{2i}\right)
	-
	\dgamma\left(\frac{1}{4}-\frac{\la-\clp*_a}{2i}\right)
	\right\rbrace
	.
	\label{clp-_den_term_dgamma_append_sec}
\end{multline}
\label{clp_den_term_dgamma_append_sec_both}
\end{subequations}
The density terms for the wide-pairs can be directly obtained from the \cref{lieb_hle_den_shftd_outside} which gives us
\begin{subequations}
\begin{align}
	\rden_{1}(\la,\wdp<+>)&=
	\frac{1}{2\pi}
	\frac{1}{(\la-\wdp)^2+\frac{1}{4}}
	,
	\label{wdp+_den_term_dgamma_append_sec}
	\\
	\rden_{1}(\la,\wdp*<->)&=
	\frac{1}{2\pi}
	\frac{1}{(\la-\wdp*)^2+\frac{1}{4}}
	.
	\label{wdp-_den_term_dgamma_append_sec}
\end{align}
\label{wdp_den_term_rat_append_sec_both}
\end{subequations}
\subsection{After the formation of 2-string and quartet}
Assuming that the close-pairs form the 2-string or quartet formations \eqref{clp_DL} in the thermodynamic limit with the vanishing string deviation $\stdv\sim e^{-M}$ we find that the density terms for the close-pairs \eqref{clp_den_term_dgamma_append_sec_both} can be put together to obtain
\begin{align}
	\rden_{1}(\la,\clp<+>)+\rden_{1}(\la,\clp<->)
	&=
	\frac{1}{2\pi}
	\frac{1}{(\la-\clp)^2+\frac{1}{4}}
	\label{clp_den_str_condn}
	.
\end{align}
Let us observe that it is a rational function which has the same form as the density term for the wide-pairs \eqref{wdp_den_term_rat_append_sec_both} which allows us to define the common density term for the complex roots $\ho\rden$:
\begin{align}
	\ho{\rden}(\la)&=
	\frac{1}{2\pi}\frac{1}{\la^2+\frac{1}{4}}
	.
	\label{den_ho_def}
\end{align}
so that \cref{ex_tot_den_expn_nostr} can be rewritten as
\begin{align}
	\rden_{e}(\nu)&=
	\rden_{g}(\nu)
	+\frac{1}{M}
	\bmsum \rden_{h}(\la-\bm\hle)
	-\frac{1}{M}
	\bmsum \ho{\rden}(\la-\bm\cid)
	+o\left(\frac{1}{M}\right)
	.
	\label{ex_tot_den_expn}
\end{align}
\section{Thermodynamic limit of the \texorpdfstring{$\phi$}{\textbackslash phi} function.}
\label{sec:tdl_phifn_append}
\index{aux@\textbf{Auxiliary functions}!phi@$\phifn$: ratio of Baxter polynomials (or similar)}%
Let us compute the thermodynamic limit of the $\phifn$ function [see \cref{defn:phifn_rat}].
Let us first begin with the function $\phifn(\la|\bm\rh,\bm\la)$ where $\bm\rh$ denotes the set of all real roots, which includes the holes $\bm\rh=\bm\rl\bm\cup\bm\hle$.
As long as we stay away from its singularity on the real line $\Im\la\neq 0$,
we can write the logarithmic derivative of this function as an integral over density terms as
\begin{align}
	\frac{d}{d\la}\log\phifn(\la|\bm\rh,\bm\la)
	&=
	M
	\int_{\Rset}\frac{1}{\la-\tau}\left(\rden_{e}(\tau)-\rden_{g}(\tau)\right)
	d\tau
	.
	\label{phifn_log_der_den_conv}
\end{align}
To compute these convolution we use the following result.
\begin{lem}
\label{lem:conv_semi-hol_fns_rat}
Let $f(\la)$ be a meromorphic function with vanishing tails $f(\la)\sim \la^{-2}$ which is also holomorphic in the upper (lower) half of the complex plane.
Then its convolution with the simple fraction is determined by,
\begin{subequations}
\begin{alignat}{2}
	\int_{\Rset}\frac{1}{\la-\tau}f(\tau)d\tau &= 0
	&\quad \text{whenever, } \Im\la < 0
	,
\shortintertext{and,}
	\int_{\Rset}\frac{1}{\la-\tau}f(\tau)d\tau &=
	f(\la)
	&\quad \text{whenever, } \Im\la > 0.
\end{alignat}
\end{subequations}
\end{lem}
According to the expansions in \cref{ex_tot_den_expn_nostr,ex_tot_den_expn} for the total density function $\rden_e$ in \cref{ex_tot_den_expn} we consider here two scenarios.
\subsection{Before the formation of 2-string and quartet.}
Let us use the expressions for the density terms obtained in \cref{clp_den_term_dgamma_append_sec_both,wdp_den_term_rat_append_sec_both,hle_den_term_dgamma_append_sec}.
Using \cref{lem:conv_semi-hol_fns_rat} to compute the convolutions \eqref{phifn_log_der_den_conv} we get the following
\begin{align}
	\phifn(\la|\bm\rh,\bm\la)&=
	\begin{dcases}
	\begin{multlined}[b]
	(2i)^{s+2n_\txtwp}
	\bmprod\frac{\la-\bmwdp*<+>}{\la-\bmwdp*<->}
	\\
	\times
	\bmprod\frac{%
	\Gamma\left(\frac{1}{2}+\frac{\la-\bmclp<+>}{2i}\right)
	\Gamma\left(\frac{1}{2}+\frac{\la-\bmclp*<->}{2i}\right)
	}{%
	\Gamma\left(1+\frac{\la-\bmclp<+>}{2i}\right)
	\Gamma\left(1+\frac{\la-\bmclp*<->}{2i}\right)
	}
	\bmprod\frac{%
	\Gamma\left(1+\frac{\la-\bm\hle}{2i}\right)
	}{%
	\Gamma\left(\frac{1}{2}+\frac{\la-\bm\hle}{2i}\right)
	}
	\end{multlined}
	&
	\text{for, }~ \Im\la>0
	;
	\\[2\jot]
	\begin{multlined}[b]
	(-2i)^{s+2n_\txtwp}
	\bmprod\frac{\la-\bmwdp<->}{\la-\bmwdp<+>}
	\\
	\times
	\bmprod\frac{%
	\Gamma\left(\frac{1}{2}-\frac{\la-\bmclp<+>}{2i}\right)
	\Gamma\left(\frac{1}{2}-\frac{\la-\bmclp*<->}{2i}\right)
	}{%
	\Gamma\left(1-\frac{\la-\bmclp<+>}{2i}\right)
	\Gamma\left(1-\frac{\la-\bmclp*<->}{2i}\right)
	}
	\bmprod\frac{%
	\Gamma\left(1-\frac{\la-\bm\hle}{2i}\right)
	}{%
	\Gamma\left(\frac{1}{2}-\frac{\la-\bm\hle}{2i}\right)
	}
	\end{multlined}
	&
	\text{for, }~ \Im\la<0
	.
	\end{dcases}
	\label{phifn_w-hle_tdl_nostr}
\end{align}
Therefore for the ratio of Baxter polynomial $\phifn(\nu|\bm\mu,\bm\la)$ we get
\begin{align}
	\phifn(\la|\bm\mu,\bm\la)&=
	\begin{dcases}
	\begin{multlined}[b]
	(2i)^{n_{r}-N_{0}}
	\bmprod\left(\la-\bmwdp<+>\right)\left(\la-\bmwdp*<+>\right)
	\\
	\times
	\bmprod\frac{%
	\Gamma\left(\frac{1}{2}+\frac{\la-\bmclp<+>}{2i}\right)
	\Gamma\left(\frac{1}{2}+\frac{\la-\bmclp*<->}{2i}\right)
	}{%
	\Gamma\left(\frac{\la-\bmclp<+>}{2i}\right)
	\Gamma\left(\frac{\la-\bmclp*<->}{2i}\right)
	}
	\bmprod\frac{%
	\Gamma\left(\frac{\la-\bm\hle}{2i}\right)
	}{%
	\Gamma\left(\frac{1}{2}+\frac{\la-\bm\hle}{2i}\right)
	}
	\end{multlined}
	&
	\text{for,}~\Im\la>0
	;
	\\[2\jot]
	\begin{multlined}[b]
	(-2i)^{n_{r}-N_{0}}
	\bmprod%
	\left(\la-\bmwdp<->\right)%
	\left(\la-\bmwdp*<->\right)%
	\\
	\times
	\bmprod\frac{%
	\Gamma\left(\frac{1}{2}-\frac{\la-\bmclp<+>}{2i}\right)
	\Gamma\left(\frac{1}{2}-\frac{\la-\bmclp*<->}{2i}\right)
	}{%
	\Gamma\left(-\frac{\la-\bmclp<+>}{2i}\right)
	\Gamma\left(-\frac{\la-\bmclp*<->}{2i}\right)
	}
	\bmprod\frac{%
	\Gamma\left(-\frac{\la-\bm\hle}{2i}\right)
	}{%
	\Gamma\left(\frac{1}{2}-\frac{\la-\bm\hle}{2i}\right)
	}
	\end{multlined}
	&
	\text{for,}~\Im\la<0
	.
	\end{dcases}
	\label{phifn_tdl_nostr}
\end{align}
\subsection{After the formation of the 2-string and quartets}
Assuming that the close-pairs forms 2-string or quartet formations \eqref{clp_DL}, we see that terms for the close-pairs in \cref{phifn_tdl_nostr,phifn_w-hle_tdl_nostr} factorises to produce
\begin{align}
	\phifn(\la|\bm\nu,\bm\la)&=
	\begin{dcases}
	(2i)^{s+\ho{N}}
	\bmprod\frac{\la-\bmwdp*<+>}{\la-\bmwdp*<->}
	\bmprod\frac{1}{\la-\bmclp<->}
	\bmprod\frac{%
	\Gamma\left(1+\frac{\la-\bm\hle}{2i}\right)
	}{%
	\Gamma\left(\frac{1}{2}+\frac{\la-\bm\hle}{2i}\right)
	}
	&
	\text{for, } \Im\la>0
	;
	\\[\jot]
	(-2i)^{s+\ho{N}}
	\bmprod\frac{\la-\bmwdp<->}{\la-\bmwdp<+>}
	\bmprod\frac{1}{\la-\bmclp<+>}
	\bmprod\frac{%
	\Gamma\left(1-\frac{\la-\bm\hle}{2i}\right)
	}{%
	\Gamma\left(\frac{1}{2}-\frac{\la-\bm\hle}{2i}\right)
	}
	&
	\text{for, } \Im\la<0
	.
	\end{dcases}
	\label{phi_wh_asym_strclp}
\end{align}
And
\begin{align}
	\phifn(\la|\bm\mu,\bm\la)&=
	\begin{dcases}
	(2i)^{-s-\ho{N}}
	\bmprod{(\la-\bmwdp<+>)}{(\la-\bmwdp*<+>)}
	\bmprod{(\la-\bmclp<+>)}
	\bmprod\frac{%
	\Gamma\left(\frac{\la-\bm\hle}{2i}\right)
	}{%
	\Gamma\left(\frac{1}{2}+\frac{\la-\bm\hle}{2i}\right)
	}
	&
	\text{for, } \Im\la>0
	;
	\\[\jot]
	(-2i)^{-s-\ho{N}}
	\bmprod{(\la-\bmwdp*<->)}{(\la-\bmwdp<->)}
	\bmprod{(\la-\bmclp<->)}
	\bmprod\frac{%
	\Gamma\left(-\frac{\la-\bm\hle}{2i}\right)
	}{%
	\Gamma\left(\frac{1}{2}-\frac{\la-\bm\hle}{2i}\right)
	}
	&
	\text{for, } \Im\la<0
	.
	\end{dcases}
	\label{phi_tdl}
\end{align}
\section{Thermodynamic limit of the ratio of the eigenvalues \texorpdfstring{$\revtf$}{}}
\label{sec:tdl_revtf_append_sec}
\index{aux@\textbf{Auxiliary functions}!rat eval@$\revtf$: ratio of eigenvalues of the transfer matrix}%
This function $\revtf$ representing the ratio of eigenvalues of the transfer matrix [see \cref{notn_def_revtf}] can be expressed in terms of the $\phifn$ function [see \cref{defn:phifn_rat}] and the exponential counting function \eqref{bae_aux_gen} as
\begin{align}
	\revtf(\la|\bm\mu,\bm\la)&=
	\frac{1+\aux_{e}(\la)}{1+\aux_{g}(\la)}
	\frac{\phi(\la-i|\bm\mu,\bm\la)}{\phi(\la|\bm\mu,\bm\la)}
	.
\end{align}
When evaluated for a ground state root $\la_{a}\in \bm\la$ we get
\begin{align}
	\revtf(\la_{a}|\bm\mu,\bm\la)&=
	\frac{1+\aux_{e}(\la_{a})}{\aux_{g}^\prime(\la_{a}|\bm\mu,\bm\la)}
	\frac{\phi(\la_{a}-i|\bm\mu,\bm\la)}{\phi'(\la_{a}|\bm\mu,\bm\la)}
	.
	\label{revtf_eval_gs_bethe_root}
\end{align}
The prime symbol $\prime$ in $\phifn'$ of the above equation \eqref{revtf_eval_gs_bethe_root} in the denominator denotes an omission of a resulting pole from the product:
\begin{align}
	\phifn'(\la_a|\bm\mu,\bm\la)=\frac{\prod_j(\la_a-\mu_k)}{\prod_{j\neq a}(\la_a-\la_j)}
\end{align}
To obtain the thermodynamic limit of $\revtf$ function \eqref{revtf_eval_gs_bethe_root} we will use the following lemma from \cite{IzeKMT99}.
\begin{lem}[\cite{IzeKMT99}]
\label{lem:IKMT}
The thermodynamic limit of the function $\phi^\prime(\la|\bm\mu,\bm\la)$ is given in terms of the jump of the function $\phi(\la|\bm\mu,\bm\la)$ across the discontinuity due to the cut on the real line. It is given by,
\begin{align}
	\phi^\prime(\la|\bm\mu,\bm\la)&=
	\frac{1}{2\pi i M \rden_{g}(\la)}
	\lim_{\epsilon\to 0}
	\left\lbrace
	\phi(\la+i\epsilon|\bm\mu,\bm\la)
	-
	\phi(\la-i\epsilon|\bm\mu,\bm\la)
	\right\rbrace
\end{align}
\end{lem}
Since the $\la_{a}\in \bm\la$ is a Bethe root, using \cref{bae_xxx} we can write $\aux_e(\la_a)$ in \eqref{revtf_eval_gs_bethe_root} as
\begin{align}
	\aux_{e}(\la_{a})=
	-\frac{\phifn(\la_{a}+i|\bm\mu,\bm\la)}{\phifn(\la_{a}-i|\bm\mu,\bm\la)}
	.
\end{align}
Therefore, the expression $\revtf(\la_{a})$ can written in terms of the $\phifn$ functions as 
\begin{align}
	\revtf(\la_{a}|\bm\mu,\bm\la)&=
	\left\lbrace
	\frac{\phifn^{-1}(\la_{a}-i|\bm\mu,\bm\la)-\phifn^{-1}(\la_{a}+i|\bm\mu,\bm\la)}{\aux_{e}^\prime(\la_{a}|\bm\mu,\bm\la)}
	\right\rbrace
	\frac{%
	\phifn(\la_{a}+i|\bm\mu,\bm\la)\phifn(\la_{a}-i|\bm\mu,\bm\la)
	}{%
	\phifn^\prime(\la_{a}|\bm\mu,\bm\la)
	}
	\label{revtf_phifn_form}
\end{align}
From the thermodynamic limit of the $\phifn$ \eqref{phi_tdl} we see the following relation:
\begin{align}
	\phifn^{-1}(\la_{a}\pm i)&=
	\frac{\bmprod (\la_{a}-\bm\hle)}{%
	\bmprod\left((\la-\bm\cid)^2+\frac{1}{4}\right)%
	}
	\lim_{\epsilon\to 0}
	\phifn(\la_{a}\pm i\epsilon|\bm\mu,\bm\la)
	\label{phifn_tdl_inv_ieps_identity}
\end{align}
where $\bm\cid=\bmclp\cup\bmwdp\cup\bmwdp*$.
Finally, using \cref{lem:IKMT} on \cref{phifn_tdl_inv_ieps_identity,revtf_phifn_form}, we obtain
\begin{align}
	\revtf(\la_{a}|\bm\mu,\bm\la)&=
	\bmprod\tanh\frac{\pi(\la_{a}-\bm\hle)}{2}
	.
	\label{revtf_tdl}
\end{align}
\end{subappendices}
\clearpage{}%

\setpartpreamble[c][.9\textwidth]{%
\vspace{3em}
\small
This part is divided into three chapters.
\Cref{chap:2sp_ff} presents our method of computing the form-factors with an example of the two-spinon form-factors.
\Cref{chap:cau_det_rep_gen,chap:gen_FF} generalises this method to form-factors of generic excitations.
The final result is obtained in \cref{chap:gen_FF}, where we also discuss a particular example of the form-factors of four-spinon excitations.
}
\part[Computation of XXX form-factors]{\mbox{Computation of the form-factors} of \mbox{the Heisenberg spin chain} in \mbox{the thermodynamic limit}}
\label{comp_ff_XXX}
\clearpage{}%
\chapter{Two-spinon form-factors}
\label{chap:2sp_ff}%
In this chapter we take the first step towards the computation of the form-factors in the thermodynamic limit.
Here we will compute the thermodynamic limit of the longitudinal form-factors of the two-spinon ($n_h=2$) triplet $(s=1)$ excitation.
\index{ff@\textbf{Form-factors}!FF@$\FF^z$: longitudinal form-factor|textbf}%
\begin{align}
	|\FF^{z}|^2=
	\frac{|\braket{\psi_{g}|\sigma^{3}_{m}|\psi_{1}^{1}(\hle_1,\hle_2)}|^2%
	}{%
	\braket{\psi_{g}|\psi_{g}}
	\braket{\psi_{1}^{1}(\hle_1,\hle_2)|\psi_{1}^{1}(\hle_1,\hle_2)}
	}
	.
	\label{long_ff_scal_prod_2sp}
\end{align}
Let $\bm\la$ denote the set of Bethe roots of the ground state $\ket{\psi_g}$ while $\bm\mu$ denote the set of Bethe roots of the excited state $\ket{\psi_{1}^{0}}$.
Their cardinalities are $n_{\bm\la}=N_0$ and $n_{\bm\mu}=N_1$.
Since the number of spinons in the current scenario is two $n_h=2$, we know from our previous discussion in \cref{chap:spectre}, that all of the Bethe roots of such an excitation are real $\bm\mu\in\Rset$.
As we saw in \cref{ff_transmap_all} from \cref{chap:qism}, the $\mathfrak{su}_2$ symmetry of the XXX model permits us to recast the form-factors \eqref{long_ff_scal_prod_2sp} in the transverse mode.
There we had also argued that the representation \eqref{ff_transmap} that is reproduced below is more convenient.
\begin{align}
	|\FF^{z}|^2
	&=
	-
	\frac{\Braket{\psi_{g}|\sigma^{-}_{m}|\psi^{0}_1(\hle_1,\hle_2)}}{\Braket{\psi_{g}|\psi_{g}}}
	\frac{\Braket{\psi^{2}_1(\hle_1,\hle_2)|\sigma^{-}_{m}|\psi_{g}}}{\Braket{\psi^{0}_1(\hle_1,\hle_2)|\psi^{0}_1}}
	.
	\label{ff_transmap_2sp}
\end{align}
When it was introduced in \cref{chap:qism}, we had simply remarked that in our method, we are obliged to take the left-action (or rightwards action) for the local spin-operators [see \cref{ff_transmap_all} and the discussion that ensues].
The reasoning behind this remark will become clear in the following paragraphs, as we introduce the procedure that we use to \emph{extract} the Gaudin matrix.
Let us first introduce the following notation:
\begin{notn}
\label{notn_bethe_root_i2}
Let the accented sets $\bm{\check{\mu}}$ and $\bm{\check{\la}}$ denote the following unions, which add the parameter $\tfrac{i}{2}$ as follows:
\index{ff@\textbf{Form-factors}!set check@$\bm{\check\alpha}$: set of spectral parameter with an extra parameter|textbf}%
\index{exc@\textbf{Excitations}!set check@$\bm{\check\alpha}$: set of spectral parameter with an extra parameter|see{Form-factors}}%
\begin{alignat}{3}
	\bm{\check{\mu}}&=\bm\mu\bm\cup\set{\tfrac{i}{2}}
	\quad
	&\text{and,}
	\quad
	\bm{\check{\la}}&=\bm\la\bm\cup\set{\tfrac{i}{2}}
	.
	\label{bethe_root_i2}
\end{alignat}
The newly added parameter will be always indexed at the end position $\check\mu_{N_0}=\frac i2$ and $\check\la_{N_0+1}=\frac i2$.
Let us recall from \cref{num_roots_xxx} that $N_0=\frac{M}{2}$ and $N_s=N_0-s$.
\end{notn}
Since the rightward action of the local spin operator $\sigma^-_m$ is determined by the identity \eqref{qism_rel_-} and it effectively amounts to adding the parameter $\frac{i}{2}$, the above \cref{notn_bethe_root_i2} helps us encapsulate its action in a very compact manner.
Using it, the expression for the longitudinal form-factor from \cref{ff_transmap} can be rewritten as
\begin{align}
	|\FF^{z}|^2
	&=
	-
	\frac{%
	\braket{\psi(\bm\la)|\psi(\bm{\check\mu})}%
	}{%
	\braket{\psi(\bm\la)|\psi(\bm\la)}%
	}%
	\frac{%
	\braket{\psi(\bm\mu)|\psi(\bm{\check\la})}%
	}{%
	\braket{\psi(\bm\mu)|\psi(\bm\mu)}%
	}%
	.
	\label{ff_transmap_2sp_mag}
\end{align}
Let us now substitute the determinant representations for the scalar products in the preceding expression \eqref{ff_transmap_2sp_mag}.
We will be using the determinant representation of Slavnov \cite{Sla89} and its variation by Foda-Wheeler \cite{FodW12a} for the scalar product in the numerators and the determinant representation of Gaudin \cite{Gau83} for the norms in the denominators.
It permits us to write the form-factors as
\begin{align}
	|\FF^{z}|^2&=
	-2
	\bmprod\frac{%
	\baxq_{g}(\bm\mu-i)
	}{%
	\baxq_{e}(\bm\mu-i)
	}
	\bmprod\frac{%
	\baxq_{e}(\bm\la-i)
	}{%
	\baxq_{g}(\bm\la-i)
	}
	\frac{%
	\det\nolimits_{N_0}\Mcal\left[\bm{\check\mu}\big\Vert\bm\la\right]
	}{%
	\det\nolimits_{N_0}\Ncal\left[\bm\la\big\Vert\bm\la\right]
	}
	\frac{%
	\det\nolimits_{N_0+1}\Mcal^{(2)}\left[\bm{\check\la}\big\Vert\bm\mu\right]
	}{%
	\det\nolimits_{N_0-1}\Ncal\left[\bm\mu\big\Vert\bm\mu\right]
	}
	.
	\label{det_rep_fini_ff}
\end{align}
Consequently, we recall that $\Mcal$ denotes the Slavnov matrix, $\Ncal$ denotes the Gaudin matrix, while the notation $\Mcal^{(2)}\big[\bm{\check\la}\big\Vert\bm\mu\big]$ is used to denote the version of the Slavnov matrix which is composed of two rectangular blocks of columns:
\index{ff@\textbf{Form-factors}!mat Sla FW@$\Mcal^{(\ell)}[\cdot\Vert\cdot]$: Foda-Wheeler version of Slavnov matrix}%
\begin{align}
	\Mcal^{(2)}\big[\bm{\check\la}\big\Vert\bm\mu\big]
	=
	\left(
	\begin{array}{c|c}
	\Mcal\big[\bm{\check\la}\big\Vert\bm\mu\big]
	&
	\Ucal\left[\bm{\check\la}\right]
	\end{array}
	\right)
	.
	\label{sla_mat_fw_version}
\end{align}
Let us recall from \cref{sla_mat_gen,van_block_fw_version,gau_mat_gen}, that all the components of the matrices $\Mcal$, $\Ncal$ and $\Ucal$ are given by the following set of expressions:
\index{ff@\textbf{Form-factors}!mat Sla@$\Mcal[\cdot\Vert\cdot]$: Slavnov matrix}%
\index{ff@\textbf{Form-factors}!mat Gau@$\Ncal[\cdot\Vert\cdot]$: Gaudin matrix}%
\begin{subequations}
\begin{align}
	\Mcal_{j,k}[\bm{\check\mu}\Vert\bm\la]
	&=
	\aux_g(\check\mu_j)
	t(\check\mu_j-\la_k)
	-
	t(\la_k-\check\mu_j)
	,
	\label{sla_comp_2sp}
	\\
	\Mcal_{j,k}[\bm{\check\la}\big\Vert\bm\mu]
	&=
	\aux_e(\check\la_j)
	t(\check\la_j-\mu_k)
	-
	t(\mu_k-\check\la_j)
	,
	\label{rect_sla_comp_2sp}
	\\
	\Ucal_{j,a}[\bm{\check\la}]
	&=
	\aux_e(\check\la_j)
	\check\la_j^{a}
	-
	(\check\la_{j}+i)^a
	;
	\label{fw_comp_2sp}
	\\[.5em]
	\Ncal_{j,k}[\bm\la\Vert\bm\la]
	&=
	\aux'_g(\la_j)\delta_{j,k}
	-
	2\pi i
	K(\la_j-\la_k)
	,
	\label{gau_gs_comp_2sp}
	\\
	\Ncal_{j,k}[\bm\mu\Vert\bm\mu]
	&=
	\aux'_e(\mu_j)\delta_{j,k}
	-
	2\pi i
	K(\mu_j-\mu_k)
	.
	\label{gau_exc_comp_2sp}
\end{align}
\label{almat_comp_2sp}
\end{subequations}
We will now present the core of our method as employed for the computation of two-spinon form-factor \eqref{ff_transmap_2sp_mag} in the thermodynamic limit.
It was first computed using the framework of the $q$-Vertex Operator Algebra ($q$-VOA) by \textcite{BouCK96}.
The two-spinon form-factors of the XXX model were also analysed by \textcite{CauH06} who corroborated the previous result by performing a numerical check of the sum rules.
Our approach gives us a mechanism to obtain the thermodynamic limits of the form-factors from the Algebraic Bethe Ansatz (ABA). In the two-spinon case, we obtained the result \cite{KitK19} that is reproduced in the following expression, written in terms of the Barnes-$G$ functions.
\begin{align}
	\left|\FF^{z}\right|^2
	&=
	\frac{2}{M^2 G^4\left(\frac{1}{2}\right)}
	\prod_{\sigma=\pm}
	\frac{%
	G(\frac{\hle_{2}-\hle_{1}}{2i\sigma})
	G(1+\frac{\hle_{2}-\hle_{1}}{2i\sigma})
	}{%
	G(\frac{1}{2}+\frac{\hle_{2}-\hle_{1}}{2i\sigma})
	G(\frac{3}{2}+\frac{\hle_{2}-\hle_{1}}{2i\sigma})
	}
	.
\end{align}
It was also compared with previous results for the two-spinon form-factors from $q$-VOA approach and it was found to be consistent with their results.
The method applied to obtain it is discussed in this chapter.
This will be divided into two sections.
In the first section, we describe the process of the Gaudin extraction to compute the ratio of determinants in \cref{ff_transmap_2sp_mag}.
It gives us a new determinant representation which contains modified Cauchy matrices.
The second section is devoted to the computation of the Cauchy determinants in the thermodynamic limit.
\section{Extraction of the Gaudin matrix}
\label{sec:gau_ex_2sp}%
Our goal is to ultimately compute the ratio of determinants in \cref{det_rep_fini_ff} in the thermodynamic limit.
In our approach, we will heavily make use of a process of \emph{matrix extraction} that allows us to replace the ratio of determinant by a single determinant. The most straightforward ways to do this here is by taking an action of the inverse Gaudin matrix: $\Fcal=\Ncal^{-1}\Mcal$, so that the determinant of the new matrix $\Fcal$ can represent the ratio of determinants in \cref{det_rep_fini_ff}.
It leads to the definitions of the matrices:
\index{ff@\textbf{Form-factors}!mat FF gs@$\Fmat[g]$: matrix obtained from the extraction of the ground state Gaudin matrix|textbf}%
\index{ff@\textbf{Form-factors}!mat FF es@$\Fmat[e]$: matrix obtained from the extraction of an excited state Gaudin matrix|textbf}%
\begin{subequations}
\begin{align}
	\Fmat[g]&=
	\Ncal^{-1}[\bm\la\Vert\bm\la]
	\cdot
	\left(\Mcal\left[\bm{\check\mu}\big\Vert\bm\la\right]\right)^T
	\label{gau_ex_I_mat_def}
	\shortintertext{and}
	\Fmat[e]&=
	\diag\left[\Ncal^{-1}[\bm\mu\Vert\bm\mu]~\Big\vert~\Id_{2}\right]
	\left(\Mcal^{(2)}\big[\bm{\check\la}\big\Vert\bm\mu\big]\right)^T
	\label{gau_ex_mat_II_def}
	.
\end{align}
	\label{gau_ex_mat_def_both}
\end{subequations}
Note that their actions are well defined since the Gaudin matrices are invertible for non-trivial on-shell Bethe vectors.
Inverse Gaudin matrix for the excited state in \cref{gau_ex_mat_II_def} is diagonally embdedded as a block, this embedding should be read as
\begin{align}
	\diag\Big[
	\Ncal^{-1}[\bm\mu\Vert\bm\mu]
	~\Big|~
	\Id_2 	
	\Big]
	=
	\begin{pmatrix}
		\Ncal^{-1}[\bm\mu\Vert\bm\mu]
		&	0
		\\
		0	& \Id_2
	\end{pmatrix}
	.
	\label{gau_mat_diag_immersion}
\end{align}
The Slavnov matrices are transposed in both expressions \eqref{gau_ex_mat_def_both} so that resulting sums always involve Bethe roots of on-shell vectors.
\begin{rem}
However, it is important to note that the word \emph{matrix extraction} (let us say $\Acal$-extraction) can be used in more broader sense than simply an action of inverse matrix $\Acal^{-1}$.
Since it is the determinant of this procedure that is a primary object of our interest, we are at liberty to replace the inverse matrix $\Acal^{-1}$ with any other matrix that we may deem fit, as long as it has same determinant as $\Acal^{-1}$.
This symmetry can have some technical advantages which we will exploit later in this thesis.
A rather trivial example of this symmetry at work can be seen in \cref{gau_ex_mat_II_def}, where we added a diagonal identity block to match the orders of two matrices.
\end{rem}
We can see that $\Fmat[e]$ is also divided in two blocks $\Fmat[e]^{\mathrm{N}}$ and $\Fmat[e]^{\mathrm{S}}$\footnotemark. The north block $\Fmat[e]^{\mathrm{N}}$ is made up of the top $N_1$ rows while the south block $\Fmat[e]^{\mathrm{S}}$ is made up of the last two rows.
From \cref{gau_ex_mat_def_both}, we can see that $\Fmat[g]$ and $\Fmat[e]^{\mathrm{N}}$ solve the systems:
\footnotetext{in reference to the cardinal directions north ($\mathrm{N}$) and south ($\mathrm{S}$).}
\index{ff@\textbf{Form-factors}!mat FF es north@\hspace{1em}$\Fmat[e]^{\textrm{N}}$: north block of \rule{3em}{1pt}|textbf}%
\index{ff@\textbf{Form-factors}!mat FF es south@\hspace{1em}$\Fmat[e]^{\textrm{S}}$: south block of \rule{3em}{1pt}|textbf}%
\begin{subequations}
\label{gau_ex_syslin_mat_form_both}
\begin{align}
	\Ncal[\bm\la\Vert\bm\la]\cdot\Fmat[g]
	&=
	\left(\Mcal\left[\bm{\check\mu}\big\Vert\bm\la\right]\right)^T
	\label{gau_ex_syslin_mat_form_I}
	\shortintertext{and}
	\Ncal[\bm\mu\Vert\bm\mu]\cdot\Fmat[e]^{\mathrm{N}}
	&=
	\left(\Mcal\left[\bm{\check\la}\big\Vert\bm\mu\right]\right)^T
	.
	\label{gau_ex_syslin_mat_form_II}
\end{align}
\end{subequations}
In \cref{gau_ex_syslin_mat_form_II} it is only the north block $\Fmat[e]^{\mathrm{N}}$ that enters our computations.
Meanwhile, we can easily check that the south block $\Fmat[e]^{\mathrm{S}}$ retains its original form since it is only acted upon by the identity matrix during the Gaudin extraction \eqref{gau_ex_mat_II_def}.
\\
The first step in our method is to show that linear equations \eqref{gau_ex_syslin_mat_form_both} give rise to the integral equations in the thermodynamic limit.
To achieve this conclusion, we use the condensation property \eqref{condn_prop_gs_gen} and whenever necessary, its generalised version \eqref{gen_condn_prop} that was discussed in \cref{sub:gen_condn_non-comp_F-zn}.
The two different scenarios present in \cref{gau_ex_syslin_mat_form_II,gau_ex_syslin_mat_form_I} demand a significantly different treatments, which is why it will be dealt here separately.
\subsection{Extraction of the first type}
\label{sub:gau_ex_I}%
The \cref{gau_ex_syslin_mat_form_I} satisfied by $\Fmat[g]$ leads us to the system of linear equations
\begin{align}
	\aux_{g}'(\la_{j})\Fmat[g]_{j,k}
	-
	2\pi i
	\sum_{l=1}^{N_0}
	K(\la_{j}-\la_{l})
	\Fmat[g]_{l,k}
	&=
	\aux_{g}(\mu_{k})\,t(\check{\mu}_{k}-\la_{j})
	-t(\la_{j}-\check{\mu}_{k})
	.
	\label{gau_ex_syslin}
\end{align}
Note that this also includes the last column for $\check{\mu}_{N_0}=\frac{i}{2}$ as per \cref{bethe_root_i2}.
In the thermodynamic limit the system of equations \eqref{gau_ex_syslin} would lead us to an integral equation for the function $\Gf[g]$, that is defined in the following:
\begin{defn}
\label{Gf_init_defn}
Let us define a meromorphic function $\Gf[g](\la,\mu)$ with the initial conditions as follows:
\index{aux@\textbf{Auxiliary functions}!gauex gs fn@$\Gf[g]$: fn. involved in the extraction of the ground state Gaudin mat.|textbf}%
\index{ff@\textbf{Form-factors}!gauex gs fn@$\Gf[g]$: fn. involved in the extraction of the ground state Gaudin mat.|see{Auxiliary functions}}%
\begin{align}
	\Gf[g](\la_{j},\check{\mu}_{k})&=\aux'_{g}(\la_{j})\Fmat[g]_{l,k}
	.
	\label{Gf_init}
\end{align}
\end{defn}
In terms of this function, the set of linear equations \eqref{gau_ex_syslin} can be rewritten as
\begin{align}
	\Gf[g](\la_{j},\check{\mu}_{k})
	-
	2\pi i
	\bmsum K(\la_{j}-\bm\la)
	\frac{%
	\Gf[g](\bm\la,\check{\mu}_{k})
	}{%
	\aux^\prime_{g}(\bm\la)
	}
	&=
	\aux_{g}(\la_{j})t(\check{\mu}_{k}-\la_{j})-t(\la_{j}-\check{\mu}_{k})
	.
	\label{gau_ex_syslin_Gf}
\end{align}
From the above \cref{gau_ex_syslin_Gf} we can now see that the function $\Gf[g](\la,\check{\mu}_{k})$ has a simple pole at $\la=\check{\mu}_{k}$ with residue
\begin{align}
	\res_{\la=\check{\mu}_{k}}\Gf[g](\la,\check{\mu}_{k})=-(1+\aux_{g}(\check{\mu}_{k}))
	.
	\label{gau_ex_Gf_res}
\end{align}
Barring an isolated case for the $\Gf[g](\la,\check{\mu}_{N_0})$ where $\check\mu_{N_0}=\frac i2$, we see that in all the remaining cases for the evaluations of $\Gf[g](\la,\mu_k)$, the poles are all taken from the set $\bm\mu$ and hence they always lie on the real line.
\\
Let us now replace the sum over residues by an integral over a rectangular contour of an appropriate width $2\alpha$.
From the poles of the function $K$ \eqref{gau_mat_K_def}, we see that the half-width cannot exceed the value: $\alpha<\frac{1}{2}$.
Meanwhile, a simple pole of $\Gf[g](\la,\mu_k)$ at $\la=\mu_k$ also falls inside the contour since $\bm\mu\subset\Rset$ and hence it needs to be isolated using \cref{gau_ex_Gf_res}.
The integrals on the edges of the contour can be ignored as long as the function $\Gf[g]$ is bounded at infinity.
Using this information, we can write
\begin{align}
	\bmsum K(\la_{j}-\bm\la)
	\frac{%
	\Gf[g](\bm\la,\mu_{k})
	}{%
	\aux'_{g}(\bm\la)
	}
	&=	
	\frac{1}{2\pi i}\left(\int_{\Rset-i\alpha}-\int_{\Rset+i\alpha}\right)	
	K(\la_{j}-\tau)
	\frac{%
	\Gf[g](\tau,\mu_{k})
	}{%
	1+\aux_{g}(\tau)
	}
	d\tau
	+
	K(\la_{j}-\mu_{k})
	.
	\label{gau_ex_oint}
\end{align}
It can be immediately rewritten as the following expression shows. It makes use of the fact that the auxiliary function $\aux$ is nothing but the exponential of the counting function $\cfn$.
\begin{multline}
	\bmsum K(\la_{j}-\bm\la)
	\frac{%
	\Gf[g](\bm\la,\mu_{k})
	}{%
	\aux'_{g}(\bm\la)
	}
	=
	-\frac{1}{2\pi i}
	\int_{\Rset+i\epsilon}
	K(\la_j-\tau)
	\Gf[g](\tau,\mu_k)
	d\tau
	+
	K(\la_j-\mu_k)
	\\
	+
	2
	\Re
	\int_{\Rset+i\alpha}
	\rden_h(\la_j-\tau)
	\Gf[g](\tau,\mu_k)
	\frac{%
	e^{2\pi i \cfn_g(\tau)}
	}{%
	1+
	e^{2\pi i \cfn_g(\tau)}
	}
	d\tau
	.
	\label{gau_ex_oint_finite_corr}
\end{multline}
Let us substitute the above expression \eqref{gau_ex_oint_finite_corr} back into \cref{gau_ex_syslin}.
This permits us to write down the following integral equation for the function $\Gf[g]$:
\begin{multline}
	\Gf[g](\la,\mu_k)
	+
	\int_{\Rset+i\epsilon}
	K(\la_j-\tau)
	\Gf[g](\tau,\mu_k)
	d\tau
	=
	(1+\aux_g(\mu_k))t(\mu_k-\la)
	+
	\\
	+
	2
	\Re
	\int_{\Rset+i\alpha}
	K(\la-\tau)
	\Gf[g](\tau,\mu_k)
	\frac{%
	e^{2\pi i \cfn_g(\tau)}
	}{%
	1+
	e^{2\pi i \cfn_g(\tau)}
	}
	d\tau
	.
	\label{gau_ex_inteq_finite_corr}
\end{multline}
From the above \cref{gau_ex_inteq_finite_corr}, we see that the function $\Gf[g]$ can be expressed in the following form:
\begin{multline}
	\Gf[g](\la,\mu_k)
	=
	2\pi i
	(1+\aux_g(\mu_k))
	\rden_2\left(\la,\mu_k+\frac{i}{2}-i0\right)
	\\
	+
	2\Re\int_{\Rset+i\alpha}
	\rden_h(\la-\tau)
	\Gf[g](\tau,\mu_k)
	\frac{%
	e^{2\pi i \cfn_g(\tau)}
	}{%
	1+
	e^{2\pi i \cfn_g(\tau)}
	}
	d\tau
	.
	\label{gau_ex_sol_finite_corr}
\end{multline}
In this decomposition, the function $\rden_h$ in the integral appears as the resolvent of the Lieb equation as it satisfies the integral equation \eqref{inteq_rden_hle}. We note that it can also be called the function $\rden_1$, which we had introduced in \cref{sub:spinon} earlier.
It satisfies a generalised version of the Lieb integral \cref{lieb_inteq_shft_scld_def}, which is duly studied in \cref{chap:den_int_aux}.
\\
Let us remark that this procedure uses the generalised condensation property in order to write \cref{gau_ex_inteq_finite_corr,gau_ex_sol_finite_corr}, that we have introduced in \cref{gen_condn_prop} from \cref{sub:gen_condn_non-comp_F-zn}.
The problems in estimating the finite-size correction terms in the case of non-compact Fermi-zone that were encountered there are also present here.
We can estimate that in the bulk of the distribution, the value of the exponential of the counting function is exponentially small.
However, this is no longer true in the tails of the Fermi-zone.
But at-least in the case where the pole at $\mu_k$ of the function $\Gf[g](\la,\mu_k)$ is in the bulk, we can assume that the correction term is subleading
\begin{align}
	\Re\int_{\Rset+i\alpha}
	\rden_h(\la-\tau)
	\Gf[g](\tau,\mu_k)
	\frac{%
	e^{2\pi i \cfn_g(\tau)}
	}{%
	1+
	e^{2\pi i \cfn_g(\tau)}
	}
	d\tau
	=
	o\left(\frac{1}{M}\right)
	\label{gau_ex_finite_corr_estmn_bulk}
\end{align}
which allows us to write
\begin{align}
	\Gf[g](\la_j,\mu_k)
	=
	2\pi i(1+\aux_g(\mu_k))
	\rden_2\left(\la_j,\mu_k+\frac{i}{2}-i0\right)
	+o\left(\frac{1}{M}\right)
	.
	\label{gau_ex_sol}
\end{align}
For the last column corresponding to $\check{\mu}_{N_0}=\frac{i}{2}$, we can simply use the \emph{regular} condensation property, as there are no poles of the function $\Gf[g]$ left on the real line in this case.
This gives us an integral equation.
\begin{align}
	\Gf[g](\la,\tfrac{i}{2})+\int_{\Rset}K(\la-\tau)\Gf[g](\tau,\tfrac{i}{2})
    d\tau
	&=
	-t(\la-\tfrac{i}{2})
	.
	\label{gau_ex_inteq_i2}
\end{align}
Since $\aux_{g}(\tfrac{i}{2})=0$ as $\check\mu_{N_0}=\frac{i}{2}$, it is the zero of the function $r$ in \cref{bae_aux_gen}.
Therefore, it leaves us with the term $t(\la-\frac{i}{2})=i p'_0(\la)$, which in-turn leads to the solution as shown in the following expression.
\begin{align}
	\Gf[g](\la,\tfrac{i}{2})&=
	-2\pi i\,
	\rden_{2}(\la)
	+
	o\left(\frac{1}{M}\right)
	.
	\label{gau_ex_sol_i2}
\end{align}
The solution $\rden_{2}(\la,\nu)$ of the Lieb type integral equation was obtained in \cref{lieb_den_shft_inside}.
We see that in both cases, we are in the central strip of its analyticity $|\nu|<\frac{1}{2}$.
Combining the results from \cref{gau_ex_sol,gau_ex_sol_i2} we can now write
\begin{align}
	\Gf[g](\la_j,\check\mu_k)&=
	\frac{\pi(1+\aux_{g}(\check\mu_k))}{\sinh\pi(\check\mu_k-\la_j)}
	+
	o\left(\frac{1}{M}\right)
	.
	\label{gau_ex_sol_all}
\end{align}
Therefore, all the elements of the matrix $\Fcal_{g}$ can be found from \cref{Gf_init}. It gives,
\begin{align}
	\Fmat[g]_{j,k}&=
	\frac{1+\aux_{g}(\check\mu_{k})}{\aux_{g}^\prime(\la_{j})}
	\left\lbrace
	\frac{\pi}{\sinh\pi(\check\mu_{k}-\la_{j})}
	+o\left(\frac{1}{M}\right)
	\right\rbrace
	.
	\label{gau_ex_mat_I_comp_2sp}
\end{align}
\subsection{Extraction of the second type}
\label{sub:gau_ex_II}%
From \cref{gau_ex_syslin_mat_form_II} we get a system of linear equations for the upper block $\Fmat[e]^{\mathrm{N}}$ of the matrix \eqref{gau_ex_mat_II_def} which is shown in the following:
\begin{align}
	\aux_{e}'(\mu_{j})\Fmat[e]^{\mathrm{N}}_{j,k}
	-	2\pi i
	\sum_{l=1}^{N_1}
	K(\mu_{j}-\mu_{l})
	\Fmat[e]^{\mathrm{N}}_{l,k}
	&=
	\aux_{e}(\mu_{k})\,t(\check{\la}_{k}-\mu_{j})
	-
	t({\mu}_{j}-\check{\la}_{k})
	.
	\label{gau_ex_syslin_II}
\end{align}
Note that it also includes the case $\check{\la}_{N_0+1}=\tfrac{i}{2}$, which represents the last column.
Similar to \cref{Gf_init_defn}, we shall again write the following definition:
\begin{defn}
Now we define a meromorphic function $\Gf[e](\mu,\la)$ with the initial conditions
\index{aux@\textbf{Auxiliary functions}!gauex es fn@$\Gf[e]$: fn. involved in the extraction of an excited state Gaudin mat.|textbf}%
\index{ff@\textbf{Form-factors}!gauex es fn@$\Gf[e]$: fn. involved in the extraction of an excited state Gaudin mat.|see{Auxiliary functions}}%
\begin{align}
	\Gf[e](\mu_{j},\check{\la}_{k})&=\aux'_{e}(\mu_{j})\Fmat[e]^{\mathrm{N}}_{j,k}
	.
	\label{Gf_init_II}
\end{align}
\end{defn}
In terms of the function $\Gf[e]$, the system of linear equations from \cref{gau_ex_mat_II_def} can be rewritten as
\begin{multline}
	\Gf[e](\mu_{j},\check{\la}_{k})
	-
	2\pi i
	\bmsum K(\mu_{j}-\bm\rh)
	\frac{%
	\Gf[e](\bm\rh,\check{\la}_{k})
	}{%
	\aux^\prime_{e}(\bm\rh)
	}
	=
	\aux_{e}(\check{\la}_{k})t(\check{\la}_{k}-\mu_{j})-t(\mu_{j}-\check{\la}_{k})
	\\
	-2\pi i
	\bmsum K(\mu_{j}-\bm\hle)
	\frac{%
	\Gf[e](\bm\hle,\check{\la}_{k})
	}{%
	\aux^\prime_{e}(\bm\hle)
	}	
	.
	\label{gau_ex_syslin_Gf_with_hle}
\end{multline}
Because it is convenient to do so, we have added the extra terms for the holes on both sides in order to write down \cref{gau_ex_syslin_Gf_with_hle}.
At this point, let us recall here that the symbol $\bm\rh$ denotes the set of all the real roots of the logarithmic Bethe \cref{log_bae}, including the holes $\bm\rh=\bm\rl\bm\cup\bm\hle$.
Since all the Bethe roots are real in the current scenario of the two-spinon form-factors, we can also write it as $\bm\rh=\bm\mu\bm\cup\bm\hle$.
This notation was first used in \cref{sec:spectre_XXX} [see \cref{def_rh_bethe_root} therein].
\\
Similar to \cref{gau_ex_Gf_res} for $\Gf[g]$, we again find that the function $\Gf[e](\nu,\check{\la}_{k})$ has a simple pole at $\nu=\check{\la}_{k}$, with its residue given by the expression:
\begin{align}
	\res_{\nu=\check{\la}_{k}}\Gf[e](\nu,\check{\la}_{k})=-(1+\aux_{e}(\check{\la}_{k}))
	.
\end{align}
Let us now follow the procedure similar to \crefrange{gau_ex_oint}{gau_ex_inteq_finite_corr}, that uses the generalised condensation property, to produce an integral equation for the function $\Gf[e]$ as shown below:
\begin{multline}
	\Gf[e](\nu,\la_k)
	+\int_{\Rset+i0}K(\nu-\tau)\Gf[e](\tau,\la_k)d\tau=
	(1+\aux_{e}(\la))\,t(\la_k-\nu)
	\\
	-2\pi i
	\bmsum K(\nu-\bm\hle)
	\frac{%
	\Gf[e](\bm\hle,\la_k)
	}{%
	\aux^\prime_{e}(\bm\hle)
	}	
	\\
	+
	2\Re\int_{\Rset+i\alpha}
	K(\nu-\tau)\Gf[e](\tau,\la_k)
	\frac{%
	e^{2\pi i \cfn_e(\tau)}
	}{%
	1+e^{2\pi i \cfn_e(\tau)}
	}
	d\tau
	.
	\label{gau_ex_inteq_II}
\end{multline}
Similar to \cref{gau_ex_sol_finite_corr}, we can see that the function $\Gf[e]$ can be decomposed as follows:
\begin{multline}
	\Gf[e](\nu,\la)=
	2\pi i(1+\aux_{e}(\la))\,\rden_{2}(\nu,\la+\tfrac{i}{2}-i0)
	-2\pi i
	\bmsum
	\rden_{1}(\nu,\bm\hle-i0)
	\frac{%
	\Gf[e](\bm\hle,\la)
	}{%
	\aux_{e}^\prime(\bm\hle)
	}
	\\
	+
	2\Re\int_{\Rset+i\alpha}
	\rden_h(\nu-\tau)\Gf[e](\tau,\la_k)
	\frac{%
	e^{2\pi i \cfn_e(\tau)}
	}{%
	1+e^{2\pi i \cfn_e(\tau)}
	}
	d\tau
	.
	\label{gau_ex_II_sol_decomp}
\end{multline}
The shifted density functions $\rden_{2}(\nu,\la+\tfrac{i}{2}-i0)$ and $\rden_{1}(\nu,\hle-i\epsilon)$ appearing in the above decomposition \eqref{gau_ex_II_sol_decomp} satisfy the integral equation which are described by \cref{lieb_inteq_shft_scld_def}.
In both cases, the shifts are small enough to lie within the central strip of the analyticity of these functions hence the solutions are given by $\rden_g$ and $\rden_h$ respectively, which are obtained in \cref{lieb_den_shft_inside,lieb_hle_den_shftd_inside}.
We substitute them here to write
\begin{align}
	\Gf[e](\nu,\la_k)&=
	\frac{%
	\pi(1+\aux_{e}(\la_k))
	}{%
	\sinh\pi(\la_k-\nu)
	}
	-2\pi i
	\bmsum
	\frac{%
	\rden_{h}(\nu-\bm\hle)
	}{%
	\aux_{e}^\prime(\bm\hle)
	}
	\Gf[e](\bm\hle,\la_k)
	+
	o\left(\frac{1}{M}\right)
	.
	\label{gau_ex_sol_II}
\end{align}
Similar to \cref{gau_ex_finite_corr_estmn_bulk}, we have assumed that the finite size correction are of sub-leading order $o\left(\frac{1}{M}\right)$, at least for the values of the parameters $\la_k$ that lie in the bulk of the Fermi-distribution.
For the last column corresponding to $\check{\la}_{\frac{M}{2}+1}=\frac{i}{2}$, the function $\Gf[e](\nu,\tfrac{i}{2})$ is regular on the real line and it suffices to use the regular condensation property to write the integral equation:
\begin{align}
	\Gf[e](\nu,\tfrac{i}{2})
	+\int_{\Rset}K(\nu-\tau)\Gf[e](\tau,\tfrac{i}{2})d\tau=
	-ip_{0}^\prime(\nu)
	-2\pi i
	\bmsum K(\nu-\bm\hle)
	\frac{%
	\Gf[e](\bm\hle,\frac{i}{2})
	}{%
	\aux^\prime_{e}(\bm\hle)
	}	
	.
\end{align}
Therefore, we see that in this case the function $\Gf[e]$ admits the decomposition:
\begin{align}
	\Gf[e](\nu,\tfrac{i}{2})&=
	\frac{%
	-i\pi
	}{%
	\cosh\pi\nu
	}
	-2\pi i
	\bmsum
	\frac{%
	\rden_{h}(\nu-\bm\hle)
	}{%
	\aux_{e}^\prime(\bm\hle)
	}
	\Gf[e](\bm\hle,\tfrac{i}{2})
	+
	o\left(\frac{1}{M}\right)
	.
	\label{gau_ex_sol_II_i2}
\end{align}
The solutions obtained in \cref{gau_ex_sol_II,gau_ex_sol_II_i2} can be combined together, to write down a single expression:
\begin{align}
	\Gf[e](\rh_{j},\check\la_k)&=
	\frac{%
	\pi(1+\aux_{e}(\check\la_k))
	}{%
	\sinh\pi(\check\la_k-\rh_{j})
	}
	-2\pi i
	\bmsum
	\frac{%
	\rden_{h}(\rh_{j}-\bm\hle)
	}{%
	\aux_{e}^\prime(\bm\hle)
	}
	\Gf[e](\bm\hle,\check\la_k)
	+
	o\left(\frac{1}{M}\right)
	.
\end{align}
Note that the terms $\Gf[e](\bm\hle,\check\la_k)$ for the holes in the above expression are yet to be determined.
To compute them, let us first introduce the following notation for the matrices.
\begin{defn}
	\index{ff@\textbf{Form-factors}!mat diag der cfn@$\Acal_{\ast}[\cdot]$: \emph{diagonal matrix} formed by derivatives of the counting function|textbf}
\label{defn:diag_cfn_der_mat}
Given a set of parameters $\bm\alpha$ and a set of Bethe roots $\bm\la$ of an on-shell Bethe vector, we define a diagonal matrix $\Acal_{\bm\la}[\bm\alpha]$, which is given by its components:
\begin{align}
	\mix{\Acal}[\bm\la]_{j,k}[\bm\alpha]
	=
	\aux'(\alpha_j|\bm\la)\delta_{j,k}
	\label{diag_cfn_der_mat}
\end{align}
The mismatched brackets in this expression are deliberate since we reserve the square brackets for the matrices.
Alternatively, we can also write
\begin{align}
	\Acal_{\bm\la}[\bm\alpha]
	=
	\diag_{\bm\alpha}
	\big[
	\aux'(\bm\alpha|\bm\la)
	\big]
\end{align}
The subscript $\bm\la$ can also be replaced by $\Acal_g[\bm\alpha]$ to denote the ground state and $\Acal_e[\bm\alpha]$ for the excited state.
\end{defn}
\begin{defn}
\label{defn:den_hle_mat_rect}
Given two sets of parameters $\bm\alpha$ and $\bm\beta$, we define the matrix $\Rcal[\bm\alpha\Vert\bm\beta]$ given by its components:
\index{ff@\textbf{Form-factors}!mat den hle@$\Rcal[\cdot\Vert\cdot]$: a matrix formed by the density terms for the holes $\rden_h$|textbf}%
\begin{align}
	\Rcal_{jk}[\bm\alpha\Vert\bm\beta]
	=
	-2\pi i
	\rden_h(\alpha_j-\beta_k)
	.
	\label{den_hle_mat_rect}
\end{align}
We can check that the density function $\rden_h$ is an even function that satisfies the integral \cref{inteq_rden_hle}. for the density term of the holes:
\begin{align}
	\rden_h(\nu)+\int_{\Rset}K(\nu-\tau)\rden_h(\tau)d\tau=K(\nu)
	.
\end{align}
It is identical to the integral equation for the resolvent of the Lieb kernel, the definition for $\Rcal$ can also be paraphrased as the matrix formed by resolvent of the Lieb equation \eqref{lieb_eq_gen}.
\end{defn}
\begin{rem}
Since the function $\rden_h$ is even, we can see that we obtain a transpose by reversing the order of two sets
\begin{align}
\Rcal[\bm\beta\Vert\bm\alpha]=(\Rcal[\bm\alpha\Vert\bm\beta])^T.
\label{rden_mat_transpose}
\end{align}
\end{rem}

\minisec{Decoupling of the hole terms}
We can see that the terms $\Gf[e](\hle_a,\check\la_k)$ can be obtained by solving the system which is coupled due to the presence of the matrix $\Rcal[\bm\hle|\bm\hle]$.
\begin{align}
	\left(\Id-\Rcal\Acal_e^{-1}[\bm\hle|\bm\hle]\right)
	\bm{\big[}\Gf[e](\bm\hle,\check\la_{k})\bm{\big]}&=
	\bm{\bigg[}
	\frac{%
	\pi(1+\aux_{e}(\check\la_{k}))
	}{%
	\sinh\pi(\check\la_{k}-\bm\hle)
	}
	\bm{\bigg]}
	.
	\label{gau_ex_II_sys_hle}
\end{align}
In the product $\Rcal\Acal_e^{-1}$ we have contracted the sum over dummy variables as shown on \cpageref{ind_free_mats}.
Since the matrix $\Acal_e$ is diagonal, it dresses the matrix of density terms $\Rcal$ to lead us to the expressions
\begin{subequations}
\begin{align}
	\left\lbrace
	\Acal^{-1}\Rcal
	\right\rbrace_{a,k}[\bm\hle|\bm\mu]
	&=
	-2\pi i \frac{%
	\rden_{h}(\mu_k-\hle_a)
	}{%
	\aux'_e(\hle_a)
	}
	,
	\shortintertext{or,}
	\left\lbrace
	\Rcal\Acal^{-1}
	\right\rbrace_{a,k}[\bm\mu|\bm\hle]
	&=
	-2\pi i \frac{%
	\rden_{h}(\mu_k-\hle_a)
	}{%
	\aux'_e(\hle_a)
	}
	.
\end{align}
\end{subequations}
However, since the derivative of the exponential counting function can be expressed as
\begin{align}
	\aux_{e}^\prime(\hle_{a})&=
	-2\pi i\, M\rden_{e}(\hle_{a})
	,
	\label{gau_ex_II_bulk_ass_hle}
\end{align}
we can see that system of \cref{gau_ex_II_sys_hle} is decoupled to the leading order as long as the hole parameters $\bm\hle$ are chosen from the bulk, which allows us to write
\begin{align}
	\Gf[e](\hle_{a},\la_{k})&=
	\frac{%
	\pi(1+\aux_{e}(\la_{k}))
	}{%
	\sinh\pi(\la_{k}-\hle_{a})
	}
	+ O\left(\frac{1}{M}\right)
	.
	\label{gau_ex_II_hle_terms_decouple}
\end{align}
Substituting the above result of \cref{gau_ex_II_hle_terms_decouple} into \cref{gau_ex_sol_II} we can now write
\begin{align}
	\Gf[e](\mu_j,\check\la_k)&=
	\frac{%
	\pi(1+\aux_{e}(\check\la_k))
	}{%
	\sinh\pi(\check\la_k-\mu_j)
	}
	-2\pi i\bmsum
	\frac{%
	\rden_{h}(\mu_{k}-\bm\hle)
	}{%
	\aux_{e}^\prime(\bm\hle)
	}
	\frac{%
	\pi(1+\aux_{e}(\la_{k}))
	}{%
	\sinh\pi(\la_{k}-\bm\hle)
	}
	+o\left(\frac{1}{M}\right)
	.
	\label{gau_ex_sol_II_decoupled}
\end{align}
Therefore, we can see from \cref{Gf_init_II} that components of the matrix $\Fmat[e]^{\mathrm{N}}$ can be expressed as
\begin{align}
	\Fmat[e]^{\mathrm{N}}_{j,k}=
	\pi
	\frac{1+\aux_{e}(\la_{k})}{\aux_{e}^\prime(\mu_{j})}
	\left\lbrace
	\frac{1}{\sinh\pi(\check\la_{k}-\mu_{j})}
	-2\pi i
	\sum_{a=1}^{n_h}
	\frac{\rden_h(\mu_j-\hle_a)}{\aux'_e(\hle_a)}
	\frac{1}{%
	\sinh\pi(\check\la_{k}-\hle_a)
	}
	+ o\left(\frac{1}{M}\right)
	\right\rbrace
	.
	\label{gau_ex_mat_II_comp_2sp}
\end{align}
Let us recall that the south block $\Fmat[e]^{\mathrm{S}}$ made up of the remaining two rows at the bottom retains its original form, which was given by the Foda-Wheeler block $\Ucal$ \eqref{fw_comp_2sp}.
Hence we can write,
\begin{subequations}
\label{gau_ex_II_FW}
\begin{align}
	\Fmat[e]^{\mathrm{S}}_{1,k}&=
	\Ucal_{1}(\check\la_k)
	=
	\aux_{e}(\check{\la}_{k})-1,
	\label{gau_ex_II_FW_deg0}
	\\
	\Fmat[e]^{\mathrm{S}}_{2,k}&=
	\Ucal_{2}(\check\la_k)
	=
	\aux_{e}(\check{\la}_{k})(\check{\la}_{k}+i)-\check{\la}_{k}.
	\label{gau_ex_II_FW_deg1}
\end{align}
\end{subequations}
\subsection{Cauchy determinant representation for the two-spinon form-factor}
\label{sub:cau_det_rep_2sp}
We will now put together all the results obtained in \cref{gau_ex_mat_I_comp_2sp,gau_ex_mat_II_comp_2sp,gau_ex_II_FW} for the components of the matrices $\Fmat[g]$ and $\Fmat[e]$. It will allow us to write down the determinant representation of the two-spinon form-factors.
Beforehand, let us introduce notations that will be used extensively from here onwards.
\begin{defn}
\label{notn:rect_cau_mat_hyper}
\index{cv@\textbf{Cauchy-Vandermonde}!mat Cau@$\Cmat[\cdot\Vert\cdot]$: Cauchy matrix (circular)|textbf}%
\begin{subequations}
Given two sets of the complex parameters $\bm\alpha$ and $\bm\beta$, the hyperbolic Cauchy matrix $\Cmat[\bm\alpha\Vert\bm\beta]$ is defined by the following expression for its components
\begin{align}
	\Cmat_{jk}[\bm\alpha\Vert\bm\beta]&=
	\frac{1}{\sinh\pi(\alpha_j-\beta_k)}
	\label{cau_hyper}
\end{align}
It can be a rectangular matrix which is the case when $n_{\bm\alpha}\neq n_{\bm\beta}$.
The determinant of a square Cauchy matrix is a well known result.
Here it will be denoted in the superalternant $\bmalt$ notation\footnote{see the \cpagerefrange{ind_free_notn}{ind_free_notn_end}}:
\begin{align}
	\det\Cmat[\bm\alpha\Vert\bm\beta]
	=
	\bmalt\sinh\pi(\bm\alpha\Vert\bm\beta)
	\label{cau_hyper_det_super-alt}
\end{align}
Since the function `$\sinh$' is an odd function, we can readily see that the reversing of arguments also leads to a change sign, in addition to the transposition:
\begin{align}
	\Cmat[\bm\beta\Vert\bm\alpha]=
	-\left(\Cmat[\bm\alpha\Vert\bm\beta]\right)^{T}
	.
	\label{cau_hyper_transp}
\end{align}
Later in \cref{sec:cau_van_mat}, we shall define a more general form of this matrix that extends the determinant formula \eqref{cau_hyper_det_super-alt} to a \textit{rectangular} case.
\end{subequations}
\end{defn}
\begin{notn}
\index{aux@\textbf{Auxiliary functions}!Phi fn@\hspace{1em}$\Phi$: hyperbolic equivalent of \rule{3em}{1pt}|textbf}%
Given the set of complex parameters $\bm\alpha,\bm\beta\subset\Cset$, we define the $\Phifn$ as
	\begin{align}
		\Phifn(\la|\bm\alpha,\bm\beta)&=
		\frac{%
		\bmprod \sinh\pi(\la-\bm\alpha)
		}{%
		\bmprod \sinh\pi(\la-\bm\beta)
		}
		.
		\label{phifn_hyper}
	\end{align}
We can see that this is a meromorphic function with poles $\la=\beta_{a}+in$ $\forall n\in\Zset$, $a\leq n_{\bm\beta}$ and zeroes at $\la=\alpha_{a}+in$ $\forall n\in \Zset$, $a\leq n_{\bm\alpha}$.
\end{notn}
We can easily see that the inverse of the Cauchy matrix \eqref{cau_hyper} can be written in terms of this $\Phifn$ function, which we shall see in the following lemma.
\begin{lem}
\label{lem:cau_hyper_inv_diag_dress}
Let $\Cmat[\bm\alpha\Vert\bm\beta]$ be a Cauchy matrix \eqref{cau_hyper}.
Its inverse $\Cmat^{-1}[\bm\beta\Vert\bm\alpha]$ can be represented as a version of itself with diagonal dressing of $\Phifn$ functions \eqref{phifn_hyper}.
\begin{align}
	\Cmat^{-1}[\bm\beta\Vert\alpha]=
	\diag_{\bm\beta}\left[
	\Phifn'\left(\bm\beta\big|\bm\alpha,\bm\beta\right)
	\right]
	\cdot
	\Cmat[\bm\beta\Vert\bm\alpha]
	\cdot
	\diag_{\bm\alpha}\left[
	\Phifn'\left(\bm\alpha\big|\bm\beta,\bm\alpha\right)
	\right]
	\label{cau_hyper_inv_diag_dress}
\end{align}
The primed symbol $\Phifn'$ signifies the omission of the vanishing term in the products contained in its definition, i.e. $\Phifn'(\beta_a|\bm\alpha,\bm\beta)=\Phifn(\beta_a|\bm\alpha,\bm{\beta_{\hat{a}}})$ where $\bm{\beta_{\hat{a}}}=\bm\beta\setminus\set{\beta_a}$. 
\end{lem}
\begin{rem}
Although this result is trivial and it does not require a proof,
it must be seen as a prelude to a non-trivial version of this result, that we use in \cref{chap:gen_FF}, and which is proved in \cref{chap:mat_det_extn}.
It is worth mentioning that the diagonal dressing interpretation that is provided in \cref{cau_hyper_inv_diag_dress} is extremely significant in this larger picture.
It allows us to exploit the full potential of the \emph{extraction}, as remarked earlier, by replacing the inverse matrix by an equivalent one with the same determinant.
\end{rem}
Let us now return to the determinant representation of the form-factors. We start with the determinant of the matrix $\Fmat[g]$.
Here we assume that %
\begin{enumerate}[noitemsep]
\item the corrections of order $o(1/M)$ in \cref{gau_ex_sol} for $\mu_k$ taken in the bulk do not contribute to the leading order of the determinant, %
\item the corrections for $\mu_k$ taken outside the bulk are also negligible in the determinant due the dominant prevalence of the bulk roots in the set $\bm\mu$.
\end{enumerate}
Taking the common term into the prefactor, we rewrite the determinant as
\begin{align}
	\det\Fmat[g]
	&=
	\pi^{N_0}
	\frac{%
	\bmprod (1+\aux_{g}(\bm\mu))
	}{%
	\bmprod \aux_g^\prime(\bm\la)
	}
	\det_{N_0}\modCau[g]
	.
	\label{gau_ex_I_cau_det_rep}
\end{align}
\index{ff@\textbf{Form-factors}!mat mod Cau gs@$\modCau[g]$: (modified) Cauchy matrix for the ground state|textbf}%
We can see that in this case, the matrix $\modCau[g]$ is purely a Cauchy matrix in the hyperbolic parametrisation, for which we had earlier introduced \cref{notn:rect_cau_mat_hyper}.
\begin{align}
	\modCau[g]&=
	-
	\Cmat\left[\bm\la\big\Vert\bm{\check{\mu}}\right]
	.
	\label{cau_det_rep_I_cau_mat}
\end{align}
Similarly we can write the determinant of the matrix $\Fmat[e]$ as
\begin{align}
	\det\Fmat[e]&=
	\pi^{N_0+1}
	\frac{%
	\bmprod (1+\aux_{e}(\bm\la))
	}{%
	\bmprod \aux_{e}^\prime(\bm\mu)
	}
	\det_{N_0+1}\modCau[e]
	.
	\label{gau_ex_II_cau_det_rep}
\end{align}
Let us note that in the process of obtaining the above \cref{gau_ex_II_cau_det_rep}, we have taken the common terms out into the prefactors and ignored the sub-leading terms.
From \cref{gau_ex_sol_II_decoupled,gau_ex_II_FW}, we can see that the matrix $\modCau[e]$ is composed of the blocks:
\index{ff@\textbf{Form-factors}!mat mod Cau es@$\modCau[e]$: (modified) Cauchy matrix for an excited state|textbf}%
\begin{align}
	\modCau[e]&=
	\left(
	\Cmat\left[\bm{\check\la}\big\Vert\bm\mu\right]
	+
	\Cmat\left[\bm{\check\la}\big\Vert\bm\hle\right]
	\cdot
	\Acal_e^{-1}
	\Rcal[\bm\hle\Vert\bm\mu]
	~
	\Big|
	~
	\bar{\Ucal}\left[\bm{\check\la}\right]
	\right)
	\label{cau_det_rep_2sp_exc_block_rep}
\end{align}
where the matrix $\Rcal[\bm\hle\Vert\bm\mu]$ was introduced in \cref{defn:den_hle_mat_rect}.
The matrix $\bar{\Ucal}$ in \cref{cau_det_rep_2sp_exc_block_rep} is related to the matrix $\Ucal$ of the Foda-Wheeler block \eqref{gau_ex_II_FW} in the Slavnov matrix.
This relation is as follows: 
\index{ff@\textbf{Form-factors}!mat mod Cau es cau@\hspace{1em}$\bar\Ucal$: renorm. $\Ucal$, as a block inside \rule{3em}{1pt}}%
\begin{align}
	\bar{\Ucal}_{ka}[\bm{\check\la}]
	&=
	\frac{\Ucal_{ka}[\bm{\check\la}]}{\pi\,(1+\aux_e(\check\la_k))}
	=
	\frac{%
	\aux_{e}(\check{\la}_{k})(\check{\la}_{k}+i)^{a-1}-\check{\la}_{k}^{a-1}
	}{%
	\pi~(\aux_{e}(\check{\la}_{k})+1)
	}
	.
	\label{gau_ex_II_det_FW_block}
\end{align}
Let us remark that we can construct a larger Cauchy matrix composed of the two rectangular Cauchy matrices that appear in the expression \eqref{cau_det_rep_2sp_exc_block_rep} above.
\begin{subequations}
\begin{align}
	\Cmat\left[\bm{\check\la}\big\Vert\bm{\rh}\right]
	=
	\left(
	\Cmat\left[\bm{\check\la}\big\Vert\bm{\rl}\right]
	~\Big|~
	\Cmat\left[\bm{\check\la}\big\Vert\bm{\hle}\right]
	\right)
	\label{big_cau_2sp_blocks}
\end{align}
by combining the individual blocks in \cref{cau_det_rep_2sp_exc_block_rep}.
In the current scenario, it turns out to be a square matrix of order $N_0+1$ with components:
\begin{align}
	\Cmat_{jk}[\bm{\check\la}|\bm{\rh}]
	=
	\frac{1}{\sinh\pi(\check\la_j-\rh_k)}
	.
	\label{big_cau_2sp}
\end{align}
\label{big_cau_2sp_all}
\end{subequations}
This larger matrix will be exploited when we compute the Cauchy determinant in the thermodynamic limit.
\par
Let us now substitute the determinants from \cref{gau_ex_I_cau_det_rep,gau_ex_II_cau_det_rep} into the expression \eqref{det_rep_fini_ff} and we obtain the following determinant representation for the form-factors in the two-spinon sector
\begin{align}
	\left|\FF^{z}\right|^2
	=
	-2\pi^{M+1}
	\frac{%
	\bmprod \revtf(\bm\la)
	}{%
	\bmprod \revtf(\bm\mu)
	}
	\frac{%
	\bmprod(\bm\mu-\bm\la)\bmprod(\bm\la-\bm\mu)%
	}{%
	\bmprod^\prime(\bm\mu-\bm\mu)
	\bmprod^\prime(\bm\la-\bm\la)
	}
	\det_{N_0}\modCau[g]
	\det_{N_0+1}\modCau[e]
	.
	\label{cau_det_rep_2sp}
\end{align}
Here we have combined some of the prefactors in the original expression \eqref{det_rep_fini_ff}, containing Baxter polynomials, the exponential counting functions and their derivatives from \cref{gau_ex_I_cau_det_rep,gau_ex_II_cau_det_rep}, in order to rewrite it in terms of the ratio of the eigenvalues of the transfer matrix.
The latter were defined in \cref{def_revtf} and its thermodynamic limit was studied in \cref{sec:tdl_revtf_append_sec}.
Its definition is reproduced again in the following:
\index{aux@\textbf{Auxiliary functions}!rat eval@$\revtf$: ratio of eigenvalues of the transfer matrix}%
\begin{align}
	\revtf(\nu)
	=
	\frac{1+\aux_{e}(\nu)}{1+\aux_{g}(\nu)}
	\frac{q_e(\nu-i)q_g(\nu)}{q_g(\nu-i)q_e(\nu)}
	.
\end{align}
\section[Thermodynamic limit]{Thermodynamic limit from the Cauchy determinant representation}
Since the matrix $\modCau[g]$ is a Cauchy matrix of hyperbolic functions \eqref{cau_det_rep_I_cau_mat} we find that its determinant is given by following expression:
\begin{align}
	\det\modCau[g]
	=\det\Cmat\left[-\bm\la\big\Vert-\bm{\check\mu}\right]
	=
	\bmalt\sinh\pi(\bm{\check\mu}\Vert\bm\la)
	.
\end{align}
Here we use the alternant $\bmalt$ notation 
which was introduced on \cpageref{ind_free_notn,ind_free_notn_end}, to write the Cauchy determinant:
\begin{align}
	\bmalt\sinh\pi(\bm{\check\mu}\Vert\bm\la)
	&=\frac{%
	\prod_{j>k}\sinh\pi(\check{\mu}_{j}-\check{\mu}_{k})
	\prod_{j<k}\sinh\pi(\la_{j}-\la_{k})
	}{%
	\prod_{j,k}\sinh\pi(\mu_{j}-\la_{k})
	}
	.
\end{align}
To compute the determinant of the matrix $\modCau[e]$, we shall first extract from it the larger Cauchy matrix defined in \cref{big_cau_2sp_all} as follows:
\index{ff@\textbf{Form-factors}!mat cvextn es@$\resmat[e]$: mat. obtained from CV extraction for an excited state}%
\begin{align}
	\Pcal&=
	\Cmat^{-1}\left[\bm\rh\big\Vert\bm{\check{\la}}\right]
	\cdot
	\modCau[e]
	.
	\label{cau_ex_2sp_mat}
\end{align}
It permits us to rewrite the representation \eqref{cau_det_rep_2sp} as
\begin{align}
	\left|\FF^{z}\right|^2
	=
	-2\pi^{M+1}
	\frac{%
	\bmprod \revtf(\bm\la)
	}{%
	\bmprod \revtf(\bm\mu)
	}
	\frac{%
	\bmalt\sinh\pi(\bm{\check\mu}\Vert\bm\la)
	\bmalt\sinh\pi(\bm{\check\la}\Vert\bm\mu)
	}{%
	\bmalt(\bm\la\Vert\bm\mu)
	\bmalt(\bm\mu\Vert\bm\la)
	}
	\det_{N_0+1}\Pcal
	.
	\label{cau_det_rep_2sp_extr_cau}
\end{align}
\subsection{Cauchy extraction}
\label{sub:cau_ex_2sp}
We will now compute the matrix $\Pcal$ and its determinant from the extraction \eqref{cau_ex_2sp_mat} of the Cauchy matrix.
Since $\modCau[g]$ matrix \eqref{blocks_gau_ex_gen_II} and $\Cmat\left[\bm{\check\la}\Vert\bm\rh\right]$ matrix \eqref{big_cau_2sp_blocks} are both divided in blocks of columns, it compels us to partition the matrix $\Pcal$ into four blocks named according to the cardinal directions as seen in the following:
\begin{align}
	\Pcal&=
	\begin{pmatrix}
		\Pcal^{\mathrm{NW}}	&	
		\Pcal^{\mathrm{NE}}	\\	
		\Pcal^{\mathrm{SW}}	&	
		\Pcal^{\mathrm{SE}}	
	\end{pmatrix}
	.
	\label{cau_ex_mat_blocks}
\end{align}
In the north-west we have a square matrix $\Pcal^{\mathrm{NW}}$ which is the largest block comprising of $N_1$ columns. In the north-east and south-west we have the rectangular matrices $\Pcal^{\mathrm{NE}}$, $\Pcal^{\mathrm{SW}}$ which forms the off-diagonal rectangular matrices comprising of two columns and two rows respectively. And finally in the south-east we have a square matrix $\Pcal^{\mathrm{SE}}$ which is another diagonal block of order two.
We can see from \cref{cau_ex_2sp_mat} that these four block matrices are given  by the following equations.
\begin{subequations}
\begin{align}
	\Pcal^{\mathrm{NW}}&=
	\Cmat^{-1}\left[\bm\mu\big\Vert\bm{\check\la}\right]
	\cdot
	\left(
	\Cmat\left[\bm{\check\la}\big\Vert\bm\mu\right]
	+
	\Cmat\left[\bm{\check\la}\big\Vert\bm\hle\right]
	\cdot
	\Acal_e^{-1}\Rcal\left[\bm\hle|\bm\mu\right]
	\right)
	,
	\label{cau_ex_cau_diag}
	\\
	\Pcal^{\mathrm{SW}}&=
	\Cmat^{-1}\left[\bm\hle\big\Vert\bm{\check\la}\right]
	\cdot
	\left(
	\Cmat\left[\bm{\check\la}\big\Vert\bm\mu\right]
	+
	\Cmat\left[\bm{\check\la}\big\Vert\bm\hle\right]
	\cdot
	\Acal_e^{-1}\Rcal\left[\bm\hle|\bm\mu\right]
	\right)
	\label{cau_ex_cau_off-diag}
	;
	\\
	\Pcal^{\mathrm{NE}}&=
	\Cmat^{-1}\left[\bm\mu\big\Vert\bm{\check\la}\right]
	\cdot
	\bar{\Ucal}[\bm{\check\la}]
	,
	\label{cau_ex_non-cau_off-diag}
	\\
	\Pcal^{\mathrm{SE}}&=
	\Cmat^{-1}\left[\bm\hle\big\Vert\bm{\check\la}\right]
	\cdot
	\bar{\Ucal}[\bm{\check\la}]
	.
	\label{cau_ex_non-cau_diag}
\end{align}
\label{cau_ex_matricial}
\end{subequations}
The Cauchy extraction in \cref{cau_ex_cau_off-diag,cau_ex_cau_diag} is a trivial one.
It can be easily seen from $\Cmat^{-1}[\bm{\mu}\Vert\bm{\check\la}]\cdot\Cmat[\bm{\check\la}\Vert\bm\mu] = \Id$ that we have
\begin{subequations}
\begin{align}
	\Pcal^{\mathrm{NW}}&= \Id,
	\\
	\Pcal^{\mathrm{SW}}&=\Acal_e^{-1}\Rcal[\bm\hle\Vert\bm\mu]
	.
\end{align}
\end{subequations}
Nonetheless, the following lemma is presented and framed in a more general setting. This is so because the reasoning used to prove this lemma is crucial for the later computations. It will shed light not only on the computation of non-trivial part of the Cauchy extraction in \cref{cau_ex_non-cau_diag,cau_ex_non-cau_off-diag} that we encounter in the current setting, but also the other instances that we will encounter in the generic setting for higher spinons.
\begin{lem}
Let $\bm\alpha$ and $\bm\beta$ denote two sets of complex parameters of the same cardinality $n_{\bm\alpha}=n_{\bm\beta}$.
\\
Let $f$ be a periodic meromorphic function satisfying $f(\la+i)=-f(\la)$, with simple poles forming the set $\bm\gamma+i\Zset$, which are all distinct from the poles of the $\Phifn$ function modulo its periodicity, i.e. $\bm\gamma+i\Zset\cap\bm\alpha+i\Zset=\varnothing$.
\\
We also assume that $f(\tau)$ is bounded at infinity and $\lim_{\tau\to\infty}f(\tau)=A$.
Then following sum can be evaluated as%
\begin{align}
	\frac{1}{\pi}
	\sum_{j=1}^{n_{\bm\alpha}}
	\Phifn'(\alpha_j|\bm\beta,\bm{\alpha})
	\frac{1}{\sinh\pi(\alpha_{j}-\beta_{k})}
	f(\alpha_j|\bm\gamma)
	&=
	-
	\bmsum_{\bm\gamma}
	\Phifn(\bm\gamma|\bm{\beta_{\hat{k}}},\bm{\alpha})
	\res f(\bm\gamma)
	.
	\label{Phifn_sum_period_method_trivial_case_result}
\end{align}
where $\bm{\beta_{\hat{k}}}=\bm\beta\setminus\set{\beta_k}$.
\label{lem:Phifn_sum_period_method_trivial_case}
\end{lem}
\begin{proof}
The sum in \cref{Phifn_sum_period_method_trivial_case_result} is taken over the poles of the function $\Phifn(\tau|\bm\beta,\bm{\alpha})$. Since the cardinalities of the two sets match $n_{\bm\nu}=n_{\bm{\alpha}}$, we can see that function tends to
	\begin{align}
		\lim_{\tau\to\infty}\Phifn(\tau|\bm\beta,\bm{\alpha})=1.
  \end{align}
We can choose to take the parameters $\bm\alpha$, $\bm\beta$ and $\bm\gamma$ in the fundamental domain.
We can see that it has simple poles at $\tau\in\bm\alpha$, $\tau=\beta_k$ and $\tau\in\gamma$ inside the fundamental domain.
We can also see that the residue of the pole for $\tau=\beta_k$ vanishes since it is a zero of the $\Phifn$ function.
Therefore we can write this sum as contour integral:
\begin{multline}
	\frac{1}{\pi}
	\sum_{j=1}^{n_{\bm\alpha}}
	\Phifn'(\alpha_j|\bm\beta,\bm\alpha)
	\frac{1}{\sinh\pi(\beta_{k}-{\alpha}_{j})}
	f(\alpha_j)
	=
	\frac{1}{2\pi i}
	\oint_{\partial(\Scal_{\bm\alpha}\setminus\bm{\gamma^{\text{in}}})}
	\Phifn(\tau|\bm\beta,\bm{\alpha})
	\frac{1}{\sinh\pi(\beta_{k}-\tau)}
	f(\tau)
	d\tau
	\\
	\frac{1}{2\pi i}
	\oint_{\partial\Scal_{\bm\alpha}}
	\Phifn(\tau|\bm\beta,\bm{\alpha})
	\frac{1}{\sinh\pi(\tau-\beta_{k})}
	f(\tau)
	d\tau
	+
	\bmsum
	\Phifn(\bm{\gamma^{\text{in}}}|\bm{\beta_{\hat{k}}},\bm{\alpha})
	\res f(\bm{\gamma^{\text{in}}})
	\label{phifn_sum_periodic_trivial_case_in_res}
\end{multline}
where $\Scal_{\bm\alpha}$ is rectangular region that contains all the poles $\bm\alpha$. We have to exclude those extra poles $\bm{\gamma^{\text{in}}}$ which inside this rectangle. The boundary of this punctured rectangle $\Scal_{\bm\alpha}\setminus\bm{\gamma^{\text{in}}}$ is the contour chosen in the above integral.
Since the integrand goes to zero exponentially for large values of $\tau$, the integral on the vertical edges can be made to vanish by taking the length of the rectangle to infinity.
Therefore by exploiting the periodicity of the integrand on the upper edge of the contour, we can obtain a new contour which is complementary to the original one and it does not contain any poles of the type $\tau\in\bm\alpha+i\Zset$.
Therefore we can write
\begin{multline}
	\frac{1}{2\pi i}
	\oint_{\partial\Scal_{\bm\alpha}}
	\Phifn(\tau|\bm\beta,\bm\alpha)
	\frac{f(\alpha_j)}{\sinh\pi({\alpha}_{j}-\beta_{k})}
	d\tau
	=
	-
	\frac{1}{2\pi i}
	\oint_{\partial\Scal_{\bm{\hat\alpha}}}
	\Phifn(\tau|\bm\beta,\bm\alpha)
	\frac{f(\tau)}{\sinh\pi(\tau-\beta_{k})}
	d\tau
	\\
	=
	-
	\bmsum
	\Phifn(\bm{\gamma^{\text{out}}}|\bm{\beta_{\hat{k}}},\bm{\alpha})
	\res f(\bm{\gamma^{\text{out}}})
	\label{phifn_sum_periodic_trivial_case_out_res}
\end{multline}
since the poles $\bm\gamma^{\text{out}}=\bm\gamma\setminus\bm{\gamma^{\text{in}}}$ which were missed in the original contour will be included.
This is illustrated in \cref{fig:cau_ex_periodic_contours}.
From \cref{phifn_sum_periodic_trivial_case_out_res,phifn_sum_periodic_trivial_case_in_res} we get the result in \cref{Phifn_sum_period_method_trivial_case_result}.
\begin{figure}[bt]
\centering
\includegraphics[width=\textwidth]{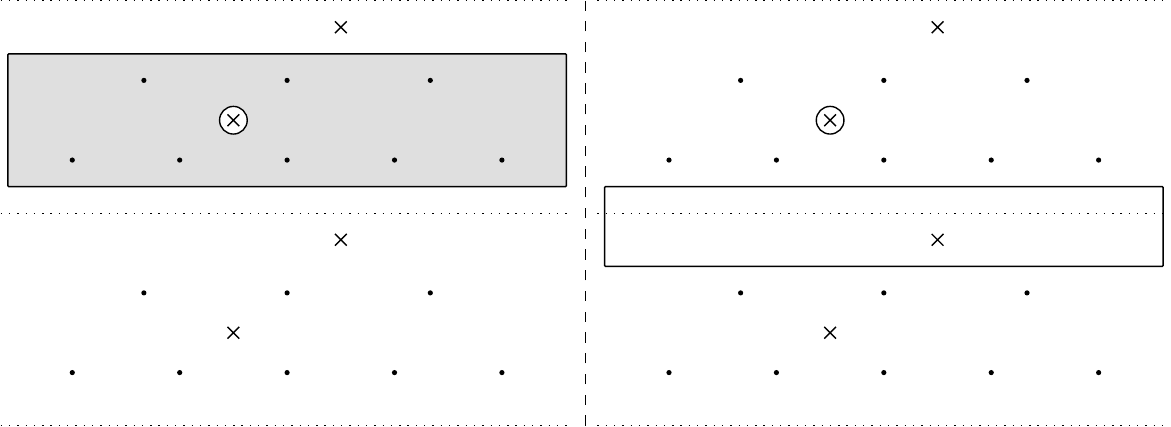}
\caption[Reformulation of integrals for Cauchy extraction using periodicity (a simplified demonstration)]{An illustrative example of use of periodicity to change the contour of integration. Here a cross `\tikzcross' represents a pole of the function $f$ in $\bm\gamma$. On the left hand-side we have the contour which was originally defined and on the right-hand side we have the newly obtained contour by virtue of the periodicity of the integrand.}
\label{fig:cau_ex_periodic_contours}
\end{figure}
\end{proof}
For the Cauchy extraction on the Foda-Wheeler block of columns we get the two blocks $\Pcal^{\mathrm{NE}}$ and $\Pcal^{\mathrm{SE}}$. Here we shall consider them together in a combined east block:
\begin{align}
	\Pcal^{\textrm{E}}=\begin{pmatrix} \Pcal^{\mathrm{NE}} \\ \Pcal^{\mathrm{SE}} \end{pmatrix}.
\end{align}
From \cref{cau_ex_non-cau_diag,cau_ex_non-cau_off-diag} we get the following summation for this block:
\begin{align}
	\Pcal^{\textrm{E}}_{j,a}
	&=
	\Phifn'(\rh_j|\bm{\check\la},\bm\rh)
	\frac{1}{\pi}
	\sum_{k=1}^{N_0+1}
	\Phifn'(\check\la_k|\bm\rh,\bm{\check\la})
	\frac{1}{\sinh\pi(\rh_{j}-\check{\la}_{k})}
	\frac{%
	\aux_{e}(\check{\la}_{k})
	(\check{\la}_{k}+i)^{a-1}
	-
	\check{\la}_{k}^{a-1}
	}{%
	\aux_{e}(\check{\la}_{k})+1
	}
	.
	\label{cau_ex_2sp_FW_sum}
\end{align}
It is similar to the summations in \cref{Phifn_sum_period_method_trivial_case_result}, except that the function $f$ is given in this case by \cref{gau_ex_II_FW} for the Foda-Wheeler columns, which is not a periodic function.
We see that it has simple poles in $\bm\rh$, all of which have zero residue except for the one $\tau=\rh_j$.
Thus we can write the integral:
\begin{multline}
	\Pcal^E_{j,a}
	=
	\\
	\frac{1}{2\pi i}
	\Phifn'(\rh_j|\bm{\check\la},\bm\rh)
	\left(
	\int_{\Rset-i\alpha}
	-
	\int_{\Rset+\frac{i}{2}+i\alpha}
	-
	\oint_{\varepsilon_{\rh_{j}}}
	\right)
	\frac{%
	\Phifn(\tau|\bm\rh,\bm{\check\la})
	}{%
	\sinh\pi(\tau-\rh_j)%
	}
	\frac{%
	\aux_{e}(\tau)
	(\tau+i)^{a-1}
	-
	\tau^{a-1}
	}{%
	\aux_{e}(\tau)+1
	}	d\tau
	.
	\label{cau_ex_fw_oint}
\end{multline}
Note that the $\alpha$ in this expression is a positive real constant whose value lies in the $0<\alpha<\frac{1}{2}$.
We can make sure that the integrand goes exponentially to zero at infinity, therefore the integrals on the vertical edges at the boundary vanish.
The $\partial\varepsilon_{\rh_j}$ is the infinitesimal neighbourhood of the extra simple pole $\tau=\rh_k$.
Hence, it is the integral over the parallel lines in the outer contour that remain the subject of our investigation.
Here we invoke the estimation of the counting function which holds in the bulk.
As a result, we see that the function $\aux_e$ is exponentially small in $M$ on the branch $\Rset+\frac{i}{2}+i\alpha$ and it is exponentially large in $M$ on the branch $\Rset-i\alpha$.
With this assumption, we can write
\begin{multline}
	\Pcal^{\mathrm{E}}_{j,a}
	=
	-
	\frac{1}{2\pi i}
	\Phifn'(\rh_j|\bm{\check\la},\bm\rh)
	\left\lbrace
	\int_{\Rset-i\alpha}
	\frac{%
	(\tau+i)^{a-1}
	\Phifn(\tau|\bm\rh,\bm{\check\la})
	}{%
	\sinh\pi(\tau-\rh_j)%
	}d\tau	
	+
	\int_{\Rset+\frac{i}{2}+i\alpha}
	\frac{%
	\tau^{a-1}
	\Phifn(\tau|\bm\rh,\bm{\check\la})
	}{%
	\sinh\pi(\tau-\rh_j)%
	}	d\tau
	\right\rbrace
	\\
	-
	\frac{\rh_j^{a-1}+(\rh_j+i)^{a-1}}{\aux'_e(\rh_j)}
	.
\end{multline}
Let us now use the periodicity property of the $\Phifn$ function.
Under the change of variables $\tau\to\tau+i$, we see that the lower branch of integration contour can be transformed according to the following:
\begin{align}
	\int_{\Rset-i\alpha}
	\frac{%
	(\tau+i)^{a-1}
	\Phifn(\tau|\bm\rh,\bm{\check\la})
	}{%
	\sinh\pi(\tau-\rh_j)%
	}d\tau	
	=
	-
	\int_{\Rset+i-i\alpha}
	\frac{%
	\tau^{a-1}
	\Phifn(\tau|\bm\rh,\bm{\check\la})
	}{%
	\sinh\pi(\tau-\rh_j)%
	}d\tau	
	.
	\label{periodicity_lower_int_transf}
\end{align}
After substituting \cref{periodicity_lower_int_transf} back, we obtain a new closed contour by adding vanishing contributions of edges.
Since the integrand is regular inside this contour, we get a null result for this part,
\begin{align}
	\left(
	\int_{\Rset+\frac{i}{2}+i\alpha}
	-
	\int_{\Rset+i-i\alpha}
	\right)
	\frac{%
	\tau^{a-1}
	\Phifn(\tau|\bm\nu,\bm{\check\la})
	}{%
	\sinh\pi(\nu_{k}-\tau)%
	}d\tau
	=
	0
	.
	\label{cau_ex_fw_new_contour_zero}
\end{align}
Therefore we are only left with the infinitesimal part of the contours in \cref{cau_ex_fw_oint}, which leads to the residue term for an extra pole.
Based on this observation, we shall now define the column vectors $\Wcal_a[\bm{\rh^+}]$ as follows:
\index{ff@\textbf{Form-factors}!mat cvextn es fw@\hspace{1em}$\Wcal$: Foda-Wheeler block inside \rule{3em}{1pt}}%
\begin{align}
	\Wcal_a(\rh^+_k)
	&=
	\Pcal^{\mathrm{E}}_{a,k}
	=
	-
	\frac{\rh_j^{a-1}+(\rh_j+i)^{a-1}}{\aux'_e(\rh_j)}
	;
	&
	\text{for }\quad a=1,2
	.
	\label{cau_ex_2sp_FW_result}
\end{align}
The matrix $\Wcal[\bm{\rh^+}]$ is composed of the two columns:
\begin{align}
	\Wcal[\bm{\rh^+}]
	=
	\left(
	\Wcal_1[\bm{\rh^+}]
	~
	\big|
	~
	\Wcal_2[\bm{\rh^+}]
	\right)
	.
\end{align}
\minisec{Reduction in the size of the determinant}
With this result, all the four blocks \eqref{cau_ex_mat_blocks} of the matrix $\Pcal$ are now determined.
We can express the matrix $\Pcal$ as follows:
\begin{align}
	\Pcal=
	\begin{pmatrix}
		\Id 	&		\Wcal[\bm\mu]
		\\[\jot]
		\Acal_e^{-1}\Rcal[\bm\hle|\bm\mu] &	\Wcal[\bm\hle]
	\end{pmatrix}
	.
\end{align}
Using \cref{lem:mat_det_red}, we can reduce the matrix $\Pcal$ to a smaller square matrix $\Qcal$ of order two, such that their determinants are related to each other through the following relation:
\begin{align}
	\det_{N_0+1}\Pcal=
	\frac{1}{\aux'_e(\hle_1)\aux'_e(\hle_2)}
	\det_{2}\Qcal
	.
\end{align}
Notice that here we also extract the diagonal matrix $\pi\Acal_e^{-1}[\bm\hle]$ in addition to the reduction through \cref{lem:mat_det_red}, this tells us that the components of the reduced matrix can be expressed as
\begin{align}
	\Qcal=
	\Acal_e\Wcal[\bm\hle]-\Rcal[\bm\hle|\bm\mu]\cdot\Wcal[\bm\mu]
\end{align}
where $\Acal_e\Wcal$ denotes the product of diagonal matrix $\Acal_e[\bm\hle]$ \eqref{diag_cfn_der_mat} with $\Wcal[\bm\hle]$.
Let us write down the components of the matrix $\Qcal$ explicitly:
\begin{subequations}
\begin{align}
	\Qcal_{a,1}
	&=
	-2
	-
	4\pi i
	\bmsum
	\frac{\rden_h(\bm\mu-\hle_a)}{\aux'_e(\bm\mu)}
	,
	\label{reduced_mat_fw_2sp_deg0}
	\\
	\Qcal_{a,2}
	&=
	-2\left(\hle_a+\frac{i}{2}\right)
	-
	4\pi i
	\bmsum
	\frac{\rden_h(\bm\mu-\hle_a)}{\aux'_e(\bm\mu)}
	\left(\bm\mu+\frac{i}{2}\right)
	.
	\label{reduced_mat_fw_2sp_deg1}
\end{align}
\label{reduced_mat_fw_2sp}
\end{subequations}
The sums in both these cases can be written as integrals:
\begin{align}
	2\pi i
	\bmsum
	\frac{\rden_h(\bm\mu-\hle_a)}{\aux'_e(\bm\mu)}
	\left(
	\bm\mu+\frac{i}{2}
	\right)^{a-1}
	=
	-
	\int_{\Rset}
	\rden_h(\tau-\hle_a)
	\left(
	\tau-\frac{i}{2}
	\right)^{a-1}
	d\tau
	.
	\label{cau_ex_fw_red_mat_sum_condn}
\end{align}
Due to the regularity of the integrand, it suffices to use the regular condensation property for this purpose.
The integral of this form due to \cref{reduced_mat_fw_2sp_deg0} can be readily computed
\begin{align}
	\int_{\Rset}
	\rden_h(\tau-\hle_a)
	d\tau	
	=
	\frac{1}{2}
	\label{int_rden_h_for_FW}
\end{align}
whereas for the integral coming from the column of the degree 1 \eqref{reduced_mat_fw_2sp_deg1} we need to break up the integral in two parts as follows:
\begin{align}
	\int_{\Rset}\left(\tau+\frac{i}{2}\right)\rden_{h}(\tau-\hle_{a})d\tau
	&=
	\int_{\Rset}\left(\tau-\hle_{a}\right)\rden_{h}(\tau-\hle_{a})
    d\tau
	+
	\left(\hle_{a}+\frac{i}{2}\right)\int_{\Rset}\rden_{h}(\tau-\hle_{a})
    d\tau
	.
\end{align}
Here we can see that the first part is an integral over an odd function and it vanishes.
The second part can be computed similar to \cref{int_rden_h_for_FW}.
It gives rise to
\begin{align}
	\int_{\Rset}\left(\tau+\frac{i}{2}\right)\rden_{h}(\tau-\hle_{a})d\tau
	=
	\frac{\hle_a+\frac{i}{2}}{2}	
	\label{int_rden_h_for_FW_deg1_symm}
	.
\end{align}
Overall, the substitution of \cref{int_rden_h_for_FW,int_rden_h_for_FW_deg1_symm} back into \cref{cau_ex_fw_red_mat_sum_condn,reduced_mat_fw_2sp} tell us that
\begin{align}
	\Qcal
	=
	\Acal_e\Wcal[\bm\hle]
	-
	\frac{1}{2}
	\Acal_e\Wcal[\bm\hle]
	=
	\frac{1}{2}
	\Acal_e\Wcal[\bm\hle]
	.
	\label{res_mat_fw_2sp_result}
\end{align}
We also easily see from \cref{cau_ex_2sp_FW_result} that $\Qcal$ is a Vandermonde matrix of size two:
\begin{align}
	\Qcal=
	\vmat\left[\bm\hle+\tfrac{i}{2}\right]
	=
	\begin{pmatrix}
		-1	&	-\left(\hle_1+\frac{i}{2}\right)
		\\[\jot]
		-1	&	-\left(\hle_2+\frac{i}{2}\right)
	\end{pmatrix}
	.
\end{align}
Therefore its determinant is given by,
\begin{align}
	\det\Pcal
	=
	\det\Qcal
	=
	(\hle_2-\hle_1)
	.
\end{align}
\section{Thermodynamic limit from an infinite product form}
\label{sec:tdl_inf_prod}
After the extraction of the Cauchy matrix, we see that the determinant representation for the two-spinon form factors can be written as
\begin{align}
	\left|\FF^{z}\right|^2
	&=-2\pi^{M+1}
	\frac{\hle_{2}-\hle_{1}}{\aux_{e}^\prime(\hle_{1})\aux_{e}^\prime(\hle_{2})}
	\frac{\bmprod\revtf(\bm\la)}{\bmprod\revtf(\bm\mu)}
	\frac{%
	\bmalt\sinh\pi(\bm{\check\mu}\Vert\bm\la)
	\bmalt\sinh\pi(\bm{\check\la}\Vert\bm\nu)
	}{%
	\bmalt^2(\bm\la\Vert\bm\mu)
	}
	.
	\label{ff_2sp_pref_alts_init}
\end{align}
Note that we have used the superalternant notation for the Cauchy determinants $\bmalt$, which is introduced on \cpagerefrange{ind_free_notn}{ind_free_notn_end}.
First of all, let us remark that the terms involving $\check{\mu}_{N_0}=\frac{i}{2}$ and $\check{\la}_{N_0+1}=\frac{i}{2}$ mutually cancel out to from the product of Cauchy determinants.
It leaves us with 
\begin{align}
	\bmalt\sinh\pi(\bm{\check\mu}\Vert\bm\la)
	\bmalt\sinh\pi(\bm{\check\la}\Vert\bm\nu)
	=
	\frac{-1}{\cosh\pi(\hle_{1})\cosh\pi(\hle_{2})}
	\bmalt\sinh\pi(\bm{\mu}\Vert\bm\la)
	\bmalt\sinh\pi(\bm{\la}\Vert\bm\nu)
	\label{trig_alt_triplet_i2_sepn}
\end{align}
We will now further expand the second hyperbolic Cauchy determinants by factoring out the terms with hole parameters.
\begin{align}
	\bmalt\sinh\pi(\bm\la\Vert\bm\nu)
	&=
	\sinh\pi(\hle_{1}-\hle_{2})
	\prod_{a=1}^{2}
	\frac{%
	\bmprod\sinh\pi(\hle_{a}-\bm\mu)
	}{%
	\bmprod\sinh\pi(\hle_{a}-\bm\la)
	}
	\bmalt\sinh\pi(\bm\la\Vert\bm\mu)
	\label{trig_alt_hle_param_sepn}
\end{align}
Let us now replace the derivatives of the counting functions in terms of the densities, assuming the holes are in the bulk, we can see that it is the ground state density term that dominates.
Thus we can write,
\begin{align}
	\aux_{e}^\prime(\hle_{a})
	&=
	\frac{%
	i\pi \, M
	}{%
	\cosh\pi\hle_{a}
	}
	+
	O\left(1\right)
	\label{cfn_der_exc_leading_order}
\end{align}
as we also ignore the hole-density terms, which are put into the correction of the first sub-leading order $O(1)$.
We have computed the thermodynamic limit of the function $\revtf(\la)$ in \cref{sec:tdl_revtf_append_sec}.
As we substitute \cref{cfn_der_exc_leading_order,trig_alt_hle_param_sepn,trig_alt_triplet_i2_sepn} into \cref{ff_2sp_pref_alts_init} we obtain the following representation:
\begin{align}
	\left|\FF^{z}\right|^2
	&=
	\frac{2\pi^{M-1}}{M^2}
	(\hle_{2}-\hle_{1})\sinh\pi(\hle_{2}-\hle_{1})
	\frac{\bmprod\revtf(\bm\la)}{\bmprod\revtf(\bm\mu)}
	\prod_{a=1}^{2}
	\frac{%
	\bmprod\sinh\pi(\hle_{a}-\bm\mu)
	}{%
	\bmprod\sinh\pi(\hle_{a}-\bm\la)
	}
	\frac{%
	\bmalt^2\sinh\pi(\bm\la\Vert\bm\mu)
	}{%
	\bmalt^2(\bm\la\Vert\bm\mu)
	}
	.
	\label{ff_2sp_pref_alts_hle_seperated}
\end{align}
Let us recall that the thermodynamic limit of the $\revtf$ function was obtained in \cref{revtf_tdl}, which is reproduced below:
\begin{align}
	\frac{%
	\bmprod \revtf(\bm\la)
	}{%
	\bmprod \revtf(\bm\mu)
	}
	&=
	\prod_{a=1}^{2}
	\frac{%
	\bmprod \tanh\frac{\pi(\hle_{a}-\bm\la)}{2}
	}{%
	\bmprod \tanh\frac{\pi(\hle_{2}-\bm\mu)}{2}
	}
	\label{revtf_tdl_pref_2sp}
\end{align}
Let us now combine it with the other terms to produce
\begin{align}
	\frac{%
	\bmprod\sinh\pi(\hle_{a}-\bm\mu)
	}{%
	\bmprod\sinh\pi(\hle_{a}-\bm\la)
	}&=
	\frac{1}{4}
	\frac{%
	\bmprod
	\sinh\frac{\pi(\hle_{a}-\bm\mu)}{2}
	\cosh\frac{\pi(\hle_{a}-\bm\mu)}{2}
	}{%
	\bmprod
	\sinh\frac{\pi(\hle_{a}-\bm\la)}{2}
	\cosh\frac{\pi(\hle_{a}-\bm\la)}{2}
	}
	\label{sinh_hle_split_pref_2sp}
	,
\end{align}
and,
\begin{align}
	\frac{%
	\bmalt^2\sinh\pi(\bm\la\Vert\bm\mu)
	}{%
	\bmalt^2(\bm\la\Vert\bm\mu)
	}
	=
	\pi^{-M+3}
	\frac{%
	\bmprod_{\bm\la,\bm\mu}
	\Gamma(1+\frac{\bm\la-\bm\mu}{2i\sigma})
	\Gamma(\frac{1}{2}+\frac{\bm\la-\bm\mu}{2i\sigma})
	}{%
	\bmprod_{\bm\mu}
	\Gamma(1+\frac{\bm\mu-\bm\mu}{2i\sigma})
	\Gamma(\frac{1}{2}+\frac{\bm\mu-\bm\mu}{2i\sigma})
	}
	\frac{%
	\bmprod_{\bm\la,\bm\mu}
	\Gamma(1+\frac{\bm\mu-\bm\la}{2i\sigma})
	\Gamma(\frac{1}{2}+\frac{\bm\la-\bm\mu}{2i\sigma})
	}{%
	\bmprod_{\bm\la}
	\Gamma(1+\frac{\bm\la-\bm\la}{2i\sigma})
	\Gamma(\frac{1}{2}+\frac{\bm\mu-\bm\mu}{2i\sigma})
	}
	\label{sinh_split_pref_2sp}
\end{align}
where the latter is derived from using the identity \eqref{id_Gamma_fns_trig_hyper}.  
Substituting the above expressions from \crefrange{revtf_tdl_pref_2sp}{sinh_split_pref_2sp} into \cref{ff_2sp_pref_alts_hle_seperated} allows us to rewrite it in terms of the auxiliary function $\Omegfn$ as follows:
\begin{align}
	\left|\FF^{z}\right|^2
	&=
	\frac{\pi}{2M^2}
	(\hle_2-\hle_1)\sinh\pi (\hle_2-\hle_1)
	\frac{%
	\prod_{j}\Omegfn(\mu_{j}|\bm\mu,\bm\la)
	}{%
	\prod_{j}\Omegfn(\la_{j}|\bm\mu,\bm\la)
	}
	.
	\label{ff_2sp_inf_prod_omegfn_form}
\end{align}
\begin{notn}	
The function $\Omegfn$ is defined as
\index{aux@\textbf{Auxiliary functions}!pref aux@$\Omegfn$: fn. involved in the computation of the prefactors|textbf}%
\begin{align}
	\Omegfn(\tau|\bm\mu,\bm\la)
	&=
	\left\lbrace
	\prod_{\sigma=\pm}
	\left\lbrace
	\bmprod
	\frac{%
	\Gamma^3\left(\frac{1}{2}\right)
	}{%
	\Gamma^2\left(\frac{1}{2}+\frac{\tau-\bm\hle}{2i\sigma}\right)
	}
	\right\rbrace
	\frac{%
	\bmprod
	\Gamma\left(1+\frac{\tau-\bm\la}{2i\sigma}\right)
	\bmprod
	\Gamma\left(\frac{1}{2}+\frac{\tau-\bm\la}{2i\sigma}\right)
	}{%
	\bmprod
	\Gamma\left(1+\frac{\tau-\bm\mu}{2i\sigma}\right)
	\bmprod
	\Gamma\left(\frac{1}{2}+\frac{\tau-\bm\mu}{2i\sigma}\right)
	}
	\right\rbrace
	\label{omegfn_2sp_def}
	.
\end{align}
\end{notn}
Let us now write this function as an infinite product:
\begin{align}
	\Omegfn(\tau|\bm\mu,\bm\la)
	&=
	\prod_{n=1}^{\infty}
	\Omegfn_{n}(\tau|\bm\mu,\bm\la)
	.
\end{align}
From the Weierstrass form of the $\Gamma$ function, we can see that the terms $\Omega_{n}$ in this infinite product are as follows:
\begin{align}
	\Omegfn_{n}(\tau|\bm\mu,\bm\la)
	&=
	\frac{16n^2}{\left(n-\frac{1}{2}\right)^6}
	\prod_{\sigma=\pm}
	\phifn(\tau+2ni\sigma|\bm\mu,\bm\la)
	\phifn(\tau+2ni\sigma|\bm\mu,\bm\la)
	\bmprod
	\left\lbrace
	\left(n-\frac{1}{2}+\frac{(\tau-\bm\hle)^2}{4}\right)
	\right\rbrace^2
	.
	\label{omeg_infprod_2sp}
\end{align}
The $\phifn(\tau|\bm\mu,\bm\la)$ function and its thermodynamic limit was studied in \cref{sec:tdl_phifn_append}.
Let us now substitute the result obtained there in \cref{phi_tdl} for the $\phifn$ function, once it is read in the context of a two-spinon excitations (i.e. by ignoring the terms containing the complex roots).
This substitution allows us to write the general terms $\Omegfn_n$ of the infinite product from the above \cref{omeg_infprod_2sp}, as follows:
\begin{align}
	\Omegfn_{n}(\tau|\bm\mu,\bm\la)
	&=
	\frac{%
	n^2\bmprod\left\lbrace
	\left(n-\frac{1}{2}\right)^2
	+
	\frac{(\tau-\bm\hle)^2}{4}
	\right\rbrace
	}{%
	\left(n-\frac{1}{2}\right)^6
	}
	.
\end{align}
Therefore, the infinite product \eqref{omeg_infprod_2sp} takes a new form that is equivalent to the old one in the thermodynamic limit. It reads,
\begin{align}
	\frac{%
	\prod_{j}\Omegfn(\mu_{j}|\bm\mu,\bm\la)
	}{%
	\prod_{j}\Omegfn(\la_{j}|\bm\mu,\bm\la)
	}
	&=
	\prod_{n=1}^{\infty}
	\left\lbrace
	\left(
	\frac{n-\frac{1}{2}}{n}
	\right)^2
	\left(
	\frac{%
	\Gamma\left(n+\frac{1}{2}\right)
	}{%
	\Gamma\left(n\right)
	}
	\right)^4
	\prod_{\sigma=\pm}
	\frac{%
	\Gamma^2\left(n-\frac{1}{2}+\frac{\hle_{2}-\hle_{1}}{2i\sigma}\right)
	}{%
	\Gamma^2\left(n+\frac{\hle_{2}-\hle_{1}}{2i\sigma}\right)
	}
	\right\rbrace
	.
	\label{tdl_rat_Omegfn_inf_prod}
\end{align}
We can show using \cref{lem:infprod_gamma_fn_crit} that this infinite product is \emph{fine-tuned} to produce a meaningful result, which can be expressed in terms of the Barnes-$G$ function as
\begin{align}
	\frac{%
	\prod_{j}\Omegfn(\mu_{j}|\bm\mu,\bm\la)
	}{%
	\prod_{j}\Omegfn(\la_{j}|\bm\mu,\bm\la)
	}
	&=
	\frac{%
	G^2\left(\frac{1}{2}\right)
	}{%
	G^6\left(\frac{3}{2}\right)
	}%
	\prod_{\sigma=\pm}
	\frac{%
	G^2\left(1+\frac{\hle_{2}-\hle_{1}}{2i\sigma}\right)
	}{%
	G^2\left(\frac{1}{2}+\frac{\hle_{2}-\hle_{1}}{2i\sigma}\right)
	}
	.
	\label{tdl_rat_Omegfn}
\end{align}
Let us now substitute \cref{tdl_rat_Omegfn} back into \cref{ff_2sp_inf_prod_omegfn_form}.
Meanwhile, we will also substitute the following expression that follows from the identity \eqref{id_Gamma_fns_hyper}
\begin{align}
	(\hle_2-\hle_1)\sinh\pi (\hle_2-\hle_1)=
	\frac{4\pi^2}{%
	(\hle_2-\hle_1)%
	}
	\prod_{\sigma=\pm}
	\frac{1}{%
	\Gamma(\frac{\hle_2-\hle_1}{2i\sigma})%
	\Gamma(\frac{1}{2}+\frac{z}{2i\sigma})%
	}
	.
\end{align}
As a result, we can see that \cref{ff_2sp_inf_prod_omegfn_form} can be rewritten in a close-form expression, expressed in terms of the Barnes-$G$ functions, as seen in the following:
\begin{align}
	\left|\FF^{z}\right|^2
	&=
	\frac{2}{M^2 G^4\left(\frac{1}{2}\right)}
	\prod_{\sigma=\pm}
	\frac{%
	G(\frac{\hle_{2}-\hle_{1}}{2i\sigma})
	G(1+\frac{\hle_{2}-\hle_{1}}{2i\sigma})
	}{%
	G(\frac{1}{2}+\frac{\hle_{2}-\hle_{1}}{2i\sigma})
	G(\frac{3}{2}+\frac{\hle_{2}-\hle_{1}}{2i\sigma})
	}
	.
	\label{2sp_ff_result}
\end{align}
We can check using the integral representations for the $G$ function given in \cref{barnes_log_Gfn_intrep}, that the result \eqref{2sp_ff_result} is the same as the result obtained earlier in the $q$-VOA framework:
\begin{align}
	|\FF^z|^2
	&=
	\frac{2e^{-I(\hle_2-\hle_1)}}{M^2}
	&
	\text{where,}
	\quad
	I(\nu)
	=
	\int_{0}^{\infty}
	\frac{dt\, e^t}{t}
	\frac{%
	\cos(2\nu t)\cosh2t-1
	}{%
	\cosh t\sinh 2t
	}
\end{align}
It was originally found in this form by \textcite{BouCK96} for the two-spinon form-factors. There they used a method which is based on $q$-vertex operator algebra \cite{JimM95}.
The fact that we are able to compare the results obtained from completely different baseline approaches is a very important step forward in the direction of ABA based form-factor approach.
The $q$-vertex operator algebra method has also been used to compute the exact expressions for four-spinon form-factor by \textcite{CauH06,AbaBS97} and to the massless XXZ chains by \textcite{CauKSW12}.
The natural question that arises is ``Whether and how can we generlise the method presented in this chapter to higher spinon sectors, or to anisotropic regimes?''.
We try to answer these questions for the higher spinon sectors over the next two \cref{chap:cau_det_rep_gen,chap:gen_FF}.
\clearpage{}%
\clearpage{}%
\chapter[Modified Cauchy determinant representation]{Modified Cauchy determinant representation for higher form-factors}
\label{chap:cau_det_rep_gen}
In this chapter we take a first step towards generalisation of our method developed in the last chapter to form-factors for triplet ($s=1$) excitations in the spinon sectors $n_h>2$.
\index{ff@\textbf{Form-factors}!FF@$\FF^z$: longitudinal form-factor}%
\begin{align}
	|\FF^{z}|^2=
	\frac{|\braket{\psi_{g}|\sigma^{3}_{m}|\psi_{1}^{1}(\bm\hle)}|^2%
	}{%
	\braket{\psi_{g}|\psi_{g}}
	\braket{\psi_{1}^{1}(\bm\hle)|\psi_{1}^{1}(\bm\hle)}
	}
	.
	\label{long_ff_scal_prod_gen}
\end{align}
We will find that the Cauchy determinant representation obtained for two-spinon case ($n_h=2$) \eqref{cau_det_rep_2sp} in \cref{chap:2sp_ff} can be generalised to any excitation with even spinon number $n_h$, which is small enough $n_h<<M$ to satisfy the low-lying criteria.
\\
Although we will use the same method of the Gaudin extraction to compute the ratio of determinants,
unlike the previous case of two-spinon form-factors, here we must also account for the complex Bethe roots, that are always present in the case of triplet excitations of higher spinon sectors $n_h>2$.
In this regard,
let us recall that we will be using the prescription of \textcite{DesL82} to describe the nature of complex Bethe roots in the thermodynamic limit.
In this picture, complex roots are classified into close-pairs and wide-pairs: $\bm\mu = \bm\rl \bm\cup \bmclp<+> \bm\cup \bmclp<-> \bm\cup \bmwdp<+> \bm\cup \bmwdp*<->$ [see \cref{clp_DL,wdp_DL}].
Let us also recall that complex roots are determined by set of higher-level roots $\bm\cid$ ($n_{\bm\cid}=\ho{n}$), which consists of the centres of the close-pairs and anchors of the wide-pairs: $\bm\cid= \bmclp \bm\cup \bmwdp \bm\cup \bmwdp*$. We also know that the set of higher-level roots satisfies an inhomogeneous version of the Bethe equations, called the higher-level Bethe equations \eqref{hl_bae}.
\\
The real roots are denoted by the set $\bm\rl$ while the set of all real roots including the holes is denoted by $\bm\rh$.
All the parameters will be ordered in the ascending order of the their union, so that the holes $\bm\hle$ appear at the end positions.
\par
Once again, we will use the $\mathfrak{su}_2$ symmetry to recast the longitudinal form-factor \eqref{long_ff_scal_prod_gen} in the transverse channel through \cref{ff_transmap}.
Therefore, the starting point remains the representation in terms of ratio of determinants, which is similar to \eqref{det_rep_fini_ff}.
However, there is one important difference due to the presence of complex roots.
When it comes to the close-pairs in particular, we see that due to the formation \eqref{clp_DL} of the close-pairs into 2-string or quartets, there is a singular term in the prefactors of \cref{det_rep_fini_ff}, as one the Baxter polynomial $\baxq_e$ vanishes in the numerator.
For this reason, we find it more appropriate to separate the vanishing terms early-on, which will be written in the terms of string deviation parameters $\bm\stdv$. Therefore, we rewrite \cref{long_ff_scal_prod_gen} in a generic setting as
\begin{align}
	|\FF^{z}|^2&=
	-2
	\bmprod\frac{%
	\baxq_{g}(\bm\mu-i)
	}{%
	\baxq'_{e}(\bm\mu-i)
	}
	\bmprod\frac{%
	\baxq_{e}(\bm\la-i)
	}{%
	\baxq_{g}(\bm\la-i)
	}
	\frac{1}{\bmprod{(2i\bm\stdv)}}
	\frac{%
	\det\nolimits_{N_0}\Mcal\left[\bm{\check\mu}\big\Vert\bm\la\right]
	}{%
	\det\nolimits_{N_0}\Ncal\left[\bm\la\big\Vert\bm\la\right]
	}
	\frac{%
	\det\nolimits_{N_0+1}\Mcal^{(2)}\left[\bm{\check\la}\big\Vert\bm\mu\right]
	}{%
	\det\nolimits_{N_0-1}\Ncal\left[\bm\mu\big\Vert\bm\mu\right]
	}
	.
	\label{det_rep_fini_ff_gen}
\end{align}
For the Gaudin extraction, once again we define the matrices $\Fmat[g]$ and $\Fmat[e]$ that are obtained from the action of the inverse Gaudin matrices, as follows: 
\index{ff@\textbf{Form-factors}!mat FF gs@$\Fmat[g]$: matrix obtained from the extraction of the ground state Gaudin matrix|textbf}%
\index{ff@\textbf{Form-factors}!mat FF es@$\Fmat[e]$: matrix obtained from the extraction of an excited state Gaudin matrix|textbf}%
\begin{subequations}
\begin{align}
	\Fmat[g]&=
	\Ncal^{-1}[\bm\la\Vert\bm\la]
	\cdot
	\left(\Mcal\left[\bm{\check\mu}\big\Vert\bm\la\right]\right)^T
	,
	\label{gau_ex_I_mat_def_generic}
	\shortintertext{and,}
	\Fmat[e]&=
	\diag\left[\Ncal^{-1}[\bm\mu\Vert\bm\mu]~\Big\vert~\Id_{2}\right]
	\left(\Mcal^{(2)}\big[\bm{\check\la}\big\Vert\bm\mu\big]\right)^T
	\label{gau_ex_mat_II_def_generic}
	.
\end{align}
	\label{gau_ex_mat_def_both_generic}
\end{subequations}
The diagonal embedding of the inverse Gaudin matrix in \cref{gau_ex_mat_II_def_generic} is exactly same as in \cref{gau_mat_diag_immersion}.
Computation for the most of the columns of the matrices $\Fmat[g]$ and $\Fmat[e]$ plays out on almost similar lines as it was in \cref{chap:2sp_ff} for the two-spinon case, since the significant number of Bethe roots are real for a low-lying excited state.
We will not repeat the similar part again, but we will elaborate on the changes brought out by the presence of complex roots.
In the case of the matrix $\Fmat[e]$, our computation leads to an interesting result: we see that there is an emergence of \emph{higher-level} Gaudin matrix whose extraction is embedded in matrix $\Fmat[e]$. 
The higher-level Gaudin matrix is given by an expression that is similar to \cref{gau_mat_gen}, except that we replace the counting function and Bethe roots by their higher-level counter-parts.
\index{ff@\textbf{Form-factors}!mat Gau hl@\hspace{1em}$\Ho{\Ncal}[\cdot\Vert\cdot]$: higher-level equivalent of \rule{3em}{1pt}}%
\begin{align}
	\Ho{N}_{a,b}
	=
	\aux*'(\cid_a)
	\delta_{a,b}
	-
	2\pi i
	K(\cid_a-\cid_b)
	.
	\label{hl_gaudin_intro_sec}
\end{align}	
We will show that the \emph{higher-level} version of the Gaudin extraction \eqref{gau_ex_mat_def_both_generic} can be expressed as follows:
\begin{align}
	\ho{\Scal}=
	\Ho{\Ncal}^{-1}
	\cdot
	\Ho{\Rcal}
	.
	\label{hl_ff_mat_intro_sec}
\end{align}
The resulting matrix $\ho\Scal$ describes a \emph{higher-level} block inside the matrix $\Fmat[e]$ for the form-factors.
The emergence of this \emph{higher-level structure} for the form-factor is one of the strong result in our computations. It can be compared with the emergence of the higher-level Bethe equations found by \textcite{DesL82} and \textcite{BabVV83} in the case of spectrum.
\par
Finally, let us remark that whenever we have a close-pair among the complex roots, the deviation parameters $\bm\stdv$ plays an important role of regularising the intermediate expressions. 
We will see over the course of our computations, that the determinants also becomes singular in the case of close-pairs, in such a way that it cancels out with string-deviation terms in the prefactors \eqref{det_rep_fini_ff_gen}.
At an appropriate stage in the computations, we will redefine the matrices $\Fcal$ in such a manner that leads to the cancellation of the deviation parameters $\bm\stdv$.
\section{Gaudin extraction of first type}
\label{sec:gau_ex_I_gen}%
We can write the system of linear equations for the extraction of the Gaudin matrix of the ground state \eqref{gau_ex_I_mat_def}, that is identical to \cref{gau_ex_syslin} obtained in \cref{chap:2sp_ff} for the two-spinon case.
However, hidden in the notations, we see that the set $\bm{\check{\mu}}$ also contains complex roots in the form of close-pairs and wide-pairs.
Let us reorder the set $\bm{\check{\mu}}$ in the ascending order of unions:
\begin{align}
	\bm{\check{\mu}}
	&=
	\bm\rl
	\bm\cup
	\Set{\tfrac{i}{2}}
	\bm\cup
	\set{\bmclp<+>+i\bm\stdv}
	\bm\cup
	\set{\bmclp<->-i\bm\stdv}
	\bm\cup
	\bmwdp<+>
	\bm\cup
	\bmwdp*<->
	.
	\label{gau_ex_I_exbr_partition}
\end{align}
Accordingly, we will partition the matrix $\Fmat[g]$ as follows:
\index{ff@\textbf{Form-factors}!mat FF gs Cau@\hspace{1em}$\Fmat[g]^{r}$: Cauchy block of \rule{3em}{1pt}|textbf}%
\index{ff@\textbf{Form-factors}!mat FF gs clp@\hspace{1em}$\Fmat[g]^{c\pm}$: close-pair blocks of \rule{3em}{1pt}|textbf}%
\index{ff@\textbf{Form-factors}!mat FF gs wdp@\hspace{1em}$\Fmat[g]^{w\pm}$: wide-pair blocks of \rule{3em}{1pt}|textbf}%
\begin{notn}
	\begin{subequations}
	\begin{align}
		\mix{\Fcal}[g]^{r}_{j,k}&=\mix{\Fcal}[g]_{j,k},	&	k\leq n_{r}+1
		;
		\\
		\mix{\Fcal}[g]^{\txtcp+}_{j,a}&=\mix{\Fcal}[g]_{j,n_{r}+a+1},	& a\leq n_\txtcp
		;
		\\
		\mix{\Fcal}[g]^{\txtcp-}_{j,a}&=\mix{\Fcal}[g]_{j,n_{r}+n_\txtcp+a+1}, & a\leq n_\txtcp
		;
		\\
		\mix{\Fcal}[g]^{\txtwp+}_{j,a}&=\mix{\Fcal}[g]_{j,n_{r}+2n_\txtcp+a+1}, & a\leq n_\txtwp
		;
		\\
		\mix{\Fcal}[g]^{\txtwp-}_{j,a}&=\mix{\Fcal}[g]_{j,n_{r}+2n_\txtcp+n_\txtwp+a+1}, & a\leq n_\txtwp
		.
	\end{align}
	\end{subequations}
	The superscripts indicates the nature of the excited state Bethe root, which appears in the columns of the Slavnov matrix in \cref{gau_ex_I_mat_def} for $\Fcal_{g}$.
	\begin{rem}
	Note that the column $\mix{\Fcal}[g]_{j,n_{r}+1}$ does not actually correspond to a real root, but it is associated with extra parameter $\check{\mu}_{n_{r}+1}=\tfrac{i}{2}$.
	Its inclusion in the block of columns for real roots is deliberate. It is motivated from the observation, which we have already seen in two-spinon case, that it admits functionally same expression in the thermodynamic limit, after the Gaudin extraction.
	\end{rem}
\end{notn}
Because the sum over the ground state roots in \cref{gau_ex_syslin} is practically unchanged, in comparison with the two-spinon case \cref{gau_ex_mat_I_comp_2sp}, 
all of the arguments made there to write \cref{Gf_init,gau_ex_sol_i2} can also be made here for the first $n_{r}+1$ columns.
Hence, we get an expression for these columns that is identical to \cref{gau_ex_sol}, as it is shown in the following:
\begin{align}
	\Fmat[g]^{r}_{j,k}
	=
	\frac{1+\aux_{g}(\check\mu_k)}{\aux'_g(\la_j)}
	\left\lbrace
	\frac{\pi}{\sinh\pi(\check\mu_k-\la_j)}
	+
	o\left(\frac{1}{M}\right)
	\right\rbrace
	.
	\label{gau_ex_I_sol_rl}
\end{align}
We will now look at action of the inverse Gaudin matrix on the close-pair and wide-pair columns.
\subsection{For close-pairs}
\label{sub:gau_ex_I_clp}
In this section, while we speak of all the close-pair roots collectively, we will use the notation $\mu^c=\set{\bmclp<+>+i\bm\stdv}\bm\cup\set{\bmclp<->-i\bm\stdv}$.
Once again, let us define the function $\Gf[g]$ on the grounds similar to \cref{Gf_init}.
\index{aux@\textbf{Auxiliary functions}!gauex gs fn@$\Gf[g]$: fn. involved in the extraction of the ground state Gaudin mat.}%
We can rewrite the system of equations \eqref{gau_ex_I_mat_def_generic} for the close-pair column in terms of the function $\Gf[g](\la,\mu^\txtcp_{a})$ as follows:
\begin{align}
	\Gf[g](\la_{j},\mu^\txtcp_{a})-2\pi i\bmsum K(\la_{j}-\bm\la)
	\frac{\Gf[g](\bm\la,\mu^\txtcp_{a})}{\aux_{g}^\prime(\bm\la)}
	=
	\aux_{g}(\mu^\txtcp_{a})\,t(\mu^\txtcp_{a}-\la_{j})
	-t(\la_{j}-\mu^\txtcp_{a})
	.
	\label{gau_ex_linsys_I_clp_Gf}
\end{align}
In the above \cref{gau_ex_linsys_I_clp_Gf}, we can first see that there is a simple pole of the function $\Gf[g](\la,\mu^\txtcp_{a})$ at $\la=\mu^\txtcp_{a}$.
We can see that close-pair roots lie sufficiently away from the real line except in an extreme scenario of the 3-string formation discussed in \cref{sub:hl_bae} that we do not consider here.
Therefore the situation here, as long as the condensation property for the summation is concerned, is very similar to the one we encountered in the case of a column for parameter $\frac{i}{2}$ in the two-spinon case. 
There we saw that in such a case it suffices to use the regular condensation property and it also applies here.
It allows us to write the following integral equation: 
\begin{align}
	\Gf[g](\la,\mu^\txtcp_{a})
	+
	\int_{\Rset}K(\la-\tau)\Gf(\tau,\mu^\txtcp_{a})d\tau
	&=
	\aux_{g}(\mu^\txtcp_{a})\,
	t(\mu^\txtcp_{a}-\la)
	-
	t(\la-\mu^\txtcp_{a})
	.
	\label{gau_ex_inteq_I_clp}
\end{align}
But in contrast to the case of parameter $\tfrac{i}{2}$, we see that the right hand side remains unaltered as we pass from \cref{gau_ex_linsys_I_clp_Gf} to \cref{gau_ex_inteq_I_clp}.
Let us also notice that the exponential counting function $\aux_g$ of the ground state is exponentially diverging or vanishing for the close-pairs.
It can be expressed in terms of the $\phifn$ function [see the \cref{defn:phifn_rat}] as
\begin{subequations}
\begin{align}
	\aux_{g}(\clp<+>_{a}+i\stdv_a)
	&=
	-
	(2i\stdv_a)
	\frac{\phifn'(\clp_a-\frac{i}{2})}{\phifn(\clp_a+\frac{3i}{2})}
	,
	\\
	\aux_{g}(\clp<->_{a}-i\stdv_a)
	&=
	-
	\frac{1}{2i\stdv_a}
	\frac{\phifn(\clp_a-\frac{3i}{2})}{\phifn'(\clp_a+\frac{i}{2})}	
	.
\end{align}
\end{subequations}
It is important to note that the divergent coefficient $\aux_g(\clp<->_a-i\stdv_a)$ is compensated in the determinant, due to the fact that the matrix $\Mcal$ is singular to the leading order in $\stdv_a$, as we have degenerate pair of columns in the Slavnov matrix $\Mcal$. It can be seen from the following:
\begin{align}
	t(\la-\clp<+>_a-i\stdv_a)=
	t(\clp<->_a-i\stdv_a-\la)
	+
	O(\stdv_a)
	.
\end{align}
This ensures that there is no divergence due to the counting function $\aux_g(\clp<->_a-i\stdv_a)$ at the level of determinant.
Therefore, in order to eliminate an apparent singular term, we can take the following recombination in the Slavnov matrix:
\begin{align}
 	\Mcal^{c-}_{j,a}(\bm{\check\mu}|\bm\la)
 	\leftarrow
 	\Mcal^{c-}_{j,a}(\bm{\check\mu}|\bm\la)
 	+
 	\aux_g(\clp<->_a-i\stdv_{a})
 	\Mcal^{c+}_{j,a}(\bm{\check\mu}|\bm\la)
 	.
 	\label{sla_clp_recomb}
\end{align} 
The substitution in \cref{sla_clp_recomb} is done silently\footref{foot:silent_kill}.
Since we know from \cref{aux_g_clp_prod_id}, that the exponential ground state counting function evaluated on close-pair roots, are reciprocate with respect to each other in the thermodynamic limit, since we have
\begin{align}
	\aux_g(\clp<+>_a)\aux_g(\clp<->_a)=1
	,
\end{align}
we can see that in the new expression obtained through \cref{sla_clp_recomb}, half the number of columns for the close-pairs are given by the their components:
\begin{subequations}
\begin{multline}
	\Mcal^{c-}_{j,a}=
	\aux_g(\clp<->_a-i\stdv_a)
	\left[
	t(\clp<->_a-i\stdv_a-\la_j)
	-
	t(\la_j-\clp<+>_a-i\stdv_a)
	\right]
	\\
	+
	t(\clp<+>_a+i\stdv_a-\la_j)
	-
	t(\la_j-\clp<->_a+i\stdv_a)
	+
	O(\stdv_a)
	.
	\label{gau_ex_clp_-_recomb_t}
\end{multline}
We can drop the deviation parameters in the last two regular $t$ terms, where they can be recombined to write \cref{gau_ex_clp_-_recomb_t} as follows:
\begin{multline}
	\Mcal^{c-}_{j,a}=
	\aux_g(\clp<->_a-i\stdv_a)
	\left[
	t(\clp<->_a-i\stdv_a-\la_j)
	-
	t(\la_j-\clp<+>_a-i\stdv_a)
	\right]
	\\
	+
	2\pi i
	\big\lbrace
	K(\la_j-\clp<+>_a)
	-
	K(\la_j-\clp<->_a)
	\big\rbrace
	+
	O(\stdv_a)
	.	
	\label{gau_ex_clp_-_recomb_K}
\end{multline}
\end{subequations}
Similarly since $\aux_g(\clp<+>_a+i\stdv_a)$ is vanishing and its coefficient term in $\Mcal^{c+}$ has been already taken into account through the recombination \eqref{sla_clp_recomb}, we will silently\footnotemark drop it to write
\footnotetext{i.e., a change or transformation that is done without changing the original notation\label{foot:silent_kill}}
\begin{align}
	\Mcal^{c+}_{j,a}=
	-t(\la_j-\clp<+>_a)
	+O(\stdv_a)
	.
	\label{gau_ex_clp_+_leading}
\end{align}
With all the rearrangements shown in \cref{gau_ex_clp_+_leading,gau_ex_clp_-_recomb_K}, we can rewrite the integral equations \eqref{gau_ex_inteq_I_clp} for the blocks $\Fmat[g]^{c\pm}$ as follows:
\begin{subequations}
\begin{align}
	\Gf[g](\la,\clp<+>_a)
	+
	\int_{\Rset}
	K(\la-\tau)
	\Gf[g](\la,\clp<+>)
	d\tau
	=
	-t(\la-\clp<+>)
	+
	O(\stdv_a)
	\label{gau_ex_inteq_I_clp+_leading}
\end{align}
and
\begin{multline}
	\Gf[g](\la,\clp<->_a)
	+
	\int_{\Rset}
	K(\la-\tau)
	\Gf[g](\la,\clp<->)
	d\tau
	=	
	\aux_g(\clp<->_a-i\stdv_a)
	\left[
	t(\clp<->_a-i\stdv_a-\la_j)
	-
	t(\la_j-\clp<+>_a-i\stdv_a)
	\right]
	\\
	+
	2\pi i
	\big\lbrace
	K(\la_j-\clp<+>_a)
	-
	K(\la_j-\clp<->_a)
	\big\rbrace
	+
	O(\stdv_a)
	.
	\label{gau_ex_inteq_clp_-_recomb}
\end{multline}
\end{subequations}
The solution of the integral \cref{gau_ex_inteq_I_clp+_leading,gau_ex_inteq_clp_-_recomb} can be obtained from the generalised version of the Lieb equation that is studied in \cref{chap:den_int_aux}. 
From its solution, we can write the components of the blocks $\Fmat[g]^{c\pm}$ of all the close-pair columns as shown in the following:
\begin{subequations}
\label{gau_ex_I_sol_clp}
\begin{align}
	\Fmat[g]^{c+}_{j,a}=
	\frac{1}{\aux'_g(\la_j)}
	\left\lbrace
	\frac{\pi}{\sinh(\clp<->_a-\la_j)}
	+
	o\left(\frac{1}{M}\right)
	\right\rbrace
	\label{gau_ex_I_sol_clp+}
\end{align}
and
\begin{multline}
	\Fmat[g]^{c-}_{j,a}
	=
	\frac{\aux_g(\clp<->-i\stdv_a)}{\aux'_g(\la_j)}
	\left\lbrace
	\frac{\pi}{\sinh\pi(\clp<+>_a+i\stdv_a-\la_j)}
	+
	\frac{\pi}{\sinh\pi(\clp<->_a-i\stdv_a-\la_j)}
	\right\rbrace
	\\
	+
	\frac{2\pi i}{\aux'_g(\la_j)}
	\left\lbrace
	\rden_1(\la_j,\clp<+>_a)
	-
	\rden_1(\la_j,\clp<->_a)
	+o\left(\frac{1}{M}\right)
	\right\rbrace
	.
	\label{gau_ex_I_sol_clp-}
\end{multline}
\end{subequations}
\subsection{For wide-pairs}
\label{sub:gau_ex_I_wdp}
Let us denote the complete set of wide-pairs with $\bm\mu^{w}=\bmwdp<+>\bm\cup\bmwdp*<->$.
From the system of equations for the extraction of the Gaudin matrix of the ground state \eqref{gau_ex_I_mat_def} can be rewritten in terms of the function $\Gf[g]$ as
\begin{align}
	\Gf[g](\la_{j},\mu^\txtwp_{a})-2\pi i\bmsum K(\la_{j}-\bm\la)
	\frac{\Gf[g](\bm\la,\mu^\txtwp_{a})}{\aux_{g}^\prime(\bm\la)}
	=
	\aux_{g}(\mu^\txtwp_{a})\,t(\mu^\txtwp_{a}-\la_{j})
	-t(\la_{j}-\mu^\txtwp_{a})
	.
	\label{gau_ex_linsys_I_wdp}
\end{align}
We can see that there are no poles of $\Gf[g](\la,\mu^{w}_a)$ on the real line. Hence, we can use the \emph{regular} condensation property on \cref{gau_ex_linsys_I_wdp} to obtain an integral equation for the function $\Gf[g]$, 
\begin{align}
	\Gf[g](\la,\mu^\txtwp_{a})
	+
	\int_{\Rset}K(\la-\tau)\Gf(\tau,\mu^\txtwp_{a})d\tau
	&=
	\aux_{g}(\mu^\txtwp_{a})\,
	t(\mu^\txtwp_{a}-\la)
	-
	t(\la-\mu^\txtwp_{a})
	.
	\label{gau_ex_inteq_I_wdp}
\end{align}
We have shown earlier in \cref{wdp_aux_gs_const} in \cref{chap:spectre}, that the exponential counting function is constant in the region $|\Im\la|>1$ occupied by the wide-pairs.
It is a direct consequence of a difference in the choice of branch cuts for the wide-pairs, in contrast to the close-pairs or the real roots,
that allows us to write,
\begin{align}
	\aux_{g}(\mu^\txtwp_{a})=1
	.
	\label{aux_wdp_gau_ex_const}
\end{align}
Therefore, using the relation \eqref{aux_wdp_gau_ex_const}, we can rewrite the integral \cref{gau_ex_inteq_I_wdp} as shown in the following expressions.
For all the terms due to wide-pairs in the positive half of the complex plane we get,
\begin{subequations}
\label{gau_ex_inteq_I_wdp_all}
\begin{align}
	\Gf[g](\la,\wdp<+>_{a})
	+
	\int_{\Rset}K(\la-\tau)\Gf(\tau,\wdp<+>_{a})d\tau
	&=
	K_{2}(\la,\wdp_{a}+i)
	-
	K_{2}(\la,\wdp_{a})
	\label{gau_ex_inteq_I_wdp+}
\shortintertext{while for the terms due to the wide-pairs in the negative half we get,}
	\Gf[g](\la,\wdp*<->_{a})
	+
	\int_{\Rset}K(\la-\tau)\Gf(\tau,\wdp*<->_{a})d\tau
	&=
	K_{2}(\la,\wdp*_{a})
	-
	K_{2}(\la,\wdp*_{a}-i)
	.
	\label{gau_ex_inteq_I_wdp-}
\end{align}
\end{subequations}
By comparing these integral equations to the generalised Lieb \cref{int_eq_shft_scld} studied in \cref{chap:den_int_aux}, we find that solutions can be written in the following form:
\begin{subequations}
\label{gau_ex_I_wdp_den_decomp}
\begin{align}
	\Gf[g](\la,\wdp<+>_{a})&=
	2\pi i\,\rden_{2}(\la,\wdp+i)-2\pi i\,\rden_{2}(\la,\wdp)
	,
	\label{gau_ex_I_wdp+_den_decomp}
	\\
	\Gf[g](\la,\wdp*<->_{a})&=
	2\pi i\,\rden_{2}(\la,\wdp*)-2\pi i\,\rden_{2}(\la,\wdp*-i)
	\label{gau_ex_I_wdp-_den_decomp}
	.
\end{align}
\end{subequations}
Note that in the above density functions $\rden_2$, the imaginary value of the shift always fall in the outer region. In this region, as we can see from \cref{lieb_den_shftd_outside}, the density function $\rden_2$ can be expressed in terms of the digamma\footnotemark functions as follows:
\footnotetext{the logarithmic derivative of the Gamma functions, see also \cref{chap:spl_fns}}
\begin{subequations}
\label{gau_ex_I_wdp_den_terms}
\begin{align}
	\rden_{2}(\la,\wdp+i)&=
	\frac{1}{4\pi}
	\left\lbrace
	\dgamma\left(\frac{1}{4}-\frac{\la-\wdp}{2i\sigma}\right)
	-
	2\dgamma\left(\frac{3}{4}-\frac{\la-\wdp}{2i\sigma}\right)
	+
	\dgamma\left(\frac{5}{4}-\frac{\la-\wdp}{2i\sigma}\right)
	\right\rbrace
	\label{gau_ex_I_wdp+_den_T1}
	\\
	\rden_{2}(\la,\wdp)&=
	\frac{1}{4\pi}
	\left\lbrace
	\dgamma\left(-\frac{1}{4}-\frac{\la-\wdp}{2i\sigma}\right)
	-
	2\dgamma\left(\frac{1}{4}-\frac{\la-\wdp}{2i\sigma}\right)
	+
	\dgamma\left(\frac{3}{4}-\frac{\la-\wdp}{2i\sigma}\right)
	\right\rbrace
	\label{gau_ex_I_wdp+_den_T2}
\shortintertext{and}
	\rden_{2}(\la,\wdp*)&=
	\frac{1}{4\pi}
	\left\lbrace
	\dgamma\left(-\frac{1}{4}+\frac{\la-\wdp*}{2i\sigma}\right)
	-
	2\dgamma\left(\frac{1}{4}+\frac{\la-\wdp*}{2i\sigma}\right)
	+
	\dgamma\left(\frac{3}{4}+\frac{\la-\wdp*}{2i\sigma}\right)
	\right\rbrace
	\label{gau_ex_I_wdp-_den_T1}
	\\
	\rden_{2}(\la,\wdp*-i)&=
	\frac{1}{4\pi}
	\left\lbrace
	\dgamma\left(\frac{1}{4}+\frac{\la-\wdp*}{2i\sigma}\right)
	-
	2\dgamma\left(\frac{3}{4}+\frac{\la-\wdp*}{2i\sigma}\right)
	+
	\dgamma\left(\frac{5}{4}+\frac{\la-\wdp*}{2i\sigma}\right)
	\right\rbrace
	\label{gau_ex_I_wdp-_den_T2}
\end{align}
\end{subequations}
Let us now gather all the results from \cref{gau_ex_I_sol_rl,gau_ex_I_sol_clp,gau_ex_I_wdp_den_decomp} to produce
\index{ff@\textbf{Form-factors}!mat FF gs Cau@\hspace{1em}$\Fmat[g]^{r}$: Cauchy block of \rule{3em}{1pt}}%
\index{ff@\textbf{Form-factors}!mat FF gs clp@\hspace{1em}$\Fmat[g]^{c\pm}$: close-pair blocks of \rule{3em}{1pt}}%
\index{ff@\textbf{Form-factors}!mat FF gs wdp@\hspace{1em}$\Fmat[g]^{w\pm}$: wide-pair blocks of \rule{3em}{1pt}}%
\begin{subequations}
\begin{flalign}
	\Fmat[g]^{r}_{j,k}&=
	\frac{1+\aux_{g}(\check{\mu}_{k})}{\aux_{g}^\prime(\la_{j})}
	\left\lbrace
	\frac{\pi}{\sinh\pi(\check{\mu}_{k}-\la_{j})}
	+
	o\left(\frac{1}{M}\right)
	\right\rbrace
	,
	\label{gau_ex_I_gen_sol_real}
	\\
	\Fmat[g]^{\txtcp,+}_{j,a}&=
	\frac{1}{\aux'_g(\la_j)}
	\left\lbrace
	\frac{\pi}{\sinh(\clp<->_a-\la_j)}
	+
	o\left(\frac{1}{M}\right)
	\right\rbrace
	,
	\label{gau_ex_I_gen_sol_clp+}
	\\
	\Fmat[g]^{\txtcp,-}_{j,a}&=
	\begin{multlined}[t]
	\frac{\aux_g(\clp<->-i\stdv_a)}{\aux'_g(\la_j)}
	\left\lbrace
	\frac{\pi}{\sinh\pi(\clp<+>_a+i\stdv_a-\la_j)}
	+
	\frac{\pi}{\sinh\pi(\clp<->_a-i\stdv_a-\la_j)}
	\right\rbrace
	\\
	+
	\frac{2\pi i}{\aux'_g(\la_j)}
	\left\lbrace
	\rden_1(\la_j,\clp<+>_a)
	-
	\rden_1(\la_j,\clp<->_a)
	+o\left(\frac{1}{M}\right)
	\right\rbrace
	,
	\end{multlined}
	\label{gau_ex_I_gen_sol_clp-}
	\\
	\Fmat[g]^{\txtwp,+}_{j,a}&=
	\frac{2\pi i}{\aux_{g}^\prime(\la_{j})}
	\left\lbrace
	\rden_{2}(\la_{j},\wdp+i)
	-
	\rden_{2}(\la_{j},\wdp)
	+o\left(\frac{1}{M}\right)
	\right\rbrace
	,
	\label{gau_ex_I_gen_sol_wdp+}
	\\
	\Fmat[g]^{\txtwp,-}_{j,a}&=
	\frac{2\pi i}{\aux_{g}^\prime(\la_{j})}
	\left\lbrace
	\rden_{2}(\la_{j},\wdp*)
	-
	\rden_{2}(\la_{j},\wdp*-i)
	+o\left(\frac{1}{M}\right)
	\right\rbrace
	,
	\label{gau_ex_I_gen_sol_wdp-}
\end{flalign}
	\label{blocks_gau_ex_I_gen_sol}
\end{subequations}
We will now compute the components of the matrix $\Fmat[e]$ in the thermodynamic limit.
\section{Gaudin extraction of second type}
\label{sub:gau_ex_II_gen}%
Let us first divide the matrix $\Fmat[e]$ into sub-blocks according to the partitioning of the set $\bm\mu=\bm\rl\bm\cup\bmclp<+>\bm\cup\bmclp<->\bm\cup\bmwdp<+>\bm\cup\bmwdp*<->$.
The partitioning of the $\Fmat[e]$ is shown here by the following expressions:
\index{ff@\textbf{Form-factors}!mat FF es Cau@\hspace{1em}$\Fmat[e]^{r}$: Cauchy block of \rule{3em}{1pt}}%
\index{ff@\textbf{Form-factors}!mat FF es clp@\hspace{1em}$\Fmat[e]^{c\pm}$: close-pair blocks of \rule{3em}{1pt}}%
\index{ff@\textbf{Form-factors}!mat FF es wdp@\hspace{1em}$\Fmat[e]^{w\pm}$: wide-pair blocks of \rule{3em}{1pt}}%
\begin{subequations}
\begin{align}
	\Fmat[e]^{r}_{j,k}
	&=
	\Fmat[e]_{j,k}
	,
	& j&\leq n_r
	;
	\\
	\Fmat[e]^{c+}_{a,k}
	&=
	\Fmat[e]_{n_r+a,k}
	,
	& a&\leq n_c
	;
	\\
	\Fmat[e]^{c-}_{a,k}
	&=
	\Fmat[e]_{n_r+n_c+a,k}
	,
	& a&\leq n_c
	;
	\\
	\Fmat[e]^{w+}_{a,k}
	&=
	\Fmat[e]_{n_r+2n_c+a,k}
	,
	& a&\leq n_w
	;
	\\
	\Fmat[r]^{w-}_{a,k}
	&=
	\Fmat[e]_{n_r+2n_c+n_w+a,k}
	,
	& a&\leq n_w
	.
\end{align}
\label{gau_ex_II_mat_blocks}
\end{subequations}
In addition to these blocks, we will also have a Foda-Wheeler block $\Ucal(\bm{\check\la})$ of two columns, which is unaffected by the Gaudin extraction \eqref{gau_ex_mat_II_def_generic}, since it is only acted upon by the identity matrix. Therefore, it retains its original form which we have seen in \cref{gau_ex_II_FW}.
\\
From \cref{gau_ex_mat_II_def}, we can obtain the following system of linear equations for all of the blocks of $\Fmat[e]$:
\begin{align}
	\aux_{e}'(\mu_{j})\Fmat[e]_{j,k}
	-	2\pi i
	\sum_{l=1}^{N_1}
	K(\mu_{j}-\mu_{l})
	\Fmat[e]_{l,k}
	&=
	\aux_{e}(\mu_{k})\,t(\check{\la}_{k}-\mu_{j})
	-
	t({\mu}_{j}-\check{\la}_{k})
	.
	\label{gau_ex_syslin_II_gen}
\end{align}
We will define the function $\Gf[e]$ on the same grounds as \cref{Gf_init} whenever possible.
However, the fact that we can have complex roots in the current setting forces us to reformulate our approach in certain cases.
The change here is more significant in comparison to extraction of the first type \eqref{sec:gau_ex_I_gen}, because the Gaudin extraction itself contains the complex roots.
Since the counting function is contained into the definition of $\Gf[e]$, we need to distinguish between the real, close-pair and wide-pair rows of the inverse matrix when it comes to the definition auxiliary function such as $\Gf[e]$.
\\
For the real rows, the counting function $\cfn$ is real and there is no exponential divergence in the exponential counting function $\aux_e$. 
Its derivative is given by the density function and the corrections are sub-leading for the values of real Bethe root $\rl_a$ from the bulk, thus we can write
\begin{align}
	\aux_{e}^\prime(\rl_a)&=
	2\pi i\, M\, \rden_{e}(\rl_a) + O(1)
	.
\end{align}
Here we shall define the function $\Gf[e]$ exactly as it was defined in \cref{Gf_init} for the two-spinon form-factor.
\begin{defn}
The meromorphic function $\Gf[e](\nu,\la)$ satisfying the initial condition:
\index{aux@\textbf{Auxiliary functions}!gauex es fn@$\Gf[e]$: fn. involved in the extraction of an excited state Gaudin mat.|textbf}%
\begin{align}
	\Gf[e](\rl_{j},\la_{k})&=
	\aux_{e}^\prime(\rl_{j})\Fmat[e]_{j,k}^r
	.
	\label{Gf_init_II_gen}
\end{align}
It is not hard to see that $\Gf[e](\nu,\la_k)$ has a simple pole on the real line with the residue:
\begin{align}
	\res_{\nu=\la_k}
	\Gf[e](\nu,\la_k)
	=
	-(1+\aux_e(\la_k)).
	\label{res_Gf_init_II_gen}
\end{align}
\end{defn}
The system of linear equations given in \cref{gau_ex_syslin_II_gen} for the block $\Fmat[e]^r$ can be written in terms of the function $\Gf[e]$ as
\begin{multline}
	\Gf[e](\nu,\la_{k})
	+
	\int_{\Rset+i\epsilon}
	K(\nu-\tau)
	\Gf[e](\tau,\check\la_{k})
	d\tau
	=
	(1+\aux_{e}(\check\la_{k}))\,
	t(\check\la_{k}-\nu)
	-2\pi i
	\bmsum\frac{%
	K(\nu-\bm\hle)%
	}{%
	\aux_{e}^\prime(\bm\hle)%
	}\Gf[e](\bm\hle,\check\la_k)%
	\\%
	+2\pi i\sum_{a=1}^{n_\txtcp}%
	\left\lbrace%
	K(\nu-\clp<+>_{a}+i\stdv_{a})\Fmat[e]^{\txtcp+}_{a,k}%
	+K(\nu-\clp<->_{a}-i\stdv_{a})\Fmat[e]^{\txtcp-}_{a,k}%
	\right\rbrace
	\\
	+2\pi i\sum_{a=1}^{n_\txtwp}%
	\left\lbrace%
	K(\nu-\wdp<+>_{a})\Fmat[e]^{\txtwp+}_{a,k}%
	+K(\nu-\wdp<->_{a})\Fmat[e]^{\txtwp-}_{a,k}%
	\right\rbrace%
	.
	\label{gau_ex_II_inteq_rl}
\end{multline}
The integrals for the finite-size correction terms that would normally appear in \cref{gau_ex_II_inteq_rl} are identical to those appearing in \cref{gau_ex_inteq_II} which was written in the context of two-spinon form-factors. Here, they will be omitted right from the beginning \eqref{gau_ex_II_inteq_rl} and we will always assume that the corrections are of order $o\left(\frac{1}{M}\right)$ for the roots $\la_k$ taken in the bulk; while in the extreme cases outside the bulk, they are believed to add a negligible corrections to the final result.
\\
Note that we have included the terms for an extra parameter $\check\la_{N_0+1}=\frac{i}{2}$ in the same integral \cref{gau_ex_II_inteq_rl}.
Let us now compare this integral equation to the generalised Lieb \cref{int_eq_shft_scld} studied in \cref{chap:den_int_aux}.
We can see that the solution for the function $\Gf[e]$ can be written in terms of the generalised density functions $\rden_\kappa$ as follows:
\begin{multline}
	\Gf[g](\nu,\check\la_k)
	=
	2\pi i (1+\aux_{e}(\check\la_{k}))\rden_{2}(\nu,\check\la_{k}+\tfrac{i}{2}-i\epsilon)
	-
	2\pi i \bmsum\frac{%
	\rden_{1}(\nu,\bm\hle)
	}{%
	\aux_{e}^\prime(\bm\hle)
	}\Gf[e](\bm\hle,\check\la_{k})
	\\
	+
	2\pi i\sum_{a=1}^{n_c}
	\left\lbrace
	\rden_{1}(\nu,\clp<+>_{a}+i\stdv_{a})\Fmat[e]^{\txtcp,+}_{a,k}
	+
	\rden_{1}(\nu,\clp<->_{a}-i\stdv_{a})\Fmat[e]^{\txtcp,-}_{a,k}
	\right\rbrace
	\\
	+
	2\pi i\sum_{a=1}^{n_w}
	\left\lbrace
	\rden_{1}(\nu,\wdp<+>_{a})\Fmat[e]^{\txtwp+}_{a,k}
	+
	\rden_{1}(\nu,\wdp<->_{a})\Fmat[e]^{\txtwp-}_{a,k}
	\right\rbrace
	.
	\label{gau_ex_II_gen_Gf_den_expn_nostr}
\end{multline}
The first line in this expression is common with \cref{gau_ex_II_sol_decomp}, while the new terms are expressed as a linear combination of the columns for close-pairs $\Fmat[e]^{c\pm}$ and wide-pairs $\Fmat[e]^{w\pm}$.
Let us recall the behaviour of close-pairs in the case of spectrum, where it was forced to form the 2-strings or quartets due to the presence of singularities.
The same can be expected here and we will see that the close-pair blocks can be recombined to produce simpler terms.
In this regard \cref{gau_ex_II_gen_Gf_den_expn_nostr} is only an intermediate expression. We will return to it after we have dealt with the close-pair blocks $\Fmat[e]^{c\pm}$.
\subsection{For close-pairs}
\label{sub:gau_ex_II_clp}
In the case of close-pairs, the counting function $\cfn_e$ has non-real part which makes the exponential counting function singular. 
However, at the same time, since the close-pairs are among the roots of the Bethe equation, we also have the condition: $\aux_e(\clp<\pm>_a)=-1$.
As we have seen in \cref{sub:DL_picture}, this is assured by the formation of close-pairs into either 2-string or quartets \eqref{clp_DL} in the thermodynamic limit.
An unavoidable consequence of their formation is that the derivative $\aux_e(\clp<\pm>_a\pm i\stdv_a)$ becomes singular. This singularity can be seen as a pole in the string deviation parameter $\stdv_a$.
We can also see from the following expansion:
\begin{align}
	\aux_{e}^\prime(\clp<\pm>_{a}\pm i\stdv_{a})
	&=
	-2\pi i\cfn_{e}^\prime(\clp<\pm>_{a}\pm i\stdv_{a})
	=
	-iM\, p_{0}^\prime(\clp<\pm>_{a}\pm i\stdv_{a})+2\pi i\bmsum K(\clp<\pm>_{a}\pm i\stdv_{a}-\bm\mu)
	\label{cfn_clp_der_expn}
\end{align}
that the pole in $i\stdv_a$ is contained in only one of the terms inside the sum, in \cref{cfn_clp_der_expn}.
Let us adopt the following notation for the inverse string deviation parameter:
\begin{align}
	\invstdv_{a}&=
	\frac{1}{2i\stdv_{a}}
	.
	\label{def_inv_stdv_param}
\end{align}
Then we can separate the derivative of the counting function $\aux'_e(\clp<\pm>\pm i\stdv_a)$ in \cref{cfn_clp_der_expn}, into the regular part in $\stdv_a$ and its pole in the parameter $\stdv_a$, as follows:
\begin{align}
	\aux_{e}^\prime(\clp<\pm>\pm i\stdv)&=
	\reg(\aux_{e}^\prime(\clp<\pm>))-\invstdv_a
	.
\end{align}
From this stage, let us take the limit $\stdv_a\to 0$ in the regular parts of the expressions.
Note that the pole in $\stdv_a$ has same parity in both values of the parameters $\clp<\pm>_a$, since the function $K$ is even.
Therefore we can write,
\begin{align}
	2\pi i\,K(\clp<+>_{a}+i\stdv_{a}-\clp<->_{a}+i\stdv_{a})
	=\invstdv_a-\frac{1}{2i+2i\stdv_{a}}
	.
\end{align}
Let us also note that this type of pole is present in the sum over function $K$ in the expression \eqref{gau_ex_syslin_II_gen}, seen in the case of extraction for close-pair columns.
We can also see that these two terms are the only sources of pole in $\invstdv_a$ in this expression.\footnote{here onwards, we shall drop the string deviation terms for from all the regular terms in the expressions starting from \cref{gau_ex_syslin_II_clp}
\label{foot:drop_stdv_clp_gau_ex}
}
Combining all the coefficients of $\invstdv_a$, we can write the system of linear equations \eqref{gau_ex_syslin_II_gen} in the case of close-pairs as follows:
\begin{subequations}
\begin{multline}
	\reg(\aux_{e}^\prime(\clp<+>_{a}))\,
	\Fmat[e]^{\txtcp,+}_{a,k}
	+ \invstdv_{a}	(\Fmat[e]^{\txtcp+}_{a,k} - \Fmat[e]^{\txtcp-}_{a,k} )
	\\
	\begin{aligned}
	&-2\pi i\bmsum K(\clp<+>_{a}-\bm\rl)
	\frac{%
	\Gf[e](\bm\rl,\la_{k})
	}{%
	\aux_{e}^\prime(\bm\rl)
	}
	\\
	&\quad
	-2\pi i\sum_{b=1}^{n_\txtcp} K(\clp_{a}-\clp_{b})	\Fmat[e]_{b,k}^{\txtcp+}
	-2\pi i\sum_{\underset{b\neq a}{b=1}}^{n_\txtcp} K(\clp_{a}-\clp_{b}+i)	\Fmat[e]_{b,k}^{\txtcp-} - \frac{1}{2i} \Fmat[e]^{\txtcp-}_{a,k}
	\\
	&\qquad
	-2\pi i\sum_{b=1}^{n_\txtwp} K(\clp_{a}-\wdp_{b})	\Fmat[e]_{b,k}^{\txtwp+}
	-2\pi i\sum_{b=1}^{n_\txtwp} K(\clp_{a}-\wdp*_{b}+i)	\Fmat[e]_{b,k}^{\txtwp-}
	\end{aligned}
	\\
	=
	(1+\aux_{e}(\la_{k}))\,
	t(\la_{k}-\clp<+>_{a})
	-
	2\pi i K(\clp<+>_{a}-\la_{k})
	.
	\label{gau_ex_syslin_II_clp+}
\end{multline}
And,
\begin{multline}
	\reg(\aux_{e}^\prime(\clp<->_{a}))\,
	\Fmat[e]^{\txtcp,+}_{a,k}
	- \invstdv_{a}	(\Fmat[e]^{\txtcp+}_{a,k} - \Fmat[e]^{\txtcp-}_{a,k} )
	\\
	\begin{aligned}
	&-2\pi i\bmsum K(\clp<->_{a}-\bm\rl)
	\frac{%
	\Gf[e](\bm\rl,\la_{k})
	}{%
	\aux_{e}^\prime(\bm\rl)
	}
	\\
	&\quad
	-2\pi i\sum_{\underset{b\neq a}{b=1}}^{n_\txtcp} K(\clp_{a}-\clp_{b}-i)	\Fmat[e]_{b,k}^{\txtcp+} - \frac{1}{2i} \Fmat[e]^{\txtcp+}_{a,k}
	-2\pi i\sum_{b=1}^{n_\txtcp} K(\clp_{a}-\clp_{b})	\Fmat[e]_{b,k}^{\txtcp-}
	\\
	&\qquad
	-2\pi i\sum_{b=1}^{n_\txtwp} K(\clp_{a}-\wdp_{b}-i)	\Fmat[e]_{b,k}^{\txtwp+}
	-2\pi i\sum_{b=1}^{n_\txtwp} K(\clp_{a}-\wdp*_{b})	\Fmat[e]_{b,k}^{\txtwp-}
	\end{aligned}
	\\
	=
	(1+\aux_{e}(\la_{k}))\,
	t(\la_{k}-\clp<->_{a})
	-
	2\pi i
	K(\clp<->_{a}-\la_{k})
	.
	\label{gau_ex_syslin_II_clp-}
\end{multline}
	\label{gau_ex_syslin_II_clp}
\end{subequations}
Since there are no other poles in the string deviation parameters\footref{foot:drop_stdv_clp_gau_ex}, we can conclude that close-pair rows coincide pairwise up-to the vanishing order $O(\stdv_a)$ in the string deviation parameter.
\begin{align}
	\Fmat[e]^{\txtcp+}_{a,k} = \Fmat[e]^{\txtcp-}+O(\stdv_{a})
	.
	\label{gau_ex_coin_clp_cols}
\end{align}
This prompts us to define the new rows for the close-pairs obtained by the combining the inital close-pair rows pairwise.
\begin{defn}
The new block of close-pair rows $\Fmat[e]^{c}$ is obtained by the combination
\index{ff@\textbf{Form-factors}!mat FF es clp recomb@\hspace{1em}$\Fmat[g]^{c}$: recombined close pair block of \rule{3em}{1pt}}%
\begin{align}
	\Fmat[e]^{c}_{a,k}=
	\kappa_a(\Fmat[e]^{\txtcp+}_{a,k}-\Fmat[e]^{\txtcp-}_{a,k})
	.
	\label{clp_diff_block_gau_ex_II}
\end{align}
It is worthwhile to note that we have brought in the $\invstdv_a$ from the prefactor of \cref{det_rep_fini_ff_gen} to scale the recombined column.
This scaling regularises simultaneously the prefactor and the matrix $\Fmat[e]$ for each of the close-pair.
\end{defn}
Let us also rename the remaining blocks of the close-pair $\Fmat[e]^{c-}$ and wide-pairs $\Fmat[e]^{w\pm}$ into collective block by adopting the following notation.
\begin{notn}[Higher-level block]
The block matrix $\Fmat*[e]$ is used to denote collectively the remaining close-pair block and wide-pair blocks as
\index{ff@\textbf{Form-factors}!mat FF hl@$\Fmat*[e]$: higher-level block inside the matrix $\Fmat[e]$|textbf}%
\index{ff@\textbf{Form-factors}!mat FF hl clp@\hspace{1em}$\Fmat*[e]^c$: close-pair sub-block of \rule{3em}{1pt}}%
\index{ff@\textbf{Form-factors}!mat FF hl wdp@\hspace{1em}$\Fmat*[e]^{w\pm}$: wide-pair sub-blocks of \rule{3em}{1pt}}%
\begin{subequations}
\begin{align}
\Fmat*[e]_{a,k}&= \Fmat^{\txtcp-}_{a,k},
\\
\Fmat*[e]_{N_\txtcp+a,k}&=\Fmat^{\txtwp+}_{a,k},
\\
\Fmat*[e]_{N_\txtcp+N_\txtwp+a,k}&=\Fmat^{\txtwp-}_{a,k}.
\end{align}
\label{gau_ex_II_hl_block}
\end{subequations}
Note that the block $\Fmat*[e]$ contains $\ho{n}=n_c+2n_w$ number of columns.
Based on this observation, %
one can draw the following remark.
\label{def_hl_block}
\begin{rem}
The motivation behind \cref{def_hl_block} is comparable to the one behind \cref{ntn:ho_roots} for the higher-level roots $\bm\cid$, defined earlier in \cref{chap:spectre}. We expect that the close-pair and wide-pair blocks could come together to form a collective unit which is best described by a higher-level structure that emerges in the thermodynamic limit. 
Our computation in fact demonstrate the emergence of this higher-level structure and this notation is a choice made in anticipation of it.
\end{rem}
\end{notn}
Let us now add \cref{gau_ex_syslin_II_clp-,gau_ex_syslin_II_clp+} together.
We can see that this new system of equation can be written in terms of the newly defined blocks $\Fmat[e]^{c}$ and $\Fmat*[e]$ in \cref{gau_ex_II_hl_block,clp_diff_block_gau_ex_II} as
\begin{multline}
	\aux*^\prime(\clp_{a}) \Fmat*[e]_{a,k}
	\begin{aligned}[t]
	&-2\pi i 	\bmsum K_{c}(\clp_{a}-\bm\rl)
	\frac{%
	\Gf[e](\bm\rl,\check\la_{k})
	}{%
	\aux_{e}^\prime(\bm\rl)
	}
	\\
	&\quad
	-2\pi i \sum_{b=1}^{n_\txtcp}
	\left\lbrace
	2K(\clp_{a}-\clp_{b})
	+
	K_{1/2}(\clp_{a}-\clp_{b})
	\right\rbrace
	\Fmat*[e]^{\txtcp}_{b,k}
	\\
	&\qquad
	-2\pi i \sum_{b=1}^{n_\txtwp}
	\left\lbrace
	K(\clp_{a}-\wdp_{b})
	+
	K(\clp_{a}-\wdp_{b}-i)
	\right\rbrace
	\Fmat*[e]^{\txtwp+}_{b,k}
	\\
	&\qquad
	-2\pi i \sum_{b=1}^{n_\txtwp}
	\left\lbrace
	K(\clp_{a}-\wdp*_{b})
	+
	K(\clp_{a}-\wdp*_{b}+i)
	\right\rbrace
	\Fmat*[e]^{\txtwp-}_{b,k}
	\end{aligned}
	\\
	=
	2\pi i
	(1+\aux_{e}(\check\la_{k}))\,
	K(\clp<+>_\alpha-\check\la_{k})
	-
	2\pi i
	K_{c}(\clp_a-\check\la_{k})
	.
	\label{gau_ex_syslin_II_clp_hl}
\end{multline}	
Here $K_{c}$ is the combination of the Lieb kernel $K$ that is given by,
\index{misc@\textbf{Miscellaneous functions}!Lieb clp@$K_c$: a combination of Lieb ker. K, usually written for close-pairs|textbf}%
\begin{align}
	K_{c}(\nu)
	=
	K(\nu-\tfrac{i}{2})
	+
	K(\nu+\tfrac{i}{2})
	.
\end{align}
It has the Fourier transform that factorises as follows:
\begin{align}
	\what{K_c}(t)=
	e^{-\frac{|t|}{2}}(1+e^{-|t|})
	.
	\label{ft_k_sum_clp}
\end{align}
We note that the terms $\reg(\aux'_e(\clp<\pm>_a))$ from \cref{gau_ex_syslin_II_clp} were added together since we have the degeneracy \eqref{gau_ex_coin_clp_cols} in the rows.
We have seen in \cref{exc_aux_clp_prod,hlbae_clp} of \cref{chap:spectre}, that the product of the exponential counting function for the close-pairs is related with the higher-level exponential counting function as shown in the following:
\index{aux@\textbf{Auxiliary functions}!exp cfn hl@\hspace{1em}$\aux*$: higher-level equivalent of \rule{3em}{1pt}}%
\begin{align}
	\aux_{e}(\clp<+>+i\stdv)\aux_{e}(\clp<->-i\stdv)=\ho{\aux}_{e}(\clp)(1+O(\stdv))
	.
	\label{ho_aux_der_id}
\end{align}
It allows us to see that the sum of its derivatives can be written as
\begin{align}
	\ho{\aux}_{e}^\prime(\clp_{a})&=
	\reg(\ho{\aux}_{e}^\prime(\clp<+>))
	+
	\reg(\ho{\aux}_{e}^\prime(\clp<->))
	\label{aux_hl_clp_sum}
\end{align}
At the same time, due to the pairwise degeneracy \eqref{gau_ex_coin_clp_cols} of the columns in the close-pair blocks, we can now simplify \cref{gau_ex_II_gen_Gf_den_expn_nostr} by combining the sum over close-pair blocks into a single summation.
At this juncture, let us also recall that all the density terms can also be combined to write a common density term of the higher-level roots.
\index{exc@\textbf{Excitations}!condn@\textbf{- condensation}!den_hl@$\ho{\rden}$: common density function for close-pairs and wide-pairs}%
\begin{subequations}
\begin{align}
	\ho{\rden}(\nu-\clp_a)&=\rden_1(\nu,\clp<+>_a)+\rden_1(\nu,\clp<->_a)
	,
	\\
	\ho{\rden}(\nu-\wdp_a)&=\rden_1(\nu,\wdp<+>_a)
	,
	\\
	\ho{\rden}(\nu-\wdp*_a)&=\rden_1(\nu,\wdp*<->_a).
\end{align}
\end{subequations}
Together with \cref{def_hl_block}, it permits us to rewrite \cref{ex_tot_den_expn_nostr} as
\begin{multline}
	\Gf[g](\nu,\check\la_k)
	=
	2\pi i (1+\aux_{e}(\check\la_{k}))\rden_{2}(\nu,\check\la_{k}+\tfrac{i}{2}-i0)
	\\
	-
	2\pi i \bmsum\frac{%
	\rden_{1}(\nu,\bm\hle)
	}{%
	\aux_{e}^\prime(\bm\hle)
	}\Gf[e](\bm\hle,\check\la_{k})
	+
	2\pi i\sum_{a=1}^{\ho{n}}
	\ho{\rden}(\nu-\cid_a)\Fmat*[e]_{a,k}
	.
	\label{gau_ex_II_gen_Gf_den_expn_rl_str}
\end{multline}
We have already seen a similar phenomenon in \cref{den_ho_def} of \cref{chap:spectre}, thus it must be seen as a signature of the higher-level structure for the form-factors.
Let us now come back to \cref{gau_ex_syslin_II_clp_hl}.
We can see that the sum over $K_c$ in this equation can be written as integral in the thermodynamic limit using the generalised condensation property. 
It permits us to write,
\begin{subequations}
\begin{align}
	\bmsum K_{c}(\clp_{a}-\bm\rh)
	\frac{%
	\Gf[e](\bm\rh,\la_{k})
	}{%
	\aux_{e}^\prime(\bm\rh)
	}
	=
	-
	\frac{1}{2\pi i}
	\int_{\Rset+i0}
	K_{c}(\clp_{a}-\tau)
	\Gf(\tau,\la_{k})	d\tau
	+ K_{c}(\la_{k}-\clp_{a})
	.
	\label{conv_K_clp_den_gen_condn}
\end{align}
Whereas for the sum over $\Gf[e](\bm\nu,\frac{i}{2})$, it is sufficient to use the regular condensation property in order to write
\begin{align}
	\bmsum K_{c}(\clp_{a}-\bm\rh)
	\frac{%
	\Gf[e](\bm\rh,\frac{i}{2})
	}{%
	\aux_{e}^\prime(\bm\rh)
	}
	=
	-
	\frac{1}{2\pi i}
	\int_{\Rset+i0}
	K_{c}(\clp_{a}-\tau)
	\Gf(\tau,\la_{k})	d\tau
	.
	\label{conv_K_clp_den_gen_condn_i2}
\end{align}
\label{conv_K_clp_den_gen_condn_all}
\end{subequations}
The convolution integrals with the kernel $K_{c}$ for the density terms occurring in \cref{conv_K_clp_den_gen_condn,conv_K_clp_den_gen_condn_i2} are computed in \cref{sec:den_conv_for_gau_ex_k_clp} at the end of this chapter.
Here we borrow the result that was obtained in \cref{Gf_II_gen_int_with_str_kernel_append} to write down the following result:
\begin{multline}
	\bmsum K_{c}(\clp_{a}-\bm\rl)
	\frac{%
	\Gf[e](\bm\rl,\check\la_{k})
	}{%
	\aux_{e}^\prime(\bm\rl)
	}
	=
	(1+\aux_{e}(\check\la_{k})) K(\clp<+>_{a}-\check\la_{k}) + K_{c}(\clp_{a}-\check\la_{k})
	\\
	\begin{aligned}[b]
	&-\sum_{b=1}^{n_h} K_{2}(\clp_{a}-\hle_b)
	\frac{%
	\Gf[e](\hle_b,\check\la_{k})
	}{%
	\aux_{e}^\prime(\hle_b)
	}
	\\
	&\quad
	-\sum_{b=1}^{n_c}
	\left\lbrace
	K(\clp_{a}-\clp_b)
	-
	K_{1/2}(\clp_{a}-\clp_b)
	\right\rbrace
	\Fmat*[e]^{\txtcp}_{b,k}
	\\
	&\qquad
	-\sum_{b=1}^{n_w} K(\clp_{a}-\wdp_b-i)
	\Fmat*[e]^{\txtwp+}_{b,k}
	-\sum_{b=1}^{n_w} K(\clp_{a}-\wdp*_b+i)
	\Fmat*[e]^{\txtwp-}_{b,k}
	.
	\end{aligned}
\end{multline}
Substituting this result into \cref{gau_ex_syslin_II_clp_hl} we find that all the summations over the components of the higher-level block $\Fmat*[e]$ can be combined into a single sum. It lead us to the following system of linear equations:
\index{ff@\textbf{Form-factors}!mat FF hl clp@\hspace{1em}$\Fmat*[e]^c$: close-pair sub-block of \rule{3em}{1pt}}%
\begin{align}
	\aux*_{e}^\prime(\clp_{a})\Fmat*[e]^{\txtcp}_{a,k}
	-
	2\pi i\sum_{b=1}^{\ho{n}} K(\clp_a-\cid_{b})
	\Fmat*[e]_{b,k}
	=
	-2\pi i\bmsum
	K_{2}(\clp_{a}-\bm\hle)
	\frac{%
	\Gf[e](\bm\hle,\la_{k})
	}{%
	\aux_{e}^\prime(\bm\hle)
	}
	\label{hl_gau_ex_clp}
\end{align}
Note that this system only involves higher-level roots and their counting function.
It is one of the important result leading to the higher-level structure for the form-factors, to which we will come back later.
\par
Now let us turn our attention to the block difference columns $\Fmat[e]^c$ which was defined in \cref{clp_diff_block_gau_ex_II} by taking the pairwise difference of the original close-pair rows.
From \cref{gau_ex_syslin_II_clp+}, we can now write
\begin{multline}
	\Fmat[e]^{c}_{a,k}=
	(1+\aux_{e}(\check\la_{k}))\,
	t(\check\la_{k}-\clp<+>_a)
	- K(\check\la_{k}-\clp<+>_a)
	\\
	+
	2\pi i \bmsum K(\clp<+>_a-\bm\rl)
	\frac{%
	\Gf[e](\bm\rl,\check\la_{k})
	}{%
	\aux_{e}^\prime(\bm\rl)
	}
	+ 2\pi i \sum_{a=1}^{\ho{n}}
	\chi_{a}
	\Fmat*[e]_{a,k}
	\label{gau_ex_II_clp+_str}
\end{multline}
The exact form of the coefficients $\chi_{a}$ is irrelevant to us since this sum is taken over the rows which are already contained in $\Fmat*[e]$.
Some of these are already determined by the system of \cref{hl_gau_ex_clp} and the remaining columns will be determined when we treat the wide-pair case.
Let us now write the sum over real Bethe roots as an integral in the thermodynamic limit.
Since $\Gf[e](\nu,\la_k)$ has a simple pole on the real line at $\nu=\la_k$ with the known residue \eqref{res_Gf_init_II_gen}, let us use the generalised condensation property\footnotemark in \cref{gen_condn_prop} to write,
\footnotetext{%
\label{foot:bulk_clp_implied}
the bulk assumption is invoked in order to appropriate this property here, we assume that the pole $\nu=\la_k$ of the function $\Gf[e](\nu,\la_k)$ is inside the bulk of the Fermi distribution for the ground state.
}
\begin{subequations}
\begin{align}
	\bmsum K(\clp<+>_a-\bm\rh)
	\frac{%
	\Gf[e](\bm\rh,\la_{k})
	}{%
	\aux_{e}^\prime(\bm\rh)
	}
	=	
	-
	\frac{1}{2\pi i}
	\int_{\Rset+i0}
	K(\clp^{+}-\tau)
	\Gf[e](\tau,\la_{k})
	d\tau
	+
	K(\clp<+>_a-\la_k)
	.
	\label{gau_ex_clp_diff_gen_condn}
\end{align}
On the other hand for the case $\Gf[e](\bm\nu,\frac{i}{2})$ we can simply use the regular condensation of roots to write,
\begin{align}
	\bmsum K(\clp<+>_a-\bm\rh)
	\frac{%
	\Gf[e](\bm\rh,\frac{i}{2})
	}{%
	\aux_{e}^\prime(\bm\rh)
	}
	=	
	-
	\frac{1}{2\pi i}
	\int_{\Rset}
	K(\clp^{+}-\tau)
	\Gf[e](\tau,\tfrac{i}{2})
	d\tau
	.
	\label{gau_ex_clp_diff_gen_condn_i2}
\end{align}
\label{gau_ex_clp_diff_gen_condn_all}
\end{subequations}
As we have seen in \cref{gau_ex_II_gen_Gf_den_expn_rl_str}, the function $\Gf[e]$ can be expanded into different density terms.
It leads to convolutions with density terms of the shifted $K$ kernel.
Among these density terms, $\rden_2$ and $\rden_1$ are known to satisfy the integral equations:
\begin{subequations}
\begin{gather}
	\rden_{2}(\clp<+>_{a},\la_{k}+\tfrac{i}{2})+\int_{\Rset+i\alpha} K(\clp<+>_{a}-\nu) \rden_2(\nu,\la_{k}+\tfrac{i}{2})
	d\nu
	=
	\frac{1}{2\pi i}
	t(\la_{k}-\clp<+>_{a}),
\shortintertext{and}
	\rden_{1}(\clp<+>_{a},\hle_{b})+\int_{\Rset+i\alpha} K(\clp<+>_{a}-\nu) \rden_{1}(\nu-\hle_{b})
	d\nu
	=
	K(\clp<+>_{a}-\hle_{b}).
\end{gather}
\label{inteq_rden_shft_kernel_clp_+_gau_ex}
\end{subequations}
Both integral equations are studied in \cref{chap:den_int_aux}.
Here we only need to observe that the value of the parameter $\alpha>0$ can be freely varied in the region of analyticity of its integrands.
Therefore we can use \cref{inteq_rden_shft_kernel_clp_+_gau_ex} to compute convolutions over $\rden_2$ and $\rden_1$ (or $\rden_h$) coming from the substitution of \cref{gau_ex_II_gen_Gf_den_expn_rl_str} into \cref{gau_ex_clp_diff_gen_condn_all}.
It also involves convolutions with density terms for higher-level roots $\ho{\rden}$ but these need not be computed since they only affect coefficients in the sum over $\Fmat*[e]$ which will be cancelled in the determinant.
Finally, this substitution give us the following expression for components of the block $\Fmat[e]^c$:
\begin{align}
	\Fmat[e]^c_{a,k}
	=
	\frac{\pi(1+\aux_{	e}(\check{\la}_{k}))}{\sinh\pi(\la_{k}-\clp<+>_{a})}
	-2\pi i\bmsum
	\frac{%
	\rden_{h}(\clp<+>_{a}-\bm\hle)
	}{%
	\aux_{e}^\prime(\bm\hle)
	}
	\Gf[e](\bm\hle,\la_{k})
	+ 2\pi i
	\sum_{b=1}^{\ho{n}}
	\chi_{b}^\prime\Fmat*[e]_{b,k}
	\label{gau_ex_clp_cau_cols}
\end{align}
Note that \cref{gau_ex_clp_cau_cols} also includes the case $\check\la_{N_0+1}=\frac{i}{2}$.
Although initial \cref{gau_ex_clp_diff_gen_condn,gau_ex_clp_diff_gen_condn_i2} were different due to the difference in location of the poles of the function $\Gf[e]$, we find that the expression \eqref{gau_ex_clp_cau_cols} also holds in $\check\la_{N_0+1}=\frac{i}{2}$ case.
\par
We can see that the expression that obtained here in \cref{gau_ex_clp_cau_cols} has a form resembling the expression that was obtained in \cref{gau_ex_II_gen_Gf_den_expn_rl_str} for the block of real rows $\Fmat[e]^{r}$.
This is an important observation. It will allow us to combine all the Cauchy terms in a single block in final the modified Cauchy determinant representation.
It is in this context that we find it useful to denote the union of the real roots and the positive close-pair roots as $\bm{\rl^+}=\bm\rl\cup\bmclp<+>$.
\par
A very important result was found in \cref{hl_gau_ex_clp} where we obtained a system of linear equations for the block $\Fmat*[e]^{c}$, that contains only the higher-level terms.
Let us emphasize that it bears a striking resemblance to the Gaudin extraction \eqref{gau_ex_mat_II_def}.
This comparison tells us that \cref{hl_gau_ex_clp} can be seen as an extraction of the higher-level Gaudin matrix \eqref{hl_gaudin_intro_sec}.
In the next \cref{sub:gau_ex_II_wdp}, we will extend this result to the wide-pair blocks $\Fmat*^{\txtwp+}$ and $\Fmat*^{\txtwp-}$.
\subsection{For wide-pairs}
\label{sub:gau_ex_II_wdp}
Unlike in the case of the close-pair, the derivative of the exponential counting function for the wide-pair do not contain any singular terms. 
From \cref{aux_hl_xxx} we can readily see that
\begin{subequations}
\begin{align}
	\aux_{e}^\prime(\wdp<+>_a)&=\aux*^\prime(\wdp_a)
	\\
	\aux_{e}^\prime(\wdp*<->_a)&=\aux*^\prime(\wdp*_a)
	.
\end{align}
\label{der_aux_wdp_gau_ex}
\end{subequations}
With this the system of equations \eqref{gau_ex_syslin_II_gen} for the extraction of Gaudin matrix on the block of wide-pairs rows can be written as
\begin{subequations}
\begin{multline}
	\aux*^\prime(\wdp_{a})	\Fmat*[e]^{\txtwp+}_{a,k}
	\begin{aligned}[t]
	&- 2\pi i\bmsum K(\wdp<+>-\bm\rl)
	\frac{%
	\Gf[e](\bm\rl,\la_{k})
	}{%
	\aux_e^\prime(\bm\rl)
	}
	\\
	&\quad
	- 2\pi i \sum_{b=1}^{n_\txtcp}
	\left\lbrace
	K(\wdp_{a}-\clp_{b})
	+
	K(\wdp_{a}-\clp_{b}+i)
	\right\rbrace
	\Fmat*[e]^{\txtcp}_{b,k}
	\\
	&\qquad
	-2\pi i \sum_{b=1}^{n_\txtwp} K(\wdp_{a}-\wdp_{b}) \Fmat*^{\txtwp+}_{b,k}
	-2\pi i \sum_{b=1}^{n_\txtwp} K(\wdp_{a}-\wdp*_{b}+i) \Fmat*^{\txtwp-}_{b,k}
	\end{aligned}
	\\
	=
	(1+\aux_{e}(\la_{k}))\,
	t(\la_{k}-\wdp<+>_{a})
	-
	K(\la_{k}-\wdp<+>_{a})
	.
	\label{gau_ex_II_syslin_wdp+}
\end{multline}
And,
\begin{multline}
	\aux*^\prime(\wdp*_{a})	\Fmat*[e]^{\txtwp-}_{a,k}
	\begin{aligned}[t]
	&- 2\pi i\bmsum K(\wdp*<->-\bm\rl)
	\frac{%
	\Gf[e](\bm\rl,\la_{k})
	}{%
	\aux_e^\prime(\bm\rl)
	}
	\\
	&\quad
	- 2\pi i \sum_{b=1}^{n_\txtcp}
	\left\lbrace
	K(\wdp*_{a}-\clp_{b})
	+
	K(\wdp*_{a}-\clp_{b}-i)
	\right\rbrace
	\Fmat*[e]^{\txtcp}_{b,k}
	\\
	&\qquad
	-2\pi i \sum_{b=1}^{n_\txtwp} K(\wdp*_{a}-\wdp_{b}-i) \Fmat*^{\txtwp+}_{b,k}
	-2\pi i \sum_{b=1}^{n_\txtwp} K(\wdp*_{a}-\wdp*_{b}) \Fmat*^{\txtwp-}_{b,k}
	\end{aligned}
	\\
	=
	(1+\aux_{e}(\la_{k}))\,
	t(\la_{k}-\wdp*<->_{a})
	-
	K(\la_{k}-\wdp*<->_{a})
	.
	\label{gau_ex_II_syslin_wdp-}
\end{multline}
\label{gau_ex_II_syslin_wdp}
\end{subequations}
Here we have used \cref{der_aux_wdp_gau_ex} to replace the derivatives of the exponential counting function by their higher-level counterparts.
We also used the pairwise degeneracy \eqref{gau_ex_coin_clp_cols} of the close-pair rows to recombine the terms $\Fmat[e]^{c\pm}$ in the summations. After this recombination, we will be using the notations defined in \cref{clp_diff_block_gau_ex_II,gau_ex_II_hl_block}.
\\
Let us remark that \cref{gau_ex_II_syslin_wdp} involves the sum over the meromorphic function $\Gf[e]$.
In the case of $\Gf[e](\nu,\la_k)$, it has a simple pole at $\nu=\la_k$ on the real line, with the residue that is given by \cref{res_Gf_init_II_gen}. 
We can use \cref{gen_condn_prop} to convert the sum over the function $\Gf[e]$ into integrals:
\begin{subequations}
\begin{align}
	\bmsum K(\wdp<+>_{a}-\bm\rh)
	\frac{%
	\Gf[e](\bm\rh,\la_{k})
	}{%
	\aux_{e}^\prime(\bm\rh)
	}
	&=
	-
	\frac{1}{2\pi i}
	\int_{\Rset+i0} K(\wdp<+>_{a}-\nu)
	\Gf[e](\nu,\la_{k})
	d\nu
	+ K(\la_{k}-\wdp<+>_{a})
	\label{gau_ex_II_wdp+_condn}
	,
\shortintertext{and,}
	\bmsum K(\wdp*<->_{a}-\bm\rh)
	\frac{\Gf[e](\bm\rh,\la_{k})}{\aux'_e(\bm\rh)}
	&=
	-
	\frac{1}{2\pi i}
	\int_{\Rset+i0} K(\wdp*<->_{a}-\nu)
	\Gf[e](\nu,\la_{k})
	d\nu
	+ K(\la_{k}-\wdp*<->_{a})
	.
	\label{gau_ex_II_wdp-_condn}
\end{align}
	\label{gau_ex_II_wdp_condn}
\end{subequations}
For the particular case of a sum over $\Gf[e](\bm\nu,\frac{i}{2})$, it suffices to use the regular condensation property to obtain the convolution integrals
\begin{subequations}
\begin{align}
	\bmsum K(\wdp<+>_{a}-\bm\rh)
	\frac{%
	\Gf[e](\bm\rh,\tfrac{i}{2})
	}{%
	\aux_{e}^\prime(\bm\rh)
	}
	&=
	-
	\frac{1}{2\pi i}
	\int_{\Rset} K(\wdp<+>_{a}-\nu)
	\Gf[e](\nu,\tfrac{i}{2})
	d\nu
	\label{gau_ex_II_wdp+_condn_i2}
\shortintertext{and}
	\bmsum K(\wdp*<->_{a}-\bm\rh)
	\frac{\Gf[e](\bm\rh,\tfrac{i}{2})}{\aux'_e(\bm\rh)}
	&=
	-
	\frac{1}{2\pi i}
	\int_{\Rset} K(\wdp*<->_{a}-\nu)
	\Gf[e](\nu,\tfrac{i}{2})
	d\nu
	.
	\label{gau_ex_II_wdp-_condn_i2}	
\end{align}
\label{gau_ex_II_wdp_condn_i2}	
\end{subequations}
The convolution integrals appearing in \cref{gau_ex_II_wdp_condn,gau_ex_II_wdp_condn_i2} are studied in \cref{sec:den_conv_for_gau_ex_k_clp} of \cref{chap:den_int_aux}.
Here we borrow the results obtained in \cref{Gf_conv_K-shft_wdp+_ker_append_gau_ex,Gf_conv_K-shft_wdp-_append_gau_ex} to express:
\begin{subequations}
\begin{multline}
	\bmsum K(\wdp<+>_{a}-\bm\rl)
	\frac{%
	\Gf[e](\bm\rl,\check\la_{k})
	}{%
	\aux_{e}^\prime(\bm\rl)
	}
	=
	-\frac{1+\aux_{e}(\check\la_{k})}{2\pi i} t(\check\la_{k}-\wdp<+>_{a})
	+ K(\check\la_{k}-\wdp<+>_{a})
	\\
	\begin{aligned}[b]
	&
	-\sum_{b=1}^{n_h} K_{2}(\wdp_{a}-\hle_b)
	\frac{%
	\Gf[e](\hle_b,\check\la_{k})
	}{%
	\aux_{e}^\prime(\hle_b)
	}
	\\
	&\quad
	- \sum_{b=1}^{n_\txtcp}
	K(\wdp_{a}-\clp_{b}+i)
	\Fmat*[e]^\txtcp_{b,k}
	\\
	&\qquad
	- \sum_{b=1}^{n_\txtwp}
	\left\lbrace
	K(\wdp_{a}-\wdp*_{b}+i)
	-
	K(\wdp_{a}-\wdp*_{b})
	\right\rbrace
	\Fmat*[e]^{\txtwp-}_{b,k}
	\end{aligned}
	\label{gau_ex_II_wdp+_conv_result}
	.
\end{multline}
And,
\begin{multline}
	\bmsum K(\wdp*<->_{a}-\bm\rl)
	\frac{%
	\Gf[e](\bm\rl,\check\la_{k})
	}{%
	\aux_{e}^\prime(\bm\rl)
	}
	=
	-\frac{1+\aux_{e}(\check\la_{k})}{2\pi i} t(\check\la_{k}-\wdp*<->_{a})
	+ K(\check\la_{k}-\wdp*<->_{a})
	\\
	\begin{aligned}[b]
	&-\sum_{b=1}^{n_h} K_{2}(\wdp*_{a}-\hle_b)
	\frac{%
	\Gf[e](\hle_b,\check\la_{k})
	}{%
	\aux_{e}^\prime(\hle_b)
	}
	\\
	&\quad
	- \sum_{b=1}^{n_\txtcp}
	K(\wdp*_{a}-\clp_{b}-i)
	\Fmat*[e]^\txtcp_{b,k}
	\\
	&\qquad
	- \sum_{b=1}^{n_\txtwp}
	\left\lbrace
	K(\wdp*_{a}-\wdp_{b}-i)
	-
	K(\wdp*_{a}-\wdp_{b})
	\right\rbrace
	\Fmat*[e]^{\txtwp+}_{b,k}
	\end{aligned}
	\label{gau_ex_II_wdp-_conv_result}
	.
\end{multline}
\label{gau_ex_II_wdp_conv_result}
\end{subequations}
Let us now substitute these two expressions from \cref{gau_ex_II_wdp_conv_result} into \cref{gau_ex_II_syslin_wdp}.
This substitution allow us to write down the system of equations resembling \eqref{hl_gau_ex_clp} for the extraction of higher-level Gaudin matrix:
\index{{ff@\textbf{Form-factors}!mat FF hl wdp@\hspace{1em}$\Fmat*[e]^{w\pm}$: wide-pair sub-blocks of \rule{3em}{1pt}}}%
\begin{subequations}
\begin{align}
	\aux*_{e}^\prime(\wdp_{a}) \Fmat*[e]^{\txtwp+}_{a,k}
	- 2\pi i \sum_{b=1}^{\ho{n}}
	K(\wdp_{a}-\cid_{b}) \Fmat*[e]_{b,k}
	&=
	- 2\pi i
	\bmsum
	\frac{%
	K_{2}(\wdp_{a}-\bm\hle)
	}{%
	\aux_{e}^\prime(\bm\hle)
	}
	\Gf[e](\bm\hle,\la_{k})
	,
	\label{hl_gau_ex_wdp+}
\shortintertext{and,}
	\aux*_{e}^\prime(\wdp*_{a}) \Fmat*[e]^{\txtwp-}_{a,k}
	- 2\pi i \sum_{b=1}^{\ho{n}}
	K(\wdp*_{a}-\cid_{b}) \Fmat*[e]_{b,k}
	&=
	- 2\pi i
	\bmsum
	\frac{%
	K_{2}(\wdp*_{a}-\bm\hle)
	}{%
	\aux_{e}^\prime(\bm\hle)
	}
	\Gf[e](\bm\hle,\la_{k})
	\label{hl_gau_ex_wdp-}
	.
\end{align}
\label{hl_gau_ex_wdp}
\end{subequations}
We can observe that function $K_2$ on the right-hand side of \cref{hl_gau_ex_wdp} is same as the function for the common density term of the higher-level roots $\ho{\rden}$. In fact, it is a rational function that coincides with the derivative of the bare momentum $p'_0$ since we can write
\begin{align}
	\ho{\rden}(\nu)
	=
	K_2(\nu)	
	=
	\frac{1}{2\pi}p'_0(\nu)
	.
	\label{hl_den_K2_bare_mom_rel}
\end{align}
Finally, we can combine the results of \cref{hl_gau_ex_wdp,hl_gau_ex_clp} to write a combined system of linear equation for the block $\Fmat*[e]$.
This lead us to a very important result that is summarised in the following section.
\subsection{Emergence of the higher-level Gaudin matrix and its extraction}
\label{sub:hl_gau_ex}
By combining \cref{hl_gau_ex_wdp,hl_gau_ex_clp}, we obtain the following system of linear equations satisfied by the block $\Fmat*[e]$ \eqref{gau_ex_II_hl_block} of the $\Fmat[e]$ matrix:
\index{ff@\textbf{Form-factors}!mat FF hl@$\Fmat*[e]$: higher-level block inside the matrix $\Fmat[e]$}%
\begin{align}
	\aux*_{e}^\prime(\cid_{a})	\Fmat*[e]_{a,k}
	-2\pi i\sum_{b=1}^{\ho{n}}
	K(\cid_{a}-\cid_{b}) \Fmat*[e]_{b,k}
	=
	-
	2\pi i
	\bmsum
	\frac{%
	\ho{\rden}(\cid_{a}-\bm\hle)
	}{%
	\aux_{e}^\prime(\bm\hle)
	}
	\Gf[e](\bm\hle,\la_{k})	
	.
	\label{hl_gau_ex_linsys}
\end{align}
We shall now define the higher-level version of the Gaudin matrix \eqref{gau_mat_gen} and also the same for the density matrix \eqref{den_hle_mat_rect}.
\begin{defn}[Higher-level Gaudin matrix]
\label{defn_hl_gaudin_mat}
\index{ff@\textbf{Form-factors}!mat Gau hl@\hspace{1em}$\Ho{\Ncal}[\cdot\Vert\cdot]$: higher-level equivalent of \rule{3em}{1pt}|textbf}%
For a given set of the higher-level Bethe roots $\bm\cid$ satisfying \cref{hl_bae}, we define the higher-level Gaudin matrix $\Ho{\Ncal}[\bm\cid\Vert\bm\cid]$ as a square matrix of the order $\ho{n}=n_{\bm\cid}$ described by the following expression for its components:
\begin{align}
	\Ho{\Ncal}_{ab}[\bm\cid\Vert\bm\cid]=
	\aux*'(\cid_a)\delta_{a,b}
	-
	2\pi i
	K(\cid_a-\cid_b)
	\label{hl_gaudin_mat}
	.
\end{align}
\end{defn}
\begin{notn}
\index{ff@\textbf{Form-factors}!mat den higher level@$\Ho{\Rcal}[\cdot\Vert\cdot]$: a matrix formed by the higher-level density terms $\ho{\rden}$|textbf}%
\label{defn_den_mat_hl}
Let us define the rectangular matrix $\Ho{\Rcal}[\bm\cid\Vert\bm\hle]$ of $\ho{n}$ rows and $n_h$ columns whose components are given by,
\begin{align}
	\Ho{\Rcal}_{ab}[\bm\cid\Vert\bm\hle]=
	-2\pi i
	\ho{\rden}(\cid_a-\hle_b)
	.
	\label{den_mat_hl}
\end{align}
Note that due to the relation \eqref{hle_num_ho_num_rel} between the cardinalities $n_h$ and $\ho{n}$, this matrix is always rectangular.
In the case of a triplet excitation, we have $\frac{n_h}{2}-1$ rows for $n_h$ columns. 
This can also be seen as a higher-level variant of the Izergin matrix since $\ho{\rden}$ is a rational function as we saw it in \cref{hl_den_K2_bare_mom_rel}.
\end{notn}
In terms of these matrices, we can write down the higher-level equivalent of the Gaudin extraction that gives the matrix $\Fmat*$:
\begin{align}
	\Fmat*=
	\Ho{\Ncal}^{-1}[\bm\cid\Vert\bm\cid]
	\cdot
	\Ho{\Rcal}[\bm\cid\Vert\bm\hle]
	\cdot
	\Acal_e^{-1}[\bm\hle]
	\cdot
	\Gf[e][\bm\hle\Vert\bm{\check\la}]
	\label{hl_gau_ex_coupled}
\end{align}
The matrix $\Gf[e][\bm\hle\Vert\bm{\check\la}]$ is obtained by promoting the function $\Gf[e]$ to a parametrised matrix as 
\begin{align}
	\Gf[e][\bm\hle\Vert\bm{\check\la}]
	=
	\bm{\big[}\Gf[e](\bm\hle,\bm{\check\la})\bm{\big]}_{\bm\hle,\bm{\check\la}}	.
	\label{Gfn_mat}
\end{align}
Note that all of the matrices except the $\Ho{\Ncal}$ in \cref{hl_gau_ex_coupled} are all strictly rectangular and only the matrices $\Ho{\Ncal}$ and $\Ho{\Rcal}$ are finite in the thermodynamic limit.
It prompts us to define the matrix $\Ho{\Scal}$ which is higher-level equivalent of the matrices $\Fcal$ which gives the determinant representation of the form-factor due to \cref{det_rep_fini_ff}.
We will see in \cref{chap:gen_FF} that higher-level matrix $\Ho{\Scal}$ enters the determinant representation of the form-factor.
It is evident to us that due to its rectangular form it must be embedded inside a square matrix, which in this case is $\Fmat[e]$.
\begin{defn}
\label{defn_hl_ff_mat}
The matrix $\Ho{\Scal}$ is defined as the result of the higher-level extraction:
\index{ff@\textbf{Form-factors}!mat hl gauex@$\ho\Scal$: result of the higher-level Gaudin extraction|textbf}%
\begin{align}	
	\Ho{\Scal}[\bm\cid\Vert\bm\hle]=
	\Ho{\Ncal}^{-1}[\bm\cid\Vert\bm\cid]
	\cdot
	\Ho{\Rcal}[\bm\cid\Vert\bm\hle]
	.
	\label{hl_ff_mat}
\end{align}
We can check that the matrix $\Ho{\Scal}$ has $\ho{n}$ rows and $n_h$ columns and it remains a finite matrix in the thermodynamic limit for any low-lying excitation.
\end{defn}
The matrix $\Gf[e][\bm\hle\Vert\bm{\check\la}]$ can be obtained from \cref{gau_ex_II_gen_Gf_den_expn_rl_str}.
In the case of two-spinon form-factors, we found in \cref{chap:2sp_ff} that it satisfies the system \eqref{gau_ex_II_sys_hle}, which gets decoupled to the leading order in the thermodynamic for the choice of the holes $\bm\hle$ that lie in bulk of the Fermi distribution.
\\
In the current scenario, we find that the system \eqref{gau_ex_II_gen_Gf_den_expn_rl_str} for $\Gf[e][\bm\hle\Vert\bm{\check\la}]$ is only a part of a larger system. The remaining sub-system solves higher-level equations \eqref{hl_gau_ex_coupled}. We see that both subsystems are intricately coupled, among themselves, as well to each other.
\\
Let us first partially decouple the system \eqref{gau_ex_II_gen_Gf_den_expn_rl_str} from the higher-level Gaudin extraction \eqref{hl_gau_ex_coupled}. 
This can be done by incorporating them together through the effective $\Rcal^\text{eff}$ matrix as shown below
\begin{align}
	\big(\Id-\Rcal^{\text{eff}}\Acal_e^{-1}[\bm\hle\Vert\bm\hle]\big)
	\Gf[e][\bm\hle\Vert\bm{\check\la}]
	=
	\bm{\bigg[}
	\frac{%
	\pi(1+\aux_{e}(\bm{\check\la}))
	}{%
	\sinh\pi(\bm{\check\la}-\bm\hle)
	}
	\bm{\bigg]}
	.
	\label{gau_ex_II_syslin_hle_gen}
\end{align}
The effective matrix $\Rcal^{\text{eff}}$ in \cref{gau_ex_II_syslin_hle_gen} is obtained by subtracting the dressed higher-level Gaudin matrix as follows:
\begin{align}
	\Rcal^{\text{eff}}[\bm\hle\Vert\bm\hle]
	=
	\Rcal[\bm\hle\Vert\bm\hle]
	-
	\Ho{\Rcal}[\bm\hle\Vert\bm\cid]
	\cdot
	\Ho{\Ncal}^{-1}[\bm\cid\Vert\bm\cid]
	\cdot
	\Ho{\Rcal}[\bm\cid\Vert\bm\hle]
	.
	\label{denmat_hle_dressed_gen_ff}
\end{align}
Therefore, we can use the same argument that was used in the two-spinon case, with the only difference being the replacement of the $\Rcal$ matrix with an effective version of it.
Hence, we can claim that when the hole parameters $\bm\hle$ are taken inside the bulk of the Fermi-zone, we can expect that all the components of the dressed matrix are of the order $O(1/M)$.
This ensures that the system \eqref{gau_ex_II_syslin_hle_gen} is decoupled to the leading order.
It allows us to write,
\begin{align}
	\Gf[e](\hle_a,\check\la_k)
	=
	\frac{\pi(1+\aux_e(\check\la_k))}{\sinh\pi(\check\la_k-\hle_a)}
	+
	O\left(\frac{1}{M}\right)
	.
	\label{Gf_hle_decoupled_gen}
\end{align}
Let us substitute \cref{Gf_hle_decoupled_gen} into \cref{gau_ex_II_gen_Gf_den_expn_rl_str}, it permits us to write, 
\begin{multline}
	\Gf[g](\rl_j,\check\la_k)
	=
	\frac{%
	\pi%
	(1+\aux_e(\check\la_k))
	}{\sinh\pi(\check\la_k-\rl_j)}
	-
	2\pi i
	\bmsum
	\frac{\rden_h(\rl_j-\bm\hle)}{\aux'_e(\bm\hle)}
	\frac{%
	\pi%
	(1+\aux_e(\check\la_k))
	}{\sinh\pi(\check\la_k-\bm\hle)}
	\\
	+
	\sum_{a=1}^{\ho{n}}
	\chi^{r}_{j,a}
	\Fmat*_{a,k}
	+
	o\left(\frac{1}{M}\right)
	.
	\label{gau_ex_II_Gf_sol_rl_str_decouple}
\end{multline}
In this way all the components of the block $\Fmat[e]^r$ can be obtained from here through \cref{Gf_init_II_gen}.
Since it is the determinant of the matrix $\Fmat[e]$ that it the primary object of our interest, we can also silently\footref{foot:silent_kill} remove the linear sum over the rows from the block $\Fmat*$, since all of its rows are already present inside the matrix $\Fmat[e]$.
It permits us to write,
\begin{subequations}
\begin{multline}
	\Fmat[e]^r_{j,k}
	=
	\frac{1+\aux_e(\check\la_k)}{\aux'_e(\rl_j)}
	\left\lbrace
	\frac{\pi}{\sinh\pi(\la_{k}-\rl_j)}
	\right.
	\\
	\left.
	-2\pi i\bmsum
	\frac{%
	\rden_{h}(\rl_j-\bm\hle)
	}{%
	\aux_{e}^\prime(\bm\hle)
	}
	\frac{%
	\pi
	}{%
	\sinh\pi(\check\la_k-\bm\hle)
	}
	+
	o\left(\frac{1}{M}\right)
	\right\rbrace
	.
	\label{gau_ex_II_rl_result_decoupled}
\end{multline}
Note that this expression \eqref{gau_ex_II_rl_result_decoupled} is similar to the result \eqref{gau_ex_sol_II_decoupled} in the case of two-spinon form-factor.
We can see from \cref{gau_ex_clp_cau_cols} that the same reasoning can also be extended to the block $\Fmat[e]^c$, which was obtained through the combination \eqref{clp_diff_block_gau_ex_II}.
Once again we will substitute solutions from \cref{Gf_hle_decoupled_gen} into \cref{gau_ex_clp_cau_cols} in order to cancel the extra terms in the sum running over $\Fmat*$. It permits us to write,
\begin{multline}
	\Fmat[e]^c_{a,k}
	=
	(1+\aux_e(\check\la_k))
	\left\lbrace
	\frac{\pi}{\sinh\pi(\la_{k}-\clp<+>_{a})}
	\right.
	\\
	\left.
	-2\pi i\bmsum
	\frac{%
	\rden_{h}(\clp<+>_{a}-\bm\hle)
	}{%
	\aux_{e}^\prime(\bm\hle)
	}
	\frac{%
	\pi(1+\aux_e(\check\la_k))%
	}{%
	\sinh\pi(\check\la_k-\bm\hle)
	}
	+o\left(\frac{1}{M}\right)
	\right\rbrace
	.
	\label{gau_ex_clp_cau_cols_decoupled}
\end{multline}
And finally we can see that decoupling in \cref{Gf_hle_decoupled_gen} also partially decouples the system of \cref{hl_gau_ex_coupled} for the higher-level Gaudin extraction. It is simplified to the form shown in the following:
\begin{align}
	\Fmat*_{a,k}
	=
	\sum_{b=1}^{\ho{n}}
	\Ho{\Scal}_{a,b}(\bm\cid\Vert\bm\hle)
	\frac{%
	\pi
	(1+\aux_e(\check\la_k))
	}{\sinh\pi(\check\la_k-\hle_b)}
	+
	o\left(\frac{1}{M}\right)
	.
	\label{hl_gau_ex_decouple}
\end{align}
The matrix $\Ho{\Scal}$ was introduced in \cref{defn_hl_ff_mat}.
This completes the computation of all the non-trivial rows of the matrix $\Fmat[e]$.
Apart from these three $\Fmat[e]^r$, $\Fmat[e]^c$ and $\Fmat*$, let us also remember that we also have the Foda-Wheeler block, which has retained its original form:
\begin{align}
	\Fmat[e]_{N_0,k}&=
	\aux_{e}(\check{\la}_{k})-1,
	\label{gau_ex_II_FW_deg0_gen}
	\\
	\Fmat[e]_{N_0+1,k}&=
	\aux_{e}(\check{\la}_{k})(\check{\la}_{k}+i)-\check{\la}_{k}.
	\label{gau_ex_II_FW_deg1_gen}
\end{align}
\label{blocks_gau_ex_gen_II}
\end{subequations}
We will now combine all the results from \cref{sub:gau_ex_II_gen,sec:gau_ex_I_gen} to write down a modified Cauchy determinant representation.
\section{Modified Cauchy determinant representation} %
\label{sec:mod_cau_det_rep_gen}
Let us now put together the determinant representation \eqref{det_rep_fini_ff_gen} that is produced by \cref{blocks_gau_ex_gen_II,blocks_gau_ex_I_gen_sol} for the form-factor of any generic triplet excitation.
Here it becomes useful to consider the following notation.
\begin{notn}
\index{exc@\textbf{Excitations}!set real clp plus@$\bm{\rl^+}$: set of all Bethe real and positive close-pair roots $0\leq \Im\rl^+ < 1$|see{Form-factors}}%
\index{ff@\textbf{Form-factors}!real clp plus@$\bm{\rl^+}$: set of all Bethe real and positive close-pair roots $0\leq \Im\rl^+ < 1$|textbf}%
Let $\bm{\rl^+}$ denote the set of all real roots $\bm\rl$ with the positive close-pair roots:
\label{notn:real_clp+_roots}
\begin{subequations}
\begin{align}
	\bm{\rl^+}&=
	\bm\rl\bm\cup\bmclp<+>
	,
	&\text{with the cardinality}
	\quad
	n_{\bm{\rl^+}}&=N_0-\ho{n}-1
	.
\end{align}
Similarly, let $\bm{\check{\rl}^+}$ denote the union:
\begin{align}
	\bm{\check{\rl}^+}&=
	\bm\rl\bm\cup
	\set{\tfrac{i}{2}}
	\bm\cup\bmclp<+>
	,
	&\text{with the cardinality}
	\quad
	n_{\bm{\rl^+}}&=N_0-\ho{n}
	.
\end{align}
\end{subequations}
Both the sets are indexed in the ascending order of their union.
\end{notn}
\begin{notn}
\index{exc@\textbf{Excitations}!set real clp hle plus@$\bm{\rh^+}$: set of all real roots (incl holes) and positive close-pair roots $0\leq \Im\rh^+ < 1$|see{FF.}}%
\index{ff@\textbf{Form-factors}!real clp hle plus@$\bm{\rh^+}$: set of all real roots (incl holes) and positive close-pair roots $0\leq \Im\rh^+ < 1$|textbf}%
\label{notn:real_clp+_hle_roots}
Let $\bm{\rh^+}$ denote the set of all real roots $\bm\rl$, all the positive close-pairs $\bmclp<+>$ and holes $\bm\hle$:
	\begin{align}
		\bm{\rh^+}&=
		\bm\rl\bm\cup
		\bmclp<+>\bm\cup
		\bm\hle
		,
		&\text{with the cardinality}
		\quad
		n_{\bm{\rh^+}}&=N_0+\ho{n}+1
		.
	\end{align}
It is indexed in the ascending order given by their union.
\end{notn}
\index{ff@\textbf{Form-factors}!mat mod Cau gs@$\modCau[g]$: (modified) Cauchy matrix for the ground state|textbf}%
In the determinant of the matrix $\Fmat[g]$, we take all the common terms into the prefactor to write
\begin{subequations}
\begin{align}
	\det\Fmat[g]
	&=
	\pi^{N_0}
	\frac{%
	\bmprod(1+\aux_{g}(\bm{\rl}))%
	}{%
	\bmprod \aux_{g}^\prime(\bm\la)%
	}%
	\det_{N_0}\modCau[g]
	.
	\label{cau_det_I_gen}
\end{align}
The matrix $\modCau[g]$ consists of four blocks of columns
\index{ff@\textbf{Form-factors}!mat mod Cau gs cau@\hspace{1em}$\modCau[g]^{\text{cau}}$: Cauchy block inside \rule{3em}{1pt}}%
\index{ff@\textbf{Form-factors}!mat mod Cau gs clp@\hspace{1em}$\modCau[g]^{c}$: close-pair block inside \rule{3em}{1pt}}%
\index{ff@\textbf{Form-factors}!mat mod Cau gs wdp@\hspace{1em}$\modCau[g]^{w\pm}$: wide-pair block inside \rule{3em}{1pt}}%
\begin{align}
	\modCau[g]=
	\left(
	\modCau[g]^{\text{cau}}
	~\bigg|~
	\modCau[g]^{c}
	~\bigg|~
	\modCau[g]^{w+}
	~\bigg|~
	\modCau[g]^{w-}
	\right)
	\label{cau_det_rep_gen_I_blocks}
\end{align}
\end{subequations}
The first block $\modCau[g]^{\text{cau}}$ is the largest one that is made up of $N_0-\ho{n}$ columns. It is obtained by aggregating all the Cauchy blocks $\Fmat[g]^r$ and $\Fmat[g]^{c+}$ in the original matrix $\Fmat[e]$ as seen from \cref{gau_ex_I_gen_sol_clp+,gau_ex_I_gen_sol_real}.
Using \cref{notn:rect_cau_mat_hyper} that was first invoked in \cref{chap:2sp_ff}, the Cauchy block can be written as
\begin{align}
	\modCau[g]^{\text{cau}}=-\Cmat\left[\bm\la\big\Vert\bm{\check\rl^+}\right]
	.
	\label{cau_det_rep_I_mat_cau_block}
\end{align}
Note that it is similar to \cref{cau_det_rep_I_cau_mat} for the two-spinon form-factor however it differs from \cref{cau_det_rep_I_cau_mat} in the following two ways:
\begin{enumerate}[noitemsep]
\item In this case we have a rectangular Cauchy block $\modCau[g]^{\text{cau}}$ which is contained inside the matrix $\modCau[g]$.
\item One of the two close-pair blocks $\Fmat[g]^{c+}$ \eqref{gau_ex_I_gen_sol_clp+} has been assimilated with $\Fmat[g]^{r}$ \eqref{gau_ex_I_gen_sol_real} to form the Cauchy block $\modCau[g]^{\text{cau}}$.
\end{enumerate}
The another half of the close-pair block $\Fmat[e]^{c-}$ \eqref{gau_ex_I_gen_sol_clp-} gets transformed in the recombination \eqref{sla_clp_recomb}.
Here in terms of the $\modCau[g]$ matrix, it is now denoted as $\modCau[g]^{\txtcp}$.
Together with the wide-pair blocks $\modCau[g]^{\txtwp+}$ and $\modCau[g]^{\txtwp-}$ it represents the columns which deviate from a pure Cauchy matrix [see \cref{gau_ex_I_gen_sol_clp-,gau_ex_I_gen_sol_wdp-,gau_ex_I_gen_sol_wdp+}].
These three blocks of non-Cauchy columns $\modCau[g]^{c}$, $\modCau[g]^{w+}$ and $\modCau[g]^{w-}$ contains $n_c$, $n_w$ and $n_w$ columns respectively, which in total adds up to $\ho{n}=n_c+2n_w$.
\par
Finally, all the components of the matrix $\modCau[g]$ \eqref{cau_det_rep_gen_I_blocks} can be obtained from \cref{blocks_gau_ex_I_gen_sol}. It permits us to write,
\begin{subequations}
\begin{flalign}
	\modCau[g]^{\text{cau}}_{j,k}&=
	-\Cmat_{jk}[\bm\la\Vert\bm{\check\rl^+}]
	=
	\frac{1}{\sinh\pi(\check{\rl}^+_{k}-\la_{j})}
	,
	& &k\leq N_0-\ho{n}
	;
	\label{cau_det_rep_I_gen_real_clp}
	\\
	\modCau[g]^{\txtcp}_{j,a}&=
	\begin{multlined}[t]
	\aux_{g}(\clp<->_{a})
	\left\lbrace
	\frac{1}{%
	\sinh\pi(\clp<->_{a}-\la_{j}-i\stdv_{a})
	}
	+
	\frac{1}{%
	\sinh\pi(\clp<+>_{a}-\la_{j}+i\stdv_{a})
	}	\right\rbrace
	\\
	+2 i
	\left\lbrace
	\rden_{h}(\la_{j}-\clp<+>_{a})
	-
	\rden_{h}(\la_{j}-\clp<->_{a})
	\right\rbrace
	\end{multlined}
	,
	& &a\leq n_\txtcp
	;
	\label{cau_det_rep_I_gen_clp_diff}
	\\
	\modCau[g]^{\txtwp+}_{j,a}&=
	2i
	\big\lbrace
	\rden_{2}(\la_{j},\wdp_a+i)
	-
	\rden_{2}(\la_{j},\wdp_a)
	\big\rbrace
	,
	& &a\leq n_\txtwp
	;
	\label{cau_det_rep_I_gen_wdp+}
	\\
	\modCau[g]^{\txtwp-}_{j,a}&=
	2i
	\big\lbrace
	\rden_{2}(\la_{j},\wdp*_a)
	-
	\rden_{2}(\la_{j},\wdp*_a-i)
	\big\rbrace
	\label{cau_det_rep_I_gen_wdp-}
	,
	& &a\leq n_\txtwp.
\end{flalign}
\label{cau_det_rep_I_gen_components}
\end{subequations}
Similarly we shall now write the determinant of the matrix $\Fmat[e]$ as following:
\index{ff@\textbf{Form-factors}!mat mod Cau es@$\modCau[e]$: (modified) Cauchy matrix for an excited state|textbf}%
\begin{subequations}
\begin{align}
	\det\Fmat[e]&=
	\pi^{N_0+1}
	\frac{%
	\bmprod (1+\aux_{e}(\bm\la))
	}{%
	\bmprod \aux_{e}^\prime(\bm\rl)
	}
	\det_{N_0+1}\modCau[e]
	\label{cau_det_II_gen}
\end{align}
The matrix $\modCau[e]$ in the above \cref{cau_det_II_gen} is obtained from $\Fmat[e]$, after we have taken all the common terms into the prefactor.
Meanwhile, we also choose to transpose the original matrix $\Fmat[e]$ for a notational convenience in our later computations.
Therefore, it can be seen that the matrix $\modCau[e]$ consists of three blocks of columns as shown by the following expression:
\begin{align}
	\modCau[e]
	=
	\left(
	\modCau[e]^{\text{cau}}
	~ \bigg| ~
	\modCau*
	~ \bigg| ~
	\bar{\Ucal}
	\right)
	.
	\label{cau_det_rep_gen_II_blocks}
\end{align}
\end{subequations}
The matrix $\modCau[e]^{\text{cau}}$ forms the largest block of $N_0-\ho{n}-1$ columns.
We can see that it can be represented as the following sum of over the Cauchy matrices,
\index{ff@\textbf{Form-factors}!mat mod Cau es cau@\hspace{1em}$\modCau[e]^{\text{cau}}$: Cauchy block inside \rule{3em}{1pt}}%
\begin{align}
	\modCau[e]^{\text{cau}}
	=
	\Cmat\left[\bm{\check\la}\big\Vert\bm{\rl^+}\right]
	+
	\Cmat\left[\bm{\check\la}\big\Vert\bm{\hle}\right]
	\cdot
	\Acal_e^{-1}\Rcal\left[\bm\hle\big\Vert\bm{\rl^+}\right]
	.
	\label{cau_det_rep_II_gen_cau_block}
\end{align}
Note that it is similar to the expression \eqref{cau_det_rep_2sp_exc_block_rep} obtained in the two-spinon case.
The $\Rcal$ matrix is composed of density terms $\rden_h$ as defined in \cref{den_hle_mat_rect} whereas $\Acal_e$ is the diagonal matrix \eqref{diag_cfn_der_mat}.
Their product is written with contraction of the summed over variables. We can see that it is a rectangular matrix with $n_h$ rows and $N_0-\ho{n}-1$ columns given by the following expression.
\begin{align}
	\{\Acal_e^{-1}\Rcal\}_{ak}[\bm\hle\Vert\bm{\rl^+}]
	=
	-2\pi i
	\frac{\rden_h(\rl^+_k-\hle_a)}{\aux'_e(\hle_a)}
	.
\end{align}
The Cauchy matrices $\Ccal$ appearing in \cref{cau_det_rep_II_gen_cau_block} can be combined to construct a larger Cauchy matrix.
\begin{align}
	\Cmat\left[\bm{\check\la}\big\Vert\bm{\rh^+}\right]
	=
	\left[
	\Cmat\left[\bm{\check\la}\big\Vert\bm{\rl^+}\right]
	~
	\Big|
	~
	\Cmat\left[\bm{\check\la}\big\Vert\bm{\hle}\right]
	\right]
	.
	\label{big_cau_gen_blocks}
\end{align}
It has the components which are described by,
\begin{align}
	\Cmat_{jk}[\bm{\check\la}\Vert\bm{\rh^+}]
	=
	\frac{1}{\sinh\pi(\check\la_j-\rh^+_k)}
	.
	\label{big_cau_gen}
\end{align}
This construction \eqref{big_cau_gen_blocks} is similar in certain aspects to \cref{big_cau_2sp_blocks} which was seen in the two-spinon case.
However, in this case, the Cauchy matrix $\Cmat[\bm{\check\la}\Vert\bm{\rh^+}]$ turns out to be rectangular, the number of columns $n_{\bm{\rh^+}}=N_0+\ho{n}+1$ surpassing the number $n_{\bm{\check\la}}=N_0+1$ of rows, unless $\ho{n}=0$ which is the two-spinon case $n_h=2$.
In the face of this realisation, it may seem contrary to the intuition to insist on the extraction of a rectangular matrix.
But it turns out that it is still a good approach to extract a larger matrix containing $\Cmat[\bm{\check\la}\Vert\bm{\rh^+}]$, in comparison to a smaller square Cauchy matrix.
There are some technical reason for it, which we will revisit in \cref{chap:gen_FF}.
Here, let us list few keys points drawn from the computations in this chapter that support this observation: 
\begin{enumerate}[labelsep=1.5\parindent]
\item
\begin{enumerate}[wide=0pt, labelsep=.5\parindent]
\item We have seen that the close-pair block for the form-factors is always split in two halves, including the case of the first type of the extraction $\modCau[g]$, where one of the block always bore a striking resemblance to the Cauchy block for the real roots.
\item Such a halving of the close-pair block is comparable to the similar phenomenon observed in \cref{str_condn_clp} observed in the case of spectrum.
\end{enumerate}
\item The difference of cardinalities in the pairs of sets $\bm{\rh^+}$ and $\bm{\check\la}$, and well as, $\bm{\la}$ and $\bm{\check\rl^+}$, is $\ho{n}$. This number is an invariant in the eigenspace of $n_h$-spinon sector, as it determined by the relation \eqref{hle_num_ho_num_rel}: $\ho{n}=\frac{n_h}{2}-1$.
\end{enumerate}
Therefore, it makes sense to always associate one half of the set of close-pairs with the real roots whenever it would be possible in our forthcoming computations. This justifies \cref{notn:real_clp+_roots,notn:real_clp+_hle_roots}. It also partly explains the need for rectangular Cauchy extractions. We will see that this can be indeed realised with all the technical details in \cref{chap:gen_FF}.
\par
Now let us come back to the modified Cauchy determinant representation. We have seen that the second half of the close-pair block for the matrix $\Fmat[e]$ can be mixed with the wide-pair blocks $\Fmat[e]^{c\pm}$, to form a common higher-level block $\Fmat*$ \eqref{gau_ex_II_hl_block}.
Naturally, it gives rise to the block higher-level block $\modCau*$. 
It inherits the higher-level structure \eqref{hl_gau_ex_decouple} of its parent as follows:
\index{ff@\textbf{Form-factors}!mat mod Cau es cau@\hspace{1em}$\modCau*$: higher-level block inside \rule{3em}{1pt}}%
\begin{subequations}
\begin{align}
	\modCau*
	=
	\Cmat\left[\bm{\check\la}\big\Vert\bm\hle\right]
	\cdot
	\Acal_e^{-1}\Ho{\Scal}[\bm\hle\Vert\bm\cid]
	.
	\label{hl_ff_gau_ex}
\end{align}
Note that $\Ho{\Scal}[\bm\hle\Vert\bm\cid]$ is the transpose of $\Ho{\Scal}[\bm\cid\Vert\bm\hle]$ that was introduced in \cref{defn_hl_ff_mat}.
In terms of the higher-level version of the Gaudin matrix $\Ho{\Ncal}$ \eqref{hl_gaudin_mat} and matrix of densities $\Ho{\Rcal}$ \eqref{den_mat_hl}, it can be expressed as the following:
\begin{align}
	\ho{\Scal}[\bm\hle\Vert\bm\cid]
	=
	\Ho{\Rcal}[\bm\hle\Vert\bm\cid]
	\cdot
	\Ho{\Ncal}^{-1}[\bm\cid\Vert\bm\cid]
	\label{hl_ff_gau_ex_mat_transp}
\end{align}
	\label{hl_ff_gau_ex_all}
\end{subequations}
Finally let us also recall that we have the matrix $\bar\Ucal$ in $\modCau[e]$ \eqref{cau_det_rep_gen_II_blocks} consisting of two columns, which arise from the original Foda-Wheeler block [see \cref{gau_ex_II_FW_deg0_gen,gau_ex_II_FW_deg1_gen}].
Its components are
\index{ff@\textbf{Form-factors}!mat mod Cau es cau@\hspace{1em}$\bar\Ucal$: renorm. $\Ucal$, as a block inside \rule{3em}{1pt}}%
\begin{align}
	\bar{\Ucal}_{a,k}[\bm{\check\la}]
	=
	\frac{%
	\aux_{e}(\check{\la}_{k})%
	(\check\la_{k}+i)^{a-1}-\check\la_{k}^{a-1}%
	}{%
	\pi~
	(\aux_{e}(\check{\la}_{k})+1)%
	}.
	\label{cau_ex_II_gen_mat_FW_block}
\end{align}
Let us now substitute \cref{cau_det_I_gen,cau_det_II_gen} into the determinant representation \eqref{det_rep_fini_ff_gen}.
It permits us to write a new determinant representation involving matrices $\modCau[g]$ \eqref{cau_det_rep_gen_I_blocks} and $\modCau[e]$ \eqref{cau_det_rep_gen_II_blocks}, both of which contains a large block of matrix $\Cmat[\bm{\check\rl^+}\Vert\bm\la]$, or $\Cmat[\bm{\check\la}\Vert\bm{\rh^+}]$ respectively.
\begin{multline}
	\left|\FF^{z}\right|^2=
	-2\pi^{M+1}
	\frac{%
	\bmprod \revtf(\bm\la)
	}{%
	\bmprod \revtf(\bm\rl)
	}
	\frac{%
	\bmprod (\bm\rl-\bm\la) \bmprod(\bm\la-\bm\mu)
	}{%
	\bmprod^\prime (\bm\rl-\bm\mu) \bmprod^\prime (\bm\la-\bm\la)
	}
	\\
	\times
	\frac{%
	\bmprod
	\baxq_{g}(\bmclp<+>-i)
	\baxq_{g}(\bmclp<->-i)
	}{%
	\bmprod
	\baxq_{e}^\prime(\bmclp<+>-i)
	\baxq_{e}(\bmclp<->-i)
	}
	\frac{%
	\bmprod
	\baxq_{g}(\bmwdp<+>-i)
	\baxq_{g}(\bmwdp*<->-i)
	}{%
	\bmprod
	\baxq_{e}(\bmwdp<+>-i)
	\baxq_{e}(\bmwdp*<->-i)
	}
	\\
	\times
	\det_{N_0}\modCau[g]
	\det_{N_0+1}\modCau[e]
	\label{cau_det_rep_gen}
\end{multline}
We recall that the string deviations parameters in prefactors of \cref{det_rep_fini_ff_gen} have been absorbed into the matrix $\Fmat[e]$, when we wrote the combination \eqref{clp_diff_block_gau_ex_II}.
Therefore, these singular terms are absent in this new determinant representation \eqref{cau_det_rep_gen}. The primed notation $q'_e$ for the Baxter polynomials in the denominators of \cref{cau_det_rep_gen} indicates the omission of these terms.
\\
Let us also recall that the function $\revtf$ represents the ratio of eigenvalues of the transfer matrix, introduced in \cref{notn_def_revtf}.
Here it is obtained by combining prefactors of \cref{cau_det_II_gen,cau_det_I_gen} with those in \cref{det_rep_fini_ff_gen}. The thermodynamic limit of the $\revtf$ function was obtained in \cref{sec:tdl_revtf_append_sec}.
\par
The determinant representation \eqref{cau_det_rep_gen} is called the \emph{modified} Cauchy determinant representation, since the matrices $\modCau[g]$ \eqref{cau_det_rep_gen_I_blocks} and $\modCau[e]$ \eqref{blocks_gau_ex_gen_II} contain the columns due to the complex roots which deviate from the Cauchy terms that we normally obtain for the real roots.
Interestingly, the deviation from a pure Cauchy matrix is more prevalent in the $\modCau[g]$ matrix, compared to the $\modCau[e]$.  While the former matrix $\modCau[g]$ contains the blocks $\modCau[g]^{c}$, $\modCau[g]^{w+}$ and $\modCau[g]^{w-}$ made up of partially or completely non-Cauchy terms;
we find that in the case of the matrix $\modCau[e]$, the contribution from the complex-roots depicts a novel picture in the form of higher-level Gaudin extraction \eqref{hl_ff_gau_ex_all}, that is uniform for close-pairs and wide-pairs.
\begin{subappendices}
\section{Auxiliary results: Convolutions with the density terms}
Here we compute convolutions with the intermediate function $\Gf[e]$.
As we can see from its expansion \eqref{gau_ex_II_gen_Gf_den_expn_rl_str} into the density terms, we need to compute these convolutions for $\rden_2(\nu,\la+\tfrac{i}{2}-i0)$ ($\la\in\Rset$), $\rden_h(\nu-\hle)$ and $\ho{\rden}(\nu-\hle)$.
\minisec{Fourier transforms of the density terms involved}
Let us summarise their Fourier transforms here.
Fourier transforms of the $\rden_g$ and $\rden_h$ were obtained in \cref{rden_gs_ft_xxx,rden_hle_ft} which are given below:
\begin{subequations}
\label{ft_rden_gau_ex_append_both}
\begin{align}
	\what{\rden_g}(t)&=
	\frac{e^{-\frac{|t|}{2}}}{1+e^{-|t|}},
	\label{ft_rden_gs_gau_ex_append}
	\shortintertext{and,}
	\what{\rden_h}(t)&=
	\frac{e^{-|t|}}{1+e^{-|t|}}
	.
	\label{ft_rden_hle_gau_ex_append}
\end{align}
\end{subequations}
We note that these closely related to the function $\rden_\kappa(\la,\mu)$ which is studied in \cref{chap:den_int_aux}.
For the shifted $\rden_2(\nu,\la-\tfrac{i}{2}+i0)$ function we know from \cref{ft_shft_scl_rhden} that its Fourier transform is given by the same expression \eqref{ft_rden_gs_gau_ex_append} but with the shift:
\begin{align}
	\what{\rden_2}(t,\la-\tfrac{i}{2}+i0)
	=
	\frac{e^{-\frac{|t|}{2}}e^{-i\la t}e^{-\frac{t}{2}}}{1+e^{-|t|}}
	.
	\label{ft_rden_gs_gau_ex_append_shftd}
\end{align}
The common density term for the complex root $\ho{\rden}$ was first written in \cref{hl_rden_sol} which is based on the computation from \cref{sec:den_terms_DL_comp_append}.
There we also mentioned that although it density term for the close-pair and wide-pair have same functional form, it has different Fourier transform as seen from \cref{ft_ho_rden_clp,ft_ho_rden_wdp+,ft_ho_rden_wdp-} reproduced below:
\begin{subequations}
\label{ft_ho_rden_cmplx_gau_ex_gen_append}
\begin{align}
	\what{\ho{\rden}}(t,\clp)
	&=
	e^{-|\frac{t}{2}|}e^{-i\clp t}
	,
	\label{ft_ho_rden_clp_gau_ex_gen_append}
	\\
	\what{\ho{\rden}}(t,\wdp)
	&=
	I_{t<0}(1-e^{-t})e^{-i\wdp t}
	,
	\label{ft_ho_rden_wdp+_gau_ex_gen_append}
	\\
	\what{\ho{\rden}}(t,\wdp*)
	&=
	I_{t>0}(1-e^{t})e^{-i\wdp* t}
	\label{ft_ho_rden_wdp-_gau_ex_gen_append}
	.
\end{align}
\end{subequations}
\minisec{Some useful identities}
Here we write down some of the identities that are used in this chapter:
\begin{subequations}
\label{gau_ex_append_rat_fun_ids}
\begin{flalign}
	K(\la)&= \frac{1}{2\pi i}\big\lbrace t(\la)+t(-\la)\big\rbrace
	,
	\\
	K_{c}(\la)
	&= K_{2}(\la)-K_{2/3}(\la)
	,
	\\
	K_{1/2}(\la)
	&=
	K(\la+i)+K(\la-i)
,
	\\
	t(-\la)
	&=
	t(\la-i)
	.
\end{flalign}
\end{subequations}
\subsection{Density convolutions with \texorpdfstring{close-pair kernel $K_c$}{close-pair kernel}}
\label{sec:den_conv_for_gau_ex_k_clp}
Here we will compute the convolution of the function $K_{c}$ with the densities terms that arises in \cref{conv_K_clp_den_gen_condn_all}.
Here it is important to recall that the Fourier transform of the $K_c$ \eqref{ft_k_sum_clp} is
\begin{align}
	\what{K_c}(t)=
	e^{-\frac{|t|}{2}}(1+e^{-|t|})
	.
	\label{ft_k_sum_clp_append}
\end{align}
Let us now write down the convolutions one-by-one for each of the terms:
\minisec{Convolution with the density $\rden_2$.}
First, we get from \cref{ft_rden_gs_gau_ex_append} for $\rden_g$ we get
\begin{align}
	\int_{\Rset}
	K_{c}(\clp_{a}-\tau)
	\rden_g(\tau)
	d\tau
	&=
	\frac{1}{2\pi}
	\int_{\Rset} e^{i\clp_{a}t}e^{-|t|}\,dt
	=
	K(\clp_{a})
	=
	K(\clp<+>_a-\tfrac{i}{2})
	.
	\label{gau_ex_conv_append_rden_g_i2}
\end{align}
Next for $\rden_2(\nu,\la-\tfrac{i}{2}+i0)$ we get from \cref{ft_rden_gs_gau_ex_append_shftd} a similar expression:
\begin{align}
	\int_{\Rset}
	K_{c}(\clp_{a}-\tau)
	\rden_g(\tau,\la+\tfrac{i}{2}-i0)
	d\tau
	&=
	\frac{1}{2\pi}
	\int_{\Rset} e^{i(\clp_{a}-\la+\frac{i}{2})t}e^{-|t|}
	dt
	=
	K(\clp<+>_a-\la)
	.
	\label{gau_ex_conv_append_rden_2_shftd}
\end{align}
\minisec{Convolution with the hole density term $\rden_h$.}
Here we get from \cref{ft_rden_hle_gau_ex_append}
\begin{align}
	\int_{\Rset} K_{c}(\clp_{a}-\tau) \rden_{h}(\tau-\hle_{b})d\tau
	&=
	\frac{1}{2\pi}\int_{\Rset} e^{i(\clp_{a}-\hle_{b})} e^{-\frac{3|t|}{2}} dt 
	=
	K_{2/3}(\clp_{a}-\hle_{b})
	.
	\label{gau_ex_conv_append_rden_hle}
\end{align}
\subsection*{Convolution with the complex root density term $\ho{\rden}$}
From \cref{ft_ho_rden_cmplx_gau_ex_gen_append} we get three different scenarios:
\minisec{Close-pair}
\begin{align}
	\int_{\Rset}K_{c}(\clp_{a}-\tau)\ho{\rden}(\tau-\clp_{b})d\tau
	&=
	K(\clp_{a}-\clp_{b})-K_{1/2}(\clp_{a}-\clp_{b})
	.
	\label{gau_ex_conv_append_rden_ho_clp}
\end{align}
\minisec{Wide-pair (positive)}
\begin{align}
	\int_{\Rset}K_{c}(\clp_{a}-\tau)\ho{\rden}(\tau-\wdp_{b})d\tau
	&=
	K(\clp-\wdp-i)
	.
	\label{gau_ex_conv_append_rden_ho_wdp+}
\end{align}
\minisec{Wide-pair (negative)}
\begin{align}
	\int_{\Rset}K_{c}(\clp_{a}-\nu)\ho{\rden}(\nu-\wdp*_{b})
	d\nu
	&=
	K(\clp-\wdp*+i)
	.
	\label{gau_ex_conv_append_rden_ho_wdp-}
\end{align}
\minisec{Combined result}
Combining the results from \crefrange{gau_ex_conv_append_rden_g_i2}{gau_ex_conv_append_rden_ho_wdp-} and through the expansion \eqref{gau_ex_II_gen_Gf_den_expn_rl_str} for the $\Gf[e]$ we obtain its convolution with the $K_c$:
\begin{multline}
	\frac{1}{2\pi i}
	\int_{\Rset+i\epsilon}
	K_{c}(\clp_{a}-\nu)
	\Gf[e](\nu,\la_{k})
	d\nu
	=
	\\
	(1+\aux_{e}(\la_{k}))\, K(\clp<+>_{a}-\la_{k})
	-\bmsum K_{2}(\clp_{a}-\bm\hle)
	\frac{%
	\Gf[e](\bm\hle,\la_{k})
	}{%
	\aux_{e}^\prime(\bm\hle)
	}
	\\
	-\bmsum \left\lbrace
	\left(
	K(\clp_{a}-\bmclp)
	-
	K_{1/2}(\clp_{a}-\bmclp)
	\right)
	\Fmat*[e]^{\txtcp}_{a,k}
	\right\rbrace
	\\
	-\bmsum K(\clp_{a}-\bmwdp-i)
	\Fmat*[e]^{\txtwp,+}_{a,k}
	-\bmsum K(\clp_{a}-\bmwdp*+i)
	\Fmat*[e]^{\txtwp,-}_{a,k}
	\label{Gf_II_gen_int_with_str_kernel_append}
	.
\end{multline}
\subsection{Convolutions with Lieb kernel \texorpdfstring{$K$}{} shifted into the wide-pair region}
We recall that some of these convolutions are also computed in \cref{sec:den_complex}.
Here we recall these results it in the specific context of wide-pairs.
\minisec{Convolutions with $\rden_2$.}
From \cref{ft_rden_gs_gau_ex_append} we get
\begin{subequations}
\begin{gather}
	\int_{\Rset}
	K(\wdp<+>_{a}-\tau)
	\rden_g(\tau)
	d\tau
	=
	\frac{1}{\pi}\int_{0}^{\infty}\sinh\frac{t}{2} e^{i\wdp<+>_at} dt
	=
	K_{2}(\wdp_{a})
	=
	\frac{1}{2\pi i} t(\tfrac{i}{2}-\wdp<+>_{a})
	,
	\label{conv_gau_ex_append_wdp_+_rden_g_i2}
\shortintertext{and,}
	\int_{\Rset}
	K(\wdp*<->_{a}-\tau)
	\rden_g(\tau)
	d\tau
	=
	\frac{1}{\pi}
	\int_{0}^{\infty}\sinh\frac{t}{2} e^{-i(\wdp*<->_a)t} dt
	=
	K_{2}(\wdp*_{a}-i)
	=
	\frac{1}{2\pi i}
	t(\tfrac{i}{2}-\wdp*<->_{a})
	.
	\label{conv_gau_ex_append_wdp_-_rden_g_i2}
\end{gather}
\label{conv_gau_ex_append_wdp_rden_g_i2}
\end{subequations}
Similarly for the shifted $\rden_2(\nu,\la-\tfrac{i}{2}+i0)$ \eqref{ft_rden_gs_gau_ex_append_shftd} we get
\begin{subequations}
\label{conv_gau_ex_append_wdp_rden_shft}
\begin{multline}
	\int_{\Rset}
	K(\wdp<+>_{a}-\tau)
	\rden_2(\tau,\la-\tfrac{i}{2}+i0)
	d\tau
	=
	\frac{1}{\pi}\int_{0}^{\infty}\sinh\frac{t}{2} e^{i(\wdp_a-\la_{k})t} dt
	\\
	=
	K_{2}(\wdp_{a}-\la_{k}) 
	= 
	\frac{1}{2\pi i}
	t(\la_{k}-\wdp<+>_{a})
	,
	\label{conv_gau_ex_append_wdp+_rden_shft}
\end{multline}
and, 
\begin{multline}
	\int_{\Rset}
	K(\wdp*<->_{a}-\tau)
	\rden_2(\tau,\la-\tfrac{i}{2}+i0)
	d\tau
	=
	\frac{1}{\pi}\int_{0}^{\infty}
	\sinh\frac{t}{2} e^{-t} e^{-i(\wdp*-\la_{k})t} dt
	\\
	=
	K_{2}(\wdp*_{a}-\la_{k}-i)
	=
	\frac{1}{2\pi i}
	t(\la_{k}-\wdp*<->_{a})
	.
	\label{conv_gau_ex_append_wdp-_rden_shft}
\end{multline}
\end{subequations}
\minisec{Convolutions with the density term $\rden_h$}
Now for the convolution with density term for the hole we get
\begin{subequations}
\label{conv_gau_ex_append_wdp_rden_hle}
\begin{align}
	\int_{\Rset} K(\wdp<+>_{a}-\tau) \rden_{h}(\tau-\hle_{b})d\tau
	&=
	\frac{1}{\pi}
	\int_{0}^{\infty} \sinh\frac{t}{2} e^{-t} e^{i(\wdp_{a}-\hle_{b})}
	dt
	=
	\frac{1}{2\pi i}
	t(\wdp<+>_{a}-\hle_{b})
	,
	\label{conv_gau_ex_append_wdp+_rden_hle}
\shortintertext{and,}
	\int_{\Rset} K(\wdp*<->_{a}-\tau) \rden_{h}(\tau-\hle_{b})d\tau
	&=
	\frac{1}{\pi}
	\int_{0}^{\infty}  \sinh\frac{t}{2} e^{-t} e^{-i(\wdp*_{a}-\hle_{b})}
	dt
	=
	\frac{1}{2\pi i}
	t(\hle_{b}-\wdp*<->_{a})
	\label{conv_gau_ex_append_wdp-_rden_hle}
	.
\end{align}
\end{subequations}
\subsection*{Convolutions with the density term $\ho{\rden}$.}
According to \cref{ft_ho_rden_cmplx_gau_ex_gen_append} we consider the three scenarios:
\minisec{Close-pair.}
\begin{subequations}
\label{conv_gau_ex_append_wdp_rden_ho_clp}
\begin{align}
	\int_{\Rset}
	K(\wdp<+>_{a}-\tau)
	\ho{\rden}(\tau-\clp_{b})
	d\tau
	&=
	\frac{1}{\pi}
	\int_{0}^{\infty}
	\sinh t e^{-t} e^{i(\wdp_{a}-\clp_{b})t}
	dt
	=
	K(\wdp_{a}-\clp_{b}+i)
	.
	\label{conv_gau_ex_append_wdp+_rden_ho_clp}
\shortintertext{And}
	\int_{\Rset}
	K(\wdp*<->_{a}-\tau)
	\ho{\rden}(\tau-\clp_{b})
	d\tau
	&=
	\frac{1}{\pi}
	\int_{0}^{\infty}
	\sinh t e^{-t} e^{-i(\wdp*_{a}-\clp_{b})t}
	dt
	=
	K(\wdp*_{a}-\clp_{b}-i)
	.
	\label{conv_gau_ex_append_wdp-_rden_ho_clp}
\end{align}
\end{subequations}
\minisec{Wide-pair (like-terms).}
\begin{subequations}
\label{conv_gau_ex_append_wdp_rden_ho_wdp_like}
\begin{align}
	\int_{\Rset} K(\wdp<+>_{a}-\tau)
	\ho{\rden}(\tau-\wdp_{b}) d\tau &=0
	\label{conv_gau_ex_append_wdp+_rden_ho_wdp_like}
	,
\shortintertext{and,}
	\int_{\Rset} K(\wdp*<->_{a}-\tau)
	\ho{\rden}(\tau-\wdp*_{b}) d\tau &=0
	.
	\label{conv_gau_ex_append_wdp-_rden_ho_wdp_like}
\end{align}
\end{subequations}
\minisec{Wide-pair (cross-terms).}
\begin{subequations}
\label{conv_gau_ex_append_wdp_rden_ho_wdp_cross}
\begin{multline}
	\int_{\Rset} K(\wdp<+>_{a}-\tau)
	\ho{\rden}(\tau-\wdp*_{b}) d\tau =
	\frac{1}{\pi}
	\int_{0}^{\infty}
	\sinh t (e^{-t}-1) e^{i(\wdp_{a}-\wdp*_{b})t}
	dt
	\\
	=
	K(\wdp_{a}-\wdp*_{b}+i)-K(\wdp_{a}-\wdp*_{b})
	,
	\label{conv_gau_ex_append_wdp+_rden_ho_wdp_cross}
\end{multline}
and,
\begin{multline}
	\int_{\Rset} K(\wdp*<->_{a}-\tau)
	\ho{\rden}(\tau-\wdp_{b}) d\tau =
	\frac{1}{\pi}
	\int_{0}^{\infty}
	\sinh t (e^{-t}-1) e^{-i(\wdp*_{a}-\wdp_{b})t}
	dt
	\\
	=
	K(\wdp*_{a}-\wdp_{b}-i)-K(\wdp*_{a}-\wdp_{b})
	.
	\label{conv_gau_ex_append_wdp-_rden_ho_wdp_cross}
\end{multline}
\end{subequations}
\minisec{Combined result}
With all the results from \crefrange{conv_gau_ex_append_wdp_rden_g_i2}{conv_gau_ex_append_wdp_rden_ho_wdp_cross} we find the following expressions for the convolution of the $\Gf[e]$ \eqref{gau_ex_II_gen_Gf_den_expn_rl_str} with kernel $K$ shifted to wide-pair region:
\begin{subequations}
\label{Gf_int_with_wdp_ker_append_gau_ex}
\begin{multline}
	\frac{1}{2\pi i}
	\int_{\Rset+i\epsilon}
	K(\wdp<+>_{a}-\nu)
	\Gf[e](\nu,\la_{k})
	d\nu
	=
	\frac{(1+\aux_{e}(\la_{k}))}{2\pi i} t(\la_{k}-\wdp<+>_{a})
	\\
	-\bmsum K_{2}(\wdp_{a}-\bm\hle) 
	\frac{%
	\Gf[e](\bm\hle,\la_{k})
	}{%
	\aux_{e}^\prime(\bm\hle)
	}
	\\
	-\sum_{b=1}^{n_\txtcp}	
	K(\wdp_{a}-\clp_{b}+i) \Fmat*[e]^{\txtcp}_{a,k}
	\\
	-\sum_{b=1}^{n_\txtwp}
	\left\lbrace
	K(\wdp_{a}-\wdp*_{b}+i) 
	-
	K(\wdp_{a}-\wdp*_{b}) 
	\right\rbrace
	\Fmat*[e]^{\txtwp,-}_{a,k}
	.
	\label{Gf_conv_K-shft_wdp+_ker_append_gau_ex}
\end{multline}
And,
\begin{multline}
	\frac{1}{2\pi i}
	\int_{\Rset+i\epsilon}
	K(\wdp*<->_{a}-\nu)
	\Gf[e](\nu,\la_{k})
	d\nu
	=
	\frac{(1+\aux_{e}(\la_{k}))}{2\pi i} t(\la_{k}-\wdp*<->_{a})
	\\
	-\bmsum K_{2}(\wdp*_{a}-\la_{k}) 
	\frac{%
	\Gf[e](\bm\hle,\la_{k})
	}{%
	\aux_{e}^\prime(\bm\hle)
	}
	\\
	-\sum_{b=1}^{n_\txtcp}	
	K(\wdp*_{a}-\clp_{b}-i) \Fmat*[e]^{\txtcp}_{a,k}
	\\
	-\sum_{b=1}^{n_\txtwp}
	\left\lbrace
	K(\wdp*_{a}-\wdp_{b}-i) 
	-
	K(\wdp*_{a}-\wdp_{b}) 
	\right\rbrace
	\Fmat*[e]^{\txtwp,+}_{a,k}
	.
	\label{Gf_conv_K-shft_wdp-_append_gau_ex}
\end{multline}
\end{subequations}
\end{subappendices}
\clearpage{}%
\clearpage{}%
\chapter[tocentry={Cauchy-Vandermonde extraction and reduced determinant representation}, head={CV extraction and reduced determinant representation}]{\mbox{Cauchy-Vandermonde extraction} and \mbox{reduced determinant representation} for \mbox{higher form-factors}}
\label{chap:gen_FF}%
In this chapter we will compute the thermodynamic limit of form-factors of higher-spinon excitations starting from the modified Cauchy determinant representation \eqref{cau_det_rep_gen}.
In the two-spinon case, we saw in \cref{chap:2sp_ff} that it involves extracting the Cauchy matrices into the prefactors, followed by obtaining the thermodynamic limit from the infinite product form of prefactors.
A similar procedure will be used here to show that 
form-factors of higher-spinon excitations admit a \emph{reduced} determinant representation \eqref{red_det_rep_generic}. It is reproduced in the following expression:
\begin{multline}
	\left|\FF^{z}\right|^2=
	(-1)^{\frac{n_h+2}{2}}
	M^{-n_h}
	2^{\frac{n_h(n_h-1)+2}{2}}	
	\pi^{\frac{n_h(n_h-3)+2}{2}}	
	\frac{\prod_{a=1}^{\ho{n}}\prod_{b=1}^{n_h}(\cid_a-\hle_b-\frac{i}{2})}{\prod_{a,b=1}^{\ho{n}}(\cid_a-\cid_b-i)}
	\\
	\times
	\frac{1}{G^{2n_h}(\frac{1}{2})}
	\prod_{\underset{a\neq b}{a,b=1}}^{n_h}
	\frac{%
	G(\frac{\hle_a-\hle_b}{2i})
	G(1+\frac{\hle_a-\hle_b}{2i})
	}{%
	G(\frac{1}{2}+\frac{\hle_a-\hle_b}{2i})
	G(\frac{3}{2}+\frac{\hle_a-\hle_b}{2i})
	}
	~
	\frac{%
	\det_{\ho{n}}\resmat*[g]
	\det_{n_h}\resmat*[e]
	}{\det\vmat[\bm\hle]}
	.
	\label{red_det_rep_gen_intro_sec}
\end{multline}
The word `reduced' signifies that the matrices $\resmat*[g]$ and $\resmat*[e]$ are of the finite size.
In comparison to the two-spinon case, there are some crucial differences in the computations that lead us to this result.
As we have already realised in the previous \cref{chap:cau_det_rep_gen}, the procedure of Cauchy extraction must be reformulated to accomodate the rectangular nature of the Cauchy matrices $\Cmat[\bm{\check\rl^+}\Vert\bm\la]$ and $\Cmat[\bm{\check\la}\Vert\bm{\rh^+}]$, with which we are coerced into dealing since they are the largest Cauchy matrices contained inside the determinant representation \eqref{cau_det_rep_gen}.
\\
In this chapter, we will see that a rectangular Cauchy extraction can be indeed realised using the Cauchy-Vandermonde matrix, since the latter is known to generalise the Cauchy determinant formula to the rectangular case.
We begin this chapter with a brief introduction to the Cauchy-Vandermonde (CV) matrix and its properties in \cref{sec:cau_van_mat}.
Initially, we introduce the CV matrix in the rational parametrisation due to its relative simplicity.
Its properties are then generalised to the hyperbolic case using the the dressing relation.
The discussion in this section is supplemented by \cref{chap:mat_det_extn} as well as \cref{sec:toy_cv_extn_append} at the end of the chapter where we provide a \emph{toy example} of extraction with Cauchy-Vandermonde matrices.
\\
In \cref{sec:cau_van_extn}, the theory of CV extraction is applied to the form-factors.
We start with the determinant representation \eqref{cau_det_rep_gen} and compute the actions of the inverse CV matrices $\Cmat<V>[\bm{\la}\Vert\bm{\check\rl^+}]$ and $\Cmat<V>[\bm{\rh^+}\Vert\bm{\check\la}]$ onto the matrices $\modCau[g]$ and $\modCau[e]$ respectively.
The sets $\bm{\check\rl^+}$ and $\bm{\rh^+}$ contains the positive half of the close-pairs, they were defined in \cref{notn:real_clp+_roots,notn:real_clp+_hle_roots} in the previous \cref{chap:cau_det_rep_gen}.
\\
In \cref{sec:tdl_pref_gen}, we compute the thermodynamic limit of the Cauchy-Vandermonde determinants extracted in the previous step, together with the prefactors of the determinant representation.
This procedure is similar to the two-spinon case studied in \cref{chap:2sp_ff} as we first express the infinite determinants prefactors as infinite products of auxiliary function $\Omegfn$. In this way, the thermodynamic limit of the form-factors can be accessed in terms of auxiliary function leading us to a reduced determinant representation \eqref{red_det_rep_gen_intro_sec}.
However, the complex roots adds significant differences in these computations of prefactors which will be addressed here. The discussion is supplemented by \cref{sec:pref_append_inf_prod} where all long formulae and auxiliary computations are presented.
Finally, an example of the four-spinon form-factor is discussed, starting from the reduced determinant representation \eqref{red_det_rep_gen_intro_sec} for the a higher-spinon case.
There we find that we can perform a \emph{little CV extraction}: a smaller, finite, and rational version of the CV extraction that further reudces the determinants of matrices $\resmat*[g]$ and $\resmat*[e]$ to a sum over a relatively small number of terms.
\section{Cauchy-Vandermonde matrix}
\label{sec:cau_van_mat}
Let us begin with the traditional Cauchy and Vandermonde matrices in the rational parametrisation.
We will be using the notation $\ptn\la$ (or $\ptn*\la$) to denote a partition of integers, which was introduced on the \cpageref{ptn_notn_page_begin}.
In particular, we will frequently use the partitions of consecutive integers $\ptn\delta$ and consecutive even or odd integers $\ptn\gamma$ which are defined in the following.
\begin{notn}
\label{notn:ptn_consecutive}
\index{cv@\textbf{Cauchy-Vandermonde}!ptn consec@$\ptn*\delta(n)$ or $\ptn\delta(n)$: partition of consecutive integers of length $n$|textbf}%
\begin{subequations}
Given an integer $n$ we denote the partition of consecutive integers of length $n$ by,
\begin{align}
	\ptn*\delta(n)=
	\set{n-1,n-2,\ldots,0}
	\label{ptn_consec}
	.
\end{align}
The argument for its length $n$ is often dropped from the notation unless we feel necessary to mention it explicitly.
We denote by $\ptn\delta$ the same partition with the reversed order.
\\	
Similarly we also define a partition of consecutive, even or odd integers:
\begin{equation}
\begin{aligned}
	\ptn*\gamma&=
	\set{n-1,n-3,\ldots,0}
	,
	&
	&(\text{for } n \text{ odd})
	.
	\\
	\ptn*\gamma&=
	\set{n-1,n-3,\ldots,1}
	,
	&
	&(\text{for } n \text{ even})
	.
\end{aligned}
\label{ptn_consec_even-odd}
\end{equation}
\end{subequations}
\index{cv@\textbf{Cauchy-Vandermonde}!ptn consec@$\ptn*\gamma(n)$ or $\ptn\gamma(n)$: partition of consecutive even/odd integers of length $\lfloor \frac{n}{2}\rfloor$|textbf}%
We shall denote by $\ptn\gamma$ same partition with the reversed order.
We can see that the partition \eqref{ptn_consec_even-odd} is contained in the partition \eqref{ptn_consec} $\ptn{\gamma}\subset\ptn{\delta}$ and its length is $\ell(\ptn\gamma)=\left\lfloor \frac{n}{2} \right\rfloor$.
\end{notn}
\begin{notn}[Vandermonde matrix]
\begin{subequations}
Given a set of variables $\bm x$ and partition $\ptn\delta(n)$ of length $n$, we define a Vandermonde matrix $\vmat_{\ptn\delta}$ as
\begin{align}
	\vmat<\ptn{\delta}(n)>[\bm x]
	=
	\bm{\big[}\bm x^{\ptn\delta(n)}\bm{\big]}
	=
	\begin{pmatrix}
	1	& x_{1} & \cdots &  x_{1}^{n-1}		
	\\
	\vdots	&	\vdots	& \ddots	& \vdots
	\\
	1	&	x_{n} & \cdots &  x_{n}^{n-1}		
	\end{pmatrix}
	\label{rat_van_rect}
	.
\end{align}
\index{cv@\textbf{Cauchy-Vandermonde}!mat van@$\vmat<\ptn\delta>[\cdot]$: Vandermonde matrix (rational)|textbf}%
Note that when the Vandermonde matrix $\vmat_{\ptn{\delta(n)}}(\bm x)$ is a square matrix, the length of partition is predetermined by $n=n_{\bm x}$, hence we will simply denote $\vmat<\ptn\delta>[\bm x]$, unless we want to denote a rectangular Vandermonde matrix.
\\
The determinant of a Vandermonde matrix is a totally skew-symmetric polynomial, also known as an alternant.
Here we denote it using the notation $\bmalt$ defined on \cpagerefrange{ind_free_notn}{ind_free_notn_end}.
\index{cv@\textbf{Cauchy-Vandermonde}!Det van@$\bmalt(\cdot)$: Vandermonde determinant|textbf}%
\begin{align}
	\det\vmat_{\ptn\delta}[\bm x]=\bmalt(\bm x)
	=
	\prod_{a>b}(x_a-x_b)
	.
\end{align}
\end{subequations}
\end{notn}
\begin{notn}[Rational Cauchy matrix]
\begin{subequations}
Given a set of complex variables $\bm x$ ($n_{\bm x}=n$) and $\bm y$ ($n_{\bm y}=m$) which are disjoint $\bm x\cap\bm y=\varnothing$, we define the Cauchy matrix $\cmat[\bm x\Vert\bm y]$ as
\begin{align}
	\cmat[\bm{x}\Vert\bm{y}]&=
	\bm{\bigg[}
	\frac{1}{\bm{x}-\bm{y}}
	\bm{\bigg]}
	=
	\begin{pmatrix}
	\frac{1}{x_{1}-y_{1}}	&	\cdots &	\frac{1}{x_{1}-y_{m}}
	\\
	\vdots	&	\ddots	&	\vdots	
	\\
	\frac{1}{x_{n}-y_{1}}	&	\cdots &	\frac{1}{x_{n}-y_{m}}
	\end{pmatrix}
	.
	\label{rat_cau_mat_rect}
\end{align}
When the cardinalities of the two sets are equal $n_{\bm x}=n_{\bm y}=n$, we obtain a square matrix. Its determinant is given by a well-known formula, which can be written in the superalternant notation as
\begin{align}
	\det\cmat[\bm x\Vert\bm y]
	=
	\bmalt(\bm{x}\Vert\bm{y})
	=
	\frac{%
	\bmalt(\bm{x})\bmalt(-\bm{y})
	}{%
	\bmalt(\bm{x}-\bm{y})
	}
	=
	\frac{%
	\prod_{j>k}^{n}(x_{j}-x_{k})
	\prod_{j>k}^{n}(y_{k}-y_{j})
	}{%
	\prod_{j,k=1}^{n}(x_{j}-y_{k})
	}
	.
	\label{det_sq_rat_cau_alt_notn}
\end{align}
\end{subequations}
\index{cv@\textbf{Cauchy-Vandermonde}!det Cv@$\bmalt\varphi(\cdot\Vert\cdot)$: Cauchy(-Vandermonde) determinant|textbf}%
\end{notn}
We shall now define a mixed version of a Cauchy \eqref{rat_cau_mat_rect} and Vandermonde \eqref{rat_van_rect} matrices.
It will be always defined as a square matrix of order $n+m$ consisting of rectangular Cauchy block of $m$ columns and Vandermonde block of $n$ columns.
\begin{defn}[Rational Cauchy-Vandermonde determinant]
\label{defn:rat_cau_van}
\begin{subequations}
Given set of variables $\bm x$ ($n_{\bm x}=n+m$) and $\bm y$ ($n_{\bm y}=m$) and, a partition of integers $\ptn\delta$ of length $n$. We define the matrix $\cmat_{\ptn\delta}[\bm x\Vert\bm y]$ as
\begin{align}
	\cmat_{\ptn \delta}[\bm x\Vert\bm y]
	=
	\left(
	\cmat[\bm x\Vert\bm y]
	~\bigg|~
	\vmat_{\ptn\delta(n)}[\bm x]
	\right)
	.
	\label{rat_cau_van_mat}
\end{align}	
The subscript $\ptn\delta$ tells us that the Vandermonde block is made by partition of consecutive integers \eqref{ptn_consec}. Its length is $n$ as it is determined by the difference of cardinalities of the sets $\bm x$ and $\bm y$.
We show in \cref{lem:rat_cau_van_det} that its determinant generalises the Cauchy determinant formula \eqref{det_sq_rat_cau_alt_notn} that is produced below:
\begin{align}
	\det\cmat<\ptn\delta>(\bm x|\bm y)
	=
	\bmalt(\bm x\Vert\bm y)
	=
	\frac{%
	\prod_{j>k}^{n+m}(x_{j}-x_{k})
	\prod_{j>k}^{m}(y_{k}-y_{j})
	}{%
	\prod_{j=1}^{n+m}\prod_{k=1}^{m}(x_{j}-y_{k})
	}
	.
	\label{rat_cau_van_det_super-alt_notn}
\end{align}
\end{subequations}
\end{defn}
It is worthwhile to note that just like the Vandermonde determinant is related with the symmetric Schur function $s_{\ptn\la}(\bm x)$:
\begin{subequations}
\begin{align}
	s_{\ptn\la}(\bm x)=
	\frac{\det\vmat_{\ptn\la+\ptn \delta}[\bm x]}{\det\vmat_{\ptn\delta}[\bm x]}
	\label{sym_schur_fn_van_det}
	,
\end{align}
Cauchy-Vandermonde determinants are also related with the supersymmetric Schur function $s_{\ptn\la}(\bm x\Vert\bm y)$ according to the following determinant representation:
\begin{align}
	s_{\ptn\la}(\bm x\Vert\bm y)=
	\frac{\det\cmat_{\ptn\la+\ptn \delta}[\bm x\Vert\bm y]}{\det\cmat_{\ptn\delta}[\bm x\Vert\bm y]}
	.
	\label{supersym_schur_fn_cau_van_det}
\end{align}
\end{subequations}
The representation \eqref{supersym_schur_fn_cau_van_det} of the supersymmetric Schur function was proved in \cite{MoeV03} for any arbitrary partition $\ptn\la$.
However, in this thesis, we do not need the generality that is offered by supersymmetric Schur functions, since we are only concerned with the inversion of a Cauchy-Vandermonde matrix.
For this purpose, it would be sufficient to consider a special type of the partitions $\ptn\la_r=\ptn\delta+\ptn{1^k}$, which has a single jump across an index $r$ while it is otherwise consecutive.
\begin{align}
	\ptn\la_r=	
	\set{0,1,\ldots,r-1,r+1,\ldots,n}
	.
\end{align}
For such a partition, the Schur function is reduced to elementary function [see \cref{ele_susy_append} for definition].
Hence we get the following identity:
\index{cv@\textbf{Cauchy-Vandermonde}!susy fn ele@$e_r(\cdot\Vert\cdot)$: elementary supersymmetric functions}%
\begin{align}
	e_{r}(\bm x\Vert\bm y)
	=
	\frac{\det\cmat<\ptn\la_r>[\bm x\Vert\bm y]}{\det\cmat<\ptn\delta>[\bm x\Vert\bm y]}
	;
\end{align}
which is demonstrated in \cref{lem:rat_cv_det_quotient} from \cref{chap:mat_det_extn} of this thesis.
There we find that the Cauchy-Vandermonde matrix $\cmat<\ptn\la_r>$ in this case has a form similar to \cref{rat_cau_van_mat}, while it skips over a column in the Vandermonde block, which is exactly what we require for the inversion:
\begin{multline}
	\cmat<\ptn\lambda_r>[\bm x\Vert\bm y]
	=
	\Big[
	\cmat[\bm x\Vert\bm y]
	~\Big|~
	\vmat_{\ptn\lambda_r}[\bm x]
	\Big]
	\\
	=
	\left(
	\begin{array}{ccc|cccccc}
	\frac{1}{x_1-y_1}		&		\cdots 		&		\frac{1}{x_1-y_m}
	&
	1 	& \cdots 	&	x_{1}^{r-1} & x_1^{r+1}	&	\cdots &	x_1^{n-1}
	\\
	\vdots 	&	\ddots	&	\vdots
	&
	\vdots	& \ddots & \vdots & \vdots 	& \ddots	& \vdots
	\\
	\frac{1}{x_{n+m}-y_1}		&		\cdots 		&		\frac{1}{x_{n+m}-y_m}
	&
	1 	& \cdots 	&	x_{n+m}^{r-1} & x_{n+m}^{r+1}	&	\cdots &	x_{n+m}^{n-1}
	\end{array}
	\right)
	.
	\label{susy_ele_lem_append_jump_cv_mat_form_chap}
\end{multline}
However it should be noted that the form \eqref{susy_ele_lem_append_jump_cv_mat_form_chap} cannot be used to write the identity \eqref{supersym_schur_fn_cau_van_det} for an arbitrary partition.
For an arbitrary partition, \textcite{MoeV03} have shown that we need to add Vandermonde blocks on both sides of the Cauchy matrix, where these two Vandermonde blocks are defined with partitions, that are dual to each other.
\subsection{Inversion and duality of the Cauchy-Vandermonde matrix}
\label{sub:rat_cau_van_inv}
In \cref{lem:rat_cv_inv_append}, we compute the inverse of a Cauchy-Vandermonde matrix $C_{\ptn\delta}[\bm x\Vert\bm y]$ in rational parametrisation.
There we also saw that it can be represented as a dressing of the dual Cauchy-Vandermonde matrix by diagonal matrices formed by $\phifn$ functions as
\begin{align}
	\left(\cmat_{\ptn{\delta}}[\bm{x}\Vert\bm{y}]\right)^{-1}
	&=
	\diag\Big[
	\phifn^\prime(\bm{y}|\bm{x},\bm{y})
	~\Big|~
	\Id_{n}
	\Big]
	\cdot
	\left(\cmat*_{\ptn{\delta}}[\bm{-x}\Vert\bm{-y}]\right)^{T}
	\cdot
	\diag\Big[
	\phifn^\prime(\bm{x}|\bm{y},\bm{x})
	\Big]
	\label{cau_van_inv_diag_dress}
\end{align}
where the dual Cauchy-Vandermonde matrix is obtained by replacing the columns with monomial $x^{r}$ in the Vandermonde block from \cref{rat_cau_van_mat} by supersymmetric elementary polynomials of degree $n-r-1$ [see \cref{defn:ele_susy_append} in \cref{chap:mat_det_extn}].
\begin{subequations}
\begin{gather}
	\cmat*_{\ptn\delta}[\bm x\Vert\bm y]
	=
	\left[
	\cmat[\bm x\Vert\bm y]
	~\Big|~
	\vmat*_{\ptn\delta}[\bm x\Vert\bm y]
	\right]
\shortintertext{where}
	\vmat*<\ptn\delta>_{a,r}[\bm x\Vert\bm y]
	=
	e_{n-r-1}(\bm{x_{\hat{a}}}\Vert\bm y)
	.
\end{gather}
\label{dual_rat_cau_van_mat}
\end{subequations}
\Cref{cau_van_inv_diag_dress} also tells that the determinants of the two matrices are equal.
\begin{align}
	\det \cmat_{\ptn\delta}[\bm x\Vert\bm y]=
	\det \cmat*_{\ptn\delta}[\bm x\Vert\bm y].
	\label{det_rat_cau_et_its_dual}
\end{align}
We say that \crefrange{cau_van_inv_diag_dress}{det_rat_cau_et_its_dual} portrays a \emph{duality} between these two matrices since we can also show that the inverse of $C^\ast_{\ptn\delta}$ matrix \eqref{dual_rat_cau_van_mat} can be written as the diagonal dressing of the Cauchy-Vandermonde matrix \eqref{rat_cau_van_mat} as
\begin{align}
	\left(\cmat*_{\ptn{\delta}}[\bm{x}\Vert\bm{y}]\right)^{-1}
	&=
	\diag\Big[
	\phifn^\prime(\bm{y}|\bm{x},\bm{y})
	~\Big|~
	\Id_{n}
	\Big]
	\cdot
	\left(\cmat_{\ptn{\delta}}[\bm{-x}\Vert\bm{-y}]\right)^{T}
	\cdot
	\diag\Big[
	\phifn^\prime(\bm{x}|\bm{y},\bm{x})
	\Big]
	.
	\label{dual_cau_van_inv_diag_dress}
\end{align}
This duality becomes a very powerful tool when it comes to extraction.
First of all we can compare the expressions from the two \cref{dual_cau_van_inv_diag_dress,cau_van_inv_diag_dress} with a similar expression \eqref{cau_hyper_inv_diag_dress} we wrote earlier in the case of square Cauchy matrix of the hyperbolic parametrisation in \cref{lem:cau_hyper_inv_diag_dress}.
We can see with the determinants computed from both \cref{cau_hyper_inv_diag_dress,dual_cau_van_inv_diag_dress} that both these matrices can be used for the extraction of a rational Cauchy-Vandermonde matrix.
\begin{subequations}
\begin{align}
	\det\left(\cmat_{\ptn{\delta}}[\bm{x}\Vert\bm{y}]\right)^{-1}
	&=
	\frac{%
	\det \cmat*_{\ptn\delta}[-\bm x\Vert-\bm y]%
	}{%
	\bmprod\phifn'(\bm x|\bm y,\bm x)%
	\bmprod\phifn'(\bm y|\bm x,\bm y)%
	}
	=
	\frac{1}{\bmalt(\bm x\Vert\bm y)}
\shortintertext{and}
	\det\left(\cmat*_{\ptn{\delta}}[\bm{x}\Vert\bm{y}]\right)^{-1}
	&=
	\frac{%
	\det \cmat_{\ptn\delta}[-\bm x\Vert-\bm y]%
	}{%
	\bmprod\phifn'(\bm x|\bm y,\bm x)%
	\bmprod\phifn'(\bm y|\bm x,\bm y)%
	}
	=
	\frac{1}{\bmalt(\bm x\Vert\bm y)}
	.
\end{align}
\end{subequations}
But it is very clear that the extraction with the dual \eqref{dual_cau_van_inv_diag_dress} is a more convenient choice since it prevents the supersymmetric polynomials entering the intermediate computation.
This distinction is more pronounced in the case of extraction with the Cauchy-Vandermonde matrix in hyperbolic parametrisation.
\subsection[Hyperbolic Cauchy-Vandermonde matrix and its extraction]{\mbox{Cauchy-Vandermonde matrix} in a \mbox{hyperbolic parametrisation} and its extraction}
\label{sub:hyper_cau_van_inv}
First let us introduce the following notation.
\begin{notn}
\index{cv@\textbf{Cauchy-Vandermonde}!mat VAN@$\Vmat<\ptn\gamma>[\cdot]$ or $\Vmat*<\ptn\gamma>[\cdot]$: Vandermonde matrix (circular, \textsl{equiv. upto det.})|textbf}%
\label{notn:hyp_vec_cau_van}
A vector-valued%
\footnote{Inside a matrix, it can be either a row or a column vector depending on the context. We will not explicitly write the transposition whenever it is clear from other indications.} %
function $\Vmat<\ptn\gamma>:\Cset\to\Cset^n$ is defined as
\begin{subequations}
\begin{align}
	\Vmat<\ptn{\gamma}>_{a}(\la)
	&=
	\cosh\pi\ptn*\gamma_{a}\la
	=
	\cosh\pi(n+1-2a)\la
	,
	&
	\text{for}\quad
	a&\leq \left\lceil\frac{n}{2}\right\rceil
	;
	\\
	\Vmat<\ptn{\gamma}>_{n+1-a}(\la)
	&=
	\sinh\pi\ptn*\gamma'_{a}\la
	=
	\sinh\pi(n+1-2a)\la
	,
	&
	\text{for}\quad
	a&\leq \left\lfloor\frac{n}{2}\right\rfloor
	.
\end{align}
\label{hyp_vec_cau_van}
\end{subequations}
where $\ptn\gamma$ and $\ptn*\gamma$ are the partitions introduced in \cref{notn:ptn_consecutive}.
The function $\Vmat*<\ptn\gamma>:\Cset\to\Cset^n$ is obtained by reversing the order $\Vmat*<\ptn\gamma>_{a}(\la)=\Vmat<\ptn\gamma>_{n+1-a}(\la)$.
Written explicitly, it is composed of
\begin{subequations}
\begin{align}
	\Vmat*<\ptn\gamma>_{a}(\la)
	&=
	\sinh\pi\ptn*\gamma'_{a}\la
	=
	\sinh\pi(n+1-2a)\la
	,
	&
	\text{for}\quad
	a&\leq \left\lfloor\frac{n}{2}\right\rfloor
	;
	\\
	\Vmat*<\ptn\gamma>_{n+1-a}(\la)
	&=
	\cosh\pi\ptn*\gamma_{a}\la
	=
	\cosh\pi(n+1-2a)\la
	,
	&
	\text{for}\quad
	a&\leq \left\lceil\frac{n}{2}\right\rceil
	.
\end{align}
\label{hyp_vec_cau_van_reversed}
\end{subequations}
The sub-index $\ptn\gamma$ maybe dropped from the notation $\Vmat<\ptn\gamma>$ or $\Vmat*<\ptn\gamma>$ when it is otherwise implicitly known from the context.
\end{notn}
Similar to the \cref{defn:rat_cau_van} we can define the hyperbolic version of the Cauchy-Vandermonde matrix can defined as
\begin{defn}
\label{defn:hyper_cau_van_mat}
Given two sets $\bm\alpha$ ($n_{\bm\alpha}=n+m$) and $\bm\beta$ ($n_{\bm\beta}=m$) of spectral parameters, the Cauchy-Vandermonde matrix $\Cmat<\ptn\gamma>(\bm\alpha|\bm\beta)$ is defined by its block structure
\begin{align}
	\Cmat<\ptn\gamma>[\bm\alpha\Vert\bm\beta]
	=
	\bigg[
	\diag\big[e^{-n\pi\bm\alpha}\big]
	\cdot
	\Ccal[\bm\alpha\Vert\bm\beta]
	\cdot
	\diag\big[e^{n\pi\bm\beta}\big]
	\;\bigg|\;
	\Vmat<\ptn\gamma>[\bm\alpha]
	\bigg]
	\label{hyper_cau_van_mat}
\end{align}
where the matrix $\Vmat<\ptn\gamma>[\bm\alpha]$ is composed of the row vectors given composed of the hyperbolic functions \eqref{hyp_vec_cau_van}:
\begin{align}
	\Vmat<\ptn\gamma>[\bm\alpha]
	=
	\bm{\big[}
	\Vmat<\ptn\gamma>(\bm\alpha)
	\bm{\big]}
	.
\end{align}
\end{defn}
We have obtained the form \eqref{hyper_cau_van_mat} of the Cauchy-Vandermonde matrix in \cref{sec:cau_van_mat_hyper} using the re-parametrisation:
\begin{alignat}{3}
	x_{j}&=e^{2\pi\alpha_{j}}
	,
	&
	\qquad
	&
	y_{j}&=e^{2\pi\beta_{j}}
	.
	\label{reparam_cau_van_hyper-rat}
\end{alignat}
on the Cauchy-Vandermonde matrix \eqref{rat_cau_van_mat} and its determinant \eqref{rat_cau_van_det_super-alt_notn}.
There we also show that its determinant produces
\begin{align}
	\det\Cmat<\ptn\gamma>[\bm\alpha\Vert\bm\beta]
	=
	\bmalt\sinh\pi(\bm\alpha\Vert\bm\beta)
	=
	\frac{%
	\prod_{j>k}^{n+m}\sinh\pi(\alpha_j-\beta_k)
	\prod_{j>k}^m\sinh\pi(\beta_k-\beta_j)
	}{%
	\prod_{j=1}^{n+m}
	\prod_{k=1}^{m}
	\sinh\pi(\alpha_j-\beta_k)
	}
	.
	\label{hyper_cau_van_det_superalt}
\end{align}
\begin{rem}
Let us observe that the matrices $\Cmat<\ptn\gamma>(-\bm\alpha|-\bm\beta)$ and $\Cmat<\ptn\gamma>(\bm\beta|\bm\alpha)$ are not related to each other simply by transposition property unlike in the case of a Cauchy matrix [see \cref{notn:rect_cau_mat_hyper}].
Reversing the sign also modifies the diagonal dressing terms whereas in the Vandermonde-like block we shall reverse the ordering $\Vmat<\ptn\gamma>\to \Vmat*<\ptn\gamma>$ [see \cref{notn:hyp_vec_cau_van}] to account for the change of sign (up-to determinant).
\begin{align}
	\Cmat<\ptn\gamma>[-\bm\alpha\Vert-\bm\beta]
	=
	\bigg[
	\diag\big[e^{n\pi\bm\alpha}\big]
	\cdot
	\Ccal[-\bm\alpha\Vert-\bm\beta]
	\cdot
	\diag\big[e^{n\pi\bm\beta}\big]
	\;\bigg|\;
	\Vmat*<\ptn\gamma>[\bm\alpha]
	\bigg]
	.
	\label{hyper_cau_van_mat_reversed}
\end{align}
For its determinant (or the superalternant), we can still write
\begin{align}
	\bmalt\sinh(\bm{\alpha}\Vert\bm{\beta})=\bmalt\sinh(-\bm{\beta}\Vert-\bm\alpha)
	.
\end{align}
\end{rem}
For the extraction of this matrix we find it convenient to use the formula \eqref{dual_cau_van_inv_diag_dress} which gives the following when written in the hyperbolic parametrisation
\begin{align}
	\Cmat*<\ptn\gamma>^{-1}[\bm\beta\Vert\bm\alpha]
	&=
	\diag_{\beta}\Big[
	\Phifn^\prime(\bm{\beta}|\bm{\alpha},\bm{\beta})
	\,\Big|\,
	\Id_{n}
	\Big]
	\cdot
	\left(\Cmat<\ptn\gamma>[-\bm{\alpha}\Vert-\bm{\beta}]\right)^{T}
	\cdot
	\diag_{\bm\alpha}\Big[
	\Phifn^\prime(\bm{\alpha}|\bm{\beta},\bm{\alpha})
	\Big]	
	.
	\label{hyper_cau_van_inv_dressing}
\end{align}
Alternatively we can denote it in terms of the dressed blocks $\Cdr$ and $\Vdr$ which are the hyperbolic equivalent of the similar dressed Cauchy and Vandermonde blocks in \cref{dual_cau_van_inv_diag_dress}.
\begin{align}
	\Cmat*<\ptn\gamma>^{-1}[\bm\beta\Vert\bm\alpha]
	=
	\begin{pmatrix}
	\Cdr[\bm\beta\Vert\bm\alpha]
	\\
	\Vdr<\ptn\gamma>[\bm\alpha]
	\end{pmatrix}
	\label{hyper_cau_van_inv_blocks}
	.
\end{align}
The dressed matrices $\Cdr$ and $\Vdr$ are described by,
\begin{subequations}
\begin{align}
	\Cdr[\bm\beta\Vert\bm\alpha]
	&=
	\diag_{\bm\beta}\big[e^{n\pi\bm\beta}\,\Phifn'(\bm\beta|\bm\alpha,\bm\beta)\big]
	\cdot
	\Cmat[\bm\beta\Vert\bm\alpha]
	\cdot
	\diag_{\bm\alpha}\big[e^{-n\pi\bm\alpha}\,\Phifn'(\bm\alpha|\bm\beta,\bm\alpha)\big]
	\label{hyper_cau_van_inv_cau_block}
	,
	\\[\jot]
	\Vdr<\ptn\gamma>[\bm\alpha]
	&=
	\Vmat*<\ptn\gamma>[\bm\alpha]\cdot
	\diag_{\bm\alpha}\big[
	\Phifn'(\bm\alpha|\bm\beta,\bm\alpha)
	\big]
	.
	\label{hyper_cau_van_inv_hvan_block}
\end{align}
\end{subequations}
\index{cv@\textbf{Cauchy-Vandermonde}!mat VDR@$\Vdr<\ptn\gamma>[\cdot]$ or $\Vdr*<\ptn\gamma>[\cdot]$: dressed Vandermonde matrix (circular, \textsl{equiv. upto det.})|textbf}%
We will now use \cref{cau_hyper_inv_diag_dress} in the next section to extract the Cauchy-Vandermonde matrices $\Cmat<\ptn\gamma>(\bm\la|\bm{\check\rl^+})$ and $\Cmat<\ptn\gamma>(\bm{\check\la}|\bm{\rh^+})$ from matrices $\modCau[g]$ and $\modCau[e]$ respectively.
\section{Cauchy-Vandermonde extraction for the form-factors}
\label{sec:cau_van_extn}%
In the previous chapter, we also highlighted some motives behind the choice of these sets and the corresponding Cauchy matrices for the extraction. To recall it briefly, the reasoning used the two basic observations about: \begin{enumerate}[noitemsep]%
\item splitting of close-pair blocks, one of which can be made to join the Cauchy block.
\item invariance of the number $\ho{n}$ of higher-level roots, which is also equal to the difference of cardinalities of the sets $\bm{\rh^+}$ and $\bm{\check\la}$, as well as, $\bm{\rl^+}$ and $\bm{\la}$.
\end{enumerate}
We will add to this list, a technical point that applies particularly for the extraction of the second type (i.e. for the matrix $\modCau[e]$).
An extraction of a larger CV matrix $\Cmat<\ptn\gamma>[\bm{\check\la}\Vert\bm{\rh^+}]$ containing all the Cauchy terms is always preferable to an extraction of smaller Cauchy matrix $\Cmat[\bm{\check\la}\Vert\bm\rh]$.
A construction of smaller matrix may have to exclude some of the Cauchy terms for the holes, since the cardinalities of the sets $\bm{\check\la}$ and $\bm{\rh}$ do not always match.
Such an exclusion must be avoided because: \begin{enumerate}%
\item it is subjective on the choice of the close-pair and wide-pair, a choice which we do not have, while speaking in absolute terms.
\item the inclusion of all the holes makes the computation easier since we have a rank-$n_h$ Cauchy matrix added in $\modCau[g]^{\text{cau}}$ \eqref{cau_det_rep_II_gen_cau_block}, as it is shown again in the following:
\begin{align}
	\modCau[e]^{\text{cau}}
	=
	\Cmat\left[\bm{\check\la}\big\Vert\bm{\rl^+}\right]
	+
	\Cmat\left[\bm{\check\la}\big\Vert\bm{\hle}\right]
	\cdot
	\Acal_e^{-1}\Rcal\left[\bm\hle\big|\bm{\rl^+}\right]
	.
	\label{cau_det_rep_II_gen_cau_block_chap:gen-FF}
\end{align}
\end{enumerate}
Let us denote the matrices $\resmat[g]$ and $\resmat[e]$ obtained from the extraction of the Cauchy-Vandermonde matrices from $\modCau[g]$ \eqref{cau_det_rep_gen_I_blocks} and $\modCau[e]$ \eqref{cau_det_rep_gen_II_blocks}.
These extractions are performed by taking a left-action of the inverse of dual Cauchy-Vandermonde matrices \eqref{dual_cau_van_inv_diag_dress} giving us the following expressions.
\index{ff@\textbf{Form-factors}!mat cvextn gs@$\resmat[g]$: mat. obtained from CV extraction for the ground state|textbf}%
\index{ff@\textbf{Form-factors}!mat cvextn es@$\resmat[e]$: mat. obtained from CV extraction for an excited state|textbf}%
Note that the choice of the negative sign in \cref{cau_ex_gen_I_mat} below can be attributed to the negative sign in the Cauchy block $\modCau[g]^{\text{cau}}$ \eqref{cau_det_rep_I_mat_cau_block}.
\begin{subequations}
\begin{align}
	\resmat[g]
	&=
	\Cmat*<\ptn\gamma>^{-1}[-\bm{\check\rl^+}\Vert-\bm\la]
	\cdot
	\modCau[g]
	,
	\label{cau_ex_gen_I_mat}
	\\
	\resmat[e]
	&=
	\Cmat*<\ptn\gamma>^{-1}[\bm{\rh^+}\Vert\bm{\check\la}]
	\cdot
	\diag
	\left[
	\modCau[e]
	~\Big|~
	\Id_{\ho{n}}
	\right]
	.
	\label{cau_ex_gen_II_mat}
\end{align}
\label{cau_ex_gen_mat}
\end{subequations}
The number of rows or columns in the $\Zcal$ blocks \eqref{hyper_cau_van_inv_blocks} for the extractions in both \cref{cau_ex_gen_I_mat,cau_ex_gen_II_mat} are equal to $\ho{n}$.
After this extraction, we determinants of the following two matrices into prefactors.
\begin{subequations}
\begin{align}
	\det\Cmat<\ptn\gamma>[-\bm\la\Vert-\bm{\check\rl^+}]
	&=
	\bmalt\sinh\pi(\bm{\check\rl^+}\Vert\bm\la)
	,
\shortintertext{and}
	\det\Cmat<\ptn\gamma>[\bm{\check\la}\Vert\bm{\rh^+}]
	&=
	\bmalt\sinh\pi(\bm{\check\la}\Vert\bm{\rh^+})
	.
\end{align}
\end{subequations}
The orders of these two matrices are $N_0$ and $N_0+\ho{n}+1$ respectively.
Substituting them into the representation \eqref{cau_det_rep_gen} for form-factors leads to
\begin{multline}
	\left|\FF^{z}\right|^2=
	-2\pi^{M+1}
	\frac{%
	\bmprod \revtf(\bm\la)
	}{%
	\bmprod \revtf(\bm\rl)
	}
	\frac{%
	\bmprod (\bm{\rl}-\bm\la) \bmprod(\bm\la-\bm\mu)
	}{%
	\bmprod^\prime (\bm{\rl}-\bm\mu) \bmprod^\prime (\bm\la-\bm\la)
	}
	\bmalt\sinh\pi(\bm{\check\rl^+}\Vert\bm\la)
	\bmalt\sinh\pi(\bm{\check\la}\Vert\bm{\rh^+})
	\\
	\times
	\frac{%
	\bmprod
	\baxq_{g}(\bmclp<+>-i)
	\baxq_{g}(\bmclp<->-i)
	}{%
	\bmprod
	\baxq_{e}^\prime(\bmclp<+>-i)
	\baxq_{e}(\bmclp<->-i)
	}
	\frac{%
	\bmprod
	\baxq_{g}(\bmwdp<+>-i)
	\baxq_{g}(\bmwdp*<->-i)
	}{%
	\bmprod
	\baxq_{e}(\bmwdp<+>-i)
	\baxq_{e}(\bmwdp*<->-i)
	}
	\\
	\times
	\det_{N_0}\resmat[g]
	\det_{N_0+\ho{n}+1}\resmat[e]
	.
	\label{cau_det_rep_gen_extr_cau}
\end{multline}
The expression \eqref{cau_det_rep_gen_extr_cau} that was presented above, generalises the representation \eqref{cau_det_rep_2sp_extr_cau} from the two-spinon case to the form-factors in a higher-spinon sector.
However, unlike the two-spinon case, we get two types of Cauchy extraction and resulting matrices $\resmat[g]$ and $\resmat[e]$.
We will compute them separately.
\\
It is important to observe that in both cases, we have chosen the action in such a way that ground state roots are summed over.
The reason behind it is however simple, we want to factorise the largest Cauchy portion inside the original determinants, which contains the contribution of the \emph{all} ground state roots $\bm\la$, in contrast to the excited state where it only contains the contributions of real roots and half the number of close-pairs $\bm{\rh^+}$.
\\
Since cardinalities of these sets are ordered as $n_{\bm{\check\rl^+}}<n_{\bm\la}<n_{\bm{\check\la}}<n_{\bm{\rh^+}}$, we see that there will be no mixing of Cauchy and Vandermonde blocks during the action of inverse matrix in \cref{cau_ex_gen_I_mat}.
\\
In the case of extraction \eqref{cau_ex_gen_II_mat}, we see that the cardinality of the set $\bm{\check\la}$ is inferior to $\bm{\rh^+}$.
Since we always take the Cauchy extraction such that ground state roots $\bm\la$ are summed over, we find that the extraction \eqref{cau_ex_gen_II_mat} can, in principle, involve mixing of Cauchy and Vandermonde sub-blocks.
Fortunately, this problematic scenario is averted in \eqref{cau_ex_gen_II_mat} due to the diagonal block structure of the matrix on which it applies.
\subsection{CV extraction of the first type}
\label{sub:cv_extn_I}%
Since both matrices $\Cmat*<\ptn\gamma>$ \eqref{hyper_cau_van_inv_blocks} and $\modCau[g]$ \eqref{cau_det_rep_gen_I_blocks} admit the block structure, we can divide $\resmat[g]$ into following blocks:
\begin{align}
	\resmat[g]=
	\begin{pmatrix}
	\resmat[g]^{n|\text{cau}}	&	\resmat[g]^{n|c}	&	\resmat[g]^{n|w+}	&	\resmat[g]^{n|w-}
	\\
	\resmat[g]^{s|\text{cau}}	&	\resmat[g]^{s|c}	&	\resmat[g]^{s|w+}	&	\resmat[g]^{s|w-}
	\end{pmatrix}
	.
\end{align}
All of the above blocks can be computed using \cref{lem:Phifn_sum_period_method_trivial_case} that was proved in the two-spinon case.
Moreover we can see that first two blocks $\resmat[g]^{n|\text{cau}}$ and $\resmat[g]^{s|\text{cau}}$ involve the action on a Cauchy matrix \eqref{cau_det_rep_gen_extr_cau} :
$%
	\modCau[g]^{\text{cau}}=%
	\Cmat[-\bm{\la}\Vert-\bm{\check\rl^+}].%
$
They are described by the expressions:
\begin{align}
	\resmat[g]^{n|\text{cau}}_{j,k}&=
	\Phifn'(\check\rl^+_j|\bm\la,\bm{\check\rl^+})
	\bmsum_{\bm\la}
	\Phifn'(\bm\la|\bm{\check{\rl}^+},\bm\la)
	\frac{1}{\sinh\pi(\check{\rl}^+_{k}-\bm\la)}
	\frac{e^{\ho{n}\pi\bm\la}}{\sinh\pi(\bm\la-\check\rl^+_j)}
	,
	\label{cau_ex_mat_I_cau-cau_sum}
	\\
	\resmat[g]^{s|\text{cau}}_{a,k}&=
	(-1)^{\ho{n}-1}
	\bmsum_{\bm\la}
	\Phifn'(\bm\la|\bm{\check{\rl}^+},\bm\la)
	\frac{1}{\sinh\pi(\check{\rl}^+_{k}-\bm\la)}
	\Vmat_{a}(\bm\la)
	.
	\label{cau_ex_mat_I_van-cau_sum}
\end{align}
We recall the $\Vmat_a=\Vmat<\ptn\gamma>_a$ is a vector valued function [see \cref{notn:hyp_vec_cau_van}] which corresponds to the hyperbolic Vandermonde block Cauchy-Vandermonde matrix \eqref{hyper_cau_van_mat}.
These summations can be evaluated using the result of the \cref{lem:Phifn_sum_period_method_trivial_case}.
It is easy see that we get
\begin{align}
	\resmat[g]^{n|\text{cau}}&=\Id_{N_0-\ho{n}}
	,
	\\
	\resmat[g]^{s|\text{cau}}&=0
	.
\end{align}
This also ensures that we can reduce the computation to a smaller matrix:
\index{ff@\textbf{Form-factors}!mat cvextn gs res@$\resmat*[e]$: reduced mat. obtained from CV extraction for the ground state|textbf}%
\begin{align}
	\resmat*[g]&=
	(-1)^{\ho{n}-1}
	\begin{pmatrix}
		\resmat[g]^{s|c}
		&
		\resmat[g]^{s|w+}
		&
		\resmat[g]^{s|w-}
	\end{pmatrix}
\end{align}
which is equivalent to $\resmat[g]$ up-to the computation of determinant, since
\begin{align}
 	\det_{\ho{n}}\resmat*[g]=\det_{N_0}\resmat[g]
 	.
 	\label{cau_ex_I_reduction}
\end{align}
In particular, this reduction means that the rest of the north block $\resmat[g]^{n|c}$ and $\resmat[g]^{n|w\pm}$ need not be computed.
For convenience $\resmat*[g]$ is subdivided into sub-blocks:
\begin{align}
	\resmat*[g]&=
	\begin{pmatrix}
		\resmat*[g]^{c}
		&
		\resmat*[g]^{w+}
		&
		\resmat*[g]^{w-}
	\end{pmatrix}
	.
	\label{cau_ex_I_eff_mat}
\end{align}
We will now compute matrices $\resmat*[g]^{c}$ and $\resmat[g]^{w\pm}$ forming the blocks for the contribution of close-pairs and wide-pairs.
\minisec{Close-pair block}
From \cref{cau_det_rep_I_gen_clp_diff} we can see that $\resmat*[g]^{c}$ can be expressed as a sum:
\begin{align}
	\resmat*[g]^{c}=
	\resmat*[g]^{c\textrm{-i}}
	+
	\resmat*[g]^{c\textrm{-ii}}
	.
	\label{cau_ex_I_clp_small_mat_decomp}
\end{align}
While the first term $\resmat*[g]^{c\textrm{-i}}$ comes from the action on Cauchy terms present in \cref{cau_det_rep_I_gen_clp_diff}, which can be expressed as
\begin{subequations}
\begin{align}
	\resmat*[g]^{c\textrm{-i}}_{a,b}
	&=
	\aux_{g}(\clp<->_{b}-i\stdv_{b})
	\sum_{\sigma=\pm}
	\left\lbrace
	\bmsum_{\bm\la}
	\Phifn'(\bm\la|\bm{\check\rl^+},\bm\la)
	\frac{1}{\sinh\pi(\clp<\sigma>_b+i\sigma\stdv_{b}-\bm\la)}
	\Vmat_{a}(\bm\la)
	\right\rbrace
	,
	\label{cau_ex_gen_I_clp_cau-term}
\shortintertext{%
the second term $\resmat*[g]^{c\textrm{-ii}}$ comes from the action on non-Cauchy terms in \cref{cau_det_rep_I_gen_clp_diff}%
, which can be expressed as
}
	\resmat*[g]^{c\textrm{-ii}}_{a,b}
	&=
	2 i
	\bmsum_{\bm\la}
	\Phifn'(\bm\la|\bm{\check\rl^+},\bm\la)
	\Vmat_{a}(\bm\la)
	\big\lbrace
	\rden_h(\bm\la-\clp<+>_b)
	-
	\rden_h(\bm\la-\clp<->_b)
	\big\rbrace
	.
	\label{cau_ex_gen_I_clp_non-cau-term}
\end{align}
\label{cau_ex_gen_I_clp_all-terms}
\end{subequations}
Since the action is taken with periodic function for the first term $\resmat*[g]^{c\textrm{-i}}$ \eqref{cau_ex_gen_I_clp_cau-term}, it can be computed using \cref{lem:Phifn_sum_period_method_trivial_case}.
Here we also take the limit $i\stdv_{a}\to 0$ to obtain%
\footnote{$\phifn'$ signify the omission of a vanishing term in the product.}
\begin{align}
	\resmat*[g]^{c\textrm{-i}}_{a,b}
	&=
	-
	\phifn(\clp<->_{b}-i|\bm\mu,\bm\la)
	\phifn'(\clp<->_b+i|\bm\la,\bm\mu)
	\Phifn'(\clp<+>_b|\bm{\check\rl^+},\bm\la)
	\Vmat_{a}(\clp<+>_b)
	.	
	\label{cau_ex_I_clp_cau_term_sol}
\end{align}
We can use the condensation property to express it in the form of integrals:
\begin{align}
	\resmat*[g]^{c\textrm{-ii}}_{a,b}
	&=
	\left(\int_{\Rset-i\alpha}-\int_{\Rset+i\alpha}\right)
	\Phifn(\tau|\bm{\check\rl^+},\bm\la)
	\Vmat_{a}(\tau)
	\lbrace
	\rden_{h}(\tau-\clp<+>_b)-\rden_h(\tau-\clp<->_b)
	\rbrace
	d\tau
	.
	\label{cau_ex_I_clp_non-cau_int_form}
\end{align}
It is obtained by writing the sum over residues as integrals over a rectangular contour of thickness $\alpha$.
We then see that the integrals over vertical edges vanish as the edges are pushed towards infinity due to an exponentially dininishing integrand.
This is true because cardinalities of the two sets are ordered $n_{\bm\check\rl^+}<n_{\bm\la}$ and the difference $\ho{n}$ is greater than $\max(\ptn{\gamma})$.
We will not compute these expressions exactly giving a close-form representation.
We can however simplify the integrand.
To do this, we first use the periodicity of the $\Phifn$ function to convert the integrals as
\begin{multline}
	\resmat*[g]^{c\textrm{-ii}}_{a,b}
	=
	\int_{\Rset+i-i\alpha}
	\Phifn(\tau|\bm{\check\rl^+},\bm\la)
	\Vmat_a(\tau)
	\rden_h(\tau-\clp<+>_b)
	d\tau
	\\
	+
	\int_{\Rset+i\alpha}
	\Phifn(\tau|\bm{\check\rl^+},\bm\la)
	\Vmat_{a}(\tau)
	\rden_h(\tau-\clp<->_b)
	d\tau
	\\
	+
	\int_{\Rset-i\alpha}
	\Phifn(\tau|\bm{\check\rl^+},\bm\la)
	\Vmat_a(\tau)
	\rden_h(\tau-\clp<+>_b)
	d\tau
	\\
	+
	\int_{\Rset-i+i\alpha}
	\Phifn(\tau|\bm{\check\rl^+},\bm\la)
	\Vmat_a(\tau)
	\rden_h(\tau-\clp<->_b)
	d\tau
	.
\end{multline}
Now we can use the semi-periodicity property of the density function which we obtain in \cref{chap:den_int_aux}, which says
\begin{align}
	\rden_h(\tau-\clp<+>)+\rden_h(\tau-\clp<->)
	=
	\frac{1}{2\pi i}t(\tau-\clp).
	\label{rden_h_clp_identity}
\end{align}
It permits us to write,
\begin{multline}
	\resmat*[g]^{c\textrm{-ii}}_{a,b}
	\begin{aligned}[t]
	=
	\left(
	\int_{\Rset+i-i\alpha}
	-
	\int_{\Rset+i\alpha}
	\right)
	\Phifn(\tau|\bm{\check\rl^+},\bm\la)
	\Vmat_{a}(\tau)
	~	
	\rden_h(\tau-\clp<+>_b)
	d\tau
	&
	\\
	+
	\int_{\Rset+i\alpha}
	\Phifn(\tau|\bm{\check\rl^+},\bm\la)
	\Vmat_{a}(\tau)
	~	
	t(\tau-\clp_b)
	d\tau
	&
	\end{aligned}
	\\
	\begin{aligned}[b]
	+
	\left(
	\int_{\Rset-i+i\alpha}
	-
	\int_{\Rset-i\alpha}
	\right)
	\Phifn(\tau|\bm{\check\rl^+},\bm\la)
	\Vmat_{a}(\tau)
	~	
	\rden_h(\tau-\clp<->_b)
	d\tau
	&
	\\
	+
	\int_{\Rset-i\alpha}
	\Phifn(\tau|\bm{\check\rl^+},\bm\la)
	\Vmat_{a}(\tau)
	~	
	t(\tau-\clp_b)
	d\tau
	&
	.
	\end{aligned}
\end{multline}
We can close the contour for two anti-parallel integrals since the integral over the vertical edges is vanishing at infinity.
Integrals for both of these closed contours as their integrands are analytic inside them.
Hence what remains are only two parallel integrals:
\begin{align}
	\resmat*[g]^{c\textrm{-ii}}_{a,b}
	=
	\left(
	\int_{\Rset-i\alpha}
	+
	\int_{\Rset+i\alpha}
	\right)
	\Phifn(\tau|\bm{\check\rl^+},\bm\la)
	\Vmat_{a}'(\tau)
	~
	t(\tau-\clp_b)
	d\tau
	.
	\label{cau_ex_I_clp_non-cau_term_int_form}
\end{align}
Substituting \cref{cau_ex_I_clp_non-cau_term_int_form,cau_ex_I_clp_cau_term_sol} into \cref{cau_ex_I_clp_small_mat_decomp} we get the expression:
\index{ff@\textbf{Form-factors}!mat cvextn gs res clp@\hspace{1em}$\resmat*[g]^{c}$: close-pair block inside \rule{3em}{1pt}|textbf}%
\begin{multline}
	\resmat*[g]^{\txtcp}_{a,b}=
	-
	\phifn(\clp<->_{a}-i|\bm\mu,\bm\la)
	\phifn'(\clp<->_a+i|\bm\la,\bm\mu)
	\Phifn'(\clp<+>|\check{\bm\alpha},\bm\la)
	\Vmat_a(\clp<+>_b)
	\\
	+
	\left(
	\int_{\Rset-i\alpha}
	+
	\int_{\Rset+i\alpha}
	\right)
	\Phifn(\tau|\bm{\check\rl^+},\bm\la)
	\Vmat_a(\tau)
	\,
	t(\tau-\clp_b)
	d\tau
	.
	\label{cau_ex_I_clp_int_form}
\end{multline}
\minisec{Wide-pair block}
In the case of wide-pairs, we have seen in \cref{cau_det_rep_I_gen_wdp+,cau_det_rep_I_gen_wdp-}, that the entire columns $\modCau[g]^{w\pm}$ consists only the density terms $\rden_2$ for the wide pairs which can be expressed in terms of digamma function.
More importantly, there are no Cauchy type terms in these expressions.
We obtain from the extraction \eqref{cau_ex_gen_I_mat} the following summations for these blocks:
\begin{subequations}
\begin{align}
	\resmat*[g]^{w+}_{a,b}
	&=
	2i
	\bmsum_{\bm\la}
	\Phifn'(\bm\la|\bm{\check\rl^+},\bm\la)
	\Vmat_{a}(\bm\la)
	\left\lbrace
	\rden_2(\bm\la,\wdp_b+i)
	-
	\rden_2(\bm\la,\wdp_b)
	\right\rbrace
\shortintertext{and}
	\resmat*[g]^{w-}_{a,b}
	&=
	2i
	\bmsum_{\bm\la}
	\Phifn'(\bm\la|\bm{\check\rl^+},\bm\la)
	\Vmat_{a}(\bm\la)
	\left\lbrace
	\rden_2(\bm\la,\wdp*_b)
	-
	\rden_2(\bm\la,\wdp*_b-i)
	\right\rbrace
	.
\end{align}
\end{subequations}
They be computed as integrals over a rectangular contour of width $2\alpha<1$.
We can also see that the integral over vertical edges vanishes due to the exponential behaviour of the $\Phifn$ function.
Therefore we are left with the integral on two anti-parallel branches of the contours.
\begin{subequations}
\begin{align}
	\resmat*[g]^{w+}_{a,b}
	&=
	\left(
	\int_{\Rset-i\alpha}
	-
	\int_{\Rset+i\alpha}
	\right)
	\Phifn(\tau|\bm{\check\rl^+},\bm\la)
	\Vmat_a(\tau)
	\left\lbrace
	\rden_2(\tau,\wdp_b+i)
	-
	\rden_2(\tau,\wdp_b)
	\right\rbrace
	d\tau
	\shortintertext{and}
	\resmat*[g]^{w-}_{a,b}
	&=
	\left(
	\int_{\Rset-i\alpha}
	-
	\int_{\Rset+i\alpha}
	\right)
	\Phifn(\tau|\bm{\check\rl^+},\bm\la)
	\Vmat_a(\tau)
	\left\lbrace
	\rden_2(\tau,\wdp*_b)
	-
	\rden_2(\tau,\wdp*_b-i)
	\right\rbrace
	d\tau
\end{align}
\end{subequations}
Finally using the periodicity of the $\Phifn$ function it can be rewritten in the form of a combination of mutually parallel or anti-parallel integrals.
\index{ff@\textbf{Form-factors}!mat cvextn gs res wdp@\hspace{1em}$\resmat*[g]^{w\pm}$: wide-pair blocks inside \rule{3em}{1pt}|textbf}%
\begin{subequations}
\begin{align}
	\resmat*[g]^{\txtwp+}_{a,b}
	&=
	\left(
	\int_{\Rset-i+i\alpha}
	-
	\int_{\Rset-i\alpha}
	-
	\int_{\Rset+i-i\alpha}
	+
	\int_{\Rset+i\alpha}
	\right)
	\Phifn(\tau|\bm{\check\rl^+},\bm\la)
	\Vmat_a(\tau)
	\,
	\rden_2(\tau,\wdp_b)
	d\tau
\shortintertext{and}
	\resmat*[g]^{\txtwp-}_{a,b}
	&=
	\left(
	\int_{\Rset-i+i\alpha}
	-
	\int_{\Rset-i\alpha}
	-
	\int_{\Rset+i-i\alpha}
	+
	\int_{\Rset+i\alpha}
	\right)
	\Phifn(\tau|\bm{\check\rl^+},\bm\la)
	\Vmat_a(\tau)
	\,
	\rden_2(\tau,\wdp*_b)
	d\tau
	.
\end{align}
\label{cau_ex_I_wdp_int_form}
\end{subequations}
The expressions obtained in \cref{cau_ex_I_clp_int_form,cau_ex_I_wdp_int_form} will be the final expression that we write down for the components of the matrix $\resmat*[g]$.
We do not have a reasonable approach to compute the integrals involved in these functions. Despite all the efforts, we find that there always remains parallel integrals which we cannot closed to evaluate them directly as sum over residues.
Any other asymptotic approaches towards its computation, also hit a roadblock, since we do not have a faithful asymptotic representation of the $\Phifn$ function in the thermodynamic limit, unlike its rational version $\phifn$ where it is better understood.
When a thermodynamic limit of $\Phifn$ is computed from an infinite product represented in terms of $\phifn$ functions, we find that the result has bad asymptotic properties. More catastrophically, we also find that the infinite product form in terms of $\phifn$ functions is also ill-defined for triplets and does not converge\footnotemark after the substitution of the thermodynamic limit for the $\phifn$ functions.
\\
Due to these reasons, we are forced to leave the computations for the components of the residual matrix $\resmat*[g]$, at the level of integral forms in terms of auxiliary functions $\Phifn$ as shown in \cref{cau_ex_I_clp_int_form,cau_ex_I_wdp_int_form}.
\footnotetext{%
unlike the infinite products in the prefactors which converge as we have seen in \cref{sec:tdl_inf_prod} for the two-spinon case as well as in the higher-spinon cases which we shall encounter in \cref{sec:tdl_pref_gen}
}
\subsection{CV extraction of the second type}
\label{sub:cv_extn_type_II}
Let us now compute the matrix $\resmat[e]$ from the extraction \eqref{cau_ex_gen_II_mat}.
In this extraction \eqref{cau_ex_gen_II_mat} the inverse of a dual Cauchy-Vandermonde matrix $\Cmat*<\ptn\gamma>^{-1}[\bm{\rh^+}\Vert\bm{\check\la}]$ [see \cref{hyper_cau_van_inv_dressing,hyper_cau_van_inv_blocks}] acts on the matrix composed of diagonal blocks $\modCau[e]$ \eqref{cau_det_rep_gen_II_blocks} and the identity $\Id_{\ho{n}}$.
We can therefore divide the matrix $\resmat[e]$ into blocks:
\begin{align}
	\resmat[e]
	=
	\begin{pmatrix}
		\resmat[e]^{\text{cau}}	&	\Ho{\Pcal}	&	\Wcal	&	\Zcal
	\end{pmatrix}
	=
	\begin{pmatrix}
		\resmat[e]^{n|\text{cau}}	&	\Ho{\Pcal}^{n}	&	\Wcal^{n}	&	\Zcal^{n}
		\\
		\resmat[e]^{s|\text{cau}}	&	\Ho{\Pcal}^{s}	&	\Wcal^{s}	&	\Zcal^{s}
	\end{pmatrix}
	.
	\label{cau_ex_II_mat_blocks}
\end{align}
We use the cardinal directions north and south to denote the top $N_0-\ho{n}-1=n_r+n_c$ rows and bottom $\ho{n}$ rows respectively. 
The north block corresponds to the extraction with fixed parameter from the set $\bm{\rl^+}$ for the inverse matrix.
Similarly the south corresponds to the extraction with the rows of the inverse matrix containing hole parameters $\bm\hle$.
Column-wise it is also divided in the four blocks which are described in the following.
\begin{enumerate}[wide=0pt, itemsep=1em,%
]
\index{ff@\textbf{Form-factors}!mat cvextn es hyper van@\hspace{1em}$\Zcal$: hyperbolic Vandermonde block inside \rule{3em}{1pt}}%
\index{ff@\textbf{Form-factors}!mat cvextn es hyper van@\hspace{1em}$\Zcal$: hyperbolic Vandermonde block inside \rule{3em}{1pt}|seealso{CV.}}%
\item[\fbox{\textbf{$\Zcal$}}]
The last two blocks $\Zcal^{n}$ and $\Zcal^{s}$ contain $\ho{n}$ columns each. They arise from the trivial part of the extraction in \cref{cau_ex_gen_II_mat}, using the $\Id_{\ho{n}}$. Thus we can see that hyperbolic Vandermonde columns \eqref{hyper_cau_van_inv_hvan_block} retain their original form: 
\begin{align}
	\Vdr<\ptn\gamma>[\bm{\rh^+}]
	=
	\begin{pmatrix}
		\Vdr<\ptn\gamma>[\bm{\rl^+}]
		\\
		\Vdr<\ptn\gamma>[\bm{\hle}]
	\end{pmatrix}
	.
\end{align}
We shall denote here $\Zcal^{n}=\Vdr<\ptn\gamma>[\bm{\rl^+}]$, $\Zcal^{s}=\Vdr<\ptn\gamma>[\bm\hle]$ and $\Zcal=\Vdr<\ptn\gamma>[\bm{\rh^+}]$.
Let us also drop the sub-index $\ptn\gamma$ from its notation as the partition $\ptn\gamma$ \eqref{ptn_consec_even-odd} is fixed by the quantum number $\ho{n}$ which is invariant for any given excitation.
\item[\fbox{\textbf{$\Wcal$}}]
The blocks $\Wcal^n$ and $\Wcal^s$ originate from the action on the Foda-Wheeler block $\bar\Ucal$ \eqref{cau_ex_II_gen_mat_FW_block} inside $\modCau[e]$ \eqref{cau_det_rep_gen_II_blocks}.
Therefore we can express,
\index{ff@\textbf{Form-factors}!mat cvextn es fw@\hspace{1em}$\Wcal$: Foda-Wheeler block inside \rule{3em}{1pt}}%
\begin{align}
	\Wcal=
	\Cmat*<\ptn\gamma>^{-1}[\bm{\rh^+}\Vert\bm{\check\la}]
	\cdot
	\bar\Ucal[\bm\check\la]
	.
\end{align}
It contains exactly two columns.
It is important note that $\Wcal$ also occurs in the two-spinon case.
The method of extraction for these block follows the similar trails but differs in one crucial aspect.
This difference is due to the fact that we extract a Cauchy-Vandermonde matrix here.
\item[\fbox{\textbf{$\Ho{\Pcal}$}}] The columns $\Ho\Pcal$ arises from the action on higher-level block $\modCau*$ in \cref{cau_det_rep_gen_II_blocks}.
It contains $\ho{n}$ columns.
We also know that the matrix $\modCau*$ is related to the matrix $\ho{\Scal}$ which solves the higher-level Gaudin extraction \eqref{hl_ff_gau_ex_all}.
It can be expressed as,
\index{ff@\textbf{Form-factors}!mat cvextn es hl@\hspace{1em}$\Ho\Pcal$: higher-level block inside \rule{3em}{1pt}|textbf}%
\begin{align}
	\modCau* = \Ccal[\bm{\check\la}\Vert\bm\hle]\cdot\Acal^{-1}\ho{\Scal}[\bm\hle\Vert\bm\cid]
	.
\end{align}
Since we take the left-action for this extraction we can see that it involves only an extraction on a Cauchy matrix and $\ho{\Scal}$ remains unperturbed.
\item[\fbox{\textbf{$\Pcal^{\text{cau}}$}}] The columns $\Pcal^{\text{cau}}$ comes from the action on the Cauchy block $\modCau[e]^{\text{cau}}$ in \cref{cau_det_rep_gen_II_blocks}. From \cref{cau_det_rep_I_mat_cau_block} we can see that it can be expressed as a sum of Cauchy matrices which are all contained in the Cauchy-Vandermonde matrix that we extract.
\begin{align}
	\modCau[e]^{\text{cau}}=
	\Ccal[\bm{\check\la}\Vert\bm{\rl^+}]
	+
	\Ccal[\bm{\check\la}\Vert\bm\hle]
	\cdot
	\Acal^{-1}\Rcal[\bm\hle\Vert\bm{\rl^+}]
\end{align}
\end{enumerate}
Let us first compute the components of matrices $\resmat[e]^{\text{cau}}$, $\Ho{\Pcal}$ and $\Wcal$ one-by-one.
\Cref{cau_det_rep_II_gen_cau_block} tells us that $\resmat[e]^{\text{cau}}$ is described by the following summations:
\begin{subequations}
\begin{multline}
	\resmat[e]^{n|\text{cau}}_{j,k}=
	e^{\ho{n}\pi\rl^+_j}
	\Phifn'(\rl^+_j|\bm{\check\la},\bm{\rl^+})
	\\
	\times
	\bmsum_{\bm{\check\la}}
	\Phifn'(\bm{\check\la}|\bm{\rh^+},\bm{\check\la})
	\frac{e^{-\ho{n}\pi\bm{\check\la}}}{\sinh\pi(\rl^+_j-\bm{\check\la})}		
	\left\lbrace
	\frac{1}{\sinh\pi(\bm{\check\la}-\rl^+_k)}
	+
	\bmsum_{\bm\hle}
	\frac{\Acal^{-1}\Rcal[\bm\hle\Vert\rl^+_k]}{\sinh\pi(\bm{\check\la}-\bm\hle)}
	\right\rbrace
	.
	\label{cau_ex_II_cau_rc}
\end{multline}
And
\index{ff@\textbf{Form-factors}!mat cvextn cau@\hspace{1em}$\Pcal^{\text{cau}}$: Cauchy block inside \rule{3em}{1pt}|textbf}%
\begin{multline}
	\resmat[e]^{s|\text{cau}}_{a,k}=
	e^{\ho{n}\pi\hle_a}
	\Phifn'(\hle_a|\bm{\check\la},\bm{\hle})
	\\
	\times
	\bmsum_{\bm{\check\la}}
	\Phifn'(\bm{\check\la}|\bm{\rh^+},\bm{\check\la})
	\frac{e^{-\ho{n}\pi\bm{\check\la}}}{\sinh\pi(\hle_a-\bm{\check\la})}		
	\left\lbrace
	\frac{1}{\sinh\pi(\bm{\check\la}-\rl^+_k)}
	+
	\bmsum_{\bm\hle}
	\frac{\Acal^{-1}\Rcal[\bm\hle\Vert\rl^+_k]}{\sinh\pi(\bm{\check\la}-\bm\hle)}
	\right\rbrace
	.
	\label{cau_ex_II_cau_hle}
\end{multline}
\label{cau_ex_II_cau}
\end{subequations}
From \cref{hl_ff_gau_ex} we find the following summations for $\Ho{\Pcal}$:
\begin{subequations}
\begin{align}
	\Ho{\resmat}^{n}_{j,b}
	&=
	e^{\ho{n}\pi\rl^+_j}
	\Phifn'(\rl^+_j|\bm{\check\la},\bm{\rl^+})
	\bmsum_{\bm{\check\la}}
	\Phifn'(\bm{\check\la}|\bm{\rl^+},\bm{\check\la})
	\frac{e^{-\ho{n}\pi\bm{\check\la}}}{\sinh\pi(\rl^+_j-\bm{\check\la})}		
	\bmsum_{\bm\hle}
	\frac{\Acal^{-1}\ho{\Scal}[\bm\hle\Vert\cid_b]}{\sinh\pi(\bm{\check\la}-\bm\hle)}
	\label{cau_ex_II_hl_rc}
	,
\shortintertext{and,}
	\Ho{\resmat}^{s}_{a,b}
	&=
	e^{\ho{n}\pi\hle_a}
	\Phifn'(\hle_a|\bm{\check\la},\bm{\rl^+})
	\bmsum_{\bm{\check\la}}
	\Phifn'(\bm{\check\la}|\bm{\rl^+},\bm{\check\la})
	\frac{e^{-\ho{n}\pi\bm{\check\la}}}{\sinh\pi(\hle_a-\bm{\check\la})}		
	\bmsum_{\bm\hle}
	\frac{\Acal^{-1}\ho{\Scal}[\bm\hle\Vert\cid_b]}{\sinh\pi(\bm{\check\la}-\bm\hle)}
	.
	\label{cau_ex_II_hl_hle}
\end{align}
\label{cau_ex_II_hl}
\end{subequations}
Comparing \cref{cau_ex_II_hl,cau_ex_II_cau} we find a similarity among them.
In the extraction \eqref{cau_ex_II_hl} the first type of term with a single sum is absent and the double sum which involves the hole parameter in both \cref{cau_ex_II_hl,cau_ex_II_cau}. An important difference among them is that we have replaced the matrix $\Rcal[\bm\hle\Vert\bm{\rl^+}]$ in the double sum in \cref{cau_ex_II_cau} with its higher-level equivalent $\Scal[\bm\hle\Vert\bm\hle]$ in \cref{cau_ex_II_hl}, where the latter solves the higher-level Gaudin extraction \eqref{hl_ff_gau_ex_mat_transp}.
\par
The matrix $\Wcal$ will be considered later.
We will now compute the summations running over $\bm{\check\la}$ in the above \cref{cau_ex_II_cau,cau_ex_II_hl}.
Since the cardinality of the set $\bm{\check\la}$ is less than the set $\bm{\rh^+}$ we see that the \cref{lem:Phifn_sum_period_method_trivial_case} cannot be applied as it as to compute these summations.
There are two approached to this which are equivalent.
Both these approaches are investigated in \cref{chap:mat_det_extn} for a toy example of the type:
\begin{align}
	P=C_{V}^{-1}\cdot\diag\big[C^{\textrm{L}}+C^{\textrm{R}}\cdot R~\big|~\Id\big]
	\label{toy_extn_example_cv_extn_gen}
\end{align}
which is a simplified `toy' version of the situation that is presented.
In \cref{sec:toy_cv_extn_append}, we first investigate this toy example in the rational parametrisation using two different methods which are described in the following. It is then generalised in \cref{sub:toy_cv_extn_hyper_append} to the hyperbolic parametrisation.
\begin{enumerate}[wide=0pt, label={\textbf{Method \Roman*}:}, ref={method \Roman*}]
\item The first method uses the identities for the inversion of the Cauchy-Vandermonde matrices obtained through the dressing [see \cref{cau_van_inv_diag_dress,dual_cau_van_inv_diag_dress,hyper_cau_van_inv_dressing}], to write down matrices $C_{V}^{-1}C^{\textrm{L}}$ and $C_{V}^{-1}C^{\textrm{R}}$.
	Since the resulting summations do not extend over the Vandermonde terms, we can see that as a result of this action, we shall obtain a linear combination of the Vandermonde columns $\Zcal$ of the inverse matrix, in addition to the identity block. Since it can be expressed as a linear sum over the remaining columns, we find that these extra terms can be easily cancelled in the determinant.
\label{cau_ex_II_method_1}
\item In the second method we construct a larger square Cauchy matrix (let's call it $C_L$) by adding extra variables wherever needed.
\begin{align}
	P=C_L^{-1}(C^{\textrm{L}}+C_{\textrm{R}}\cdot R|\Id)
\end{align}
The actions $C_{L}^{-1}C^{\textrm{L}}$ and $C_{L}^{-1}C^{\textrm{II}}$ can then be easily computed.
We find that the resulting summations do not encompass all the terms, hence it does not give us a large sub-block of an identity matrix.
Instead, there are correction to the identity block of the same rank as the number of extra variables that were added. Finally, we find that these rank-$n$ corrections can be cancelled using the other columns.
In the limit where these extra variables are send to infinity, we recover the same expression as \ref{cau_ex_II_method_1}.
\label{cau_ex_II_method_2}
\end{enumerate}
\par
The choice of the computation for matrices $\resmat[e]^{\text{cau}}$ and $\Ho{\resmat}$ from \cref{cau_ex_II_cau,cau_ex_II_hl} is irrelevant.
From either method, we get the matrix $\resmat[e]^{\text{cau}}$ \eqref{cau_ex_II_cau} given by the following expressions:
\begin{subequations}
\begin{align}
	\resmat[e]^{n|\text{cau}}_{j,k}
	&=
	\delta_{j,k}
	-
	\sum_{r=1}^{\ho{n}}
	\chi_{r,k}
	\Vdr_{r}[\rl^+_j]
	-
	\sum_{r=1}^{\ho{n}}
	\sum_{t=1}^{n_h}
	\chi^h_{r,t}
	~
	\Acal^{-1}\Rcal[\hle_t\Vert\rl^+_k]
	~
	\Vdr_{r}[\rl^+_j]
	,
	\label{cv_extn_II_cau_north_comps}
\shortintertext{and,}
	\resmat[e]^{s|\text{cau}}_{a,k}
	&=
	\Acal^{-1}\Rcal[\hle_a\Vert\rl^+_k]
	-
	\sum_{r=1}^{\ho{n}}
	\chi_{r,k}
	\Vdr_{r}[\hle_a]
	-
	\sum_{r=1}^{\ho{n}}
	\sum_{t=1}^{n_h}
	\chi^h_{r,t}
	~
	\Acal^{-1}\Rcal[\hle_t\Vert\rl^+_k]
	~
	\Vdr_{r}[\hle_a]
	.
	\label{cv_extn_II_cau_south_comps}
\end{align}
\label{cv_extn_II_cau_comps}
\end{subequations}
Similarly for the matrix $\Ho{\resmat}$ \eqref{cau_ex_II_hl} we get,
\begin{subequations}
\begin{align}
	\Ho{\Pcal}^{n}_{j,b}
	&=
	-
	\sum_{r=1}^{\ho{n}}
	\ho{\chi}_{r,b}
	\Vdr_{r}[\rl^+_j]
	-
	\sum_{r=1}^{\ho{n}}
	\sum_{t=1}^{n_h}
	\ho{\chi}^h_{r,t}
	~
	\Acal^{-1}\ho{\Scal}[\hle_t\Vert\cid_b]
	~
	\Vdr_{r}[\rl^+_j]
	,
	\label{cv_extn_II_hl_north_comps}
\shortintertext{and,}
	\Ho{\Pcal}^{s}_{a,b}
	&=
	\Acal^{-1}\ho{\Scal}[\hle_a\Vert\cid_b]
	-
	\sum_{r=1}^{\ho{n}}
	\ho{\chi}_{r,b}
	~
	\Vdr_{r}[\hle_a]
	-
	\sum_{r=1}^{\ho{n}}
	\sum_{t=1}^{n_h}
	\ho{\chi}^h_{r,t}
	~
	\Acal^{-1}\ho{\Scal}[\hle_t\Vert\cid_b]
	~
	\Vdr_{r}[\hle_a]
	.
	\label{cv_extn_II_hl_south_comps}
\end{align}
\label{cv_extn_II_hl_comps}
\end{subequations}
In the toy example for rational case, we get the coefficient matrices $\chi$ and $\chi^h$ that are composed of the terms containing supersymmetric elementary functions \eqref{ele_susy_append} given in \cref{defn:ele_susy_append}.
In the hyperbolic parametrisation, we would expect that supersymmetric functions in the exponential variables \eqref{reparam_cau_van_hyper-rat} or their linear combination appearing in $\chi$, $\ho{\chi}$ and $\chi^h$, $\ho{\chi}^h$ coefficient matrices.
But their exact form is not important to us since they appear as a coefficient terms of a linear sum of the columns $\Zcal$ of the same matrix $\resmat[e]$ \eqref{cau_ex_II_mat_blocks} and hence these terms get cancelled in the determinant.
\par
We will now compute the matrix $\Wcal$ from the extraction on Foda-Wheeler columns $\Vcal$ \eqref{cau_ex_II_gen_mat_FW_block}.
In this case the \ref{cau_ex_II_method_1} cannot be used due to the nature of the matrix $\Vcal$.
For the \ref{cau_ex_II_method_2},
we construct take a larger set of variables $\bm\zeta$ containing the set $\bm{\check\la}\subset\bm\zeta$ ($n_{\zeta}=N_0+\ho{n}+1$).
Let $\bm\eta$ ($n_{\eta}=\ho{n}$) denote the extra variables added to obtain it $\bm\zeta=\bm{\check\la}\bm\cup\bm\eta$.
In terms of this bigger set $\bm\zeta$ we construct a larger square Cauchy matrix $\Ccal[\bm\zeta\Vert\bm{\rh^+}]$ which is now extracted to get
\begin{align}
	\Wcal
	=
	\Ccal^{-1}(\bm{\rh^+}\Vert\bm\zeta)
	\cdot
	\begin{pmatrix}
		\Vcal[\bm{\check\la}]
	\\
	\bm0
	\end{pmatrix}
	.
	\label{big_cau_ex_FW}
\end{align}
In its computation, we find it convenient to further the divide the north-block $\Wcal^{n}$ into two to rewrite
\begin{align}
	\Wcal
	=
	\begin{pmatrix}
		\Wcal^n
		\\
		\Wcal^c
		\\
		\Wcal^s
	\end{pmatrix}
	=
	\begin{pmatrix}
		\Wcal[\bm\rl]
		\\
		\Wcal[\bmclp<+>]
		\\
		\Wcal[\bm\hle]
	\end{pmatrix}
\end{align}
where the central part $\Wcal^{c}$ consists as many columns as the number $n_c$ of close-pairs.
We can see that the components of these three blocks from the extraction \eqref{big_cau_ex_FW} are given by,
\begin{subequations}
\begin{align}
		\Wcal^{n}_{j,r}
		&=
		-
		\Phifn'(\rl_j|\bm\zeta,\bm{\rh^+})
		\frac{1}{\pi}
		\bmsum_{\bm{\check\la}}
		\Phifn'(\bm{\check{\la}}|\bm{\rh^+},\bm\zeta)
	 	\frac{1}{\sinh\pi(\bm{\check\la}-\rl_j)}
		\frac%
		{\aux_e(\bm{\check\la})(\bm{\check\la}+i)^{r}-\bm{\check\la}^{r}}%
		{1+\aux_{e}(\bm{\check\la})}%
		,
	\label{large_cau_ex_FW_north}
		\\
		\Wcal^{c}_{a,r}
		&=
		-
		\Phifn'(\clp<+>_{a}|\bm\zeta,\bm{\rh^+})
		\frac{1}{\pi}
		\bmsum_{\bm{\check\la}}
		\Phifn'(\bm{\check{\la}}|\bm{\rh^+},\bm\zeta)
	 	\frac{1}{\sinh\pi(\bm{\check\la}-\clp<+>_{a})}
		\frac%
		{\aux_e(\bm{\check\la})(\bm{\check\la}+i)^{r}-\bm{\check\la}^{r}}%
		{1+\aux_{e}(\bm{\check\la})}%
		,
	\label{large_cau_ex_FW_centre}
	\shortintertext{and}
		\Wcal^{s}_{a,r}
		&=
		-
		\Phifn'(\hle_a|\bm\zeta,\bm{\rh^+})
		\frac{1}{\pi}
		\bmsum_{\bm{\check\la}}
		\Phifn'(\bm{\check{\la}}|\bm{\rh^+},\bm\zeta)
	 	\frac{1}{\sinh\pi(\bm{\check\la}-\hle_a)}
		\frac%
		{\aux_e(\bm{\check\la})(\bm{\check\la}+i)^{r}-\bm{\check\la}^{r}}%
		{1+\aux_{e}(\bm{\check\la})}%
		.
	\label{large_cau_ex_FW_south}
\end{align}
\label{large_cau_ex_FW_comps}
\end{subequations}
These can be computed using a method which is similar to the one employed in the case of two-spinon form-factor [see \crefrange{cau_ex_2sp_FW_sum}{cau_ex_2sp_FW_result} from \cref{chap:2sp_ff}].
Let us consider the following function
\begin{align}
	f_r(\nu|\bm\alpha)
	=
	-
	\frac{1}{\pi}
	\bmsum_{\bm\alpha\bm\subset\bm\zeta}
	\Phifn'(\bm\alpha|\bm{\rh^+},\bm\zeta)
	\frac{1}{\sinh\pi(\bm\zeta-\nu)}
	\frac%
	{\aux_e(\bm\alpha)(\bm\alpha+i)^{r}-\bm\alpha^r}%
	{1+\aux_{e}(\bm\alpha)}%
\end{align}
which can be used to descried all of the terms in \cref{large_cau_ex_FW_north,large_cau_ex_FW_centre,large_cau_ex_FW_south}.
This allows us to write,
\begin{align}
	\Wcal_{j,r}=
	\Phifn(\rh^+_j|\bm\zeta,\bm{\rh^+})
	f_r(\rh^+_j|\bm{\check\la})
	=
	\Phifn(\rh^+_j|\bm\zeta,\bm{\rh^+})
	\Big\lbrace
	f_r(\rh^+_j|\bm\zeta)
	-
	f_r(\rh^+_j|\bm\eta)
	\Big\rbrace
	.
	\label{cau_ex_fw_larger_gen_fun_form}
\end{align}
Let us assume that all of the extra parameters $\bm\eta\bm\subset\Rset$ are real.
We will later see how the periodicity of the $\Phifn$ function makes it safe to assume so.
We can therefore write following integrals when $\nu\in\Rset$,
\begin{subequations}
\begin{align}
	f_r(\nu|\bm\zeta)
	=
	-
	\frac{1}{2\pi i}
	\left(
	\int_{\Rset-i\alpha}
	-
	\int_{\Rset+i\alpha}
	-
	\oint_{\varepsilon_{\nu}}
	\right)
	\Phifn(\tau|\bm{\rh^+},\bm\zeta)
	\frac{1}{\sinh\pi(\tau-\nu)}
	\frac%
	{\aux_e(\tau)(\tau+i)^r-\tau^r}%
	{\aux_e(\tau)+1}%
	d\tau
	\label{extn_fw_gen_larger_cau_real}
\end{align}
otherwise for $\nu\notin\Rset$ we can write 
\begin{align}
	f_r(\nu|\bm\zeta)
	=
	-
	\frac{1}{2\pi i}
	\left(
	\int_{\Rset-i\alpha}
	-
	\int_{\Rset+i\alpha}
	\right)
	\Phifn(\tau|\bm{\rh^+},\bm\zeta)
	\frac{1}{\sinh\pi(\tau-\nu)}
	\frac%
	{\aux_e(\tau)(\tau+i)^r-\tau^r}%
	{\aux_e(\tau)+1}%
	d\tau
	.
	\label{extn_fw_gen_larger_cau_non-real}
\end{align}
\label{extn_fw_gen_larger_cau}
\end{subequations}
The integrals on the vertical edges can be ignored since this new $\Phifn$ is bounded at infinity.
In both cases for \cref{extn_fw_gen_larger_cau_real,extn_fw_gen_larger_cau_non-real} we will first use the estimation of the function $\aux_e$ as exponentially vanishing or growing in $M$ inside the bulk.
We also subsequently use the periodicity of the $\Phifn$ to write,
\begin{multline}
	\left(
	\int_{\Rset-i\alpha}
	-
	\int_{\Rset+i\alpha}
	\right)
	\Phifn(\tau|\bm{\rh^+},\bm\zeta)
	\frac{1}{\sinh\pi(\tau-\nu)}
	\frac%
	{\aux_e(\tau)(\tau+i)^r-\tau^r}%
	{\aux_e(\tau)+1}%
	d\tau
	\\
	=	
	\left(
	\int_{\Rset+i\alpha}
	-
	\int_{\Rset+i-i\alpha}
	\right)
	\Phifn(\tau|\bm\rh^+,\bm\zeta)
	\frac{\tau^r}{\sinh\pi(\tau-\nu)}
 	d\tau
 	.
 	\label{cau_ex_II_gen_fw_large_new_contour}
\end{multline}
We can close the contour for the newly obtained anti-parallel integrals.
The contribution of the vertical edges is still vanishing.
For all values of $\nu\in\bm{\rh^+}$ we can see that the integrand is holomorphic inside the new contour and thus we get,
\begin{align}
	\left(
	\int_{\Rset+i\alpha}
	-
	\int_{\Rset+i-i\alpha}
	\right)
	\Phifn(\tau|\bm{\rh^+},\bm\zeta)
	\frac{\tau^r}{\sinh\pi(\tau-\rh^+_k)}
 	d\tau
	=0	
	.
	\label{cau_ex_gen_II_large_holo_zero}
\end{align}
Substituting \crefrange{extn_fw_gen_larger_cau}{cau_ex_gen_II_large_holo_zero} backwards into \cref{cau_ex_fw_larger_gen_fun_form} permits us to write blocks of matrix $\Wcal$ as
\begin{subequations}
\begin{align}
	\Wcal^{n}_{j,r}
	&=
	-
	\frac{(\rl_j+i)^r+\rl_j^r}{\aux'_e(\rl_j)}
	-
	\frac{1}{\pi}
	\Phifn'(\rl_j|\bm{\check\la},\bm{\rh^+})
	f_r(\rl_j|\bm\eta)
	\label{cau_ex_large_fw_II_sol_north}
	,
	\\
	\Wcal^{c}_{a,r}
	&=
	-
	\frac{1}{\pi}
	\Phifn'(\clp<+>_a|\bm{\check\la},\bm{\rh^+})
	f_r(\clp<+>_a|\bm\eta)
	\label{cau_ex_large_fw_II_sol_centre}
	,
	\\
	\Wcal^{s}_{a,r}
	&=
	-
	\frac{(\hle_a+i)^r+\hle_a^r}{\aux'_e(\hle_a)}
	-
	\frac{1}{\pi}
	\Phifn'(\hle_a|\bm{\check\la},\bm{\rh^+})
	f_r(\hle_a|\bm\eta)
	.
	\label{cau_ex_large_fw_II_sol_south}
\end{align}
\end{subequations}
It is important to note that extra terms due to $\bm\eta$ appear as a linear sum over $f_r(\rh^+_j|\bm\eta)$ given by \cref{cau_ex_fw_larger_gen_fun_form}.
Therefore it is a linear sum over the columns which arises from the action on identity block and can be immediately cancelled in the \ref{cau_ex_II_method_2}.
We therefore need not give an explicit form for the coefficients.
\begin{rem}
We claimed that it was safe to assume that added parameters are real $\bm\eta\subset\Rset$. This was due to two main reasons: 
\begin{enumerate}
\item The extra parameters are unwanted so removing those who do not fall inside the contour and they do not affect the computation.
\item These excluded parameters in the original contour can enter the newly drawn contour in \cref{cau_ex_II_gen_fw_large_new_contour}, hence reintroducing them into $f_r$ function.
\end{enumerate}
However, in the end, this difference will be unimportant since we can cancel the extra terms in the determinant.
\end{rem}
The correspondence between the \ref{cau_ex_II_method_1} and \ref{cau_ex_II_method_2} tells us that when extra parameters are send to infinity, the sum $f_r(\rh^+_j)$ becomes the linear sum over columns $\Zcal$ in \cref{cau_ex_II_mat_blocks}.
\index{ff@\textbf{Form-factors}!mat cvextn es fw@\hspace{1em}$\Wcal$: Foda-Wheeler block inside \rule{3em}{1pt}|textbf}%
\begin{subequations}
\begin{align}
	\Wcal^{n}_{j,r}
	&=
	-
	\frac{(\rl_j+i)^r+\rl_j^r}{\aux'_e(\rl_j)}
	-
	\sum_{s=1}^{\ho{n}}
	\chi^{w}_{r,s}
	\Vdr_{s}[\rl_j]
	\label{cau_ex_fw_II_sol_north}
	,
	\\
	\Wcal^{c}_{a,r}
	&=
	-
	\sum_{s=1}^{\ho{n}}
	\chi^{w}_{r,s}
	\Vdr_{s}[\clp<+>_a]
	\label{cau_ex_fw_II_sol_centre}
	,
	\\
	\Wcal^{s}_{a,r}
	&=
	-
	\frac{(\hle_a+i)^r+\hle_a^r}{\aux'_e(\hle_a)}
	-
	\sum_{s=1}^{\ho{n}}
	\chi^{w}_{r,s}
	\Vdr_{s}[\hle_a]
	.
\end{align}
\label{cau_ex_fw_II_comps}
\end{subequations}
\begin{rem}
This procedure of sending extra parameters to infinity is not new.
We would also like to point out that we already used it in \cref{sec:FW_append_pf} for the proof of the Foda-Wheeler version of the Slavnov formula \cite{FodW12a}.
This comparison also sheds a light on the issue at hand, that we face while using the \ref{cau_ex_II_method_1} for the extraction for the Foda-Wheeler block matrix $\Wcal$. Since both matrices in the extraction were obtained through limit, the order of this limits plays an important role.
The exponential divergence of the $\Phifn$ can be attributed to its pole accumulating at the infinity due to this limit. 
\par
In the case of $\resmat[e]^{\text{cau}}$ and $\Ho{\Pcal}$ matrices the extraction is more robust and both methods can be applied.
This robustness is due to the Cauchy structure of the original blocks $\modCau[e]^{\text{cau}}$ and $\modCau*$ but it won't be incorrect to say that this robustness can be also attributed to the fact that we perform the limiting procedure only once unlike the Foda-Wheeler terms $\Wcal$.
\end{rem}
\subsection{Reduced matrices}
In \cref{sec:gau_ex_I_gen} we saw that the matrix $\resmat[g]$ can be reduced to a small matrix $\resmat*[g]$ \eqref{cau_ex_I_eff_mat} of finite order $\ho{n}$ which is equivalent with the original matrix up-to the determinant \eqref{cau_ex_I_reduction}.
We saw that it is made up of sub-blocks:
\begin{align}
	\resmat*[g]
	=
	\begin{pmatrix}
		\resmat*[g]^{c}
		&
		\resmat*[g]^{w+}
		&
		\resmat*[g]^{w-}
	\end{pmatrix}
\end{align}
We also saw that its components are described by \cref{cau_ex_I_clp_int_form,cau_ex_I_wdp_int_form} for the close-pair $\resmat*[g]^{c}$ and wide-pair $\resmat[g]^{w\pm}$ respectively.
\par
Now we shall do the same for the matrix $\resmat[e]$.
The first step is to cancel the linear combinations of the Vandermonde-like $\Zcal$ block from \cref{cv_extn_II_cau_comps,cv_extn_II_hl_comps,cau_ex_fw_II_comps}.
This leads to the following block structure for the matrix $\resmat[e]$:
\index{ff@\textbf{Form-factors}!mat cvextn es fw@\hspace{1em}$\Wcal$: Foda-Wheeler block inside \rule{3em}{1pt}}%
\index{ff@\textbf{Form-factors}!mat cvextn es hyper van@\hspace{1em}$\Zcal$: hyperbolic Vandermonde block inside \rule{3em}{1pt}}%
\index{ff@\textbf{Form-factors}!mat cvextn es hl@\hspace{1em}$\Ho\Pcal$: higher-level block inside \rule{3em}{1pt}}%
\begin{align}
	\Pcal_e=
	\left[
	\begin{tikzpicture}[x=1.25mm,y=.4mm, baseline=(current bounding box.center)]
	\clip (-51,-35) rectangle (31,70);
	\draw[thin, shorten <=2, shorten >=3] (-50,70)-- node[black, midway, fill=white, circle, inner sep = 10] {$\Id_{n_r+n_\text{c}}$} (0,0);
	\fill[gray!20, opacity=.5] (-50,-30) rectangle (0,0);
	\path[thin, gray, shorten <=2, shorten >=2] (-50,-15) -- node[black, midway, circle, inner sep= 1pt] {$\Acal^{-1}\Rcal[\bm\hle|\bm{\rl^+}]$} (0,-15);
	\fill[gray!20, opacity=.5] (0,0) rectangle (30,70);
	\fill[white] (0,0) rectangle (10,70);
	\fill[white] (10,0) rectangle (20,20);
	\path[thin, gray, shorten <=2, shorten >=2] (5,0) -- node[black, midway,circle] {0} (5,70);
	\fill[thin, gray, shorten <=2, shorten >=2] (15,20) -- node[black, midway, inner sep=1] {$\Wcal^{n}$} (15,70);
	\node at (15,10) {0};
	\path[thin, gray, shorten <=2, shorten >=2] (25,0) -- node[black, midway, circle, inner sep=0] {$\Zcal^{n}$} (25,70);
	\fill[gray!60, opacity=.5] (0,0) rectangle (30,-30);
	\path[thin, gray, shorten <=2, shorten >=2] (5,0) -- node[black, midway, circle, inner sep=0] {$\Acal^{-1}\ho{\Scal}$} (5,-30);
	\path[thin, gray, shorten <=2, shorten >=2] (15,0) -- node[black, midway, inner sep=.8, rectangle] {$\Wcal^{s}$} (15,-30);
	\path[thin, gray, shorten <=2, shorten >=2] (25,0) -- node[black, midway, circle, inner sep=0] {$\Zcal^{s}$} (25,-30);
	\draw[thick] (-50,0)--(30,0);
	\draw[thick] (0,-30)--(0,70);
	\draw[thick] (10,-30)--(10,70);
	\draw[thick] (20,-30)--(20,70);
	\draw[thick] (10,20)--(20,20);
	\end{tikzpicture}
	\right]
	.
	\label{cau_ex_II_gen_blocks_diagram}
\end{align}
According to the \cref{lem:mat_det_red} we can construct a smaller matrix $\resmat*[e]$ of order $n_h$, which is composed of the blocks:
\index{ff@\textbf{Form-factors}!mat cvextn es res@$\resmat*[e]$: reduced mat. obtained from CV extraction for an excited state|textbf}%
\index{ff@\textbf{Form-factors}!mat hl gauex@$\ho\Scal$: result of the higher-level Gaudin extraction}%
\begin{align}
	\resmat*[e]=
	\frac{1}{\pi}
	\begin{pmatrix}
	\ho{\Scal}
	&
	\Acal\Wcal^{\text{eff}}
	&
	\Acal\Zcal^{\text{eff}}
	\end{pmatrix}
	\label{cau_ex_II_eff_mat_def}
\end{align}
so that it is equivalent to the $\resmat[e]$ upto the determinants, since we have
\begin{align}
	\det_{N_0+\ho{n}+1}\Pcal_e
	=
	\frac{1}{\bmprod\aux'_e(\bm\hle)}
	\det_{n_h}\Qcal_e
	\label{cau_ex_II_reduction}
	.
\end{align}
In this process, we also extracted a determinant of a diagonal matrix $\Acal^{-1}$ into the prefactor.
The effective matrices $\Wcal^{\text{eff}}$ and $\Zcal^{\text{eff}}$ are thus given by,
\index{ff@\textbf{Form-factors}!mat cvextn es res fw eff@\hspace{1em}$\Wcal^{\text{eff}}$: effective Foda-Wheeler block inside \rule{3em}{1pt}}%
\begin{align}
	\Acal\Wcal^{\text{eff}}[\bm\hle]
	&=
	\Acal[\bm\hle]
	\Wcal[\bm\hle]
	-
	\Rcal[\bm\hle\Vert\bm\rl]\cdot\Wcal[\bm\rl]
	,
	\label{cau_ex_II_FW_eff_mat_form}
\shortintertext{and}
	\Acal\Zcal^{\text{eff}}[\bm\hle]
	&=
	\Acal[\bm\hle]
	\Zcal[\bm\hle]
	-
	\Rcal[\bm\hle\Vert\bm{\rl^+}]
	\cdot
	\Zcal[\bm{\rl^+}]
	.
	\label{cau_ex_II_hyper_Van_eff_mat_form}
\end{align}
\index{ff@\textbf{Form-factors}!mat cvextn es res van eff@\hspace{1em}$\Zcal^{\text{eff}}$: effective Vandermonde block inside \rule{3em}{1pt}}%
Let us now compute effective matrices $\Wcal^{\text{eff}}$ and $\Zcal^{\text{eff}}$.
We can see that see that components of the $\Wcal^{\text{eff}}$ matrix \eqref{cau_ex_II_FW_eff_mat_form} can be obtained from the following:
\begin{align}
	\Wcal^{\text{eff}}_{a,r+1}
	=	
	-
	\frac{2(\hle_a+\tfrac{i}{2})^r}{\aux_e'(\hle_a)}
	-
	2\pi i
	\bmsum_{\bm\rl}
	\frac{2%
	(\bm\rl+\tfrac{i}{2})^r
	}{\aux_e'(\bm\rl)}
	\rden_h(\bm\rl-\hle_a)
	,
	\qquad
	(r=0,1)
	.
	\label{cau_ex_II_fw_eff_block_mat}
\end{align}
Its computation is similar to the one performed in \crefrange{reduced_mat_fw_2sp}{res_mat_fw_2sp_result} for the two-spinon form-factor.
The summation is taken over the real Bethe roots only.
It can be computed using the regular condensation property and thus we obtain the result
\begin{align}
	\Acal\Wcal^\text{eff}_{a,r+1}
	=
	-
	\left(\hle_a+\frac{i}{2}\right)^r
	,
	\qquad
	(r=0,1)
	.
	\label{cau_ex_II_fw_eff_block_comps}
\end{align}
For the effective hyperbolic Vandermonde block $\Zcal^{\text{eff}}$ we get from \cref{cau_ex_II_hyper_Van_eff_mat_form}
\begin{align}
	\Zcal^\text{eff}_{a,b}
	&=
	\Phifn'(\hle_{a}|\bm{\check\la},\bm{\rh^+})
	\Vmat*_{a}(\hle_b)
	+
	\frac{2\pi i}{\aux_e(\hle_a)}
	\bmsum_{\bm{\rl^+}}
	\rden_h(\hle_a-\bm{\rl^+})
	\Phifn'(\bm{\rl^+}|\bm{\check\la},\bm{\rh^+})
	\Vmat*_{a}(\bm{\rl^+})
	.
	\label{cau_ex_II_hvan_eff_block_mat}
\end{align}
where we have also used \cref{hyper_cau_van_inv_hvan_block} to write $\Vdr_a(\nu)=\Vmat*_a(\nu)\Phifn'(\nu|\bm{\check\la},\bm{\rh^+})$.
These can be computed in the similar manner as we did for the close-pair Cauchy extraction of the first type in \cref{sec:gau_ex_I_gen}.
Let us consider the function $f_a$ defined by the summation:
\begin{align}
	f_a(\nu|\bm{\rh^+})
	=
	2\pi i
	\bmsum_{\bm{\rh^+}}
	\rden_h(\nu-\bm{\rh^+})
	\Phifn'(\bm{\rh^+}|\bm{\check\la},\bm{\rh^+})
	\Vmat*_{a}(\bm{\rh^+})
	.
	\label{hvan_eff_block_gen_fun}
\end{align}
It is sufficient to define it over the real values of $\nu\in\Rset$ since we want to compute these sums for $\nu\in\bm\hle$, which are always real.
Note that this sum is taken over all the roots $\bm{\rh^+}$ which includes the holes unlike the original sum which runs over $\bm{\rl^+}$, excluding the holes.
In terms of this function we can write the effective $\Zcal$ block as
\begin{align}
	\Acal\Zcal^\text{eff}_{a,b}
	=
	\Big(
	\Acal-\Rcal
	\Big)
	\Vdr_{ab}[\bm\hle]
	+
	f_a(\hle_b|\bm{\rh^+})
	.
	\label{cau_ex_II_hvan_block_intermediate_fun_form}
\end{align}
Now we write the function $f_a$ as an integral over the residues of the $\Phifn$ function.
Since the set ${\rh^+}$ contains the positive close pair, we choose the contour which passes through $\Rset-i\alpha$ and $\Rset+i-i\alpha$ with $\alpha<\frac{1}{2}$.
We also remark that the function $\rden_h$ is analytic in this region.
Thus we can rewrite the function $f_a$ as
\begin{align}
	f_a(\nu|\bm{\rh^+})
	=
	\pi
	\left(
	\int_{\Rset-i\alpha}
	-
	\int_{\Rset+i\alpha}
	\right)
	\rden_h(\nu-\tau)
	\Phifn(\tau|\bm{\check\la},\bm{\rh^+})
	\Vmat*_a(\tau)
	d\tau
	.
\end{align}
Using the periodicity of the $\Phifn$ and $\Vmat*_a$ functions and the semi-periodicity of the $\rden_h$ function \eqref{rden_h_clp_identity}, it can be converted according to the following:
\begin{align}
	f_a(\nu|\bm{\rh^+})
	=
	\pi
	\int_{\Rset-i\alpha}
	p'_0(\nu-\tau)
	\Phifn(\nu|\bm{\check\la},\bm{\rh^+})
	\Vmat*_a(\tau)
	d\tau
	.
	\label{hvan_block_gen_fun_int_form}
\end{align}
Substituting it in \cref{cau_ex_II_hvan_block_intermediate_fun_form} gives us the following expression:
\begin{align}
	\Acal\Zcal^\text{eff}_{a,b}
	=
	\Big(
	\Acal-\Rcal
	\Big)
	\Vdr_a[\hle_b]
	+
	\pi
	\int_{\Rset-i\alpha}
	p'_0(\nu-\tau)
	\Phifn(\nu|\bm{\check\la},\bm{\rh^+})
	\Vmat*_a(\tau)
	d\tau
	.
	\label{cau_ex_II_hvan_block_int_form}
\end{align}
We do not compute the integral with the $\Phifn$ explicitly. This will be the final expression that we write for the block $\Zcal^{\text{eff}}$, similar to blocks of $\resmat*[g]$ which were also represented in terms of integrals over $\Phifn$ function.
\par
Finally let us recall that the matrix $\ho{\Scal}$ is a result of the higher-level Gaudin extraction.
It is unaffected in the reduction of size and thus retains its original form \eqref{hl_ff_gau_ex_all}. It reads,
\begin{align}
	\ho{\Scal}[\bm\hle\Vert\bm\cid]
	=
	\Ho{\Rcal}[\bm\hle\Vert\bm\cid]
	\cdot
	\Ho{\Ncal}^{-1}[\bm\cid\Vert\bm\cid]
	.
	\label{cau_ex_II_hl_gauex}
\end{align}
Equivalently, we can write the following system of linear equation for the matrix $\ho{\Scal}$:
\begin{align}
	\aux*'(\cid_a)
	\ho{\Scal}_{a,b}
	-
	\sum_{r=1}^{\ho{n}}
	K(\cid_a-\cid_r)
	\ho{\Scal}_{b,r}
	=
	-2\pi i\ho{\rden}(\cid_a-\hle_b)
	=
	-ip'_0(\cid_a-\hle_b)
	.
	\label{hl_gau_ex_hl_ff_pref_chap}
\end{align}
This is a system of finite size $\ho{n}$ and hence it can be solved with the traditional methods of the linear algebra.
\begin{rem}
Let us also remark that we are usually concerned with the excitations with very small number $\ho{n}$ of higher-level roots as it can be seen from the relation \eqref{hle_num_ho_num_rel}, e.g. $\ho{n}=1$ for a four-spinon triplet, $\ho{n}=2$ for a six-spinon triplet and so on.
Therefore the higher-level Gaudin extraction need not be seen as another difficulty.
\par
On the contrary, an emergence of the higher-level Bethe Gaudin matrix and its role in describing one of the blocks in the final determinant $\resmat*[e]$ is a very significant result.
Drawing parallels with the emergence of higher-level Bethe equations in the case of spectrum \cite{DesL82,BabVV83}, we can see that this result does a same thing for the form-factors.
In the case of spectrum we see that density of the excited state is dominated by the density function for ground state, we thus take the ratio of two exponential counting functions $\aux_e$ and $\aux_g$ and found that the higher-level Bethe equation emerges for the complex roots due to this factorisation.
Similarly here we find that a higher-level Gaudin extraction emerges for the complex roots due to the factorisation of Gaudin and then Cauchy-Vandermonde matrices.
\end{rem}
\section{Thermodynamic limit from the infinite product form}
\label{sec:tdl_pref_gen}
Substituting \cref{cau_ex_II_reduction,cau_ex_I_reduction} into the expression \eqref{cau_det_rep_gen_extr_cau} gives us the representation:
\begin{multline}
	\left|\FF^{z}\right|^2=
	-2\pi^{M+1}%
	\frac{1}{\bmprod\aux'_e(\bm\hle)}
	\frac{%
	\bmprod \revtf(\bm\la)
	}{%
	\bmprod \revtf(\bm{\rl})
	}
	\frac{%
	\bmprod (\bm{\rl}-\bm\la) \bmprod(\bm\la-\bm\mu)
	}{%
	\bmprod^\prime (\bm{\rl}-\bm\mu) \bmprod^\prime (\bm\la-\bm\la)
	}
	\\
	\times
	\frac{%
	\bmprod
	\baxq_{g}(\bmclp<+>-i)
	\bmprod
	\baxq_{g}(\bmclp<->-i)
	}{%
	\bmprod
	\baxq_{e}^\prime(\bmclp<+>-i)
	\bmprod
	\baxq_{e}(\bmclp<->-i)
	}
	\frac{%
	\bmprod
	\baxq_{g}(\bmwdp<+>-i)
	\bmprod
	\baxq_{g}(\bmwdp<->-i)
	}{%
	\bmprod
	\baxq_{e}(\bmwdp<+>-i)
	\bmprod
	\baxq_{e}(\bmwdp<->-i)
	}
	\\
	\times
	\bmalt\sinh\pi(\bm{\check\rl^+}\Vert\bm\la)
	~
	\bmalt\sinh\pi(\bm{\check\la}\Vert\bm{\rh^+})
	\\
	\times
	\det_{\ho{n}}\resmat*[g]
	\det_{n_h}\resmat*[e]
	.
	\label{gen_pref_FF_expn}
\end{multline}
Let us first recall in the following table the cardinalities of all the different sets involved.
\begin{table}[h]
\label{tab:cards_rappel}
\begin{center}
\begin{tabular}{|c||c|c|c||c|c||c|c|c||c|c|}
\hline
$\bm\la$						&%
$\bm\mu$						&
$\bm\hle$						&
$\bm\cid$						&%
$\bmclp$						&
$\bmwdp$						&%
$\rl$								&
$\bm{\rl^+}$				&
$\bm{\rl^\pm}$			&%
$\bm{\rh}$					&
$\bm{\rh^+}$				
\\
\hline
$N_0$
&
$N_1$
&
$n_h$
&
$\ho{n}$
&
$n_c$
&
$n_w$
&
$n_r$
&
$n_r+n_c$
&
$n_r+2n_c$
&
$n_r+n_h$
&
$N_0+1+\ho{n}$
\\
\hline
\end{tabular}
\end{center}
\caption[Cardinalities of sets for some key examples	]{Cardinalities of sets for some key examples (recalled)}
\end{table}
\par
Let us also recall that numbers $N_0$ and $N_1$ are related to the length of the chain $M$ through the following expression:
\begin{align}
	N_s=\frac{M}{2}-s
	.
\end{align}
We know that these numbers are related to each other. One of these relations is the difference of real roots:
\begin{align}
	N_0-n_r
	=
	1+2n_\txtcp+2n_\txtwp
	.
	\label{gen_pref_qno_diffreal_cp_wp}
\end{align}
It is related to the number of spinons $n_h$, or equivalently the number of higher-level roots $\ho{n}=\frac{n_h}{2}-1$, through the following relations:
\begin{subequations}
\begin{align}
	N_0-n_r
	&=
	1+\ho{n}+n_\txtcp
	\label{gen_pref_qno_diffreal_ho_cp}
\shortintertext{or,}
	N_0-n_r
	&=
	\frac{1}{2}(n_h+2n_\txtcp).
	\label{gen_pref_qno_diffreal_hle_cp}
\end{align}
\label{gen_pref_qno_diffreal_cp_both}
\end{subequations}
The main advantage of the above two relations \eqref{gen_pref_qno_diffreal_cp_both} over \eqref{gen_pref_qno_diffreal_cp_wp} is that it leaves the number of wide-pair implicit.
To this effect, we further introduce the following quantum numbers which will it easier to follow some of the exponents and constants that arise in our computation.
\begin{notn}
\label{def:pref_tdl_qno}
Quantum numbers $p$ and $q$ are defined in terms of $n_h$ and $n_r$ as
\begin{align}
	\frac{p+q}{2}&=n_h
	,
	&
	\frac{p-q}{2}&=N_0-n_r
	=
	\frac{1}{2}n_h+n_c
	.
	\label{gen_pref_qno_pq_sum-diff}
\end{align}
We can see that they can be expressed as
\begin{align}
	p&
	=\frac{3}{2}n_h+n_\txtcp
	=3+3\ho{n}+n_\txtcp
	,
	&
	q&
	=\frac{1}{2}n_h-n_\txtcp
	=1+\ho{n}-n_\txtcp
	.
	\label{gen_pref_qno_pq_def}
\end{align}
\end{notn}
\begin{notn}
\label{def:pref_tdl_Qno}
We also define the quantum numbers $P$ and $Q$ as
\begin{subequations}
\label{gen_pref_qno_PQ_def}
\begin{align}
	P&=\left(N_0-n_r\right)p=(1+\ho{n}+n_{\txtcp})(3+3\ho{n}+n_\txtcp),
	\label{gen_pref_qno_P_def}
	\\
	Q&=\left(N_0-n_r\right)q=(1+\ho{n}+n_\txtcp)(1+\ho{n}-n_\txtcp).
	\label{gen_pref_qno_Q_def}
\end{align}
\end{subequations}
We can see that the sum and difference of $P$ and $Q$ are given by,
\begin{subequations}
\label{gen_pref_qno_PQ_sum-diff}
\begin{align}
	P-Q&=
	\frac{1}{2}(n_h+2n_\txtcp)^2
	,
	\label{gen_pref_qno_PQ_diff}
	\\
	P+Q&=
	n_h^2+2n_h n_\txtcp
	.
	\label{gen_pref_qno_PQ_sum}
\end{align}
\end{subequations}
\end{notn}
We will now compute the prefactor in the expression \eqref{gen_pref_FF_expn} in the thermodynamic limit.
This computation is divided into four stages to make it easier to follow.
At the end of each stage we revise the formula \eqref{gen_pref_FF_expn} for the form-factor.
The last two stages are computationally intensive, some of the intermediate formulae are put in \cref{sec:pref_append_inf_prod} at the end of this chapter.
\minisec{Stage one}
The product over holes $\bmprod_{\bm\hle}\cosh\pi\bm\hle$ comes from the determinant $\det(\pi\Acal^{-1})$ which was extracted when we defined $\resmat*[e]$ in \cref{cau_ex_II_eff_mat_def,cau_ex_II_reduction}.
We then use the fact that all the holes are chosen inside the bulk of the Fermi-distribution, thus up-to the leading order we can write
\begin{align}
	\aux'_e(\hle_a)
	=
	-2\pi i M \rden_e(\hle_a) + O(1)
	=
	\frac{-i\pi M}{\cosh\pi\hle_a} + O(1)
	.
\end{align}
We also find that the same product is contained in the Cauchy-Vandermonde determinants in the form of cross products for the combination of $\check\la_{N_0+1}=\frac{i}{2}$ with any of the holes $\hle_a\in\bm\hle$.
We can separate all the terms due to $\check{\rl^+}_{n_r+1}=\frac{i}{2}$ and $\check\la_{N_0+1}=\frac{i}{2}$ from these products.
We see that they mutually cancel out to produce the following simplification in \cref{gen_pref_FF_expn}:
\begin{align}
	\bmprod i{\cosh\pi\bm\hle}
	\bmalt\sinh\pi(\bm{\check\rl^+}\bm\Vert\bm\la)
	\bmalt\sinh\pi(\bm{\check\la}\bm\Vert\bm{\rh^+})
	=	
	\bmalt\sinh\pi(\bm{\rl^+}\bm\Vert\bm\la)
	\bmalt\sinh\pi(\bm\la\bm\Vert\bm{\rh^+})
	.
	\label{cv_det_pref_hle_i2_recomb}
\end{align}
Let us now use the periodicity of the hyperbolic function and the condition \eqref{clp_DL} for the formation of the close-pairs to write
\begin{align}
	\sinh\pi(\la-\clp<+>)=\sinh\pi(\clp<->-\la)
	.
\end{align}
We ignore the string deviation terms as they are exponentially small in the thermodynamic limit.
It permits us to recast the product of Cauchy-Vandermonde determinant in \cref{cv_det_pref_hle_i2_recomb} into the product:
\begin{multline}
	\bmalt\sinh\pi(\bm{\rl^+}\bm\Vert\bm\la)
	\bmalt\sinh\pi(\bm\la\bm\Vert\bm{\rh^+})	
	\\
	=
	\frac%
	{%
	\bmprod^\prime\sinh\pi(\bm\rl\bm-\bm{\rl^\pm})
	}{%
	\bmprod\sinh\pi(\bm\rl-\bm\la)
	}%
	\frac%
	{%
	\bmprod^\prime\sinh\pi(\bm\la\bm-\bm\la)
	}{%
	\bmprod\sinh\pi(\bm\la\bm-\bm{\rl^\pm})
	}%
	\frac%
	{%
	\bmprod\sinh\pi(\bm{\rl^\pm}\bm-\bm\hle)
	}{%
	\bmprod\sinh\pi(\bm\la\bm-\bm\hle)
	}%
	\\
	\times
	\bmprod\sinh\pi(\bmclp<+>\bm-\bm\hle)
	\bmprod^\prime\sinh\pi(\bmclp\bm-\bmclp)
	\bmalt\sinh\pi\bm\hle
	.
	\label{gen_pref_hle_terms_sepn}
\end{multline}
Here $\bm{\rl^\pm}=\rl\bm\cup\bmclp<+>\bm\cup\bmclp<->$ denotes the set of real and all close-pair roots, in contrast to $\bm{\rl^+}$ which only contains the positive close-pairs.
In other words it is the set of Bethe roots except the wide pairs $\bm{\rl^\pm}=\bm\mu\setminus(\bmwdp<+>\cup\bmwdp<->)$.
Let us now substitute \cref{gen_pref_hle_terms_sepn,cv_det_pref_hle_i2_recomb} into the expression \eqref{gen_pref_FF_expn}.
It allows us to write
\begin{multline}
	\left|\FF^{z}\right|^2=
	\frac{%
	-2\pi^{M-n_h+1}%
	}{%
	M^{n_h}
	}
	\begin{aligned}[t]
	&
	\frac{%
	\bmprod (\bm{\rl}-\bm\la) \bmprod(\bm\la-\bm\mu)
	}{%
	\bmprod^\prime (\bm{\rl}-\bm\mu) \bmprod^\prime (\bm\la-\bm\la)
	}
	\\
	&
	\times
	\frac{%
	\bmprod \revtf(\bm\la)
	}{%
	\bmprod \revtf(\bm\rl)
	}
	\frac%
	{%
	\bmprod^\prime\sinh\pi(\bm\rl\bm-\bm{\rl^\pm})
	}{%
	\bmprod\sinh\pi(\bm\rl\bm-\bm\la)
	}%
	\frac%
	{%
	\bmprod^\prime\sinh\pi(\bm\la\bm-\bm\la)
	}{%
	\bmprod\sinh\pi(\bm\la\bm-\bm{\rl^\pm})
	}%
	\frac%
	{%
	\bmprod\sinh\pi(\bm\rl\bm-\bm\hle)
	}{%
	\bmprod\sinh\pi(\bm\la\bm-\bm\hle)
	}%
	&
	\hfill
	(\ast)
	\\
	&
	\quad
	\times
	\frac{%
	\bmprod
	\baxq_{g}(\bmclp<+>-i)
	\bmprod
	\baxq_{g}(\bmclp<->-i)
	}{%
	\bmprod
	\baxq_{e}^\prime(\bmclp<+>-i)
	\bmprod
	\baxq_{e}(\bmclp<->-i)
	}
	\frac{%
	\bmprod
	\baxq_{g}(\bmwdp<+>-i)
	\bmprod
	\baxq_{g}(\bmwdp<->-i)
	}{%
	\bmprod
	\baxq_{e}(\bmwdp<+>-i)
	\bmprod
	\baxq_{e}(\bmwdp<->-i)
	}
	&
	\\
	&
	\qquad
	\times
	\bmprod_{\bmclp<+>,\bm\hle}\sinh\pi(\bmclp<+>\bm-\bm\hle)
	\bmprod^\prime\sinh\pi(\bmclp\bm-\bmclp)
	\bmalt\sinh\pi\bm\hle
	&
	\end{aligned}
	\\
	\times
	\det_{\ho{n}}\resmat*[g]
	\det_{n_h}\resmat*[e]
	.
	\label{gen_pref_FF_expn_step_1}
\end{multline}
The terms on the line marked with $(\ast)$ are computed in the next stage.
\minisec{Stage two}
Let us substitute the following expression for the thermodynamic limit of the function $\revtf$ in \cref{gen_pref_FF_expn_step_1}. We recall that it represents the ratio of eigenvalues \eqref{def_revtf} and its thermodynamic limit is computed in \cref{revtf_tdl}.
\begin{align}
	\revtf(\tau)=\bmprod\tanh\frac{\pi(\tau-\bm\hle)}{2}
	.
	\label{gen_pref_revtf_tdl}
\end{align}
Let us also split all the $\sinh$ terms of the line $(\ast)$ of \cref{gen_pref_FF_expn_step_1} in two parts according to
\begin{align}
	\sinh\pi\nu=2\sinh\frac{\pi\nu}{2}\cosh\frac{\pi\nu}{2}
	.
	\label{gen_pref_sinh_split}
\end{align}
Substitution of the above two \cref{gen_pref_revtf_tdl,gen_pref_sinh_split} in \cref{gen_pref_FF_expn_step_1} allows us to rewrite its line $(\ast)$ as
\begin{multline}
	\frac{%
	\bmprod \revtf(\bm\la)
	}{%
	\bmprod \revtf(\bm\rl)
	}
	\frac%
	{%
	\bmprod^\prime\sinh\pi(\bm\rl\bm-\bm{\rl^\pm})
	}{%
	\bmprod\sinh\pi(\bm\rl\bm-\bm\la)
	}%
	\frac%
	{%
	\bmprod^\prime\sinh\pi(\bm\la\bm-\bm\la)
	}{%
	\bmprod\sinh\pi(\bm\la\bm-\bm{\rl^\pm})
	}%
	\frac%
	{%
	\bmprod\sinh\pi(\bm\rl\bm-\bm\hle)
	}{%
	\bmprod\sinh\pi(\bm\la\bm-\bm\hle)
	}%
	\\
	\begin{aligned}[t]
	=
	2^{-\frac{P-Q}{2}}
	2^{-(n_{r}+\frac{M}{2})}
	\frac%
	{%
	\bmprod^\prime\sinh\frac{\pi(\bm\rl\bm-\bm{\rl^\pm})}{2}\bmprod\cosh\frac{\pi(\bm\rl\bm-\bm{\rl^\pm})}{2}
	}{%
	\bmprod\sinh\frac{\pi(\bm\rl\bm-\bm\la)}{2}\bmprod\cosh\frac{\pi(\bm\rl\bm-\bm\la)}{2}
	}%
	\\
	\times
	\frac%
	{%
	\bmprod^\prime\sinh\frac{\pi(\bm\la\bm-\bm\la)}{2}\bmprod\cosh\frac{\pi(\bm\la\bm-\bm\la)}{2}
	}{%
	\bmprod\sinh\frac{\pi(\bm\la\bm-\bm{\rl^\pm})}{2}\bmprod\cosh\frac{\pi(\bm\la\bm-\bm{\rl^\pm})}{2}
	}%
	\end{aligned}
	\\
	\times
	\frac%
	{%
	\bmprod\cosh^2\frac{\pi(\bm\rl\bm-\bm\hle)}{2}
	}{%
	\bmprod\cosh^2\frac{\pi(\bm\la\bm-\bm\hle)}{2}
	}%
\end{multline}
which can be re-expressed in terms of the Gamma functions with the identity [to be added] as
\begin{multline}
	\frac{%
	\bmprod \revtf(\bm\la)
	}{%
	\bmprod \revtf(\bm\rl)
	}
	\frac%
	{%
	\bmprod^\prime\sinh\pi(\bm\rl\bm-\bm{\rl^\pm})
	}{%
	\bmprod\sinh\pi(\bm\rl\bm-\bm\la)
	}%
	\frac%
	{%
	\bmprod^\prime\sinh\pi(\bm\la\bm-\bm\la)
	}{%
	\bmprod\sinh\pi(\bm\la\bm-\bm{\rl^\pm})
	}%
	\frac%
	{%
	\bmprod\sinh\pi(\bm\rl\bm-\bm\hle)
	}{%
	\bmprod\sinh\pi(\bm\la\bm-\bm\hle)
	}%
	\\[\jot]
	=	
	2^{-\frac{P+Q}{2}}
	\pi^{-(P-Q)-(n_r+N_0)}
	\frac{\bmprod^\prime(\bm\rl\bm-\bm{\rl^\pm})\bmprod^\prime(\bm\la\bm-\bm\la)}{\bmprod(\bm\rl\bm-\bm\la)\bmprod(\bm\la\bm-\bm{\rl^\pm})}
	\\
	\begin{aligned}[b]
	\times
	\prod_{\sigma=\pm1}
	\left\lbrace
	\frac{%
	1
	}{%
	\bmprod\Gamma^2\left(\frac{1}{2}+\frac{\bm\rl\bm-\bm\hle}{2i\sigma}\right)
	}%
	\frac{%
	\bmprod\Gamma\left(1+\frac{\bm\rl\bm-\bm\la}{2i\sigma}\right)
	\Gamma\left(\frac{1}{2}+\frac{\bm\rl\bm-\bm\la}{2i\sigma}\right)
	}{%
	\bmprod\Gamma\left(1+\frac{\bm\rl\bm-\bm{\rl^\pm}}{2i\sigma}\right)
	\Gamma\left(\frac{1}{2}+\frac{\bm\rl\bm-\bm{\rl^\pm}}{2i\sigma}\right)
	}%
	\right.
	&	
	\hfill (\ast)
	\\
	\left.
	\times
	\bmprod\Gamma^2\left(\frac{1}{2}+\frac{\bm\la\bm-\bm\hle}{2i\sigma}\right)
	\frac{%
	\bmprod\Gamma\left(1+\frac{\bm\la\bm-\bm{\rl^\pm}}{2i\sigma}\right)
	\Gamma\left(\frac{1}{2}+\frac{\bm\la\bm-\bm{\rl^\pm}}{2i\sigma}\right)
	}{%
	\bmprod\Gamma\left(1+\frac{\bm\la\bm-\bm\la}{2i\sigma}\right)
		\Gamma\left(\frac{1}{2}+\frac{\bm\la\bm-\bm\la}{2i\sigma}\right)
	}%
	\right\rbrace
	.
	&
	\hfill (\ast)
	\end{aligned}
	\label{pref_to_gamma}
\end{multline}
The last two terms marked with $(\ast)$ in the above expression in \cref{pref_to_gamma} have a similar but reciprocate form.
It prompts us to define the following auxiliary function.
\begin{defn}
\label{defn:gen_pref_aux_fn}
\index{aux@\textbf{Auxiliary functions}!pref aux@$\Omegfn$: fn. involved in the computation of the prefactors|textbf}%
The auxiliary function $\Omega(\tau)$ is defined by the following rational form in Gamma functions:
\begin{align}
	\Omega(\tau)=
	\prod_{\sigma=\pm 1}
	\left\lbrace
	\frac{%
	\Gamma^{p}(\frac{1}{2})
	}{%
	\bmprod\Gamma^2\left(\frac{1}{2}+\frac{\tau-\bm\hle}{2i\sigma}\right)%
	}%
	\frac{%
	\bmprod\Gamma\left(1+\frac{\tau\bm-\bm\la}{2i\sigma}\right)
	\Gamma\left(\frac{1}{2}+\frac{\tau\bm-\bm\la}{2i\sigma}\right)
	}{%
	\bmprod\Gamma\left(1+\frac{\tau\bm-\bm{\rl^\pm}}{2i\sigma}\right)
	\Gamma\left(\frac{1}{2}+\frac{\tau\bm-\bm{\rl^\pm}}{2i\sigma}\right)
	}%
	\right\rbrace
	\label{gen_pref_aux_fn_def}
\end{align}
Let us recall that the integer $p$ was introduced in \cref{def:pref_tdl_qno}.
Let's also recall that $\bm{\rl^\pm}$ denotes the union $\bm{\rl^\pm}=\bm\rl\cup\bmclp<+>\cup\bmclp<->$.
\end{defn}
The rational terms in \cref{pref_to_gamma,gen_pref_FF_expn_step_1} cancel out only partially
when the wide-pairs are present since $\bm{\rl^\pm}=\bm\mu\setminus(\bmwdp<+>\cup\bmwdp*<->)$.
\begin{align}
	\frac{\bmprod^\prime(\bm\rl\bm-\bm{\rl^\pm})\bmprod^\prime(\bm\la\bm-\bm\la)}{\bmprod(\bm\rl\bm-\bm\la)\bmprod(\bm\la\bm-\bm{\rl^\pm})}
	\cdot
	\frac{%
	\bmprod (\bm{\rl}-\bm\la) \bmprod(\bm\la-\bm\mu)
	}{%
	\bmprod^\prime (\bm{\rl}-\bm\mu) \bmprod^\prime (\bm\la-\bm\la)
	}
	=
	\frac{%
	\bmprod_{\bm\la}(\bmwdp<+>-\bm\la)
	\bmprod_{\bm\la}(\bmwdp*<->-\bm\la)
	}{%
	\bmprod_{\bm\rl}(\bmwdp<+>-\bm\rl)
	\bmprod_{\bm\rl}(\bmwdp*<->-\bm\rl)
	}	
	\label{gen_pref_rat_terms_cancel_partial_wdp}
\end{align}
\Cref{gen_pref_rat_terms_cancel_partial_wdp,gen_pref_aux_fn_def,pref_to_gamma} permits us to rewrite \cref{gen_pref_FF_expn_step_1} in the notation of the auxiliary functions $\Omegfn$ and $\phifn$ [see \cref{defn:gen_pref_aux_fn,defn:phifn_rat}], as it is shown in the following:
\begin{multline}
	\left|\FF^{z}\right|^2=
	- 2^{-\frac{P+Q-2}{2}}
	\pi^{Q-q+1}	
	M^{-n_h}
	\bmprod_{\bmclp<+>,\bm\hle}\sinh\pi(\bmclp<+>\bm-\bm\hle)
	\bmprod^\prime\sinh\pi(\bmclp\bm-\bmclp)
	\bmalt\sinh\pi\bm\hle
	\\
	\begin{aligned}[t]
	&
	\times
	\bmprod_{\bmwdp<+>} \phifn\left(\bmwdp<+>\big|\bm\la,\bm\rl\right)
	&\cdot\quad&
	\bmprod_{\bmwdp*<->} \phifn\left(\bmwdp*<->\big|\bm\la,\bm\rl\right)
	\\
	&
	\times
	\bmprod_{\bmclp<+>} \phifn'\left(\bmclp<+>-i\big|\bm\la,\bm{\mu}\right)
	&\cdot\quad&
	\bmprod_{\bmclp<->} \phifn\left(\bmclp<->-i\big|\bm\la,\bm\mu\right)
	\\
	&
	\times
	\bmprod_{\bmwdp<+>} \phifn\left(\bmwdp<+>-i\big|\bm\la,\bm{\mu}\right)
	&\cdot\quad&
	\bmprod_{\bmwdp*<->} \phifn\left(\bmwdp*<->-i\big|\bm\la,\bm{\mu}\right)
	\end{aligned}
	\\
	\times
	\frac{\bmprod\Omega(\bm\rl)}{\bmprod\Omega(\bm\la)}
	\cdot
	\det_{\ho{n}}\resmat*[g]
	\det_{n_h}\resmat*[e]
	.
	\label{pref_to_omega_conversion_step_2}
\end{multline}
Note that in $\phifn'(\clp<+>_{a}|\bm\la,\bm\mu)$, the $'$ symbol denotes omission of the pole and the explicit dependence on the string deviation parameters is dropped.
We will now compute the thermodynamic limit of the auxiliary function $\Omegfn$ in the next stage.
\minisec{Stage three}
Let us use the infinite product representation of the $\Gamma$ function to write an infinite product form for the auxiliary function $\Omegfn$:
\begin{align}
	\Omega(\tau)&=
	\prod_{n=1}^{\infty}
	\Omega_{n}(\tau)
	.
	\label{gen_pref_aux_fun_infprod}
\end{align}
This infinite product can be commuted with the product over the function $\Omegfn$ for real roots in \cref{pref_to_omega_conversion_step_2}.
This gives the infinite product form for
\begin{align}
	\frac{\bmprod\Omega(\bm\rl)}{\bmprod\Omega(\bm\la)}
	=
	\prod_{n=1}^{\infty}
	\frac{\bmprod\Omega_n(\bm\rl)}{\bmprod\Omega_n(\bm\la)}
	\label{inf_prod_ratio_omegfn_gen_pref}
\end{align}
The general term $\Omegfn$ in this product can be written in terms of the $\phifn$ functions as
\begin{multline}
	\Omega_{n}(\tau)=
	\frac{(2n)^{2q}}{\left(2n-1\right)^{2p}}
	\prod_{\sigma=\pm1}
	\phi\left(\tau+2in\sigma\Big|\bm{\rl^\pm},\bm\la\right)%
	\\
	\times
	\prod_{\sigma=\pm1}
	\phi\left(\tau+(2n-1)i\sigma\Big|\bm{\rl^\pm},\bm\la\right)%
	\Bigg\lbrace
	\bmprod_{\bm\hle}
	\left(
	\left(2n-1\right)^2
	+
	{(\tau-\bm\hle)^2}%
	\right)
	\Bigg\rbrace
	^{2}
	.
	\label{pref_gen_aux_fun_gen-term_phifn_form}
\end{multline}
Its thermodynamic limit is computed by substituting that of the $\phifn$ functions in \cref{sub:pref_append_omeg_n_tdl}.
There we obtained the expression \eqref{pref_gen_append_tdl_omeg_n} for it. Substituting this expression for $\Omegfn_n$ in \cref{inf_prod_ratio_omegfn_gen_pref} gives us the thermodynamic limit of the general term in this infinite product \eqref{gen_pref_ratio_omegfn_inf_prod_gen-term_tdl_append} which is reproduced in the following expression:
\begin{multline}
	\frac{\bmprod\Omega_{n}(\bm\rl)}{\bmprod\Omega_{n}(\bm\la)}
	=
	\frac{(2n-1)^{2P}}{(2n)^{2Q}}
	\begin{aligned}[t]
	\bigg\lbrace
	&
	\bmprod_{\bmclp<+>}
	\phi\left(\bmclp<+>-2in\Big|\bm\rl,\bm\la\right)
	\phi\left(\bmclp<+>-(2n-1)i\Big|\bm\rl,\bm\la\right)
	\\
	\times
	&
	\bmprod_{\bmclp<->}
	\phi\left(\bmclp<->+2in\Big|\bm\rl,\bm\la\right)
	\phi\left(\bmclp<->+(2n-1)i\Big|\bm\rl,\bm\la\right)
	\bigg\rbrace
	\end{aligned}
	\\
	\begin{aligned}[b]
	\times
	\bmprod_{\bmwdp<+>}
	\frac{%
	\phi\left(\bmwdp<+>+2(n-1)i\Big|\bm\rl,\bm\la\right)}
	{%
	\phi\left(\bmwdp<+>+2ni\Big|\bm\rl,\bm\la\right)}
	\cdot
	\bmprod_{\bmwdp*<->}
	\frac{%
	\phi\left(\bmwdp*<->-2(n-1)i\Big|\bm\rl,\bm\la\right)
	}{%
	\phi\left(\bmwdp*<->-2ni\Big|\bm\rl,\bm\la\right)
	}
	&
	\\
	\times
	\left\lbrace
	\bmprod_{\bm\hle}
	\phi\left(\bm\hle+(2n-1)i\Big|\bm\rl,\bm\la\right)
	\phi\left(\bm\hle-(2n-1)i\Big|\bm\rl,\bm\la\right)
	\right\rbrace
	&
	.
	\end{aligned}
	\label{gen_pref_ratio_omegfn_inf_prod_gen-term_tdl}
\end{multline}
We now compute in \cref{sub:ratio_omegfn_tdl_append} its thermodynamic limit using the asymptotic form \eqref{phifn_tdl_rl_only} of the $\phifn$ function. 
There we find that the ratio of $\Omegfn_n$ has the thermodynamic limit \eqref{gen_pref_rat_omegfn_gen-term_ratio_tdl_append} which is reproduced in the following expression:
\begin{multline}
	\frac{\bmprod\Omega_{n}(\bm\rl)}{\bmprod\Omega_{n}(\bm\la)}
	=
	2^{-N_h^2}\frac{(2n-1)^{2P}}{(2n)^{2Q}}
	\\[\jot]
	\begin{aligned}[t]
	&\times&
	\bigg\lbrace
	&
	\bmprod_{\bmclp}
	\frac{1}{%
	((\bmclp\bm-\bmclp)^2+(2n)^2)
	}
	\frac{1}{%
	((\bmclp\bm-\bmclp)^2+(2n-1)^2)
	}
	\bigg\rbrace
	\\[\jot]
	&\times&
	\bigg\lbrace
	&
	\bmprod_{\bmclp,\bmwdp}
	\frac{%
	(\bmclp\bm-\bmwdp-(2n+1)i)
	(\bmclp\bm-\bmwdp-(2n-2)i)
	}{%
	(\bmclp\bm-\bmwdp-2ni)
	(\bmclp\bm-\bmwdp-(2n-1)i)
	}
	\\[\jot]
	&\times&
	&
	\bmprod_{\bmclp,\bmwdp*}
	\frac{%
	(\bmclp\bm-\bmwdp*+(2n+1)i)
	(\bmclp\bm-\bmwdp*+(2n-2)i)
	}{%
	(\bmclp\bm-\bmwdp*+2ni)
	(\bmclp\bm-\bmwdp*+(2n-1)i)
	}
	\bigg\rbrace
	\\[\jot]
	&\times&
	\bigg\lbrace
	&
	\bmprod_{\bmwdp,\bmwdp*}
	\frac{(\bmwdp\bm-\bmwdp*+(2n-2)i)(\bmwdp\bm-\bmwdp*+(2n+1)i)}{(\bmwdp\bm-\bmwdp*+(2n-1)i)(\bmwdp\bm-\bmwdp*+2ni)}
	\\[\jot]
	&\times&
	&
	\bmprod_{\bmwdp,\bmwdp*}
	\frac{(\bmwdp*\bm-\bmwdp-(2n-2)i)(\bmwdp*\bm-\bmwdp-(2n+1)i)}{(\bmwdp*\bm-\bmwdp-(2n-1)i)(\bmwdp*\bm-\bmwdp-2ni)}
	\bigg\rbrace
	\\[\jot]
	&\times&
	\bigg\lbrace
	&
	\bmprod_{\bmclp,\bm\hle}
	\frac{1}{%
	((2n-\frac{1}{2})^2+(\bmclp\bm-\bm\hle)^2)
	}
	\frac{1}{%
	((2n-\frac{3}{2})^2+(\bmclp\bm-\bm\hle)^2)
	}
	\bigg\rbrace
	\end{aligned}
	\\[\jot]
	\times
	\prod_{\sigma=\pm 1}
	\bmprod
	\left\lbrace
	\frac{%
	\Gamma\left(n-\frac{1}{2}+\frac{\bm\hle-\bm\hle}{2i\sigma}\right)
	}{%
	\Gamma\left(n+\frac{\bm\hle-\bm\hle}{2i\sigma}\right)
	}
	\right\rbrace
	.
	\label{gen_pref_rat_omegfn_gen-term_ratio_tdl}
\end{multline}
We compute in \cref{sub:tdl_rat_omegfn_append} the infinite product \eqref{inf_prod_ratio_omegfn_gen_pref} with the help of this expression.
Here we only note that the above expression \eqref{gen_pref_rat_omegfn_gen-term_ratio_tdl} can be written entirely in terms of the $\Gamma$ function [see \cref{gen_pref_infprod_gamma_append}]. Thus we can access this infinite product by comparing it with the Weierstrass form of the Barnes G-function, which is exactly what is done in \cref{sub:tdl_rat_omegfn_append}.
There we find in \cref{tdl_rat_omegfn_append} that the infinite product \eqref{inf_prod_ratio_omegfn_gen_pref} is well defined and it converges to the following expression:
\footnote{Let us recall the notation $\bmalt^2 f(\bm\alpha)=\bmalt f(\bm\alpha)\bmalt f(-\bm\alpha)$ which introduced on the \cpageref{ind_free_alt_prod}.}
\begin{multline}
	\frac{\bmprod\Omega(\bm\rl)}{\bmprod\Omega(\bm\la)}
	=
	\frac{(2i)^{n_\txtcp n_h}}{\pi^{\frac{1}{2}n_h^2+Q+n_\txtcp}}
	\begin{aligned}[t]
	&
	\bmprod_{\bmclp}^\prime\frac{\bmclp\bm-\bmclp}{\sinh\pi(\bmclp\bm-\bmclp)}
	&\cdot~&
	\bmprod_{\bmclp,\bm\hle}
	\frac{1}{\sinh\pi(\bmclp<+>\bm-\bm\hle)}
	\\
	\times
	&
	\bmprod_{\bmclp,\bmwdp}
	\frac{%
	\bmclp\bm-\bmwdp
	}{%
	\bmclp\bm-\bmwdp-i
	}
	&\cdot~&
	\bmprod_{\bmclp,\bmwdp*}
	\frac{%
	\bmclp\bm-\bmwdp*
	}{%
	\bmclp\bm-\bmwdp*+i
	}
	\\
	\times
	&
	\bmprod_{\bmwdp,\bmwdp*}
	\frac{%
	\bmwdp\bm-\bmwdp*
	}{%
	\bmwdp\bm-\bmwdp*+i
	}
	&\cdot~&
	\bmprod_{\bmwdp,\bmwdp*}
	\frac{%
	\bmwdp*\bm-\bmwdp
	}{%
	\bmwdp*\bm-\bmwdp-i
	}
	\end{aligned}
	\\
	\times
	\frac{1}{G^{2n_{h}}(\frac{1}{2})}
	\bmalt_{\bm\hle}^2
	\frac{%
	G^{2}\left(1+\frac{\bm\hle}{2i}\right)
	}{%
	G^{2}\left(\frac{1}{2}+\frac{\bm\hle}{2i}\right)
	}
	.
	\label{tdl_rat_omegfn_gen}
\end{multline}
Note that the anomalous phase term $i^{n_hn_c}$ in the above expression \eqref{tdl_rat_omegfn_gen} can be dropped since the product of quantum numbers $n_hn_c$ is multiple of four.
This can be easily seen from the following expressions:
\begin{align}
	n_c=\frac{1}{2}n_h-2n_w-1
	.
	\label{n_clp_n_hle_rel}
\end{align}
which tells us that either $n_h$ is multiple of four or both $n_h$ and $n_c$ are even integers.\footnote{Let us also recall from \cref{chap:spectre} that number of holes $n_h$ is always even.}
\par
Let us now substitute this result \eqref{tdl_rat_omegfn_gen} into \cref{pref_to_omega_conversion_step_2}.
We see that both products $\bmprod'\sinh\pi(\bmclp-\bmclp)$ and $\bmprod\sinh\pi(\bmclp<+>-\bm\hle)$ are cancelled out in this process.
Meanwhile let us also use the following expression to see how $\bmalt\sinh\pi\bm\hle$ in \cref{pref_to_omega_conversion_step_2} and the Barnes G-functions in the above expression \eqref{tdl_rat_omegfn_gen} can also be recombined during this substitution.
\begin{align}
	\bmalt\sinh\pi(\bm\hle)
	=
	\frac{%
	(2\pi)^{n_h(n_h-1)}
	(\bmalt-\hle)^{-1}
	}{%
	\bmalt^2
	\Gamma\left(\frac{\bm\hle}{2i}\right)
	\bmalt^2
	\Gamma\left(\frac{1}{2}+\frac{\bm\hle}{2i}\right)
	}
	.
\end{align}
It allows us to rewrite \cref{pref_to_omega_conversion_step_2} with the following expression:
\begin{multline}
	\left|\FF^{z}\right|^2=
	(-1)^{\frac{n_h(n_h-1)}{2}+1}
	~
	2^{\frac{n_h(n_h-2)}{2}+1}	
	~
	\pi^{\frac{n_h(n_h-3)}{2}+1}	
	M^{-n_h}
	\\
	\begin{aligned}[t]
	&
	\times
	\bmprod'_{\bmclp,\bm\cid}(\bmclp-\bm\cid)
	&\cdot~&
	\bmprod_{\bmclp<+>} \phifn'\left(\bmclp<+>-i\big|\bm\la,\bm{\mu}\right)
	&\cdot~&
	\bmprod_{\bmclp<->} \phifn\left(\bmclp<->-i\big|\bm\la,\bm\mu\right)
	\\
	&
	\times
	\bmprod_{\bmwdp,\bmwdp*}|(\bmwdp-\bmwdp*)|^2
	&\cdot~&
	\bmprod_{\bmwdp<+>} \phifn\left(\bmwdp<+>\big|\bm\la,\bm\rl\right)
	&\cdot~&
	\bmprod_{\bmwdp*<->} \phifn\left(\bmwdp*<->\big|\bm\la,\bm\rl\right)
	\\
	&
	\times
	\bmprod_{\bmwdp,\bmwdp*}|\bmwdp-\bmwdp*+i|^2
	&\cdot~&
	\bmprod_{\bmwdp<+>} \phifn\left(\bmwdp<+>-i\big|\bm\la,\bm{\mu}\right)
	&\cdot~&
	\bmprod_{\bmwdp*<->} \phifn\left(\bmwdp*<->-i\big|\bm\la,\bm{\mu}\right)
	\end{aligned}
	\\
	\times
	\frac{1}{G^{2n_h}(\frac{1}{2})}
	\bmalt^2
	\frac{%
	G(\frac{\bm\hle}{2i})
	G(1+\frac{\bm\hle}{2i})
	}{%
	G(\frac{1}{2}+\frac{\bm\hle}{2i})
	G(\frac{3}{2}+\frac{\bm\hle}{2i})
	}
	~
	\frac{%
	\det_{\ho{n}}\resmat*[g]
	\det_{n_h}\resmat*[e]
	}{%
	\bmalt(\bm\hle)
	}
	.
	\label{pref_gen_step3}
\end{multline}
In the fourth and the final stage we compute the $\phifn$ functions in the prefactor of the above expression \eqref{pref_gen_step3}.
\minisec{Stage four: Reduced determinant representation}.
In \cref{sub:pref_phifn_cmplx_append} we found that the product of the $\phifn$ functions for the close-pairs is given by the following expression in the thermodynamic limit:
\begin{align}
	\bmprod\phifn'(\bmclp<+>-i|\bm\la,\bm\mu)
	\bmprod\phifn(\bmclp<->-i|\bm\la,\bm\mu)
	=
	\bmprod'\frac{1}{\bmclp-\bmclp}
	\bmprod\frac{1}{\bmclp-\bmwdp}
	\bmprod\frac{1}{\bmclp-\bmwdp*}
	\frac{\bmprod(\bmclp-\bm\hle-\frac{i}{2})}{\bmprod(\bmclp-\bm\cid-i)}
	.
\end{align}	
Whereas for the product of $\phifn$ functions for the wide-pairs were found to be given by the following expression:
\begin{multline}
	\bmprod\phifn(\bmwdp<+>|\bm\la,\bm\rl)
	\bmprod\phifn(\bmwdp*<->|\bm\la,\bm\rl)
	\bmprod\phifn(\bmwdp<+>-i|\bm\la,\bm\mu)
	\bmprod\phifn(\bmwdp*<->-i|\bm\la,\bm\mu)
	\\
	=
	\bmprod(\bmwdp-\bmclp+i)
	\bmprod(\bmwdp*-\bmclp-i)
	\bmprod\frac{\bmwdp-\bmwdp*+i}{\bmwdp-\bmwdp*}
	\bmprod\frac{\bmwdp*-\bmwdp-i}{\bmwdp*-\bmwdp}
	\\
	\times
	\frac{\bmprod(\bmwdp-\bm\hle-\frac{i}{2})}{\bmprod(\bmwdp-\bm\cid-i)}
	\frac{\bmprod(\bmwdp*-\bm\hle-\frac{i}{2})}{\bmprod(\bmwdp*-\bm\cid-i)}
\end{multline}
Substituting these two expressions back into \cref{pref_gen_step3} allows us to write the following \emph{reduced} determinant representation for the form-factors:
\begin{multline}
	\left|\FF^{z}\right|^2=
	(-1)^{\frac{n_h+2}{2}}
	M^{-n_h}
	2^{\frac{n_h(n_h-2)+2}{2}}	
	\pi^{\frac{n_h(n_h-3)+2}{2}}	
	\frac{\bmprod(\bm\cid-\bm\hle-\frac{i}{2})}{\bmprod(\bm\cid-\bm\cid-i)}
	\\
	\times
	\frac{1}{G^{2n_h}(\frac{1}{2})}
	\bmprod'
	\frac{%
	G(\frac{\bm\hle-\bm\hle}{2i})
	G(1+\frac{\bm\hle-\bm\hle}{2i})
	}{%
	G(\frac{1}{2}+\frac{\bm\hle-\bm\hle}{2i})
	G(\frac{3}{2}+\frac{\bm\hle-\bm\hle}{2i})
	}
	~
	\frac{%
	\det_{\ho{n}}\resmat*[g]
	\det_{n_h}\resmat*[e]
	}{\det\vmat[\bm\hle]}
	.
	\label{red_det_rep_generic}
\end{multline}
This will be the final result of our computations here.
We will discuss its merits and demerits thoroughly in the conclusions.
Let us only remark here that the strongest aspect of this result is that we compute the prefactors, including an infinite Cauchy matrix in the thermodynamic limit.
However, the weakness lies in the fact that we do not obtain a closed form expression [see \cref{cau_ex_I_clp_int_form,cau_ex_I_wdp_int_form,cau_ex_II_hvan_block_int_form}] for all of the terms involved in the matrices $\resmat*[g]$ and $\resmat*[e]$.
But the fact we find higher-level Gaudin matrix \eqref{hl_gau_ex_hl_ff_pref_chap} inside the matrix $\resmat*[e]$ is one of the strongest aspect of this result.
\par
Finally before we end this chapter let us consider some examples of the representations \eqref{red_det_rep_generic} for the form-factors in lower sectors.
The trivial yet important example is the case $n_h=2$ for which we already obtained an exact result in \cref{chap:2sp_ff} for the two-spinon form-factor. Let us compare the result \eqref{2sp_ff_result} with what we have obtained here.
We can very easily check that
\begin{enumerate}
\item The determinant of matrix $\resmat[g]$ is absent since $\ho{n}=0$, hence we substitute $\det\resmat[g]=1$.
\item The matrix $\resmat[e]$ is the Vandermonde matrix of order two $\resmat[e]=\vmat[\bm\hle+\frac{i}{2}]$ and hence its cancels out with the denominator. 
\item The prefactor for the two-spinon ($n_h=2$) case is $2M^{-2}$.
\end{enumerate}
Hence the results \eqref{red_det_rep_generic} is compatible with our previous result \eqref{2sp_ff_result} and this comparison is a helpful tool in determining the fidelity of our computations, particular in the prefactors.
We will now see a simplest non-trivial example that follows from our computations in \cref{chap:cau_det_rep_gen,chap:gen_FF}, namely the four-spinon form-factor.
\subsection{Example: Four-spinon case and little CV extraction}
\label{sub:4sp_ff_example}
The four-spinon triplet excitations are determined by the four hole parameters $\bm\hle$, $n_h=4$.
We saw in \cref{sub:DL_picture} that it consists of two complex roots which forms a 2-string $\set{\clp<+>,\clp<->}$. It is important to remark these roots cannot form any other configuration such as a quartet or wide-pair since the latter two require $\ho{n}\geq 4$ whereas for $n_h=2$, we have $\ho{n}=1$ from the relation \eqref{hle_num_ho_num_rel}.
The centre $\clp$ of the 2-string is a real parameter satisfying the higher-level Bethe equation \eqref{hl_bae} for the four-spinon excitation is simply given by
\begin{align}
	\frac{\bmprod(\clp-\bm\hle+\frac{i}{2})}{\bmprod(\clp-\bm\hle-\frac{i}{2})}
	=1
	\label{hl_bae_4sp}
	.
\end{align}
When simplified, it takes the form of a cubic polynomial \eqref{hl_bae_4sp_poly}:
\begin{align}
	4\clp^3
	- 3\clp^2
	\sum_a \hle_a	
	+\clp
	\left(2\sum_{a\neq b}\hle_a\hle_b-1\right)
	-
	\left(\sum_{a\neq b\neq c}\hle_a\hle_b\hle_c-\frac{1}{4}\sum_{a}\hle_a\right)
	=
	0
	.
	\label{hl_bae_4sp_poly_ff}
\end{align}
It admits three real solutions which tells us that there are three different positions available for the centre of a 2-string, once all the hole parameters are fixed.
We have discussed this in details in \cref{sub:hl_bae}.
\par
Since we have $\ho{n}=1$, the matrix $\resmat*[g]$ in \cref{red_det_rep_generic} becomes a singleton.
From \cref{cau_ex_I_clp_int_form} we see that it can be expressed as following:
\begin{multline}
	\Jcal_{g}
	=
	\det\resmat*[g]
	=
	-\phifn\big(\clp-\tfrac{3i}{2}\big|\bm\mu,\bm\la\big)
	\phifn'\big(\clp+\tfrac{i}{2}\big|\bm\la,\bm\mu\big)
	\Phifn'\big(\clp<+>\big|\bm{\check{\rl}^+},\bm\la\big)
	\\
	+
	2\Re\int_{\Rset+i\alpha}
	\Phifn\left(\tau\big|\bm{\check\rl^+},\bm\la\right)
	t(\tau-\clp)
	d\tau
	.
	\label{cau_ex_I_det_mat_4sp}
\end{multline}
Similarly the higher-level Gaudin matrix is also a singleton matrix since we have $\ho{n}=1$. 
The term $\aux*'(\clp)$ can be computed from the logarithmic derivative of the higher-level counting function \eqref{aux_hl_xxx}.
It gives us,
\begin{align}
  \aux*'(\clp)
  =
  -\log\aux*'(\clp)
  =
  2\pi i
  K(0)
  -
  2\pi i
  \bmsum
  K_2(\clp-\bm\hle)
  .
\end{align}  
Therefore we can write,
\begin{align}
	\Ho{\Ncal}
	=
	\aux*'(\clp) - 2\pi i K(0)
	=
	-
	2\pi i
	\bmsum
	K_2(\clp-\bm\hle)
	.
	\label{hl_gau_mat_4sp}
\end{align}
Let recall from \cref{den_ho_def} that the density term $\ho\rden$ for the higher-level roots is a rational function $K_2$ with the decomposition into the simple fractions:
\begin{align}
	\ho{\rden}(\la)
	=
	K_2(\la)
	=
	\frac{1}{2\pi i}
	\left\lbrace
	\frac{1}{\la-\frac{i}{2}}
	-
	\frac{1}{\la+\frac{i}{2}}
	\right\rbrace
	\label{rden_hl_decomp_ff_4sp}
\end{align}
and let us also recall that the matrix $\ho{\Rcal}$ is composed of the density terms $\ho{\rden}$ \eqref{den_mat_hl}.
This tells us that the higher-level Gaudin extraction $\ho{\Scal}=\ho{\Rcal}\Ho\Ncal^{-1}$ \eqref{cau_ex_II_hl_gauex} becomes extremely simple and its results can be written as
\begin{align}
	\ho{\Scal}_{a}
	=
	\frac{\ho{\rden}(\clp-\hle_a)}{\bmsum\ho{\rden}(\clp-\bm\hle)}
	.
	\label{hl_gau_ex_mat_4sp}
\end{align}
Simultaneously, the decomposition of the function $\ho{\rden}$ in \cref{rden_hl_decomp_ff_4sp} into the simple fractions also permits us to construct two new matrices $\Tcal^{+}$ and $\Tcal^{-}$ by developing on the first column $\ho{\Scal}$ in $\resmat*[e]$, such that
\begin{align}
	\det\resmat*[e]
	=
	\frac{%
	\det_4\Tcal^{+}
	-
	\det_4\Tcal^{-}
	}{
	\bmsum\ho{\rden}(\clp-\bm\hle)}
	.
	\label{cau_ex_II_mat_4sp_decomp}
\end{align}
We can see that matrices $\Tcal^+$ and $\Tcal^-$ that we constructed here are modified rational Cauchy-Vandermonde matrices.
A block of first three columns forms Cauchy-Vandermonde matrices while the last column contains $\Acal_e\Zcal^{\text{eff}}$ from \cref{cau_ex_gen_II_mat}, as it can be seen from the following two expressions:
\begin{subequations}
\begin{align}
	\Tcal^{+}
	&=
	\left(
	\cmat<\ptn\delta(2)>\left[\bm\hle+\tfrac{i}{2}\Big\Vert\set{\clp}\right]
	~\Big|~
	\Acal\Zcal^{\text{eff}}
	\right)
	,
\shortintertext{and}
	\Tcal^{-}
	&=
	\left(
	\cmat<\ptn\delta(2)>\left[\bm\hle+\tfrac{i}{2}\Big\Vert\set{\clp+i}\right]
	~\Big|~
	\Acal\Zcal^{\text{eff}}
	\right)
	.
\end{align}
\end{subequations}
Similarly, the prefactors can also be combined together in such a way that it forms a rational Cauchy-Vandermonde determinant.
There are two ways in which this can be arranged. Either we can combine the rational terms in \cref{red_det_rep_generic} directly to obtain the determinants:
\begin{subequations}
\begin{align}
	\frac{\bmprod(\clp-\bm\hle-\frac{i}{2})}{(-i)	\bmalt(\bm\hle)}
	&=
	\frac{i}{\bmalt(\set{\clp}\Vert\bm\hle+\frac{i}{2})}
	=
	\frac{i}{\cmat<\ptn\delta>\left[\bm\hle+\tfrac{i}{2}\Big\Vert\set{\clp}\right]}
	.
	\label{pref_rat_cau_van_det_4sp_type-I}
\shortintertext{Or we can first use the higher-level Bethe equations \eqref{hl_bae_4sp} and thus obtain the following determinant:}
	\frac{\bmprod(\clp-\bm\hle+\frac{i}{2})}{(-i)\bmalt(\bm\hle)}
	&=
	\frac{i}{\bmalt(\set{\clp+i}\Vert\bm\hle+\frac{i}{2})}
	=
	\frac{i}{\cmat<\ptn\delta>\left[\bm\hle+\tfrac{i}{2}\Big\Vert\set{\clp+i}\right]}
	\label{pref_rat_cau_van_det_4sp_type_II}
	.
\end{align}
\label{pref_rat_cau_van_det_4sp_both}
\end{subequations}
The higher-level Bethe equations \eqref{hl_bae_4sp} tells us that these two expressions are equal.
Let us now substitute \cref{cau_ex_I_det_mat_4sp} and \crefrange{cau_ex_II_mat_4sp_decomp}{pref_rat_cau_van_det_4sp_both} into the representation \eqref{red_det_rep_generic}. It allows to write
\begin{multline}
	\left|
	\FF^z
	\right|^2
	=
	-
	\frac{32\pi^3}{M^4}
	\frac{1}{G^8(\frac{1}{2})}
	\bmalt^2
	\frac{%
	G(\frac{\bm\hle}{2i})
	G(1+\frac{\bm\hle}{2i})
	}{%
	G(\frac{1}{2}+\frac{\bm\hle}{2i})
	G(\frac{3}{2}+\frac{\bm\hle}{2i})
	}
	\\
	\times
	\frac{\Jcal_g}{\bmsum\ho{\rden}(\clp-\bm\hle)}
	\left\lbrace
	\frac{\det\Tcal^{+}}{\det\cmat<\ptn\delta>\left[\set{\clp}\Vert\bm\hle+\frac{i}{2}\right]}
	-
	\frac{\det\Tcal^{-}}{\det\cmat<\ptn\delta>\left[\set{\clp+i}\Vert\bm\hle+\frac{i}{2}\right]}
	\right\rbrace
	.
	\label{4sp_ff_rat_det_exprn}
\end{multline}
Let us now extract the rational Cauchy-Vandermonde matrices in the denominator.
It is important to note that we must take the action of the inverse Cauchy-Vandermonde matrix and not its dual, since the matrices $\Tcal^\pm$ already contains the Vandermonde columns.
\begin{subequations}
\label{4sp_ff_little-cv-extn}
\begin{align}
	\Jcal_{e}^{+}
	&=
	\cmat<\ptn\delta>^{-1}\left[\set{\clp}\Big\Vert\bm\hle+\tfrac{i}{2}\right]\cdot
	\Tcal^+
	,
	\shortintertext{and}
	\Jcal_{e}^{-}
	&=
	\cmat<\ptn\delta>^{-1}\left[\set{\clp+i}\Big\Vert\bm\hle+\tfrac{i}{2}\right]\cdot
	\Tcal^-
	.
\end{align}
\end{subequations}
The inverse matrices in the above expressions \eqref{4sp_ff_little-cv-extn} are given by the duality found in \cref{cau_van_inv_diag_dress,dual_rat_cau_van_mat}, it means that, in a general setting we would have to deal with extraction sums involving supersymmetric polynomials $e_a(\bm x\Vert\bm y)$ [see \cref{defn:ele_susy_append}].
However, in the four-spinon case, this issue do not arise, since we can easily check that there is an identity block of order three in the resultant matrix $\Jcal_e$.
Hence for the determinant, we only need to compute the non-trivial diagonal element, which is obtained by the extraction with the constant supersymmetric polynomial $e_0$.
Thus we can write that
\begin{align}
	\det\Jcal_e^{+}
	=
	\Jcal^{+}_{e;44}
	&=
	\sum_{a=1}^{4}
	\frac{\hle_a-\clp+\frac{i}{2}}{\bmprod'(\hle_a-\bm\hle)}
	\Acal\Zcal^{\text{eff}}[\hle_a]
	,
	\shortintertext{and}
	\det\Jcal_e^{-}
	=
	\Jcal^{-}_{e;44}
	&=
	\sum_{a=1}^{4}
	\frac{\hle_a-\clp-\frac{i}{2}}{\bmprod'(\hle_a-\bm\hle)}
	\Acal\Zcal^{\text{eff}}[\hle_a]
	.
\end{align}
Let us now observe that the difference $\Jcal_e=\det\Jcal^{+}_e-\det\Jcal^{-}_e$ is given by the expression
\begin{align}
	\Jcal_{e}
	&=
	\sum_{a=1}^{4}
	\frac{i}{\bmprod'(\hle_a-\bm\hle)}
	\Acal\Zcal^{\text{eff}}[\hle_a]	
	.
	\label{4sp_ff_little_cv-extn_result}
\end{align}
Let us also recall that the effective $\Zcal$ matrix can be written as
\begin{multline}
	\Acal\Zcal^{\text{eff}}[\hle_a]
	=
	\aux'_e(\hle_a)
	\Phifn'\left(\hle_a\big|\bm{\check\la},\bm{\rh^+}\right)
	-
	2\pi i
	\bmsum
	\rden_h(\hle_a-\bm\hle)
	\Phifn'\left(\bm\hle\big|\bm{\check\la},\bm{\rh^+}\right)
	\\
	+
	\pi
	\int_{\Rset+i\alpha}
	\ho{\rden}(\hle_a-\tau)
	\Phifn\left(\hle_a\big|\bm{\check\la},\bm{\rh^+}\right)
	d\tau
	.
	\label{4sp_hvan_col}
\end{multline}
Finally, we substitute \crefrange{4sp_ff_little-cv-extn}{4sp_ff_little_cv-extn_result} in \cref{4sp_ff_rat_det_exprn} to obtain the following representation for the four-spinon form-factor:
\begin{align}
 	\left|\FF^z\right|^2
 	=
	-
	\frac{32\pi^3}{M^{4}}
	\frac{1}{G^8(\frac{1}{2})}
	\bmprod'
	\frac{%
	G(\frac{\bm\hle-\bm\hle}{2i})
	G(1+\frac{\bm\hle-\bm\hle}{2i})
	}{%
	G(\frac{1}{2}+\frac{\bm\hle-\bm\hle}{2i})
	G(\frac{3}{2}+\frac{\bm\hle-\bm\hle}{2i})
	}
	\frac{%
	\Jcal_g
	\Jcal_e
	}{%
	\bmsum \ho{\rden}(\clp-\bm\hle)
	}
	.
	\label{4sp_ff_result}
\end{align}
Note that all the prefactors as well as the term $\Jcal_e$ obtained in \cref{4sp_ff_little_cv-extn_result} is independent of the choice of the string centre $\clp$.
The term $\Jcal_g$ \eqref{cau_ex_I_det_mat_4sp} and the denominator $\bmsum\ho{\rden}(\clp-\bm\hle)$ depends implicitly on the string centre $\clp$ which is determined as one of the three roots of the cubic polynomial \eqref{hl_bae_4sp_poly_ff}.
\begin{subappendices}
\section{Toy example : Extraction of the Cauchy-Vandermonde matrix}
\label{sec:toy_cv_extn_append}
\subsection[In rational parametrisation]{CV extraction in rational parametrisation}
\label{sub:toy_cv_extn_append}
We will now study a simpler example of the Cauchy-Vandermonde extraction \eqref{toy_extn_example_cv_extn_gen}.
Here we want to extract the common Cauchy matrix $\cmat[\bm x\Vert\bm y]$ from the sum:
\begin{align}
	C[\bm x\Vert\bm{y^{\textrm{L}}}]+ C[\bm x\Vert\bm{y^{\textrm{R}}\cdot R}]
	\label{toy_cv_extn_rat_append_sum_cau_mats}
\end{align}
The sets $\bm x$ and $\bm y$ have cardinalities $n_{\bm x}=m$ and $n_{\bm y}=m+n$. The latter is partitioned into two subsets $\bm{x}=\bm{x^{\textrm{L}}}\bm\cup\bm{x^{\textrm{R}}}$ such that their cardinalities are $n_{\bm{y^{\textrm{L}}}}=m$ and $n_{\bm{y^{\textrm{R}}}}=n$.
The extraction is taken on the larger matrix where \cref{toy_cv_extn_rat_append_sum_cau_mats} is embedded diagonally as
\begin{align}
	P=
	\cmat*<\ptn\delta>^{-1}[\bm y\Vert\bm x]
	\cdot
	\begin{pmatrix}
	\cmat[\bm x\Vert\bm{y^{\textrm{L}}}]+ \cmat[\bm x\Vert\bm{y^{\textrm{R}}}]\cdot R
	& 0
	\\
	0 & \Id_{n}
	\end{pmatrix}
	.
	\label{toy_cv_extn_rat_extn_mat_append}
\end{align}
Here we compute the action in this toy example in two different manners which leads to the same final conclusion.
\begin{enumerate}[wide=0pt, label={\textbf{Method \Roman*}:},ref={method \Roman*}]
\item 
\label{rat_cv_extn_toy_method-1}
We see that from the inversion $\cmat*<\ptn\delta>^{-1}[\bm y\Vert\bm x]\cmat*<\ptn\delta>[\bm x\Vert\bm y]=\Id$ and the diagonal dressing \eqref{inv_rat_van_dual} we get
\begin{align}
	\delta_{j,k}=
	\phifn'(y_{j}|\bm{x},\bm{y})
	\left\lbrace
	\sum_{a=1}^{m}
	\phifn'(x_{a}|\bm{y},\bm{x})
	\frac{1}{y_{j}-x_{a}}
	\frac{1}{x_{a}-y_{k}}
	+
	\sum_{a=1}^{n}
	(-1)^{n-a}
	y_{j}^{a-1}
	e_{n-a}(\bm{{y}_{\hat{k}}}\Vert\bm{x})
	\right\rbrace
	.
	\label{toy_cv_extn_rat_inv_id_append}
\end{align}
Hence the partial sums in the extraction on the Cauchy-block in \cref{toy_cv_extn_rat_extn_mat_append} can be expressed as
\begin{multline}
	P_{j,k}=\delta_{j,k}
	+R_{r,k}\delta_{j,m+r}
	\\
	-
	\phifn'(y_j|\bm x, \bm y)
	\sum_{a=1}^{n}
	(-1)^{n-a}
	y_{j}^{a-1}
	e_{n-a}(\bm{{y}_{\hat{k}}}\Vert\bm{x})
	\\	
	-
	\phifn'(y_j|\bm x, \bm y)
	\sum_{a=1}^{n}
	\sum_{b=1}^{n}
	(-1)^{n-a}
	y_{j}^{a-1}	
	e_{n-a}(\bm{y}_{\what{m+b}}\Vert\bm{x})
	R_{b,k}
	.
	\label{ratn_cau_extn_cau_block}
\end{multline}
Whereas the action on identity block simply gives back the supersymmetric Vandermonde block in the inverse matrix which we shall denote
\begin{align}
	Z_{a}[y_j]=
	P_{j,m+a}=
	(-1)^{a-1}
	\phifn'(y_{j}|\bm y,\bm x)
	y_{j}^{a-1}
	.
	\label{toy_cv_extn_rat_susy_van_block}
\end{align}
We see that the extra terms in the form of summations in \cref{ratn_cau_extn_cau_block} are linear combinations of the columns from $Z$.
\begin{align}
	P_{j,k}
	=
	\delta_{j,k}
	+
	R_{r,k}\delta_{j,m+r}
	-
	\sum_{a=1}^{n}
	\chi_a(y_k)
	Z_{a}[y_j]
	-
	\sum_{a=1}^{n}
	\left(
	\sum_{b=1}^{n}
	\chi_{a}(y_{m+b})
	R_{b,k}
	\right)
	Z_{a}[y_j]
	.
	\label{toy_rat_cv_extn_actn_cau_append}
\end{align}
These linear sums over the columns from block $Z$ can be silently cancelled without affecting its determinant to obtain
\begin{align}
	P
	=
	\begin{pmatrix}
		\Id &	Z^{\textrm{N}}
		\\
		R^T &	Z^{\textrm{S}}
	\end{pmatrix}
	.
	\label{ratn_cv_extn_residual_matrix_blocks_toy_append}
\end{align}
Finally using \cref{lem:mat_det_red} we can construct a smaller matrix $Q$ as follows:
\begin{align}
	Q=Z^{\textrm{S}}-R^T Z^{\textrm{N}}
	\label{ratn_toy_cv_extn_reduc_mat_append}
\end{align}
such that it is equivalent up-to the original matrix $P$ up-to the evaluation of its determinant:
\begin{align}
	\det_{m+n}P = \det_{n}Q
	.
	\label{ratn_toy_cv_extn_reduction_append}
\end{align}
\item
\label{append_cv_ex_toy_method-2_larger_cau}
Let us form a larger square matrix $\cmat[\bm z\Vert\bm y]$ where we add the extra variables $\bm z=\bm x\cup\bm w$ such that the cardinality of the extended set $\bm z$ is at par with $\bm y$ i.e. $n_{\bm z}=n_{\bm y}=m+n$.
Thus the number of added variables is $n_{\bm w}=n$.
This larger matrix is extracted to form
\begin{align}
	P(\bm w)=
	\cmat^{-1}[\bm y\Vert\bm z]
	\cdot
	\begin{pmatrix}
	\cmat[\bm x\Vert\bm{y^{\textrm{L}}}]+ \cmat[\bm x\Vert\bm{y^{\textrm{R}}}]\cdot R
	& 0
	\\
	0 & \Id_{n}
	\end{pmatrix}
	\label{toy_cv_ex_rat_larger_cau_mat_append}
\end{align}
The action on the Cauchy block gives us the partial sum.
These can be written as
\begin{multline}
	P_{j,k}(\bm w)
	=
	\delta_{j,k}
	+
	\delta_{j,m+a}R_{a,k}
	\\	
	-
	\phifn'(y_j|\bm z,\bm y)
	\sum_{a=1}^{n}
	\phifn'(w_a|\bm y,\bm z)
	\frac{1}{y_j-w_a}
	\frac{1}{w_a-y_k}
	\\
	-
	\phifn'(y_j|\bm z,\bm y)
	\sum_{a=1}^{n}
	\sum_{b=1}^{n}
	\phifn'(w_a|\bm y,\bm z)
	\frac{1}{y_j-w_a}
	\frac{1}{w_a-y_{m+b}}
	R_{b,k}
	.
	\label{toy_rat_cv_extn_large_cau_append_sum_extra_terms}
\end{multline}
Whereas the action on the identity block gives back the columns corresponding to $\bm w$ which are denoted with $Z$:
\begin{align}
	P_{j,m+a}
	=
	Z[y_j\Vert w_a]
	=
	\phifn'(y_j|\bm z,\bm y)
	\phifn'(w_a|\bm y,\bm w)
	\frac{1}{y_j-w_a}
	.
	\label{rat_toy_extn_larger-cau_Z-bl}
\end{align}
Similar to \cref{rat_cv_extn_toy_method-1}, we find that the summation in \cref{toy_rat_cv_extn_large_cau_append_sum_extra_terms} are taken over the extra columns $Z$.
\begin{align}
	P_{j,k}(\bm w)
	=
	\delta_{j,k}
	+
	\delta_{j,m+a}	R_{a,k}
	-
	\sum_{a=1}^{n}
	\frac{1}{w_a-y_k}
	Z[y_j\Vert w_a]
	-
	\sum_{a,b=1}^n
	\frac{R_{b,k}}{w_a-y_{m+b}}
	Z[y_j\Vert w_a]
	.
	\label{rat_toy_extn_larger-cau_cau-bl}
\end{align}
Thus we can see that sum over the columns $Z$ can be cancelled.
Finally we take a series of limits where the extra variables $\bm w$ are send to infinity
\begin{align}
	P(\bm w)\to P(\bm w^{(1)}) \to \cdots P(\bm w^{(k)}) \to \cdots \to P.
	\label{toy_rat_cv_extn_seq_lim_mat}
\end{align}
We note that such as a procedure was also used in \cref{chap:mat_det_extn} to write an alternative proof of \cref{lem:rat_cau_van_det} for the determinant of the Cauchy-Vandermonde matrix.
We define the set $\bm{w^{(k)}}$ as $\set{w_{k+1},\ldots,w_{n}}$.
Here we need to see the effect of these limits on the block $Z$.
But an important distinction here is that the extra variables $w_1^{-1}$ are contained inside $\phifn$ function, since we have
\begin{align}
	Z[y_j\Vert w_a]	
	=
	\phifn'(y_j|\bm z\setminus \set{w_a},\bm y)
	\phifn'(w_a|\bm y,\bm z)
	.
\end{align} 
As a result here the limits are taken in different manner compared to what we do in the \cref{lem:rat_cau_van_det}.
At the first iteration, we multiply the first column with $w_1^{-1}$ before the limit $w_1\to\infty$ is taken on every column:
\begin{subequations}
\begin{align}
	&
	&Z^{(1)}_1[y_j]&=\lim_{w_1\to\infty}w_1^{-1} Z[y_j\Vert w_1]=\phifn'(y_j|\bm{z^{(1)}},\bm y)
	\\
	&(1<a\leq n)
	&Z^{(1)}_a[y_j]&=\lim_{w_1\to\infty}Z[y_j\Vert w_a ]
	=
	\phifn'(y_j|\bm{z^{(1)}}\setminus\set{w_a},\bm y)
	\phifn'(w_a|\bm y,\bm{z^{(1)}})
	.
\end{align}
\label{toy_rat_cv_extn_seq_lim_mat_itn_1}
\end{subequations}
For the successive iterations we can write
\begin{subequations}
\begin{align}
	&
	&Z^{(k)}_k[y_j]&=\lim_{w_k\to\infty}w_k^{-k} Z^{(k-1)}_{k}[y_j\Vert w_k]
	=
	\phifn'(y_j|\bm{z^{(k)}},\bm y)
	\\
	&(1\leq a<k)
	&Z^{(k)}_a[y_j]&=\lim_{w_k\to\infty} Z^{(\kappa-1)}_{a}[y_j]+w_k Z^{(k)}_{a+1}[y_j]
	=
	y_j^{k-a}\phifn'(y_j|\bm{z^{(k)}},\bm y)
	\\
	&(k<a\leq n)
	&Z^{(k)}_a[y_j]
	&=
	\lim_{w_k\to\infty} Z^{(k)}_{a}[y_j]
	=
	\phifn'(y_j|\bm{z^{(k-1)}}\setminus \set{w_a},\bm y)
	\phifn'(w_a|\bm y,\bm{z^{(k)}})
	.
\end{align}
\label{toy_rat_cv_extn_seq_lim_mat_itn_k}
\end{subequations}
At the end of this procedure, we end up with the matrix $P$ with the same form as \cref{ratn_cv_extn_residual_matrix_blocks_toy_append}, seen in the following:
\begin{align}
	P
	=
	\begin{pmatrix}
		\Id_n 	&		Z^{\textrm{N}}
		\\
		R^T 		&		Z^{\textrm{S}}
	\end{pmatrix}
	.
\end{align}
\end{enumerate}
The equivalence between the two methods presented above is not a plain coincidence, it follows from the fact that the inverse of a Cauchy-Vandermonde matrix can be constructed using the procedure of taking limit similar to \crefrange{toy_rat_cv_extn_seq_lim_mat}{toy_rat_cv_extn_seq_lim_mat_itn_k}.
It is also important to remark that this limiting procedure always gives us the inverse of the dual matrix \eqref{dual_cau_van_inv_diag_dress}.
We now generalise this example to the hyperbolic case through the re-parametrisation \eqref{reparam_cau_van_hyper-rat}.
\subsection[In hyperbolic parametrisation]{Toy example in hyperbolic parametrisation}
\label{sub:toy_cv_extn_hyper_append}
Let us now consider the hyperbolic version of the example \eqref{toy_cv_extn_rat_extn_mat_append}:
\begin{align}
	\Pcal=
	\Cmat*<\ptn\delta>^{-1}[\bm \beta\Vert\bm \alpha]
	\cdot
	\begin{pmatrix}
	\Cmat[\bm \alpha\Vert\bm{\beta^{\textrm{L}}}]+ \Cmat[\bm \alpha\Vert\bm{\beta^{\textrm{R}}}]\cdot R
	& 0
	\\
	0 & \Id_{n}
	\end{pmatrix}
	.
	\label{toy_hcv_extn_mat_append-sec}
\end{align}
The cardinalities of the sets $\bm \alpha$ and $\bm\beta$ are $n_{\bm\alpha}=m$ and $n_{\bm\beta}=m+n$. The latter set $\bm\beta$ is partitioned into two subsets $\bm\beta=\bm{\beta^{\textrm{L}}}\cup\bm{\beta^{\textrm{R}}}$ such that $n_{\bm\beta^{\textrm{R}}}=n$.
Note that in this example we choose to work with the $\Cmat<\ptn\delta>$ [see \cref{hcv_traditional_mat_comps}] instead of $\Cmat<\ptn\gamma>$ [see the \cref{defn:hyper_cau_van_mat} or \eqref{hyper_cau_van_mat_append}].
These two are related to each other through a simple recombination of rows.
While the matrix $\Cmat<\ptn\delta>$ consists of the exponentials in the Vandermonde block, we can recover the matrix $\Cmat<\ptn\gamma>$ from it by recombining the exponential terms to form the hyperbolic Vandermonde block $\Xcal$, which is expressed in terms of the $\sinh$ and $\cosh$ functions.
This recombination is discussed in \cref{sec:cau_van_mat_hyper}.
It does not alter substantially the process of extraction and the its only visible effect can be seen in difference of normalisation constants in their determinant [see \cref{hcv_mat_traditional_det,hcv_mat_det_append}]:
\begin{align}
	\det\Cmat<\ptn\delta>=2^{-\frac{n(n-1)}{2}}\det\Cmat<\ptn\gamma>
	.
\end{align}
\par
Since the hyperbolic and the rational versions Cauchy-Vandermonde are related through the parametrisation \eqref{reparam_cau_van_hyper-rat} shown below:
\begin{alignat}{3}
	x_{j}&=e^{2\pi\alpha_{j}}
	,
	&
	\qquad
	&
	y_{j}&=e^{2\pi\beta_{j}}
	.
	\label{reparam_cau_van_hyper-rat_append-sec}
\end{alignat}
Hence, it would be sufficient to see how the results from both methods in the above example compare with the hyperbolic case through this re-parametrisation.
\begin{enumerate}[wide=0pt, label={\textbf{Method \Roman*:}},ref={method \Roman*}]
\item
Under the re-parametrisation  \eqref{reparam_cau_van_hyper-rat_append-sec} the hyperbolic Cauchy matrix can be expressed according to \cref{hyp_cau_mat_rat_dressing_append-chap} as
\begin{align}
	\Cmat[\bm\alpha(\bm x)\Vert\bm\beta(\bm x)]=
	\diag\left[2\bm{x}^{\frac{1}{2}}\right]
	\cmat[\bm{x}\Vert\bm{y}]
	\diag\left[\bm{y}^{\frac{1}{2}}\right]
	.
	\label{hyp_cau_mat_rat_dressing_append-sec}
\end{align}
Similarly the hyperbolic Cauchy-Vandermonde matrix in the rational parametrisation \eqref{hcv_mat_traditional_dressing_rat_param} can be expressed as%
\footnote{let us recall that here $n_{\bm x}=n_{\bm \alpha}<n_{\bm y}=n_{\bm\beta}$ unlike the example \eqref{hcv_mat_traditional_dressing_rat_param} in \cref{chap:mat_det_extn}.}
\begin{align}
	\Cmat<\ptn\delta>[\bm\alpha(\bm x)\Vert\bm\beta(\bm y)]
	=
	\diag\left[
	2\bm x^{\frac{n+1}{2}}
	~\Big|~
	\Id_n
	\right]
	\cmat<\ptn\delta>[\bm x\Vert\bm y]
	\diag\left[
	\bm y^{-\frac{n-1}{2}}
	\right]
	.
	\label{hcv_mat_traditional_dressing_rat_param_append-sec}
\end{align}
The inverse of a dual Cauchy-Vandermonde matrix can be expressed with the diagonal dressing of the $\Phifn$ functions as
\begin{align}
	\Cmat*<\ptn\delta>[\bm\beta\Vert\bm\alpha]
	=
	\diag_{\bm\beta}\Big[
	\Phifn'(\bm\beta|\bm\alpha,\bm\beta)
	\Big]
	\cdot
	\left(\Cmat<\ptn\delta>[-\bm\alpha\Vert-\bm\beta]\right)^T
	\cdot
	\diag_{\bm\alpha}\Big[
	\Phifn'(\bm\alpha|\bm\beta,\bm\alpha)
	~\Big|~
	\Id_n
	\Big]
	.
	\label{hcv_dual_inv_dressing_append-sec}
\end{align}
Under the re-parametrisation \eqref{reparam_cau_van_hyper-rat_append-sec} to the rational variables the $\Phifn$ transforms into its rational variant $\phifn$ as
\begin{subequations}
\label{Phifn_rat_param_trans_append_both}
\begin{align}
	\Phifn'(\alpha(x_j)|\bm\beta(\bm y),\bm\alpha(\bm x))
	&=
	2^{-n-1}
	x_{j}^{-\frac{n+2}{2}}
	\left(\bmprod x \bmprod y^{-1}\right)^{\frac{1}{2}}
	\phifn'(x_j|\bm y,\bm x)
	\label{Phifn_rat_param_trans_append_lag}
\shortintertext{and}
	\Phifn'(\beta(y_j)|\bm\alpha(\bm x),\bm\beta(\bm y))
	&=
	2^{n-1}
	y_j^{\frac{n-2}{2}}
	\left(\bmprod x^{-1} \bmprod y\right)^{\frac{1}{2}}
	\phifn'(y_j|\bm y,\bm x)
	.
	\label{Phifn_rat_param_trans_append_lead}
\end{align}
\end{subequations}
Substituting \cref{hcv_mat_traditional_dressing_rat_param_append-sec,Phifn_rat_param_trans_append_both} into \cref{hcv_dual_inv_dressing_append-sec} tells us that the inverse of the dual Cauchy-Vandermonde matrix in the rational parametrisation can be expressed as
\begin{subequations}
\label{toy_hcv_extn_inv_mat_rat_param_append_all}
\begin{multline}
	\Cmat*<\ptn\delta>^{-1}[\bm\alpha(\bm x)\Vert\bm\beta(\bm y)]
	\\
	=
	\diag_{\bm x}\Big[
	\bm{y^{-\frac{1}{2}}}\phifn'(\bm y|\bm x,\bm y)
	\Big]
	\cdot
	\left(\cmat<\ptn\delta>(-\bm x\Vert-\bm y)\right)^T
	\cdot
	\diag\Big[
	2^{-1}
	\bm{x^{-\frac{1}{2}}}\phifn'(\bm x|\bm y,\bm x)
	~\Big|~
	D_n
	\Big]
	\label{toy_hcv_extn_inv_mat_rat_param_append}
\end{multline}
where $\Dcal_n$ is a diagonal matrix:
\begin{align}
	D_n
	=
	2^{n-1}
	\left(\bmprod x^{-1} \bmprod y\right)^{\frac{1}{2}}
	\Id_n
	.
	\label{hcv_toy_inv_mat_rat_param_diag_block}
\end{align}
\end{subequations}
When we take action \eqref{toy_hcv_extn_mat_append-sec} is taken on the Cauchy matrix, we can see from \cref{toy_hcv_extn_inv_mat_rat_param_append,hyp_cau_mat_rat_dressing_append-sec} that when it is expressed in the rational parametrisation:
\begin{multline}
	\Cmat*<\ptn\delta>^{-1}[\bm\beta\Vert\bm\alpha]
	\begin{pmatrix}
	\Cmat[\bm\alpha\Vert\bm{\beta^{\textrm{L}/\textrm{R}}}]	&	0
	\\
	0	&	\Id_n
	\end{pmatrix}
	=
	\diag\Big[
	\bm y^{-\frac{1}{2}}
	\Big]
	\cdot
	\\
	\left\lbrace
	\diag_{\bm y}\Big[
	\phifn'(\bm y|\bm x,\bm y)
	\Big]
	\cdot
	\left(\cmat<\ptn\delta>(-\bm x\Vert-\bm y)\right)^T
	\cdot
	\diag_{\bm x}\Big[
	\phifn'(\bm x|\bm y,\bm x)
	~\Big|~
	D_n
	\Big]
	\cdot
	\begin{pmatrix}
	\cmat[\bm x\Vert\bm{y^{{\textrm{L}/\textrm{R}}}}]
	&	0
	\\
	0	& \Id_n
	\end{pmatrix}
	\right\rbrace
	\quad (\ast)
	\\
	\cdot
	\diag\Big[
	{(\bm y^{\textrm{L}/\textrm{R}}})^{\frac{1}{2}}
	~\Big|~
	\Id_n
	\Big]
	.
	\label{toy_hcv_extn_conv_to_rat_append}
\end{multline}
It is important to note that the system on the line marked $(\ast)$ in the above expression \eqref{toy_hcv_extn_conv_to_rat_append} is identical to the rational case that we saw earlier.
In particular it means that we can carry out the intermediate steps from \crefrange{toy_cv_extn_rat_inv_id_append}{toy_rat_cv_extn_actn_cau_append} in the rational parametrisation.
When reconverted to the hyperbolic parametrisation we get from \cref{toy_rat_cv_extn_actn_cau_append} the following expression:
\begin{align}
	\Pcal_{j,k}
	=
	\delta_{j,k}
	+
	\Rcal_{r,k}
	\delta_{j,m+r}
	-
	\sum_{a=1}^{n}
	\chi_a(\beta_k)
	\Zcal_{\ptn\delta;a}(\beta_j)
	-
	\sum_{a=1}^{n}
	\left(
	\sum_{b=1}^{n}
	\chi_a(\beta_{m+b})
	\Rcal_{b,k}
	\right)
	\Zcal_{\ptn\delta;a}(\beta_j)
	.
\end{align}
The exact form of the terms $\chi_a$ is unimportant since they are cancelled in the determinant.
It also means that results of \crefrange{ratn_cv_extn_residual_matrix_blocks_toy_append}{ratn_cv_extn_residual_matrix_blocks_toy_append} can be extended to the hyperbolic parametrisation as
\begin{align}
 	\Pcal
 	=
 	\begin{pmatrix}
 	\Id_m 	& 		\Zcal_{\ptn\delta}^{\textrm{N}}
 	\\[\jot]
 	\Rcal^T	&	\Zcal_{\ptn\delta}^{\textrm{S}}
 	\end{pmatrix}
\end{align}
and thus $\det\Pcal=\det\Qcal$ where $\Qcal=\Zcal_{\ptn\delta}^{\textrm{S}}-\Rcal^T\Zcal_{\ptn\delta}^{\textrm{N}}$.
Finally we can also easily check that the $\Zcal_{\ptn\delta}$ obtained through the reconversion of the product as
\begin{align}
	\Zcal_{\ptn\delta}=D_n\cdot Z[\bm y]
\end{align}
through the re-parametrisation \eqref{reparam_cau_van_hyper-rat_append-sec} leads to the hyperbolic Vandermonde of the $\Cmat*<\ptn\delta>^{-1}$ [see \cref{hcv_traditional_mat_comps}]:
\begin{align}
	\Zcal_{\ptn{\delta};j,a}
	=
	e^{\pi\beta_j(2a-n-1)}
	\Phifn'(\beta_j|\bm\alpha,\bm\beta)
	.
\end{align}
It is related to the Vandermonde block $\Vdr<\ptn\gamma>$ that we saw in \cref{hyper_cau_van_inv_hvan_block} in its version $\Cmat<\ptn\gamma>$ by simple recombination \eqref{hcv_append_hvan_recomb}.
\item 
Similar to \cref{toy_cv_ex_rat_larger_cau_mat_append} in \cref{append_cv_ex_toy_method-2_larger_cau} used in the rational case, we start with the extraction of larger Cauchy matrix:
\begin{align}
	\Pcal(\bm \eta)
	=
	\Cmat^{-1}[\bm\beta\Vert\bm\zeta]
	\cdot
	\begin{pmatrix}
	\Cmat[\bm\alpha\Vert\bm{\beta^{\textrm{L}}}]
	+
	\Cmat[\bm\alpha\Vert\bm{\beta^{\textrm{L}}}]
	\Rcal
	&	0
	\\
	0	& \Id_n
	\end{pmatrix}
\end{align}
where $\bm\zeta=\bm\alpha\cup\bm\eta$ and $n_{\bm\eta}=n$.
It is not hard to see that similar to \cref{rat_toy_extn_larger-cau_Z-bl,rat_toy_extn_larger-cau_cau-bl} were we have,
\begin{subequations}
\begin{multline}
	\Pcal_{j,k}(\bm\eta)
	=
	\delta_{j,k}
	+
	\delta_{j,m+a}	\Rcal_{a,k}
	-
	\sum_{a=1}^{n}
	\frac{1}{\sinh\pi(\eta_a-\beta_k)}
	\Zcal[\beta_j\Vert\eta_a]
	\\
	-
	\sum_{a,b=1}^n
	\frac{\Rcal_{b,k}}{\sinh\pi(\eta_a-\beta_{m+b})}
	\Zcal[\beta_j\Vert\eta_a]
	.
	\label{hyper_toy_extn_larger-cau_cau-bl}
\end{multline}
Where,
\begin{align}
	\Zcal[\beta_j\Vert\eta_a]
	=
	\Pcal_{j,m+a}(\bm\eta)
	=
	\Phifn'(\beta_j|\bm\zeta,\bm\beta)
	\Phifn'(\eta_a|\bm\beta,\bm\zeta)
	\frac{1}{\sinh\pi(\eta_a-\beta_j)}
	.
\end{align}
\end{subequations}
First of all we can see that the linear sum over $\Zcal$ can be cancelled in the \eqref{hyper_toy_extn_larger-cau_cau-bl} to write
\begin{align}
	\Pcal(\bm\eta)
	=
	\begin{pmatrix}
	\Id_m 	&		\Zcal^{(0)}[\bm\beta^{\textrm{L}}|\bm\eta]
	\\
	\Rcal^T 	&	\Zcal^{(0)}[\bm\beta^{\textrm{R}}|\bm\eta]
	\end{pmatrix}
	.
\end{align}
We now see that the variables $\bm\eta$ are only contained in the $\Zcal$ part.
We remove these extra variable through the sequence of limit
\begin{align}
	\Zcal^{(0)} \to
	\Zcal^{(1)} \to
	\cdots 	\to
	\Zcal^{(n)}=\Zcal_{\ptn\delta}.
\end{align}
Let us remark that initially we had
\begin{align}
	\Zcal^{(0)}_{j,a}
	=
	\Phifn'(\beta_j|\bm\zeta\setminus\set{\eta_a},\bm\beta)
	\Phifn'(\eta_a|\bm\beta,\bm\zeta)
\end{align}
Let us compute the first iteration $\eta_1\to\infty$.
Since we can see that
\begin{align}
	\Phifn'(\eta_1|\bm\beta,\bm\zeta)
	\sim_{\eta_1\to\infty}
	e^{\pi\eta_1}
	\bmprod_{\bm\beta}e^{-\pi\bm\beta}
	\bmprod_{\bm\zeta}e^{\pi\bm\zeta^{(1)}}
\end{align}
we will find that
\begin{subequations}
\begin{flalign}
	\Zcal^{(1)}_{j,1}&=
	\lim_{\eta_1\to\infty}
	e^{-\pi\eta_1}
	\Zcal^{(0)}_{j,a}
	=
	\left(
	\bmprod_{\bm\beta}e^{-\pi\bm\beta}
	\right)
	\left(
	\bmprod_{\bm\zeta}e^{\pi\bm\zeta^{(1)}}
	\right)
	\Phifn'(\beta_j|\bm\zeta^{(1)},\bm\beta)
	,
	\\
	\Zcal^{(1)}_{j,a}&=
	\lim_{\eta_1\to\infty}
	\Zcal^{(0)}_{j,a}
	=
	e^{\pi(\eta_a-\beta_j)}
	\Phifn'(\beta_j|\bm\zeta^{(1)}\setminus\set{\eta_a},\bm\beta)
	\Phifn'(\eta_a|\bm\beta,\bm\zeta^{(1)})
	.
\end{flalign}
\end{subequations}
In the subsequent iterations where we send $\eta_k\to\infty$, we take the limit in the following manner
\begin{subequations}
\begin{multline}
	\Zcal^{(k)}_{j,k}=
	\lim_{\eta_k\to\infty}
	e^{-\pi(2k-1)\eta_k}
	\Zcal^{(k-1)}
	\\
	=
	e^{-\pi(k-1)\beta_j}
	\left(
	\bmprod_{\bm\beta}e^{-\pi\bm\beta}
	\right)
	\left(
	\bmprod_{\bm\zeta^{(k)}}e^{\pi\bm\zeta^{(k)}}
	\right)
	\Phifn'(\beta_j|\bm\zeta^{(k)},\bm\beta)
	.
\end{multline}
\begin{multline}
	\Zcal^{(k)}_{j,a}=
	\lim_{\eta_k\to\infty}
	\Zcal^{(k-1)}_{j_a}-e^{2\pi\eta_k}\Zcal^{(k)}_{j,a+1}
	\hfill
	(a<k)
	\\
	=
	e^{\pi (k-1-2a)\beta_j}
	\left(
	\bmprod_{\bm\beta}e^{-\pi\bm\beta}
	\right)
	\left(
	\bmprod_{\bm\zeta^{(k)}}e^{\pi\bm\zeta^{(k)}}
	\right)
	\Phifn'(\beta_j|\bm\zeta^{(k)},\bm\beta)
	.
\end{multline}
And,
\begin{align}
	\Zcal^{(k)}_{j,a}
	&=
	\lim_{\eta_k\to\infty}
	\Zcal^{(k-1)}
	\hfill
	=
	e^{\pi k(\eta_a-\beta_j)}
	\Phifn'(\beta_j|\bm\zeta^{(k)}\setminus\set{\eta_a},\bm\beta)
	\Phifn'(\eta_a|\bm\beta,\bm\zeta^{(k)})
	&
	&(k>a)
	.
\end{align}
\end{subequations}
In the end, when all the extra variables are send to infinity, we would get
\begin{align}
	\Zcal^{(n)}_{j,a}
	=
	e^{\pi(n-1-2a)\beta_j}
	\left(
	\bmprod_{\bm\beta}e^{-\pi\bm\beta}
	\right)
	\left(
	\bmprod_{\bm\alpha}e^{\pi\bm\alpha}
	\right)
	\Phifn'(\beta_j|\bm\alpha,\bm\beta)
	.
\end{align}
Let us now extract the product over exponentials from all the $n$ columns of the block $\Zcal^{(n)}$, it gives us $\Vdr<\ptn\delta>$:
\begin{align}
	\Vdr<\ptn\delta>_{j,a}
	=
	e^{\pi(n-1-2a)\beta_j}
	\Phifn'(\beta_j|\bm\alpha,\bm\beta)
	.
\end{align}
The extracted terms form the diagonal dressing for the Cauchy-Vandermonde matrix that we have extracted which can be seen from the following expression:
\begin{align}
	\Cmat<\ptn\delta>[\bm\alpha\Vert\bm\beta]
	=
	\diag\Big[
	e^{\pi n\bm\alpha}
	~\Big|~
	\Id_n
	\Big]
	\left\lbrace
	\left(\prod_{a=1}^{n}
	\lim_{\eta_a\to\infty}
	e^{\pi(2a-1)\eta_a}
	\right)
	\Cmat[\bm\zeta\Vert\bm\beta]
	\right\rbrace
	\diag\Big[
	e^{-\pi n\bm\beta}
	\Big]
	.
\end{align}
This ensures that we have the correct normalisation in the determinant [see \cref{hcv_mat_traditional_det}].
Finally we can recombine the columns of $\Vdr<\ptn\delta>$ according to \cref{hcv_append_hvan_recomb} to obtain $\Vdr<\ptn\gamma>$.
\end{enumerate}
\section{Formulae : Infinite products involving \texorpdfstring{$\phifn$}{\textbackslash phi} functions}
\label{sec:pref_append_inf_prod}
The thermodynamic limit of the $\phifn$ was originally obtained in the expression \eqref{phi_tdl} in \cref{sec:tdl_phifn_append} at the end of \cref{chap:spectre}.
Here we will need the thermodynamic limit of different variants of the $\phifn$ which are reproduced in the following expressions.
\begingroup
\allowdisplaybreaks
\begin{subequations}
\begin{align}
	\phi(\tau|\bm{\rl^\pm},\bm\la)&=
	\begin{dcases}
	(2i)^{\frac{n_{h}}{2}}
	\bmprod{(\tau-\bmclp<+>)}
	\bmprod\frac{(\tau-\bmwdp*<+>)}{(\tau-\bmwdp*<->)}
	\bmprod\frac{%
	\Gamma\left(\frac{\tau-\bm\hle}{2i}\right)
	}{%
	\Gamma\left(\frac{1}{2}+\frac{\tau-\bm\hle}{2i}\right)
	}
	&
	\text{for, } \Im\tau>0
	\\
	(-2i)^{\frac{n_{h}}{2}}
	\bmprod{(\tau-\bmclp<->)}
	\bmprod\frac{(\tau-\bmwdp<->)}{(\tau-\bmwdp<+>)}
	\bmprod\frac{%
	\Gamma\left(-\frac{\tau-\bm\hle}{2i}\right)
	}{%
	\Gamma\left(\frac{1}{2}-\frac{\tau-\bm\hle}{2i}\right)
	}
	&
	\text{for, } \Im\tau<0		
	\end{dcases}
	.
	\label{phifn_tdl_w-o_wide}
	\\
	\phi(\tau|\bm\rl,\bm\la)&=
	\begin{dcases}
	(2i)^{-\frac{n_{h}}{2}}
	\bmprod\frac{1}{(\tau-\bmclp<->)}
	\bmprod\frac{(\tau-\bmwdp*<+>)}{(\tau-\bmwdp*<->)}
	\bmprod\frac{%
	\Gamma\left(\frac{\tau-\bm\hle}{2i}\right)
	}{%
	\Gamma\left(\frac{1}{2}+\frac{\tau-\bm\hle}{2i}\right)
	}
	&
	\text{for, } \Im\tau>0
	\\
	(-2i)^{-\frac{n_{h}}{2}}
	\bmprod\frac{1}{(\tau-\bmclp<+>)}
	\bmprod\frac{(\tau-\bmwdp<->)}{(\tau-\bmwdp<+>)}
	\bmprod\frac{%
	\Gamma\left(-\frac{\tau-\bm\hle}{2i}\right)
	}{%
	\Gamma\left(\frac{1}{2}-\frac{\tau-\bm\hle}{2i}\right)
	}
	&
	\text{for, } \Im\tau<0		
	\end{dcases}
	.
	\label{phifn_tdl_rl_only}
	\\
	\phifn(\tau|\bm\la,\bm\mu)&=
	\begin{dcases}
	(2i)^\frac{n_h}{2}
	\bmprod
	\frac{1}{(\tau-\bm\cid-\frac{i}{2})}
	\bmprod\frac{%
	\Gamma\left(\frac{1}{2}+\frac{\la-\bm\hle}{2i}\right)
	}{%
	\Gamma\left(\frac{\la-\bm\hle}{2i}\right)
	}
	&
	\text{for, } \Im\la>0
	\\
	(-2i)^{\frac{n_h}{2}}
	\bmprod
	\frac{1}{(\tau-\bm\cid+\frac{i}{2})}
	\bmprod\frac{%
	\Gamma\left(\frac{1}{2}-\frac{\tau-\bm\hle}{2i}\right)
	}{%
	\Gamma\left(-\frac{\tau-\bm\hle}{2i}\right)
	}
	&
	\text{for, } \Im\la<0
	\end{dcases}
	\label{phifn_tdl_inversed}
	.
\end{align}
\end{subequations}
\endgroup
\subsection{The auxiliary function \texorpdfstring{$\Omegfn_n$}{\textbackslash Omega\_n}}
\label{sub:pref_append_omeg_n_tdl}
Here we compute the thermodynamic limit of $\Omegfn_n$ using that of the $\phifn$ function \eqref{phifn_tdl_w-o_wide}.
Hence,
\begin{subequations}
\begin{multline}
	\prod_{\sigma=\pm 1}
	\phi(\tau+2in\sigma|\bm{\rl^{\pm}},\bm\la)
	=
	\frac{1}{2^{n_h}}
	\left\lbrace
	\bmprod
	\left[
	(\tau-\bmclp)^2+\left(2n-\frac{1}{2}\right)^2
	\right]
	\right\rbrace
	\\
	\left\lbrace
	\bmprod
	\frac{%
	\tau-\bmwdp*+\left(2n-\frac{1}{2}\right)i
	}{%
	\tau-\bmwdp*+\left(2n+\frac{1}{2}\right)i
	}
	\cdot
	\bmprod
	\frac{%
	\tau-\bmwdp-\left(2n-\frac{1}{2}\right)i
	}{%
	\tau-\bmwdp-\left(2n+\frac{1}{2}\right)i
	}
	\right\rbrace
	\prod_{\sigma=\pm 1}
	\left\lbrace
	\bmprod
	\frac{%
	\Gamma\left(n+\frac{\tau-\bm\hle}{2i\sigma}\right)
	}{%
	\Gamma\left(n+\frac{1}{2}+\frac{\tau-\bm\hle}{2i\sigma}\right)
	}
	\right\rbrace
\end{multline}
and
\begin{multline}
	\prod_{\sigma=\pm 1}
	\phi(\tau+(2n-1)i\sigma|\bm{\rl^\pm},\bm\la)
	=
	\frac{1}{2^{n_h}}
	\left\lbrace
	\bmprod
	\left[
	(\tau-\bmclp)^2+\left(2n-\frac{3}{2}\right)^2
	\right]
	\right\rbrace
	\\
	\left\lbrace
	\bmprod
	\frac{%
	\tau-\bmwdp*+\left(2n-\frac{3}{2}\right)i
	}{%
	\tau-\bmwdp*+\left(2n-\frac{1}{2}\right)i
	}
	\cdot
	\bmprod
	\frac{%
	\tau-\bmwdp-\left(2n-\frac{3}{2}\right)i
	}{%
	\tau-\bmwdp-\left(2n-\frac{1}{2}\right)i
	}
	\right\rbrace
	\prod_{\sigma=\pm 1}
	\left\lbrace
	\bmprod
	\frac{%
	\Gamma\left(n-\frac{1}{2}+\frac{\tau-\bm\hle}{2i\sigma}\right)
	}{%
	\Gamma\left(n+\frac{\tau-\bm\hle}{2i\sigma}\right)
	}
	\right\rbrace
\end{multline}
\end{subequations}
Therefore we have
\begin{multline}
	\Omega_{n}(\tau)=
	\frac{(2n)^{2q}}{(2n-1)^{2p}}
	\left\lbrace
	\bmprod\left((\tau-\bmclp)^2+\left(2n-\frac{3}{2}\right)^2\right)
	\bmprod\left((\tau-\bmclp)^2+\left(2n-\frac{1}{2}\right)^2\right)
	\right\rbrace
	\\
	\prod_{\sigma=\pm 1}
	\left\lbrace
	\bmprod\frac{\tau-\bmwdp*+\left(2n-\frac{3}{2}\right)i}{\tau-\bmwdp*+\left(2n+\frac{1}{2}\right)i}
	\bmprod\frac{\tau-\bmwdp-\left(2n-\frac{3}{2}\right)i}{\tau-\bmwdp-\left(2n+\frac{1}{2}\right)i}
	\right\rbrace
	\left\lbrace
	\bmprod\left(
	(\tau-\bm\hle)^2+\left(2n-\frac{1}{2}\right)
	\right)
	\right\rbrace
	\label{pref_gen_append_tdl_omeg_n}
\end{multline}
It gives us the thermodynamic limit for the following fraction is the general term of the infinite product \eqref{inf_prod_ratio_omegfn_gen_pref}
\begin{multline}
	\frac{\bmprod\Omega_{n}(\bm\rl)}{\bmprod\Omega_{n}(\bm\la)}
	=
	\frac{(2n-1)^{2P}}{(2n)^{2Q}}
	\left\lbrace
	\bmprod
	\phi\left(\bmclp<+>-2in\Big|\bm\rl,\bm\la\right)
	\phi\left(\bmclp<+>-(2n-1)i\Big|\bm\rl,\bm\la\right)
	\right.
	\\
	\left.
	\times
	\bmprod
	\phi\left(\bmclp<->+2in\Big|\bm\rl,\bm\la\right)
	\phi\left(\bmclp<->+(2n-1)i\Big|\bm\rl,\bm\la\right)
	\right\rbrace
	\\
	\times
	\left\lbrace
	\bmprod
	\frac{%
	\phi\left(\bmwdp<+>+2(n-1)i\Big|\bm\rl,\bm\la\right)}
	{%
	\phi\left(\bmwdp<+>+2ni\Big|\bm\rl,\bm\la\right)}
	\bmprod
	\frac{%
	\phi\left(\bmwdp*<->-2(n-1)i\Big|\bm\rl,\bm\la\right)
	}{%
	\phi\left(\bmwdp*<->-2ni\Big|\bm\rl,\bm\la\right)
	}
	\right\rbrace
	\\
	\times
	\left\lbrace
	\bmprod
	\phi\left(\bm\hle+(2n-1)i\Big|\bm\rl,\bm\la\right)
	\phi\left(\bm\hle-(2n-1)i\Big|\bm\rl,\bm\la\right)
	\right\rbrace
	.
	\label{gen_pref_ratio_omegfn_inf_prod_gen-term_tdl_append}
\end{multline}
\subsection{Ratio of \texorpdfstring{$\Omega_n$}{\textbackslash Omega\_n functions}}
\label{sub:ratio_omegfn_tdl_append}
For the ratio of $\Omegfn_n$ we obtained the above expression \eqref{gen_pref_ratio_omegfn_inf_prod_gen-term_tdl_append}.
We now compute its thermodynamic limit for its different components using \cref{phifn_tdl_rl_only}.
Using 
\minisec{Close-pair terms}
\begin{subequations}
\begin{multline}
	\bmprod\phi(\bmclp<+>-2ni|\bm\rl,\bm\la)\phi(\bmclp<->+2ni|\bm\rl,\bm\la)
	\\
	=
	2^{-n_h n_\txtcp}
	\left\lbrace
	\bmprod\frac{1}{(\bmclp\bm-\bmclp)^2+4n^2}
	\right\rbrace
	\left\lbrace
	\bmprod\frac{\bmclp\bm-\bmwdp-(2n-1)i}{\bmclp\bm-\bmwdp-2ni}
	\bmprod\frac{\bmclp\bm-\bmwdp*+(2n-1)i}{\bmclp\bm-\bmwdp*+2ni}
	\right\rbrace
	\\
	\times
	\prod_{\sigma=\pm 1}
	\left\lbrace
	\bmprod\frac{%
	\Gamma\left(n-\frac{1}{4}+\frac{\bmclp\bm-\bm\hle}{2i\sigma}\right)
	}{%
	\Gamma\left(n+\frac{1}{4}+\frac{\bmclp\bm-\bm\hle}{2i\sigma}\right)
	}
	\right\rbrace
	\label{pref_tdl_aux_cp_even}
\end{multline}
and
\begin{multline}
	\bmprod\phi(\bmclp<+>-(2n-1)i|\bm\rl,\bm\la)\phi(\bmclp<->+(2n-1)i|\bm\rl,\bm\la)
	\\
	=
	2^{-n_h n_\txtcp}
	\left\lbrace
	\bmprod\frac{1}{(\bmclp\bm-\bmclp)^2+(2n-1)^2}
	\right\rbrace
	\left\lbrace
	\bmprod\frac{\bmclp\bm-\bmwdp-(2n-2)i}{\bmclp\bm-\bmwdp-(2n-1)i}
	\bmprod\frac{\bmclp\bm-\bmwdp*+(2n-2)i}{\bmclp\bm-\bmwdp*+(2n-1)i}
	\right\rbrace
	\\
	\times
	\prod_{\sigma=\pm 1}
	\left\lbrace
	\bmprod\frac{%
	\Gamma\left(n-\frac{3}{4}+\frac{\bmclp\bm-\bm\hle}{2i\sigma}\right)
	}{%
	\Gamma\left(n-\frac{1}{4}+\frac{\bmclp\bm-\bm\hle}{2i\sigma}\right)
	}
	\right\rbrace
	\label{pref_tdl_aux_cp_odd}
\end{multline}
\label{pref_tdl_aux_cp_both}
\end{subequations}
Put together, close-pair contribution to \cref{gen_pref_ratio_omegfn_inf_prod_gen-term_tdl_append} is given by
\begin{multline}
	\bmprod
	\phi(\bmclp<+>-2ni|\bm\rl,\bm\la)
	\phi(\bmclp<->+2ni)|\bm\rl,\bm\la) 
	\phi(\bmclp<+>-(2n-1)i|\bm\rl,\bm\la)
	\phi(\bmclp<->+(2n-1)i|\bm\rl,\bm\la)
	\\
	=
	\left\lbrace
	\bmprod\frac{1}{(\bmclp\bm-\bmclp)^2+(2n)^2}
	\bmprod\frac{1}{(\bmclp\bm-\bmclp)^2+(2n-1)^2}
	\right\rbrace
	\\
	\times
	\left\lbrace
	\bmprod\frac{\bmclp\bm-\bmwdp-2(n-1)i}{\bmclp\bm-\bmwdp-2ni}
	\bmprod\frac{\bmclp\bm-\bmwdp*+2(n-1)i}{\bmclp\bm-\bmwdp*+2ni}
	\right\rbrace
	\\
	\times
	\left\lbrace
	\bmprod\frac{1}{(2n-\frac{3}{2})^2+(\bmclp\bm-\bm\hle)^2}
	\right\rbrace
	.
	\label{pref_tdl_aux_cp}
\end{multline}
\minisec{Wide-pair terms}
Wide-pair contribution to \cref{gen_pref_ratio_omegfn_inf_prod_gen-term_tdl_append} is 
\begin{multline}
	\bmprod\frac%
	{%
	\phi(\wdp<+>+2(n-1)i|\bm\rl,\bm\la)%
	}%
	{
	\phi(\wdp<+>+2ni|\bm\rl,\bm\la)%
	}%
	\frac{%
	\phi(\wdp*<->-2(n-1)i|\bm\rl,\bm\la)%
	}{%
	\phi(\wdp*<->-2n|\bm\rl,\bm\la%
	)}
	\\
	=
	\left\lbrace
	\bmprod%
	\frac{%
	\bmwdp\bm-\bmclp+(2n+1)i
	}{%
	\bmwdp\bm-\bmclp+(2n-1)i
	}
	\frac{%
	\bmwdp*\bm-\bmclp-(2n+1)i
	}{%
	\bmwdp*\bm-\bmclp-(2n-1)i
	}
	\right\rbrace
	\\
	\times
	\left\lbrace
	\bmprod\frac{(\bmwdp\bm-\bmwdp*+(2n-2)i)(\bmwdp\bm-\bmwdp*+(2n+1)i)}{(\bmwdp\bm-\bmwdp*+(2n-1)i)(\bmwdp\bm-\bmwdp*+2ni)}
	\right.
	\\
	\left.
	\bmprod\frac{(\bmwdp*\bm-\bmwdp-(2n-2)i)(\bmwdp*\bm-\bmwdp-(2n+1)i)}{(\bmwdp*\bm-\bmwdp-(2n-1)i)(\bmwdp*\bm-\bmwdp-2ni)}
	\right\rbrace
	\\
	\times
	\bmprod\frac{%
	\Gamma\left(n-\frac{3}{4}+\frac{\bmwdp\bm-\bm\hle}{2i}\right)
	\Gamma\left(n+\frac{3}{4}+\frac{\bmwdp\bm-\bm\hle}{2i}\right)
	\Gamma\left(n-\frac{3}{4}-\frac{\bmwdp*\bm-\bm\hle}{2i}\right)
	\Gamma\left(n+\frac{3}{4}-\frac{\bmwdp*\bm-\bm\hle}{2i}\right)
	}{%
	\Gamma\left(n-\frac{1}{4}+\frac{\bmwdp\bm-\bm\hle}{2i}\right)
	\Gamma\left(n+\frac{1}{4}+\frac{\bmwdp\bm-\bm\hle}{2i}\right)
	\Gamma\left(n-\frac{1}{4}+\frac{\bmwdp*\bm-\bm\hle}{2i}\right)
	\Gamma\left(n+\frac{1}{4}+\frac{\bmwdp*\bm-\bm\hle}{2i}\right)
	}
	.
	\label{pref_tdl_aux_wp}
\end{multline}
\minisec{Hole terms}
Hole contribution to \cref{gen_pref_ratio_omegfn_inf_prod_gen-term_tdl_append} is
\begin{multline}
	\bmprod
	\phi(\bm\hle+(2n-1)i|\bm\rl,\bm\la)
	\phi(\bm\hle-(2n-1)i|\bm\rl,\bm\la)
	\\
	=
	2^{-N_h^2}
	\left\lbrace
	\frac{1}{(2n-\frac{1}{2})^2+(\bmclp\bm-\bm\hle)^2}
	\right\rbrace
	\left\lbrace
	\bmprod
	\frac{%
	(\bm\hle-\bmwdp-(2n-\frac{3}{2})i)
	(\bm\hle-\bmwdp*+(2n-\frac{3}{2})i)
	}{%
	(\bm\hle-\bmwdp-(2n-\frac{1}{2})i)
	(\bm\hle-\bmwdp*+(2n-\frac{1}{2})i)
	}
	\right\rbrace
	\\
	\times
	\prod_{\sigma=\pm 1}
	\left\lbrace
	\bmprod
	\frac{%
	\Gamma\left(n-\frac{1}{2}+\frac{\bm\hle-\bm\hle}{2i\sigma}\right)
	}{%
	\Gamma\left(n+\frac{\bm\hle-\bm\hle}{2i\sigma}\right)
	}
	\right\rbrace
	.
	\label{pref_tdl_aux_hle}
\end{multline}
\minisec{Result}
\Cref{pref_tdl_aux_cp,pref_tdl_aux_wp,pref_tdl_aux_hle} can be put together in \cref{gen_pref_ratio_omegfn_inf_prod_gen-term_tdl_append} to produce
\begingroup
\allowdisplaybreaks
\begin{multline}
	\frac{\bmprod\Omega_{n}(\bm\rl)}{\bmprod\Omega_{n}(\bm\la)}
	=
	2^{-N_h^2}\frac{(2n-1)^{2P}}{(2n)^{2Q}}
	\\[\jot]
	\begin{aligned}[t]
	&\times&
	\bigg\lbrace
	&
	\bmprod_{\bmclp}
	\frac{1}{%
	((\bmclp\bm-\bmclp)^2+(2n)^2)
	}
	\frac{1}{%
	((\bmclp\bm-\bmclp)^2+(2n-1)^2)
	}
	\bigg\rbrace
	\\[\jot]
	&\times&
	\bigg\lbrace
	&
	\bmprod_{\bmclp,\bmwdp}
	\frac{%
	(\bmclp\bm-\bmwdp-(2n+1)i)
	(\bmclp\bm-\bmwdp-(2n-2)i)
	}{%
	(\bmclp\bm-\bmwdp-2ni)
	(\bmclp\bm-\bmwdp-(2n-1)i)
	}
	\\[\jot]
	&\times&
	&
	\bmprod_{\bmclp,\bmwdp*}
	\frac{%
	(\bmclp\bm-\bmwdp*+(2n+1)i)
	(\bmclp\bm-\bmwdp*+(2n-2)i)
	}{%
	(\bmclp\bm-\bmwdp*+2ni)
	(\bmclp\bm-\bmwdp*+(2n-1)i)
	}
	\bigg\rbrace
	\\[\jot]
	&\times&
	\bigg\lbrace
	&
	\bmprod_{\bmwdp,\bmwdp*}
	\frac{(\bmwdp\bm-\bmwdp*+(2n-2)i)(\bmwdp\bm-\bmwdp*+(2n+1)i)}{(\bmwdp\bm-\bmwdp*+(2n-1)i)(\bmwdp\bm-\bmwdp*+2ni)}
	\\[\jot]
	&\times&
	&
	\bmprod_{\bmwdp,\bmwdp*}
	\frac{(\bmwdp*\bm-\bmwdp-(2n-2)i)(\bmwdp*\bm-\bmwdp-(2n+1)i)}{(\bmwdp*\bm-\bmwdp-(2n-1)i)(\bmwdp*\bm-\bmwdp-2ni)}
	\bigg\rbrace
	\\[\jot]
	&\times&
	\bigg\lbrace
	&
	\bmprod_{\bmclp,\bm\hle}
	\frac{1}{%
	((2n-\frac{1}{2})^2+(\bmclp\bm-\bm\hle)^2)
	}
	\frac{1}{%
	((2n-\frac{3}{2})^2+(\bmclp\bm-\bm\hle)^2)
	}
	\bigg\rbrace
	\end{aligned}
	\\
	\times
	\prod_{\sigma=\pm 1}
	\bmprod
	\left\lbrace
	\frac{%
	\Gamma\left(n-\frac{1}{2}+\frac{\bm\hle-\bm\hle}{2i\sigma}\right)
	}{%
	\Gamma\left(n+\frac{\bm\hle-\bm\hle}{2i\sigma}\right)
	}
	\right\rbrace
	.
	\label{gen_pref_rat_omegfn_gen-term_ratio_tdl_append}
\end{multline}
\endgroup
\subsection{Ratio of \texorpdfstring{$\Omega$}{\textbackslash Omega}}
\label{sub:tdl_rat_omegfn_append}
Notice that here we obtain all possible cross terms between different species of parameters, i.e. holes, close-pairs and wide-pairs. All of these are new occurrences except the hole-hole term which is similar to the two-spinon case \eqref{tdl_rat_Omegfn_inf_prod}. A most problematic of these could be wide-pair/ wide pair term as it does not lie on the unit circle.
\par
In order to examine the convergence and find a close-form in terms of the Barnes function, we can first rewrite this expression in terms of the Gamma functions.
\begin{multline}
	\frac{\bmprod\Omega_{n}(\bm\rl)}{\bmprod\Omega_{n}(\bm\la)}
	=
	\frac{%
	\Gamma^{2P}\left(n+\frac{1}{2}\right)
	\Gamma^{2Q}\left(n\right)
	}{%
	\Gamma^{2P}\left(n-\frac{1}{2}\right)
	\Gamma^{2Q}\left(n+1\right)
	}
	\prod_{\sigma=\pm 1}
	\left\lbrace
	\bmprod
	\frac{%
	\Gamma\left(n+\frac{\bmclp\bm-\bmclp}{2i\sigma}\right)
	\Gamma\left(n-\frac{1}{2}+\frac{\bmclp\bm-\bmclp}{2i\sigma}\right)
	}{%
	\Gamma\left(n+1+\frac{\bmclp\bm-\bmclp}{2i\sigma}\right)
	\Gamma\left(n+\frac{1}{2}+\frac{\bmclp\bm-\bmclp}{2i\sigma}\right)
	}
	\right\rbrace
	\\
	\times
	\left\lbrace
	\bmprod
	\frac{%
	\Gamma\left(n-\frac{1}{2}-\frac{\bmclp\bm-\bmwdp}{2i}\right)
	\Gamma^2\left(n-\frac{\bmclp\bm-\bmwdp}{2i}\right)
	\Gamma\left(n+\frac{3}{2}-\frac{\bmclp\bm-\bmwdp}{2i}\right)
	}{%
	\Gamma\left(n-1-\frac{\bmclp\bm-\bmwdp}{2i}\right)
	\Gamma^2\left(n+\frac{1}{2}-\frac{\bmclp\bm-\bmwdp}{2i}\right)
	\Gamma\left(n+1-\frac{\bmclp\bm-\bmwdp}{2i}\right)
	}
	\right.
	\\
	\left.
	\bmprod
	\frac{%
	\Gamma\left(n-\frac{1}{2}+\frac{\bmclp\bm-\bmwdp*}{2i}\right)
	\Gamma^2\left(n+\frac{\bmclp\bm-\bmwdp*}{2i}\right)
	\Gamma\left(n+\frac{3}{2}+\frac{\bmclp\bm-\bmwdp*}{2i}\right)
	}{%
	\Gamma\left(n-1+\frac{\bmclp\bm-\bmwdp*}{2i}\right)
	\Gamma^2\left(n+\frac{1}{2}+\frac{\bmclp\bm-\bmwdp*}{2i}\right)
	\Gamma\left(n+1+\frac{\bmclp\bm-\bmwdp*}{2i}\right)
	}
	\right\rbrace
	\\
	\times
	\left|
	\bmprod
	\frac{%
	\Gamma\left(n-\frac{1}{2}+\frac{\bmwdp\bm-\bmwdp*}{2i}\right)
	\Gamma^2\left(n+\frac{\bmwdp\bm-\bmwdp*}{2i}\right)
	\Gamma\left(n+\frac{3}{2}+\frac{\bmwdp\bm-\bmwdp*}{2i}\right)
	}{%
	\Gamma\left(n-1+\frac{\bmwdp\bm-\bmwdp*}{2i}\right)
	\Gamma^2\left(n+\frac{1}{2}+\frac{\bmwdp\bm-\bmwdp*}{2i}\right)
	\Gamma\left(n+1+\frac{\bmwdp\bm-\bmwdp*}{2i}\right)
	}
	\right|^2
	\\
	\times
	\prod_{\sigma=\pm 1}
	\left\lbrace
	\frac{%
	\Gamma\left(n-\frac{1}{4}+\frac{\bmclp\bm-\bm\hle}{2i\sigma}\right)
	\Gamma\left(n-\frac{3}{4}+\frac{\bmclp\bm-\bm\hle}{2i\sigma}\right)
	}{%
	\Gamma\left(n+\frac{3}{4}+\frac{\bmclp\bm-\bm\hle}{2i\sigma}\right)
	\Gamma\left(n+\frac{1}{4}+\frac{\bmclp\bm-\bm\hle}{2i\sigma}\right)
	}
	\right\rbrace
	\prod_{\sigma=\pm 1}
	\left\lbrace
	\bmprod
	\frac{%
	\Gamma\left(n-\frac{1}{2}+\frac{\bm\hle-\bm\hle}{2i\sigma}\right)
	}{%
	\Gamma\left(n+\frac{\bm\hle-\bm\hle}{2i\sigma}\right)
	}
	\right\rbrace
	.
	\label{gen_pref_infprod_gamma_append}
\end{multline}
Let us study all the cross-terms one-by-one.
For each of these terms we test the criteria found in \cref{lem:infprod_gamma_fn_crit}.
We can also observe that since all the terms are real and expressed as conjugated products, it would be sufficient to look at the real parts of the numerators and denominators in these terms.
\minisec{Close-pair / close-pair term}
\begin{align}
	\prod_{\sigma=\pm 1}
	\left\lbrace
	\bmprod
	\frac{%
	\Gamma\left(n+\frac{\bmclp\bm-\bmclp}{2i\sigma}\right)
	\Gamma\left(n-\frac{1}{2}+\frac{\bmclp\bm-\bmclp}{2i\sigma}\right)
	}{%
	\Gamma\left(n+1+\frac{\bmclp\bm-\bmclp}{2i\sigma}\right)
	\Gamma\left(n+\frac{1}{2}+\frac{\bmclp\bm-\bmclp}{2i\sigma}\right)
	}
	\right\rbrace
\end{align}
\begin{subequations}
\begin{gather}
	\bmsum \Re(\textbf{num})-\bmsum \Re(\textbf{den})
	=
	-4n_\txtcp^2
\\
	\bmsum \Re(\textbf{num})^2-\bmsum \Re(\textbf{den})^2
	=
	-2n_\txtcp^2(4n+1)
	.
\end{gather}
\end{subequations}
\minisec{Cross terms involving wide-pairs}
\begin{subequations}
\begin{align}
	\bmprod
	\frac{%
	\Gamma\left(n-\frac{1}{2}-\frac{\bmclp\bm-\bmwdp}{2i}\right)
	\Gamma^2\left(n-\frac{\bmclp\bm-\bmwdp}{2i}\right)
	\Gamma\left(n+\frac{3}{2}-\frac{\bmclp\bm-\bmwdp}{2i}\right)
	}{%
	\Gamma\left(n-1-\frac{\bmclp\bm-\bmwdp}{2i}\right)
	\Gamma^2\left(n+\frac{1}{2}-\frac{\bmclp\bm-\bmwdp}{2i}\right)
	\Gamma\left(n+1-\frac{\bmclp\bm-\bmwdp}{2i}\right)
	}
	\\
	\bmprod
	\frac{%
	\Gamma\left(n-\frac{1}{2}+\frac{\bmclp\bm-\bmwdp*}{2i}\right)
	\Gamma^2\left(n+\frac{\bmclp\bm-\bmwdp*}{2i}\right)
	\Gamma\left(n+\frac{3}{2}+\frac{\bmclp\bm-\bmwdp*}{2i}\right)
	}{%
	\Gamma\left(n-1+\frac{\bmclp\bm-\bmwdp*}{2i}\right)
	\Gamma^2\left(n+\frac{1}{2}+\frac{\bmclp\bm-\bmwdp*}{2i}\right)
	\Gamma\left(n+1+\frac{\bmclp\bm-\bmwdp*}{2i}\right)
	}
\\
	\bmprod
	\frac{%
	\Gamma\left(n-\frac{1}{2}+\frac{\bmwdp\bm-\bmwdp*}{2i}\right)
	\Gamma^2\left(n+\frac{\bmwdp\bm-\bmwdp*}{2i}\right)
	\Gamma\left(n+\frac{3}{2}+\frac{\bmwdp\bm-\bmwdp*}{2i}\right)
	}{%
	\Gamma\left(n-1+\frac{\bmwdp\bm-\bmwdp*}{2i}\right)
	\Gamma^2\left(n+\frac{1}{2}+\frac{\bmwdp\bm-\bmwdp*}{2i}\right)
	\Gamma\left(n+1+\frac{\bmwdp\bm-\bmwdp*}{2i}\right)
	}
\end{align}
\end{subequations}
\begin{subequations}
\begin{gather}
	\bmsum \Re(\textbf{num})-\bmsum \Re(\textbf{den})
	=
	0
	\\
	\bmsum \Re(\textbf{num})^2 - \bmsum \Re(\textbf{den})^2
	=
	0
	.
\end{gather}
\end{subequations}
\minisec{Close-pair / hole term}
\begin{align}
	\prod_{\sigma=\pm 1}
	\left\lbrace
	\frac{%
	\Gamma\left(n-\frac{1}{4}+\frac{\bmclp\bm-\bm\hle}{2i\sigma}\right)
	\Gamma\left(n-\frac{3}{4}+\frac{\bmclp\bm-\bm\hle}{2i\sigma}\right)
	}{%
	\Gamma\left(n+\frac{3}{4}+\frac{\bmclp\bm-\bm\hle}{2i\sigma}\right)
	\Gamma\left(n+\frac{1}{4}+\frac{\bmclp\bm-\bm\hle}{2i\sigma}\right)
	}
	\right\rbrace
\end{align}
\begin{subequations}
\begin{gather}
	\bmsum \Re(\textbf{num})-\bmsum \Re(\textbf{den})
	=
	-4n_\txtcp n_h
	\\
	\bmsum \Re(\textbf{num})^2-\bmsum \Re(\textbf{den})^2
	=
	-8n\,n_\txtcp n_h
\end{gather}
\end{subequations}
\minisec{Hole / hole term}
\begin{align}
	\prod_{\sigma=\pm 1}
	\left\lbrace
	\bmprod
	\frac{%
	\Gamma\left(n-\frac{1}{2}+\frac{\bm\hle-\bm\hle}{2i\sigma}\right)
	}{%
	\Gamma\left(n+\frac{\bm\hle-\bm\hle}{2i\sigma}\right)
	}
	\right\rbrace	
\end{align}
\begin{subequations}
\begin{gather}
	\bmsum \Re(\textbf{num})-\bmsum \Re(\textbf{den})
	=
	-n_h^2
	\\
	\bmsum \Re(\textbf{num})^2-\bmsum \Re(\textbf{den})^2
	=
	(-2n+\frac{1}{2})n_h^2.
\end{gather}
\end{subequations}
Hence, overall for all the cross-terms combined we have
\begin{subequations}
\begin{gather}
	\bmsum \Re(\textbf{num})-\bmsum \Re(\textbf{den})
	=
	-(n_h+2n_\txtcp)^2
	\\
	\bmsum \Re(\textbf{num})^2-\bmsum \Re(\textbf{den})^2
	=
	-2n(n_h+2n_\txtcp)^2+(\frac{1}{2}n_h^2-2n_\txtcp^2).
\end{gather}
\end{subequations}
\minisec{Other terms}
\begin{align}
	\frac{%
	\Gamma^{2P}\left(n+\frac{1}{2}\right)
	\Gamma^{2Q}\left(n\right)
	}{%
	\Gamma^{2P}\left(n-\frac{1}{2}\right)
	\Gamma^{2Q}\left(n+1\right)
	}
\end{align}
\begin{subequations}
\begin{gather}
	\bmsum \Re(\textbf{num})-\bmsum \Re(\textbf{den})
	=
	2(P-Q)
	\\
	\bmsum \Re(\textbf{num})^2-\bmsum \Re(\textbf{den})^2
	=
	4n(P-Q)-2Q
	.
\end{gather}
\end{subequations}
The comparison of the two computations which leads to the cancellation can be easily seen from \cref{gen_pref_qno_PQ_diff,gen_pref_qno_Q_def} which transforms quantum numbers $n_h$ and $n_\txtcp$ into $P$ and $Q$.
\begin{multline}
	\frac{\bmprod\Omega(\bm\rl)}{\bmprod\Omega(\bm\la)}
	=
	\frac{(2i)^{n_\txtcp n_h}}{\pi^{\frac{1}{2}n_h^2+Q+n_\txtcp}}
	\bmprod^\prime\frac{\bmclp\bm-\bmclp}{\sinh\pi(\bmclp\bm-\bmclp)}
	\bmprod\frac{1}{\sinh\pi(\bmclp<+>\bm-\bm\hle)}
	\\
	\times
	\bmprod
	\frac{%
	\bmclp\bm-\bmwdp
	}{%
	\bmclp\bm-\bmwdp-i
	}
	\bmprod
	\frac{%
	\bmclp\bm-\bmwdp*
	}{%
	\bmclp\bm-\bmwdp*+i
	}
	~
	\bmprod
	\frac{%
	\bmwdp\bm-\bmwdp*
	}{%
	\bmwdp\bm-\bmwdp*+i
	}
	\bmprod
	\frac{%
	\bmwdp*\bm-\bmwdp
	}{%
	\bmwdp*\bm-\bmwdp-i
	}
	\\
	\times
	\frac{1}{G^{2n_{h}}(\frac{1}{2})}
	\prod_{\sigma=\pm 1}
	\bmalt
	\frac{%
	G^{2}\left(1+\frac{\bm\hle\bm-\bm\hle}{2i\sigma}\right)
	}{%
	G^{2}\left(\frac{1}{2}+\frac{\bm\hle\bm-\bm\hle}{2i\sigma}\right)
	}
	.
	\label{tdl_rat_omegfn_append}
\end{multline}
\subsection{\texorpdfstring{$\phifn$}{\textbackslash phi} function involving complex roots in the prefactors}
\label{sub:pref_phifn_cmplx_append}
We compute the remaining $\phifn$ in the prefactor of \cref{pref_gen_step3}.
\minisec{Close-pair terms}
From \cref{phifn_tdl_inversed} we get
\begin{align}
	\phifn'(\clp<+>_a-i|\bm\la,\bm\mu)
	\phifn(\clp<->_a-i|\bm\la,\bm\mu)
	=
	\bmprod'\frac{1}{\clp_a-\bm\cid}
	\frac{\bmprod(\clp_a-\bm\hle-\frac{i}{2})}{\bmprod(\clp_a-\bm\cid-i)}
	.
\end{align}	
It is important to note that for a close pair we need to choose the branch $\Im\nu<0$ for both $\phifn'(\clp<+>_a-i)$ and $\phifn(\clp<->_a-i)$ since $|\Im\nu|<1$ for a close-pair.
The primed symbol $\phifn'$ means we have removed the vanishing term $\clp<+>_a-\clp<->_a-i$ from the denominator.
\\
Hence,
\begin{align}
	\bmprod\phifn'(\bmclp<+>-i|\bm\la,\bm\mu)
	\bmprod\phifn(\bmclp<->-i|\bm\la,\bm\mu)
	=
	\bmprod'\frac{1}{\bmclp-\bm\cid}
	\frac{\bmprod(\bmclp-\bm\hle-\frac{i}{2})}{\bmprod(\bmclp-\bm\cid-i)}
	.
\end{align}	
\minisec{Wide-pair terms}
From \cref{phifn_tdl_rl_only,phifn_tdl_inversed} we get the following expression for the product
\begin{multline}
	\phifn(\wdp<+>_a|\bm\la,\bm\rl)	
	\phifn(\wdp*<->_a|\bm\la,\bm\rl)	
	\phifn(\wdp<+>_a-i|\bm\la,\bm\mu)	
	\phifn(\wdp*<->_a-i|\bm\la,\bm\mu)
	\\
	=
	\bmprod(\wdp_a-\bmclp+i)
	\bmprod(\wdp*_a-\bmclp-i)
	\bmprod\frac{\wdp_a-\bmwdp*+i}{\wdp_a-\bmwdp*}
	\bmprod\frac{\wdp*_a-\bmwdp-i}{\wdp*_a-\bmwdp}
	\\
	\times
	\frac{\bmprod(\wdp_a-\bm\hle-\frac{i}{2})}{\bmprod(\wdp_a-\bm\cid-i)}
	\frac{\bmprod(\wdp*_a-\bm\hle-\frac{i}{2})}{\bmprod(\wdp*_a-\bm\cid-i)}
	.
\end{multline}
Here we choose the positive branch $\Im\nu>0$ for $\phifn(\wdp<+>)$ and $\phifn(\wdp<+>-i)$ while the negative branch $\Im\nu<0$ for $\phifn(\wdp<->)$ and $\phifn(\wdp<->-i)$.
\\
Hence,
\begin{multline}
	\bmprod
	\phifn(\bmwdp<+>|\bm\la,\bm\rl)
	\bmprod	
	\phifn(\bmwdp*<->|\bm\la,\bm\rl)
	\bmprod	
	\phifn(\bmwdp<+>-i|\bm\la,\bm\mu)
	\bmprod	
	\phifn(\bmwdp*<->-i|\bm\la,\bm\mu)
	\\
	=
	\bmprod(\bmwdp-\bmclp+i)
	\bmprod(\bmwdp*-\bmclp-i)
	\bmprod\frac{\bmwdp-\bmwdp*+i}{\bmwdp-\bmwdp*}
	\bmprod\frac{\bmwdp*-\bmwdp-i}{\bmwdp*-\bmwdp}
	\\
	\times
	\frac{\bmprod(\bmwdp-\bm\hle-\frac{i}{2})}{\bmprod(\bmwdp-\bm\cid-i)}
	\frac{\bmprod(\bmwdp*-\bm\hle-\frac{i}{2})}{\bmprod(\bmwdp*-\bm\cid-i)}
	.
\end{multline}
\end{subappendices}
\clearpage{}%

\part[Summary and conclusions]{Summary \\ and \\ conclusions}
\label{concl}
\clearpage{}%
\chapter{Conclusions}
\label{conclusion}

\addsec{Summary of results}
We have computed the thermodynamic limit of a longitudinal form-factor of the XXX chain starting from its determinant representation.
First of all, in \cref{lem:seln_criteria} we found that only the triplet excitations have non-trivial form-factors for the XXX model.
The triplet excitations are classified into spinon sectors according to the number of holes $n_h$ which is always an even integer.
A simplest example of it is a two-spinon $n_h=2$ excitation which is free of spinon bound states and hence it does not contain any complex Bethe roots.
We have computed the thermodynamic limit of the form-factor for this two-spinon triplet excitation \cite{KitK19} which was reproduced here in \cref{chap:2sp_ff}. The result \eqref{2sp_ff_result} can be shown as follows: 
\begin{align}
	\left|\FF^{z}(\hle_1,\hle_2)\right|^2
	&=
	\frac{2}{M^2 G^4\left(\frac{1}{2}\right)}
	\prod_{\sigma=\pm}
	\frac{%
	G(\frac{\hle_{2}-\hle_{1}}{2i\sigma})
	G(1+\frac{\hle_{2}-\hle_{1}}{2i\sigma})
	}{%
	G(\frac{1}{2}+\frac{\hle_{2}-\hle_{1}}{2i\sigma})
	G(\frac{3}{2}+\frac{\hle_{2}-\hle_{1}}{2i\sigma})
	}
	.
	\label{2sp_ff_result_concl}
\end{align}
This computation can be summarised as a three step process comprising of the extraction of Gaudin matrices followed by extraction of a Cauchy matrix and the computation of the infinite Cauchy determinants and its prefactors in the thermodynamic limit using the infinite product form.
To realise this computation we invoke a general version of the condensation property that allows us compute summations involving meromorphic function as the integrals with the densities.
It was introduced in \cref{gen_condn_prop} for the compact Fermi distribution and we conjecture that it also applies to the non-compact Fermi distribution of our XXX model in zero external field, at-least for the spectral parameters which are in the bulk of the Fermi distribution.
We also compared our with the previous result \cite{BouCK96,BouKM98} obtained using the method based on the framework of $q$-vertex operator algebra and found that our result is in agreement with it.
Particularly our result \eqref{2sp_ff_result_concl} contains the rational form in the Barnes-G function and using the integral representation for the logarithm of the Barnes-G function \eqref{barnes_log_Gfn_intrep} one readily finds the special function $I$ with the following integral representation as it appears in \cite{BouCK96}:
\begin{align}
	I(z)
	=
	\int_{0}^{\infty}
	\frac{dt\, e^t}{t}
	\frac{%
	\cos(2z t)\cosh2t-1
	}{%
	\cosh t\sinh 2t
	}
	.
	\label{2sp_ff_res_int_form_concl}
\end{align}
\par
We generalised our approach over to the computation of longitudinal form-factor for a generic triplet excitations in the higher spinon sectors in \cref{chap:cau_det_rep_gen,chap:gen_FF}.
All the triplet excitations in the higher spinon sectors contains complex Bethe roots which are determined in terms of the spectral parameters for the holes through a set of higher-level version of \cite{DesL82} Bethe equations \eqref{hl_bae} involving the set $\bm\cid$ of centres and anchors of the close-pairs and wide-pairs respectively.
These higher-level Bethe equations are obtained by factorising out from Bethe equations for a low-lying excited state its dominant part described by ground state root density function.
Our method based on the extraction of the Gaudin matrix allows us to do the same for the form-factors of a low-lying excited state as we find in \cref{chap:cau_det_rep_gen} that the \emph{higher-level} Gaudin matrix $\Ho{\Ncal}$ emerges from our treatment:
\begin{align}
	\Ho{\Ncal}_{a,b}
	=
	\aux*'(\cid_a)
	-
	2\pi i K(\cid_a-\cid_b)
	.
	\label{hl_gau_mat_concl}
\end{align}
We find that similar to the Gaudin extraction $\Fcal=\Ncal^{-1}\Mcal$, we have a \emph{higher-level} its higher-level equivalent $\ho{\Scal}=\Ho{\Ncal}^{-1}\Ho{\Rcal}$ of the Gaudin extraction and the resultant matrix $\ho{\Scal}$ plays an important part in our computation by occupying a block \eqref{cau_ex_II_reduction} inside the final determinant representation \eqref{red_det_rep_generic} of the reduced size seen below:
\begin{multline}
	\left|\FF^{z}(\set{\hle_a}_{a=1}^{n_h})\right|^2=
	(-1)^{\frac{n_h+2}{2}}
	M^{-n_h}
	2^{\frac{n_h(n_h-2)+2}{2}}	
	\pi^{\frac{n_h(n_h-3)+2}{2}}	
	\frac{\prod_{a=1}^{\ho{n}}\prod_{b=1}^{n_h}(\cid_a-\hle_b-\frac{i}{2})}{\prod_{a,b=1}^{\ho{n}}(\cid_a-\cid_b-i)}
	\\
	\times
	\frac{1}{G^{2n_h}(\frac{1}{2})}
	\prod_{\underset{a\neq b}{a,b=1}}^{n_h}
	\frac{%
	G(\frac{\hle_a-\hle_b}{2i})
	G(1+\frac{\hle_a-\hle_b}{2i})
	}{%
	G(\frac{1}{2}+\frac{\hle_a-\hle_b}{2i})
	G(\frac{3}{2}+\frac{\hle_a-\hle_b}{2i})
	}
	~
	\frac{%
	\det_{\ho{n}}\resmat*[g]
	\det_{n_h}\resmat*[e]
	}{\det\vmat[\bm\hle]}
	.
	\label{red_det_rep_generic_concl}
\end{multline}
The matrix $\resmat*[g]$ is made up of the components shown in \cref{cau_ex_I_clp_int_form,cau_ex_I_wdp_int_form}, which were expressed them in terms of the integrals over $\Phifn$ functions which enters our final expression as it is.
The matrix $\resmat*[e]$ is made up of three blocks, one of which is the result $\ho{\Scal}$ of the higher-level Gaudin extraction that we discussed above while the other two blocks are the effective Foda-Wheeler columns $\Wcal^{\text{eff}}$ \eqref{cau_ex_II_fw_eff_block_mat} and the effective dressed Vandermonde columns $\Zcal^{\text{eff}}$ \eqref{hyper_cau_van_inv_hvan_block} of the hyperbolic parametrisation.
The latter block $\Zcal^{\text{eff}}$ was only expressed in terms of the integrals over the $\Phifn$ functions.
\par
Comparing this result with the BJMST fermionic approach \cite{JimMS11} which also deals with the complex Bethe roots (or spinon bound states) we find that our result agrees with their prediction since we have found that the thermodynamic form-factors of a low-lying excitations are given by smaller reduced determinants $\det\Qcal_{g}$ and $\det\Qcal_e$.
Our approach can in principle can allow us to go further and obtain a closed-form expression however this would require that the $\Phifn$ integrals are exactly computed.
In this regard, the fact that integrals with an auxiliary $\Phifn$ function enters our final expression is a shortcoming of this result.
This problem was successfully resolved in the two-spinon case in \cref{chap:2sp_ff} since there we managed to eliminate all the $\Phifn$ integrals using the periodicity argument \eqref{cau_ex_fw_new_contour_zero} for the finite $\Phifn$ function and we never had to invoke the thermodynamic limit of the $\Phifn$ function.
Unfortunately this method does not extend to the Cauchy extraction in generic case and it always leaves behind the integrals over $\Phifn$.
Moreover we find that unlike the rational $\phifn$ function, its hyperbolic version $\Phifn$ does not have a well defined thermodynamic limit.
A computation based on the expansion of $\Phifn$ function as an infinite product over $\phifn$ function shows that the limit is divergent for the excitation $s\neq 0$. For the singlet excitations $s=0$ where,
\begin{subequations}
\label{Phifn_singlet_asym_anom_concl_both}
\begin{align}
	\Phifn(\la)\sim_{\la\to\infty} 1
	\label{Phifn_singlet_asym_anom_concl_finite}
\end{align}
we find that $\Phifn$ computed as an infinite product converges to $\tilde{\Phifn}$ which has an anomalous asymptotic behaviour:
\begin{align}
	\tilde{\Phifn}(\la) \sim_{\la\to\infty} \frac{1}{\la}.
	\label{Phifn_singlet_asym_anom_concl_tdl}
\end{align}
\end{subequations}
This means that we cannot rely on the $\tilde{\Phi}$ to compute these integrals in $\Zcal^{\text{eff}}$ and $\resmat*[g]$.
We also know that this problem is independent of the bulk assumption that was used for the Gaudin extraction.
This is because of the fact that the Gaudin extraction for the complex roots can be realised with the regular condensation property [see \cref{thm:cndn_YY_Koz}] since there are poles of the integrand on the real line, as we may recall from \cref{gau_ex_inteq_I_clp,gau_ex_inteq_I_wdp} in \cref{chap:cau_det_rep_gen}.
Therefore to resolve this problem we may need to consider the sub-leading terms while computing the $\Phifn$ asymptotically as an infinite product over $\phifn$ functions, or otherwise consider a new approach that  circumvents this problem entirely (which may also involve the periodicity argument).
Indeed from experience we know that two-spinon form-factor when computed using the double sum method \eqref{ff_trans_double_sum} always leads to the same type of $\Phifn$ integrals that cannot be resolved, and it was found that the Foda-Wheeler determinant representation \eqref{det_rep_fw} \cite{FodW12a} helps us circumvent this issue.
This anecdote begs that the second option for this problem cannot be easily discarded.
\par
Regardless of the (temporary) irreducibility of the $\Phifn$ function integrals, we can still compute the reduced determinants $\det\Qcal_{g/e}$ in \cref{red_det_rep_generic_concl} in the four-spinon case.
This we done in \cref{sub:4sp_ff_example} to shown an example of the four-spinon form-factor where we found that four-spinon form-factor can be represented as
\begin{align}
 	\left|\FF^z(\hle_1,\hle_2,\hle_3,\hle_4)\right|^2
 	=
	-
	\frac{32\pi^3}{M^{4} G^8(\frac{1}{2})}
	\prod_{a\neq b}
	\frac{%
	G(\frac{\hle_a-\hle_b}{2i})
	G(1+\frac{\hle_a-\hle_b}{2i})
	}{%
	G(\frac{1}{2}+\frac{\hle_a-\hle_b}{2i})
	G(\frac{3}{2}+\frac{\hle_a-\hle_b}{2i})
	}
	\frac{%
	\Jcal_g
	\Jcal_e
	}{%
	\sum_{a=1}^{n_h} \ho{\rden}(\clp-\hle_a)
	}
	.
	\label{4sp_ff_result_concl}
\end{align}
The term $\Jcal_g$ in this expression can be written in the form of the $\Phifn$ integrals, as it was seen in \cref{cau_ex_I_det_mat_4sp} while, $\Jcal_e$ can be expressed as a sum over the terms, which themselves can be expressed with the $\Phifn$ integrals, as it was shown in \cref{4sp_hvan_col,4sp_ff_little_cv-extn_result}.
The result for the four-spinon form-factor was also computed using the $q$-vertex operator algebra method in \cite{AbaBS97,CauH06}.
An exact comparison of the results for the four-spinon form-factors (and higher spinons) from both the methods is made difficult since these two methods work in different bases and little is known about the mapping between these two bases.
The mapping from a subspace of fixed $n_h$ and $\ho{n}$ in the ABA framework to its equivalent in $q$-VOA is an endomorphism that mixes non-trivially all the $P(n_h,\ho{n})$ eigenvectors.
In the case of two-spinon form-factors, it becomes trivial since we had one-dimensional subspace $P(2,0)=1$ as it can be seen from \cref{tab:exc_DL}.
Therefore it is possible to compare the results of the two-spinon form-factor exactly and we have managed to do the same in \cite{KitK19}.
As we proceed to the higher spinon sectors $n_h>2$, we need to account for the mixing of the bases and hence an exact comparison is rendered difficult without the precise knowledge of this mapping of bases is not possible.
However, we can still compare the prefactors in our result \eqref{2sp_ff_result_concl} with that of \cite{AbaBS97,CauH06} since we have seen here that the prefactors are largely formed by the infinite Cauchy matrix which is dominated by contribution of the ground state and real excited state roots alone and does not depend upon the choice of basis.
From the integral representation for the Barnes-G function \eqref{barnes_log_Gfn_intrep} we can see that
\begin{align}
	\frac{1}{G^8(\frac{1}{2})}
	\prod_{a\neq b}
	\frac{%
	G(\frac{\hle_a-\hle_b}{2i})
	G(1+\frac{\hle_a-\hle_b}{2i})
	}{%
	G(\frac{1}{2}+\frac{\hle_a-\hle_b}{2i})
	G(\frac{3}{2}+\frac{\hle_a-\hle_b}{2i})
	}
	=
	\exp\left(-\sum_{a<b}I(\hle_a-\hle_b)\right)	
\end{align}
where the special function $I$ is same as in \cref{2sp_ff_res_int_form_concl}.
This shows that the prefactors of the results from the both methods are compatible for the four-spinon form-factor.
We also believe that if a close-form representation for the four-spinon form-factor were to be obtained, this can also shed light on problem of the mapping between the bases in algebraic Bethe ansatz and $q$-vertex operator algebra frameworks. 
\addsec{Perspectives for the future work}
The method presented in this thesis provides a promising approach for the computation of the form-factors of integrable quantum spin chains models through the algebraic Bethe ansatz.
It inherits the wider applicability from the algebraic Bethe ansatz and we hope that it can extended to the broad range of quantum integrable systems.
\par
The closest relative of the XXX model is undoubtedly its anisotropic XXZ model which is divided into two principle regimes: massive $\Delta>1$ and massless $|\Delta|<1$.
Let us recall from \cref{gau_ex_sol,gau_ex_sol_II} that the Gaudin extraction for the columns of real roots always results in an expression that we can express with the ground state density function with some shift.
The same could be said about the Gaudin extraction in anisotropic XXZ model.
\\
In the massive regime it would means that the Gaudin extraction gives the density terms which we can express in terms of elliptic functions.
It is known that elliptic functions also forms the Cauchy matrices similar to the hyperbolic version that we saw in \cref{cau_hyper} and it can be expected that we can obtain an elliptic version of the Cauchy-Vandermonde with the arguments similar to \cref{sec:cau_van_mat_hyper}.
This could give representation for the thermodynamic form-factors in massive regime $\Delta>1$ obtained through the algebraic Bethe ansatz which we can also compare with the form-factors \cite{Cas20} from $q$-vertex operator algebra \cite{JimM95}.
\\
The massless XXZ model is more interesting since it falls in quantum critical regime.
In this case, we believe that for the non-commensurate values of $\gamma\notin\pi\Qcal$ in the range $0<\gamma<\frac{\pi}{2}$, the longitudinal form-factors is largely described by excitations with one \emph{negative parity} root.
As we have remarked in \cref{item:xxz_massless_rem_neg_par} of \cref{sec:xxz_spectre}, these excitations are immediate equivalent of the first descendants of triplets in the XXX case.
This computation would again involve determinant representations of hyperbolic Cauchy matrices, however unlike for the XXX model, the period of these functions will be rescaled by the factor of $\gamma$, as it can be easily seen from the density function \eqref{lieb_den_sol}.
The commensurability of the period with $\pi$ plays an important role here, which we have already mentioned. In the non-commensurate case, we would expect that the results for these XXZ form-factors from our method can be expressed in terms of $q$-Barnes G functions. It would be an interesting prospect to compare it with the results of \cite{CauKSW12} obtained using the $q$-vertex operator method \cite{JimM95}.
\par
The $\mathfrak{su}_2$ quantum spin-$\frac{1}{2}$ chain is the most fundamental example quantum integrable models.
Through the construction of solutions of the Yang-Baxter equation and with the different representation of the $\Rm$-matrix, a large collection of quantum integrable models have been revealed.
However, most of the analysis of these models is in the early stages.
The quantum spin chains was also first example of a quantum integrable model where determinant representations and the resolution of the inverse problem were found, expanding them to other quantum integrable models is still an ongoing process which has culminated only for a handful of models such as the higher spin chains \cite{MaiT00,CasM07}.
In addition to the existence of scalar product and quantum inverse scattering relations, a meaningful application of our method to any other type of model would require us to access it on an additional criteria.
Since our computation relies heavily on the fact that the matrix of density terms obtained from the Gaudin extraction contains an infinite Cauchy matrix, we need to first see whether this models admits the density function which forms a Cauchy matrix.
Although we have only dealt with the Cauchy(-Vandermonde) matrix in the rational and hyperbolic parametrisation, an elliptic version of it is also known to exist and it would be reasonable to expect that the essential properties, such as duality, can also be extended to the elliptic case.
This covers, in principle, the generalisability of our method to the XXZ model in massive regime $\Delta>1$, where matrix of density function is written would involve elliptic functions.
However, based on this observation, a more nuanced stance can also be taken with respect to the generalisability and the requirement of the Cauchy structure. 
A weaker version of this criteria can be framed by rewording it in terms of the essential properties of the Cauchy matrix, which are summarised briefly in the following points:
\begin{enumerate}
\item The inverse of a matrix that is composed by the ground state density function can be expressed with a diagonal dressing that is similar to \cref{hcv_dual_inv_dressing_append-sec}. In addition to this, it would be advantageous if it also has a duality that is similar to the Cauchy-Vandermonde matrix. 
\item The infinite determinant can be computed through infinite product form or integral representations of the special functions.
\end{enumerate}
These are the two essential properties that we need in our computations, hence the it would be interesting to see if a more general class of integrable systems can be found with these properties and to see what class of special functions these properties belong.
\par
The study quantum integrable models at finite temperature as well the out-of-equilibrium integrable models has also emerged in recent years.
Form-factors based approach plays an important role in these computations \cite{CasDY16,DugGK13}.
It could be interesting to see whether our approach can be used in these settings.
Although this would require us to move away from the low-lying excitations, the compatibility of the Destri-Lowenstein picture and string picture can help us make this transition.
However, the reduced determinant in this far from ground state limit will no longer be finite size determinants.
A closed form expression for the $\Phifn$ integrals would be thus desired to make such an approach efficient.
\par
Finally, it is important to see whether the generalised condensation property that we apply for the meromorphic functions can be proven rigorously in the non-compact Fermi distributions and whether the finite-size corrections could be estimated and we believe an approach with non-linear integral equation can be used to do so.
However, we do realise that it would still us require to tame the extreme cases with the poles lying at the edges of Fermi-distribution, in which case the problem gets deeply interlinked with the problem of finite size corrections to the Cauchy determinants.
A more robust approach for computing the $\Phifn$ function integrals would be helpful in tackling the latter problem.
\clearpage{}%

\addpart{Appendices}
\begin{appendices}
\clearpage{}%
\chapter{Special functions}
\label{chap:spl_fns}
\section{Euler's \texorpdfstring{$\Gamma$}{Gamma} function}
\label{sec:eu_gamma_fn}
The Gamma $\Gamma$ function can be defined with its Weierstrass infinite product form \cite{Bat81}:
\begin{align}
	\Gamma(z)
	=
	\frac{e^{\gamma z}}{z}
	\prod_{n=1}^{\infty}
	\frac{ne^{\frac{z}{n}}}{n+z}
	.
	\label{gamma_fn_wei_form}
\end{align}
Here $\gamma$ denotes the Euler-Mascheroni constant defined as the limit of the following sequence which takes the approximate value of $\gamma\simeq 0.577$.
\begin{align}
	\gamma=\lim_{n\to\infty}(\sum_{p=1}^{n}\frac{1}{n}-\log n)
	.
	\label{eul_masch_const}
\end{align}
The Euler's $\Gamma$ function has the properties:
\begin{subequations}
\begin{flalign}
&\textbf{Initial value:} & \Gamma(1)&=1,	&&\hfill&
\label{prop_gamma_fun_init}
\\
&\textbf{Recurrence:} & \Gamma(z+1)&=z\Gamma(z). &&\hfill&
\label{prop_gamma_fun_recur}
\end{flalign}
\label{props_gamma_fun}
\end{subequations}
The $\Gamma(z)$ has simple poles in $z\in -i\Nset=\set{0,-i,-2i,\ldots}$, the residue at any of these poles can be obtained using the property \eqref{prop_gamma_fun_recur}.
\minisec{Integral representation}
Its logarithm admits the integral representation \cite{GraRJ07,Nie06}:
\begin{align}
	\log\Gamma(z)
	=
	\int_{0}^{\infty}
	\frac{dt}{z}
	\left\lbrace
	(z-1)e^{-t}
	+
	\frac{e^{-zt}-e^{-t}}{1-e^{-t}}
	\right\rbrace
	,
	&
	&(\Re z>0)
	.
	\label{log_gamma_intrep}
\end{align}
\minisec{Asymptotic form}
The asymptotic form of the $\log\Gamma$ function for the large values of $|z|$ is given by, \cite{GraRJ07,MagO48i}
\begin{multline}
	\log\Gamma(z)
	=
	z\log z -z -\frac{\log z}{2}
	+\log \sqrt{2\pi}
	\\
	+\frac{1}{2}
	\sum_{n=1}^{\infty}
	\frac{n}{(n+1)(n+2)}
	\sum_{m=1}^{\infty}
	\frac{1}{(z+m)^{n+1}},
	\qquad
	(|\arg z|<\pi).
	\label{stitling_gamma_fn_asym_form}
\end{multline}
\minisec{Relationship with the trigonometric and hyperbolic functions}
The Gamma function is related with the trigonometric or the hyperbolic sine and cosine function through the identities:
\begin{subequations}
\begin{align}
	\sin\pi z&=	
	\frac{\pi z}{\Gamma(1+z)\Gamma(1-z)}
	&
	\cos\pi z&=
	\frac{\pi}{\Gamma(\frac{1}{2}+z)\Gamma(\frac{1}{2}-z)}
	;
	\label{id_Gamma_fns_trig}
	\\
	\sinh\pi z&=	
	\frac{\pi z}{\Gamma(1+iz)\Gamma(1-iz)}
	&
	\cosh\pi z&=
	\frac{\pi}{\Gamma(\frac{1}{2}+iz)\Gamma(\frac{1}{2}-iz)}
	.
	\label{id_Gamma_fns_hyper}
\end{align} 
\label{id_Gamma_fns_trig_hyper}
\end{subequations}
\minisec{Convergence of infinite rational products}
A comparison of any infinite product involving rational terms with the Weierstrass form \eqref{gamma_fn_wei_form} gives an important criteria for the absolute convergence of such products.
\begin{lem}[\cite{WhiW02,Bat81}]
\label{lem:abs_conv_rat_infprd}
Let $\bm\alpha\subset\Cset$ and $\bm\beta\subset\Cset$ denote a set of complex variables with cardinalities $p=n_{\bm\alpha}$ and $q=n_{\bm\beta}$ respectively.
The necessary and sufficient condition for infinite product
\begin{align}
	\prod_{n=1}^{\infty}
	\frac{\bmprod(n-\bm\alpha)}{\bmprod(n-\bm\beta)}
	=
	\prod_{n=1}^{\infty}
	\frac{(n-\alpha_1)\cdots (n-\alpha_p)}{(n-\beta_1)\cdots (n-\beta_q)}
\end{align}
to be absolutely convergent is $p=q$ and
\begin{align}
	\bmsum\alpha-\bmsum\beta=0.
\end{align}
In which case, it converges to the ratio of Gamma functions:
\begin{align}
	\prod_{n=1}^{\infty}
	\frac{\bmprod(n-\bm\alpha)}{\bmprod(n-\bm\beta)}
	=
	\frac{%
	\bmprod\Gamma(1-\bm\beta)
	}{%
	\bmprod\Gamma(1-\bm\alpha)
	}
	.
\end{align}
\end{lem}
\section{Digamma function \texorpdfstring{$\dgamma$}{\textbackslash psi} : Logarithmic derivative of the \texorpdfstring{$\Gamma$}{Gamma} function}
\label{sec:dgamma}
The logarithmic derivative of the $\Gamma$ function is denoted as $\dgamma(z)=\log\Gamma'(z)$.
It has the recurrence property:
\begin{flalign}
&\textbf{recurrence:}	& \dgamma(z+1)&=\frac{1}{z}+\dgamma(z).	&& \hfill	&
\end{flalign}
Some of the particular values of the digamma function are as follows:
\begin{equation}
\begin{aligned}
	\dgamma(1)&=-\gamma,
	&
	&\quad&
	\dgamma\left(\tfrac{1}{4}\right)&=-\gamma-\tfrac{\pi}{2}-3\log 2
	,
	\\
	\dgamma\left(\tfrac{1}{2}\right)&=-\gamma-2\log 2
	,
	&
	&\quad&
	\dgamma\left(\tfrac{3}{4}\right)&=-\gamma+\tfrac{\pi}{2}-3\log 2
	.
\end{aligned}
\end{equation}
\minisec{Series representation}
The digamma function can be represented as infinite series of the rational terms \cite{GraRJ07}:
\begin{align}
	\dgamma(z)&=
	-\gamma+\sum_{n=0}^{\infty}
	\left\lbrace
	\frac{1}{n+1}-\frac{1}{z+n}
	\right\rbrace
	.
\end{align}
\minisec{Integral representation}	
From the \cref{log_gamma_intrep} we can write the integral representation due to Gauss:
\begin{subequations}
\begin{align}
	\dgamma(z)
	&=
	\int_{0}^{\infty}
	\left\lbrace
	\frac{1}{t}e^{-t}-\frac{e^{-zt}}{1-e^{-t}}
	\right\rbrace
	dt
	,
	&
	&(\Re z>0)
	.
	\label{dgamma_intrep_gauss}
\shortintertext{This is equivalent to the integral representation:}
	\dgamma(z)
	&=
	-\gamma
	+
	\int_{0}^{\infty}
	\frac{e^{-t}-e^{-zt}}{1-e^{-t}}
	dt
	,
	&
	&(\Re z>0)
	.
	\label{dgamma_intrep_eu-masch}
\end{align}
\label{dgamma_intrep_both}
\end{subequations}
\section{Barnes \texorpdfstring{$G$}{G}-function}
\label{sec:barnes-Gfn}
\Textcite{Bar99} generalised the Gamma function to a class of transcendental functions $\Gamma_n$ called multiple Gamma functions with the recursive property \eqref{prop_gamma_fun_recur}.
\begin{align}
	\Gamma_{n+1}(z+1)
	=
	\frac{\Gamma_{n+1}(z)}{\Gamma_{n}(z)}
	.
\end{align}
In this notation $\Gamma_{0}$ is the rational identity function $\Gamma_0(z)=z$ and $\Gamma_{1}$ is the Euler's Gamma function from \eqref{sec:eu_gamma_fn}.
For each $n\geq 2$ we obtain a higher transcendent of the $\Gamma_{n-1}$. Out of these variants, the double Gamma function $\Gamma_{2}$ is of our particular interest.
It is customary and also convenient to write down its reciprocate $G(z)=\Gamma_2^{-1}(z)$ which is called the \emph{Barnes} $G$-function.
It has the following properties:
\begin{subequations}
\begin{flalign}
&\textbf{Initial value:}	& G(1)&=1,	&& \hfill &
\label{barnes_Gfn_init_prop}
\\
&\textbf{Recurrence:} &	G(z+1)&=\Gamma(z)G(z).	&&	\hfill &
\label{barnes_Gfn_recur_prop}
\end{flalign}
\end{subequations}
\minisec{Weierstrass infinite product form}
The Barnes $G$-function can be represented in the Weierstrass infinite product form \cite{Bar99}:
\begin{align}
	G(z+1)
	=
	(2\pi)^{\frac{z}{2}}
	e^{-\frac{z(z-1)}{2}-\frac{\gamma z^2}{2}}
	\prod_{n=1}^{\infty}
	\left\lbrace
	\frac{%
	\Gamma(n)
	}{%
	\Gamma(z+n)
	}
	e^{z\psi(n)+\frac{z^2}{2}\psi'(n)}
	\right\rbrace
	.	
	\label{barnes_Gfn_wei_form}
\end{align}
The $\dgamma$ and $\dgamma'$ represents the digamma function from the \cref{sec:dgamma} above and its derivative respectively.
\minisec{Integral representation}
Its logarithm admits the following integral representation \cite{ChoS09,Vig79}:
\begin{align}
	\log G(z+1)
	=
	\int_{0}^{\infty}
	\frac{e^{-t}}{t}
	\frac{%
	e^{-zt}
	+
	zt+\frac{z^2t^2}{2}-1
	}{%
	(1-e^{-t})^2
	}
	dt
	-
	(1+\gamma)
	\frac{z^2}{2}
	+
	\frac{3}{2}\log\pi
	,
	&
	&(\Re z>-1)
	.
	\label{barnes_log_Gfn_intrep}
\end{align}
\minisec{Particular value}
\Textcite{Bar99} also gave the particular value:
\begin{align}
	G\left(\frac{1}{2}\right)
	=
	2^{\frac{1}{24}}
	\cdot
	\pi^{-\frac{1}{4}}
	\cdot
	e^{\frac{1}{8}}
	\cdot
	A^{-\frac{3}{2}}
	\label{barnes_Gfn_1-2}
\end{align}
where $A$ is called the Glaisher-Kinkeline \cite{Gla77} constant given by the integral:
\begin{align}
	A=
	2^{\frac{7}{36}}
	\pi^{-\frac{1}{6}}
	\exp\left\lbrace
	\frac{1}{3}
	+
	\frac{2}{3}\int_{0}^{\frac{1}{2}}\Gamma(t+1)dt%
	\right\rbrace
	.
	\label{gla_kin_const_A}
\end{align}
\minisec{Asymptotic form}
The Stirling approximation formula for logarithm of the Barnes $G$-function for large values of $x\in\Rset$ and a complex $\alpha\in\Cset$ \cite{ChoS09} is given by,
\begin{multline}
	\log G(x+\alpha+1)
	=
	\left\lbrace
	\frac{x+\alpha}{2}
	\right\rbrace
	\log(2\pi)
	-\log A
	-\frac{3x^2}{4}
	-\alpha x
	+\frac{1}{12}
	\\
	+
	\left\lbrace
	\frac{(x+\alpha)^2}{2}
	-\frac{1}{12}
	\right\rbrace
	\log x
	+O(x^{-1})
	.
	\label{barnes_Gfn_assym_form}
\end{multline}
\minisec{Convergence of infinite product involving Gamma functions}
Similar to the \cref{lem:abs_conv_rat_infprd} we can prove the following result for the absolute convergence of an infinite product involving the $\Gamma$ function:
\begin{lem}
\label{lem:infprod_gamma_fn_crit}
Given the sets $\bm\alpha\subset\Cset$ and $\bm\beta\subset\Cset$ of the complex variables with cardinalities $p=n_{\bm\alpha}$ and $q=n_{\bm\beta}$ respectively. 
The necessary and sufficient condition for an infinite product
\begin{align}
	\prod_{n=1}^{\infty}
	\frac{%
	\bmprod\Gamma(n-\bm\alpha)
	}{%
	\bmprod\Gamma(n-\bm\beta)
	}
	=
	\prod_{n=1}^{\infty}
	\frac{%
	\Gamma(n-\alpha_1)\cdots \Gamma(n-\alpha_p)
	}{%
	\Gamma(n-\beta_1)\cdots \Gamma(n-\beta_q)
	}
\end{align}
to converge is $p=q$,
\begin{subequations}
\begin{align}
\bmsum \bm{\alpha}-\bmsum \bm{\beta}&=0,
\shortintertext{and}
\bmsum\bm\alpha^2-\bmsum\bm\beta^2&=0.
\end{align}
\label{infprod_gamma_fn_crit}
\end{subequations}
In which case it converges to the following rational form in $G$-functions
\begin{align}
	\prod_{n=1}^{\infty}
	\frac{%
	\bmprod\Gamma(n-\bm\alpha)
	}{%
	\bmprod\Gamma(n-\bm\beta)
	}
	=
	\frac{%
	\bmprod G(1-\bm\beta)
	}{%
	\bmprod G(1-\bm\alpha)
	}
	.
	\label{infprod_gamma_fn_barnes_Gfn}
\end{align}
\end{lem}
\begin{proof}
From the Stirling's approximation \eqref{stitling_gamma_fn_asym_form} we find that for large $n$ the general term is comparable to
\begin{align}
	\frac{%
	\bmprod\Gamma(n-\bm\beta)
	}{%
	\bmprod\Gamma(n-\bm\alpha)
	}
	\sim
	\left(\frac{2\pi}{\sqrt{n}}\right)^{q-p}
	\left(\frac{n}{e}\right)^{(q-p)n}
	n^{\bm\Delta}
	\cdot
	\left(1-\frac{\bm\Delta}{2n}+\frac{\bm{\Delta^2}}{n}\right)
	.
\end{align}
where $\bm\Delta=\bmsum\bm\alpha-\bmsum\bm\beta$ and $\bm{\Delta^2}=\bmsum\bm\alpha^2-\bmsum\bm\beta^2$.
This gives us the necessary condition for the absolute convergence $p=q$, $\bm\Delta=0$ and $\bm{\Delta^2}=0$.
It is also sufficient condition as we can insert the exponential terms to obtain
\begin{align}
	\prod_{n=1}^{\infty}
	\frac{%
	\bmprod\Gamma(n-\bm\alpha)
	}{%
	\bmprod\Gamma(n-\bm\beta)
	}
	=
	\prod_{n=1}^{\infty}
	\frac{%
	\bmprod_{\bm\alpha}\left\lbrace\Gamma(n-\bm\alpha)e^{-\bm\alpha\psi(n)-\frac{\bm\alpha^2}{2}\psi'(n)}\right\rbrace
	}{%
	\bmprod_{\bm\beta}\left\lbrace\Gamma(n-\bm\beta)e^{-\bm\beta\psi(n)-\frac{\bm\beta^2}{2}\psi'(n)}\right\rbrace
	}
	=
	\frac{%
	\bmprod G(1-\bm\beta)
	}{%
	\bmprod G(1-\bm\alpha)
	}
	.
\end{align}
\end{proof}
\begin{rem}
The condition $p=q$ can be relaxed provided we add in the denominator extra $p-q$-fold $\Gamma(n)$ terms:
\begin{align}
	\frac{%
	\bmprod G(1-\bm\beta)
	}{%
	\bmprod G(1-\bm\alpha)
	}	
	=
	\prod_{n=1}^{\infty}
	\frac{1}{\Gamma^{p-q}(n)}
	\frac{%
	\bmprod\Gamma(n-\bm\alpha)
	}{%
	\bmprod\Gamma(n-\bm\beta)
	}
	.
\end{align}
In other words, when $p\neq q$, we can simply add $|p-q|$ fold zeroes to obtain the sets with equal cardinalities, without altering the condition \eqref{infprod_gamma_fn_crit}.
\end{rem}
\clearpage{}%
\clearpage{}%
\chapter{Integral equations and convolutions of the density terms}
\label{chap:den_int_aux}%
Here we will study generalised versions of integral equation for the function $\rden_{\kappa}$:
\index{exc@\textbf{Excitations}!condn@\textbf{- condensation}!den_generic@$\rden_\kappa(\cdot,\alpha)$: generic density function with shift $\alpha$}%
\begin{align}
	\rden_{\kappa}(\la)+
	\int_{\Rset}K(\la-\tau)\rden_{\kappa}(\tau)d\tau&=
	K_{\kappa}(\la)
	\label{shft_lieb_sclf}
\end{align}
where the parameter $\la$ or the integrated variable $\tau$ (or both) are allowed to be non-real.
\section{Shifted density terms}
\label{sec:shftd_den_terms_append}
Let us define the shifted density term as a bivariate function $\rden_\kappa(\la,\mu)$ where $\la\in\Rset$ and $\mu\in\Cset$ satisfying the integral equation:
\begin{align}
	\rden_{\kappa}(\la,\mu)+
	\int_{\Rset}K(\la-\tau)\rden_{\kappa}(\tau,\mu)
    d\tau
	&=
	K_{\kappa}(\la-\mu)
	.
	\label{int_eq_shft_scld}
\end{align}
We can see that depending upon the value of the imaginary part $\Im\mu$, the function $K_{\kappa}$ has differing configuration of its poles.
Its Fourier transform is simply given by a shift as long as we confine that the parameter $\mu$ to the $|\Im\mu|<\frac{1}{\kappa}$, we see that there is one pole on the both sides of the real line. For the values of $\mu$ outside this strip $|\Im\mu|>\frac{1}{\kappa}$ we see that one of the pole crosses the real line to move to the other side.
This means that the function $K_\kappa(\la-\mu)$ is holomorphic in one of the halves of the complex plane while it has two simple poles in the other half.
This leads to the Fourier transforms\footnote{Here $I_{t>0}$ and $I_{t<0}$ denotes the characteristic functions with support on $\set{t\in\Rset|t>0}$ and $\set{t\in\Rset|t<0}$ resp. They  are related to the Heaviside step function.}:
\index{misc@\textbf{Miscellaneous functions}!Lieb K gen@$K_\kappa$: generic Lieb kernel}%
\begin{align}
	\what{K_{\kappa}}(t,\mu)=\begin{dcases}
	e^{-|\frac{t}{\kappa}|}\, e^{-i\mu t};& \kappa\left|\Im\mu\right|<1,%
	\\%
	-2I_{t>0}\sinh \left(\frac{t}{\kappa}\right)\, e^{-i\mu t};	&	\kappa\Im\mu<-1,%
	\\%
	2I_{t<0}\sinh \left(\frac{t}{\kappa}\right)\, e^{-i\mu t};	&	\kappa\Im\mu>1.%
	\end{dcases}
	\label{ft_lieb_kernel_shft_scl}
\end{align}
Note that this shift does not affect the Lieb kernel $K$ on the left-hand side of the integral equation \eqref{int_eq_shft_scld}.
Thus we can see that the Fourier transform of the function $\rden_\kappa(\la,\mu)$ can be written as
\begin{align}
	\what{\rden_{\alpha}}(t,\mu)
	&=
	\frac{\what{K_{\alpha}}(t,\mu)}{1+e^{-|t|}}
	.
\end{align}
Therefore it is also divided into the three branches according to value of $\Im\alpha$ where the Fourier transform of the $\rden_{\kappa}(\nu,\alpha)$ assumes different forms.
In particular for values of the scaling factor that interests us $\kappa=1$ and $\kappa=2$, we get
\begin{subequations}
\begin{align}
	\what{\rden_{2}}(t,\mu)
	&=
\begin{dcases}
	\frac{e^{-i\mu t}}{2\cosh\frac{t}{2}};	&	|\Im\mu|<\frac{1}{2},%
	\\%
	I_{t<0}\frac{2\sinh {\frac{t}{2}}\,e^{-i\mu t}}{1+e^{t}};		&	\Im\mu>\frac{1}{2},%
	\\
	-I_{t>0}\frac{2\sinh {\frac{t}{2}}\,e^{-i\mu t}}{1+e^{-t}};		&	\Im\mu<-\frac{1}{2}.%
\end{dcases}
\intertext{And}
	\what{\rden_{1}}(t,\mu)
	&=
\begin{dcases}
	\frac{e^{-|t|}e^{-i\mu t}}{1+e^{-|t|}};	&	|\Im\mu|<1,%
	\\%
	I_{t<0}(1-e^{-t})e^{-i\mu t};		&	\Im\mu>1,%
	\\
	I_{t>0}(1-e^{t})e^{-i\mu t};		&	\Im\mu<-1.%
\end{dcases}
\label{ft_shft_scl_rhden}
\end{align}
\end{subequations}
The solution for $\rden_2$ in the central strip of analyticity $|\Im\alpha|<\frac{1}{2}$ can be expressed in terms of the hyperbolic function:
\begin{align}
	\rden_{2}(\nu,\alpha)&=
	\rden_{2}(\nu-\alpha)
	=
	\frac{1}{2\cosh\pi(\nu-\alpha)},
	&
	-\frac{1}{2}<\Im\alpha<\frac{1}{2}
	.
	\label{lieb_den_shft_inside}
\end{align}	
We also note that inside this strip it only depends on the difference $\rden_{2}(\la,\mu)=\rden_{2}(\la-\mu)$ where the latter is an analytic continuation of $\rden_{2}(\la)$ for $\la\in\Rset$ to the region $|\Im\la|<\frac{1}{2}$.
Outside this central strip, we can see from the integral representation \eqref{dgamma_intrep_both} that $\rden_2(\nu,\mu)$ can be expressed as a sum of digamma functions $\dgamma$:
\begin{align}
	\rden_{2}(\la,\mu)&=
	\frac{1}{4\pi}
	\left\lbrace
	\dgamma\left(-\frac{1}{4}-\frac{\la-\mu}{2i\sigma}\right)
	-
	2\dgamma\left(\frac{1}{4}-\frac{\la-\mu}{2i\sigma}\right)
	+
	\dgamma\left(\frac{3}{4}-\frac{\la-\mu}{2i\sigma}\right)
	\right\rbrace
	,
	& \sigma\Im\mu>\frac{1}{2}
	,\,
	\sigma=\pm 1.
	\label{lieb_den_shftd_outside}
\end{align}
\par
For the solutions of $\rden_{1}(\la,\mu)$ \eqref{ft_shft_scl_rhden} this is reverse. We find that in the central strip it can only be expressed in terms of digamma functions
\begin{align}
	\rden_{1}(\la,\mu)&=
	\frac{1}{4\pi}
	\sum_{\sigma=\pm 1}
	\left\lbrace
	\dgamma\left(\frac{1}{2}+\frac{\la-\mu}{2i\sigma}\right)
	-
	\dgamma\left(1+\frac{\la-\mu}{2i\sigma}\right)
	\right\rbrace
	\label{lieb_hle_den_shftd_inside}
\end{align}
where it depends on the difference $\rden_1(\la,\mu)=\rden_1(\la-\mu)$ and the latter $\rden_1(\la)$ ($|\Im\la|<1$) is the analytic continuation of $\rden_2(\la)$ ($\la\in\Rset$).
Whereas in the region outside this central strip $|\Im\mu|>1$, we find that it admits a rational form
\begin{align}
	\rden_{1}(\la,\mu)&=
	\frac{1}{2\pi i}
	t\left(\sigma(\la-\mu)\right)
	=
	\frac{1}{2\pi}\frac{1}{(\tau-\mu)(\tau-\mu+i\sigma)}
	,	&\sigma \Im\mu>1.
	\label{lieb_hle_den_shftd_outside}
\end{align}
\minisec{Effect of moving the contour of integration}
Let us now consider the case where $\la=\nu+i\alpha$ in $\rden_\kappa(\la,\mu)$ is complex.
It satisfies the integral equation
\begin{align}
	\rden_\kappa(\la,\mu)
	+
	\int_{\Rset+i\alpha}
	K(\la-\tau)
	\rden_\kappa(\tau,\mu)
	d\tau
	=
	K_\kappa(\la-\mu)
	.
	\label{int_eq_shft_scld_cmplx}
\end{align}
With a change of parameter of integration we can see that this integral equation is identical to that of $\rden_\kappa(\nu,\mu-i\alpha)$ written according to \cref{int_eq_shft_scld}:
\begin{align}
	\rden(\nu,\mu-i\alpha)
	+
	\int_{\Rset}
	K(\nu-\tau)
	\rden_\kappa(\tau,\mu-i\alpha)
    d\tau
	=
	K_\kappa(\nu-\mu+i\gamma)
	=
	K_\kappa(\la-\mu)
	.
\end{align}
Thus we see that it is sufficient to consider $\la\in\Rset$ as we did in \cref{int_eq_shft_scld} and the more general case is nothing but shifting of the contour of integration.
\section{Complexified density function}
\label{sec:den_complex}
Let us now consider another version of the $\rden_\kappa$ function where the parameter $\la$ is allowed to complex such that we have
\begin{align}
	\rden_{\kappa}(\la)+
	\int_{\Rset}K(\la-\tau)\rden_{\kappa}(\tau)d\tau
	&=
	K_{\kappa}(\la)
	\qquad
	\text{with}
	\quad
	\la\in \Cset
	.
	\label{lieb_inteq_complex}
\end{align}
Note that there is clear distinction between \cref{lieb_inteq_complex,int_eq_shft_scld_cmplx} that here the integral is taken over the real line.
This means that the equation \eqref{lieb_inteq_complex} has ceased to an integral equation for $\la\notin \Cset$ and it can be simply solved by calculating the convolution of the kernel $K$ with the density function $\rden_\kappa$.
A crucial point to note here is that kernel $K$ is now shifted and we must pay attention to the its poles.
We will now compute these convolutions for the case of $\kappa=2$ and $\kappa=1$.
\subsection{For \texorpdfstring{$\kappa=2$}{\textbackslash kappa = 2}}
\minisec{\texorpdfstring{When $|\Im\la|<1$ (Close-pair region)}{Close-pair region} :}
Since the poles of $K(\tau-\la)$ for $|\Im\la|<1$ lies on each side of the real line, we get the Fourier transform in this case from \cref{ft_lieb_kernel_shft_scl,ft_shft_scl_rhden} which is factorised to write the following integrals:
\begin{align}
	\int_{\Rset}K(\la-\tau)\rden_2(\tau)d\tau
	&=
	\frac{1}{2\pi}
	\int_{0}^{\infty}
	e^{-\frac{t}{2}}
	(e^{i\la t}+e^{-i\la t})
	dt
	-
	\frac{1}{2\pi}
	\int_0^\infty
	\frac{e^{-\frac{t}{2}}(e^{i\la t}+e^{-i\la t})}{1+e^{-t}}
	dt
	\label{conv_K_shft_clp_rden_2_FT}
\end{align}
which gives 
\begin{align}
	\int_{\Rset}K(\la-\tau)\rden_2(\tau)d\tau
	&=
	K_2(\la)
	-
	\frac{1}{2\cosh\pi\la}
	.
	\label{conv_K_shft_clp_rden_2_sol}
\end{align}
Substituting this in \cref{lieb_inteq_complex} for $\kappa=2$ we find that
\begin{align}
	\rden_{2}(\la)=\frac{1}{2\cosh\pi\la}
	,
	\quad
	|\Im\la|<1
	\label{lieb_den_clp}
\end{align}
In other words, we only found that the ground state density function can be analytically continued to the entire $|\Im\la|<1$ that we associate in the Destri-Lowenstein picture with the close-pairs.
\minisec{\texorpdfstring{For $|\Im\la|>1$ (Wide-pair region)}{Wide-pair region} :}
In the outside region $|\Im\la|>1$, the kernel $K(\tau-\la)$ has both its poles lying on same side of the real line.
This changes its Fourier transform \eqref{ft_lieb_kernel_shft_scl}, which is reflected in the computation of its convolution with the density $\rden_2$ as seen from the following expressions:
\begin{subequations}
\begin{align}
	\int_{\Rset}K(\la-\tau)\rden_2(\tau)d\tau
	&=
	\frac{1}{\pi}
	\int_{0}^{\infty}
	e^{i\la \sigma t}
	\sinh\tfrac{t}{2}
	dt
	,
	&
	\text{where } 
	~
	\sigma\Im\la>1
	~(\sigma=\pm 1)
\shortintertext{which gives,}
	\int_{\Rset}K(\la-\tau)\rden_2(\tau)d\tau
	&=
	\frac{1}{2\pi}\frac{1}{\la^2+\frac{1}{4}}
	,
	&
	|\Im\la|>1
	.
\end{align}
\end{subequations}
This means that $\rden_2(\la)$ vanishes in the region $|\Im\la|>1$ that is associated with the wide-pairs in the Destri-Lowenstein picture,
\begin{align}
	\rden_{2}(\la)&=0
	,
	\quad
	|\Im\la|>1
	.
	\label{lieb_den_wdp}
\end{align}
\subsection{For \texorpdfstring{$\kappa=1$}{\textbackslash kappa=1}}
\minisec{\texorpdfstring{When $|\Im\la|<1$ (Close-pair region)}{Close-pair region} :}
Similar to the $\kappa=2$ case, we get from \cref{ft_shft_scl_rhden} the Fourier transform of the convolution
\begin{align}
	\int_{\Rset}K(\la-\tau)\rden_1(\tau)d\tau
	&=
	\frac{1}{2\pi}
	\int_{0}^{\infty}
	e^{-t}
	(e^{i\la t}+e^{-i\la t})
	dt
	-
	\frac{1}{2\pi}
	\int_0^\infty
	\frac{e^{-t}(e^{i\la t}+e^{-i\la t})}{1+e^{-t}}
	dt
	\label{conv_K_shft_clp_rden_hle_2_FT}
\end{align}
which gives,
\begin{align}
	\int_{\Rset}K(\la-\tau)\rden_2(\tau)d\tau
	&=
	K(\la)
	-
	\rden_h(\la)
	.
	\label{conv_K_shft_clp_rden_hle_2_sol}
\end{align}
Substituting this in \cref{lieb_inteq_complex} for $\kappa=2$ we find that
\begin{align}
	\rden_{2}(\la)=\rden_h(\la)
	,
	\quad
	|\Im\la|<1
	\label{lieb_den_hle_clp}
\end{align}
\minisec{\texorpdfstring{For $|\Im\la|>1$ (Wide-pair region)}{Wide-pair region} :}
From \cref{ft_shft_scl_rhden,ft_lieb_kernel_shft_scl} we get 
\begin{subequations}
\begin{align}
	\int_{\Rset}K(\la-\tau)\rden_1(\tau)d\tau
	&=
	\frac{1}{\pi}
	\int_{0}^{\infty}
	(e^{\sigma t}-1)
	e^{i\la \sigma t}
	dt
	,
	&
	\text{where } 
	~
	\sigma\Im\la>1
	~(\sigma=\pm 1)
\shortintertext{which gives,}
	\int_{\Rset}K(\la-\tau)\rden_2(\tau)d\tau
	&=
	\frac{1}{2\pi i}t(-\sigma\la)
	=
	\frac{1}{2\pi}\frac{1}{\la(\la-i\sigma)}
	,
	&
	|\Im\la|>1
	.
\end{align}
\end{subequations}
Substituting it into the integral equation \eqref{lieb_inteq_complex} for $\kappa=1$ gives us,
\begin{align}
	\rden_{1}(\la)&=\frac{1}{2\pi i}t(\sigma\la)
	=
	\frac{1}{2\pi}\frac{1}{\la(\la+i\sigma)}
	,
	\quad
	|\Im\la|>1
	.
	\label{rden_hle_wdp}
\end{align}
\clearpage{}%
\clearpage{}%
\chapter{Matrices, determinants and their extractions}
\label{chap:mat_det_extn}%
\section{Reduction of matrices up-to their determinants}
Here we present the result of the following lemma which we used to reduce the order of determinant.
\begin{lem}
\label{lem:mat_det_red}
Let $M$ be a matrix of order $m+n$ composed of the blocks:
\begin{align}
 	M=
 	\begin{pmatrix}
 	\Id_{m}	&	C
 	\\
 	B 			&	D
 	\end{pmatrix}
 	.
 	\label{mat_det_red_lem_block_mat}
\end{align} 
Then a smaller matrix $P$ of order $n$ can be constructed as
\begin{align}
 P=D-BC
 \label{mat_det_red_lem_reduced_mat}
\end{align} 
such that it is equivalent to the original matrix $M$ up-to the evaluation of their determinants,
\begin{align}
	\det_n P
	=
	\det_{m+n} M
	.
 \label{mat_det_red_lem_reduced_mat_det}
\end{align}
\end{lem}
\begin{proof}
Since the matrix $M$ contains an identity block $I_m$, we need to consider only permutations which stabilise the subset $[1;m]$ of indices.
This permits us to write the determinant as
\begin{multline}
	\det M
	= 
	\begin{aligned}[t]
	&\sum_{\tau\in\mathfrak{S}_{n}}
	(-1)^\tau M_{11}\ldots M_{mm}
	M_{m+\tau_{1},m+1}\ldots M_{m+\tau_{m},m+n}
	\\
	-&
	\sum_{\ell = 1}^{n}
	\sum_{k=1}^{m}
	\sum_{\tau\in\mathfrak{S}_{n}}
	(-1)^\tau
	\left\lbrace
	\begin{multlined}
	M_{11}	\ldots	M_{m+\tau_{\ell},k}	\ldots	M_{mm}\,
	\\
	\times
	M_{m+\tau_{1},m+1}	\ldots M_{k,m+\ell}	\ldots	M_{m+\tau_{m},m+n}
	\end{multlined}
	\right\rbrace
	\\
	+&
	\underset{\ell_1\neq \ell_2}{\sum_{\ell_{1}=1}^{n}\sum_{\ell_{2}=1}^{n}}
	\underset{k_1\neq k_2}{\sum_{k_{1}=1}^{m}\sum_{k_{2}=1}^{m}}
	\sum_{\tau\in\mathfrak{S_{n}}}
	(-1)^\tau
	\left\lbrace
	\begin{multlined}
	M_{11}	\ldots	M_{m+\tau_{\ell_{1}},k_{1}}	\ldots	M_{m+\tau_{\ell_{2}},k_{2}}	\ldots	M_{mm}\,
	\\
	M_{m+\tau_{1},m+1}	\ldots M_{k_{1},m+\ell_{1}}	\ldots M_{k_{2},m+\ell_{2}}	\ldots	M_{m+\tau_{m},m+n}
	\end{multlined}
	\right\rbrace
	\\
	&\vdots
	\end{aligned}
	\\
	\text{All the iterations up-to} : \set{k_1,\ldots, k_{r_\text{max}}}, \set{\ell_1,\ldots, \ell_{r_\text{max}}};~ (r_{\text{max}} = \text{min}(m,n))
	.
\end{multline}
It can be rewritten as
\begin{multline}
	\det{M} =
	\\
	\sum_{r=0}^{r_{\text{max}}}
	(-1)^{r}
	\underset{k_{i}\neq k_{j}\, \forall i,j\leq r} {\sum_{k_1=1}^{m}\ldots\sum_{k_r=1}^{m}}
	\,
	\underset{\ell_{i}\neq \ell_{j}\, \forall i,j\leq r} {\sum_{\ell_1=1}^{n}\ldots\sum_{\ell_r=1}^{n}}
	\,
	\sum_{\tau\in\mathfrak{S}_{n}}
	(-1)^{\tau}
	\prod_{p=1}^{r} M_{m+\tau_{p},k_{p}} M_{k_{p},\ell_{p}}
	\underset{q\neq \ell_{1},\ldots \ell_{r}}{\prod_{q=1}^{n}}
	M_{m+\tau_{q},m+q}
	.
\end{multline}
In this sum, we can relax the condition $k_{i}\neq k_{j}$ by adding monomials with $k_{p} = k_{p^\prime}$ for some $p,p^\prime \leq r~ (r\leq r_\text{max})$.
Let $\ell_{p}$ and $\ell_{p^\prime}$ denote the corresponding pre-indices associated via the terms $M_{k_p,\ell_p}$ and $M_{k_{p^\prime},\ell_{p^\prime}}$ in the product over $p$,
we find that such a monomial is invariant under a permutation of $n+\ell_{p}$ and $n+\ell_{p^\prime}$.
Therefore it appears twice in the determinant sum and hence it would cancels out with itself.
With this observation, we can write
\begin{equation}
	\det M = \sum_{\tau\in\mathfrak{S}_{n}} (-1)^{\tau} \prod_{q=1}^{m} \left\lbrace M_{m+\tau_{q},m+q}-\sum_{k=1}^{n} (M_{m+\tau_{q},k}M_{k,m+q})\right\rbrace
\end{equation}
which corresponds to the determinant of the matrix $P$ constructed as $P=D-BC$ in \cref{mat_det_red_lem_reduced_mat}.
\end{proof}
\begin{coro}
Let $M$ be any matrix divided into the blocks:
	\begin{align}
	M=
	\begin{pmatrix}
	A 	&		C 
	\\
	B 	& 	D
	\end{pmatrix}
\end{align}
where the diagonal block $A$ is invertible.
Then a smaller matrix $P$ can be written as
\begin{align}
	P=D-BA^{-1}C
\end{align}
which is equivalent to $M$ up-to the evaluation of its determinant,
\begin{align}
	\det_{m+n} M
	= 
	\det_{m}A
	\det_{n} P
	.
\end{align}
\end{coro}
\begin{lem}
\label{lem:rank-1_det}
Let $M$ be a square matrix of order $m$ and $P$ a rank-1 matrix of the same order $m$.
Then,
\begin{align}
	\det(M+P)
	=
	\det M
	+
	\sum_{a=1}^{m}
	\det M^{(a)}
\end{align}
where the matrix $M^{(a)}$ is obtained by replacing the $a$\textsuperscript{th} column with the corresponding column vector forming $P$.
\end{lem}
\begin{proof}
It follows from the standard development of the determinant:
\begin{align}
	\det(M+P)
	=
	\det M
	+
	\sum_{a=1}^{m}
	\det M^{(a)}
	+
	\sum_{a< b}
	\det M^{(ab)}
	+
	\cdots
	.
\end{align}
However since the matrix $P$ is rank-1, all the higher order terms vanish.
\end{proof}
\section{Cauchy-Vandermonde matrix in rational parametrisation}
\label{sec:cvdet}%
\index{cv@\textbf{Cauchy-Vandermonde}|seealso{Index-free}}%
A Cauchy-Vandermonde matrix $\cmat<\ptn\delta>$ [see \cref{defn:rat_cau_van}] is composed of the two rectangular blocks $\cmat$ and $\vmat<\ptn\delta>$ as follows:
\index{cv@\textbf{Cauchy-Vandermonde}!mat cv@$\cmat_{\ptn\delta}[\cdot\Vert\cdot]$: Cauchy-Vandermonde matrix (rational)|textbf}%
\begin{align}
	\cmat_{\ptn \delta}[\bm x\Vert\bm y]
	=
	\left(
	\cmat[\bm x\Vert\bm y]
	~\bigg|~
	\vmat_{\ptn\delta}[\bm x]
	\right)
	.
	\label{rat_cv_mat_append}
\end{align}
Here $\cmat[\bm x\Vert\bm y]$ is Cauchy matrix:
\begin{align}
	\cmat(x_j\Vert y_k)
	=
	\frac{1}{x_j-y_k}
\end{align}
and $\vmat<\ptn\delta>[\bm x]$ is a rectangular Vandermonde matrix:
\begin{align}
	\vmat<\ptn\delta>_{a}[x_j]
	=
	x_{j}^{a-1}.
\end{align}
\minisec{Determinant of rational Cauchy-Vandermonde matrix}
\begin{lem}
\label{lem:rat_cau_van_det}
The determinant of the Cauchy-Vandermonde matrix $\cmat<\ptn\delta>[\bm x\Vert\bm y]$ with the set of variables $\bm x$ ($n_{\bm x}=m+n$) and $\bm y$ ($n_{\bm y}=m$) is given by the formula:
\begin{align}
	\det \cmat<{\ptn{\delta}}>(\bm{x}\Vert\bm{y})
	&=
	\bmalt(\bm{x}\Vert\bm{y})
	=
	\frac{%
	\prod_{j>k}^{n+m}(x_j-x_k)\prod_{j<k}^{m}(y_j-y_k)
	}{%
	\prod_{j=1}^{n+m}\prod_{k=1}^{m}(x_j-y_k)
	}
	.
	\label{cau_van_rat_det_formula_append}	
\end{align}
\index{cv@\textbf{Cauchy-Vandermonde}!det Cv@$\bmalt\varphi(\cdot\Vert\cdot)$: Cauchy(-Vandermonde) determinant|textbf}%
\end{lem}
\begin{proof}
We give two versions of the proof here.
The first one is inductive and computes the determinant with increasing the sizes of Vandermonde block.
The second is limiting procedure that produces the Cauchy-Vandermonde matrix, starting from a larger square Cauchy matrix.
\begin{enumerate}[label=\textbf{Method \Roman*.}, ref={method \Roman*}, align=left, leftmargin=0pt]
\item Here we will do recursion on the integer $n$ which is the difference of the cardinalities $n=n_{\bm x}-n_{\bm y}$. This also represents the number of Vandermonde columns.
For $n=0$, we have a trivial case of a square Cauchy determinant.
Therefore let us start with a mixed matrix with one Vandermonde column. Evidently this column is a constant vector $\vmat[x_j]=1$, hence by developing the determinant on this column we get,
\begin{align}
	\det\cmat<\ptn\delta(1)>[\bm x\Vert\bm y]
	=
	\sum_{k=1}^{m+1}
	(-1)^{m+k+1}
	\det\cmat[\bm{x_{\hat{k}}}\Vert \bm y]
	=
	\bmalt(\bm x\Vert\bm y)
	\sum_{k=1}^{m+1}
	\phifn'(x_k|\bm y,\bm x)
	.
\end{align}
We can now see compute the summation over the $\phifn$ function in the above expression.
We find that it has zero residue for all the poles $x_k=x_{k'}$ and hence as an entire function which is bounded, we can see that it is a constant function, since
\begin{align}
	\sum_{k=1}^{m+1}
	\phifn'(x_k|\bm y, \bm x)
	=
	\sum_{k=1}^{m+1}
	\frac{%
	\bmprod(x_k-\bm y)
	}{%
	\bmprod(x_k-\bm{x_{\hat{k}}})
	}
	=
	1
	.
\end{align}
This demonstrates that the formula \eqref{cau_van_rat_det_formula_append} holds for $n=\ell(\ptn\delta)=1$.
Let us assume that it holds up-to certain $n=\ell(\ptn\delta)\in\Nset$, we now show that it holds for $\ell(\ptn\delta)=n+1$ using the similar approach.
We see that the development on the last column of the Vandermonde matrix leads to the summation:
\begin{align}
	\det\cmat<\ptn\delta(n+1)>[\bm x\Vert\bm y]
	=
	\sum_{k=1}^{m+n+1}
	(-1)^{m+n+k+1}
	x_{k}^{n}
	\det\cmat<\ptn\delta(n)>[\bm{x_{\hat{k}}}\Vert\bm y]
	=
	\bmalt(\bm x\Vert\bm y)
	\sum_{k=1}^{m+n+1}
	x_{k}^{n}
	\phifn'(x_{k}|\bm y,\bm x)
	.
\end{align}
We can again easily see that the resulting summation forms a bounded multivariate entire function and it takes thus a constant value $1$.
This proves inductively the result \eqref{cau_van_rat_det_formula_append} for the determinant of a Cauchy-Vandermonde matrix.
\item 
Let us construct a bigger Cauchy matrix $\cmat[\bm x\Vert\bm z]$ where $\bm z=\bm y\cup\bm w$. The number of extra variables added $n_{\bm w}=n$ is taken such that it forms a square Cauchy matrix.
Let us now show that by taking the limits where the extra roots $\bm w$ are send to infinity, we obtain a Cauchy-Vandermonde matrix through a sequential procedure of taking this limit
\begin{align}
	\cmat[\bm x\Vert\bm z]
	\to
	\cmat<\ptn\delta(1)>[\bm x\Vert\bm{z_{(1)}}]
	\to
	\cdots
	\to
	\cmat<\ptn\delta(k)>[\bm x\Vert\bm{z_{(k)}}]
	\to
	\cdots
	\to
	\cmat<\ptn\delta(n)>[\bm x\Vert\bm y]
\end{align}
where $\bm{z_{(k)}}=\bm y\cup \set{w_{k+1},\ldots, w_{n}}$.
For the first iteration we multiply the column for $w_1$ with the first order monomial $w_1$ and take the limit $w_1\to\infty$ to obtain 
\begin{align}
	\lim_{w_1\to\infty}	
	(-w_1)
	\frac{1}{x_j-w_1}
	=
	1
	.
\end{align}
For the successive iteration we multiply the column for $w_k$ with a monomial of order $k$:
\begin{align}
	(-w_{k})^{k}
	\frac{1}{x_j-w_k}
	=
	\sum_{r=0}^{\infty}(-w_k)^{k-r-1} x_j^r
	.
\end{align}
At this point we also subtract the linear combination of the Vandermonde columns of the inferior order $r<k$ to cancel the divergent terms, note that this manipulation does not affect its determinant.
In this way we obtain at each recursion we obtain the intermediate matrix:
\begin{align}
	\cmat<\ptn\delta(k)>[\bm x\Vert\bm{z_{(k)}}]
	=
	\left(
	\cmat[\bm x\Vert \bm y]
	~\big|~
	\vmat<\ptn\delta(k)>[\bm x]
	~\big|~
	\cmat[\bm x\Vert \bm{w_{(k)}}]
	\right)
	.
\end{align}
At the end of this procedure, the final expression would give us the Cauchy-Vandermonde matrix $\cmat<\ptn\delta>[\bm x\Vert\bm y]$.
\par
We now see that this limiting procedure for the determinant leads us to
\begin{align}
	\det\cmat<\ptn\delta>[\bm x\Vert\bm y]
	=
	\left\lbrace
	\prod_{a=1}^{n}
	\lim_{w_a\to\infty}
	(-w_a)^{a}
	\right\rbrace
	\bmalt(\bm x\Vert\bm z)
	=
	\bmalt(\bm x\Vert\bm y)
	.
\end{align}
\end{enumerate}
This demonstrates the result \eqref{cau_van_rat_det_formula_append} using the two methods.
\end{proof}
\subsection{Inversion of the rational Cauchy-Vandermonde matrix}
To express the inverse matrices, we will use the symmetric and supersymmetric function that we have defined below.
\subsubsection{Symmetric and supersymmetric function}
\begin{defn}[Elementary and total symmetric functions]
\label{defn:sym_fns_append_ele-tot}
Let $\bm z$ denote a set of variables. The elementary symmetric function of $e_r(\bm z)$ and total symmetric functions $h_r(\bm z)$ are the polynomials of degree $r$ defined as
\index{cv@\textbf{Cauchy-Vandermonde}!sym fn ele@$e_r(\cdot)$: elementary symmetric functions|textbf}%
\index{cv@\textbf{Cauchy-Vandermonde}!sym fn ele@$h_r(\cdot)$: total symmetric functions|textbf}%
\begin{subequations}
\begin{align}
	e_r(\bm x)&=	
	\sum_{a_1<a_2<\cdots<a_r}
	x_{a_1}x_{a_2}\cdots x_{a_{r}}
	,
	\label{ele_symm_append}
	\shortintertext{and}
	h_r(\bm x)&=
	\sum_{\underset{\alpha_{j}\geq 0,~\forall j\leq n_{\bm x}}{\alpha_{1}+\cdots+\alpha_{n_{\bm x}}=r}}
	x_{1}^{\alpha_{1}}x_{2}^{\alpha_{2}}\cdots x_{n_{\bm x}}^{\alpha_{n_{\bm x}}}
	.
	\label{tot_symm_append}
\end{align}
\end{subequations}
Their generating functions $E(z)$ and $H(z)$ can be written as
\begin{subequations}
\begin{align}
	E(z)&=\bmprod (z-\bm x) = \sum_{r=0}^{n_{\bm x}}(-1)^{r} z^{n_{\bm x}-r} e_r(\bm x),
	\label{ele_symm_gen-fn_append}
	\\
	H(z)&=\frac{1}{\bmprod(z-\bm x)}
	=
	\sum_{r=0}^{\infty} z^{-n_{\bm x}-r} h_r(\bm x)
	.
	\label{tot_symm_gen-fn_append}
\end{align}
\end{subequations}
\end{defn}
For the trivial case of the degree $r=0$, both of these functions takes the value $e_0(\bm x)=h_0(\bm x)=1$.
We can also easily see that for the degree 1 both the elementary and total symmetric functions are equal $e_1(\bm x)=h_1(\bm x)$.
We now generalise this to the supersymmetric case.
\begin{defn}[Supersymmetric elementary function]
\label{defn:ele_susy_append}
A supersymmetric elementary function $e_r(\bm x\Vert\bm y)$ of degree $r$ over two sets of variables $\bm x$ and $\bm y$ is defined, as an alternating convolution sum of the total and elementary symmetric functions:
\index{cv@\textbf{Cauchy-Vandermonde}!susy fn ele@$e_r(\cdot\Vert\cdot)$: elementary supersymmetric functions|textbf}%
\begin{align}
	e_r(\bm x\Vert\bm y)
	&=
	\sum_{a=0}^{r}
	(-1)^a
	e_{r-a}(\bm x)
	h_{r}(\bm y)
	.
	\label{ele_susy_append}
\end{align}
It has the $\phifn(z|\bm x, \bm y)$ as its generating function:
\begin{align}
	\phifn(z|\bm x,\bm y)
	=
	\frac{\bmprod(z-\bm x)}{\bmprod(z-\bm y)}
	=
	\sum_{r=0}^{\infty}
	(-1)^r
	z^{n_{\bm x}-n_{\bm y}-r}
	e_r(\bm x\Vert\bm y)
	.
	\label{ele_susy_gen-fn_append}
\end{align}
\end{defn}
Evidently for the trivial case of the degree 0, it is a constant $e_0(\bm x\Vert\bm y)=1$.
For the non-trivial degrees $r>0$, we have the following property for the coinciding sets of variables.
\begin{lem}[Conjugacy of symmetry functions]
\label{lem:susy_conjug_prop}
Elementary supersymmetric function of a non-trivial degree $r>0$ vanishes when the two sets of variables coincide, as it can be seen from the following:
\begin{align}
	e_{r>0}(\bm x\Vert\bm x)=0.
\end{align}
This translates to the following identity for the elementary symmetric and total symmetric functions:
\begin{align}
	\sum_{a=0}^{r}(-1)^a e_{r-a}(\bm x) h_{a}(\bm x)=0.
	\label{sym_fns_conjug_rel}
\end{align}
\end{lem}
\begin{proof}
It can easily shown from the generating function $\phifn$ \eqref{ele_susy_gen-fn_append} for the elementary supersymmetric function since we have since $\phifn(z|\bm x,\bm x)=1$.
\end{proof}
\subsubsection{Inverse of Cauchy and Vandermonde matrices}
Let us begin the inverse of a square Cauchy and a square Vandermonde matrix. Although these are well known results, their computation give valuable insights on the inversion of the hybrid Cauchy-Vandermonde matrix.
\begin{lem}
The inverse of the Cauchy matrix $\cmat(\bm{x}\Vert\bm{y})$ is given by,
\begin{align}
	\cmat^{-1}(y_k\Vert x_j)
	&=
	\frac{\bmprod(x_{k}-\bm{y})}{\bmprod^\prime(x_{k}-\bm{x})}
	\frac{\bmprod(y_{j}-\bm{x})}{\bmprod^\prime(y_{j}-\bm{y})}
	\frac{1}{y_{j}-x_{k}}
	\label{inv_rat_cau_mat}
	,
\end{align}
\end{lem}
\begin{proof}
It suffices to compute the determinant of the cofactors, since all of them are determinant of the Cauchy matrices of lower order, it can be computed directly.
\end{proof}
\begin{lem}
The components of the inverse of the Vandermonde matrix can be written as 
\begin{align}
	\vmat*<\ptn\delta>^{-1}_{a,k}[\bm{x}]&=
	\frac{(-1)^{n-a}e_{n-a}(\bm{x}_{\hat{k}})}{\bmprod^\prime(x_{k}-\bm{x})}
	\label{inv_rat_van_mat}
	,
\end{align}
where $e_{\alpha}(\bm{z})$ denotes the elementary symmetric polynomial of the degree $\alpha$ in variables $\bm{z}$.
\end{lem}
\begin{proof}
There exists many different proofs for this result.
Here we develop the determinant formula for the Vandermonde matrix in such a way that
\begin{align}
	\bmalt(\bm{x})&=
	\prod(x_{k}-\bm{x}_{\hat{k}})
	\bmalt(\bm{x}_{\hat{k}})
	=
	\sum_{a=0}^{n-1}(-1)^{n-a} x_{k}^{a} e_{n-a}(\bm{x}_{\hat{k}})
	\bmalt(\bm{x}_{\hat{k}})
	.
\end{align}
Now comparing this sum to the development of the determinant on the column $j$ and row $k$ give us the cofactor of the Vandermonde matrix.
In this way all cofactors can be found and hence the inverse can be obtained.
\end{proof}
\begin{rem}
This gives us some useful insights for the inversion of a Cauchy-Vandermonde matrix.
In its determinant formula \eqref{cau_van_rat_det_formula_append}, here we do a similar expansion:
\begin{align}
	\bmalt(\bm x\Vert\bm y)
	=
	\frac{\bmprod(x_k-\bm{x_{\hat{k}}})}{\bmprod(x_k-\bm y)}
	\bmalt(\bm{x_{\hat{k}}}\Vert\bm y)
	=
	\sum_{a=0}^{\infty}(-1)^{n-a} x_{k}^{a} e_{n-a}(\bm{x}_{\hat{k}}\Vert\bm y)
	\bmalt(\bm{x}_{\hat{k}}\Vert\bm y)
	,
\end{align}
where $e_{\alpha}(\bm u\Vert \bm v)$ denotes the supersymmetric elementary polynomial [see the \cref{defn:ele_susy_append}].
This gives a hint that inverse Cauchy-Vandermonde matrix contains supersymmetric elementary functions, which is exactly what we find in the following lemma.
\end{rem}
To demonstrate this result, let us first prove the following lemma.
\begin{defn}
\label{defn:cv_mat_jump}
A partition $\ptn\lambda_r$ is defined as sum of partition $\ptn\lambda_r=\ptn\delta+\ptn{1^r}$, it represents a partition of consecutive integers which jumps over an index $r$ :
\begin{align}
	\ptn\lambda_r=
	\set{0,1,2,\ldots,r-1,r+1,\ldots,n+1}
	.
	\label{jump_ptn_append}
\end{align}
Consequently, a Cauchy-Vandermonde matrix $\cmat<\ptn\la_r>$ can be defined as a matrix containing a Vandermonde block of columns which skips over the column of degree $r$ :
\begin{multline}
	\cmat<\ptn\lambda_r>[\bm x\Vert\bm y]
	=
	\Big[
	\cmat[\bm x\Vert\bm y]
	~\Big|~
	\vmat_{\ptn\lambda_r}[\bm x]
	\Big]
	\\
	=
	\left(
	\begin{array}{ccc|cccccc}
	\frac{1}{x_1-y_1}		&		\cdots 		&		\frac{1}{x_1-y_m}
	&
	1 	& \cdots 	&	x_{1}^{r-1} & x_1^{r+1}	&	\cdots &	x_1^{n-1}
	\\
	\vdots 	&	\ddots	&	\vdots
	&
	\vdots	& \ddots & \vdots & \vdots 	& \ddots	& \vdots
	\\
	\frac{1}{x_{n+m}-y_1}		&		\cdots 		&		\frac{1}{x_{n+m}-y_m}
	&
	1 	& \cdots 	&	x_{n+m}^{r-1} & x_{n+m}^{r+1}	&	\cdots &	x_{n+m}^{n-1}
	\end{array}
	\right)
	.
	\label{susy_ele_lem_append_jump_cv_mat}
\end{multline}
\end{defn}
\begin{lem}
\label{lem:rat_cv_det_quotient}
The supersymmetric elementary function $e_r(\bm x\Vert\bm y)$ can be representation by ratio of determinants: 
\begin{align}
	e_r(\bm x\Vert\bm y)&=
	\frac{%
	\det\cmat<\ptn\lambda_r>[\bm x\Vert\bm y]%
	}{%
	\det\cmat<\ptn\delta>[\bm x\Vert\bm y]%
	}
	\label{susy_ele_fn_rat_det_append}
	.
\end{align}
\end{lem}
\begin{proof}
This result is obviously true for $r=0$ where we have the trivial polynomial of degree zero $e_{0}(\bm x\Vert\bm y)=1$.
Let us start with the first non-trivial case $r=1$. 
We develop the determinant appearing in the numerator of \cref{susy_ele_fn_rat_det_append} on the last column:
\begin{align}
	\det\cmat<\ptn\lambda_1>[\bm x\Vert\bm y]
	=
	(-1)^{n+m}
	\frac{\bmalt(-\bm y)}{\bmprod(\bm x-\bm y)}
	\sum_{a=1}^{n+m}(-1)^{a}x_{a}^{n}%
	\bmalt(\bm{x_{\hat{a}}})
	E(x_a|\bm y)
	.
	\label{det_cv_jump-1_expn_elem_sym_fn}
\end{align}
Here the function $E$ denotes the generating function \eqref{ele_symm_gen-fn_append} for the elementary symmetric function.
Let us substitute \cref{ele_symm_gen-fn_append} while exchanging the order of two summations in the resulting expression.
It permits us to rewrite \cref{det_cv_jump-1_expn_elem_sym_fn} as
\begin{align}
	\det\cmat<\ptn\lambda_1>[\bm x\Vert\bm y]
	=
	(-1)^{n+m}
	\frac{\bmalt(-\bm y)}{\bmprod(\bm x-\bm y)}
	\sum_{r=0}^{m}
	(-1)^r
	e_r(\bm y)
	\det\left(
	\vmat<\ptn\delta(n+m-1)>[\bm x]
	~\big|~
	x^{n+m-r}
	\right)
	.
	\label{det_cv_jump-1_expn_jump_van_det}
\end{align}
where the determinant of the Vandermonde matrix comes from the interior summation inside the double-sum:
\begin{align}
	\det\left(
	\vmat<\ptn\delta(n+m-1)>[\bm x]
	~\big|~
	x^{n+m-r}
	\right)
	=
	(-1)^{n+m}
	\sum_{a=1}^{n+m}
	(-1)^{a}
	x_{a}^{n+m-r}
	\bmalt(\bm{x_{\hat{a}}})
\end{align}
which can be evaluated as
\begin{align}
	\det\left(
	\vmat<\ptn\delta(n+m-1)>[\bm x]
	~\big|~
	x^{n+m-r}
	\right)
	=\begin{dcases}
		0,	&	r\geq 2
		\\
		\bmalt(\bm x),	&	r=1
		\\
		e_{1}(\bm x)\bmalt(\bm x),	& r=0.
	\end{dcases}
	.
\end{align}
Substituting it back into \cref{det_cv_jump-1_expn_jump_van_det} leads to the result:
\begin{align}
	\det\cmat<\ptn\lambda_1>[\bm x\Vert\bm y]
	=
	\det\cmat<\ptn\delta>[\bm x\Vert \bm y]
	e_{1}(\bm x\Vert\bm y)
\end{align}
where we use have used the relation:
\begin{align}
	e_1(\bm x\Vert\bm y)
	=e_1(\bm x)-e_1(\bm y)
\end{align}
that follows from its \cref{defn:ele_susy_append} since we  $e_0(\bm x)=h_0(\bm x)=1$ and $e_1(\bm x)=h_1(\bm x)$.
We shall now prove the result \eqref{susy_ele_fn_rat_det_append} by induction.
\par
Let us assume that this result holds for all the degrees $r\leq \alpha-1$ up-to certain positive integer $\alpha$.
We develop the determinant of $\cmat<\ptn\la_\alpha>$ on the last column to obtain the summation:
\begin{align}
	\det\cmat<\ptn\la_\alpha>[\bm x\Vert\bm y]
	=
	(-1)^{n+m}
	\frac{%
	\bmalt(-\bm y)
	}{%
	\bmprod(\bm x-\bm y)
	}
	\sum_{a=1}^{n+m}
	(-1)^{a}
	x_a^{n}
	E(x_a|\bm y)
	e_{\alpha-1}(\bm{x_{\hat{a}}}\Vert\bm y)
	\bmalt(\bm{x_{\hat{a}}})
	.
	\label{det_cv_jump_expn_elem_fns}
\end{align}
Let us substitute now \cref{ele_symm_gen-fn_append,ele_susy_append} for the generating function $E(x_a|\bm y)$ and the supersymmetric function $e_{\alpha-1}(\bm{x_{\hat{a}}}\Vert\bm y)$ into the above expression.
We also exchange the order of summations in the resulting triple-sum to write \cref{det_cv_jump_expn_elem_fns} as
\begin{multline}
	\det\cmat<\ptn\lambda_\alpha>[\bm x\Vert\bm y]
	=
	(-1)^{n+m}
	\frac{\bmalt(-\bm y)}{\bmprod(\bm x-\bm y)}
	\\
	\times
	\Bigg\lbrace
	\sum_{r=0}^{m}
	(-1)^r
	e_r(\bm y)
	\sum_{s=0}^{\alpha-1}
	(-1)^{s}
	h_{s}(\bm y)
	\det\left(
	\vmat<\ptn\lambda_{\alpha-s-1}(n+m-1)>[\bm x]
	~\big|~
	x^{n+m-r}
	\right)
	\Bigg\rbrace
	.
	\label{det_cv_jump_expn_van_dets}
\end{multline}
The innermost sum that runs over $x_a$ variables is hidden inside the Vandermonde determinant in the above expression \eqref{det_cv_jump_expn_van_dets} which can be evaluated as
\begin{align}
	\det\left(
	\vmat<\ptn\lambda_{\alpha-s-1}(n+m-1)>[\bm x]
	~\big|~
	x^{n+m-r}
	\right)=
	\big\lbrace
	\delta_{r,0}e_{\alpha-s}(\bm x)
	-
	(-1)^{\alpha-s}
	\delta_{r,\alpha-s}
	e_0(\bm x)
	\big\rbrace
	\bmalt(\bm x)
	.
	\label{eval_jump_van_det_arbit}
\end{align}
Substituting it back into \cref{det_cv_jump_expn_van_dets} gives us
\begin{align}
	\det\cmat<\ptn\la_\alpha>[\bm x\Vert\bm y]
	=
	\left\lbrace
	\sum_{s=0}^{\alpha-1}
	(-1)^s
	h_s(\bm y)
	e_{\alpha-s}(\bm x)
	-
	e_0(\bm x)
	\sum_{s=0}^{\alpha-1}
	(-1)^{s}
	h_{s}(\bm y)
	e_{\alpha-s}(\bm y)
	\right\rbrace
	\bmalt[\bm x\Vert\bm y]
	.
	\label{cv_det_jump_expn_sym_fns_conv}
\end{align}
Using the conjugacy relation \eqref{sym_fns_conjug_rel} in \cref{cv_det_jump_expn_sym_fns_conv} we obtain the result:
\begin{align}
	\det\cmat<\ptn\la_\alpha>[\bm x\Vert\bm y]
	=
	\det\cmat<\ptn\delta>[\bm x\Vert\bm y]
	e_{\alpha}(\bm x\Vert\bm y)
	.
\end{align}
\end{proof}
\begin{rem}
A most general form of this result for any partition $\ptn\lambda$ was found in \cite{MoeV03}. For a more detailed overview of the subject, one can also refer to the thesis \cite{Moe07}.
\end{rem}
\begin{prop}
\label{lem:rat_cv_inv_append}
The inverse of the Cauchy-Vandermonde matrix \eqref{rat_cv_mat_append} contains two blocks.
The Cauchy block has the components given by the expressions identical to those obtained for the inverse of a square Cauchy matrix \eqref{inv_rat_cau_mat}
\begin{align}
	\cmat<\ptn\delta>^{-1}_{j,k}[\bm y\Vert\bm x]&=
	\frac{\bmprod(x_{k}-\bm{y})}{\bmprod^\prime(x_{k}-\bm{x})}
	\frac{\bmprod(y_{j}-\bm{x})}{\bmprod^\prime(y_{j}-\bm{y})}
	\frac{1}{y_{j}-x_{k}};
	& k\leq m
	.
	\label{inv_rat_cvmat_cau_block_append}
\end{align}
The Vandermonde block of the inverse Cauchy-Vandermonde matrix has components given by the expression:
\begin{align}
	\cmat<\ptn\delta>^{-1}_{m+a,k}[\bm y\Vert\bm x]
	&=
	(-1)^{n-a}
	\frac{%
	\bmprod(x_{k}-\bm{y})
	}{%
	\bmprod^\prime(x_{k}-\bm{x})
	}
	e_{n-a}(\bm{x}_{\hat{k}}\Vert\bm{y})
	;	& a&\leq  n
	\label{inv_rat_cvmat_van_block_append}
\end{align}
which is similar to the expression \eqref{inv_rat_van_mat} for the inverse of the square Vandermonde matrix, while the symmetric functions are replaced by their supersymmetric equivalents.
\end{prop}
\begin{proof}
It is the corollary to the \cref{lem:rat_cv_det_quotient}.
\end{proof}
\subsection{Duality of the Cauchy-Vandermonde matrices}
Let us begin with the inverse of a Cauchy matrix.
We can see that it can be expressed as a diagonal dressing of itself, as we can see from the following expression:
\begin{align}
	\cmat^{-1}[\bm{y}\Vert\bm{x}]&=
	\diag\big[
	\phi^\prime(\bm{y}|\bm{x},\bm{y})
	\big]
	\cmat[\bm y\Vert\bm x]
	\diag\big[
	\phi^\prime(\bm{x}|\bm{y},\bm{x})
	\big]
	.
	\label{rat_cau_sq_inv_dressing_append}
\end{align}
\index{cv@\textbf{Cauchy-Vandermonde}!mat Cau@$\Cmat[\cdot\Vert\cdot]$: Cauchy matrix (circular)}%
It is also important to note that the determinant of the two diagonal matrices can be written as
\begin{align}
	\prod_{j\leq n}
	\frac{%
	\bmprod(y_{j}-\bm{x})
	}{%
	\bmprod^\prime(y_{j}-\bm{y})
	}
	\prod_{k\leq n}
	\frac{%
	\bmprod(x_{k}-\bm{y})
	}{%
	\bmprod^\prime(x_{k}-\bm{x})
	}
	=
	\bmalt(\bm{x}\Vert\bm{y})
	\bmalt(\bm{y}\Vert\bm{x})
	=
	\bmalt^2(\bm{x}\Vert\bm{y})
	,
\end{align}
which simply affirms the trivial result:
\begin{align}
	\det \cmat^{-1}[\bm x\Vert\bm y]
	=
	\frac{\bmalt(\bm{y}\Vert\bm{x})}{\bmalt^2(\bm{x}\Vert\bm{y})}
	=
	\frac{1}{\bmalt(\bm{x}\Vert\bm{y})}
	.
\end{align}
For the inverse of the Vandermonde matrix \eqref{inv_rat_van_mat} however we start to see non-trivial nature of this dressing.
In this case we find that the inverse is a diagonally dressed version of an another matrix $\vmat*$ which we call its dual:
\begin{align}
	\vmat<\ptn{\delta}>^{-1}[\bm x]&=
	\vmat*<\ptn{\delta}>[-\bm{x}]
	\diag\big[
	\bmprod^\prime(x_{k}-\bm{x})^{-1}
	\big]_{k\leq n}
	.
	\label{inv_rat_van_dressing_append}
\end{align}
We can see from \cref{inv_rat_van_mat} that the dual Vandermonde matrix is composed of the elementary polynomials as follows:
\begin{align}
	\vmat*<\ptn\gamma>_a(x_k)
	=
	e_{n-a}(\bm{x_{\hat{k}}})
	.
\end{align}
For the determinant we find that
\begin{align}
	\det \vmat^{-1}[\bm{x}] =
	\frac{\det \vmat*[\bm{-x}]}{\bmalt^2(\bm{x})}
	=
	\frac{1}{\bmalt(\bm{x})}
	.
\end{align}
Moreover we can see that inverse of the dual Vandermonde matrix can be expressed as a diagonal dressing of the Vandermonde matrix:
\begin{align}
	\vmat*<\ptn\delta>^{-1}[\bm{x}]&=
	\vmat<\ptn{\delta}>[-\bm x]
	\diag\Big[
	\bmprod^\prime(x_{k}-\bm{x})^{-1}
	\Big]_{k\leq n}
	.
	\label{inv_rat_van_dual}
\end{align}
A similar project can be realised for the mixed Cauchy-Vandermonde matrix.
Let us define the supersymmetric dual of the Vandermonde matrix
\begin{align}
	\vmat<\ptn{\delta}>_{a,k}[\bm x\Vert\bm y]
	&=
	e_{n-a}(\bm{x}_{\hat{k}}\Vert\bm{y})
\end{align}
\index{cv@\textbf{Cauchy-Vandermonde}!susy fn ele@$e_r(\cdot\Vert\cdot)$: elementary supersymmetric functions}%
From the result of \cref{lem:rat_cv_inv_append} we can see that the inverse Cauchy-Vandermonde matrix can be expressed as the dressing of the dual Cauchy-Vandermonde matrix $\cmat*<\ptn\delta>$:
\begin{align}
	\cmat<\ptn\delta>^{-1}[\bm y\Vert\bm x]
	&=
	\diag\big[
	\phifn^\prime(\bm{y}|\bm{x},\bm{y})
	~\big|~
	\Id_{n}
	\big]
	~
	\left(\cmat*<\ptn\delta>[-\bm x\Vert -\bm y]\right)^T
	\,
	\diag\big[
	\phifn^\prime(\bm{x}|\bm{y},\bm{x})
	\big]
	.
\end{align}
The dual matrix is composed of the two blocks:
\begin{align}
	\cmat*<\ptn\delta>[\bm x\Vert\bm y]
	=
	\left[
	\cmat[\bm x\Vert\bm y]
	~\Big|~
	\vmat*<\delta>[\bm x]
	\right]
	.
\end{align}
We find that the determinant of the Cauchy-Vandermonde matrix and its dual are equal and hence we can write that
\begin{align}
	\det\cmat*<\ptn\delta>^{-1}[\bm x\Vert\bm y]
	=
	\frac{%
	\bmalt(-\bm x\Vert -\bm y)
	}{%
	\bmalt^2(\bm x\Vert\bm y)
	}
	=
	\frac{1}{\bmalt(\bm x\Vert\bm y)}
	.
\end{align}
Finally, the most important result for us is that the inverse of the dual Cauchy-Vandermonde matrix is given by the diagonal dressing of the original Cauchy-Vandermonde matrix:
\begin{align}
	\cmat*<\ptn\delta>^{-1}[\bm y\Vert\bm x]
	&=
	\diag\Big[
	\phifn^\prime(\bm{y}|\bm{x},\bm{y})
	\Big|
	\Id_{n}
	\Big]
	\cdot
	\left(\cmat<\ptn\delta>[-\bm x\Vert -\bm y]\right)^T
	\cdot
	\diag\Big[
	\phifn^\prime(\bm{x}|\bm{y},\bm{x})
	\Big]
	.
\end{align}
\section{Cauchy-Vandermonde matrix in hyperbolic parametrisation}
\label{sec:cau_van_mat_hyper}
The hyperbolic version of the Cauchy matrix is given by,
\begin{align}
	\Cmat[\bm\alpha\Vert\bm\beta]
	&=
	\bm{\bigg[}
	\frac{1}{\sinh\pi(\bm\alpha-\bm\beta)}
	\bm{\bigg]}
	.
	\label{hyper_cau_mat_def_append-chap}
\end{align}
It can be compared with the rational Cauchy matrix [see \cref{det_sq_rat_cau_alt_notn}] through change of variables which is shown below:
\begin{alignat}{3}
	x_{j}&=e^{2\pi\alpha_{j}}
	,
	&
	\qquad
	&
	y_{j}&=e^{2\pi\beta_{j}}
	.
	\label{CVmat_ratn_reparametrsn_map}
\end{alignat}
This allows us to write that
\begin{align}
	\Cmat_{j,k}[\bm\alpha(\bm x)\Vert\bm\beta(\bm y)]
	&=
	\frac{%
	2
	\sqrt{x_{j}y_{k}}
	}{%
	x_{j}-y_{k}
	}.
	\label{hyper_cau_rat_reparam_append-chap}
\end{align}
From this expression we can also verify that
\begin{align}
	\det \Cmat[\bm\la\Vert\bm\mu]
	&=
	\bmalt\sinh\pi(\bm\la\Vert\bm\mu)
	.
\end{align}
From the expression \eqref{hyper_cau_rat_reparam_append-chap}, we can see that in this parametrisation, hyperbolic Cauchy matrix is a  diagonally dressed rational Cauchy matrix:
\begin{align}
	\Cmat[\bm\alpha\Vert\bm\beta]=
	\diag\left[\bm{x}^{\frac{1}{2}}\right]
	\cmat[\bm{x}\Vert\bm{y}]
	\diag\left[2\bm{y}^{\frac{1}{2}}\right]
	.
	\label{hyp_cau_mat_rat_dressing_append-chap}
\end{align}
This also allows to compute the inverse of the hyperbolic matrix starting from \cref{rat_cau_sq_inv_dressing_append} in the dressed form:
\begin{align}
	\Cmat^{-1}[\bm\beta\Vert\bm\alpha]
	&=
	\diag\Big[
	\Phi^\prime(\bm\beta|\bm\alpha,\bm\beta)
	\Big]
	\Cmat[\bm\beta\Vert\bm\alpha]
	\diag\Big[
	\Phi^\prime(\bm\alpha|\bm\beta,\bm\alpha)
	\Big]
	.
	\label{hyper_cau_sq_inv_dressed}
\end{align}
\subsection{Construction of the hyperbolic Cauchy-Vandermonde matrix}
We now generalise this to the mixed Cauchy-Vandermonde matrix.
Let $\bm\alpha$ and $\bm\beta$ be the set of complex variables of the cardinalities $n_{\bm\alpha}=n+m$ and $n_{\bm\beta}=m$.
Let us first define the hyperbolic equivalent of the CV matrix $\Cmat<\ptn\delta>$ as
\begin{subequations}
\begin{flalign}
	&&
	\Cmat<\ptn\delta>_{j,k}&=\frac{e^{-n\pi(\alpha_j-\beta_k)}}{\sinh\pi(\alpha_j-\beta_k)},
	&	k\leq m
	\\
	&&	
	\Cmat<\ptn\delta>_{j,m+a}&=e^{n\pi\alpha_j(2a-n-1)}
	,
	& a\leq n
	.
\end{flalign}
\label{hcv_traditional_mat_comps}
\end{subequations}
In the parametrisation \eqref{CVmat_ratn_reparametrsn_map} we can thus write
\begin{align}
	\Cmat<\ptn\delta>[\bm\alpha(\bm x)\Vert\bm\beta(\bm y)]
	=
	\diag\left[
	\bm x^{-\frac{n-1}{2}}
	\right]
	\cmat<\ptn\delta>[\bm x\Vert\bm y]
	\diag\left[
	2\bm y^{\frac{n+1}{2}}
	~\Big|~
	\Id_n
	\right]
	\label{hcv_mat_traditional_dressing_rat_param}
	.
\end{align}
We can verify that the its determinant is given by,
\begin{align}
	\det\Cmat<{\ptn{\delta}}>[\bm\alpha\Vert\bm\beta]
	&=
	2^{-\frac{n(n)-1)}{2}}
	\bmalt\sinh\pi(\bm\alpha\Vert\bm\beta)
	.
	\label{hcv_mat_traditional_det}
\end{align}
We can recombine the hyperbolic version of the Vandermonde block:
\begin{align}
	\left[
	\begin{array}{c|c|c|c}
 		e^{-(n-1)\pi\bm\alpha} & e^{-(n-3)\pi\bm\alpha} & \cdots 	& e^{(n-1)\pi\bm\alpha}
	\end{array}
	\right]
 		&\to
 		\Vmat<\ptn\gamma>[\bm \alpha]
 	\label{hcv_append_hvan_recomb}
	,
\end{align}
where $\Vmat<\ptn\gamma>(\alpha)$ is a row vector [see \cref{notn:hyp_vec_cau_van}]:
\begin{subequations}
\begin{align}
	\Vmat<\ptn\gamma>_{a}(\alpha)
	&=
	\cosh\pi(n+1-2a)\alpha
	,
	&
	\text{for}\quad
	a&\leq \left\lceil\frac{n}{2}\right\rceil
	;
	\\
	\Vmat<\ptn\gamma>_{a}(\alpha)
	&=
	\sinh\pi(n+1-2a)\alpha
	,
	&
	\text{for}\quad
	a&\leq \left\lfloor\frac{n}{2}\right\rfloor
	.
\end{align}
\label{hcv_append_hvan_vec_defn}
\end{subequations}
This allows us to redefine the hyperbolic Cauchy-Vandermonde matrix as
\index{cv@\textbf{Cauchy-Vandermonde}!mat CV@$\Cmat<\ptn\gamma>[\cdot\Vert\cdot]$: Cauchy-Vandermonde matrix (circular)|textbf}%
\begin{align}
	\Cmat<\ptn\gamma>[\bm\alpha\Vert\bm\beta]
	=
	\bigg[
	\diag\big[e^{-n\pi\bm\alpha}\big]
	\cdot
	\Ccal[\bm\alpha\Vert\bm\beta]
	\cdot
	\diag\big[e^{n\pi\bm\beta}\big]
	\;\bigg|\;
	\van[\ptn\gamma][\bm\alpha]
	\bigg]
	\label{hyper_cau_van_mat_append}
	.
\end{align}
Note that the definition of the partition $\ptn\gamma$ \eqref{ptn_consec_even-odd} [see \cpageref{ptn_notn_page_begin}, or \cref{ptn_consec_even-odd}] also follows from here.
The effect of this recombination \eqref{hcv_append_hvan_recomb} on the determinant is that all the factors of the 2 in \cref{hcv_mat_traditional_det} are absorbed in the determinant.
This permits us to write
\begin{align}
	\det\Cmat<{\ptn{\gamma}}>[\bm\alpha\Vert\bm\mu]&=
	\bmalt\sinh\pi(\bm\alpha\Vert\bm\mu)
	\label{hcv_mat_det_append}
\end{align}
Let us remark that it is more natural and convenient to use the matrix \eqref{hyper_cau_van_mat_append} over the matrix \eqref{hcv_traditional_mat_comps} that was originally defined.
Indeed it is this latter form \eqref{hyper_cau_van_mat_append} that we have introduced and used in \cref{chap:gen_FF}.
The block $\Vmat<\ptn\gamma>$ in it is called the hyperbolic Vandermonde block in this context.
\subsection{Inverse of the hyperbolic CV matrix and its duality}
\label{sub:hcv_append_inv}
We can take the inverse of the hyperbolic Cauchy-Vandermonde matrix from its relationship \eqref{hcv_mat_traditional_dressing_rat_param}, however it would inadvertently involves the supersymmetric polynomials in the exponential variables \eqref{CVmat_ratn_reparametrsn_map}.
But the duality that we saw in the rational case provides a convenient alternative.
We can show that the inverse of the dual hyperbolic Cauchy-Vandermonde matrix can be represented as a dressing:
\begin{align}
	\Cmat*<\ptn\gamma>^{-1}[\bm\beta\Vert\bm\alpha]
	&=
	\diag\Big[
	\Phifn^\prime(\bm{\beta}\Vert\bm{\alpha},\bm{\beta})
	~\Big|~
	\Id_{n}
	\Big]
	\cdot
	(\Cmat<\ptn\gamma>[-\bm\alpha|-\bm\beta])^T
	\cdot
	\diag\Big[
	\Phifn'(\bm\alpha|\bm\beta,\bm\alpha)
	\Big]
	.
	\label{trig_cau_van_inv_dressing}
\end{align}
so that the its determinant can be written as
\begin{align}
	\det\Cmat*<\ptn\gamma>^{-1}[\bm\beta\Vert\bm\alpha]
	&=
	\frac{%
	\bmalt\sinh\pi(\bm{-\alpha}\Vert\bm{-\beta})
	}{%
	\bmalt^2\sinh\pi(\bm\alpha\Vert\bm\beta)
	}
	=
	\frac{%
	1
	}{%
	\bmalt\sinh\pi(\bm\alpha\Vert\bm\beta)
	}
	.
\end{align}
If we explicitly write the elements of this dual inverse of the trigonometric (hyperbolic) Cauchy-Vandermonde matrix, we obtain
\begin{subequations}
\begin{flalign}
	\Cmat*<\ptn\gamma>^{-1}_{j,k}[\bm\beta\Vert\bm\alpha]
	&=
	\Phi^\prime(\beta_{j}|\bm\alpha,\bm\beta)
	\Phi^\prime(\alpha_{k}|\bm\beta,\bm\alpha)
	\frac{%
	e^{-\pi\ell(\beta_{j}-\alpha_{k})}
	}{%
	\sinh\pi(\beta_{j}-\alpha_{k})
	}
	;
	&
	j&\leq m
	.
	\label{hcv_append_inv_cau_block}
	\\
	\Cmat*<\ptn\gamma>^{-1}_{m+a,k}[\bm\beta\Vert\bm\alpha]
	&=
	\Phi^\prime(\alpha_{k}|\bm\beta,\bm\alpha)
	\Vmat*<\ptn\gamma>_{a}[\alpha_k]
	,
	& 
	a&\leq n
	.
	\label{hcv_append_inv_hvan_block}
\end{flalign}
\end{subequations}
Here $\Vmat*$ denotes the vector with indices reversed in \cref{hcv_append_hvan_vec_defn}. Alternatively, we can also obtain the $\Vmat*$ by exchanging the $\sinh$ and $\cosh$ terms in \cref{hcv_append_hvan_vec_defn} [see \cref{notn:hyp_vec_cau_van}].
\begin{rem}
All the results obtained for the hyperbolic parametrisation apply automatically to trigonometric parametrisation.
This means these results gives us a complete picture portrayed by Cauchy-Vandermonde matrices and their extractions, in rational and circular parametrisations.
An extension of these results to the elliptic parametrisation is also known to exist, however we do not discuss it here since our computations for the XXX model do not require it.
\end{rem}
\clearpage{}%

\end{appendices}
\backmatter
\addpart{References}
\printbibliography[notcategory=ignore]
\indexprologue{The page numbers are added for the definitions and key mentions of the symbols. The \textbf{bold} page numbers refers to the page(s) where the definitions can be found.}
\printindex
\cleardoubleoddemptypage
\thispagestyle{empty}
\cleardoubleevenemptypage
\thispagestyle{empty}
\begin{picture}(80,0)
	\put(0,0){\makebox(0,0)[lc]%
	{\includegraphics*[height=40pt]{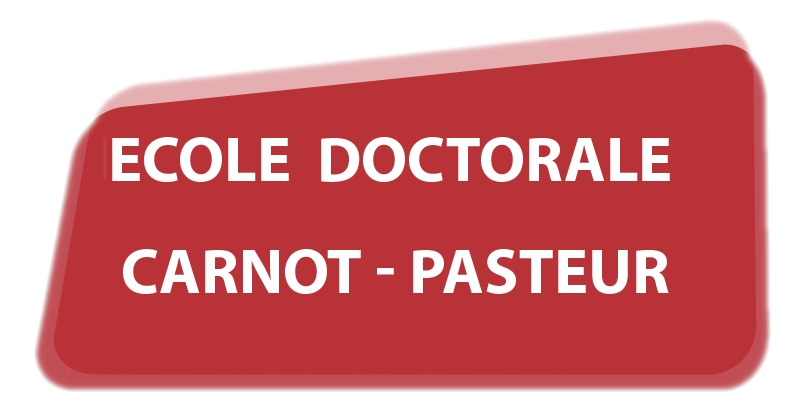}}
	}
\end{picture}
\vspace{\stretch{3}}
\begin{resumebox}{Résumé}
\textbf{Titre :} \frtitle
\\[1ex]
\textbf{mots clefs:} \frkeywords.
\begin{multicols}{2}
\small
Les systèmes intégrables quantiques restaient longtemps un domaine où des méthodes mathématiques modernes permettaient d’accéder aux résultats intéressants pour l‘étude de systèmes physiques.
Le calcul exacte, numérique et asymptotique de fonction de corrélation reste un de sujets les plus importants de la théorie de modèles intégrables quantiques. 
Dans ce cadre l’approche basée sur le calcul des facteurs de forme s’est révélée la plus efficace.
Dans ce thèse, une méthode alternative fondée sur l'ansatz de Bethe algébrique est développée pour calculer des facteurs de formes dans la limite thermodynamique.
Elle est appliqué et décrit dans le contexte de chaîne de spin isotrope XXX, qui est un des cas plus intéressant des modèles critiques où la zone de Fermi est non-compacte.
Dans le cas particulière des facteurs de formes à deux-spinons, on obtient un résultat exact en forme close qui est comparable à celui-ci obtenu initialement dans le formalisme de l'algèbre des opérateurs de $q$-vertex.
Cette méthode est aussi généralisée au calcul des facteurs de formes dans les secteurs de spinons plus hauts, donnant une représentation en déterminants réduits, dont une structure de haut-niveau à l'échelle des facteurs de formes est révélée.
 \end{multicols}
\end{resumebox}
\vspace{\stretch{1}}
\begin{resumebox}{Abstract}
\textbf{Title :} \engtitle
\\[1ex]
\textbf{Keywords:} \engkeywords. 
\begin{multicols}{2}
\small
Since a long-time, the quantum integrable systems have remained an area where modern mathematical methods have given an access to interesting results in the study of physical systems.
The exact computations, both numerical and asymptotic, of the correlation function is one of the most important subject of the theory of the quantum integrable models.
In this context an approach based on the calculation of form factors has been proved to be a more effective one.
In this thesis, we develop a new method based on the algebraic Bethe ansatz for the computation of the form-factors in thermodynamic limit.
It is both applied to and described in the context of isotropic XXX Heisenberg chain, which is one of the examples of an interesting case of critical models where the Fermi-zone is non-compact.
In a particular case of two-spinon form-factors, we obtain an exact result in a closed-form which matches the previous result obtained from an approach based on $q$-vertex operator algebra.
This method is then generalised to form-factors in higher spinon sectors where we find a reduced determinant representation for the form-factors, in which a higher-level structure for the form-factors is revealed.
 \end{multicols}
\end{resumebox}
\vspace{\stretch{3}}
\begin{picture}(80,0)
	\put(0,2){\makebox(0,0)[lc]%
	{\includegraphics*[height=3.5em]{logos/logo_ubfc.png}}
	}
\end{picture}
\put(15,5){%
\parbox[c][1ex][s]{15em}
{
	\scriptsize
	\usekomafont{footnote}
	\footnotesize Université de Bourgogne Franche-Comté
	\par\vfill
	32, avenue de l'Observatoire
	\par\vfill
	25000 Besnaçon
}
}
\end{document}